\documentclass[11pt,a4paper,reqno]{article}

\usepackage{SdS}
\usepackage[mode=buildnew]{standalone}

\usepackage[sortcites,
backend=biber,
date=year,
style=numeric-comp,
doi=false,
isbn=false,
url=false,
eprint=true,
maxbibnames=5]{biblatex} 
 
\DeclareNameAlias{author}{family-given}
\AtEveryBibitem{
  \clearfield{month}
  \clearfield{url}
  \clearfield{urlyear}
  \clearfield{urlmonth}
  \clearfield{primaryclass}
  \clearfield{eprintclass}
  \ifentrytype{online}{
    \clearfield{year}}
  {
    \clearfield{eprint}}
}
\renewbibmacro*{volume+number+eid}{%
  \printfield{volume}%
  \setunit*{\addnbspace}
  \printfield{number}%
  \setunit{\addcomma\space}%
  \printfield{eid}}
\DeclareFieldFormat[article]{volume}{\textbf{#1}}
\DeclareFieldFormat[article]{number}{\mkbibparens{#1}}
\DeclareFieldFormat*{title}{\textit{#1}}
\DeclareFieldFormat{journaltitle}{#1\isdot}
\renewbibmacro{in:}{}
\title{Linear stability of the slowly-rotating \KdS{} family}
\author{Allen Juntao Fang\thanks{Sorbonne
    Universit\'{e}, CNRS, Laboratoire Jacques-Louis Lions (LJLL),
    F-75005 Paris, France, fanga@ljll.math.upmc.fr}} 
\date{}

\begin{document}
\maketitle

\begin{abstract}
  In this paper, we prove that the slowly-rotating \KdS{} family of
  black holes are linearly stable as a family of solutions to the
  Einstein vacuum equations with $\Lambda>0$ in harmonic (wave)
  gauge. This article is part of a series that provides a novel proof
  of the full nonlinear stability of the slowly-rotating \KdS{}
  family. This paper and its follow-up offer a self-contained
  alternative approach to nonlinear stability of the \KdS{} family
  from the original work of Hintz, Vasy \cite{hintz_global_2018} by
  interpreting quasinormal modes as $H^k$ eigenvalues of an operator
  on a Hilbert space, and using integrated local energy decay
  estimates to prove the existence of a spectral gap. In particular,
  we avoid the construction of a meromorphic continuation of the
  resolvent. We also do not compactify the spacetime, thus avoiding
  the use of $b$-calculus and instead only use standard
  pseudo-differential arguments in a neighborhood of the trapped set;
  and avoid constraint damping altogether. The methods in the current
  paper offer an explicit example of how to use the vectorfield method
  to achieve resolvent estimates on a trapping background. 
\end{abstract}


\tableofcontents

\section{Introduction}

The aim of this paper is to offer the linear theory needed for a novel
proof of the global nonlinear stability of the slowly-rotating \KdS{}
family of black hole solutions to Einstein's vacuum equations with a
positive cosmological constant. In the process of doing so, we also
hope to contribute a deeper understanding of quasinormal modes of wave
equations using a vectorfield approach.

\subsection{The black hole stability problem }
The Einstein vacuum equations (EVE), which govern Einstein's theory of
relativity under the assumption of vacuum, are given by
\begin{equation}
  \label{linear:eq:intro:EVE}
  \Ric(g) - \Lambda g = 0,
\end{equation}
where $g$ is a Lorentzian metric with signature $(-1,+1,+1,+1)$ on a
manifold $\mathcal{M}$ (which for the sake of this paper, we will
assume is $3+1$ dimensional), $\Ric$ denotes its Ricci tensor, and
$\Lambda$ is the cosmological constant.

Expanding the Ricci tensor $\Ric(g)$ as a differential operator acting
on the metric tensor $g$, equation \eqref{linear:eq:intro:EVE} is a fully
nonlinear second-order partial differential equation (PDE) for the
components of the metric tensor $g$, which is invariant under
diffeomorphism. A well-posedness theory for EVE was shown by
Choquet-Bruhat \cite{choquet-bruhat_theoreme_1952} and extended in
Choquet-Bruhat-Geroch \cite{choquet-bruhat_global_1969} by setting EVE
in harmonic gauge (also known as wave gauge in the literature),
showing the existence of a maximal globally hyperbolic development for
a sufficiently smooth initial data triplet
$(\Sigma, \InducedMetric, k)$.  We discuss the initial value problem,
along with wave coordinates, in detail in section \ref{linear:sec:EVE}. For a
detailed classical treatment, we refer to the reader to Chapter 10 in
\cite{wald_general_1984}, and for a more modern treatment, to
\cite{sbierski_existence_2016}.

Among the most interesting solutions to EVE are
the families of black hole solutions.  An explicit two-parameter
family of black hole solutions to EVE with $\Lambda>0$ is given by the
\KdS{} family of spacetimes $(\mathcal{M}, g_{M, a})$.
The fixed four-dimensional manifold is taken to be isomorphic to
$\Real^+_{\tStar}\times(0,\infty)_r\times\Sphere^2$, and the black
hole parameters $(M, a)$ denote the mass parameter and the angular
momentum of the black hole respectively. We will use $b=(M, a)$
to denote the black hole parameters. See Section \ref{linear:sec:KdS} for a
more detailed presentation of the \KdS{} family. 

The \KdS{} family black holes, like their asymptotically flat cousin
the Kerr family of black hole solutions to Einsteins equations with
vanishing cosmological constant, represent a family of rotating,
uncharged black hole solutions. In the particular case where $a=0$,
the black hole is a static, non-rotating, uncharged black hole, the
\KdS{} family reduces to the \SdS{} sub-family of black hole solutions
(mirroring the $a=0$ Schwarzschild sub-family of the Kerr family of
solutions in the asymptotically flat case).

The classical black hole stability problem is concerned with whether
the Kerr family of black hole solutions is stable as a
family\footnote{For black hole solutions, the question is generally
  stated as whether a \textit{family} of black hole solutions is
  stable rather than a particular black hole solution. This is in line
  with the expectation that nontrivial perturbations of a black hole
  should alter the mass and angular momentum of the black hole.} in
the sense that initial data triplets sufficiently close to the initial
data triplet of a given Kerr solution have a maximal development with
a domain of outer communication which globally approaches a nearby
Kerr solution. An equivalent question of stability can be formulated
for the \KdS{} family, which retains some of the crucial geometric
difficulties as Kerr (superradiance and trapping in particular), while
featuring substantially easier analysis overall due to exponential
decay at the linear level.

This paper and its sequel in \cite{fang_nonlinear_2021}
seek to address exactly the question of stability for the case of the
\KdS{} family.
\begin{theorem} [Nonlinear stability of the slowly-rotating \KdS{}
  family, informal statement]
  \label{linear:thm:intro:NL-stab:informal}
  Suppose $(\InducedMetric,k)$ are smooth initial data on some
  $3$-dimensional hypersurface $\Sigma_0$ satisfying the constraint
  conditions, which are close to the initial data
  $(\InducedMetric_{b^0}, k_{b^0})$ of a slowly-rotating \KdS{}
  spacetime $g_{b^0}$ in some initial data norm. Then there exists a
  solution $g$ to EVE with a positive cosmological constant,
  \eqref{linear:eq:intro:EVE}, such that $g$ attains the initial data at
  $\Sigma_0$ and there exist black hole parameters
  $b=(M_\infty, a_\infty)$ such that
  \begin{equation*}
    g - g_{b_\infty} = O(e^{-\SpectralGap \tStar}),\quad \tStar\to \infty
  \end{equation*}
  for some real $\SpectralGap>0$ constant that is independent of the
  initial data. 
\end{theorem}
This result was in fact first proven by Hintz and Vasy in their
seminal work \cite{hintz_global_2018}. Their proof uses a modification
of the harmonic gauge, to treat Einstein's equations as a system of
quasilinear wave equations for the metric
\begin{equation*}
  g^{\alpha\beta}\p_\alpha\p_\beta g_{\mu\nu} = \mathcal{N}_{\mu\nu}(g, \p g).
\end{equation*}
An important first step to analyzing the nonlinear stability of the
slowly-rotating \KdS{} family is to consider the linear stability of
the gauged Einstein's equations linearized around a member of the
slowly-rotating \KdS{} family. In the case of \KdS, the
linearized system is expected to exhibit exponential decay, from which
the full nonlinear stability quickly follows. Hintz and Vasy in
\cite{hintz_global_2018} approach proving exponential decay for the
linearized system by building on the numerous works on the theory of
resonances on black hole spacetimes (see for instance
\cite{bony_decay_2008,dyatlov_asymptotic_2012,dyatlov_asymptotics_2015}),
in particular relying on a series of works by themselves and Dyatlov
on applying the scattering resolvent method to the study of wave-type
equations on asymptotically (Kerr-)de Sitter backgrounds
\cite{hintz_asymptotics_2018, hintz_global_2016, hintz_resonance_2017,
  hintz_semilinear_2015, dyatlov_spectral_2016, vasy_microlocal_2013}.

\subsection{The theory of scattering resonances}

The theory of scattering resonances has proven particularly powerful
in analyzing asymptotic behavior of wave-like equations on
asymptotically (Kerr-)de Sitter spacetimes as they are well-adapted to
spectral methods, although recently, significant progress has also
been made adapting the method to asymptotically flat spacetimes
\cite{hintz_sharp_2022,vasy_resolvent_2020,hafner_linear_2021}. 

The theory of resonances on black hole spacetimes is
founded on finding an appropriate notion of characteristic
frequencies, and using them to analyze the asymptotic behavior of the
system at hand (for an in-depth explanation of the scattering
resonance method, see \cite{dyatlov_mathematical_2019}). For a
stationary linear operator $\LinearOp$, these characteristic
frequencies (resonances) are typically defined as the poles of a
meromorphic continuation of the resolvent
$\widehat{\LinearOp}(\sigma)^{-1}$ into the lower half-space of the
complex plane $\Complex$\footnote{In the notation used in this
  article, consistent with most of the literature, the lower
  half-space of the complex plane is the half-space corresponding to
  exponential decay. However, in some articles in the literature, the
  half-space corresponding to exponential decay is instead the left
  half-space.}. The location of the resonances within the complex
plane then determine the asymptotic behavior of solutions $\psi$ to
$\LinearOp\psi=0$. In particular, if it is possible to show that all the
resonances are located in the lower half space, then a simple contour
deformation argument proves exponential decay (see for example the
analysis of the scalar wave equation on a \SdS{} background in
\cite{bony_decay_2008}). For the Einstein system in harmonic gauge
linearized around a slowly-rotating \KdS{} background, it turns out
that it is not true that the resonances are all located in the
exponentially decaying half-space of $\Complex$.  However, by proving
a high-frequency resolvent estimate, it remains possible to derive the
existence of a high-frequency spectral gap\footnote{This is typically
  the primary difficulty in the method of scattering resonances, and
  has often relied on the use of Melrose's b-calculus and microlocal
  analysis. }. In the case of wave-type equations on \KdS, this relies
on a detailed understanding of the nature of the instability of the
trapped set in \KdS{} \cite{dyatlov_spectral_2016, bony_decay_2008,
  dyatlov_asymptotics_2015, hintz_non-trapping_2014}.  Following a
contour deformation argument, this allows a decomposition of any
solution into a finite sum of linear obstacles to decay which grow at
a bounded exponential rate, and an exponentially decaying remainder.

For a general nonlinear wave equation, this is as far as the linear
analysis could go. Such a result, of course is entirely unsatisfactory
for nonlinear decay. After all, growth at the linear level
heuristically leads to even more growth at the nonlinear level, and
exponential linear growth would destroy any hopes for nonlinear
decay. What turns out to save nonlinear stability in the case of the
Einstein equations is the geometric structure of Einstein's equations
themselves. Indeed, at the level of the linearized Einstein equations
linearized around \SdS, it is known that the non-decaying resonances
are in fact unphysical
\cite{ishibashi_stability_2003,kodama_brane_2000,kodama_master_2003}
(see Section \ref{linear:sec:mode-stability} for an detailed discussion, and
Section 4.1 and Section 7 of \cite{hintz_global_2018} for comparison).
Such a result however, is unknown for Einstein's equations linearized
around even slowly-rotating \KdS{} backgrounds. To overcome this,
Hintz and Vasy introduce constraint damping to the gauged linearized
Einstein equations and use perturbation theory on the
constraint-damped gauged linearized Einstein equations to deduce that
the remaining finite linear obstacles to decay are in fact unphysical,
thus allowing them to conclude exponential decay at the linear level.
Finally, to extend the linear stability results to nonlinear stability, Hintz
and Vasy make use of a Nash-Moser argument
\cite{hintz_global_2016,hintz_global_2018}.

\subsection{The vectorfield method} 

On the other hand, much recent progress in understanding stability of
black hole spacetimes (and asymptotic behavior of linear fields on
black hole backgrounds more generally), has followed from the
development of the vectorfield method. First used to analyze scalar
wave equations, the vectorfield method proved incredibly relevant to
the question of black hole stability as a critical component of the
breakthrough proof of nonlinear stability of Minkowski spacetime as a
solution to EVE with $\Lambda=0$
\cite{christodoulou_global_1993}. Over the years, many alternative
proofs for the stability of Minkowski space using variants of the
vectorfield method have been developed \cite{hintz_stability_2020,
  lindblad_global_2010-1, klainerman_evolution_2003}. The vectorfield
method has also proven useful beyond the study of just Minkowski
space, having led to substantial developments in the study of linear
waves on black hole backgrounds \cite{dafermos_red-shift_2009,
  dafermos_decay_2010, schlue_decay_2013, dafermos_decay_2016,
  marzuola_strichartz_2010, tataru_local_2010,
  ionescu_global_2015}. More directly relevant to the subject of this
paper, the vectorfield method has also led to numerous developments in
the understanding of stability of black hole spacetimes. This includes
linearized stability of the Schwarzschild family
\cite{dafermos_linear_2019} (see also \cite{johnson_linear_2018} and
\cite{hung_linear_2020}), the nonlinear stability of the Schwarzschild
family \cite{klainerman_global_2020, dafermos_non-linear_2021},
linearized stability of the Kerr family
\cite{andersson_stability_2019}, as well as the nonlinear stability of
slowly-rotating Kerr \cite{klainerman_kerr_2021}.

The vectorfield method, as originally conceived, works by exploiting
Killing, conformal Killing, and almost-conformal Killing symmetries to
define vectorfield multipliers and commutators that can be used to
derive various energy estimates. Exploiting symmetries though can only
be as strong as one has symmetries to exploit, and in the case of
Schwarzschild(-de Sitter) and Kerr(-de Sitter) spacetimes, there are
frequently not enough (conformal) Killing vectorfields to exploit
directly. To compensate for this lack of symmetries, extensions of the
classical vectorfield method have constructed new vectorfields that
have coercive deformation tensors on different regions of spacetime,
and have extended the method to consider more general Killing tensors
\cite{andersson_hidden_2015}. The starting point, and typically the
most involved portion of the method is the proof of an
\textit{integrated local energy decay} (ILED) estimate, also commonly
referred to in the literature as a \textit{Morawetz estimate}
\cite{pieter_blue_semilinear_2003}.

Despite the differences in the origins of the scattering resonance
method and the vectorfield method, Morawetz estimates and resolvent
estimates are known to be intrinsically tied. On non-trapping asymptotically
flat space-times, Metcalfe, Sterbenz, and Tataru have shown an
equivalence between resolvent estimates and Morawetz estimates (modulo
certain conditions on the real axis) for symmetric wave-type operators
\cite{metcalfe_local_2017}.  This suggests that despite the initial
differences between the approaches to black-hole stability by the
scattering resonance and the vectorfield communities, there are
actually deep similarities between the two methods. This is further
suggested by recent work done by Warnick and Gajic
\cite{warnick_quasinormal_2015,gajic_quasinormal_2019} using ideas
developed in the new vectorfield method to reinterpret the resonances
of the scattering resonance method as true eigenvalues of an operator
on a Hilbert space.

Finally, we mention that there has recently been work done by
Mavrogiannis using the vectorfield method to prove exponential decay for
scalar quasilinear waves on \SdS{} without reference to spectral
theory \cite{mavrogiannis_morawetz_2021,
  mavrogiannis_quasilinear_2021}.

\subsection{Statement of the main result}

Motivated by these insights, in this paper and its sequel, we attempt
to bridge the gaps between these two methods in a novel proof of the
full nonlinear stability of the slowly-rotating \KdS{} family. Our
result aims to use ideas and results from the vectorfield method in
order to recover the necessary exponential decay at the linear level
to be applied in a full proof of nonlinear stability. The statement
that we wish to prove is then as follows (for the more formal
statement, see Theorem \ref{linear:thm:Main}).
\begin{theorem}[Main Theorem, version 1]
  \label{linear:thm:intro:main}
  Given initial data $(\InducedMetric,k)$ close to the initial data of
  a slowly-rotating \KdS{} black hole $(\InducedMetric_{b}, k_{b})$,
  there exists a solution to the linearized Einstein vacuum equations
  in harmonic gauge that is exponentially decaying to an infinitesimal
  diffeomorphism of some nearby linearized slowly-rotating \KdS{} metric
  $g_{b}'(b')$.
\end{theorem}

We highlight the
main differences between our proof and the original proof for linearized
stability of Hintz and
Vasy (see Theorem 10.5 of \cite{hintz_global_2018}) below.
\begin{enumerate}
\item We do not compactify the spacetime, and as a result, avoid any
  use of the $b$-calculus. In fact, we only use classical
  pseudo-differential arguments in a neighborhood of trapping in a
  single part of the derivation of the Morawetz estimate (in Section
  \ref{linear:ILED:near}) to overcome the frequency-dependent nature of the
  trapped set in \KdS, avoiding all other microlocal arguments and
  techniques in the rest of the proof.  This is in line with the idea
  that for slowly-rotating \KdS{} spacetimes, trapping is the only
  frequency-dependent behavior\footnote{For slowly-rotating \KdS{}
    spacetimes, like in slowly-rotating Kerr spacetimes, superradiance
  is both separated from trapping and effectively handled by
  redshift. Passing to the more general case, superradiance is also a
  frequency-dependent obstacle to stability.}. It may in fact be possible to remove
  even this frequency-based argument by following similar methods as
  \cite{andersson_hidden_2015} to produce a purely physical argument,
  but this is not pursued further here.
\item By using the vectorfield method, we are able to adapt Warnick's
  work on anti-de Sitter spaces \cite{warnick_quasinormal_2015} to
  \KdS, resulting in a characterization of quasinormal modes as
  eigenvalues of an operator on a Hilbert space, thus
  avoiding the need to construct a meromorphic continuation of the
  resolvent.
\item We ascertain the ``spectral gap'' explicitly for the gauged
  linearized Einstein equation using a Morawetz estimate, following the
  approach taken by Tataru and Tohaneanu in \cite{tataru_local_2010}
  to prove a similar Morawetz estimate for scalar
  waves on Kerr. Doing so allows our proof to remain self-contained
  while also providing an explicit example of the equivalence between
  Morawetz estimates and resolvent estimates on a trapping background.
\item We obtain the desired mode stability result (a statement of the
  non-physical nature of the non-decaying quasinormal mode solutions
  of the gauged linearized Einstein operator) by perturbing a
  geometric mode stability result for \SdS{} (see Theorem
  \ref{linear:thm:Kodama-Ishibashi}). However, unlike the similar
  perturbation done by Hintz and Vasy in \cite{hintz_global_2018}, we
  do not introduce constraint damping, and instead proceed by a
  precise analysis of the linearized constraint equations and the
  constraint propagation equation induced by the Einstein vacuum
  equations.
\end{enumerate}
These goals at the linear level will allow us in
\cite{fang_nonlinear_2021} to give a novel self-contained proof of the
full nonlinear stability of the slowly-rotating \KdS{} family that
uses a bootstrapping argument, rather than a Nash-Moser argument to
close nonlinear stability.

We provide a brief discussion of the main difficulties involved in
achieving the goals of the paper. It is well known that on black hole
spacetimes, there tend to be two underlying geometric difficulties
standing in the way of a proof of linear stability\footnote{These
  geometric difficulties are inherent to the geometry of black hole
  spacetimes rather than Einstein's equations, and are present even
  for simpler problems, such as decay for scalar waves on black hole
  backgrounds.}. The first is the issue of superradiance and the
loss of a timelike Killing vector at the horizons. In fact on \KdS{}
spacetimes, much like on its Kerr cousin, there is no global timelike
Killing vectorfield on the domain of exterior
communication\footnote{On \SdS, there is a timelike Killing
  vectorfield up to the horizons, where it becomes null. See Section
  \ref{linear:sec:setup:killing} for more details.}. What has arisen as a
powerful solution in the new vectorfield method is defining a new
vectorfield that captures exactly within a neighborhood of the
horizon, the exponential decay corresponding to the \textit{redshift
  effect} \cite{dafermos_red-shift_2009,dafermos_lectures_2008}. The
second, and more problematic difficulty is the existence of trapped
null geodesics on the interior of the black hole exterior
region. Fortunately, on Kerr(-de Sitter), these trapped null geodesics
are unstable, in the sense that locally, energy disperses away from
the trapped set \cite{dyatlov_spectral_2016, hintz_non-trapping_2014,
  marzuola_strichartz_2010, tataru_local_2010, andersson_hidden_2015,
  dafermos_decay_2016}. This is the main idea which we will rely on to
prove the desired Morawetz estimates.

Finally, there is the difficulty of proving a mode stability
statement. While there is a powerful geometric mode stability (GMS)
statement on the ungauged linearized Einstein equations linearized
around a member of the \SdS{} family \cite{kodama_master_2003,
  kodama_brane_2000, ishibashi_stability_2003}, perturbation theory
does not directly yield a geometric mode stability statement for even
nearby \KdS{} black holes. Indeed, rather than working at the level of
the ungauged equations, we will use perturbation theory at the level
of the gauged Einstein equations and the constraint propagation
equation, which are principally wave and more amenable to
perturbative methods.

\subsection{Outline of the paper}

In Section \ref{linear:sec:Setup}, we set up the main geometric aspects of
the problem. In particular, we construct various regular coordinate
systems on the \KdS{} family, and show the instability of the trapped
set. We also identify the vectorfield multipliers
that will be used subsequently to prove Killing and redshift energy
estimates respectively.
In Section \ref{linear:sec:EVE}, we then give a brief overview of
Einstein's vacuum equations and harmonic gauge, defining the gauged
linearized Einstein operator that will be the focus of the rest of the
paper. We also compute the important properties of the subprincipal
symbol at the trapped set and the horizons that will be crucial to the
rest of the proof. 
In Section \ref{linear:sec:QNM}, we define the
solution semigroup associated to a strongly hyperbolic operator, and
use these $C^0$-semigroups to define the quasinormal
spectrum. We also in this section define the Laplace-transformed
operator which will prove as a useful intermediary in studying the
quasinormal spectrum. 

Having established the important definitions, we state the main
theorem in Section \ref{linear:sec:main-theorem}. We also provide a
breakdown of the main intermediary results: exponential decay up to a
compact perturbation, and geometric mode stability.  Section
\ref{linear:sec:energy-estimates} consists of the estimates which are blind
to the presence of the trapped set, namely, the standard Killing
energy estimate, and (enhanced) redshift estimate, which will serve
as the technical basis for proving a Fredholm alternative for the
spectrum.

In Section \ref{linear:sec:freq-analysis}, we set up the necessary tools for
the frequency analysis in Section \ref{linear:sec:ILED}. The Morawetz
estimate in Section \ref{linear:sec:ILED} is both the most difficult and the
most important step to proving exponential decay up to a compact
perturbation. It requires analyzing both the physical and frequency
space behavior of trapping and defining energies that are well-adapted
to capturing the behavior at the trapped set. These estimates are the
key to proving a spectral gap for the spectrum, and form the core of
the paper.  In Section \ref{linear:sec:asymptotic-expansion}, we show how to
use the energy estimates that we have derived in the previous sections
to deduce information about the spectrum, in particular, showing that
there are only a finite number of spectral obstacles to exponential
decay of solutions to the linearized equations.

The mode stability of the gauged linearized Einstein operator is dealt
with in Section \ref{linear:sec:mode-stability}. This is another key
component to the proof. We begin with a review of mode stability in
\SdS, for which strong results are already known, before showing that
the unphysical nature of quasinormal mode solutions is conserved under small
perturbations to nearby linearizations of
Einstein's equations around \KdS.

Finally, in Section \ref{linear:sec:proof-of-main-thm}, we  prove the
main theorem, using the tools developed in all the preceding
sections.

\subsection{Acknowledgments}

The author would like to acknowledge J\'{e}r\'{e}mie Szeftel for
his encouragement and support. This work is supported by the ERC grant
ERC-2016 CoG 725589 EPGR.

\section{Geometric set up}
\label{linear:sec:Setup}

In this section, we define key geometric objects that
we will make use of later.

\subsection{Notational conventions}

Many of the inequalities in this paper feature implicit constants. We
use the following notation. Generically, $\epsilon, \delta$, are
auxiliary small constants. $C$ is used to denote large auxiliary
constants. $C(\epsilon)$ (or equivalently $C(\delta)$) indicates a
large constant depending on monotonically on $\epsilon$ such that
$\lim_{\epsilon\to 0}C(\epsilon) = \infty$. In particular,
inequalities with $C(\epsilon)$ still hold if $C(\epsilon)$ is
replaced with a larger constant. Subscripts in constants are then used
to denote additional dependencies of the constants.

Throughout the paper, Greek indices will be used to indicate the
spacetime indices $\{0,1,2,3\}$, lower-case Latin indices will be used
to represent the spatial indices $\{1,2,3\}$, and upper-case Latin
indices will be used to represent angular indices. We also denote
$\ImagUnit:= \sqrt{-1}$ to avoid confusion with the index $i$.

We will use the standard musical isomorphism notation to denote
$X^\flat$ the canonical one-form associated to a vectorfield $X$, and
$\omega^\sharp$ the canonical vectorfield associated to a one-form
$\omega$. 

We will use $T^*\Manifold$ $S^2T^*\Manifold$ to refer to the cotangent
bundle and the bundle of symmetric two tensors on $\Manifold$
respectively.

We will use $\nabla$ to denote the full spacetime covariant
derivative, and $D = \ImagUnit\p$. 

If $M$ and $N$ are matrices, we write
\begin{equation*}
  M\cdot N = M^TN,
\end{equation*}
where $M^T$ denotes the transpose of $M$. 

\subsection{The \KdS{} family}
\label{linear:sec:KdS}

The \KdS{} family of black holes, which will be presented explicitly
in what follows, is a family of stationary black hole solutions to
the Einstein vacuum equations (EVE) with a positive cosmological
constant $\Lambda>0$. The two-parameter family is parameterized by
\begin{enumerate}
\item the mass of the black hole $M$ and,
\item the angular momentum of the black hole $a = |\AngularMomentum|$,
  where $\frac{\AngularMomentum}{|\AngularMomentum|}$ is the axis of
  symmetry of the black hole. 
\end{enumerate}
We will denote by $B$ the set of black-hole parameters $(M, a)$.  In
this paper, we will not deal with the full \KdS{} family, but instead
are primarily concerned with only a subfamily characterized by two
features: first, that the mass of the black hole is subextremal and
satisfies $1-9\Lambda M^2>0$; and that the black hole is slowly
rotating, $\frac{\abs*{a}}{M},\frac{\abs*{a}}{\Lambda}\ll 1$. The
subextremality of the mass ensures that the event horizon and the
cosmological horizon remain physically separated, and the slow
rotation ensures that the trapped set remains physically separated
from both the cosmological and the black-hole ergoregions.

\subsubsection{The \SdS{} metric}
\label{linear:sec:SdS}

Given a cosmological constant $\Lambda$, and a black hole mass $M>0$
such that
\begin{equation}
  \label{linear:eq:SdS:non-degeneracy-condition}
  1 - 9\Lambda M^2>0,
\end{equation}
we denote by
\begin{equation*}
  b_0 = (M, \mathbf{0})
\end{equation*}
the black hole parameters for a mass-subextremal \SdS{} black
hole. The \SdS{} family represents a family of spherical, non-rotating
black hole solutions to Einstein's equations with positive
cosmological constant. On the domain of outer communication (also
known in the literature as the static region),

\begin{equation}
  \label{linear:eq:static-region-def}
  \StaticRegion := \Real_t\times \widetilde{\Sigma}_{\tStar}^\circ, 
\end{equation}
where
\begin{equation}
  \label{linear:eq:Sigma-tilde:def}
  \widetilde{\Sigma}_{\tStar}:= [r_{b,\EventHorizonFuture}, r_{b,\CosmologicalHorizonFuture}]\times\Sphere^2,
\end{equation}
with $r_{b_0,\EventHorizonFuture}$,
$r_{b_0, \CosmologicalHorizonFuture}$ defined below, the \SdS{} metric
a can be expressed in standard Boyer-Lindquist coordinates by
\begin{equation}
  \label{linear:eq:SdS:metric-def:BL}
  g_{b_0} = -\mu_{b_0}dt^2 + \mu_{b_0}^{-1}dr^2 + r^2\UnitSphereMetric,
\end{equation}
where
\begin{equation*}
  \mu_{b_0}(r) = 1 - \frac{2M}{r}- \frac{\Lambda r^2}{3},
\end{equation*} 
and $\UnitSphereMetric$ denotes the standard metric on
$\UnitSphere^2$. The subextremal mass restriction in
(\ref{linear:eq:SdS:non-degeneracy-condition}) guarantees that $\mu_{b_0}(r)$
has three roots: two positive simple roots, and one negative simple
root, 
\begin{equation*}
  r_{b_0,-}<0<r_{b_0,\EventHorizonFuture}<r_{b_0,\CosmologicalHorizonFuture}<\infty.
\end{equation*}
The $r$-constant hypersurfaces defined by
$\{r = r_{b_0, \EventHorizonFuture}\}, \{r = r_{b_0,
  \CosmologicalHorizonFuture}\}$ are respectively the \textit{(future)
  event horizon} and the \textit{(future) cosmological horizon}, and
bound the domain of outer communications,
$\Real_t^+\times (r_{b_0, \EventHorizonFuture}, r_{b_0,
  \CosmologicalHorizonFuture})\times \UnitSphere^2$, on which the form
of the metric in (\ref{linear:eq:SdS:metric-def:BL}) is valid.

The form of the metric $g_{b_0}$ in Boyer-Lindquist coordinates has
one major short-coming, namely, its apparent singularity when
$\mu_{b_0}=0$, which occurs exactly at the event horizon and the
cosmological horizon. Fortunately, this is only a coordinate
singularity, and can be resolved by a change of coordinates. We
present two candidates in the following section.

\subsubsection{Regular coordinates on \SdS}
\label{linear:sec:SdS:regular}

We now construct some coordinate systems on \SdS{} that are
regular at the horizons, and will be used in calculations throughout
the rest of the paper. To fix this, we introduce the
change of coordinates
\begin{equation*}
  \tStar = t - F_{b_0}(r),
\end{equation*}
where $F_{b_0}$ is a smooth function of $r$. In the coordinates
$(\tStar , r, \omega)$, the \SdS{} metric, $g_{b_0}$, and the inverse
metric, $G_{b_0}$, take the form
\begin{equation}
  \label{linear:eq:SdS:regular}
  \begin{split}
    g_{b_0} &= -\mu_{b_0} \,d\tStar^2
    + 2F_{b_0}'(r)\mu_{b_0} \,d\tStar dr
    +(\mu_{b_0} F_{b_0}'(r)^2 +\mu_{b_0}^{-1})\,dr^2
    + r^2\UnitSphereMetric,\\
    G_{b_0} &= -\left(\mu_{b_0}^{-1} +\mu_{b_0} F_{b_0}'(r)^2\right)\p_{\tStar}^2
    + 2F_{b_0}'(r)\mu_{b_0}\,\p_{\tStar}\p_r
    + \mu_{b_0}\p_r^2 + r^{-2}\UnitSphereInvMetric.
  \end{split}
\end{equation}
We will require two main conditions be satisfied by the new $(\tStar,
r, \omega)$ coordinate system: that the $\tStar$-constant
hypersurfaces are uniformly spacelike even slightly beyond the
horizons, and that in a neighborhood of the trapped set, the new
coordinate system is identical to the Boyer-Lindquist coordinate
system $(t,r,\omega)$. 

\begin{lemma}
  \label{linear:lemma:SdS:Kerr-star-regular-coordinates}
  Fix some interval
  $\Interval_{b_0}\subset (r_{b_0,\EventHorizonFuture},
  r_{b_0,\CosmologicalHorizonFuture})$. Then there exists a choice of
  $F_{b_0}(r)$ such that
  \begin{enumerate}
  \item the $\tStar$-constant hypersurfaces are space-like. That is,
    that
    \begin{equation*}
      -\frac{1}{\mu_{b_0}} + F_{b_0}'(r)^2\mu_{b_0} < 0;
    \end{equation*}
  \item $F_{b_0}(r)$ satisfies that
    \begin{equation*}
    F_{b_0}(r) \ge 0,\qquad  r\in(r_{b_0,\EventHorizonFuture},r_{b_0,\CosmologicalHorizonFuture}),
  \end{equation*}
  with equality for $r\in \Interval_{b_0}$.
  \end{enumerate}
\end{lemma}
\begin{proof}
  See appendix~\ref{linear:appendix:lemma:SdS:Kerr-star-regular-coordinates}.
\end{proof}

Crucially, we can extend $F_{b_0}$ in an arbitrary manner smoothly
beyond the horizons $\EventHorizonFuture$,
$\CosmologicalHorizonFuture$. This allows us to consider the
extended domain
\begin{align}  
  \StaticRegionWithExtension
  &:=
    \Real_{\tStar}
    \times (r_{b_0, \EventHorizonFuture} - \varepsilon_{\StaticRegionWithExtension}, r_{b_0, \CosmologicalHorizonFuture} + \varepsilon_{\StaticRegionWithExtension})
    \times \Sphere^2\label{linear:eq:extended-region-def},\\
  \Sigma
  &:= (r_{b_0, \EventHorizonFuture} - \varepsilon_{\StaticRegionWithExtension}, r_{b_0, \CosmologicalHorizonFuture} + \varepsilon_{\StaticRegionWithExtension}) \label{linear:eq:Sigma-def},
\end{align}
for some $\varepsilon_{\StaticRegionWithExtension}>0$ small, on which
$g_{b_0}$ as defined by the expression in \eqref{linear:eq:SdS:regular} is a
smooth Lorentzian metric satisfying Einstein's equations. We also
define
\begin{equation}
  \label{linear:eq:extended-horizon-def}
  \EventHorizonFuture_-:=\curlyBrace*{r_{b_0, \EventHorizonFuture} - \varepsilon_{\StaticRegionWithExtension}},\qquad
  \CosmologicalHorizonFuture_+:=\curlyBrace*{r_{b_0, \CosmologicalHorizonFuture} + \varepsilon_{\StaticRegionWithExtension}}.
\end{equation}

It will often be useful to consider the outgoing and ingoing
Eddington-Finkelstein coordinates given by the case where
$F_{b_0}'(r) = \pm \mu_{b_0}(r)^{-1}$ respectively for specific
calculations. In this case, the \SdS{} metric and inverse metric are
\begin{equation}
  \begin{split}
    g_{b_0} &= -\mu_{b_0}dt_0^2 \pm 2dt_0dr + r^2\UnitSphereMetric,\\
    G_{b_0} &= \mp 2 \p_{t_0}\p_r + \mu_{b_0}\p_r^2 + r^{-2}\UnitSphereInvMetric.
  \end{split}
\end{equation}

\begin{figure}[h]
  \centering
  \includegraphics[width=.8\textwidth]{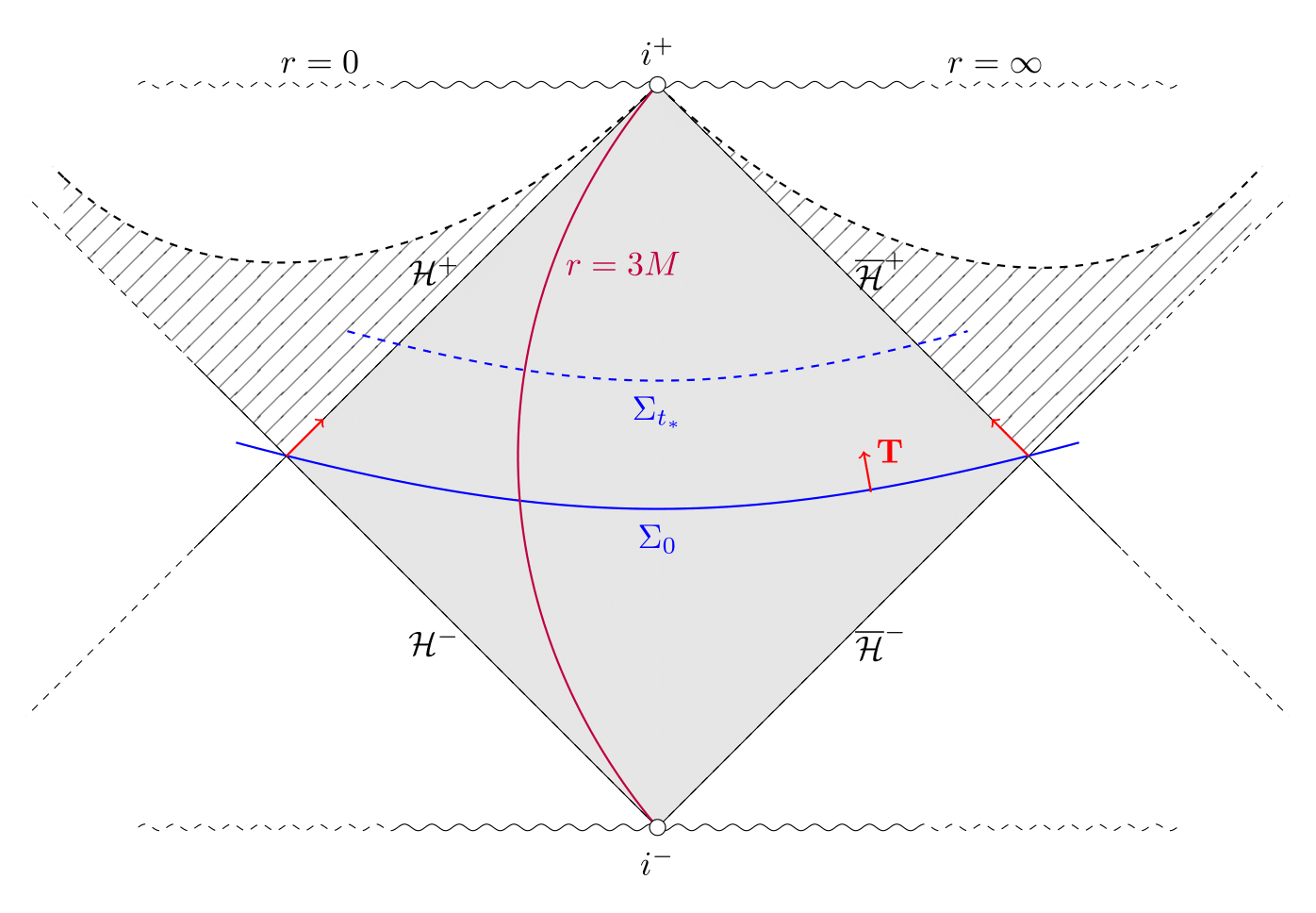}
  \caption{Penrose diagram of \SdS{} space. The solid gray region
    represents the domain of outer communication, in which the metric
    in \eqref{linear:eq:SdS:metric-def:BL} is valid. The gray dashed regions
    beyond the horizons represent the extent to which we can smoothly
    extend the metric in \eqref{linear:eq:SdS:metric-def:BL} by the regular
    coordinates in Lemma
    \ref{linear:lemma:SdS:Kerr-star-regular-coordinates}. Here,
    $\EventHorizon^{\pm}$ and $\CosmologicalHorizon^{\pm}$ denote the
    future/past event and cosmological horizons respectively, and
    $i^\pm$ denotes future/past timelike infinity. We have also
    highlighted two level sets of the $\tStar$ coordinate, the photon
    sphere $r=3M$, and the Killing vectorfield $\KillT$, which is
    future-oriented timelike on the domain of outer communication, and
    null on the horizons.}
  \label{linear:fig:SdS-penrose}
\end{figure}

\subsubsection{The \KdS{} metric}
\label{linear:sec:KdS:metric}

The \SdS{} family represents a stationary, spherically symmetric
family of black-hole solutions to Einstein's equations. On the other
hand, the \KdS{} family represents a stationary,
\textit{axi-symmetric}, family of black-hole solutions to Einstein's
equations, of which the \SdS{} family is a sub-family. In this
section, we detail various useful coordinate systems that we use
subsequently. Throughout the paper, we are mainly interested in \KdS{}
metrics that are slowly-rotating, i.e. that
$a = \abs{\AngularMomentum}\ll \Lambda, M$, and thus close to a \SdS{}
relative.

\begin{definition}
  In the Boyer-Lindquist coordinates
  $(t, r, \theta, \varphi)\in \Real\times (r_{b, \EventHorizonFuture},
  r_{b, \CosmologicalHorizonFuture})\times (0,\pi)\times
  \Sphere^1_\varphi$ (with $r_{b, \EventHorizonFuture}$,
  and $r_{b, \CosmologicalHorizonFuture}$ defined below), the \KdS{}
  metric $g_b:=g(M,\AngularMomentum)$, and inverse metric
  $G_b := G(M, \AngularMomentum)$ take the form:
  \begin{equation}
    \label{linear:eq:KdS-metric:static}
    \begin{split}
      g_b &= \rho_b^2\left(\frac{dr^2}{\Delta_b}
        +  \frac{d\theta^2}{\varkappa_b}\right)
      + \frac{\varkappa_b\sin^2\theta}{(1+\lambda_b)^2\rho_b^2 }\left(a\,dt- (r^2+a^2)\,d\varphi\right)^2
      - \frac{\Delta_b}{(1+\lambda_b)^2\rho_b^2}(dt-a\sin^2\theta\,d\varphi)^2,\\
      G_{b} &=  \frac{1}{\rho_b^2}\left(\Delta_{b}\p_r^2 + \varkappa_b\p_\theta^2\right)
      + \frac{(1+\lambda_b)^2}{\rho_b^2\varkappa_b\sin^2\theta}\left(a\sin^2\theta \p_{t} + \p_\varphi\right)^2
      - \frac{(1+\lambda_b)^2}{\Delta_b\rho_b^2}\left(
        \left(r^2+a^2\right)\p_t + a\p_\varphi
      \right)^2,
    \end{split}  
  \end{equation}
  where
  \begin{gather*}
    \Delta_b := (r^2+a^2)\left(1-\frac{1}{3}\Lambda r^2\right) - 2M r,\qquad
    \rho_b^2 := r^2+a^2\cos^2\theta,\\
    \lambda_b :=\frac{1}{3}\Lambda a^2,\qquad
    \varkappa_b:=1+\lambda_b \cos^2\theta. 
  \end{gather*}
\end{definition}

\begin{remark}
  It is easy to observe that the metric reduces to the \SdS{} metric
  $g_{b_0}$ expressed in Boyer-Lindquist coordinates in
  \eqref{linear:eq:SdS:metric-def:BL} when $b=b_0$ ($a=0$). When $a\neq 0$, the
  spherical coordinates $(\theta, \varphi)$ are chosen so that
  $\frac{\AngularMomentum}{\abs{\AngularMomentum}}\in \Sphere^2$ is
  defined by $\theta=0$, and $\p_{\varphi}$ generates counter-clockwise
  rotation around the axis of rotation.  
\end{remark}

\begin{definition}
  As in the \SdS{} case, we define the \emph{event horizon} and the
  \emph{cosmological horizon} of $g_{b}$, denoted by $\EventHorizonFuture$,
  $\CosmologicalHorizonFuture$ to be the $r$-constant hypersurfaces
  \begin{equation*}
    \EventHorizonFuture:=\{r=r_{b, \EventHorizonFuture}\},\qquad
    \CosmologicalHorizonFuture:=\{r=r_{b, \CosmologicalHorizonFuture}\},
  \end{equation*}
  respectively, where
  $r_{b, \EventHorizonFuture} < r_{b,\CosmologicalHorizonFuture}$ are
  the two largest distinct positive roots of $\Delta_b$.
\end{definition}

A consequence of the implicit function theorem is that these roots
depend smoothly on the black hole parameters $b=(M, a)$.  Since these
two horizons are null, the domain of exterior communications, bounded
by $r_{b, \EventHorizonFuture}, r_{b, \CosmologicalHorizonFuture}$ is
a causal domain that is foliated by compact space-like hypersurfaces.
This point is used in a crucial way throughout what follows.

Much like the case in \SdS, the Boyer-Lindquist form of the \KdS{}
metric $g_b$ in (\ref{linear:eq:KdS-metric:static}) has a singularity at both
the event horizon and the cosmological horizon. As we have already
discussed for the \SdS{} case, this is merely a coordinate
singularity, and it is possible to construct a smooth coordinate
system that extends beyond the horizons. 

We construct such a new, Kerr-star, coordinate system explicitly (see
similar constructions in \cite[Section 5.5][]{dafermos_lectures_2008},
\cite[Section 4][]{tataru_local_2010}, \cite[Section
3.2][]{hintz_global_2018}). First define the new variables
\begin{equation}
  \label{linear:eq:KdS-regular-coordinates:Hintz-Vasy-regular-def}
  \tStar = t - F_b(r),\quad \phiStar = \varphi - \Phi_b(r),
\end{equation}
where $F_b$ and $\Phi_b$ are smooth functions on
$(r_{b_0, \EventHorizonFuture} + \varepsilon_{\StaticRegionWithExtension},
r_{\CosmologicalHorizonFuture} -
\varepsilon_{\StaticRegionWithExtension})$.
We can then compute that the metric takes the form
\begin{equation}
  \label{linear:eq:KdS:regular}
  \begin{split}
    g_b =&{}  \frac{\varkappa_b \sin^2\theta}{(1+\lambda_b)^2\rho_b^2}
    \left(
      a(d\tStar+ F_b'\,dr) - (r^2+a^2)(d\phiStar + \Phi_b'\,dr)
    \right)^2\\
    & - \frac{\Delta_b}{(1+\lambda_b)\rho_b^2}\left(
      d\tStar + F_b'\,dr
      - a \sin^2\theta(d\phiStar +  \Phi_b'\,dr)
      - \frac{(1+\lambda_b)\rho_b^2}{\Delta_b} \,dr
    \right)^2\\
    & + \frac{2}{1+\lambda_b}\left(
      d\tStar + F_b'\,dr
      - a \sin^2\theta(d\phiStar + \Phi_b'\,dr)
      - \frac{(1+\lambda_b)\rho_b^2}{\Delta_b} \,dr
    \right)dr
    +\frac{\rho_b^2}{\varkappa_b}\,d\theta^2.  
  \end{split}  
\end{equation}

We pick the $F_b, \Phi_b$ so that the $(\tStar,r,\theta,\phiStar)$
coordinate system expends smoothly beyond the horizons, is identical
to the Boyer-Lindquist coordinates on a small neighborhood of $r=3M$,
and such that the $\tStar$-constant hypersurfaces are spacelike.
\begin{lemma}
  \label{linear:lemma:KdS:Kerr-star-regular-coordinates}
  Fix an interval $\Interval_b := (r_1, r_2)$ such that
  $r_{b_0,\EventHorizonFuture}+\epsilon_{\StaticRegionWithExtension}<r_1<r_2<r_{b_0,\CosmologicalHorizonFuture}
  - \epsilon_{\StaticRegionWithExtension}$. Then we can pick
  $F_b, \Psi_b$ so that
  \begin{enumerate}
  \item the choice extends the choice of regular coordinates for
    \SdS{} in \eqref{linear:eq:SdS:regular} in the sense that when
    $b=b_0$,
    \begin{equation*}
      F_{b} = F_{b_0},\qquad \Phi_b = 0;
    \end{equation*}
  \item $F_b(r) \ge 0$ for $r\in (r_{b, \EventHorizonFuture}, r_{b,
      \CosmologicalHorizonFuture})$ with equality for $r\in
    \Interval_b$;
  \item the $\tStar$-constant hypersurfaces are space-like, and in
    particular, defining
    \begin{equation}
      \label{linear:eq:GInvdtdt-def}
      \GInvdtdt_b:= - \frac{1}{G_b(d\tStar, d\tStar)},
    \end{equation}
    we have that
    \begin{equation}
      \label{linear:eq:GInvdtdt-prop}
      1\lesssim \GInvdtdt_b \lesssim 1
    \end{equation}
    uniformly on $\StaticRegionWithExtension$;
  \item the metric $g_b$ is smooth on $\StaticRegionWithExtension$. 
  \end{enumerate}
\end{lemma}
\begin{proof}
  See appendix \ref{linear:appendix:lemma:KdS:Kerr-star-regular-coordinates}.
\end{proof}

It will often be convenient to perform calculations on \KdS{} with the
choice of 
\begin{equation*}
  F_{b}'(r) = \pm 
  \frac{(1+\lambda_b)(r^2+a^2)}{\Delta_b},
  \qquad
  \Phi_b'(r) = \pm \frac{(1+\lambda_b)a}{\Delta_b},
\end{equation*}
which correspond to the outgoing and ingoing Eddington-Finkelstein
coordinates, in which case the \KdS{} metric and inverse metric can be
expressed in the coordinates $(t_0, r, \theta, \varphi_0)$ by
\begin{align}
  g_b ={}& \frac{\varkappa_b\sin^2\theta}{(1+\lambda_b)^2\rho_b^2}(ad{t_0}-(r^2+a^2)d{\varphi_0})^2
           - \frac{\Delta_b}{(1+\lambda_b)^2\rho_b^2}\left(
           d{t_0} - a(\sin^2\theta)d{\varphi_0}\right)^2  \notag \\
    &\mp \frac{2}{(1+\lambda_b)}(d{t_0}- a(\sin^2\theta)d{\varphi_0})\,dr + \frac{\rho_b^2}{\varkappa_b}d\theta^2,\label{linear:eq:KdS:null-metric} \\
  \rho_b^2G_b ={}&
  \Delta_b\left(\p_r \mp
    \frac{1+\lambda_b}{\Delta_b}\left((r^2+a^2)\p_{t_0}+a\p_{\varphi_0}\right)\right)^2
  -\frac{(1+\lambda_b)^2}{\Delta_b}\left((r^2+a^2)\p_{t_0} +
    a\p_{\varphi_0}\right)^2 + \rho_b^2\OCal, \label{linear:eq:KdS:null-inv-metric}\\
  \rho_b^2\OCal ={}&
  \varkappa_b\p_{\theta}^2 + \frac{(1+\lambda_b)^2}{\varkappa_b\sin^2\theta}\left(a\sin^2\theta\p_{t_0} + \p_{\varphi_0}\right)^2 \notag.
\end{align}

In the Eddington-Finkelstein coordinates, the \KdS{} metric can be
written in a more condensed format. 
\begin{definition}
  Define the vectorfields
  \begin{equation}
    \label{linear:eq:KdS:Rhat-That-def}
    \begin{split}
      \HprVF=
      \begin{cases}
        \p_r - \frac{1}{\Delta_b}\HawkingVF, & r > 3M,\\
        \p_r + \frac{1}{\Delta_b}\HawkingVF, & r < 3M,
      \end{cases}
                                             &\qquad
                                               \HawkingVF = (1+\lambda_b)\left((r^2+a^2)\p_{t_0} + a\p_{\varphi_0}\right),\\
      \HprVF' =
      \begin{cases}
        \Delta_b^{-2}\p_r\Delta_b\HawkingVF \mp\frac{1}{\Delta_b}\HawkingVF', & r>3M,\\
        -\Delta_b^{-2}\p_r\Delta_b\HawkingVF \mp\frac{1}{\Delta_b}\HawkingVF', & r<3M,
      \end{cases}
                                             &\qquad
                                               \Delta_b^{-2}\p_r\Delta_b\HawkingVF \mp\frac{1}{\Delta_b}\HawkingVF',
                                               \qquad
                                               \HawkingVF' = 2(1+\lambda_b)r\p_{t_0}.  
    \end{split}
  \end{equation}
  Notice that we can then write,
  \begin{equation}
    \rho_b^2 G_b = -\frac{1}{\Delta_b}\HawkingVF^2 + \Delta_b\HprVF^2 + \rho_b^2\OCal.
  \end{equation}
  We will also define the operator
  \begin{equation}
    \label{linear:eq:CartarOp:def}
    \rho\CartarOp h := \left(\varkappa^2\abs*{\p_\theta h}^2
      + \frac{(1+\lambda)^2}{\varkappa\sin^2\theta}\abs*{\left(a\sin^2\theta \p_{t_0} + \p_{\varphi_0}\right)h}^2\right)^{\frac{1}{2}},
  \end{equation}
  defined so that
  \begin{equation*}
    \rho^2\abs*{\CartarOp h}^2 = \OCal^{\alpha\beta}\p_\alpha h \p_\beta h. 
  \end{equation*}
\end{definition}
\begin{remark}
  Observe that $\HprVF$ and $\HprVF'$ as defined in
  \eqref{linear:eq:KdS:Rhat-That-def} are not smooth vectorfields, since they
  have a singularity at $r=3M$. This is fortunately not an issue for
  us, since everywhere that we use $\HprVF$ and $\HprVF'$, we will be
  restricted to disjoint subset of $\StaticRegionWithExtension$
  localized away from $r=3M$. 
\end{remark}

\subsection{The Killing vectorfields $\KillT$, $\KillPhi$}
\label{linear:sec:setup:killing}

\begin{definition}
  Define the vectorfields $\KillT = \p_{\tStar}$,
  $\KillPhi=\p_{\phiStar}$ using the Kerr-star coordinates in
  \eqref{linear:eq:KdS:regular}. From the fact that the expression of $g_b$
  in the Kerr-star coordinates is independent of $\tStar$ and
  $\phiStar$, we immediately have that $\KillT, \KillPhi$ are Killing
  vectorfields.
\end{definition}

\begin{definition}
  On \KdS{} spacetimes, the \emph{ergoregion} is defined by
  \begin{equation*}
    \Ergoregion := \{(t,x): g_b(\KillT,\KillT)(t,x) > 0\}.
  \end{equation*}
  We define the boundary of the ergoregion, the set of points where
  $\KillT$ is null, as the \emph{ergosphere}.
\end{definition}
\begin{remark}
  \label{linear:remark:SdS:ergosphere}
  Observe that for the \SdS{} sub-family, the ergosphere is exactly
  the event horizon and the cosmological horizon, and $\KillT$ is
  timelike on the whole of the interior of the domain of outer
  communication. 
\end{remark}

The following lemma shows that on slowly-rotating \KdS{} spacetimes,
the ergoregion consists of two components,
$\Ergoregion_{\EventHorizonFuture}$, which contains
$\EventHorizonFuture$, and $\CosmologicalHorizonFuture$, which
contains $\Ergoregion_{\CosmologicalHorizonFuture}$. 
\begin{lemma}
  For any fixed $\delta_{\Horizon}\ll 1$, let
  $\BHParamNbhd_{\delta_{\Horizon}}$ be the set of \KdS{} black
  hole parameters such that
  \begin{equation*}
    \abs*{\sup_{\Ergoregion_{\EventHorizonFuture}}r - r_{b_0,\EventHorizonFuture}}
    + \abs*{\inf_{\Ergoregion_{\CosmologicalHorizonFuture}}r - r_{b_0, \CosmologicalHorizonFuture}}
    \le \delta_{\Horizon},
  \end{equation*}
  so that for $b\in \BHParamNbhd_{\delta_{\Horizon}}$, the two components of the
  ergoregion $\Ergoregion_{\EventHorizonFuture}$ and
  $\Ergoregion_{\CosmologicalHorizonFuture}$ are physically separated
  and lie within a small neighborhood of the event and cosmological
  horizons respectively. Then, for $|a|\ll \Lambda, M$ and $1-9\Lambda
  M^2 > 0$, $b\in \BHParamNbhd_{\delta_{\Horizon}}$. 
\end{lemma}

\begin{proof}
  In view of Remark \ref{linear:remark:SdS:ergosphere}, we have that
  \begin{equation*}
    \sup_{\Ergoregion_{\EventHorizonFuture}}r =
    r_{b_0,\EventHorizonFuture},
    \quad
    \inf_{\Ergoregion_{\CosmologicalHorizonFuture}}r = r_{b_0, \CosmologicalHorizonFuture},
  \end{equation*}
  and thus
  \begin{equation*}
    \lim_{a\to 0}\sup_{\Ergoregion_{\EventHorizonFuture}}r - r_{b_0,
      \EventHorizonFuture} = \lim_{a\to 0}
    \inf_{\Ergoregion_{\CosmologicalHorizonFuture}}r - r_{b_0,
      \CosmologicalHorizonFuture} = 0,
  \end{equation*}
  directly yielding the conclusion.
\end{proof}

\subsection{The horizon generators $\HorizonGen_{\EventHorizonFuture},
  \HorizonGen_{\CosmologicalHorizonFuture}$}
\label{linear:sec:setup:horizon-gen}

\begin{definition}
  In \KdS{} spacetimes, the horizons are generated by the following
  Killing \emph{horizon generators}, which are null on their respective horizons,
  and time-like in a neighborhood of their respective horizons:
  \begin{equation*}
    \HorizonGen_{\EventHorizonFuture} := \KillT + \frac{a}{r_{\EventHorizonFuture}^2 +a^2}\KillPhi,
    \qquad \HorizonGen_{\CosmologicalHorizonFuture}  := \KillT + \frac{a}{r_{\CosmologicalHorizonFuture}^2 + a^2}\KillPhi.
  \end{equation*} 
  Let us denote also
  \begin{equation*}
    \HorizonGenPhi_{\EventHorizonFuture} := \frac{a}{r_{\EventHorizonFuture}^2 + a^2},
    \qquad
    \HorizonGenPhi_{\CosmologicalHorizonFuture} := \frac{a}{r_{\CosmologicalHorizonFuture}^2 + a^2}.
  \end{equation*}
  In particular, in the \SdS{} subfamily,
  $\HorizonGen_{\EventHorizonFuture} =
  \HorizonGen_{\CosmologicalHorizonFuture} = \KillT$.   
\end{definition}

\begin{definition}
  Associated to each horizon is a \emph{surface gravity}, $\SurfaceGravity_{\Horizon}$,
  defined by
  \begin{equation}
    \label{linear:eq:surface-grav-equation}
    \nabla_{\HorizonGen_{\Horizon}} \HorizonGen = \SurfaceGravity_{\Horizon} \HorizonGen.
  \end{equation}
  On \SdS, the values for the black hole surface gravity and the
  cosmological surface gravity respectively are
  \begin{equation*}
    \SurfaceGravity_{b_0,\EventHorizonFuture} = \frac{1}{2}\p_r\mu_{b_0}(r_{\EventHorizonFuture}),\qquad
    \SurfaceGravity_{b_0,\CosmologicalHorizonFuture} = -\frac{1}{2}\p_r\mu_{b_0}(r_{\CosmologicalHorizonFuture}).
  \end{equation*}
  Notice that $\SurfaceGravity_{b_0,\EventHorizonFuture},
  \SurfaceGravity_{b_0,\CosmologicalHorizonFuture}>0$, and that this
  positivity persists in the \KdS{} case, where
  \begin{equation*}
    \SurfaceGravity_{\EventHorizonFuture} = \frac{r_{\EventHorizonFuture}\left(1-\frac{2}{3}\Lambda r_{\EventHorizonFuture}^2 - \frac{\Lambda}{3}a^2 \right)- M}{r_{\EventHorizonFuture}^2 + a^2},\qquad
    -\SurfaceGravity_{\CosmologicalHorizonFuture} = \frac{r_{\CosmologicalHorizonFuture}\left(1-\frac{2}{3}\Lambda r_{\CosmologicalHorizonFuture}^2 - \frac{\Lambda}{3}a^2 \right)- M}{r_{\CosmologicalHorizonFuture}^2 + a^2}.
  \end{equation*}
\end{definition}

\subsection{Energy momentum tensor and divergence formulas}
\label{linear:sec:EMT-and-div-thm}
In this section, we define the energy momentum tensor and some basic
divergence properties. Given a complex matrix-valued function $h$ on
$\mathcal{M}$, let us denote its complex conjugate by
$\overline{h}$. Moreover, for any 2-tensor $h_{\mu\nu}$, we denote its
symmetrization
\begin{equation*}
  h_{(\mu\nu)} = \frac{1}{2}\left(
    h_{\mu\nu} + h_{\nu\mu}
  \right).
\end{equation*}
\begin{definition}
  We define the \emph{energy-momentum tensor} to be the
  symmetric 2-tensor:
  \begin{equation*}
    \EMTensor_{\mu\nu}[h]
    = \nabla_{(\mu}\overline{h}\cdot\nabla_{\nu)}h
    - \frac{1}{2}g_{\mu\nu}\nabla_\alpha \overline{h}\cdot  \nabla^\alpha h.
  \end{equation*}
\end{definition}
The energy-momentum tensor satisfies the following divergence
property:
\begin{equation}
  \label{linear:eq:EMTensor:divergence-property}
  \nabla_\mu \tensor[]{\EMTensor}{^\mu_\nu}[h]
  = \Re\left(\nabla_\nu\overline{h} \cdot \ScalarWaveOp[g]h\right),
\end{equation}
where we denote by 
\begin{equation*}
  \ScalarWaveOp[g] = \nabla^\alpha \partial_\alpha
\end{equation*}
the scalar wave operator\footnote{We will use
  $\nabla^\alpha\nabla_\alpha$ to denote the tensorial
  wave operator.}.

This property will be the key to producing the various divergence
equations we use to derive the relevant energy estimates in the
subsequent sections.
\begin{definition}
  Let $X$ be a smooth vectorfield on $\StaticRegionWithExtension$, $m$
  be a smooth one-form on $\StaticRegionWithExtension$, and
  $q$ be a smooth function on $\StaticRegionWithExtension$.  We will
  refer to $X$ as the \emph{(vectorfield) multiplier}, to $m$ as the
  \emph{auxiliary zero-order corrector}, and to $q$ as
  the \emph{Lagrangian corrector}. Then define
  \begin{equation}
    \label{linear:eq:J-K-currents:def}
    \begin{split}
      \JCurrent{X,q,m}_\mu[{h}]
      &= X^\nu\tensor[]{\EMTensor}{_{\mu\nu}}[{h}] 
      + \frac{1}{2}q\nabla_\mu (\abs*{h}^2)
      - \frac{1}{2}\nabla_\mu q\abs*{{h}}^2
      + \frac{1}{2}m_\mu \abs*{h}^2,\\
      \KCurrent{X,q,m}[{h}]
      &= \DeformationTensor{X}{}\cdot \EMTensor[h]
      + q \nabla^\alpha{h}\cdot \nabla_\alpha\overline{{h}}
      + \frac{1}{2}\nabla_{m^\sharp}\abs{h}^2
      + \frac{1}{2}\nabla^\alpha (m_\alpha - \p_\alpha q) \abs*{{h}}^2,      
    \end{split}
  \end{equation}
  where
  \begin{equation*}
    \DeformationTensor{X}{_{\mu\nu}} := \frac{1}{2}\left(\nabla_\mu X_\nu + \nabla_\nu X_\mu\right)
  \end{equation*}
  denotes the \emph{deformation tensor} of $X$. 
\end{definition}

It will also be convenient to define the Laplace-transformed analogues
of $\JCurrent{X,q,m}[h]$, $\KCurrent{X,q,m}[h]$.
\begin{definition}
  \label{linear:def:J-K-Laplace-currents:def}
  Let $X$ be a smooth vectorfield on $\StaticRegionWithExtension$, $m$
  be a smooth one-form on $\StaticRegionWithExtension$, and
  $q$ be a smooth function on $\StaticRegionWithExtension$, and let
  $u = u(x)$. Then we define
  \begin{equation}
    \label{linear:eq:J-K-Laplace-currents:def}
    \begin{split}
      \JLaplaceCurrent{X,q,m}[u] &= e^{-2\Im\sigma\tStar}\JCurrent{X,q,m}\left[e^{-\ImagUnit\sigma\tStar}u\right]\\
      \KLaplaceCurrent{X,q,m}[u] &= e^{-2\Im\sigma\tStar}\KCurrent{X,q,m}\left[e^{-\ImagUnit\sigma\tStar}u\right],
    \end{split}    
  \end{equation}
  which are both $\tStar$-independent.
\end{definition}

Throughout the paper, we will use the following consequence of the
divergence theorem. 
\begin{prop}
  \label{linear:prop:energy-estimates:spacetime-divergence-prop}
  Let $X$ denote a sufficiently regular vectorfield on
  $\StaticRegionWithExtension$, a \KdS{} black hole spacetime, and
  $\DomainOfIntegration$ denote the spacetime region bounded by
  $\Sigma_{t_1}, \Sigma_{t_2}, \EventHorizonFuture_-$, and
  $\CosmologicalHorizonFuture_+$. Moreover, denote
  \begin{equation*}
    \EventHorizonFuture_{t_1,t_2} := \EventHorizonFuture_-\bigcap \{t_1\le \tStar\le t_2\},\qquad
    \CosmologicalHorizonFuture_{t_1,t_2} := \CosmologicalHorizonFuture_+\bigcap \{t_1\le \tStar\le t_2\}.
  \end{equation*}

  Then we have the following divergence
  property: 
  \begin{equation*}    
    -\int_{\DomainOfIntegration} \nabla_g\cdot X ={}\int_{\Sigma_{t_2}}X\cdot n_{\Sigma_{t_2}}  -
    \int_{\Sigma_{t_1}}X\cdot n_{\Sigma_{t_1}} +
    \int_{\EventHorizonFuture_{t_1,t_2}}
    X \cdot n_{\EventHorizonFuture_-}  +
    \int_{\CosmologicalHorizonFuture_{t_1,t_2}} X\cdot n_{\CosmologicalHorizonFuture_+}
    .
  \end{equation*}
  Here, $n_{\Sigma_t}$ is the future-directed unit normal on
  $\Sigma_t$, and $n_{\Horizon}$ denotes the (time-like)
  future-directed unit normal to $\Horizon$.
\end{prop}

It will also be convenient to apply the divergence property to a
spacetime domain with boundaries along $\EventHorizonFuture$ and
$\CosmologicalHorizonFuture$.

\begin{prop}
  \label{linear:prop:energy-estimates:spacetime-divergence-prop:exact-horizon}
  Let $X$ denote a sufficiently regular vectorfield on
  $\StaticRegionWithExtension$, a \KdS{} black hole spacetime, and
  $\DomainOfIntegration$ denote the spacetime region bounded by
  $\widetilde{\Sigma}_{t_1}, \widetilde{\Sigma}_{t_2}, \EventHorizonFuture$, and
  $\CosmologicalHorizonFuture$. Moreover, denote
  \begin{equation*}
    \EventHorizonFuture_{t_1,t_2} := \EventHorizonFuture\bigcap \{t_1\le \tStar\le t_2\},\qquad
    \CosmologicalHorizonFuture_{t_1,t_2} := \CosmologicalHorizonFuture\bigcap \{t_1\le \tStar\le t_2\}.
  \end{equation*}

  Then we have the following divergence
  property: 
  \begin{equation*}    
    -\int_{\DomainOfIntegration} \nabla_g\cdot X ={}\int_{\widetilde{\Sigma}_{t_2}}X\cdot n_{\widetilde{\Sigma}_{t_2}}  -
    \int_{\widetilde{\Sigma}_{t_1}}X\cdot n_{\widetilde{\Sigma}_{t_1}} +
    \int_{\EventHorizonFuture_{t_1,t_2}}
    X \cdot \HorizonGen_{\EventHorizonFuture}  +
    \int_{\CosmologicalHorizonFuture_{t_1,t_2}} X\cdot \HorizonGen_{\CosmologicalHorizonFuture}
    .
  \end{equation*}
  Here, $n_{\widetilde{\Sigma}_t}$ is the future-directed unit normal on
  $\widetilde{\Sigma}_t$, and we recall from Section \ref{linear:sec:setup:horizon-gen}
  that $\HorizonGen_{\Horizon}$ are the Killing null generators of
  $\Horizon$.
\end{prop}

For some of the estimates in this paper, the following analogue of
Propositions \ref{linear:prop:energy-estimates:spacetime-divergence-prop} and
\ref{linear:prop:energy-estimates:spacetime-divergence-prop:exact-horizon}
will be more useful.

\begin{corollary}
  \label{linear:cor:div-them:spacelike}
  Suppose $X$ is a vectorfield on a \KdS{} spacetime $g$. Then, for
  sufficiently regular $X$, the following relations hold:

  \begin{align*}    
    -\int_{\Sigma_{\tStar}}\nabla_g\cdot X\, \sqrt{\GInvdtdt}
    &= \p_{\tStar}\int_{\Sigma_{\tStar}}X \cdot n_{\Sigma_{\tStar}}
    + \int_{\CosmologicalHorizonFuture_+\bigcap \Sigma_{\tStar}} X \cdot n_{\CosmologicalHorizonFuture_+}
    +\int_{\EventHorizonFuture_-\bigcap\Sigma_{\tStar}}X \cdot n_{\EventHorizonFuture_-},\\
    -\int_{\widetilde{\Sigma}_{\tStar}}\nabla_g\cdot X\, \sqrt{\GInvdtdt}
    &= \p_{\tStar}\int_{\widetilde{\Sigma}_{\tStar}}X \cdot n_{\widetilde{\Sigma}_{\tStar}}
    + \int_{\CosmologicalHorizonFuture\bigcap \widetilde{\Sigma}_{\tStar}} X \cdot \HorizonGen_{\CosmologicalHorizonFuture}
    +\int_{\EventHorizonFuture\bigcap\widetilde{\Sigma}_{\tStar}}X \cdot \HorizonGen_{\EventHorizonFuture},          
  \end{align*}
  where we recall $\GInvdtdt$ as defined in \eqref{linear:eq:GInvdtdt-def}. 
\end{corollary}

\begin{proof}
  This follows immediately from Proposition
  \ref{linear:prop:energy-estimates:spacetime-divergence-prop} and the
  co-area formula. 
\end{proof}

In proving energy estimates, we will typically apply the
divergence formulas in Proposition
\ref{linear:prop:energy-estimates:spacetime-divergence-prop} and Corollary
\ref{linear:cor:div-them:spacelike} in the form
\begin{equation}
  \label{linear:eq:div-them:J-K-currents}
  \nabla_g\cdot \JCurrent{X,q,m}[{h}] = \Re\squareBrace*{
    (X+q)\overline{{h}} \cdot
      \ScalarWaveOp[g]{h} 
  }
  + \KCurrent{X,q,m}[{h}].
\end{equation}
We will refer to $X$ and $q$ as the \emph{vectorfield (multiplier)},
and the \emph{Lagrangian correction} respectively.

By combining Proposition
\ref{linear:prop:energy-estimates:spacetime-divergence-prop} and Corollary
\ref{linear:cor:div-them:spacelike} with the divergence relation
\eqref{linear:eq:div-them:J-K-currents}, we will have the divergence
relation over a spacetime domain $\DomainOfIntegration =
[t_1,t_2]_{\tStar}\times\Sigma$.
\begin{align}  
  \Re\int_{\DomainOfIntegration} (X + q)\overline{{h}}\cdot
  \ScalarWaveOp[g_b]{h}
  + \int_{\DomainOfIntegration} \KCurrent{X,q,m}[{h}]
  ={}& -\int_{\Sigma_{{t}_2}} \JCurrent{X,q,m}[{h}]\cdot n_{\Sigma_{{t}_2}} 
       + \int_{\Sigma_{{t}_1}} \JCurrent{X,q,m}[{h}]\cdot  n_{\Sigma_{{t}_1}}\notag  \\      
  &- \int_{\EventHorizonFuture_{t_1,t_2}} \JCurrent{X,q,m}[{h}]\cdot  n_{\EventHorizonFuture} 
    - \int_{\CosmologicalHorizonFuture_{t_1,t_2}} \JCurrent{X,q,m}[{h}]\cdot n_{\CosmologicalHorizonFuture}; \label{linear:eq:div-thm:spacetime}
\end{align}
and its equivalent formulation over a space-like slice
\begin{align}
  &\Re\int_{\Sigma_{\tStar}}(X+q)\overline{{h}}\cdot
    \ScalarWaveOp[g_b]{h} \,\sqrt{\GInvdtdt}
    + \int_{\Sigma_{\tStar}}\KCurrent{X,q,m}[{h}]\,\sqrt{\GInvdtdt}\notag \\
  ={}& - \p_{\tStar}\int_{\Sigma_{\tStar}}\JCurrent{X,q,m}[{h}] \cdot n_{\Sigma_{\tStar}}       
       - \int_{\EventHorizonFuture\bigcap\Sigma_{\tStar}}\JCurrent{X,q,m}[{h}]\cdot  n_{\EventHorizonFuture}
       - \int_{\CosmologicalHorizonFuture\bigcap\Sigma_{\tStar}} \JCurrent{X,q,m}[{h}] \cdot n_{\CosmologicalHorizonFuture}. \label{linear:eq:div-thm:spacelike}
\end{align}
Analogous statements as above hold for the case
$\DomainOfIntegration = [t_1,t_2]_{\tStar}\times\widetilde{\Sigma}$.

\subsection{The redshift vectorfields $\RedShiftN$}
\label{linear:sec:set-up:redshift}

In this subsection, we recall the construction of the redshift
vectorfield $\RedShiftN$.

\begin{prop}
  \label{linear:prop:redshift:N-construction}
  Let $b=(M, a)$, $|a|\ll M, \Lambda$ be the black hole parameters for
  a slowly-rotating \KdS{} black hole, and let $\Sigma$ be a
  $\tStar$-constant uniformly spacelike hypersurface. Moreover, fix
  some vectorfield $X$ which is tangent to both $\EventHorizonFuture$
  and $\CosmologicalHorizonFuture$, and some
  $\varepsilon_{\RedShiftN}>0$. Then, there exist a stationary time-like
  vectorfield $\RedShiftN$, positive constants $c_{\RedShiftN}$ and
  $C_{\RedShiftN}$, and parameters
  $r_{\EventHorizonFuture}< r_0<r_1<R_1<R_0<
  r_{\CosmologicalHorizonFuture}$ such that the following conditions
  are fulfilled. 
  \begin{enumerate}
  \item On $r\le r_0$ or $r>R_0$,
    \begin{equation}
      \label{linear:eq:redshift:DeformTen-redshift-control}
      \KCurrent{\RedShiftN, 0, 0}[h] \sqrt{\GInvdtdt}\ge
      \left(\SurfaceGravity_{\Horizon}-\varepsilon_{\RedShiftN}\right)
       \JCurrent{\RedShiftN, 0, 0}[h]\cdot n_{\Sigma}
      + \abs*{X h}^2,
    \end{equation}
    where $\Horizon = \EventHorizonFuture$ for $r\le r_0$, and
    $\Horizon = \CosmologicalHorizonFuture$ for $r>R_0$
  \item For $r_0\le r\le R_0$,
    \begin{equation}
      \label{linear:eq:redshift:DeformTen-redshift-weak-control}
      -\KCurrent{\RedShiftN, 0, 0}[{h}] \le C_{\RedShiftN}
      \JCurrent{\RedShiftN, 0, 0}[{h}]\cdot n_{\Sigma}.
    \end{equation}
  \item  For $r_1\le r\le R_1$,
    \begin{equation}
      \label{linear:eq:redshift:N=T}
      \RedShiftN=\KillT.
    \end{equation}
  \item  For $\Horizon = \EventHorizonFuture,
    \CosmologicalHorizonFuture$, there exists some $c>0$ such that 
    \begin{equation}
      \label{linear:eq:redshift:divergence}  
      \evalAt*{\nabla_{g_b}\cdot\RedShiftN}_{\Horizon} < -c.
    \end{equation}  
  \end{enumerate}  
\end{prop}
\begin{proof}
  See appendix \ref{linear:appendix:prop:redshift:N-construction}.
\end{proof}
\begin{remark}
  Observe that the $\RedShiftN$ constructed in Proposition
  \ref{linear:prop:redshift:N-construction} depends on the vectorfield $X$
  and the constant $\varepsilon_{\RedShiftN}$ chosen. In practice, we
  take $|X|$ sufficiently large and $\varepsilon_{\RedShiftN}$
  sufficiently small so $\RedShiftN$ is fixed throughout the remainder
  of the paper.
\end{remark}

The following technical lemma (see Lemma
3.11 in \cite{warnick_quasinormal_2015} for the anti-de-Sitter equivalent)
constructs the $\RedShiftK_i$ vectorfields which will be used to
define suitable Sobolev spaces in Section \ref{linear:sec:Sobolev-spaces}.

\begin{lemma}
  \label{linear:lemma:enhanced-redshift:Ka-construction}
  There exists a finite collection of vectorfields
  $\curlyBrace*{\RedShiftK_i}_{i=1}^N$
  with the following properties:
  \begin{enumerate}
  \item $\RedShiftK_i$ are stationary, smooth vectorfields on $\StaticRegionWithExtension$.
  \item Near $\EventHorizonFuture$, $\RedShiftK_1$ is future-oriented null with
    $g(\RedShiftK_1,\HorizonGen_{\EventHorizonFuture}) = -1$, and near
    $\CosmologicalHorizonFuture$, $\RedShiftK_1$ is future-oriented null with
    $g(\RedShiftK_1,\HorizonGen_{\CosmologicalHorizonFuture}) = -1$.
  \item $\RedShiftK_i$ are tangent to both $\EventHorizonFuture$ and
    $\CosmologicalHorizonFuture$ for $2\le i\le N$.
  \item If $X$ is any vectorfield supported in
    $\StaticRegionWithExtension$, then there exist smooth functions
    $x^i$, not necessarily unique, such that
    \begin{equation*}
      X =  \sum_i x^i\RedShiftK_i.
    \end{equation*}
  \item We have the following decomposition of the deformation
    tensor of $\RedShiftK_i$,
    \begin{equation}
      \label{linear:eq:enhanced-redshift:Deform-Tens-decomp}
      \DeformationTensor{\RedShiftK_i}=\sum_{j,k}f^{jk}_i \RedShiftK^\flat_j \otimes_s \RedShiftK^\flat_k,
    \end{equation}
    for stationary functions $f^{jk}_i = f^{kj}_i\in
    C^\infty_c(\StaticRegionWithExtension)$, and on
    $\Horizon \in \curlyBrace*{\EventHorizonFuture, \CosmologicalHorizonFuture}$, 
    \begin{align*}
      f^{11}_1 &= \SurfaceGravity_{\Horizon},\\
      f^{11}_i &=0,\quad i\neq 1.
    \end{align*}
  \end{enumerate}
\end{lemma}
\begin{proof}
  See appendix
  \ref{linear:appendix:lemma:enhanced-redshift:Ka-construction}. 
\end{proof}

\subsection{The almost Killing timelike vectorfield $\TFixer$}
\label{linear:sec:setup:TFixer}

There are no globally timelike Killing vectorfields on $a\neq 0$
\KdS{} backgrounds. However, we can define a vectorfield which is
Killing outside of two disconnected components that avoid the horizons
as well as a neighborhood of $r=3M$, which has an $O(a)$ deformation
tensor and is timelike up to the horizons, where it becomes null. 

\begin{lemma}
  \label{linear:lemma:T-Fixer:construction}
  There exists a function $\tilde{\chi}(r) \in
  C^\infty(\StaticRegionWithExtension)$ such that
  \begin{equation}
    \label{linear:eq:TFixer:def}
    \TFixer = \KillT + a\tilde{\chi}(r)\KillPhi, 
  \end{equation}
  satisfies the properties
  \begin{enumerate}
  \item $
    \supp \tilde{\chi}(r) \subset [r_{\EventHorizonFuture} - \varepsilon_{\StaticRegionWithExtension},
    r_0] \bigcup
    [R_0, r_{\CosmologicalHorizonFuture}+ \varepsilon_{\StaticRegionWithExtension}]$;
  \item is timelike on
  $\StaticRegion\backslash\{\EventHorizonFuture\bigcup\CosmologicalHorizonFuture\}$,
  and exactly null on both $\EventHorizonFuture$ and
  $\CosmologicalHorizonFuture$; and
  \item the deformation
  tensor of $\TFixer$ is given by
  \begin{equation}
    \label{linear:eq:TFixer:DeformTensor}
    \DeformationTensor{\TFixer} = a\,d\tilde{\chi}\otimes \KillPhi^\flat,
  \end{equation}
  where we denote by $X^\flat$ the canonical one-form for any
  vectorfield $X$. 
  \end{enumerate}
\end{lemma} 

\begin{proof}
  See appendix \ref{linear:appendix:lemma:T-Fixer:construction}.
\end{proof}

\subsection{Sobolev spaces}
\label{linear:sec:Sobolev-spaces}

In this section, we some useful Sobolev spaces that will feature in
what follows. 
\begin{definition}
  Let $h:\StaticRegionWithExtension\to \Complex^D$. Then denoting by
  $\DomainOfIntegration$ a subset of $\StaticRegionWithExtension$, we
  define the regularity spaces
  \begin{equation*}
    L^2(\DomainOfIntegration) := \curlyBrace*{h: \int_{\DomainOfIntegration}\abs*{h}^2 <\infty}, \qquad
    H^k(\DomainOfIntegration) := \curlyBrace*{h: \RedShiftK^\alpha h\in L^2(\DomainOfIntegration), |\alpha|\le k}. 
  \end{equation*}
\end{definition}

Let us also define two $L^2$ inner products on spacelike
slices. 
\begin{definition}
  We define the $\InducedLTwo$ and $\LTwo$ inner products on the spacelike
  slice $\Sigma_{\tStar}$ by
  \begin{equation*}
    \begin{split}
      \bangle{{h}_1, {h}_2}_{\InducedLTwo(\Sigma_{\tStar})}
      &=  \int_{\Sigma_{\tStar}} {h}_1 \cdot \overline{h}_2\, \sqrt{\GInvdtdt_b},\\
      \bangle{{h}_1, {h}_2}_{\LTwo(\Sigma_{\tStar})}
      &= \int_{\Sigma_{\tStar}}{h}_1\cdot\overline{h}_2,
    \end{split}
  \end{equation*}
  where $\GInvdtdt_b$ is as defined in \eqref{linear:eq:GInvdtdt-def}. 
\end{definition}
\begin{remark}
   Observe that due to \eqref{linear:eq:GInvdtdt-prop}, the two norms are
  equivalent to each other. Furthermore, despite the dependence on the
  choice of $g_b$ reflected in the presence of $\GInvdtdt_b$ in the
  definition of $\InducedLTwo(\Sigma)$, the $\InducedLTwo(\Sigma)$
  norm defined for differing slowly-rotating \KdS{} metrics $g_b$ are
  all equivalent to each other.
\end{remark}

We likewise have the following higher-regularity Sobolev norms.
\begin{definition}
  \label{linear:def:Hk-def}
  For any $u:\Sigma\to\Complex^D$, let
  $\upsilon:\StaticRegionWithExtension\to \Complex^D$ be the unique
  lifting satisfying
  \begin{equation}
    \label{linear:eq:stationary-extension-def}
    \upsilon\vert_{\Sigma} = u,\quad \KillT \upsilon = 0.
  \end{equation}  
  Then define the regularity spaces $H^k(\Sigma), \InducedHk{k}(\Sigma)$ by:
  \begin{equation*}
    \begin{split}
      H^k(\Sigma)&:= \curlyBrace*{u:\left.\RedShiftK^\alpha
          \upsilon\right\vert_{\Sigma} \in \LTwo(\Sigma), |\alpha|\le
        k}, \\
      \InducedHk{k}(\Sigma)&:=\curlyBrace*{u:\left.\RedShiftK^\alpha
          \upsilon\right\vert_{\Sigma} \in \InducedLTwo(\Sigma), |\alpha|\le
        k}, 
    \end{split}
  \end{equation*}
  where $\alpha$ is a multi-index and $\RedShiftK_i$ are vectorfields satisfying
  the requirements of Lemma \ref{linear:lemma:enhanced-redshift:Ka-construction}.

  With the same $\RedShiftK_i$, we abuse notation to define the norm $\HkWithT{k}(\Sigma_{\tStar})$ by:
  \begin{equation*}
    \norm*{h}_{\HkWithT{k}(\Sigma_{\tStar})}^2 := \sum_{\abs*{\alpha}\le k} \norm*{\RedShiftK^\alpha h}_{\InducedLTwo(\Sigma_{\tStar})}^2, 
  \end{equation*}
  where $h:\StaticRegionWithExtension\to\Complex^D$. 
\end{definition}
\begin{remark}  
  Observe that different choices of the
  family $\curlyBrace{\RedShiftK_i}_{i=1}^N$ will result in different, though
  equivalent, $H^k(\Sigma), \InducedHk{k}(\Sigma),
  \HkWithT{k}(\Sigma)$ norms. 
\end{remark}
\begin{remark}
  \label{linear:remark:Hk-equivalence}
  At first glance, it may appear that since the construction
  $\{\RedShiftK_i[g_b]\}$ family in Lemma
  \ref{linear:lemma:enhanced-redshift:Ka-construction} depends on the chosen
  \KdS{} metric, the $H^k[g_b](\DomainOfIntegration)$ norms defined in
  Definition \ref{linear:def:Hk-def} are also dependent on the chosen \KdS{}
  metric. However, recall that by its construction in Proposition
  \ref{linear:prop:redshift:N-construction}, $\RedShiftN$ is uniformly
  timelike on $\StaticRegionWithExtension$ and transverse to both the
  event horizon and cosmological horizon for all slowly-rotating
  \KdS{} black hole backgrounds. Combined with Lemma
  \ref{linear:lemma:enhanced-redshift:Ka-construction}, we have that in fact
  for any slowly-rotating \KdS{} metric $g_b$, the
  $H^k[g_b](\DomainOfIntegration)$ norm is equivalent to
  \begin{equation*}
    H^k_{\star}(\DomainOfIntegration):=\curlyBrace*{u: Z^\alpha u\in L^2(\DomainOfIntegration), |\alpha|\le k, Z\in \{\RedShiftN, \curlyBrace*{\RedShiftK_i[b_0]}_{i=2}^{N}\}}.
  \end{equation*}
  A similar equivalence holds for the $H^k(\Sigma)$,
  $\InducedHk{k}(\Sigma)$, and $\HkWithT{k}(\Sigma)$ norms. 
\end{remark}

We also define the following Laplace-transformed Sobolev norms:
\begin{definition}
  \label{linear:def:Hk-sigma}
  Let $u:\Sigma\to\Complex^D$. Then, we define the
  \emph{Laplace-transformed Sobolev norms} by:
  \begin{align*}
    \norm{u}_{\InducedHk{1}_\sigma(\Sigma)}^2 &= \norm{u}_{\InducedHk{1}(\Sigma)}^2 + \norm{\sigma u}_{\InducedLTwo(\Sigma)}^2,\qquad
    \norm{u}_{\InducedHk{k}_\sigma(\Sigma)}^2 = \norm{u}_{\InducedHk{k}(\Sigma)}^2 + \norm{\sigma u}_{\InducedHk{k-1}_{\sigma}(\Sigma)}^2,
  \end{align*}
  so that
  $\InducedHk{k}_\sigma(\Sigma) = \InducedHk{k}(\Sigma) \bigcap
  \sigma^{-1}\InducedHk{k-1}_\sigma(\Sigma)$.
\end{definition}
\begin{remark}
  We define the Laplace-transformed vectorfields
  $\widehat{\RedShiftK}(\sigma)u=
  \left. e^{\ImagUnit\sigma\tStar}\RedShiftK(e^{-\ImagUnit\sigma\tStar}\upsilon)
  \right\vert_{\Sigma_{\tStar}}$, where
  $\upsilon$ is the stationary extension of $u$ defined in
  \eqref{linear:eq:stationary-extension-def}.  We can also characterize
  \begin{equation*}
    \norm{u}_{\InducedHk{k}_\sigma(\Sigma)}^2 = \sum_{|\alpha|\le k}\norm*{\widehat{\RedShiftK}^\alpha(\sigma)u}_{\InducedLTwo(\Sigma)}^2.
  \end{equation*}
\end{remark}

We use the vectorfields defined in Theorem
\ref{linear:lemma:enhanced-redshift:Ka-construction} to define the following
higher-regularity Sobolev spaces.
\begin{definition}
  Suppose $\LinearOp$ is any strongly hyperbolic operator on a \KdS{}
  black hole. Given some
  $(\psi_0,\psi_1)\in H^k_{\local}(\Sigma, \Complex^D)\times
  H^{k-1}_\local(\Sigma, \Complex^D)$, letting
  $\curlyBrace{\RedShiftK_i}_{i=1}^N$ be as constructed in Lemma
  \ref{linear:lemma:enhanced-redshift:Ka-construction}. We then define
  $\LSolHk{k}(\Sigma)$ to be the space consisting of $(\psi_0,\psi_1)$ such
  that
  \begin{equation*}
    \norm{(\psi_0, \psi_1)}^2_{\LSolHk{k}(\Sigma)} :=
    \norm*{\psi_0}^2_{\Hk{k}(\Sigma)} + \norm{\psi_1}^2_{\Hk{k-1}(\Sigma)}.
  \end{equation*}
\end{definition}

Finally, we also define the following weighted Sobolev spaces.
\begin{definition}
  For $\alpha \in \Real$, $h:
  \StaticRegionWithExtension\to\Complex^D$, define the \emph{weighted Sobolev norm} 
  \begin{align*}
    \norm{h}_{H^{k,\alpha}(\StaticRegionWithExtension)}^2:=
    \int_0^\infty e^{2\alpha\tStar} \norm*{h}_{\HkWithT{k}(\Sigma_{\tStar})}^2\,d\tStar.
  \end{align*}
\end{definition}

\subsection{Strongly hyperbolic operators}

\begin{definition}
  \label{linear:def:strongly-hyperbolic-operator}
  Given a linear second-order differential operator $\LinearOp$,
  acting on complex matrix functions
  \begin{equation*}
    h: \StaticRegionWithExtension\to\Complex^D,
  \end{equation*}
  we call $\LinearOp$ a \emph{strongly hyperbolic operator} on a
  background Lorentzian metric $g$ if $\LinearOp$ can be expressed as
  \begin{equation}
    \label{linear:eq:strongly-hyperbolic-operator}
    \LinearOp{h} = \ScalarWaveOp[g]{h} + \SubPOp[{h}] + \PotentialOp{h},
  \end{equation}
  where $\ScalarWaveOp[g]$ denotes the scalar geometric wave operator
  $\nabla^\alpha\p_\alpha$,
  \begin{equation}
    \label{linear:eq:strongly-hyperbolic-operator-s-def}
    \SubPOp = S^\alpha\p_\alpha
  \end{equation}
  is a smooth vectorfield-valued matrix, and $\PotentialOp$ is a
  smooth matrix potential. We will often refer to $\SubPOp$ as the
  \emph{subprincipal operator} of $\LinearOp$, and
  $\PotentialOp$ as the \emph{potential operator} of $\LinearOp$. 

  As we will show in Lemma \ref{linear:lemma:EVE:GHC-quasilinear}, the
  gauged linearized Einstein operator is an example of a strongly
  hyperbolic operator.
\end{definition}

Given some strongly hyperbolic $\LinearOp$ on background metric $g$,
we define the following quantities,
\begin{equation}
  \label{linear:eq:SHorizonControl:def}
  \begin{split}
    \SHorizonControl{\LinearOp}[\EventHorizonFuture]
    &= \sup_{\EventHorizonFuture}\curlyBrace*{
      -g\left(\HorizonGen_{\EventHorizonFuture}, \Re[\overline{\xi}\cdot \SubPOp\xi]\right) : \xi\in \Complex^N,\abs{\xi}=1},\\
    \SHorizonControl{\LinearOp}[\CosmologicalHorizonFuture]
    &= \sup_{\CosmologicalHorizonFuture}\curlyBrace*{
      -g\left(\HorizonGen_{\CosmologicalHorizonFuture}, \Re[\overline{\xi}\cdot \SubPOp\xi]\right) : \xi\in \Complex^N,\abs{\xi}=1},\\
    \SHorizonControl{\LinearOp}^*[\EventHorizonFuture]
    &= \inf_{\EventHorizonFuture}\curlyBrace*{
      -g\left(\HorizonGen_{\EventHorizonFuture}, \Re[\overline{\xi}\cdot \SubPOp\xi]\right) : \xi\in \Complex^N,\abs{\xi}=1},\\
    \SHorizonControl{\LinearOp}^*[\CosmologicalHorizonFuture]
    &= \inf_{\CosmologicalHorizonFuture}\curlyBrace*{
      -g\left(\HorizonGen_{\CosmologicalHorizonFuture}, \Re[\overline{\xi}\cdot \SubPOp\xi]\right) : \xi\in \Complex^N,\abs{\xi}=1},
  \end{split}
\end{equation}
where we use the shorthand notation
\begin{equation*}
  g\left(\HorizonGen, \Re[\overline{\xi}\cdot \SubPOp\xi]\right)
  := g_{\alpha\beta}\HorizonGen^\alpha\Re[\overline{\xi}\cdot S^\beta\xi].
\end{equation*}
The quantities in \eqref{linear:eq:SHorizonControl:def} will play a critical
role in our formulation of quasinormal modes (see Section
\ref{linear:sec:QNM}).

\begin{definition}
  Given a strongly hyperbolic operator $\LinearOp$ of the form in equation
  \eqref{linear:eq:strongly-hyperbolic-operator}, we can
  define its adjoint\footnote{For a matrix
    $M_{ab}$, we use the convention that its adjoint is
    \begin{equation*}
      M^*_{ab} = M_{ba}.
    \end{equation*}
  }  $\LinearOp^\dagger$ as
  \begin{equation}
    \label{linear:eq:L-dagger-def}
    \LinearOp^\dagger {h} = \ScalarWaveOp[g]{h} - \SubPOp^\dagger[{h}] + \PotentialOp^*{h} - (\nabla_{g}\cdot\SubPOp^\dagger){h},
  \end{equation}
  where $\SubPOp^\dagger$ is a vectorfield-valued matrix such that
  \begin{equation*}
    \SubPOp^{\dagger} = (S^\mu)^*\p_\mu,
  \end{equation*}
  with $S^\mu$ as defined in
  \eqref{linear:eq:strongly-hyperbolic-operator-s-def}, and where we 
  understand the adjoint of a matrix to be its transpose.
\end{definition}
It is immediately clear that $\LinearOp^\dagger$ is itself a strongly
hyperbolic operator on $\StaticRegionWithExtension$. Moreover, we have
the following adjoint relation.
\begin{corollary}
  \label{linear:corollary:L-dagger-adjoint}
  Then for ${h}_1, {h}_2\in C^2(\StaticRegionWithExtension,\Complex^D)$:
  \begin{equation*}
    \begin{split}
      \bangle{{h}_1,\LinearOp {h}_2}_{\InducedLTwo(\Sigma_{\tStar})} -
      \bangle{\LinearOp^\dagger {h}_1, {h}_2}_{\InducedLTwo(\Sigma_{\tStar})}
      ={}&
      \p_{\tStar}\int_{\Sigma_{\tStar}} K\cdot n_\Sigma
      + \int_{\EventHorizonFuture\bigcap\Sigma_{\tStar}}K\cdot n_{\EventHorizonFuture} 
      + \int_{\CosmologicalHorizonFuture\bigcap\Sigma_{\tStar}}K\cdot n_{\CosmologicalHorizonFuture} ,
    \end{split}
  \end{equation*}
  where
  \begin{equation*}
    K_\mu = \overline{{h}}_1\cdot  \nabla_\mu {h}_2
    - \nabla_\mu\overline {h}_1\cdot {h}_2
    - \overline{{h}}_1 \cdot S_\mu {h}_2. 
  \end{equation*}
\end{corollary}
\begin{proof}
  The proof follows directly from applying Corollary
  \ref{linear:cor:div-them:spacelike} with the vectorfield $K_\mu$. 
\end{proof}

\section{Einstein's equations}
\label{linear:sec:EVE}

In this section, we introduce the generalized harmonic coordinates and
the hyperbolic initial value problem formulation of Einstein's
equations as a system of evolution equations.  Then we detail some key
properties of the linearized Einstein operator that will be crucial in
the ensuing stability analysis.

\subsection{Harmonic gauge}
\label{linear:sec:GHC-Cauchy-problem}

Recall that Einstein's vacuum equations with a cosmological
constant $\Lambda$ for $g$, a $(-, +, +, +)$ Lorentzian metric  on a
smooth manifold $\mathcal{M}$, are
\begin{equation}
  \label{linear:eq:EVE:Full}
  \Ric(g) - \Lambda g= 0.
\end{equation}
For any globally hyperbolic $(\mathcal{M}, g)$ solution to
(\ref{linear:eq:EVE:Full}), and spacelike hypersurface
$\Sigma_0\subset\mathcal{M}$, the induced Riemannian metric $\InducedMetric$
on $\Sigma_0$ and the second fundamental form $k(X,Y)$ of $\Sigma_0$
satisfy the \emph{constraint equations}
\begin{equation}
  \label{linear:eq:EVE:constraint-eqns}
  \begin{split}
    R(\InducedMetric) + (\Trace_{\InducedMetric}k)^2 - \abs{k}_{\InducedMetric}^2 &= - 2\Lambda,\\
    \nabla_{\InducedMetric} \cdot k - d \Trace_{\InducedMetric} k &= 0,
  \end{split}
\end{equation}
where $R(\InducedMetric)$ is the scalar curvature of $\InducedMetric$,
and $\nabla_{\InducedMetric}\cdot k$ denotes the divergence of $k$ with
respect to the covariant derivative of $\InducedMetric$.  The Cauchy
problem for Einstein's equations then asks, given an initial data set
consisting of the triple $(\Sigma_0, \InducedMetric, k)$, where
$\InducedMetric$ is a Riemannian metric on the smooth $3$-manifold
$\Sigma_0$ and $k$ is a symmetric 2-tensor on $\Sigma_0$ such that
$(\InducedMetric, k)$ satisfy \eqref{linear:eq:EVE:constraint-eqns}, for a
Lorentzian 4-manifold $(\mathcal{M}, g)$ and an embedding
$\Sigma_0\hookrightarrow \mathcal{M}$ such that $\InducedMetric$ is
the induced metric on $\Sigma_0$, and $k$ is the second fundamental
form of $\Sigma_0$ in $\mathcal{M}$. We denote initial data triplets
with $(\InducedMetric , k)$ satisfying the constraint equations
\eqref{linear:eq:EVE:constraint-eqns} to be \textit{admissible} initial data
triplets.

It is well-known that (\ref{linear:eq:EVE:Full}) is a quasilinear
second-order partial differential system of equations for the metric
coefficients $g_{\mu\nu}$. In local coordinates, we can write
\eqref{linear:eq:EVE:Full} as
\begin{equation*}
  \Ric(g)_{\mu\nu}
  = - \frac{1}{2}g^{\alpha\beta}\p_{\alpha}\p_{\beta}g_{\mu\nu}
  + \nabla_{(\mu}\Gamma(g)_{\nu)} + \mathcal{N}(g, \p g),\qquad
  \Gamma(g)^\mu:= g^{\alpha\beta}\ChristoffelTypeTwo[g]{\mu}{\alpha\beta},
\end{equation*}
where the nonlinear term $\mathcal{N}(g, \p g)$ involves at most one
derivative of $g$. As a result of the presence of the
$\nabla_{(\mu}\Gamma(g)_{\nu)}$ term EVE lacks any useful
structure. However, as was first demonstrated by Choquet-Bruhat, this
problem can be overcome by using the general covariance of Einstein's
equations and choosing \textit{wave coordinates}. With this choice
Einstein's equations become a quasilinear hyperbolic system of
equations
\cite{choquet-bruhat_theoreme_1952,choquet-bruhat_global_1969}.

\begin{definition}
  Define the \emph{constraint operator}
  \begin{equation*}
    \Constraint(g, g^0)_\mu =
    g_{\mu\chi}g^{\nu\lambda}\left(\ChristoffelTypeTwo[g]{\chi}{\nu\lambda}
      - \ChristoffelTypeTwo[g^0]{\chi}{\nu\lambda} \right),
  \end{equation*}
  where $g^0$ is a fixed background metric which solves Einstein's
  equations\footnote{In our case, we 
  typically take $g^0$ to be the background metric around which we 
  linearize the equations.}.
  Then we say that a Lorentzian metric $g$ satisfies the
  \emph{harmonic coordinate condition} (with respect to $g^0$) if
  \begin{equation}
    \label{linear:eq:GHC-condition:Christoffel}
    g_{\mu\chi}g^{\nu\lambda}\left(\ChristoffelTypeTwo[g]{\chi}{\nu\lambda} - \ChristoffelTypeTwo[g^0]{\chi}{\nu\lambda} \right) = 0.
  \end{equation}
\end{definition}
\begin{remark}
  Instead of $0$ on the right-hand side, if we instead pick some
  one-form $\bH(g)$ depending on $g$ but not its derivatives, we would
  obtain \emph{generalized harmonic coordinates}.
\end{remark}

Crucial to the utility of harmonic coordinates is that
they are propagated by a hyperbolic operator.
\begin{definition}
  Given a smooth one-form $\psi$, we define the
  \emph{constraint propagation operator},
  \begin{equation}
    \label{linear:eq:constraint-propagation-op:def}
    \ConstraintPropagationOp_g\psi :=2\nabla_g\cdot  \TraceReversal_g\nabla_g\otimes\psi
    = - \VectorWaveOp[g]\psi + \Ric(g)(\psi, \cdot),
  \end{equation}
  where 
  \begin{equation}
    \label{linear:eq:constraint-propagation-aux}
    \begin{split}
      \nabla_g\cdot: S^2 T^*M \to T^*M,&\qquad \nabla_g\cdot U_{\mu\nu} := \nabla^\mu U_{\mu\nu}\\
      \nabla_g\otimes: T^*M \to S^2 T^*M,&\qquad \nabla_g\otimes \omega_\mu := -\frac{1}{2}\left(\nabla_\mu \omega_\nu +\nabla_\nu\omega_\mu\right),\\
      \TraceReversal_g: S^2T^*M\to S^2T^*M,&\qquad \TraceReversal_gU_{\mu\nu} := U_{\mu\nu} - \frac{1}{2}\Trace_gU g_{\mu\nu}.
    \end{split}        
  \end{equation}
  From \eqref{linear:eq:constraint-propagation-op:def},
  $\ConstraintPropagationOp$ is a manifestly hyperbolic operator on
  $\psi$.
\end{definition}

\begin{lemma}
  \label{linear:lemma:EVE:nonlinear-constraint-prop}
  Any solution $g$ to
  \begin{equation}
    \label{linear:eq:EVE:gauged-eq}
    \Ric(g) - \Lambda g + \nabla_g\otimes\Constraint(g,g^0) =0 
  \end{equation}
  must also satisfy
  \begin{align*}
    \ConstraintPropagationOp_g\Constraint(g,g^0)= 0.
  \end{align*}
\end{lemma}
\begin{proof}
  The conclusion follows directly by applying the twice-contracted
  second Bianchi identity to \eqref{linear:eq:EVE:gauged-eq}.
\end{proof}

\begin{prop}
  \label{linear:prop:EVE:GHC-EVE-equivalence}
  If $g\in S^2T^*\StaticRegionWithExtension$ satisfies the gauged
  Einstein equation
  \begin{equation}
    \label{linear:eq:EVE:gauged}
    \begin{split}
      \Ric(g) - \Lambda g + \nabla_g\otimes\Constraint(g,g^0) ={}&0
    \end{split}
  \end{equation}
  and moreover, $g$ satisfies the gauge constraint on $\Sigma_0$,
  \begin{equation*}
    \Constraint(g, g^0)\vert_{\Sigma_0} = 0,
  \end{equation*}
  then $g$ is a solution to the ungauged Einstein vacuum equations (\ref{linear:eq:EVE:Full}). 
\end{prop}

\begin{proof}
  This follows directly from Lemma
  \ref{linear:lemma:EVE:nonlinear-constraint-prop} and uniqueness of
  solutions for hyperbolic PDEs.
\end{proof}

\subsection{The linearized Einstein equations}
\label{linear:sec:linearized-EVE}

We now introduce the linearized Einstein equation, for a more in-depth
introduction, we refer the reader to Section 3 of
\cite{graham_einstein_1991}. Directly linearizing \eqref{linear:eq:EVE:Full}
around $g_b$ yields the ungauged linearized Einstein equation
\begin{equation}
  \label{linear:eq:linearized-EVE-ungauged}
  D_{g_b}(\Ric - \Lambda)({h}) = 0. 
\end{equation}
Given admissible initial data $(\Sigma_0, \InducedMetric_0, k_0)$ for $g_b$,
we define the \emph{linearized constraint equation} as the
linearization of \eqref{linear:eq:EVE:constraint-eqns} around
$(\InducedMetric_0, k_0)$ in terms of the linearized metric
$\InducedMetric'$ and the linearized second fundamental form
$k'$. 
An initial data triplet $(\Sigma_0, \InducedMetric', k')$ linearized
around $(\InducedMetric_b, k_b)$ is an \textit{admissible} initial
data triplet for Einstein equations linearized around $g_b$ if
$(\InducedMetric', k')$ satisfy the linearized constraint equations.
Linearizing the gauged Einstein equations in \eqref{linear:eq:EVE:gauged-eq},
we have the linearized gauged Einstein equations
\begin{equation}
  \label{linear:eq:linearized-gauged-EVE}
  D_{g_b}(\Ric - \Lambda)({h}) - \nabla_{g_b}\otimes D_{g_b}\Constraint(g, g^0)(h) = 0. 
\end{equation}

\begin{definition}
  Define the \emph{linearized gauge constraint} \begin{equation*}
    \begin{split}
      \Constraint_{g_b}{h}:={}&D_{g_b}\Constraint(g_b+{h}, g_b)({h})\\
      ={}& -\nabla_{g_b}\cdot\TraceReversal_{g_b}{h}.
    \end{split}    
  \end{equation*}
\end{definition}
We have the following linearized equivalent of Lemma
\ref{linear:lemma:EVE:nonlinear-constraint-prop}.
\begin{lemma}
  \label{linear:lemma:EVE:linearized-constraint-prop}
  Let ${h}$ solve \eqref{linear:eq:linearized-gauged-EVE}. Then ${h}$
  also satisfies
  \begin{equation*}
    \ConstraintPropagationOp_{g_b}(\Constraint_{g_b}{h}) = 0, \qquad
    \ConstraintPropagationOp_{g_b}\psi = \VectorWaveOp[g_b]\psi - \Lambda\psi, 
  \end{equation*}
  where $\VectorWaveOp[g_b] = \nabla^\alpha \nabla_\alpha $ denotes
  the wave operator acting on 1-tensors. 
\end{lemma}
\begin{proof}
  The lemma follows directly by applying the twice-contracted
  linearized second Bianchi identity to the gauged linearized Einstein
  equation.
\end{proof}
\begin{remark}
  From Lemma \ref{linear:lemma:EVE:linearized-constraint-prop}, it is clear
  that if $(\evalAt*{\Constraint_{g_b}{h}}_{\Sigma_0}, \evalAt*{\LieDerivative_{\KillT}\Constraint_{g_b}{h}}_{\Sigma_0}) = (0, 0)$, then
  $\Constraint_{g_b}{h} = 0$ uniformly.  
\end{remark}

Finally, we remark that any solution ${h}$ to the ungauged linearized
Einstein's equations \eqref{linear:eq:linearized-EVE-ungauged} can be put
into the linearized gauge $C_{g_b}({h}) = 0$ by finding some
infinitesimal diffeomorphism
$\nabla_{g_b}\otimes \omega$\footnote{Observe that in terms of the Lie
  derivative, we have that
  \begin{equation*}
    \nabla_{g_b}\otimes \omega = \frac{1}{2}\LieDerivative_{\omega^\sharp}g_b.
  \end{equation*}
}
such that
\begin{equation}
  \Constraint_{g_b}({h} + \nabla_{g_b}\otimes \omega) = 0,
\end{equation}
as general covariance implies that
\begin{equation*}
  D_{g_b}(\Ric - \Lambda)(\nabla_{g_b}\otimes \omega) = 0
\end{equation*}
for any one-form $\omega\in C^\infty(\StaticRegionWithExtension,
T^*\StaticRegionWithExtension)$. This is equivalent to finding some
$\omega$ such that
\begin{equation}
  \label{linear:eq:linearized-gauge-fixer}
  \Box_{g_b}^\Upsilon\omega = 2\Constraint_{g_b}({h}),\qquad
  \Box_{g_b}^{\Upsilon} = -2\Constraint_{g_b}\circ \nabla_{g_b}\otimes ,
\end{equation}
which is principally $\VectorWaveOp[g_b]$, and in fact, in our case we
can calculate that
\begin{equation*}
  \Box_{g_b}^\Upsilon = \VectorWaveOp[g_b] - \Lambda. 
\end{equation*}
Solving for $\omega$ satisfying \eqref{linear:eq:linearized-gauge-fixer} with
Cauchy data
$\left.(\omega, \LieDerivative_{\KillT}\omega)\right\vert_{\Sigma_0}=
0$ then ensures that ${h} +\nabla_{g_b}\otimes \omega$ has the same initial
data as ${h}$.

\subsection{Properties of $\LinEinstein_{g_b}$}
\label{linear:sec:LinEinstein-properties}

Using wave coordinates, we can compute the exact quasilinear structure
of Einstein's equations. 
\begin{lemma}
  \label{linear:lemma:EVE:GHC-quasilinear}
  Let $g_b + h$ be a solution to the gauged Einstein vacuum equations
  in \eqref{linear:eq:EVE:gauged} 
  where we choose $g^0=g_b$. Then, $h$ solves
  \begin{equation*}
    \LinEinstein_{g_b}h = \NCal(h,\p h, \p\p h),
  \end{equation*}
  where
  \begin{align}
    \LinEinstein_{g_b}h &:= \frac{1}{2}\TensorWaveOp[g_b]h
                          + \mathcal{R}_g(h),\label{linear:eq:LinEinstein-def}\\
    \mathcal{R}_{g_b}({h})_{\mu\nu}&:= {h}^{\alpha\lambda}(R_{g_b})_{\alpha\mu\nu\lambda}
                                     - \Lambda {h}_{\mu\nu}\label{linear:eq:LinEinstein-R-def}, 
  \end{align}
  where $\TensorWaveOp[g]:=\nabla^\alpha\nabla_\alpha$ denotes the
  wave operator acting on $2$-tensors, $R_{g_b}$ is the Riemann
  curvature tensor of $g_b$, and $\NCal$ is a quasilinear nonlinear
  term in ${h}$.  In particular, $\LinEinstein_{g_b}$ is a strongly
  hyperbolic operator, as defined in Definition
  \ref{linear:def:strongly-hyperbolic-operator}.
\end{lemma}
\begin{proof}
  The conclusion follows from equation (2.4) in
  \cite{graham_einstein_1991} and the harmonic gauge condition in
  \eqref{linear:eq:GHC-condition:Christoffel}.  It should be noted that when
  writing $\LinEinstein_{g_b}h$ as a strongly hyperbolic operator,
  \begin{equation*}
    \LinEinstein_{g_b} = \ScalarWaveOp[g_b] + \SubPOp_{b} + \PotentialOp_b,
  \end{equation*}
  the exact coefficients in $\SubPOp_b$ and $\PotentialOp_b$ will
  depend not only on the coordinate system chosen, but also on the
  particular frame used to split
  $S^2T^*\StaticRegionWithExtension$. To see that
  $\LinEinstein_{g_b}$ can be written as a global system of strongly
  hyperbolic equations, it suffices to take some global frame on
  $\StaticRegionWithExtension$ (for example, the Cartesian frame). 
\end{proof}
\begin{definition}
  We refer to $\LinEinstein = \LinEinstein_g$ defined in
  \eqref{linear:eq:LinEinstein-def} as the \emph{gauged linearized Einstein
    operator}.
\end{definition}

We can also compute the value of
$\SHorizonControl{\LinEinstein_{g_{b_0}}}[\Horizon]$.
\begin{lemma}
  \label{linear:lemma:SubPOp:horizons}
  Let $g_{b_0}$ be a fixed member of the \SdS{} family. Then,
  \begin{equation*}
    \SHorizonControl{\LinEinstein_{g_{b_0}}}[\Horizon] = 4\SurfaceGravity_{b_0,\Horizon},\qquad \Horizon=\EventHorizonFuture,\CosmologicalHorizonFuture,
  \end{equation*}
  where we recall the definition of
  $\SHorizonControl{\LinearOp}[\Horizon]$ from \eqref{linear:eq:SHorizonControl:def}.
\end{lemma}
\begin{proof}
  See appendix \ref{linear:appendix:lemma:SubPOp:horizons}.
\end{proof}

\subsection{Initial data}

In this section, we will construct the mapping $i_{b, \phi}$ between
admissible initial data triplets $(\Sigma_0, \InducedMetric_0, k_0)$
for the Cauchy problem for the ungauged Einstein equations
\eqref{linear:eq:EVE:Full} and the admissible initial data $(h_0, h_1)$ for
the Cauchy problem for the gauged Einstein equations. Recall
from Section \ref{linear:sec:GHC-Cauchy-problem} that an important property
of this mapping is that a metric perturbation $h$ such that
$\gamma_0(h)=(h_0, h_1)$ satisfies the gauge constraint
$\Constraint(g_b+h, g^0)\vert_{\Sigma_0} = 0$. To construct
$i_{b, \phi}$ we need to specify a choice of $g^0$. In the remainder
of this paper, it will be convenient to choose $g^0 = g_b$.

Consider the \KdS{} initial data triplet
$(\Sigma_0, \InducedMetric_b, k_b)$ that launches $g_b$, so that in
particular, $\InducedMetric_b$ and $k_b$ denote the induced metric and
second fundamental form on $\Sigma_0$ by $g_b$.  We will construct
$i_{b, \phi}$ mapping $(\Sigma_0, \InducedMetric_b, k_b)$ into Cauchy
data for \eqref{linear:eq:EVE:gauged} launching the \KdS{} solution
$\phi^*g_b$. That is, we will have that
\begin{equation*}
  i_{b, \phi}(\phi^*\InducedMetric_b, \phi^*k_b) = (0,0).
\end{equation*}
The linearization of this mapping will also produce the correctly
gauged initial data for the gauged linearized Einstein equation
linearized around $g_b$.

\begin{prop}
  \label{linear:prop:initial-data:ib-construction}
  Fix some one-form $\vartheta$, and denote by $\phi$ the
  diffeomorphism generated by $\vartheta^\sharp$. Then 
  there exist neighborhoods of symmetric two-tensors on $\Sigma_0$
  \begin{equation*}
    H\subset C^1(\Sigma_0; S^2T^*\Sigma_0),\qquad
    K\subset C^0(\Sigma_0; S^2T^*\Sigma_0)
  \end{equation*}
  of $\InducedMetric_{b_0}$ and $k_{b_0}$, respectively, so that $\InducedMetric_b\in H$,
  $k_b\in K$ for all $b\in \BHParamNbhd$, where $\BHParamNbhd$ is a
  sufficiently small neighborhood of black-hole parameters of $b_0$;
  and moreover, for each $b\in \BHParamNbhd$, there exists a map
  \begin{equation*}
    \begin{split}
      i_{b, \phi}:& \left(H\bigcap C^m(\Sigma_0; S^2T^*\Sigma_0)\right)
      \times \left(K\bigcap C^{m-1}(\Sigma_0;S^2T^*\Sigma_0)\right)\\
      &\to C^m(\Sigma_0; S^2T^*_{\Sigma_0}\StaticRegionWithExtension)\times C^{m-1}(\Sigma_0;S^2T^*_{\Sigma_0}\StaticRegionWithExtension),
    \end{split}
  \end{equation*}
  that is smooth for $m\ge 1$ depending smoothly on $b$, such that 
  \begin{enumerate}
  \item if $h$ is some symmetric two-tensor such that $\gamma_0(h) =
    i_{b, \phi}(\InducedMetric_0, k_0)$, and $g = \phi^*(g_b+h)$, then
    \begin{equation*}
      (\InducedMetric,  k) = (\InducedMetric_0, k_0),
    \end{equation*}
    where $(\InducedMetric, k)$ are the induced metric and the second
    fundamental form respectively of $\phi(\Sigma_0)$ induced by
    $g$. Moreover, $g_b+h$ satisfies the gauge constraint
    \begin{equation*}
      \Constraint(g_b+h,g_b)\vert_{\Sigma_0} = 0;
    \end{equation*}
  \item if $(\InducedMetric_b, k_b)$ is the admissible initial data launching the \KdS{}
    metric $g_b$, then
    \begin{equation*}
      i_{b, \phi}(\phi^*\InducedMetric_b, \phi^*k_b)  = (0, 0);
    \end{equation*}
  \item $(g_0, g_1) = i_{b, \phi}(\InducedMetric, k)$ satisfies the
    condition
    \begin{align*}
      \norm{(g_0, g_1)}_{\LSolHk{k}(\Sigma_0)}
      \lesssim{}&  \sum_{0\le |I|\le k}\norm*{\p_{x}^I(\InducedMetric - \InducedMetric_b)}_{L^2(\Sigma_0)} + \sum_{0\le |I|\le k+1}\norm*{\p_x^I\vartheta}_{L^2(\Sigma_0)}\\
      &+ \sum_{0\le |I|\le k-1}\norm*{\p_x^I (k - k_b)}_{L^2(\Sigma_0)},
    \end{align*}
    where $I$ is a multi-index. 
  \end{enumerate}
\end{prop}
\begin{proof}
  See appendix \ref{linear:appendix:prop:initial-data:ib-construction}. 
\end{proof}

The linearization of $i_{b, \phi}$ constructed above yields the correctly
gauged Cauchy data for the linearized gauged Einstein equation, just
as $i_{b, \phi}$ itself yields the correctly gauged Cauchy data for the
nonlinear gauged Einstein equation.
\begin{corollary}
  \label{linear:corollary:initial-data:ib-linearized}
  Fix $b\in \BHParamNbhd$ and a one-form $\vartheta$ generating the
  diffeomorphism $\phi$. Suppose
  $(\Sigma_0, \InducedMetric_b, k_b)$ is the smooth admissible initial
  data triplet launching $g_b$. Then, let $(\InducedMetric', k')$ be
  smooth solutions of the linearized constraint equations linearized
  around $(\InducedMetric_b, k_b)$, and let
  \begin{equation*}
    D_{(\InducedMetric_b, k_b)}i_{b, \phi}(\InducedMetric', k') = ({h}_0,
  {h}_1).
  \end{equation*}
  Finally, let ${h}\in S^2T^*\StaticRegionWithExtension$
  be a metric perturbation inducing $({h}_0, {h}_1)$ on $\Sigma_0$ so that
  $\gamma_0({h}) = ({h}_0, {h}_1)$, where
  \begin{equation}
    \label{linear:eq:gamma-0:def}
    \gamma_0(h):= (\evalAt*{h}_{\Sigma_0}, \evalAt*{\LieDerivative_{\KillT} h}_{\Sigma_0}).
  \end{equation}

  Then ${h}$ induces the linearized metric $\InducedMetric'$ and
  second fundamental form $k'$ on $\Sigma_0$, and satisfies the
  linearized gauge constraint on $\Sigma_0$,
  \begin{equation*}
    D_{g_b}\Constraint(g_b+{h}, g_b)({h})\vert_{\Sigma_0} = 0.
  \end{equation*}
  Moreover, if
  $({h}_0, {h}_1) = D_{(\InducedMetric_b, k_b)}i_{b,\phi}(\InducedMetric', k')$, then
  \begin{equation*}
    \norm*{({h}_0, {h}_1)}_{\LSolHk{k}(\Sigma_0)}
    \lesssim{} \sum_{1\le|I|\le k}\norm*{\p_x^I \InducedMetric'}_{L^2(\Sigma_0)}
    + \sum_{1\le|I|\le k+1}\norm*{\p_x^I\vartheta}_{L^2(\Sigma_0)}
    + \sum_{1\le |J| \le k-1}\norm*{\p_x^J k'}_{L^2(\Sigma_0)}.
  \end{equation*}
\end{corollary}
\begin{proof}
  The statement follows directly by linearizing the construction of
  $i_{b,\phi}$ in appendix \ref{linear:appendix:prop:initial-data:ib-construction}.
\end{proof}

\section{Quasinormal Spectrum}
\label{linear:sec:QNM}

In this section, we define the quasinormal spectrum and establish the
basic theory necessary to use the quasinormal spectrum to analyze the
behavior of solutions to initial value problems. The definition of the
$\LSolHk{k}$-quasinormal modes and their relation to the
Laplace-transformed operator follows closely to the original work done
in \cite{warnick_quasinormal_2015} in the anti-de Sitter
case. However, we provide a more detailed analysis of the
$\LSolHk{k}$-quasinormal spectrum and its relation to the initial
value problem following an analysis analogous to that in Section 5 of
\cite{hintz_global_2018}. We remark that Sections
\ref{linear:sec:sol-op-semigroup}, \ref{linear:sec:Hk-QNM-Spectrum}, and
\ref{linear:sec:Laplace-transformed-op} make no assumptions on the particular
choice of strongly hyperbolic operator. Sections
\ref{linear:sec:Hk-QNM-orthogonality}, \ref{linear:sec:QNM-and-IVP}, and
\ref{linear:sec:QNM-perturbation-theory} introduce certain assumptions (see
Assumption \ref{linear:ass:QNM}) on the operators which in particular, will
be shown in Section \ref{linear:sec:asymptotic-expansion} to be satisfied for
the linearized gauged Einstein operator.

\subsection{Solution operator semigroup}
\label{linear:sec:sol-op-semigroup}

We begin with a definition of the solution operator
semigroup, which maps initial data to the evolution of a solution to
$\LinearOp\psi=0$ with the given initial data. In this section, we
work on a fixed slowly-rotating \KdS{} background $g_b$ and drop the
$b$ subscript, denoting $g=g_b$, $\GInvdtdt = \GInvdtdt_b$. We also
work with a general strongly hyperbolic operator $\LinearOp$, although
our results will clearly also apply to $\LinEinstein_g$. 

\begin{lemma}
  Any strongly hyperbolic operator on a fixed slowly-rotating \KdS{}
  black hole $g=g_b$
  \begin{equation}
    \label{linear:eq:QNM:L-form}
    \LinearOp = \ScalarWaveOp[g] + \SubPOp + \PotentialOp,\qquad
    \SubPOp = S^\mu\p_\mu,
  \end{equation}
  where $S^\mu$ and $\PotentialOp$ are matrices of smooth functions, and
  $\ScalarWaveOp[g]$ denotes the Laplace-Beltrami operator of $g$, can
  be rewritten as
  \begin{equation}
    \label{linear:eq:Sol-Op:L-def}
    \LinearOp = P_2+ \frac{1}{\ImagUnit}P_1D_{\tStar} + \frac{1}{\GInvdtdt} D_{\tStar}^2,
  \end{equation}
  where $D_{\tStar} = \ImagUnit \p_{\tStar}$, $\GInvdtdt$ is defined in
  \eqref{linear:eq:GInvdtdt-def}, and $P_i$ are $i^{\text{th}}$-order differential
  operators such that
  \begin{equation}
    \label{linear:eq:P-i-def}
    \begin{split}
      P_1 &= g^{i\tStar}\p_i + S^{\tStar},\\
      P_2 &= g^{ij}\p_i\p_j + S^i\p_i + \PotentialOp.
    \end{split}  
  \end{equation}
  
  Moreover, let $(h_0, h_1)\in \LSolHk{1}(\Sigma_0)$.  Then we
  can rewrite the Cauchy problem
  \begin{equation}
    \label{linear:eq:Sol-Op:Cauchy}
    \begin{split}
      \LinearOp h &= 0,\\
      \gamma_0(h) &= (h_0, h_1),
    \end{split}
  \end{equation}
  as a first-order system
  \begin{equation}
    \label{linear:eq:Sol-Op:first-order}
    \begin{split}
      \KillT
      \begin{pmatrix}
        h\\ h'
      \end{pmatrix}
      = 
      \begin{pmatrix}
        0 &  1\\
        \GInvdtdt P_2& \GInvdtdt  P_1
      \end{pmatrix}
      \begin{pmatrix}
        h\\
        h'
      \end{pmatrix},\qquad
      \left.\begin{pmatrix}
          h\\
          h'
        \end{pmatrix}\right\vert_{\tStar=0} =
      \begin{pmatrix}
        h_0\\
        h_1
      \end{pmatrix}.
    \end{split}  
  \end{equation}
\end{lemma}
\begin{proof}
  It is a simple computation to verify that \eqref{linear:eq:Sol-Op:L-def}
  and \eqref{linear:eq:Sol-Op:first-order} holds for $P_2$ and $P_1$ as
  defined in \eqref{linear:eq:P-i-def}.
\end{proof}

\begin{definition}  
  Define the \emph{solution operator}
  associated to $\LinearOp$
  \begin{equation*}
    \begin{split}
      \SolOp(\tStar): \LSolHk{1}(\Sigma)&\to \LSolHk{1}(\Sigma),\\
      (h_0, h_1)&\mapsto \evalAt*{(h, \p_{\tStar}h)}_{\Sigma_{\tStar}},
    \end{split}
  \end{equation*}
  be the solution operator of the Cauchy problem in
  \eqref{linear:eq:Sol-Op:Cauchy} mapping the initial data to the solution at time $\tStar$. 
\end{definition}
Recall from Lemma \ref{linear:lemma:EVE:GHC-quasilinear} that $\LinEinstein_{g_b}$
is a strongly hyperbolic linear operator. In subsequent sections, we
use subscripts to denote the specific infinitesimal generator of
interest. As noted by Warnick in \cite{warnick_quasinormal_2015} on
anti-de-Sitter spacetimes, $\SolOp$ defines a $C^0$-semigroup (see
\cite{engel_one-parameter_2000}).

\begin{prop}
  \label{linear:prop:QNM:C0-semigroup}
  Let $\SolOp(\tStar)$ be the solution operator for
  $\LinearOp \psi= 0$ a strongly hyperbolic operator on a
  slowly-rotating \KdS{} background. Then the one-parameter family of
  operators $\SolOp(\tStar)$ defines a $C^0$-semigroup on
  $\LSolHk{k}(\Sigma)$.
\end{prop}
\begin{proof}
  See appendix \ref{linear:appendix:prop:QNM:C0-semigroup}.
\end{proof}
Associated to the $C^0$-semigroup is the closed \emph{infinitesimal
  generator} of the semigroup.
\begin{definition}
  \label{linear:def:inf-gen}
  Define
  \begin{equation*}
    D^k({\InfGen}) := \curlyBrace*{\psi\in
      \LSolHk{k}(\Sigma):\ImagUnit\lim_{\tStar\to 0+}\frac{\SolutionOp(\tStar)\psi -\psi}{\tStar} \in \LSolHk{k}(\Sigma)}
  \end{equation*}
  the domain of
  \begin{equation*}
    {\InfGen}\psi := \ImagUnit\lim_{\tStar\to 0+}\frac{\SolutionOp(\tStar)\psi -\psi}{\tStar}.
  \end{equation*}
  The unbounded operator $(D^k(\InfGen), \InfGen), $ is the
  \emph{infinitesimal generator} of the semigroup
  $\SolutionOp(\tStar)$ on $H^k(\Sigma)$. For $\LinearOp$ as in
  \eqref{linear:eq:Sol-Op:L-def}, we have that
  \begin{equation*}
    \InfGen = \ImagUnit \begin{pmatrix}
      0 &   1\\
      \GInvdtdt P_2&  \GInvdtdt  P_1
    \end{pmatrix}.
  \end{equation*}
\end{definition}

Observe that $(D^k(\InfGen), \InfGen)$ for different values of $k$
differ only in their domain. Moreover,
$D^k(\InfGen)\subset D^{k-1}(\InfGen)$, and that
$(D^k(\InfGen), \InfGen)$ and $(D^{k-1}(\InfGen), \InfGen)$ agree on
$D^k(\InfGen)\bigcap D^{k-1}(\InfGen)$.  We will need the following
classical properties of infinitesimal generators (See Corollary II.1.5
in \cite{engel_one-parameter_2000}):

\begin{prop}
  \label{linear:prop:basic-C0-semigroup-properties}
  The operator $(D^k(\InfGen), \InfGen)$ satisfies the following
  properties. 
  \begin{enumerate}
  \item The domain $D^k(\InfGen)$ is dense in
    $\LSolHk{k}(\Sigma)$.
  \item $(D^k(\InfGen), \InfGen)$ is
    a closed operator.
  \item There exists some $\GronwallExp$ such that the resolvent
    $(\InfGen-\sigma)^{-1}$ exists and is a bounded linear
    transformation of $\LSolHk{k}(\Sigma)$ onto $D^k(\InfGen)$ for
    $\Im(\sigma)>\GronwallExp$. In particular, 
    \begin{equation*}
      \norm*{\SolOp(\tStar)\mathbf{h}}_{\LSolHk{k}(\Sigma)} \lesssim e^{\GronwallExp\tStar}\norm*{\mathbf{h}}_{\LSolHk{k}(\Sigma)}
    \end{equation*}
  \end{enumerate}
\end{prop}
\begin{proof}
  These are well-known properties of infinitesimal
  generators.
\end{proof}

We also have the following relationship between the resolvent
$(\InfGen-\sigma)^{-1}$ and the solution operator.

\begin{lemma}
  \label{linear:lemma:semigroup-resolvent-Cauchy-sol-relation}
  Fix $({h}_0, {h}_1)\in \LSolHk{k}(\Sigma)$, and $\LinearOp$ a
  strongly hyperbolic linear operator on a slowly-rotating \KdS{}
  background, and consider a solution ${h}$ to the Cauchy problem
  \begin{equation*}
    \begin{split}
      \LinearOp {h} &= 0\\
      \gamma_0({h}) &= ({h}_0, {h}_1).  
    \end{split}    
  \end{equation*}
  Moreover, let $(D^k(\InfGen), \InfGen)$ be the infinitesimal
  generator for the $C^0$-solution semigroup $\SolOp(\tStar)$. Then
  for $\Im \sigma > \GronwallExp$ the constant in Proposition
  \ref{linear:prop:basic-C0-semigroup-properties},
  \begin{equation*}
    (\InfGen-\sigma)^{-1}\mathbf{h}_0 = \ImagUnit \int_{\Real^+} e^{\ImagUnit\sigma\sStar}\SolOp(\sStar)\mathbf{h}_0\,d\sStar,\qquad
    \mathbf{h}_0 = \begin{pmatrix}
        {h}_0\\
        {h}_1
      \end{pmatrix}.
  \end{equation*}
\end{lemma}
\begin{proof}
  This is straightforward from applying Lemma
  \ref{linear:lemma:semigroup:resolvent-sol-op-relation} and taking
  $\tStar\to\infty$. 
\end{proof}

\subsection{$\LSolHk{k}$-quasinormal spectrum and modes}
\label{linear:sec:Hk-QNM-Spectrum}
Having established a $C^0$-semigroup, a closed, densely-defined
infinitesimal generator $(D^k(\InfGen), \InfGen)$, we can now analyze
the asymptotic behavior of solutions to $(D_{\tStar}-\InfGen)\psi = 0$
via the quasinormal spectrum.
\begin{definition}
  \label{linear:def:QNM-spectrum}
  Let $\LinearOp$ be a strongly
  hyperbolic operator on a slowly-rotating \KdS{} background, and let
  $(D^k(\InfGen), \InfGen)$ be the infinitesimal generator of the
  associated semigroup on $\LSolHk{k}(\Sigma)$. Then
  $\sigma\in\Complex$ belongs to the $\LSolHk{k}$-\emph{quasinormal
    spectrum} of $\LinearOp$, denoted by $\QNFk{k}(\LinearOp)$, if
  \begin{enumerate}
  \item $\Im\sigma > \frac{1}{2}\max_{\Horizon=\EventHorizonFuture,\CosmologicalHorizonFuture}\left(\SHorizonControl{\LinearOp}[\Horizon] -
      \left(2k + \frac{1}{2}\right)\SurfaceGravity_{\Horizon} \right)$, 
  \item $\sigma$ is in the spectrum of $(D^k(\InfGen), \InfGen)$. 
  \end{enumerate}
  If $\sigma$ is an eigenvalue of $(D^k(\InfGen), \InfGen)$, it is
  called an $\LSolHk{k}$-\emph{quasinormal frequency} and its
  corresponding eigenfunctions, $\LSolHk{k}$-\emph{quasinormal mode
    solutions}. 
\end{definition}
\begin{remark}
  An advantage of this method of Definition \ref{linear:def:QNM-spectrum} is
  that we do not have to construct a meromorphic extension of the
  resolvent. In this methodology, the faster the decay of a
  quasinormal mode, the higher regularity we require in order to study
  it.
\end{remark}
\begin{remark}
  \label{linear:remark:QNM-spectrum:suboptimal-regularity}
  The restriction of the $\LSolHk{k}$-quasinormal spectrum of
  $\LinearOp$ to the half-space
  \begin{equation}
    \label{linear:eq:QNM-spectrum:suboptimal-regularity}
    \Im\sigma > \frac{1}{2}\max_{\Horizon=\EventHorizonFuture,\CosmologicalHorizonFuture}\left(\SHorizonControl{\LinearOp}[\Horizon] -
      \left(2k + \frac{1}{2}\right)\SurfaceGravity_{\Horizon} \right).
  \end{equation}
  In fact, using the methods in this paper, for any fixed
  $\varepsilon>0$, the
  $\LSolHk{k}$-quasinormal spectrum for a strongly hyperbolic operator
  on a sufficiently slowly-rotating \KdS{} background can be shown to be well-defined on
  the half-space
  \begin{equation}
    \label{linear:eq:QNM-spectrum:optimal-regularity}
    \Im\sigma >\frac{1}{2}
  \max_{\Horizon=\EventHorizonFuture,\CosmologicalHorizonFuture}\left(\SHorizonControl{\LinearOp}[\Horizon]
    - \left(2k+1\right)\SurfaceGravity_{\Horizon} + \varepsilon_{\Omega}\right).
  \end{equation}  
  The loss of a $\varepsilon_{\Omega}>0$ when compared to the regularity
  levels in  \cite{warnick_quasinormal_2015} is due to the fact that
  \KdS{} is not a globally stationary spacetime, and is instead only
  locally stationary (for a more in-depth discussion, see Section 5 of \cite{warnick_quasinormal_2015}).

  While the restriction to the half-space in
  \eqref{linear:eq:QNM-spectrum:suboptimal-regularity} is not optimal, it is
  nevertheless consistent with the application of the linear theory
  developed in the current paper to the context of proving nonlinear
  stability of \KdS{} in \cite{fang_nonlinear_2021}.  In particular,
  with the restriction in
  \eqref{linear:eq:QNM-spectrum:suboptimal-regularity}, the threshold
  regularity level $k_0$ in \eqref{linear:eq:threshold-reg-def} for the
  gauged linearized Einstein operator $\LinEinstein_{g_{b}}$ is
  $\frac{5}{2}+O(a)$.
\end{remark}

\subsection{Laplace-transformed operator}
\label{linear:sec:Laplace-transformed-op}

In this section, we define the Laplace-transformed operator and see
how it relates to the infinitesimal generator
$(D^k(\InfGen), \InfGen)$. We derive resolvent estimates for the
Laplace-transformed operator using the vectorfield method in Sections
\ref{linear:sec:energy-estimates} and \ref{linear:sec:ILED}. 

\begin{definition}
  Given a linear operator $L$, we construct the
  \emph{Laplace-transformed operator} of $L$ by:
  \begin{equation*}
    \widehat{\LinearOp}(\sigma)u = \left.e^{\ImagUnit\sigma\tStar}L(e^{-\ImagUnit\sigma\tStar}u)\right\vert_{\Sigma_{\tStar}}.
  \end{equation*}
  Thus, the Laplace transform of a strongly hyperbolic linear operator
  $\LinearOp = P_2 + P_1 D_{\tStar} + \GInvdtdt^{-1} D_{\tStar}^2$ is
  \begin{equation*}
    \widehat{\LinearOp}(\sigma) = P_2 + \sigma P_1 + \sigma^2 \GInvdtdt^{-1}.
  \end{equation*}
\end{definition}

We then define a family of domains for
$\widehat{\LinearOp}(\sigma)$, $D^k(\widehat{\LinearOp}(\sigma))$
for $k\in\mathbb{N}$ to be the closure of $C^\infty_0(\Sigma,
\Complex^D)$ with respect to the graph norm
$\norm{u}_{\InducedHk{k-1}(\Sigma)}^2 +
\norm*{\widehat{\LinearOp}(\sigma)u}_{\InducedHk{k-1}(\Sigma)}^2$, and
$D^k_\sigma(\widehat{\LinearOp}(\sigma))$ for $k\in\mathbb{N}$ to be
the closure of $C^\infty_0(\Sigma, \Complex^D)$ with respect to the
graph norm $\norm{u}_{\InducedHk{k-1}_\sigma(\Sigma)}^2 +
\norm*{\widehat{\LinearOp}(\sigma)u}_{\InducedHk{k-1}_\sigma(\Sigma)}^2$.

\begin{definition}  
  Given $u\in \mathcal{D}'(\Sigma, \Complex^D)$, define the
  \emph{adjoint Laplace-transformed operator}
  $\widehat{\LinearOp}^\dagger(\sigma)$ by
  \begin{equation*}
    \widehat{\LinearOp}^\dagger(\sigma) u
    = \left.e^{\ImagUnit\overline{\sigma}\tStar}\LinearOp^\dagger(e^{-\ImagUnit\overline{\sigma}\tStar}u)\right\vert_{\Sigma_{\tStar}},
  \end{equation*}
  where $\LinearOp^\dagger$ is as defined in \eqref{linear:eq:L-dagger-def}.
\end{definition}

Define the domain of $\widehat{\LinearOp}^\dagger(\sigma)$,
$D^k(\widehat{\LinearOp}^\dagger(\sigma))$, to be the closure of
$\mathcal{D}'(\Sigma; \Complex^D)$ with respect to the graph norm
$\norm{u}_{\InducedHk{k-1}(\Sigma)}^2 +
\norm{\widehat{\LinearOp}^\dagger(\sigma)u}_{\InducedHk{k-1}(\Sigma)}^2$. With this
domain, $\widehat{\LinearOp}^\dagger(\sigma)$ is a closed, densely
defined operator. 

\begin{lemma}
  The operator $\widehat{\LinearOp}^\dagger(\sigma)$ defined above is the
  adjoint of $\widehat{\LinearOp}(\sigma)$ with respect to the inner product
  \begin{equation*}
    \bangle{u_1, u_2}_{\LTwo(\Sigma)} = \int_{\Sigma} u_1\cdot\overline{u_2}.
  \end{equation*}
\end{lemma}

\begin{proof}
  This follows easily from Corollary \ref{linear:corollary:L-dagger-adjoint}
  applied to
  ${h}_1 = e^{-\ImagUnit\overline{\sigma}\tStar}\upsilon_1$,
  ${h}_2 = e^{\ImagUnit\overline{\sigma}\tStar}\upsilon_2$, where
  $\upsilon_1, \upsilon_2$ are the unique stationary lifts of
  $u_1, u_2 \in C^\infty_0(\Sigma; \Complex^D)$ respectively, and that
  $\Sigma$ is a compact interval.
\end{proof}

The main motivation for considering the Laplace-transformed operator
$\widehat{\LinearOp}(\sigma)$ comes from the following lemma.
\begin{lemma}
  \label{linear:lemma:laplace:inverse-A}
  Let $\LinearOp$ be a strongly hyperbolic operator such that
  \begin{equation*}
    \LinearOp = \ScalarWaveOp[g] + \SubPOp + \PotentialOp
    =P_2+ P_1D_{\tStar} + \frac{1}{\GInvdtdt} D_{\tStar}^2,
  \end{equation*}
  with $P_i$ as defined in \eqref{linear:eq:P-i-def}, and let $\InfGen$ be
  the infinitesimal generator of $\LinearOp$.  Then the resolvent
  $(\InfGen-\sigma)^{-1}$ is a bounded linear operator from
  $\LSolHk{k}(\Sigma)\to D^k(\InfGen)$ if and only if
  $\widehat{\LinearOp}(\sigma)^{-1}:\InducedHk{k-1}(\Sigma)\to
  D^k(\widehat{\LinearOp}(\sigma))$ exists as a bounded operator with
  $D^k(\widehat{\LinearOp}(\sigma))\subset \InducedHk{k}(\Sigma)$. In
  particular,
  \begin{equation}
    \label{linear:eq:laplace:inverse-A}
    (\InfGen - \sigma)^{-1}= 
    \begin{pmatrix}
      -1&0\\
      \sigma& 1
    \end{pmatrix}
    \begin{pmatrix}
      \GInvdtdt^{-1}\widehat{\LinearOp}(\sigma)^{-1}&0\\
      0& 1
    \end{pmatrix}
    \begin{pmatrix}
      \GInvdtdt P_1+ \sigma&1\\
      1& 0
    \end{pmatrix}.
  \end{equation}
\end{lemma}
\begin{proof}
  We can directly calculate that
  \begin{equation*}
    (\InfGen -\sigma) = 
    \begin{pmatrix}
      0& 1 \\
      1&- \left(\GInvdtdt P_1+ \sigma\right)
    \end{pmatrix}
    \begin{pmatrix}
      \GInvdtdt \widehat{\LinearOp}(\sigma) & 0\\
      0 & 1
    \end{pmatrix}
    \begin{pmatrix}
      -1&0\\
       \sigma& 1
    \end{pmatrix}.
  \end{equation*}
  The first and third matrices are both invertible, with
  \begin{equation*}
    \begin{pmatrix}
      -1&0\\
      \sigma& 1
    \end{pmatrix} =
    \begin{pmatrix}
      -1&0\\
      \sigma& 1
    \end{pmatrix}^{-1},\qquad
   \begin{pmatrix}
      \GInvdtdt P_1+ \sigma&1\\
      1& 0
    \end{pmatrix} = \begin{pmatrix}
      0& 1 \\
      1&- \left(\GInvdtdt P_1+ \sigma\right)
    \end{pmatrix}^{-1}.
  \end{equation*}
  Checking the domains of definition of all the operators, we see that
  if
  $\widehat{\LinearOp}(\sigma)^{-1}:\InducedHk{k-1}(\Sigma)\to
  D^k(\widehat{\LinearOp}(\sigma))$ is well-defined,
  \begin{equation*}
    (\InfGen - \sigma)^{-1}\circ (\InfGen-\sigma)=\Identity_{D^k(\InfGen)},
    \qquad (\InfGen-\sigma)\circ(\InfGen-\sigma)^{-1}=\Identity_{\LSolHk{k}(\Sigma)},
  \end{equation*}
  and thus, (\ref{linear:eq:laplace:inverse-A}) holds.

  To show that $(\InfGen - \sigma)^{-1}$ is bounded,
  we first recall that $P_1$ is a bounded map from $\InducedHk{1}(\Sigma)$ to
  $\LTwo(\Sigma)$. Now observe that the right-hand side of
  (\ref{linear:eq:laplace:inverse-A}) is bounded if $\widehat{\LinearOp}(\sigma)^{-1}$ is
  a bounded operator from $\InducedHk{k-1}(\Sigma)$ to $\InducedHk{k}(\Sigma)$.
\end{proof}

Lemma \ref{linear:lemma:laplace:inverse-A} allows us to analyze the family of
operators $\widehat{\LinearOp}(\sigma)$ in place of
$(\InfGen - \sigma)$. We have a similar corollary using
$\InducedHk{k}_\sigma(\Sigma)$ norms. 
\begin{corollary}
  Fix a (not necessarily bounded) subset $\Omega\subset \Complex$.
  The family of resolvents $(\InfGen-\sigma)^{-1}$ exists and is a
  bounded linear transformation of $\LSolHk{k}(\Sigma)$ onto
  $D^k(\InfGen)$ for all $\sigma\in\Omega$ if and only if
  $\widehat{\LinearOp}(\sigma)^{-1}:\InducedHk{k-1}_\sigma(\Sigma) \to
  D^k_\sigma(\widehat{\LinearOp}(\sigma))$ exists as a bounded
  operator with
  $D^k_\sigma(\widehat{\LinearOp}(\sigma))\subset
  \InducedHk{k}_\sigma(\Sigma)$ for all $\sigma\in \Omega$.
\end{corollary}
\begin{proof}
  The proof follows from the expression for $(\InfGen
  -\sigma)^{-1}$ in equation \eqref{linear:eq:laplace:inverse-A}.
\end{proof}

\subsection{$\LSolHk{k}$-quasinormal mode solutions and orthogonality}
\label{linear:sec:Hk-QNM-orthogonality}

In this subsection, we define the $\LSolHk{k}$-quasinormal mode
solutions (also known as resonant states), and establish an
orthogonality condition. We then show how to relate information about
the $\LSolHk{k}$-quasinormal spectrum back to solutions of initial
value problems. 

\begin{assumption}  
  \label{linear:ass:QNM}
  Throughout this section and the remainder of Section \ref{linear:sec:QNM}, we
  assume that the linear operator $\LinearOp$ and the
  infinitesimal generator $\InfGen$ of the solution semigroup of
  $\LinearOp$ 
  satisfy the following properties:
  \begin{enumerate}
  \item $\LinearOp$ is a strongly hyperbolic linear operator on a
    slowly-rotating \KdS{} background, $g_b$.
  \item There exists some $\SpectralGap>0$ and $k_0\ge 0$ such that the
    resolvent
    \begin{equation*}
      (\InfGen-\sigma)^{-1}:\LSolHk{k}(\Sigma)\to D^k(\InfGen)
    \end{equation*} is a meromorphic operator for
    $\Im\sigma>-\SpectralGap$ and $k>k_0$.
  \item Moreover, for the same $\SpectralGap>0$, there exists some $C>0$
    such that the resolvent $(\InfGen-\sigma)^{-1}$ is a uniformly
    bounded operator for $\Im\sigma>-\SpectralGap$ and $\abs*{\sigma}>C$.
  \end{enumerate}
\end{assumption}

\begin{remark}
  Since the $\LSolHk{k}$-quasinormal spectrum coincides with the poles
  of the resolvent $(\InfGen-\sigma)^{-1}$ in the half-plane
  \begin{equation*}
    \Im\sigma > \frac{1}{2}\max_{\Horizon = \EventHorizonFuture, \CosmologicalHorizonFuture} \left(
      \SHorizonControl{\LinearOp}[\Horizon] - \left(2k+\frac{1}{2}\right)\SurfaceGravity_{\Horizon}
    \right),
  \end{equation*}
  and since the poles of a meromorphic function are discrete, the
  second assumption in Assumption \ref{linear:ass:QNM} ensures that the
  $\LSolHk{k}$-quasinormal spectrum of $\LinearOp$ is discrete.
\end{remark}
\begin{remark}
  As will be shown later in Section \ref{linear:sec:intermediary-results},
  the linearized gauged Einstein operator $\LinEinstein_{g_b}$ does
  indeed satisfy the above assumptions. See Theorems
  \ref{linear:thm:meromorphic:main-A} and
  \ref{linear:thm:resolvent-estimate:inf-gen}.
\end{remark}

We define the $\LSolHk{k}$-\emph{resonant states}, or
$\LSolHk{k}$-\emph{quasinormal mode solutions}, of $\LinearOp$. Consider some
${h} = ({h}_0, {h}_1)$ a $\sigma_0$-frequency
$\LSolHk{k}$-quasinormal mode of $\LinearOp$ for some
$\sigma_0\in\QNFk{k}(\LinearOp)$. Using Lemma
\ref{linear:lemma:laplace:inverse-A}, we see that $({h}_0, {h}_1)$ must
satisfy
\begin{equation*}
  \widehat{\LinearOp}(\sigma_0){h}_0 = 0,\qquad -\ImagUnit\sigma_0 {h}_0 = {h}_1.
\end{equation*}

\begin{definition}
  Given some $\sigma_0\in \QNFk{k}(\LinearOp)$, we define the space of
  $\LSolHk{k}$-\emph{quasinormal mode solutions} of $\LinearOp$ with
  frequency $\sigma_0$ to be the set
  \begin{equation*}
    \QNMk{k}(\LinearOp, \sigma_0) = \curlyBrace*{\mathbf{v} = \sum_{k=0}^n e^{-\ImagUnit\sigma_0\tStar}\tStar^k \mathbf{u}_k(x): n\in \Natural_0, (D_{\tStar} - \InfGen)\mathbf{v} = 0, \mathbf{u}_k\in C^\infty(\Sigma)}.
  \end{equation*}
  Furthermore for a subset
  $\Xi\subset\curlyBrace*{\sigma\in\Complex: \Im\sigma >
    \max_{\Horizon=\EventHorizonFuture,\CosmologicalHorizonFuture}\left(\frac{1}{2}\SHorizonControl{\LinearOp}[\Horizon] +
      (\frac{1}{2}-k)\SurfaceGravity_{\Horizon}\right)}$, which contains only finitely many
  resonances of $\LinearOp$, we can define the set of
  $\LSolHk{k}$-quasinormal mode solutions with frequencies in $\Xi$ by 
  \begin{equation*}
    \QNMk{k}(\LinearOp, \Xi) := \bigoplus_{\sigma\in\Xi\bigcap\QNFk{k}(\LinearOp)}\QNMk{k}(\LinearOp, \sigma).
  \end{equation*}  
\end{definition}
It will be useful to have a frequency-space characterization of the
space of $\LSolHk{k}$-quasinormal mode solutions of $\LinearOp$.

\begin{prop}
  \label{linear:prop:QNM:freq-characterization}
  Fix $\sigma_0\in \QNFk{k}(\LinearOp)$. Then an equivalent characterization of
  $\QNMk{k}(\LinearOp, \sigma_0)$ is the set
  \begin{equation}
    \label{linear:eq:QNM:freq-characterization}
    \QNMk{k}(\LinearOp, \sigma_0) = \curlyBrace*{
      \Residue_{\sigma=\sigma_0}\left(e^{-\ImagUnit\sigma\tStar}(\InfGen
      - \sigma)^{-1}p(\sigma)\right): p(\sigma)\in P\left(\sigma, C^\infty(\Sigma)\right)},
  \end{equation}
  where we denote by $P(\sigma, \mathcal{D}'(\Sigma))$ the set of all
  polynomials in $\sigma$ with coefficient in $\mathcal{D}'(\Sigma)$. 
\end{prop}
\begin{proof}
  See appendix \ref{linear:appendix:prop:QNM:freq-characterization}. 
\end{proof}

We can likewise characterize the dual $\LSolHk{k}$-quasinormal mode solutions
of $\LinearOp$ using the $\bangle*{\cdot,\cdot}_{L^2(\Sigma)}$ inner
product,
\begin{equation*}
  \begin{split}
    \QNMk{k*}(\LinearOp, \sigma)
    = \curlyBrace*{
      \mathbf{v} = \sum_{k=0}^{n}e^{-\ImagUnit\overline{\sigma}\tStar}\tStar^k\mathbf{u}_k(x): n\in \Natural_0, (D_{\tStar} - \InfGen^*)\mathbf{v} = 0, \mathbf{u}_k\in\mathcal{D}'(\Sigma)
    }.  
  \end{split}
\end{equation*}
Following the same reasoning as above, we also have the
frequency-space characterization of the dual $\LSolHk{k}$-quasinormal
modes
\begin{equation*}
  \begin{split}
    \QNMk{k*}(\LinearOp, \sigma) 
    ={}&\curlyBrace*{
      \Residue_{\zeta =
        \overline{\sigma}}\left(e^{-\ImagUnit\zeta\tStar}(\InfGen^*-\zeta)^{-1}p(\zeta)\right):
      p(\zeta)\in P(\zeta,
        \mathcal{D}'(\Sigma))
    },  
  \end{split}
\end{equation*}
where $\InfGen^*$ is the infinitesimal generator of the dual solution
semigroup of $\LinearOp$. 

We can then  establish an orthogonality
condition to a finite set of $\LSolHk{k}$-quasinormal frequencies (see
the similar Proposition 5.7, Corollary 5.8 in \cite{hintz_global_2018}
formulated at the level of the Laplace-transformed operator).
\begin{prop}
  \label{linear:prop:lin-theory:lambda-map-def}
  Let $\Xi=\curlyBrace{\sigma_j}_{j=1}^{N_\Xi}\subset
  \QNFk{k}(\LinearOp)$ be a finite set of $\LSolHk{k}$-quasinormal
  frequencies, and fix $\beta > \max\curlyBrace{-\Im \sigma_j:
    \sigma_j\in\Xi}$. 
  Then define the continuous linear map\footnote{The time regularity
    in the definition of $\lambda$ can be improved beyond $H^{-1}$ to
    incorporate more negative order Sobolev spaces. But since these do
  not make an appearance in this paper, we do not optimize this
  further.}
  \begin{equation*}
    \begin{split}
      \lambda: H^{-1, \beta}(\Real_+, \LSolHk{k-1}(\Sigma))&\to\mathcal{L}(\QNMk{k*}(\LinearOp, \Xi), \overline{\Complex}),\\
      F&\mapsto \bangle{F, \cdot}_{L^2(\StaticRegionWithExtension)}, 
    \end{split}
  \end{equation*}
  mapping $F$ to a $\Complex$-anti-linear function on
  $\mathcal{L}(\QNMk{k*}(\LinearOp, \Xi))$.

  Then $\lambda(F) = 0$ if
  and only if $(\InfGen-\sigma)^{-1}\widehat{F}(\sigma)$ is
  holomorphic in a neighborhood of $\Xi$.
\end{prop}

\begin{proof}
  See appendix \ref{linear:appendix:prop:lin-theory:lambda-map-def}.
\end{proof}

\subsection{Quasinormal spectrum and the initial value problem}
\label{linear:sec:QNM-and-IVP}

We are now ready to use the $\LSolHk{k}$-quasinormal spectrum to
analyze the Cauchy problem given by
\begin{equation}
  \label{linear:eq:lin-theory:Cauchy-problem}
  \begin{cases}
    \LinearOp {h} &= f,  \\
    \InitData({h}) &= ({h}_0,{h}_1).
  \end{cases}
\end{equation}
We will require that the forcing term and initial data of the Cauchy
problem in \eqref{linear:eq:lin-theory:Cauchy-problem} have certain decay and
regularity properties.
\begin{definition}
  Let $k,\alpha\in\Real$. Then we define the space of data with
  regularity $k$,and  decay $\alpha$ to be 
  \begin{equation*}
    D^{k,\alpha}(\StaticRegionWithExtension):= H^{k-1,\alpha}(\StaticRegionWithExtension) \oplus \LSolHk{k}(\Sigma_0).
  \end{equation*}  
  We then define the norm
  \begin{equation*}
    \norm{(f, {h}_0, {h}_1)}_{D^{k,\alpha}(\StaticRegionWithExtension)}
    := \norm{f}_{H^{k-1,\alpha}(\StaticRegionWithExtension)}
    + \norm{({h}_0, {h}_1)}_{\LSolHk{k}(\Sigma_0)}. 
  \end{equation*}
\end{definition}

We have the first preliminary expression of ${h}(\tStar)$ in terms of
the resolvent $(\InfGen-\sigma)^{-1}$. 
\begin{lemma}
  \label{linear:lemma:contour-formulation:gronwall}
  Let $\LinearOp$ be a strongly hyperbolic linear operator on a
  slowly-rotating \KdS{} background $g$, and let
  $(f, {h}_0, {h}_1)\in D^{k,\alpha}(\StaticRegionWithExtension)$
  for some $k, \alpha>0$
  . Then if ${h}$ is a solution
  to the Cauchy problem \eqref{linear:eq:lin-theory:Cauchy-problem}, 
  there exists some $\GronwallExp>0$ such that 
  
  \begin{equation}
    \label{linear:eq:Inverse-Laplace-ImSigma=M}
    \mathbf{h}(\tStar)
     = \int_{\Im\sigma=\GronwallExp}e^{-\ImagUnit\sigma\tStar}(\InfGen - \sigma)^{-1}\widehat{F}(\sigma)\,d\sigma, 
   \end{equation}
   where
   \begin{equation}
     \label{linear:eq:F-from-f-ID-def}
     \mathbf{h} =
     \begin{pmatrix}
       {h}\\ \KillT{h}
     \end{pmatrix},\qquad 
     F(\tStar) = \frac{1}{\ImagUnit}
     \begin{pmatrix}
         \delta_0(\tStar){h}_0\\
       - \GInvdtdt f(\tStar) + \delta_0(\tStar){h}_1
     \end{pmatrix}, 
   \end{equation}
   where $\delta_0$ is the Dirac delta, and
   \begin{equation*}
     \widehat{F}(\sigma) = \int_0^{\infty} F(\tStar)e^{\ImagUnit\sigma\tStar}\,d\tStar
   \end{equation*}
   denotes the Laplace transform of $F$.  
\end{lemma}

\begin{proof}  
  From Duhamel's principle we have that
  \begin{equation}
    \label{linear:eq:Inverse-Laplace-ImSigma=M:Duhamel-app}
    \mathbf{h}(\tStar, \cdot) = \int_{0}^{\tStar} \SolOp(\tStar-\sStar) \begin{pmatrix}
      \delta_0(\sStar){h}_0(\cdot)\\
      -  \GInvdtdt f(\sStar, \cdot) + \delta_0(\sStar){h}_1(\cdot)
    \end{pmatrix}\,d\sStar    
  \end{equation}
  is a solution to the Cauchy problem in
  \eqref{linear:eq:lin-theory:Cauchy-problem}. Using Lemma
  \ref{linear:lemma:semigroup:resolvent-sol-op-relation} and letting the
  limits of integration tend to $+\infty$, we have that for
  $\Im \sigma>\GronwallExp$,
  \begin{equation*}
     (\InfGen-\sigma)^{-1}\widehat{F}(\sigma) = \ImagUnit\int_{\Real^+}e^{\ImagUnit\sigma\sStar}\SolOp(\sStar)\widehat{F}(\sigma)\,d\sStar. 
  \end{equation*}
  As a result, Laplace-transforming both sides of
  (\ref{linear:eq:Inverse-Laplace-ImSigma=M:Duhamel-app}), we have that for
  $\Im\sigma> \GronwallExp$,
  \begin{equation*}
    \widehat{\mathbf{h}}(\sigma) = (\InfGen - \sigma)^{-1}\widehat{F}(\sigma).
  \end{equation*}
  As a result at the cost of slightly increasing $\GronwallExp$, we  
  can take the inverse Laplace transform over the contour
  $\Im\sigma=\GronwallExp$, and
  \begin{equation}
    \mathbf{h}(\tStar) = \int_{\Im\sigma=\GronwallExp}e^{-\ImagUnit\sigma\tStar}(\InfGen - \sigma)^{-1}\widehat{F}(\sigma)\,d\sigma,
  \end{equation}
  as desired. 
\end{proof}

Next, we show that under certain conditions on the resolvent
$(\InfGen-\sigma)^{-1}$, solutions to the Cauchy problem decay
exponentially up to a finite number of non-decaying linear obstacles.

\begin{prop}
  \label{linear:prop:asymptotic-exp:general}
  Let $\LinearOp$ be some strongly hyperbolic linear operator
  satisfying the assumptions in Assumption \ref{linear:ass:QNM}. Also fix $k_0>0$ such that
  \begin{equation*}
    \max_{\Horizon = \EventHorizonFuture, \CosmologicalHorizonFuture}\frac{1}{2}\SHorizonControl{\LinearOp}[\Horizon] + \left(\frac{1}{4} - k_0 \right)\SurfaceGravity_{\Horizon} < - \SpectralGap, 
  \end{equation*}
  where $\SpectralGap$ is specified as in Assumption \ref{linear:ass:QNM}.
  
  Furthermore for $k\ge k_0$, let
  $\Xi= \curlyBrace*{\sigma_j}_{j=1}^{N_{\LinearOp}}$ denote the set
  of all $\LSolHk{k}$-quasinormal frequencies of $\LinearOp$ with
  $\Im\sigma>-\SpectralGap$. Let
  $(f, h_0, h_1)\in D^{k+1,\SpectralGap}(\StaticRegionWithExtension)$,
  and let $h$ be the solution to the Cauchy problem given by
  \begin{equation*}
    \begin{split}
      \LinearOp h &= f \\
      \gamma_0(h) &= (h_0, h_1).
    \end{split}
  \end{equation*}
  For $\mathbf{v} =
  \begin{pmatrix}
    h\\ \KillT h
  \end{pmatrix}$, we have that
  \begin{equation}
    \label{linear:eq:asymptotic-exp:general:pre-mode-stability}
    \mathbf{v} = \tilde{\mathbf{v}} + \sum_{j=1}^{N_{\LinearOp}}\sum_{\ell=1}^{d_j}e^{-\ImagUnit\sigma_j\tStar}\tStar^{\ell} \mathbf{u}_{j\ell}(x), 
  \end{equation}
  where $d_j$ is the multiplicity of $\sigma_j$, 
  $\displaystyle \sum_{j=1}^{N_{\LinearOp}}\sum_{\ell=1}^{d_j}e^{-\ImagUnit\sigma_j\tStar}\tStar^\ell
  \mathbf{u}_{j\ell}(x) \in \QNMk{k}(\LinearOp, \Xi)$,  and where
  $\tilde{\mathbf{v}}$ satisfies the decay bound
  \begin{equation*}
    \begin{split}
      \norm{\tilde{\mathbf{v}}}_{\LSolHk{k}(\Sigma_{\tStar})} \lesssim e^{-\SpectralGap\tStar}
        \norm{(f, h_0, h_1)}_{D^{k+1, \SpectralGap}(\StaticRegionWithExtension)},
    \end{split}
  \end{equation*}
  for $k>k_0$, and where there is a continuous mapping 
  \begin{equation*}
     (f, h_0, h_1)\mapsto\sum_{j=1}^{N_{\LinearOp}}\sum_{\ell=1}^{d_j}e^{-\ImagUnit\sigma_j\tStar}\tStar^\ell \mathbf{u}_{j\ell}(x).
  \end{equation*}  
\end{prop}
\begin{remark}
  Observe that as written, the individual summands
  $e^{-\ImagUnit\sigma_j\tStar}\tStar^\ell\mathbf{u}_{j\ell}$ in
  \eqref{linear:eq:asymptotic-exp:general:pre-mode-stability} are not
  necessarily $\LSolHk{k}$-quasinormal modes. However, we can rewrite
  \begin{equation*}
    \sum_{j=1}^{N_{\LinearOp}}\sum_{\ell=1}^{d_j}e^{-\ImagUnit\sigma_j\tStar}\tStar^\ell\mathbf{u}_{j\ell}(x)
    = \sum_{j=1}^{N_{\LinearOp}}\sum_{\ell=1}^{d_j}a_{j\ell}\left(\sum_{k=0}^{n_{j\ell}}e^{-\ImagUnit\sigma_j\tStar}\tStar^k\tilde{\mathbf{u}}_{j\ell k}\right),
  \end{equation*}
  where $a\in \Real$, and
  $\sum_{k=0}^{n_{j\ell}}e^{-\ImagUnit\sigma_j\tStar}\tStar^k\tilde{\mathbf{u}}_{j\ell
    k}$
    are indeed $\LSolHk{k}$-quasinormal mode solutions. 
\end{remark}
\begin{proof}
  Recall from Lemma \ref{linear:lemma:contour-formulation:gronwall} that
  denoting by $\InfGen$ the infinitesimal generator of the solution
  semigroup associated to $\LinearOp$, we
  can write
  \begin{equation*}
    \mathbf{h}(\tStar):= \begin{pmatrix}
      {h}(\tStar)\\
      \KillT{h}(\tStar)
    \end{pmatrix}
    = \int_{\Im\sigma=\GronwallExp}e^{-\ImagUnit\sigma\tStar}(\InfGen - \sigma)^{-1}\widehat{F}(\sigma)\,d\sigma,   
  \end{equation*}
  where
  \begin{equation*}
    F(\tStar, x) =\frac{1}{\ImagUnit}
     \begin{pmatrix}
         \delta_0(\tStar){h}_0\\
       - \GInvdtdt f(\tStar) + \delta_0(\tStar){h}_1
     \end{pmatrix}.
  \end{equation*}
  Recall from Lemma \ref{linear:lemma:contour-formulation:gronwall} that
  $(\InfGen - \sigma)^{-1}$ is a holomorphic family of operators for
  $\Im\sigma>\GronwallExp$. We have that
  $(\InfGen-\sigma)^{-1}\widehat{F}(\sigma)$ is meromorphic with only
  finitely many poles on the half-plane $\Im\sigma> -\SpectralGap$, we
  now deform the contour of the inverse Fourier transform in
  \eqref{linear:eq:Inverse-Laplace-ImSigma=M} from $\Im\sigma = \GronwallExp$
  to $\Im\sigma = -\SpectralGap$. To this end, consider the contours
  \begin{align*}
    \gamma_{\pm C} = \{\pm C + \ImagUnit s: -\SpectralGap\le s\le \GronwallExp\}.
  \end{align*}
  Using the construction of the Laplace-transformed norms
  \begin{equation}
    \label{linear:eq:asymptotic-expansion:aux1}
    \norm*{(\InfGen - \sigma)^{-1}\widehat{F}(\sigma)}_{\LSolHk{k}_\sigma(\Sigma)}
    \lesssim \frac{1}{(1+|\sigma|)} \norm*{(\InfGen - \sigma)^{-1}\widehat{F}(\sigma)}_{\LSolHk{k+1}_\sigma(\Sigma)}.
  \end{equation}
  By assumption $(\InfGen-\sigma)^{-1}$ is a uniformly bounded
  operator for $\Im\sigma>-\SpectralGap$, $|\sigma|>C$. Thus we know
  that for $C$ and $k$ sufficiently large,
  \begin{equation}
    \label{linear:eq:asymptotic-expansion:aux2}
    \norm*{(\InfGen - \sigma)^{-1}\widehat{F}(\sigma)}_{\LSolHk{k}_{\sigma}(\Sigma)}
    \lesssim \norm*{\widehat{F}(\sigma)}_{\LSolHk{k}_{\sigma}(\Sigma)}
    \lesssim \norm*{(f, h_0, h_1)}_{D^{k+1, -\Im\sigma}(\StaticRegionWithExtension)}.
  \end{equation}
  As a result, 
  \begin{equation*}
    \norm*{\int_{\gamma_{\pm C}}e^{-\ImagUnit\sigma\tStar}(\InfGen - \sigma)^{-1}\widehat{F}(\sigma)}_{\LSolHk{k}_{\sigma}(\Sigma)}\,d\sigma
    \lesssim \frac{e^{\GronwallExp\tStar}}{1+C}\norm*{(f, h_0, h_1)}_{D^{k+1, \SpectralGap}(\StaticRegionWithExtension)},
  \end{equation*}
  and for $k$ sufficient large, it is clear that for all $\tStar>0$,
  \begin{equation*}
    \lim_{C\to\infty}\int_{\gamma_{\pm C}}e^{-\ImagUnit\sigma\tStar}(\InfGen - \sigma)^{-1}\widehat{F}(\sigma)\,d\sigma = 0.
  \end{equation*}
  Then using Cauchy's integral formula we can perturb the contour of
  integration to obtain that
  \begin{align}
    \label{linear:asymptotic-expansion:proof:main}
    \mathbf{h}(\tStar, \cdot)
    ={} \int_{\Im\sigma = -\SpectralGap} e^{-\ImagUnit\sigma\tStar}(\InfGen-\sigma)^{-1}\widehat{F}(\sigma)\,d\sigma
    + \sum_{1\le j\le N_{\LinearOp}} \Residue_{\zeta = \sigma_j}\left(
      e^{-\ImagUnit\zeta\tStar}(\InfGen-\zeta)^{-1}\widehat{F}(\zeta)
      \right).
  \end{align}
  From 
  \eqref{linear:eq:asymptotic-expansion:aux2}, we see that
  \begin{equation*}
    \norm*{\int_{\Im\sigma = -\SpectralGap} e^{-\ImagUnit\sigma\tStar}(\InfGen-\sigma)^{-1}\widehat{F}(\sigma)\,d\sigma}_{\LSolHk{k}_{\sigma}(\Sigma)}
      \lesssim e^{-\SpectralGap \tStar}\norm*{(f, h_0, h_1)}_{D^{k+1, \SpectralGap}(\StaticRegionWithExtension)},
  \end{equation*}
  while we can recall from Proposition
  \ref{linear:prop:QNM:freq-characterization} that
  \begin{equation*}
    \sum_{1\le j\le N_{\LinearOp}} \Residue_{\zeta = \sigma_j}\left(
      e^{-\ImagUnit\zeta\tStar}(\InfGen-\zeta)^{-1}\widehat{F}(\zeta)
    \right)\in \QNFk{k}(\LinearOp, \Xi)
  \end{equation*}
  where $\Xi$ is some open subset containing all the non-decaying
  $\LSolHk{k}$-quasinormal modes of $\LinearOp$. Thus, $h$ has exactly
  the desired form in \eqref{linear:eq:asymptotic-exp:general:pre-mode-stability}.
\end{proof}

Next we show that we can produce exponentially decaying solutions to
the Cauchy problem provided we can modify the forcing term and the
initial data within some finite-dimensional space of modifications
(compare with Corollary 5.8 in \cite{hintz_global_2018}).
\begin{corollary}
  \label{linear:coro:lambda-IVP}
  Let $\LinearOp$ be some strongly hyperbolic linear operator
  satisfying Assumption \ref{linear:ass:QNM}, and fix $k_0>0$ such that
  \begin{equation*}
    \max_{\Horizon = \EventHorizonFuture, \EventHorizonFuture}\frac{1}{2}\SHorizonControl{\LinearOp}[\Horizon] + \left(\frac{1}{2} - k_0 \right)\SurfaceGravity_{\Horizon} < - \SpectralGap,
  \end{equation*}
  with $\SpectralGap$ as in Assumption \ref{linear:ass:QNM}.  
  
  Then, for $k> k_0$, let
  $\Xi= \curlyBrace*{\sigma_j}_{j=1}^{N_{\LinearOp}}$ denote the set
  of all $\LSolHk{k}$-quasinormal frequencies of $\LinearOp$ with
  $\Im\sigma>-\SpectralGap$, and let
  $\ZCal\subset D^{k,\SpectralGap}(\StaticRegionWithExtension)$ be a
  finite-dimensional linear subspace. We define the map
  \begin{equation}
    \label{linear:eq:lambda-IVP-def}
    \begin{split}
      \lambda_{IVP}: D^{k,\SpectralGap}(\StaticRegionWithExtension)
      &\to \mathcal{L}(\QNMk{k*}(\LinearOp, \Xi), \overline{\Complex})\\
      \lambda_{IVP}(f, {h}_0, {h}_1)&:= \lambda\left(
        \frac{1}{\ImagUnit}
        \begin{pmatrix}
          \delta_0 h_0\\
          -\GInvdtdt f+\delta_0{h}_1
        \end{pmatrix}\right),
    \end{split}
  \end{equation}
  where $\lambda$ is as constructed in Proposition
  \ref{linear:prop:lin-theory:lambda-map-def},  $\lambda_{\ZCal}$ denotes
  its restriction to $\ZCal$, and $\delta_0 = \delta_0(\tStar)$
  denotes the Dirac delta,
  \begin{equation*}
    \lambda_{\ZCal}: \ZCal \to
    \mathcal{L}(\QNMk{k*}(\LinearOp, \Xi), \overline{\Complex}),
    \qquad \lambda_{\ZCal}:=\lambda_{IVP}\vert_{\ZCal}.
  \end{equation*}

  Then, if $\lambda_{\ZCal}$ is surjective, for any choice of
  $(f, {h}_0, {h}_1)\in D^{k,
    \SpectralGap}(\StaticRegionWithExtension)$, there exists an
  element $z = (\tilde{f}, \tilde{{h}}_0, \tilde{{h}}_1) \in \ZCal$
  such that the initial value problem
  \begin{equation*}
    \begin{split}
      \LinearOp {h} &= f + \tilde{f},\\
      \gamma_0({h}) &= ({h}_0 + \tilde{{h}}_0, {h}_1 + \tilde{{h}}_1)
    \end{split}
  \end{equation*}
  has an exponentially decaying solution ${h}$ that satisfies the
  estimate
  \begin{equation}
    \label{linear:eq:lambda-IVP-solution-inequality}
    \norm{h}_{\HkWithT{k}(\Sigma_{\tStar})}
    \lesssim e^{-\SpectralGap \tStar} \norm{(f, {h}_0, {h}_1)}_{D^{k+1,\SpectralGap}(\StaticRegionWithExtension)}. 
  \end{equation}
  If moreover,
  $\lambda_{\ZCal}$ is bijective, then $z$ is unique and the map $(f,
  {h}_0, {h}_1) \to z$ is continuous.
\end{corollary}

\begin{proof}
  As in Proposition \ref{linear:prop:asymptotic-exp:general},
  we perturb the contour of integration in
  \eqref{linear:eq:Inverse-Laplace-ImSigma=M} to write that for $\mathbf{h} =
    \begin{pmatrix}
      {h}\\ \KillT{h}
    \end{pmatrix}$, where ${h}$ is a solution to
    \eqref{linear:eq:lin-theory:Cauchy-problem},
  \begin{equation}
    \label{linear:eq:phi-contour-deformation:freq-form}
    \mathbf{h}(\tStar, \cdot)
    = \int_{\Im\sigma = -\SpectralGap} e^{-\ImagUnit\sigma\tStar}(\InfGen-\sigma)^{-1}\widehat{F}(\sigma)\,d\sigma
    + \sum_{1\le j\le N_{\LinearOp}} \Residue_{\zeta = \sigma_j}\left(
      e^{-\ImagUnit\zeta\tStar}(\InfGen-\zeta)^{-1}\widehat{F}(\zeta)
    \right),
  \end{equation}
  where $F(\tStar, x)$ is as defined in \eqref{linear:eq:F-from-f-ID-def}. 

  Now we consider finding some $z\in\ZCal$ for which
  \begin{equation*}
    \lambda_{\ZCal}(z) =
    - \lambda_{IVP}(f, {h}_0, {h}_1)\in \mathcal{L}(\QNMk{k*}(\LinearOp,
    \Xi), \overline{\Complex}). 
  \end{equation*}
  This is clearly possible if $\lambda_{\ZCal}$ is surjective. If in
  addition, $\lambda_{\ZCal}$ is bijective, then we can consider the
  mapping $(f, {h}_0, {h}_1)\mapsto z$ defined by
  \begin{equation*}
    \lambda_{\ZCal}^{-1}\circ \lambda_{IVP}(f, {h}_0, {h}_1) = z
  \end{equation*}
  which is clearly a linear and continuous, and therefore bounded,
  map. Moreover, if
  \begin{equation*}
    z = \lambda_{\ZCal}^{-1}\circ \lambda_{IVP}(f,
    {h}_0, {h}_1) = (\tilde{f}, \tilde{{h}}_0, \tilde{{h}}_1),
  \end{equation*}
  then by construction
  \begin{equation*}
    \tilde{F} =  \frac{1}{\ImagUnit}\begin{pmatrix}
      \delta_0({h}_0+\tilde{{h}}_0)\\
      -\GInvdtdt (f + \tilde{f}) + \delta_0({h}_1 + \tilde{{h}}_1)
    \end{pmatrix},
  \end{equation*}
  satisfies 
  \begin{equation*}
    \sum_{1\le j\le N_{\LinearOp}} \Residue_{\zeta = \sigma_j}\left(
      e^{-\ImagUnit\zeta\tStar}(\InfGen-\zeta)^{-1}\widehat{\tilde{F}}(\zeta)
    \right) = 0.
  \end{equation*}
  As a result, using \eqref{linear:eq:phi-contour-deformation:freq-form}, we
  have that
  \begin{equation*}
    \mathbf{h}(\tStar, \cdot)
    = \int_{\Im\sigma = -\SpectralGap} e^{-\ImagUnit\sigma\tStar}(\InfGen-\sigma)^{-1}\widehat{\tilde{F}}(\sigma)\,d\sigma,
  \end{equation*}
  and the bound in \eqref{linear:eq:lambda-IVP-solution-inequality}
  immediately follows from Plancherel.
\end{proof}

\subsection{Perturbation theory of the quasinormal spectrum}
\label{linear:sec:QNM-perturbation-theory}

In this section, we explore the perturbation theory for the
$\LSolHk{k}$-quasinormal frequencies and modes. To this end, let us
consider a family of stationary strongly hyperbolic linear operators
$\{\LinearOp_w\}_{w\in W}$ on a family of slowly-rotating \KdS{}
backgrounds $b(w)$, where $W\subset \Real^{N_W}$ is a finite
dimensional open neighborhood of some fixed $w_0\in \Real^{N_W}$.

Let $\InfGen_w$ be the infinitesimal generator associated
to the operator $\LinearOp_w$. In what follows, we will apply the
results of this section to both the gauged linearized Einstein
operator $\LinEinstein_{g_b}$, and the constraint propagation operator
$\ConstraintPropagationOp_{g_b}$, perturbing results obtained on a
fixed \SdS{} background to a nearby \KdS{} background (see Section
\ref{linear:sec:KdS-QNM-perturb}). 

The following proposition is a collection of the main basic
perturbation theory we will use, and is analogous to Proposition 5.11
in \cite{hintz_global_2018}. We have restated the result below in
terms of the infinitesimal generators $\InfGen_w$. 
\begin{prop}
  \label{linear:prop:QNM-perturb:gen}
  Let $k>0$ for some fixed $k_0$ such that
  \begin{equation}
    \label{linear:eq:QNM-perturb:gen:threshold-reg-def}
    \sup_{w\in W; \Horizon = \EventHorizonFuture, \CosmologicalHorizonFuture}\left(\frac{1}{2}\SHorizonControl{\LinearOp_{b(w)}}[\Horizon] -\left( k_0  - \frac{1}{4}\right) \SurfaceGravity_{b(w), \Horizon}\right) < 0,
  \end{equation}
  and let
  \begin{equation*}
    \Omega\subset \curlyBrace*{\sigma\in\Complex: \Im\sigma >\max_{\Horizon=\EventHorizonFuture, \CosmologicalHorizonFuture}
      \frac{1}{2}\SHorizonControl{\LinearOp_{w_0}}[\Horizon] -
      \left(k-\frac{1}{4}\right)\SurfaceGravity_{b_0,\Horizon}
    } 
  \end{equation*}
  be a non-empty pre-compact set such that
  \begin{equation*}
    \QNFk{k}(\LinearOp_{w_0}) \cap \p\Omega = \emptyset.
  \end{equation*}
  Then the following hold.
  \begin{enumerate}[label=(\roman{enumi})]
  \item \label{linear:prop:QNM-perturb:gen:item1} The set $I:=\curlyBrace{(w,\sigma)\in W\times \Omega :
      (\InfGen_w - \sigma)^{-1}\text{ is bounded}}$ is open. 
  \item \label{linear:prop:QNM-perturb:gen:item2}Let
    $\mathcal{L}_{weak}(X_1, X_2)$ denote the space of bounded linear
    operators mapping $X_1\to X_2$ equipped with the weak operator
    topology, and $\mathcal{L}_{op}(X_1, X_2)$ the space of bounded
    linear operators mapping $X_1 \to X_2$ equipped with the norm
    topology.  Then the map
    \begin{equation*}
      I\ni (w,\sigma)\mapsto (\InfGen_w - \sigma)^{-1} \in
      \mathcal{L}_{weak}(\LSolHk{k}(\Sigma), D^{k}(\InfGen))
    \end{equation*}
    is
    continuous for all $k>k_0$ and also as a map into
    $\mathcal{L}_{op}(\LSolHk{k+\epsilon}(\Sigma), D^{k-\epsilon}(\InfGen))$.
  \item \label{linear:prop:QNM-perturb:gen:item3} The set $\QNFk{k}(\LinearOp_w)\bigcap \Omega$ depends
    continuously on $w$ in the Hausdorff distance sense, and the total
    rank
    \begin{equation*}
      d:= \sum_{\sigma\in \QNFk{k}(\LinearOp_w)\bigcap\Omega}\text{rank}_{\zeta=\sigma} (\InfGen_w-\zeta)^{-1}
    \end{equation*}
    is constant for $w$ near $w_0$. 
  \item \label{linear:prop:QNM-perturb:gen:item4} The total space of quasinormal modes
    $\QNMk{k}(\LinearOp_w, \Omega)\subset C^\infty(\StaticRegionWithExtension)$
    depends continuously on $w$ in the sense that there exists a
    continuous map
    \begin{equation*}
      W\times \Complex^d\to C^\infty(\StaticRegionWithExtension)
    \end{equation*}
    such that $\QNFk{k}(\LinearOp_w, \Omega)$ is the image of
    $\{w\}\times \Complex^d$. 
  \item \label{linear:prop:QNM-perturb:gen:item5} Likewise, the total space
    of dual states
    $\QNMk{k*}(\LinearOp_w, \Omega)\subset
    H_{loc}^{1-k_0}(\StaticRegionWithExtension)$ depends continuously
    on $w$.
  \end{enumerate}
\end{prop}
\begin{proof}
  See appendix \ref{linear:appendix:prop:QNM-perturb:gen}.
\end{proof}

\section{Main Theorem} \label{linear:sec:main-theorem}

\subsection{Statement of the main theorem}
We are now ready to state the main theorem.
\begin{theorem}[Main Theorem, version 2]
  \label{linear:thm:Main}
  Fix $k>k_0$, where $k_0$ is as defined in
  \eqref{linear:eq:threshold-reg-def}. Then, let
  $(\underline{g}', k')\in \Hk{k+1}(\Sigma_0; S^2T^*\Sigma_0)\oplus
  \Hk{k}(\Sigma_0; S^2T^*\Sigma_0)$ be solutions of the linearized constraint
  equations linearized around the initial data
  $(\InducedMetric_{b}, k_{b})$ of a slowly rotating \KdS{} background
  $(\StaticRegionWithExtension, g_b)$. Then there exists a solution
  ${h}$ to the initial value problem
  \begin{equation*}
    \begin{cases}
      \LinEinstein_{g_b} {h} =0 &\text{in }\StaticRegionWithExtension,\\
      \gamma_0({h}) = D_{(\InducedMetric_b, k_b)}i_{b,\Identity}(\InducedMetric', k')&\text{on }\Sigma_0,
    \end{cases}
  \end{equation*}
  with $i_{b, \Identity}$ defined in Proposition
  \ref{linear:prop:initial-data:ib-construction}.  Moreover, there exists some
  $\SpectralGap>0$,  some finite-dimensional family of 1-forms
  $\Theta\in C^\infty(\StaticRegionWithExtension)$ parametrized by
  $\vartheta: \Real^{N_\Theta}\to \Theta$, and some
  $b'\in T_bB, \omega\in \Real^{N_{\Theta}}$ such that
  \begin{equation*}
    {h} = g_b'(b') +  \nabla_{g_b}\otimes \vartheta(\omega) + \tilde{{h}}, \qquad
    g_b'(b'):= \frac{\p g_b}{\p b}b',
  \end{equation*}
  where
  $\tilde{{h}}$ satisfies the bounds
  \begin{align*}
    \sup_{\tStar}e^{\SpectralGap\tStar}\norm{\tilde{h}}_{\HkWithT{k}(\Sigma_{\tStar})}
    \lesssim  \norm*{\InducedMetric'}_{\Hk{k+1}(\Sigma_0)} + \norm*{k'}_{\Hk{k}(\Sigma_0)},
  \end{align*}
  and $b'$ and $\omega$ are small in the sense that
  \begin{equation*}
    \abs*{b'} + \abs*{\omega}
    \lesssim   \norm*{\InducedMetric'}_{\Hk{k+1}(\Sigma_0)} + \norm*{k'}_{\Hk{k}(\Sigma_0)}. 
  \end{equation*}
\end{theorem}
The main theorem in the case of $g_{b}=g_{b_0}$ a \SdS{} metric will
be proven in Section \ref{linear:sec:lin-stability:SdS}, and the general
slowly-rotating case in Section \ref{linear:sec:proof-of-main-thm}.

\begin{corollary}
  \label{linear:coro:Main:eventual-constraint}
  Fix $k$ as in Theorem \ref{linear:thm:Main}, and
  $(f, h_0, h_1)\in
  D^{k,\SpectralGap}(\StaticRegionWithExtension)$, where
  $\SpectralGap$ is as in Theorem \ref{linear:thm:Main}. Moreover, assume
  that there exists some $\TStar>0$ such that for all $\tStar>\TStar$,
  \begin{equation*}
    f(\tStar, \cdot) = 0 ,\qquad \left.\Constraint_{g_b}(h)\right\vert_{\Sigma_{\tStar}}=0,
  \end{equation*}
  and let $h$ be the solution to the Cauchy problem
  \begin{equation}
    \label{linear:eq:Main:eqns}
    \begin{split}
      \LinEinstein_{g_b}h &= f,\\
      \gamma_0(h) &= (h_0, h_1). 
    \end{split}
  \end{equation}
  Then, there exists some
  $\SpectralGap>0$,  some finite-dimensional family of 1-forms
  $\Theta\in C^\infty(\StaticRegionWithExtension)$ parametrized by
  $\vartheta: \Real^{N_\Theta}\to \Theta$, and some
  $b'\in T_bB, \omega\in \Real^{N_{\Theta}}$ such that
  \begin{equation*}
    {h} = g_b'(b') +  \nabla_{g_b}\otimes \vartheta(\omega) + \tilde{{h}}, \qquad
    g_b'(b'):= \frac{\p g_b}{\p b}b',
  \end{equation*}
  where
  $\tilde{{h}}$ satisfies the bounds
  \begin{align*}
    \sup_{\tStar}e^{\SpectralGap\tStar}\norm{\tilde{h}}_{\HkWithT{k}(\Sigma_{\tStar})}
    \lesssim  \norm*{\InducedMetric'}_{\Hk{k+1}(\Sigma_0)} + \norm*{k'}_{\Hk{k}(\Sigma_0)},
  \end{align*}
  and $b'$ and $\omega$ are small in the sense that
  \begin{equation}
    \label{linear:eq:Main:param-continuity}
    \abs*{b'} + \abs*{\omega}
    \lesssim   \norm*{\InducedMetric'}_{\Hk{k+1}(\Sigma_0)} + \norm*{k'}_{\Hk{k}(\Sigma_0)}. 
  \end{equation}
\end{corollary}

\subsection{Main intermediary results}
\label{linear:sec:intermediary-results}

In this section, we describe the main intermediary steps taken in
proving Theorem \ref{linear:thm:Main}. 
\begin{enumerate}
\item A Fredholm alternative result on the
  $\LSolHk{k}$-quasinormal modes for the gauged linearized Einstein
  operator $\LinEinstein_{g_b}$, which gives us the discreteness of the
  $\LSolHk{k}$-quasinormal spectrum, and the meromorphy of the
  resolvent $(\InfGen-\sigma)^{-1}$. See Theorem
  \ref{linear:thm:meromorphic:main-A}. 
\item A high-frequency estimate for the resolvent
  $(\InfGen-\sigma)^{-1}$ on a region of the complex plane, proving
  the existence of a spectral gap underneath the real axis, and
  $\LSolHk{k}$-quasinormal-mode-free regions of $\Complex$. See
  Theorem \ref{linear:thm:resolvent-estimate:inf-gen}. 
\item A mode stability statement for the ungauged linearized Einstein
  operator, which allows us to characterize the non-decaying modes of
  the gauged linearized Einstein operator as unphysical. See Theorem
  \ref{linear:thm:mode-stability-v1}. 
\end{enumerate}
\begin{remark}
  We remark that the first two steps are the latter two assumptions in
  Assumption \ref{linear:ass:QNM}. 
\end{remark}

Recall that we denote by $\LinEinstein_{g_b}$ the gauged
linearized Einstein operator around $g_b$,
\begin{equation*}
  \LinEinstein_{g_b}  = D_{g_b}(\Ric-\Lambda) - \nabla_{g_b}\otimes \Constraint_{g_b},
\end{equation*}
which by Lemma \ref{linear:lemma:EVE:GHC-quasilinear} is a strongly
hyperbolic operator. To $\LinEinstein_{g_b}$, we can associate the
solution operator $\SolOp_{g_b}(\tStar)$, and the infinitesimal generator
of the solution semigroup $\InfGen_{g_b}$. The first result we will need is a
Fredholm alternative for $\InfGen_{g_b}$. 

\begin{theorem}
  \label{linear:thm:meromorphic:main-A}
  For $g=g_b$ a sufficiently
  slowly-rotating \KdS{} background, let
  $\LinEinstein_{g_b} = \LinEinstein$ be the gauged linearized
  Einstein operator linearized around $g_b$. Let
  $(D^k(\InfGen_{g_b}), \InfGen_{g_b}) = (D^k(\InfGen), \InfGen)$ be
  the infinitesimal generator of the solution semigroup of
  $\LinEinstein_{g_b}$ on $\LSolHk{k}(\Sigma)$. Then, for
  \begin{equation*}
    2\Im(\sigma)> \max_{\Horizon=\EventHorizonFuture,\CosmologicalHorizonFuture}
    \left(\SHorizonControl{\LinEinstein}[\Horizon]
      - \left( 2k + \frac{1}{2} \right)\SurfaceGravity_{\Horizon}
      \right),
  \end{equation*}
  one of the following must be true:
  \begin{enumerate}
  \item $\sigma$ is in the resolvent set of
    $(D^k(\InfGen), \InfGen)$, 
  \item $\sigma$ is an eigenvalue of $(D^k(\InfGen),
    \InfGen)$ with finite multiplicity. 
  \end{enumerate}
  In particular, the latter possibility is true only for isolated
  values of $\sigma$. This implies that the resolvent is meromorphic
  on the specified half plane, and that the residues at the poles are
  finite rank operators.
\end{theorem}

\begin{remark}
  When compared to the approach of proving the stability of \KdS{} in
  \cite{hintz_global_2018}, and the preceding works in
  \cite{hintz_global_2016, vasy_microlocal_2013, hintz_resonance_2017,
    hintz_semilinear_2015},
  this theorem is equivalent to the idea of proving a
  \textit{meromorphic continuation}. A slight difference is that in
  the mentioned works, a meromorphic continuation is typically
  directly proven for the entire lower half-plane, while we will only
  need a meromorphic continuation that extends a finite amount below
  the real axis.

  Also, note that like Warnick
  \cite{warnick_quasinormal_2015,gajic_quasinormal_2019}, we construct
  the quasinormal modes as true eigenvalues of an operator on a
  Hilbert space, rather than as the poles of a cutoff resolvent. 
\end{remark}

\begin{proof}
  See Section \ref{linear:sec:meromorphic-continuation}. 
\end{proof}

To prove asymptotic stability, we will need to locate the
aforementioned eigenvalues. To this end, we first show that there are
only a finite number of non-decaying $\LSolHk{k}$-quasinormal modes,
and then we show that those finite non-decaying quasinormal modes are
in some sense unphysical.

To show that there are only a finite number of non-decaying
$\LSolHk{k}$-quasinormal modes, we prove the existence of a
high-frequency spectral gap.
\begin{theorem}
  \label{linear:thm:resolvent-estimate:inf-gen}
  Let $g_b$ be a slowly-rotating \KdS{} metric, and let
  $\LinEinstein$, $\InfGen$ denote the gauged linearized Einstein
  operator, and the infinitesimal generator of the solution semigroup
  associated to $\LinEinstein$ respectively.  Then there exists some
  $\SpectralGap>0$ and $C>0$ such that the resolvent
  $(\InfGen -\sigma)^{-1}$ exists and is a uniformly bounded linear
  transformation of $\LSolHk{k}(\Sigma)$ onto $D^k(\InfGen)$,
  satisfying the bound
  \begin{equation}
    \label{linear:eq:resolvent-estimate:inf-gen:main}
    \norm{(\InfGen-\sigma)^{-1}\mathbf{h}_0}_{D^k(\InfGen)} \lesssim \norm{\mathbf{h}_0}_{\LSolHk{k}(\Sigma)},\qquad
    \mathbf{h}_0:=
    \begin{pmatrix}
      h_0\\h_1
    \end{pmatrix}
  \end{equation}
  for all $\Im\sigma\ge -\SpectralGap$, $\abs*{\sigma}\ge C$,
  $k> k_0$, where
  \begin{equation}
    \label{linear:eq:threshold-reg-def}
    \SurfaceGravity_{\Horizon} \left(2k_0 + \frac{1}{2}\right)
    > 2\SHorizonControl{\LinEinstein}[\Horizon] , \quad
    \Horizon = \EventHorizonFuture, \CosmologicalHorizonFuture. 
  \end{equation}
\end{theorem}

\begin{remark}
  Observe that in particular, using the computation of
  $2\SHorizonControl{\LinEinstein}[\Horizon]$ in Lemma
  \ref{linear:lemma:SubPOp:horizons}, for a sufficiently
  slowly-rotating \KdS{} metric $g_b$, the choice $k_0>2$ satisfies
  the condition in \eqref{linear:eq:threshold-reg-def}. In the proof
  of nonlinear stability of the slowly-rotating \KdS{} family
  presented in \cite{fang_nonlinear_2021}, we will effectively take
  $k_0=3$ to avoid the use of fractional functional spaces.
\end{remark}

\begin{remark}
  Theorem \ref{linear:thm:resolvent-estimate:inf-gen} plays the role of 
  Theorem 4.3 in \cite{hintz_global_2018}. Both statements show the
  existence of a high-frequency spectral gap. The main difference
  comes from the method of proof. We provide a self-contained proof
  that circumvents the use of $b$-pseudo-differential calculus, a
  compactified spacetime, and frequency-based arguments outside of
  a neighborhood of the trapped set.
\end{remark}

\begin{proof}
  See Section \ref{linear:sec:Resolvent-Estimates}. 
\end{proof}
We will refer to $k_0$ as the threshold regularity level. The reliance
of the spectral gap on sufficiently large $k_0$ reflects that we need
the regularity level to be sufficiently large that the
$\LSolHk{k}$-quasinormal modes are decaying (see Section
\ref{linear:sec:meromorphic-continuation} for a more precise discussion). 

Due to Theorem \ref{linear:thm:resolvent-estimate:inf-gen}, we know that any
non-decaying quasinormal mode must lie in a compact region of the
complex plane. Moreover, since the eigenvalues of
$(D^k(\InfGen), \InfGen)$ are isolated, there can only be a finite
number of them in any compact region of $\Complex$. Thus, Theorems
\ref{linear:thm:meromorphic:main-A} and \ref{linear:thm:resolvent-estimate:inf-gen}
will together show that the gauged linearized Einstein equations decay
exponentially up to a compact perturbation, and will be proven in
Section \ref{linear:sec:asymptotic-expansion}. 

The results on the $\LSolHk{k}$-quasinormal spectrum in Theorems
\ref{linear:thm:meromorphic:main-A} and \ref{linear:thm:resolvent-estimate:inf-gen}
allow us to prove an asymptotic expansion for the gauged linearized
Einstein system. 
\begin{corollary}
  \label{linear:coro:asymptotic-expansion}
  Fix $b\in \BHParamNbhd$ a set of black-hole parameters for a
  slowly-rotating \KdS{} black-hole, and let $k>k_0$, where $k_0$ is
  the threshold regularity level defined in \eqref{linear:eq:threshold-reg-def},
  Let $(h_0, h_1)\in \LSolHk{k}$, and $f\in
  H^{k-1}_0(\StaticRegionWithExtension)$. 
  Then if ${h}$ is
  a solution of the initial value problem
  \begin{equation*}
    \begin{cases}
      \LinEinstein_{g_b}{h} &=  f\\
      \gamma_0({h}) &= (h_0, h_1),
    \end{cases}
  \end{equation*} 
  then, there exists $\SpectralGap, \GronwallExp>0$ such that ${h}$ has an asymptotic expansion
  \begin{equation}
    \label{linear:eq:asymptotic-exp:pre-mode-stability}
    {h} = \tilde{h} + \sum_{j=1}^{N_{\LinEinstein_{g_b}}}\sum_{\ell=1}^{d_j}e^{-\ImagUnit\sigma_j\tStar}\tStar^\ell {u}_{j\ell}(x), 
  \end{equation}
  where $\tilde{h}$ satisfies the decay bounds
  \begin{equation*}
    \begin{split}
      \norm{\tilde{h}}_{\HkWithT{k}(\Sigma_{\tStar})} &\lesssim e^{-\SpectralGap\tStar}
        \norm{(f, h_0, h_1)}_{D^{k+1, \SpectralGap}(\StaticRegionWithExtension)}.
    \end{split}
  \end{equation*} 
\end{corollary}
\begin{proof}
  Theorems
  \ref{linear:thm:meromorphic:main-A} and \ref{linear:thm:resolvent-estimate:inf-gen}
  confirm that $\LinEinstein_{g_b}$ is a strongly hyperbolic linear
  operator satisfying the assumptions in Proposition
  \ref{linear:prop:asymptotic-exp:general}. The result then follows
  directly. 
\end{proof}

In the asymptotic expansion for $h$ in Corollary
\ref{linear:coro:asymptotic-expansion}, there are a finite number of
non-decaying $\LSolHk{k}$-quasinormal mode solutions.  Thus, to show
the desired exponential decay, it remains to show that these
non-decaying $\LSolHk{k}$-quasinormal mode solutions are unphysical.

\begin{theorem}[Mode stability of $\LinEinstein_{g_b}$, version 1]
  \label{linear:thm:mode-stability-v1}
  If ${h}$ if a non-decaying $\LSolHk{k}$-quasinormal mode solution
  of $\LinEinstein_{g_b}$, then either
  \begin{enumerate}
  \item there exists some linearized \KdS{} metric $b'$, and some
    one-form $\omega$ such that
    \begin{equation}
      {h} = g_{b}'(b') + \nabla_{g_b}\otimes \omega,
    \end{equation}
    or;
  \item ${h}$ does not satisfy the linearized gauge constraint
    conditions. That is, that
    \begin{equation*}
      \Constraint_{g_b}{h} \neq 0.
    \end{equation*}
  \end{enumerate}
\end{theorem}
As a result if ${h}$ is a non-decaying $\LSolHk{k}$-quasinormal mode
solution, then ${h}$ is not a physical mode solution. These notions
are expanded upon in Section \ref{linear:sec:mode-stability}, and a
more precise statement of mode stability of $\LinEinstein_{g_b}$ is given
in Proposition \ref{linear:prop:KdS-QNM-perturb}.

\begin{remark}
  At the level of mode stability, the main difference between the
  present work and \cite{hintz_global_2018} is that Hintz and
  Vasy introduce constraint damping to treat the non-decaying
  resonances in \KdS. Doing this allows them to consider arbitrary
  initial data, which may or may not be admissible. We instead do not
  introduce constraint damping, and work only with admissible initial
  data that satisfies exactly the linearized gauge constraint, and
  thus generate true solutions to the linearized Einstein equations in
  harmonic gauge.
\end{remark}

\subsection{Structure of the remainder of the paper}

We outline the remaining parts of the paper. Section
\ref{linear:sec:energy-estimates} sets up the necessary estimates to prove
Theorem \ref{linear:thm:resolvent-estimate:inf-gen}, all of  which are
insensitive to the issue of trapping. 

Section \ref{linear:sec:ILED} set up the estimates necessary to prove Theorem
\ref{linear:thm:resolvent-estimate:inf-gen}. These estimates are organized so
that the majority of them avoid the trapped set. The estimate taking
place in a neighborhood of the trapped set is reserved for Section
\ref{linear:sec:ILED-trapping:KdS}, which is the most technically difficult
part of the paper. In particular, the estimate in Section
\ref{linear:sec:ILED-trapping:KdS} involves 

Section \ref{linear:sec:asymptotic-expansion} then contains the proofs for
the main intermediate results in Theorem \ref{linear:thm:meromorphic:main-A}
and Theorem \ref{linear:thm:resolvent-estimate:inf-gen}, which are direct
applications of the estimates of Section \ref{linear:sec:energy-estimates}
and Section \ref{linear:sec:ILED}.

Section \ref{linear:sec:mode-stability} is of a wholly different flavor than
the previous sections, dealing with the mode stability of the
linearized gauged Einstein operator $\LinEinstein_{g_b}$, and is a
straightforward application of a geometric mode stability statement
originally proven by Kodama and Ishibashi
\cite{ishibashi_stability_2003,kodama_master_2003}, and stated in
Theorem \ref{linear:thm:Kodama-Ishibashi} in the form presented in
\cite{hintz_global_2018}, and the perturbation theory in Proposition
\ref{linear:prop:QNM-perturb:gen}. 

\subsection{Choice of constants}

In what follows, the proof involves the choice of several
constants. We review these constants and their relationship to each
other.

\begin{enumerate}
\item $C_{\RedShiftN}(\epsilon)$ and $C_{\TFixer}(\delta)$ are
  (large )implicit constants that are introduced in Section
  \ref{linear:sec:energy-estimates}, chosen so that
  \begin{equation*}
    C_{\RedShiftN}(\epsilon) \gg \max\left(1, \frac{1}{\epsilon}\right), \qquad
    C_{\TFixer}(\delta) \gg \max\left(1, \frac{1}{\delta}\right). 
  \end{equation*}
\item $\gamma$ is the size of a lower-order correction term added to
  the $\TFixer$-energy estimate and the redshift estimate in Section
  \ref{linear:sec:energy-estimates}. When applied in Section
  \ref{linear:sec:meromorphic-continuation}, $\delta_0$, $\epsilon_0$ will be
  some fixed
  \begin{equation*}
    \delta_0, \epsilon_0 \ll 1,
  \end{equation*}
  and $\gamma$ will be chosen such that
  \begin{equation*}
    \gamma \gg C_{\RedShiftN}(\epsilon_0)C_{\TFixer}(\delta_0).
  \end{equation*}
\item $\delta_{\Horizon}$ is a smallness constant that measures the
  amount that we extend $\Sigma$ beyond the horizons. In Section
  \ref{linear:ILED:full}, we choose $\delta_{\Horizon}$ such that
  \begin{equation*}
    \delta_{\Horizon}^2\ll \max\left(
      \frac{r_0}{r_{\EventHorizonFuture}}, \frac{r_{\CosmologicalHorizonFuture}}{R_0}
    \right).
  \end{equation*}
\item $\delta_r$ and $\delta_\zeta$ are smallness constants measuring
  the size of the localization around the trapped set we take in
  Section \ref{linear:sec:ILED}.
\item $\varepsilon_{\TrappedSet}$ is a smallness constant that
  measures the smallness of the skew-adjoint component of the
  subprincipal symbol at the trapped set, and we will take
  \begin{equation*}
    \varepsilon_{\TrappedSet} \ll 1
  \end{equation*}
  in Section \ref{linear:ILED:near}.
\item $\GronwallExp$ will denote the maximal exponential growth rate
  of solutions to the linearized Einstein equations. 
\item $\COuter$ and $\CInnerNT$ are large constants that produces 
  large bulk terms in the Morawetz estimates in Sections
  \ref{linear:sec:ILED:nontrapping} and \ref{linear:sec:ILED:nontrapping-freq} which
  will be chosen so that
  \begin{equation*}
    \COuter \gg \frac{1}{\delta_r^2}, \qquad \CInnerNT\gg \frac{1}{\delta_\zeta}. 
  \end{equation*}
\item We will use $C_0$ to denote the large lower-order error that is
  incurred in all of our high-frequency estimates. It is the largest
  constant in the proofs for high-frequency resolvent estimates, and
  is chosen so that
  \begin{equation*}
    \GronwallExp \ll \COuter, \CInnerNT \ll C_0.
  \end{equation*}
\item $\mathring{c}$, $\check{c}$, and $\breve{c}$ are constants
  introduced in Section \ref{linear:ILED:full} to help glue the
  Morawetz estimates on the various regions of phase space
  together. For a more precise description of their relationship with
  the other constants here, see Section \ref{linear:ILED:full}.
\item $b=(M,a)$ and $b_0 = (M,a)$ denote the \KdS{} and \SdS{} black
  hole parameters we will consider. We treat $\Lambda$ as a fixed
  constant, and $M$ satisfying the mass-subextremal condition
  in~\eqref{linear:eq:SdS:non-degeneracy-condition}, and $a$ such that
  \begin{equation*}
    0\le a \ll \min\left(
      \varepsilon_{\RedShiftN},
      \delta_r,
      \delta_\zeta,
      \delta_{\Horizon},
      \frac{\epsilon_0}{C_{\RedShiftN}(\epsilon_0)C_{\TFixer}(\delta_0)},
      \frac{r_0}{r_{\EventHorizonFuture}} - 1,
      \frac{r_{\CosmologicalHorizonFuture}}{R_0} - 1
    \right).
  \end{equation*}
\end{enumerate}

\section{Energy estimates}
\label{linear:sec:energy-estimates}

In this section, we prove a variety of energy estimates that will be
used throughout the paper. In particular, we will prove a Killing
energy estimate, a redshift energy estimate, and an enhanced redshift
energy estimate. The Morawetz estimate, or
its resolvent estimate equivalent, is the subject of Section
\ref{linear:sec:ILED}, due to the different approach and difficulties
involved in the proof. All the estimates in this section are proven
relying purely on the physical space vectorfield method, using
vectorfields as multipliers and commutators.

We prove the estimates in this section for any strongly hyperbolic
linear operator on a slowly-rotating \KdS{} background of the form
\begin{equation}
  \label{linear:eq:energy-estimates:strongly-hyperbolic-operator}
  \LinearOp{h} = \ScalarWaveOp[g_b] {h} + \SubPOp[{h}] + \PotentialOp{h},
\end{equation}
where $\SubPOp$ is a matrix-valued vectorfield, and $\PotentialOp$ is
a smooth matrix valued potential.  Recall from Lemma
\ref{linear:lemma:EVE:GHC-quasilinear} that the gauged linearized Einstein
operator is itself a strongly hyperbolic linear operator. As we are
only working with a general strongly hyperbolic linear operator, we do
not need any additional structural assumptions from Einstein's
equations.  In particular, the estimates proven in this section are
blind to the presence of the trapped set. This should be contrasted
with the derivations of the Morawetz estimates in Section
\ref{linear:sec:ILED}, where trapping plays a crucial role and we require a
precise structure in Einstein's equations at the trapped set to close
the argument.

\subsection{$\TFixer$-energy estimates}
\label{linear:sec:T-fixer-energy-estimates}

The Killing estimate will be derived via using the almost Killing
vectorfield $\TFixer$ of Lemma \ref{linear:lemma:T-Fixer:construction} as a
multiplier. We define the $\TFixer$-energy by
\begin{equation}
  \label{linear:eq:Energy-Kill:def}
  \EnergyKill(\tStar)[{h}] = \int_{\widetilde{\Sigma}_{\tStar}} \JCurrent{\TFixer, 0, 0} [{h}]\cdot n_{\Sigma},
\end{equation}
where $\JCurrent{X,q,m}[h]$ is as defined in
\eqref{linear:eq:J-K-currents:def}, with $(X,q,m) = (\TFixer, 0, 0)$, and we
recall the definition of $\widetilde{\Sigma}$ in \eqref{linear:eq:Sigma-tilde:def}.

\begin{theorem}
  \label{linear:thm:Killing-estimate}
  We have the following estimates. 
  \begin{enumerate}[label=(\roman{enumi})]
  \item \label{linear:item:Killing-estimate:one}There exists a $C>0$ such
    that
    \begin{equation}
      \label{linear:eq:Killing-estimate:one}
      \EnergyKill(\tStar)[{h}]
      \lesssim  \norm{\nabla {h}}_{L^2(\widetilde{\Sigma}_{\tStar})}^2          
    \end{equation}
    for all ${h}\in C^\infty(\StaticRegionWithExtension,\Complex^N)$. 
  \item \label{linear:item:Killing-estimate:two}For any vectorfield $X$
    tangent to both $\CosmologicalHorizonFuture$ and
    $\EventHorizonFuture$, there exists $C_X$ depending on $X$ such
    that for all ${h}$,
    \begin{equation}
      \label{linear:eq:Killing-estimate:two}
      \norm{X{h}}^2_{\LTwo(\widetilde{\Sigma}_{\tStar})} \le C_X \EnergyKill(\tStar)[{h}].
    \end{equation}
  \item \label{linear:item:Killing-estimate:three}For any $\epsilon>0$, there
    exists a constant $C(\epsilon)$ such that
    \begin{equation}
      \label{linear:eq:Killing-estimate:three}
      \p_{\tStar} \EnergyKill(\tStar)[{h}] \le{} \epsilon\left(
        \norm{{h}}^2_{\InducedHk{1}(\widetilde{\Sigma}_{\tStar})} + \norm{\LinearOp{h}}^2_{\InducedLTwo(\widetilde{\Sigma}_{\tStar})}
      \right) + C(\epsilon)\norm{\KillT{h}}_{\InducedLTwo(\widetilde{\Sigma}_{\tStar})}^2 
      + a C(\epsilon) \norm{h}^2_{\InducedHk{1}(\widetilde{\Sigma}_{\tStar})}.
    \end{equation}
    
  \item \label{linear:item: Killing-estimate:four} If ${h}$ is additionally
    assumed to vanish on the horizons, there exists a constant
    $C(\epsilon)$ independent of $\gamma$ such that for any
    $\epsilon>0$, the following estimate holds
    \begin{equation}
      \label{linear:eq:Killing-estimate:four}
      -\p_{\tStar} \EnergyKill(\tStar)[{h}] \le{} \epsilon\left(
        \norm{{h}}^2_{\InducedHk{1}(\widetilde{\Sigma}_{\tStar})} + \norm{\LinearOp{h}}^2_{\InducedLTwo(\widetilde{\Sigma}_{\tStar})}
      \right) + C(\epsilon)\norm{\KillT{h}}_{\InducedLTwo(\widetilde{\Sigma}_{\tStar})}^2
      + a C(\epsilon) \norm{h}^2_{\InducedHk{1}(\widetilde{\Sigma}_{\tStar})}.
    \end{equation}
  \end{enumerate}
\end{theorem}

\begin{proof}
  To prove \ref{linear:item:Killing-estimate:one} and
  \ref{linear:item:Killing-estimate:two}, it is sufficient to observe that
  $\TFixer$ by definition is time-like in $\StaticRegion$ as defined
  in \eqref{linear:eq:static-region-def}, and null
  along both the event horizon and the cosmological horizon.
  To establish \ref{linear:item:Killing-estimate:three}, we apply the
  divergence theorem of Corollary \ref{linear:cor:div-them:spacelike} on
  $\Sigma_{\tStar}$ with the vectorfield
  $\JCurrent{\TFixer, 0, 0}[{h}]$,
  \begin{equation} \label{linear:eq:Killing-estimate:Div-Theorem-app-aux-1}  
    - \int_{\widetilde{\Sigma}_{\tStar}}\nabla\cdot\JCurrent{\TFixer, 0, 0}[{h}]\, \sqrt{\GInvdtdt}
    = \p_{\tStar}\EnergyKill(\tStar)[h]
    + \int_{\CosmologicalHorizonFuture\bigcap\widetilde{\Sigma}_{\tStar}}
    \JCurrent{\TFixer, 0, 0}[{h}] \cdot \HorizonGen_{\CosmologicalHorizonFuture}
    +\int_{\EventHorizonFuture\bigcap\widetilde{\Sigma}_{\tStar}}
    \JCurrent{\TFixer, 0, 0}[{h}] \cdot \HorizonGen_{\EventHorizonFuture},
  \end{equation}
  and handle the terms individually. The first term on the right hand
  side is left alone as it gives rise to the derivative of the
  energy.
  Consider surface integrals at the horizons. Since $\TFixer$ is
  null, future-directed,  and tangent to both $\EventHorizonFuture$ and
  $\CosmologicalHorizonFuture$, the surface integrals at the
  horizons are positive. That is,
  \begin{equation*}
    \int_{\EventHorizonFuture\bigcap \widetilde{\Sigma}_{\tStar}}\JCurrent{\TFixer, 0, 0}[{h}]\cdot \HorizonGen_{\EventHorizonFuture} \ge 0,\qquad
    \int_{\CosmologicalHorizonFuture\bigcap\widetilde{\Sigma}_{\tStar}} \JCurrent{\TFixer, 0, 0}[{h}]\cdot \HorizonGen_{\CosmologicalHorizonFuture} \ge 0.
  \end{equation*}
  We now deal with the divergence term on the right-hand side of
  (\ref{linear:eq:Killing-estimate:Div-Theorem-app-aux-1}). Using
  \eqref{linear:eq:J-K-currents:def} and the
  divergence property of the energy-momentum tensor in
  \eqref{linear:eq:div-them:J-K-currents}, we have that
  \begin{align*}
    \nabla\cdot\JCurrent{\TFixer, 0, 0}[{h}]
    &=  \Re\left[
      \TFixer\overline{{h}}\cdot
      \ScalarWaveOp[g]{h}      
      \right]
      + \KCurrent{\TFixer, 0, 0}[{h}]\\
    &= \Re\left[
      \TFixer\overline{{h}}\cdot\left(
      \LinearOp{h} -\SubPOp[{h}] -\PotentialOp[h]
      \right)
      \right]
       + \KCurrent{\TFixer, 0, 0}[{h}].
  \end{align*}
  We first consider
  \begin{equation}
    \label{linear:eq:Killing:L-times-Vfield:aux1}
    \Re\left[
      \TFixer\overline{h} \cdot \left(
      \LinearOp{h} -\SubPOp[{h}] -\PotentialOp[h]
      \right)
    \right].
  \end{equation}
  These terms can each be directly controlled by Cauchy-Schwarz, using
  the fact that $\TFixer = a\,\tilde{\chi}(r)\KillPhi$,
  \begin{equation*}
    \begin{split}
      &\int_{\widetilde{\Sigma}_{\tStar}}
      \Re\left[
        \TFixer\overline{{h}}\cdot
        \left(
          \LinearOp{h} -\SubPOp[{h}] -\PotentialOp[h]
        \right)
      \right]
      \,\sqrt{\GInvdtdt}\\
      \le{}& \epsilon\left(
        \norm{{h}}^2_{\InducedHk{1}(\widetilde{\Sigma}_{\tStar})} + \norm{\LinearOp{h}}^2_{\InducedLTwo(\widetilde{\Sigma}_{\tStar})}
      \right) + C(\epsilon)\norm{\KillT{h}}_{\InducedLTwo(\widetilde{\Sigma}_{\tStar})}^2
      + a C(\epsilon) \norm{h}^2_{\InducedHk{1}(\widetilde{\Sigma}_{\tStar})}.
    \end{split}
  \end{equation*}
  We now consider the deformation tensor term. Using 
  equation (\ref{linear:eq:TFixer:DeformTensor}), we see that
  \begin{equation}
    \label{linear:eq:Killing-estimate:aux1}
    \int_{\widetilde{\Sigma}_{\tStar}}\KCurrent{\TFixer, 0, 0}[h]\,\sqrt{\GInvdtdt}
    \lesssim{} 
     a\norm{{h}}^2_{\InducedHk{1}(\widetilde{\Sigma}_{\tStar})}
    + a \norm{\KillT h}^2_{\InducedLTwo(\widetilde{\Sigma}_{\tStar})}.
  \end{equation}
  We conclude the proof of \ref{linear:item:Killing-estimate:three} by taking
  $C(\epsilon)$ sufficiently large. 

  The final statement, \ref{linear:item: Killing-estimate:four}, is
  proven in the same way as
  \ref{linear:item:Killing-estimate:three}. Since we have the additional
  assumption that ${h}$ vanishes on $\EventHorizonFuture$ as well
  as $\CosmologicalHorizonFuture$, we can neglect the surface
  integrals at the horizons, and only need to estimate the
  divergence term. This can be done using Cauchy-Schwarz in the same
  manner as \eqref{linear:eq:Killing-estimate:aux1}. 
\end{proof}
The presence of $O(a)$ terms on the right hand
side of our estimates arise from the issue of
superradiance. Since we are working in the slowly
rotating case, these superradiant terms will be handled by an
appropriate redshift argument, which is the subject of Theorem
\ref{linear:thm:redshift-energy-estimate}. 

We also prove the following corollary of Theorem
\ref{linear:thm:Killing-estimate} which will be useful in Section
\ref{linear:sec:asymptotic-expansion} for proving Theorem
\ref{linear:thm:meromorphic:main-A}.

\begin{corollary}
  \label{linear:coro:Killing-estimate-with-gamma}
  Define the energy
  \begin{equation}
    \label{linear:eq:EnergyKill-gamma:def}
    \EnergyKill_\gamma(\tStar)[{h}]
    = \int_{\widetilde{\Sigma}_{\tStar}} \JCurrent{\TFixer, 0, -\gamma \TFixer^\flat} [{h}]\cdot n_{\widetilde{\Sigma}_{\tStar}},
  \end{equation}
  where we recall the definition of $\JCurrent{X,q,m}[h]$ from
  \eqref{linear:eq:J-K-currents:def}. 
  
  We have the following estimates.
  \begin{enumerate}[label=(\roman{enumi})]
  \item \label{linear:item:Killing-estimate-with-gamma:one} For all
    ${h}\in C^\infty(\StaticRegionWithExtension,\Complex^D)$,
    \begin{equation}
      \label{linear:coro:Killing-estimate-with-gamma:one}
      \EnergyKill_\gamma(\tStar)[{h}]
      \lesssim  \norm{h}_{H^1(\widetilde{\Sigma}_{\tStar})}^2
      + \norm{\KillT h}_{L^2(\widetilde{\Sigma}_{\tStar})}^2
      + \gamma\norm*{h}_{L^2(\widetilde{\Sigma}_{\tStar})}^2.
    \end{equation}
    
  \item \label{linear:item:Killing-estimate-with-gamma:three}For any
    $\epsilon>0$, there exists a constant $C(\epsilon)$ independent of
    $\gamma$, and a constant $C$ independent of both $\epsilon$ and
    $\gamma$ such that
    \begin{align}      
      \p_{\tStar} \EnergyKill_{\gamma}(\tStar)[{h}]
      \le{}& \epsilon\left(
             \norm{{h}}^2_{\InducedHk{1}(\widetilde{\Sigma}_{\tStar})} + \norm{(\LinearOp - \gamma){h}}^2_{\InducedLTwo(\widetilde{\Sigma}_{\tStar})}
             \right)\notag \\
           &+ C(\epsilon)\norm{\KillT{h}}_{\InducedLTwo(\widetilde{\Sigma}_{\tStar})}^2 
             + a C(\epsilon) \norm{h}_{\InducedHk{1}(\widetilde{\Sigma}_{\tStar})}^2
              + aC \gamma\norm*{h}^2_{\InducedLTwo(\Sigma_{\tStar})}. \label{linear:eq:Killing-estimate-with-gamma:three}
    \end{align}
        
  \item \label{linear:item:Killing-estimate-with-gamma:four} If ${h}$ is additionally assumed to vanish on the horizon,
    there exists a constant $C(\epsilon)$ independent of $\gamma$, and
    a constant $C$ independent of both $\gamma$ and $\epsilon$ such that for
    any $\epsilon>0$, the following estimate holds
    \begin{align*}      
      -\p_{\tStar} \EnergyKill_\gamma(\tStar)[{h}]
      \le{}& \epsilon\left(
             \norm{{h}}^2_{\InducedHk{1}(\widetilde{\Sigma}_{\tStar})} + \norm{(\LinearOp - \gamma){h}}^2_{\InducedLTwo(\widetilde{\Sigma}_{\tStar})}
              \right)\notag \\
           & + C(\epsilon)\norm{\KillT{h}}_{\InducedLTwo(\widetilde{\Sigma}_{\tStar})}^2
             + a C(\epsilon) \norm{h}_{\InducedHk{1}(\widetilde{\Sigma}_{\tStar})}^2
              + a C\gamma\norm*{h}^2_{\InducedLTwo(\Sigma_{\tStar})}. \label{linear:eq:Killing-estimate-with-gamma:four}  
    \end{align*} 
  \end{enumerate}
\end{corollary}

\begin{proof}
  The first conclusion, follows directly from
  \eqref{linear:eq:Killing-estimate:one} and the definition of
  $\JCurrent{\TFixer, 0, \gamma\TFixer}[h]$.

  To prove \eqref{linear:eq:Killing-estimate-with-gamma:three}, observe that
  for $\Horizon=\EventHorizonFuture,\CosmologicalHorizonFuture$, 
  \begin{equation*}
    \int_{\Horizon\bigcap \widetilde{\Sigma}_{\tStar}} \JCurrent{\TFixer, 0, -\gamma\TFixer^\flat}[h]\cdot \HorizonGen_{\Horizon}
    = \int_{\Horizon\bigcap \widetilde{\Sigma}_{\tStar}} \JCurrent{\TFixer, 0, 0}[h]\cdot \HorizonGen_{\Horizon}
    .
  \end{equation*}
  Moreover, for $\gamma>0$, the flux terms vanish, since for
  $\Horizon = \EventHorizonFuture, \CosmologicalHorizonFuture$,
  \begin{equation*}
    \int_{\Horizon\bigcap\widetilde{\Sigma}_{\tStar}} \gamma |h|^2g\left(\HorizonGen_{\Horizon}, \HorizonGen_{\Horizon}\right) = 0.
  \end{equation*}
  As a result, following the proof of
  \eqref{linear:eq:Killing-estimate:three}, we see that we again have
  \begin{equation*}
    \p_{\tStar}\EnergyKill(\tStar)[h] \le \abs*{\int_{\widetilde{\Sigma}_{\tStar}}\nabla\cdot  \JCurrent{\TFixer, 0, -\gamma\TFixer^\flat}[h]\sqrt{\GInvdtdt}}.
  \end{equation*}
  Using the divergence property \eqref{linear:eq:div-them:J-K-currents}, we
  have that
  \begin{equation*}
    \nabla\cdot  \JCurrent{\TFixer, 0, -\gamma\TFixer^\flat}[h] = \Re\left[\TFixer\overline{h}\cdot
      \left((\LinearOp -\gamma)h - \SubPOp[h] - \PotentialOp[h]\right)\right]
    + \KCurrent{\TFixer, 0, -\gamma\TFixer^\flat}[h]
    +  \Re\left[\TFixer \overline{h}\cdot \gamma h\right],
  \end{equation*}
  where we now have using the construction of $\TFixer$ in Lemma
  \ref{linear:lemma:T-Fixer:construction} that for some $C>0$ independent of
  $\gamma$, 
  \begin{equation*}
    \abs*{\int_{\widetilde{\Sigma}_{\tStar}}\KCurrent{\TFixer, 0, -\gamma\TFixer^\flat}[h]\,\sqrt{\GInvdtdt}
      +  \Re\bangle*{\TFixer h, \gamma h}_{\InducedLTwo(\widetilde{\Sigma}_{\tStar})}} \le C\left(
    a\norm*{h}_{\InducedHk{1}(\widetilde{\Sigma}_{\tStar})}^2
    + \norm*{\KillT h}_{\InducedLTwo(\widetilde{\Sigma}_{\tStar})}^2
    +  a\gamma\norm*{h}_{\InducedLTwo(\widetilde{\Sigma}_{\tStar})}^2\right).
  \end{equation*}
  Repeating the proof of \eqref{linear:eq:Killing-estimate:three} as in
  Theorem \ref{linear:thm:Killing-estimate} yields the conclusions of the
  corollary. The proof of \eqref{linear:eq:Killing-estimate:four} follow
  similarly.
\end{proof}

\subsection{Redshift estimates}
\label{linear:sec:redshift-multiplier}

The redshift estimates will allow us to extend energy estimates to the
horizons. This will be useful in proving both the meromorphic
continuation and the resolvent estimates. We will use the vectorfield
$\RedShiftN$ as a multiplier to derive the redshift estimates. 
We define the redshift energy on $\Sigma_{\tStar}$,
\begin{equation}
  \label{linear:eq:redshift-energy:def}
  \RedShiftEnergy(\tStar)[{h}] = \int_{\Sigma_{\tStar}}\JCurrent{\RedShiftN, 0, 0}[{h}]\cdot n_{\Sigma_{\tStar}},
\end{equation}
and the redshift region
\begin{equation}
  \label{linear:eq:redshift-reg:def}
  \RedShiftReg:=[r_{\EventHorizonFuture}-\varepsilon_{\StaticRegionWithExtension}, r_0)\bigcup(R_0,
  r_{\CosmologicalHorizonFuture}+\varepsilon_{\StaticRegionWithExtension}].
\end{equation}

\begin{theorem}
  \label{linear:thm:redshift-energy-estimate}
  Let $\varepsilon_{\RedShiftN}>0$ be fixed. 
  \begin{enumerate}[label=(\roman{enumi})]
  \item \label{linear:item:RedShiftEstimate:one}We have that
    \begin{equation}
      \label{linear:eq:RedShiftEstimate:one}
      \RedShiftEnergy(\tStar)[{h}]
      \lesssim \norm{\nabla {h}}^2_{L^2(\Sigma_{\tStar})} + \norm*{h}_{L^2(\Sigma_{\tStar})}^2
      \lesssim \RedShiftEnergy(\tStar)[{h}].
    \end{equation}
  \item \label{linear:item:RedShiftEstimate:two}For any $\epsilon>0$, there exists some
    constant $C_{\RedShiftN}(\epsilon)>0$ such that
    \begin{align}
      \p_{\tStar} \RedShiftEnergy(\tStar)[{h}]
      \le{}& \max_{\Horizon=\EventHorizonFuture, \CosmologicalHorizonFuture}\left(
             \SHorizonControl{\LinearOp}[\Horizon]-\SurfaceGravity_{\Horizon} +\varepsilon_{\RedShiftN} + \epsilon \right)\RedShiftEnergy(\tStar)[{h}] \notag\\
           &+ C_{\RedShiftN}(\epsilon)\left(\norm{\LinearOp{ h}}_{\InducedLTwo(\Sigma_{\tStar})}^2
             + \EnergyKill(\tStar)[{h}]
             + \norm{h}_{\InducedLTwo(\Sigma_{\tStar})}^2
             \right) . \label{linear:eq:redshift:main-est}
    \end{align}
  \item \label{linear:item:RedShiftEstimate:three}If in addition, ${h}$ is
    assumed to vanish on $\EventHorizonFuture_-$ and
    $\CosmologicalHorizonFuture_-$ as defined in
    \eqref{linear:eq:extended-horizon-def}, then for any $\epsilon>0$, there
    exists $C(\epsilon)$ such that
    \begin{align}
      -\p_{\tStar} \RedShiftEnergy(\tStar)[{h}]
      \le{}& \max_{\Horizon=\EventHorizonFuture, \CosmologicalHorizonFuture}\left(
             -\SHorizonControl{\LinearOp}^*[\Horizon] + \SurfaceGravity_{\Horizon} + \varepsilon_{\RedShiftN} +\epsilon\right)\RedShiftEnergy(\tStar)[{h}]\notag \\
           &+C_{\RedShiftN}\left(\norm{\LinearOp{h}}_{\InducedLTwo(\Sigma_{\tStar})}^2
             + \EnergyKill(\tStar)[\chi_\bullet h]
             + \norm{h}_{\InducedLTwo(\Sigma_{\tStar})}^2
             \right), \label{linear:eq:redshift:dual-est}
    \end{align}
    for some $\chi_\bullet(r)$ such that 
    \begin{equation*}
      \chi_\bullet(r)=
      \begin{cases}
        0&r< r_{\EventHorizonFuture}\left(1+ \frac{r_0}{2r_{\EventHorizonFuture}}\right),\\
        1& r>r_0. 
      \end{cases}
    \end{equation*}
  \item \label{linear:item:RedShiftEstimate:four}If instead, ${h}$ is
    supported on $\RedShiftReg$ as defined in
    \eqref{linear:eq:redshift-reg:def}, then for any $\epsilon>0$, there
    exists $C(\epsilon)$ such that
    \begin{equation}
      \label{linear:eq:redshift:morawetz-est-aux}
      \p_{\tStar}\RedShiftEnergy(\tStar)[{h}] \le
      \max_{\Horizon=\EventHorizonFuture, \CosmologicalHorizonFuture}\left(
        \SHorizonControl{\LinearOp}[\Horizon] - \SurfaceGravity_{\Horizon} + \varepsilon_{\RedShiftN} + \epsilon\right)
      \RedShiftEnergy(\tStar)[{h}]
      + C(\epsilon)\norm{\LinearOp{h}}_{\InducedLTwo(\Sigma_{\tStar})}^2
      + C(\epsilon)\norm*{h}_{\InducedLTwo(\Sigma_{\tStar})}^2
      .
    \end{equation}
  \end{enumerate}
\end{theorem}

\begin{proof}
  Using the fact that
  \begin{equation*}
    g\left(\frac{n_{\Sigma}}{\sqrt{\GInvdtdt}}, \HorizonGen\right) = -1,
  \end{equation*}
  we decompose $\SubPOp[h]$ into
  \begin{equation}
    \label{linear:eq:redshift:S-decomp-transverse-tangential}
    \SubPOp = \frac{1}{\sqrt{\GInvdtdt}}s n_{\Sigma} + \SubPOp',
  \end{equation}
  where $\SubPOp'$ is tangent to the horizons, and $n_{\Sigma}$
  denotes the unit normal to $\Sigma$.  We now fix $\RedShiftN$ to be
  as constructed in Proposition \ref{linear:prop:redshift:N-construction}
  with $X=\SubPOp'$ and $\varepsilon_{\RedShiftN}$.
  
  To prove \ref{linear:item:RedShiftEstimate:one}, a similar argument as in
  the case of the $\TFixer$-energy shows that the redshift energy
  $\RedShiftEnergy(\tStar)[{h}]$ is positive. Then, since
  $\RedShiftN$ is everywhere timelike,
  $\RedShiftEnergy(\tStar)[{h}]$ controls all derivatives,
  including those transverse to the horizons, and as a result, also
  the $L^2$ norm.

  To prove \ref{linear:item:RedShiftEstimate:two}, we apply the divergence
  theorem in Corollary \ref{linear:cor:div-them:spacelike} with $X=\JCurrent{\RedShiftN, 0, 0}[{h}]$,
  \begin{align}    
      &\p_{\tStar}\int_{\Sigma_{\tStar}}\JCurrent{\RedShiftN, 0, 0}[h]\cdot n_{\Sigma_{\tStar}}
    + \int_{\EventHorizonFuture_-\bigcap \Sigma_{\tStar}}\JCurrent{\RedShiftN, 0, 0}[{h}]\cdot n_{\EventHorizonFuture_-}
    +     \int_{\CosmologicalHorizonFuture_+\bigcap \Sigma_{\tStar}}\JCurrent{\RedShiftN, 0, 0}[{h}]\cdot n_{\CosmologicalHorizonFuture_+}\notag \\
    ={}& - \int_{\Sigma_{\tStar}}\nabla\cdot \JCurrent{\RedShiftN, 0, 0}[{h}]\sqrt{\GInvdtdt}.     \label{linear:eq:RedShift:div-thm-app}    
  \end{align}
  By construction, $\RedShiftN$ is timelike and future-directed everywhere. As a
  result,
  \begin{equation*}
    \int_{\EventHorizonFuture_-}\JCurrent{\RedShiftN, 0, 0}[{h}]\cdot n_{\EventHorizonFuture_-} \ge 0,
    \quad
    \int_{\CosmologicalHorizonFuture_+}\JCurrent{\RedShiftN, 0, 0}[{h}]\cdot n_{\CosmologicalHorizonFuture_+}\ge 0.
  \end{equation*}
  It thus remains to estimate the divergence term on the right-hand
  side of \eqref{linear:eq:RedShift:div-thm-app}. Using
  (\ref{linear:eq:EMTensor:divergence-property}), 
  \begin{equation}
    \label{linear:eq:redshift:expanded-div-first-pass}
    \nabla\cdot \JCurrent{\RedShiftN, 0, 0}[{h}] =  \Re\left[
      \RedShiftN\overline{{h}}\cdot \ScalarWaveOp[g]{h}
    \right]
    +\KCurrent{\RedShiftN, 0, 0}[h].
  \end{equation}
  We consider the terms in (\ref{linear:eq:redshift:expanded-div-first-pass})
  individually.
  First, using the definition of $\LinearOp$,
  \begin{equation*}
    \begin{split}
      \Re\left[
        \RedShiftN\overline{{h}}\cdot \ScalarWaveOp[g] {h}\right]
      &= \Re\left[\RedShiftN\overline{{h}}\cdot \left(\LinearOp {h}-\SubPOp[{h}]- \PotentialOp {h}\right) \right].  
    \end{split}     
  \end{equation*}
  Then, using Cauchy-Schwartz and \eqref{linear:eq:RedShiftEstimate:one},
  \begin{equation}
    \label{linear:eq:redshift:main-est:aux1}
    \bangle*{\RedShiftN h, \LinearOp h - \PotentialOp h }_{\InducedLTwo(\Sigma)}
    \le{} \epsilon\RedShiftEnergy(\tStar)[{h}] + C(\epsilon)\left(
      \norm*{\LinearOp{h}}_{\InducedLTwo(\Sigma)}^2
      + \norm{h}_{\InducedLTwo(\Sigma)}^2
    \right).       
  \end{equation}
  
  We now control the principal bulk term,
  $\KCurrent{\RedShiftN,0,0}[{h}]-\Re\left[\RedShiftN\overline{{h}}\cdot
    \SubPOp[{h}]\right]$.  Observe that
  \begin{equation*}
    \SHorizonControl{\LinearOp}^*[\Horizon]|\xi|^2 \le \left.\Re[\xi^*\cdot s\xi]\right\vert_{\Horizon} \le \SHorizonControl{\LinearOp}[\Horizon]|\xi|^2,
  \end{equation*}
  where $s$ is as defined in
  \eqref{linear:eq:redshift:S-decomp-transverse-tangential}.  Then, from
  \eqref{linear:eq:redshift:DeformTen-redshift-control}, we have that
  \begin{equation*}
    \Re \bangle*{\RedShiftN h, \SubPOp[h]}_{\InducedLTwo(\RedShiftReg)}
    - \int_{\RedShiftReg} \KCurrent{\RedShiftN, 0, 0}[h]\,\sqrt{\GInvdtdt}
    \le \max_{\Horizon=\EventHorizonFuture, \CosmologicalHorizonFuture}\left(
    \SHorizonControl{\LinearOp}[\Horizon] - \SurfaceGravity_{\Horizon} + \varepsilon_{\RedShiftN}
    \right)\int_{\RedShiftReg}\JCurrent{\RedShiftN,0, 0}[h]\cdot n_{\RedShiftReg},
  \end{equation*}
  where we recall the definition of $\RedShiftReg$ in
  \eqref{linear:eq:redshift-reg:def}, and we have used the fact that the
  redshift bulk $\KCurrent{\RedShiftN, 0, 0}[h]$ controls a
  sufficiently large amount of the derivatives tangential to the
  horizons.  It remains then to control
  $\KCurrent{\RedShiftN,0,0}[{h}]-\Re\left[\RedShiftN\overline{{h}}\cdot\SubPOp[{h}]\right]$
  on $\Sigma\backslash\RedShiftReg$. But on
  $\Sigma\backslash\RedShiftReg$, we recall from
  \eqref{linear:eq:Killing-estimate:two} that $\EnergyKill(\tStar)[h]$
  controls the full $L^2$ norm of $\nabla h$. Thus, we have that for
  some $C>0$,
  \begin{equation}
    \label{linear:eq:redshift-main-est:aux2}
    \Re\bangle*{\RedShiftN h, \SubPOp[h]}_{\InducedLTwo(\Sigma)}
    - \int_{\Sigma} \KCurrent{\RedShiftN, 0, 0}[h]\,\sqrt{\GInvdtdt}
    \le \max_{\Horizon=\EventHorizonFuture, \CosmologicalHorizonFuture}\left(
    \SHorizonControl{\LinearOp}[\Horizon] - \SurfaceGravity_{\Horizon} + \varepsilon_{\RedShiftN}
    \right)\RedShiftEnergy(\tStar)[{h}]
    + C\EnergyKill(\tStar)[{h}].
  \end{equation}
  Combining the estimates in \eqref{linear:eq:redshift:main-est:aux1} and
  \eqref{linear:eq:redshift-main-est:aux2} and taking $C(\epsilon)$
  sufficiently large yields the conclusion. 

  To prove \ref{linear:item:RedShiftEstimate:three}, we repeat the same
  estimates but using $-\RedShiftN$ as a multiplier.  The only change
  that we need to make to the argument is the analysis on the boundary
  flux terms, which are no longer positive. However, this is easily
  handled since we assumed that $h$ vanishes exactly at
  $\EventHorizonFuture_-$ and $\CosmologicalHorizonFuture_+$.
  
  Finally, we prove \ref{linear:item:RedShiftEstimate:four}. This can be done
  repeating the proof above, but on $\RedShiftReg$ instead of
  $\Sigma$. We then observe that the only place we needed to use
  $\EnergyKill(\tStar)[{h}]$ to control derivatives was 
  on  $\Sigma\backslash\RedShiftReg$. On $\RedShiftReg$, we can use
  \eqref{linear:eq:redshift:DeformTen-redshift-control} to achieve
  arbitrary control of tangential derivatives.
\end{proof}

An important first corollary of the redshift energy estimates is an
application of Gronwall's Lemma: 
\begin{corollary}
  \label{linear:corollary:naive-energy-estimate}
  Fix $\TStar>0$ and let $\psi\in L^2(\Real^+, \Hk{1}(\Sigma))$
  with $\KillT\psi\in L^2(\Real^+, \LTwo(\Sigma))$ be a weak
  solution to the Cauchy problem
  \begin{equation}
    \label{linear:eq:naive-energy:Cauchy-prob}
    \begin{split}
      \LinearOp \psi &= f,\\
      \gamma_0(\psi)&=(\psi_0,\psi_1)
    \end{split}
  \end{equation}
  where $\LinearOp$ is the gauged linearized Einstein operator
  above,
  $\psi_0\in \Hk{1}(\Sigma), \psi_1\in \LTwo(\Sigma)$. Then
  $\psi\in C^0(\Real^+,\Hk{1}(\Sigma))$ with
  $\KillT\psi\in C^0(\Real^+, \LTwo(\Sigma))$, and there exists
  some constant $\GronwallExp$, depending on the black-hole
  background $g$, such that
  \begin{equation}
    \label{linear:eq:naive-energy-estimate}    
    \sup_{\tStar\le \TStar}e^{-\GronwallExp\tStar}\norm{\psi}_{\Hk{1}(\Sigma_{\tStar})}
    \lesssim \left(\norm{\psi_0}_{\Hk{1}(\Sigma)}
      + \norm{\psi_1}_{\LTwo(\Sigma)}
      + \int_0^\TStar e^{-\GronwallExp\tStar}\norm{f}_{\InducedLTwo(\Sigma_{\tStar})}\,d\tStar\right).            
  \end{equation}
  In particular, there exists some $\GronwallExp$ such that for
  $\Im\sigma > \GronwallExp$, 
  \begin{equation}
    \label{linear:eq:naive-energy-estimate:freq}
    \norm{u}_{\InducedHk{1}_\sigma(\Sigma)}  \lesssim \norm{\widehat{\LinearOp}(\sigma)u}_{\InducedLTwo(\Sigma)}.
  \end{equation}  
\end{corollary}

\begin{proof}
  To prove the first statement, suppose $\psi_0,\psi_1$ are in fact
  smooth, and induce a smooth solution of
  (\ref{linear:eq:naive-energy:Cauchy-prob}). Now, let $\DomainOfIntegration$
  be the spacetime subregion of $\StaticRegionWithExtension$ bounded by
  $\Sigma_{\TStar}, \Sigma_0$.

  We now repeat the proof of (\ref{linear:eq:redshift:main-est}) but applying
  the divergence theorem with $\RedShiftN$ as a multiplier over
  $\DomainOfIntegration$ instead of over a single spacelike
  hypersurface. Again observing that
  $\RedShiftN$ is timelike and future-directed
  everywhere, we have that
  \begin{equation*}
    \int_{\EventHorizonFuture_-}\JCurrent{\RedShiftN, 0, 0}[h]\cdot n_{\EventHorizonFuture_-} \ge 0,
    \quad
    \int_{\CosmologicalHorizonFuture_+}\JCurrent{\RedShiftN, 0, 0}[h]\cdot n_{\CosmologicalHorizonFuture_+}\ge 0.
  \end{equation*}
  Then observe that
  \begin{equation*}
    \int_{\DomainOfIntegration}\abs*{\nabla\cdot\JCurrent{\RedShiftN, 0, 0}[h]}
    \le  \norm{\RedShiftN h}_{L^2(\DomainOfIntegration)}^2
    + \norm{\ScalarWaveOp[g]h}_{L^2(\DomainOfIntegration)}^2
    + \int_{\DomainOfIntegration} \abs*{\KCurrent{\RedShiftN, 0, 0}[h]}.
  \end{equation*}
  Then recalling \eqref{linear:eq:RedShiftEstimate:one} and
  \eqref{linear:eq:redshift:main-est}, we have that there exists some
  $C, \GronwallExp>0$ such that
  \begin{equation}
    \label{linear:eq:naive-energy:aux1}
    \RedShiftEnergy(\TStar)[h]
    \le{} \RedShiftEnergy(0)[h]
    + C \norm{\LinearOp {h}}_{L^2(\DomainOfIntegration)}^2 
    + \GronwallExp \int_0^\TStar \RedShiftEnergy(\tStar)[h]\,d\tStar.    
  \end{equation}
  Applying Gronwall's inequality then immediately yields
  \eqref{linear:eq:naive-energy-estimate}.  
  
  To prove the Laplace-transformed statement in
  \eqref{linear:eq:naive-energy-estimate:freq}, we take a $\p_{\tStar}$
  derivative of equation \eqref{linear:eq:naive-energy:aux1},
  \begin{equation}
    \label{linear:eq:naieve-energy:aux-spacelike}
    \p_{\tStar}\RedShiftEnergy(\tStar)[h] - \GronwallExp \RedShiftEnergy(\tStar)[h]
    \lesssim \norm{\LinearOp {h}}_{\InducedLTwo(\DomainOfIntegration)}^2.      
  \end{equation}
  Considering \eqref{linear:eq:naieve-energy:aux-spacelike} in the specific
  case where $h(\tStar, x) = e^{-\ImagUnit\sigma\tStar} u(x)$, we have that
  \begin{equation*}
    \left(\Im \sigma - \GronwallExp\right)\RedShiftEnergy(\tStar)[e^{-\ImagUnit\sigma\tStar}u]
    \lesssim
    C \norm{\LinearOp e^{-\ImagUnit\sigma\tStar}u}_{\InducedLTwo(\Sigma_{\tStar})}. 
  \end{equation*}
  For $\Im\sigma>\GronwallExp$, we see that the left-hand side is
  positive. Multiplying both sides by $e^{2\Im\sigma\tStar}$ to remove any
  $\tStar$-dependency, we then have that
  \begin{equation*}
    \left(\Im \sigma - \GronwallExp\right)\left(
      \norm*{u}_{H^1(\Sigma_{\tStar})} + \norm*{\sigma u}_{L^2(\Sigma_{\tStar})}
    \right)
    \lesssim
    \norm{\widehat{\LinearOp}(\sigma)u}_{\InducedLTwo(\Sigma_{\tStar})}
  \end{equation*}
  as desired.   
\end{proof}

\begin{remark}
  We remark that the second corollary above is exactly the statement
  that the resolvent of the operator $\LinearOp$ is bounded in
  some upper-half space. Thus, from the redshift energy alone we are
  able to deduce a large resonance-free region. 
\end{remark}

\begin{corollary}
  \label{linear:coro:redshift-energy-estimate-with-gamma}
  Fix some $\varepsilon_{\RedShiftN}>0$, and define
  \begin{equation}
    \label{linear:eq:RedShiftEnergy-with-gamma:def}
    \RedShiftEnergy_\gamma(\tStar)[h] = \int_{\Sigma_{\tStar}}\JCurrent{\RedShiftN,0, -\gamma\RedShiftN^\flat}[h]\cdot n_{\Sigma_{\tStar}},
  \end{equation}
  where we recall the definition of $\JCurrent{X,q,m}[h]$ in
  \eqref{linear:eq:div-them:J-K-currents}.  
  Then we have the following estimates. 
  \begin{enumerate}[label=(\roman{enumi})]
  \item \label{linear:item:RedShiftEstimate-with-gamma:one}We have that
    \begin{equation}
      \label{linear:eq:RedShiftEstimate-with-gamma:one}
      \RedShiftEnergy_\gamma(\tStar)[{h}]
      \lesssim \norm{h}^2_{H^1(\Sigma_{\tStar})}
      + \norm{\KillT h}_{L^2(\Sigma_{\tStar})}^2
      +  \gamma\norm*{h}^2_{L^2(\Sigma_{\tStar})}
      \lesssim \RedShiftEnergy_\gamma(\tStar)[{h}].
    \end{equation}
  \item \label{linear:item:RedShiftEstimate-with-gamma:two}Then for any
    $\epsilon>0$, there exists some
    constant $C_{\RedShiftN}(\epsilon)>0$ such that
    \begin{align}
      \p_{\tStar} \RedShiftEnergy_\gamma(\tStar)[{h}]
      \le{}& \max_{\Horizon=\EventHorizonFuture, \CosmologicalHorizonFuture}\left(
             \SHorizonControl{\LinearOp}[\Horizon]
             -\SurfaceGravity_{\Horizon}
             +\varepsilon_{\RedShiftN}
             +\epsilon
             \right)\RedShiftEnergy_\gamma(\tStar)[{h}] \notag\\
           &+ C_{\RedShiftN}(\epsilon)\left(\norm{(\LinearOp - \gamma){ h}}_{\InducedLTwo(\Sigma_{\tStar})}^2
             + \EnergyKill_\gamma(\tStar)[{h}]
             + \norm{h}_{\InducedLTwo(\Sigma_{\tStar})}^2
             \right) . \label{linear:eq:redshift-with-gamma:main-est}
    \end{align}
  \item \label{linear:item:RedShiftEstimate-with-gamma:three}If in addition,
    ${h}$ is assumed to vanish on $\EventHorizonFuture_-$ and
    $\CosmologicalHorizonFuture_+$, then for any $\epsilon>0$, there
    exists some constant $C_{\RedShiftN}(\epsilon)>0$ such that
    \begin{align}
      -\p_{\tStar} \RedShiftEnergy_\gamma(\tStar)[{h}]
      \le{}& \max_{\Horizon=\EventHorizonFuture, \CosmologicalHorizonFuture}\left(-\SHorizonControl{\LinearOp}^*[\Horizon] + \SurfaceGravity_{\Horizon} + \varepsilon_{\RedShiftN} + \epsilon\right)\RedShiftEnergy_\gamma(\tStar)[{h}]\notag \\
           &+C_{\RedShiftN}(\epsilon)\left(\norm{(\LinearOp - \gamma){h}}_{\InducedLTwo(\Sigma_{\tStar})}^2
             + \EnergyKill_\gamma(\tStar)[h]
             + \norm{h}_{\InducedLTwo(\Sigma_{\tStar})}^2
             \right). \label{linear:eq:redshift-with-gamma:dual-est}
    \end{align}
  \end{enumerate}
\end{corollary}
\begin{proof}
  We again fix $\RedShiftN$ as constructed in Proposition
  \ref{linear:prop:redshift:N-construction} with $\varepsilon_{\RedShiftN}$,
  and $X = \SubPOp'$.  
  The first conclusion in \eqref{linear:eq:RedShiftEstimate-with-gamma:one}
  follows directly from \eqref{linear:eq:RedShiftEnergy-with-gamma:def} and
  \eqref{linear:eq:RedShiftEstimate:one}. 

  To prove \eqref{linear:eq:redshift-with-gamma:main-est}, observe that for
  $\Horizon = \EventHorizonFuture_-, \CosmologicalHorizonFuture_+$, we
  have that
  \begin{equation*}
    \int_{\Horizon}\JCurrent{\RedShiftN, 0, -\gamma\RedShiftN^\flat}[h]\cdot n_{\Horizon}
    = \int_{\Horizon}\JCurrent{\RedShiftN, 0, 0}[h]\cdot n_{\Horizon}
    - \int_{\Horizon}\gamma|h|^2g(\RedShiftN, n_{\Horizon}). 
  \end{equation*}
  Moreover, for $\gamma>0$, we have that for
  $\Horizon = \EventHorizonFuture_-, \CosmologicalHorizonFuture_+$,
  \begin{equation*}
     \int_{\Horizon}\gamma|h|^2g(\RedShiftN, n_{\Horizon})< 0.
   \end{equation*}
   As a result, following the proof of \eqref{linear:eq:redshift:main-est},
   we have that
   \begin{equation*}
     \p_{\tStar}\RedShiftEnergy_\gamma(\tStar)[h]\le \int_{\Sigma_{\tStar}}\nabla \cdot \JCurrent{\RedShiftN, 0, -\gamma\RedShiftN^\flat}[h]\,\sqrt{\GInvdtdt}.
   \end{equation*}
   Using the divergence property \eqref{linear:eq:div-them:J-K-currents}, we
   have that
   \begin{equation*}
     \nabla\cdot \JCurrent{\RedShiftN, 0, -\gamma\RedShiftN^\flat}[h] = \Re\left[
       \RedShiftN \overline{h}\cdot \left((\LinearOp - \gamma)h - \SubPOp[h] - \PotentialOp[h]\right)
     \right]
     + \KCurrent{\RedShiftN, 0, -\gamma\RedShiftN^\flat}[h]
     + \Re\left[\RedShiftN\overline{h}\cdot \gamma h \right],
   \end{equation*}
   where we recall that
   \begin{equation*}
     \KCurrent{\RedShiftN, 0, -\gamma\RedShiftN^\flat}[h]
     + \Re\left[\RedShiftN\overline{h}\cdot \gamma h \right]
     = \KCurrent{\RedShiftN, 0, 0}[h] - \frac{1}{2}\gamma\nabla\cdot \RedShiftN |h|^2. 
   \end{equation*}
   From Proposition \ref{linear:prop:redshift:N-construction}, we have that
   $\nabla\cdot\RedShiftN<0$ in a small neighborhood of the
   horizons. Thus, we have that there exists some $C>0$ such that
   \begin{equation*}
     \int_{\Sigma_{\tStar}}\KCurrent{\RedShiftN, 0, -\gamma\RedShiftN^\flat}[h]\,\sqrt{\GInvdtdt}
       + \Re\bangle*{\RedShiftN h,\gamma h}_{\InducedLTwo(\Sigma_{\tStar})} 
     > \int_{\Sigma_{\tStar}}\KCurrent{\RedShiftN, 0, 0}[h]\,\sqrt{\GInvdtdt}
     - C \EnergyKill_{\gamma}(\tStar)[h]
   \end{equation*}
   Then, repeating the proof of
   \eqref{linear:eq:redshift:main-est} yields
   \eqref{linear:eq:redshift-with-gamma:main-est} directly. 
\end{proof}

\subsection{$\RedShiftN$ as a commutator}
\label{linear:sec:redshift-commutator}

We use the redshift vectorfield $\RedShiftN$ not only as a multiplier
but also as a commutator in order to derive higher-order energy
estimates. This will improve the domain on which we define quasinormal
modes.  

\begin{theorem}
  \label{linear:thm:redshift-commutation:main}
  Consider some ${h}\in C^\infty_0(\StaticRegionWithExtension,
  \Complex^D)$. Define $f$ as
  \begin{equation}
    \label{linear:eq:redshift-commutation:orig-eq}
    \LinearOp {h} =f.
  \end{equation}
  Then,
  \begin{enumerate}
  \item There exists a finite set of vectorfields
    $\curlyBrace*{\RedShiftK_i}_{i=1}^N$, which span the set of
    smooth vectorfields over $\StaticRegionWithExtension$, such that
    for $\mathbf{h} = ({h}, \RedShiftK_1{h}, \cdots, \RedShiftK_N{h})$,
    $\mathbf{h}$ satisfies
    \begin{equation}
      \label{linear:eq:enhanced-redshift:new-system}
      \bL\mathbf{h} = \blittleF 
    \end{equation}
    where $\bL$ is a strongly hyperbolic operator constructed from $L$,
    which acts on vectors in $\Complex^{D'}$, $D' = N(D+1)$, and $f'$
    is defined by
    \begin{equation}
      \label{linear:eq:enhanced-redshift:new-forcing}
      \blittleF_0 = f,\quad \blittleF_i=\RedShiftK_if.
    \end{equation}
    We also have
    \begin{equation}
      \label{linear:eq:enhanced-redshift:improvement}
      \SHorizonControl{\bL}[\mathcal{H}] = \SHorizonControl{\LinearOp}[\mathcal{H}] - 2\SurfaceGravity_{\mathcal{H}},
    \end{equation}
    where $\mathcal{H}$ is either $\EventHorizonFuture$ or
    $\CosmologicalHorizonFuture$. 
  \item Conversely, suppose that $\mathbf{h}\in
    C^\infty_0(\StaticRegionWithExtension; \Complex^{D'})$
    solves
    \eqref{linear:eq:enhanced-redshift:new-system} with $\blittleF$ having the form
    \eqref{linear:eq:enhanced-redshift:new-forcing} for some $f$. Then,
    defining ${h}:= \mathbf{h}_0$, and $\delta\mathbf{h}= (\mathbf{h}_i - 
    \RedShiftK_i{h})_{i=1,\cdots N}$, we have that $\delta\mathbf{h}$ satisfies
    \begin{equation*}
      \bL'\delta\mathbf{h} = 0
    \end{equation*}
    for a strongly hyperbolic operator $\bL'$ acting on vectors of
    dimension $D'N$ such that
    $\SHorizonControl{\bL'}= \SHorizonControl{\bL}$. If the initial
    conditions imply that
    $\left.\delta\mathbf{h}\right\vert_{\Sigma_0}=0$, then actually we
    have that ${h}$ solves $\LinearOp{h} = f$.
  \end{enumerate}
\end{theorem}

It is a simple calculation to verify the following commutation lemma.
\begin{lemma}
  Let $K$ be a smooth vectorfield on
  $\StaticRegionWithExtension$. Then, for sufficiently smooth ${h}$,
  \begin{equation*}
    [\ScalarWaveOp[g],K]{h}
    = 2\nabla_\mu\left(\DeformationTensor{K}^{\mu\nu}\nabla_\nu{h} \right)
    - \nabla_\nu\left(\Trace_g \DeformationTensor[]{K}\right)\nabla^\nu {h}. 
  \end{equation*}
\end{lemma}

We now have the tools necessary to prove the main theorem of this
section.
\begin{proof}[Proof of Theorem \ref{linear:thm:redshift-commutation:main}]
  Let us first prove the first part of the theorem. The main idea here
  will be to commute the equation
  \begin{equation}
    \label{linear:eq:MainEqn}
    \LinearOp{h} = f
  \end{equation}
  with the set of vectorfields $\RedShiftK_i$ constructed in Lemma
  \ref{linear:lemma:enhanced-redshift:Ka-construction}. Then, by rewriting
  the resulting equation as a system, we will analyze the
  subprincipal operator at $\EventHorizonFuture$ and
  $\CosmologicalHorizonFuture$ to verify that it satisfies 
  (\ref{linear:eq:enhanced-redshift:improvement}).

  \noindent{\textbf{Step 1: Commuting the equation.}}
  Commuting equation (\ref{linear:eq:MainEqn}) with the $\RedShiftK_i$ vectorfields, we have
  \begin{equation}
    \label{linear:eq:enhanced-redshift:eq-commute-K} 
    \begin{split}
      \RedShiftK_i f
      ={}&\ScalarWaveOp[g]\RedShiftK_i{h}
      + \SubPOp [\RedShiftK_i{h}]
      + \PotentialOp\RedShiftK_i{h}
      + [\RedShiftK_i, \SubPOp]{h}
      +[\RedShiftK_i, \PotentialOp]{h} \\
      &+ 2\nabla_\mu \left(\DeformationTensor{\RedShiftK_i}^{\mu\nu}\nabla_\nu {h} \right)
      - \nabla_\nu \left(\Trace_g \DeformationTensor{\RedShiftK_i}\right)\nabla^\nu {h}.
    \end{split}    
  \end{equation}
  Using (\ref{linear:eq:enhanced-redshift:Deform-Tens-decomp}), we can write
  that
  \begin{equation*}
    \begin{split}
      \nabla_\mu \left(\DeformationTensor{\RedShiftK_i}^{\mu\nu}\nabla_\nu {h}\right) &= \nabla_\mu \left(f^{jk}_i(\RedShiftK_j)^\mu \RedShiftK_k{h} \right)\\
      &= \RedShiftK_jf^{jk}_i\RedShiftK_k{h} + f^{jk}_i(\nabla \cdot \RedShiftK_j)\RedShiftK_k {h} + f^{jk}_i\RedShiftK_j\RedShiftK_k{h},
    \end{split}
  \end{equation*}
  where repeated $j, k$ indices are summed over $1, \cdots, N$.
  Applying (\ref{linear:eq:enhanced-redshift:Deform-Tens-decomp}) again,
  \begin{equation*}
    \Trace_g \DeformationTensor[]{\RedShiftK_i} = \nabla\cdot \RedShiftK_i
    = f^{jk}_ig(\RedShiftK_j, \RedShiftK_k). 
  \end{equation*}
  Thus, we can rewrite
  \begin{equation*}
    \RedShiftK_if
    = \ScalarWaveOp[g]\RedShiftK_i{h}
    +\SubPOp[\RedShiftK_i{h}]
    + \PotentialOp \RedShiftK_i{h}
    - {\SAux_i}^j \RedShiftK_j{h}
    - \SAux_i {h}
    + [\RedShiftK_i, \PotentialOp]{h} ,
  \end{equation*}
  where $\tensor[]{\SAux}{_i^j}$ and  $\SAux_i$ are smooth
  vectorfields on $\StaticRegionWithExtension$ given by 
  \begin{equation*}
    \begin{split}
      \tensor[]{\SAux}{_i^j} &= 2f^{jk}_i\RedShiftK_k,\\
      \SAux_i &= \left(\RedShiftK_jf^{jk}_i + f^{jk}_i\nabla\cdot \RedShiftK_j\right)\RedShiftK_k - \nabla(\nabla\cdot \RedShiftK_i) + [\RedShiftK_i,\SubPOp].
    \end{split}
  \end{equation*}
  To verify the improvement that we gain from commuting, i.e. that
  $\SHorizonControl{\bL}[\mathcal{H}]
  =\SHorizonControl{\LinearOp}[\mathcal{H}]-2\SurfaceGravity_{\mathcal{H}}$,
  it is necessary to analyze the first-order terms
  $\tensor[]{\SAux}{_i^j}$ and $\SAux_i$ at the horizons.  Since we
  are working in a system after a commutation, it will turn out that
  despite containing first-order derivatives of ${h}$, $\SAux_i$ will
  turn out to be part of the zero-order potential operator in the new
  system, while $\tensor[]{\SAux}{_i^j}$ will contribute be the main
  subprincipal term of the new system. To this end, let us first
  consider the term $\tensor[]{\SAux}{_i^j}$. Define
  $\widetilde{\SurfaceGravity} := \widetilde{\SurfaceGravity}(r)$ to be a smooth cutoff
  function function such that
  \begin{equation}
    \label{linear:eq:redshift-commute:surface-grav-cutoff-def}
    \widetilde{\SurfaceGravity} =
    \begin{cases}
      \SurfaceGravity_{\EventHorizonFuture} &\text{ near } \EventHorizonFuture, \\
      0 &\text{ away from }\EventHorizonFuture, \CosmologicalHorizonFuture,\\
      \SurfaceGravity_{\CosmologicalHorizonFuture} &\text{ near }\CosmologicalHorizonFuture.
    \end{cases}
  \end{equation}
  Adding and subtracting
  $2\widetilde{\SurfaceGravity}\RedShiftK_1\RedShiftK_i{h}$,
  \begin{equation*}
    \tensor{\SAux}{_i^j}\RedShiftK_j {h} = 2\iota \widetilde{\SurfaceGravity} \RedShiftK_1 \RedShiftK_i{h}
    + 2f^{jk}_i\RedShiftK_j\RedShiftK_k{h} - 2\widetilde{\SurfaceGravity} \RedShiftK_1\RedShiftK_i{h}.     
  \end{equation*}
  We then use $[\RedShiftK_i,\RedShiftK_j] = \RedShiftK_i\RedShiftK_j - \RedShiftK_j\RedShiftK_i$ to write
  \begin{equation*}
    \begin{split}
      f^{jk}_i\RedShiftK_j\RedShiftK_k{h} - \widetilde{\SurfaceGravity}
      \RedShiftK_1\RedShiftK_i{h} ={}& \left(f_i^{jk}
        - \widetilde{\SurfaceGravity}\tensor[]{\delta}{_1^j}\tensor[]{\delta}{_i^k} \right)\RedShiftK_j\RedShiftK_k{h}
      + \tensor[]{\alpha}{_i^j}(\RedShiftK_j\RedShiftK_1{h}
      - \RedShiftK_1\RedShiftK_j{h}
      + [\RedShiftK_1,\RedShiftK_j]{h})\\
      ={}&\left(f_i^{jk} - \widetilde{\SurfaceGravity}
        \tensor[]{\delta}{_1^j}\tensor[]{\delta}{_i^k} +
        \tensor[]{\alpha}{_i^j}\tensor[]{\delta}{_1^k}
        - \tensor[]{\delta}{_1^j}\tensor[]{\alpha}{_i^k}\right)\RedShiftK_j\RedShiftK_k{h}
      + \tensor[]{\alpha}{_i^j}[\RedShiftK_1,\RedShiftK_j]{h}
    \end{split}
  \end{equation*}
  for any $\tensor[]{\alpha}{_i^j}$. We then define
  $\tensor[]{\alpha}{_i^j}$ such that on $\EventHorizonFuture$ and
  $\CosmologicalHorizonFuture$, 
  \begin{equation*}
    \tensor[]{\alpha}{_i^j}= 
    \begin{cases}
      f^{1j}_i - \widetilde{\SurfaceGravity} \tensor[]{\delta}{_i^j} & j\neq 1,\\
      0 & j=1.
    \end{cases}
  \end{equation*}
  Next, we define
  \begin{equation*}
    \tensor[]{\SAux}{_i^{jk}} = f^{jk}_i - \widetilde{\SurfaceGravity} \delta_1^j \delta_i^k + \alpha^j_i\delta_1^k -\delta_1^j \alpha_i^k,
  \end{equation*}
  which has the crucial property that $\tensor[]{\SAux}{_i^{1j}}$ vanishes on both
  $\EventHorizonFuture$ and $\CosmologicalHorizonFuture$, and  $\tensor[]{\VAux}{_i^k}\in
  C^\infty(\StaticRegionWithExtension)$ so that 
  \begin{equation*}
    \SAux_i + 2\tensor[]{\alpha}{_i^j}[\RedShiftK_1, \RedShiftK_j]{h} = \tensor[]{\VAux}{_i^j}\RedShiftK_j{h}. 
  \end{equation*}
  We can now rewrite (\ref{linear:eq:enhanced-redshift:eq-commute-K}) in the
  following manner:
  \begin{equation}
    \label{linear:eq:enhanced-redshift:eq-commute-K-a}
    \RedShiftK_i f = \ScalarWaveOp[g] \RedShiftK_i{h}
    + \left(
      \tensor[]{\delta}{_i^j}(\SubPOp - 2\widetilde{\SurfaceGravity}\RedShiftK_i)
      - 2 \tensor[]{\SAux}{_i^{jk}}\RedShiftK_k
    \right)\RedShiftK_j{h}
    + \left(\PotentialOp\tensor[]{\delta}{_i^j}
      - \tensor[]{\VAux}{_i^j}\right)\RedShiftK_j{h}
    + [\RedShiftK_i,\PotentialOp]{h},
  \end{equation}
  which can be written as a system of equations
  for $\mathbf{h}_i$, where $\mathbf{h}_i = \RedShiftK_i{h}$ for $i>0$, and
  $\mathbf{h}_0={h}$. 
  \begin{equation*}
    K_if = \ScalarWaveOp[g]{}\mathbf{h}_i + \tensor[]{\bS}{_i^j}\mathbf{h}_j + \tensor[]{\bV}{_i^j}\mathbf{h}_j 
  \end{equation*}
  where the Einstein summation notation denotes summation over $j=1, \cdots, N$, and
  \begin{equation*}
    \begin{split}
      \tensor[]{\bS}{_i^j} &= (\SubPOp - 2 \widetilde{\SurfaceGravity} \RedShiftK_1)\tensor[]{\delta}{_i^j}
                             - 2 \tensor[]{\SAux}{_i^{jk}}\RedShiftK_k,\quad j\neq 0\\
      \tensor[]{\bS}{_i^0} &=0,\\
      \tensor[]{\bV}{_i^j} &= \PotentialOp\tensor[]{\delta}{_i^j} - \tensor[]{\VAux}{_i^j},\quad j\neq 0\\
      \tensor[]{\bV}{_i^0} &= [\RedShiftK_i, \PotentialOp].
    \end{split}
  \end{equation*}
  We now have equations for $\mathbf{h}_i, i\neq 0$. It remains to derive
  an equation for $\mathbf{h}_0={h}$.  In particular, we cannot use the
  original equation $\bL\mathbf{h}_0 = \LinearOp \mathbf{h}_0$ since we want
  to derive an equation consistent with
  $\SHorizonControl{\bL}[\mathcal{H}] =\SHorizonControl{\LinearOp
  }[\mathcal{H}] - 2\SurfaceGravity_{\mathcal{H}}$. Instead, let us
  rewrite the main equation (\ref{linear:eq:MainEqn}) using the fact that
  $\mathbf{h}_1 = \RedShiftK_1\mathbf{h}_0$,
  \begin{equation}
    \label{linear:eq:enhanced-redshift:eq-commute-K-0}
    f = \ScalarWaveOp[g]\mathbf{h}_0 + (\SubPOp - 2\widetilde{\SurfaceGravity} \RedShiftK_1)\mathbf{h}_0 + \PotentialOp\mathbf{h}_0 + 2\widetilde{\SurfaceGravity} \mathbf{h}_1 .
  \end{equation}
  We then define
  \begin{equation*}
    \begin{split}
      \tensor[]{\bS}{_0^j}&=(\SubPOp - 2\widetilde{\SurfaceGravity} \RedShiftK_1)\tensor[]{\delta}{_0^j},\\
      \tensor[]{\bV}{_0^j} &= \PotentialOp\tensor[]{\delta}{_0^j}
      + 2\widetilde{\SurfaceGravity}\tensor[]{\delta}{_1^j},\\
      \blittleF_0&=f,\\
      \blittleF_i &= \RedShiftK_if.
    \end{split}
  \end{equation*}
  With these definitions, we see that we can combine the equations
  (\ref{linear:eq:enhanced-redshift:eq-commute-K-a}) and
  (\ref{linear:eq:enhanced-redshift:eq-commute-K-0}) to write
  \begin{equation*}
    \bL\mathbf{h}_i=\blittleF_i,\qquad \bL\mathbf{h}_i:= \Box_g\mathbf{h}_i + \tensor[]{\bS}{_i^j}\mathbf{h}_j + \tensor[]{\bV}{_i^j}\mathbf{h}_j. 
  \end{equation*}
  Then, since we have that $\tensor[]{\SAux}{_a^{1b}}$ vanishes on the horizons, we see
  that
  \begin{equation*}
    \left.g(\HorizonGen_{\EventHorizonFuture}, \tensor[]{\bS}{_i^j})\right\vert_{\EventHorizonFuture}
    = \left(
      g(\HorizonGen_{\EventHorizonFuture}, \SubPOp) + 2\SurfaceGravity_{\EventHorizonFuture}
    \right)\tensor[]{\delta}{_i^j},
    \qquad
    \left.g(\HorizonGen_{\CosmologicalHorizonFuture}, \tensor[]{\bS}{_i^j})\right\vert_{\CosmologicalHorizonFuture}
    = \left(
      g(\HorizonGen_{\CosmologicalHorizonFuture}, \SubPOp) + 2\SurfaceGravity_{\CosmologicalHorizonFuture}
    \right)\tensor[]{\delta}{_i^j}.
  \end{equation*}
  This directly implies that
  \begin{equation*}
    \begin{split}
      \SHorizonControl{\bL}[\EventHorizonFuture] &= \SHorizonControl{L}[\EventHorizonFuture] - 2\SurfaceGravity_{\EventHorizonFuture},\\
      \SHorizonControl{\bL}[\CosmologicalHorizonFuture] &= \SHorizonControl{L}[\CosmologicalHorizonFuture] - 2\SurfaceGravity_{\CosmologicalHorizonFuture}.
    \end{split}
  \end{equation*}

  \noindent{\textbf{Step 2: Verifying the propagation of the constraint.}}
  We now prove the second part of the theorem. The goal in this part
  of the theorem will be to show that the extended system of equations
  propagates the constraints $\RedShiftK_i\mathbf{h}_0 - \mathbf{h}_i$. 
  
  Define
  \begin{equation*}
    \delta\mathbf{h}_i = \RedShiftK_i\mathbf{h}_0 - \mathbf{h}_i.
  \end{equation*}
  Then we can rewrite the equation for $\mathbf{h}_0$ in equation
  (\ref{linear:eq:enhanced-redshift:eq-commute-K-0}), as
  \begin{equation*}
    \LinearOp\mathbf{h}_0 - 2\widetilde{\SurfaceGravity} (\delta\mathbf{h}_1) = \blittleF_0.
  \end{equation*}
  Commuting this equation with $\RedShiftK_i$ and repeating the
  algebra above leading to equation (\ref{linear:eq:enhanced-redshift:eq-commute-K-a})
  we recover
  \begin{equation*}
    \RedShiftK_i\blittleF_0
    = \ScalarWaveOp[g](\RedShiftK_i\mathbf{h}_0)
    + \tensor[]{\bV}{_i^j}\RedShiftK_j\mathbf{h}_0
    + \tensor[]{\bS}{_i^j}\RedShiftK_j\mathbf{h}_0
    + \tensor[]{\bS}{_i^0}\mathbf{h}_0
    + \tensor[]{\bV}{_i^0}\mathbf{h}_0
    - 2\widetilde{\SurfaceGravity}\RedShiftK_i(\delta\mathbf{h}_1). 
  \end{equation*}
  If  we have that 
  \begin{equation*}
    0 =\RedShiftK_i\blittleF_0 - \blittleF_i,\quad 1\le i\le N,
  \end{equation*}
  then, 
  \begin{equation*}
    0 = \ScalarWaveOp[g](\delta\mathbf{h}_i)
    + \tensor[]{\bS}{_i^j}(\delta\mathbf{h}_j)
    + \tensor[]{\bV}{_i^j}(\delta\mathbf{h}_j)
    - 2\RedShiftK_i(\widetilde{\SurfaceGravity} \delta\mathbf{h}_1).
  \end{equation*}
  This can be rewritten as
  \begin{equation*}
    \bL'(\delta\mathbf{h}_i)=0, \qquad
    \bL'(\delta\mathbf{h}_i)
    :=\ScalarWaveOp[g](\delta\mathbf{h}_i)
    + \tensor[]{{\bS'}}{_i^j}(\delta\mathbf{h}_j)
    + \tensor[]{{\bV'}}{_i^j}(\delta\mathbf{h}_j) ,
  \end{equation*}
  where $\bL'$ is a strongly hyperbolic operator with
  \begin{align*}
    \tensor[]{{\bS'}}{_i^j}(\delta\mathbf{h}_j)
    &:=\tensor[]{\bS}{_i^j}(\delta\mathbf{h}_j)
      - 2\widetilde{\SurfaceGravity}\RedShiftK_i(\delta\mathbf{h}_1),\\
    \tensor[]{{\bV'}}{_i^j}(\delta\mathbf{h}_j)
    &:= \tensor[]{\bV}{_i^j}(\delta\mathbf{h}_j)
      - 2\RedShiftK_i\widetilde{\SurfaceGravity}(\delta\mathbf{h}_1).
  \end{align*}
  From the hyperbolic nature of this system, it is clear that having
  initial data $\delta\mathbf{h}\vert_{\Sigma_0}=0$ implies that
  $\delta\mathbf{h}=0$ identically and thus the extended system
  \eqref{linear:eq:enhanced-redshift:new-system} reduces back to the original
  uncommuted equation \eqref{linear:eq:redshift-commutation:orig-eq}.  It
  remains to check that
  $\SHorizonControl{\bL'}[\mathcal{H}] =
  \SHorizonControl{\bL}[\mathcal{H}]$. To this end, we can evaluate
  \begin{align*}
    \left.g(\HorizonGen, \tensor[]{{\bS'}}{_i^j})\right\vert_{\mathcal{H}}
    &= \left.g(\HorizonGen, \tensor[]{\bS}{_i^j})\right\vert_{\mathcal{H}}
      + 2\SurfaceGravity_{\mathcal{H}} \tensor[]{\delta}{_i^1}\tensor[]{\delta}{_1^j},\\
    &= \left(\left.g(\HorizonGen, \SubPOp)\right\vert_{\mathcal{H}}
      + 2\SurfaceGravity_{\mathcal{H}}\right)\tensor[]{\delta}{_i^j}
      + 2\SurfaceGravity_{\mathcal{H}} \tensor[]{\delta}{_i^1}\tensor[]{\delta}{_1^j},
  \end{align*}
  so $\SHorizonControl{\bL'}[\mathcal{H}] =
  \SHorizonControl{\bL}[\mathcal{H}].$
\end{proof}

We thus have the following higher-regularity equivalent of Corollary
\ref{linear:corollary:naive-energy-estimate}:
\begin{corollary}
  \label{linear:corollary:naive-energy-estimate:higher-order}
  Let $\LinearOp$ be a strongly hyperbolic operator on a
  slowly-rotating \KdS{} background, and suppose that
  $\psi\in L^2(\Real_+,\Hk{1}(\Sigma))$ satisfies
  $\KillT\psi\in L^2(\Real_+, \LTwo(\Sigma))$ and is a weak solution
  of the Cauchy problem
  \begin{equation*}
    \begin{split}
      \LinearOp \psi &= f,\\
      \gamma_0(\psi) &=(\psi_0, \psi_1),
    \end{split}
  \end{equation*}
  where $(\psi_0, \psi_1)\in \LSolHk{k}(\Sigma)$, and $f\in
  H^k(\StaticRegionWithExtension)$. 
  \begin{enumerate}
  \item  Then, for any $\TStar>0$ and $k\ge 1$, $\psi$ satisfies the energy estimate
    \begin{equation*}
      \sup_{\tStar \le \TStar}e^{-\GronwallExp\tStar}\norm*{\psi}_{\HkWithT{k}(\Sigma)}
      \le C \norm*{(\psi_0, \psi_1)}_{\LSolHk{k}(\Sigma_0)}
      + \int_{0}^{\TStar}e^{-\GronwallExp\tStar}\norm{f}_{\HkWithT{k-1}(\Sigma_{\tStar})}\,d\tStar
    \end{equation*}
    for some constants $C, \GronwallExp$ depending on $g$, $\LinearOp$, and $k$.
  \item On the Laplace-transformed side, there exists $\GronwallExp$
    and $C$ such that for $\Im\sigma > \GronwallExp$,
    \begin{equation*}
      \norm{u}_{\InducedHk{k}_\sigma(\Sigma)}\lesssim \norm{\widehat{\LinearOp}(\sigma)u}_{\InducedHk{k-1}_\sigma(\Sigma)}.
    \end{equation*}
  \end{enumerate}
\end{corollary}

\begin{proof}
  We prove the result by commuting the equation with the vectorfields
  $\RedShiftK_i$ a sufficient number of times and then applying
  Corollary \ref{linear:corollary:naive-energy-estimate}. 
\end{proof}

We have the following regularity statement.
\begin{corollary}
  \label{linear:coro:redshift-regularity}
  Let $\LinearOp$ be a strongly hyperbolic operator on some
  slowly-rotating \KdS{} background $g$, and let $u \in \InducedHk{1}(\Sigma)$
  be a weak solution of
  \begin{equation}
    \label{linear:eq:KdS-elliptic-regularity:Cauchy-system}
    \widehat{\LinearOp}(\sigma)u = f,
  \end{equation}
  where
  $\Im\sigma > \frac{1}{2} \max_{\Horizon=\EventHorizonFuture,
    \CosmologicalHorizonFuture}\left(\SHorizonControl{\LinearOp}[\Horizon]
    - \SurfaceGravity_{\Horizon}\right)$. Then if
  $f\in C^\infty_0(\Sigma)$, we must also have that
  $u\in C^{\infty}_0(\Sigma)$.
\end{corollary}
\begin{proof}
  Define ${h}=e^{-\ImagUnit\sigma\tStar}u$, which by 
  hypothesis satisfies
  \begin{equation*}
    \LinearOp {h} = e^{-\ImagUnit\sigma\tStar}f. 
  \end{equation*}
  Now fix $\epsilon >0$ and recall from Theorem
  \ref{linear:thm:redshift-energy-estimate} part \ref{linear:item:RedShiftEstimate:two}
  that
  \begin{align*}
    2\left(\Im\sigma- \max_{\Horizon=\EventHorizonFuture, \CosmologicalHorizonFuture}(\SHorizonControl{\LinearOp}[\Horizon] - \SurfaceGravity_{\Horizon} + \epsilon)\right)\RedShiftEnergy(\tStar)[{h}]
    \le{}
    C(\epsilon)\left(
    \norm{\LinearOp{h}}_{\InducedLTwo(\Sigma_{\tStar})}^2
    +\norm{h}_{\InducedLTwo(\Sigma_{\tStar})}^2
    + \EnergyKill(\tStar)[{h}]
    \right).
  \end{align*}
  Then there exists some $k_0$ sufficiently large so that 
  \begin{equation*}
    2\left(\Im\sigma- \max_{\Horizon=\EventHorizonFuture, \CosmologicalHorizonFuture}(\SHorizonControl{\LinearOp}[\Horizon]- (2k_0+1)\SurfaceGravity_{\Horizon} + \epsilon)\right)\RedShiftEnergy(\tStar)[{h}]
    \ge C(\epsilon)\EnergyKill(\tStar)[{h}].
  \end{equation*}
  Using Theorem
  \ref{linear:thm:redshift-commutation:main} we recall that by commuting
  \eqref{linear:eq:KdS-elliptic-regularity:Cauchy-system} with the vectorfields
  $\curlyBrace*{\RedShiftK_i}_{i=1}^N$, we can construct
  a new system of equations
  \begin{equation}
    \label{linear:eq:KdS-elliptic-regularity:commuted-system}
    \bL \mathbf{h} = e^{-\ImagUnit\sigma\tStar}\blittleF
  \end{equation}
  such that
  $\SHorizonControl{\bL}[\Horizon] =
  \SHorizonControl{\LinearOp}[\Horizon] -
  2k\SurfaceGravity_{\Horizon}$. Then commuting through $k>k_0$ number
  of times, applying Theorem \ref{linear:thm:redshift-energy-estimate}, and
  multiplying both sides of the equation by
  $e^{2\Im\sigma\tStar}$, we have that
  \begin{equation*}
    \norm{u}_{\InducedHk{k+1}_\sigma(\Sigma)} \lesssim \norm{f}_{\InducedHk{k}_\sigma(\Sigma)}.
  \end{equation*}
  for $k>k_0$. We conclude by allowing $k$ to be arbitrarily large. 
\end{proof}

\section{Tools for frequency analysis}
\label{linear:sec:freq-analysis}

In this section we introduce tools used in the frequency analysis of
Section \ref{linear:sec:ILED}.

\subsection{Pseudodifferential analysis}
\label{linear:sec:pseudodifferential-analysis}

In this section, we introduce the basics of the classical
pseudo-differential analysis we will be using. We first introduce the
necessary pseudo-differential calculus on $\Real^n$ before defining
pseudo-differential operators on manifolds, which is what we will
actually use in the pseudo-differential arguments involved in Section
\ref{linear:sec:ILED}. For an in-depth reference, we refer the reader to
Chapter 1 of \cite{alinhac_pseudo-differential_2007}, Chapter 18 of
\cite{hormander_analysis_2007}, or Chapters 1-4 of
\cite{taylor_pseudodifferential_1991}.

\begin{definition}
  \label{linear:def:psido:basic}
  For $m\in \Real$, we define $\Psi^{m}(\Real^d)$ to be the class of order-$m$
  symbols on $\Real^d$, consisting of $C^\infty$ functions
  $a(\STPoint,\zeta)$ such that
  \begin{equation*}
    \abs*{D_{\STPoint}^\beta D_\zeta^\alpha a(\STPoint, \zeta) }\le C_{\alpha\beta}\bangle*{\zeta}^{m-|\alpha|}
  \end{equation*}
  for all multi-indexes $\alpha$, where $\bangle*{\zeta} =
  (1+|\zeta|^2)^{\frac{1}{2}}$. We also define the symbol class
  \begin{equation*}
    \Psi^{-\infty} := \bigcap_{m}\Psi^{m}.
  \end{equation*}
  To each symbol is its associated \emph{pseudo-differential operator}
  acting on Schwartz functions $\phi$,
  \begin{equation*}
    a(\STPoint, D)\phi(\STPoint) = \Op(a)\phi(\STPoint) := (2\pi)^{-d}\int_{\Real^d}e^{\ImagUnit\STPoint\cdot \zeta}a(\STPoint, \zeta)\widehat{\phi}(\zeta)\,d\zeta,
  \end{equation*}
  where $\widehat{\phi}$ is the Fourier transform of $\phi$. 
\end{definition}

\begin{remark}
  By abuse of notation, we will understand symbols $a$ homogeneous of
  degree $m$ on $\abs*{\zeta}>1$ to also be symbols in $S^m$, since
  they can be corrected to be proper symbols by some cutoff in
  $S^{-\infty}$.
\end{remark}

We review the basic properties of the pseudo-differential symbol
calculus (see for
example Theorem I.3.2.3, Theorem I.4.1, and Corollary I.4.1 in
\cite{alinhac_pseudo-differential_2007}).
\begin{prop}
  \label{linear:prop:PsiDO:calculus}
  For $m, m_1, m_2\in \Real$, let $a\in \SymClass^m(\Real^d)$,
  $a_1\in \SymClass^{m_1}(\Real^d)$, and
  $a_2\in\SymClass^{m_2}(\Real^d)$.
  \begin{enumerate}    
  \item We have that $a(\STPoint,D)^* = a^*(\STPoint,D)$, where
    \begin{equation*}
      a^*(\STPoint,\zeta) \sim \sum_{\alpha}\frac{1}{\alpha!}\p_\zeta^\alpha D_{\STPoint}\overline{a}(\STPoint,\zeta). 
    \end{equation*}
  \item  We have that
    \begin{equation*}
      a_1(\STPoint,D)\circ a_2(\STPoint,D) = b(\STPoint,D),
    \end{equation*}
    where
    \begin{equation*}
      b = a_1\# a_2 \sim \sum_{\alpha}\frac{1}{\alpha!}\p_\zeta^\alpha a_1 D_\STPoint^\alpha a_2.
    \end{equation*}
  \item We have that
    \begin{equation*}
      \squareBrace*{a_1(\STPoint,D), a_2(\STPoint,D)} = b(\STPoint, D),
    \end{equation*}
    where,
    \begin{equation*}
      b = \PoissonB*{a_1,a_2} + \SymClass^{m_1 + m_2 - 2},
    \end{equation*}
    where
    \begin{equation*}
      \PoissonB*{f,g} = \sum \left(\frac{d f}{d\zeta^i}\frac{d g}{d \STPoint^i} - \frac{d g}{d\zeta^i}\frac{d f}{d \STPoint^i}\right)
    \end{equation*}
    denotes the \emph{Poisson bracket} of $f$ and $g$. 
  \end{enumerate}  
\end{prop}

\begin{remark}
  An immediate consequence of Proposition \ref{linear:prop:PsiDO:calculus} is
  that if $a(\STPoint, \zeta)$ is a real homogeneous symbol, then
  \begin{equation*}
    a(\STPoint, D) + \frac{1}{2V}a(\STPoint, D)(V)
  \end{equation*}
  is formally skew-adjoint with respect to the inner product
  \begin{equation*}
    \int_{\Real^d} u\cdot \overline{v}\, V\abs*{dx}.
  \end{equation*}
  Similarly, if $b(\STPoint, \zeta)$ is an imaginary homogeneous
  symbol, then
  \begin{equation*}
    a(\STPoint, D) - \frac{1}{2V}a(\STPoint, D)(V)
  \end{equation*}
  is formally self-adjoint with respect to the same inner product. 
\end{remark}

\begin{definition}
  We call a symbol $a(\STPoint, \zeta)\in \Psi^m(\Real^d)$ and its
  corresponding operator $a(\STPoint, D)$ an \emph{elliptic symbol} and
  an \emph{elliptic operator} respectively if there exists some $c, C$ such
  that for $\bangle*{\zeta}>C$,
  \begin{equation*}
    \abs*{a(\STPoint, \zeta)}\ge c \bangle*{\zeta}^m.
  \end{equation*}
\end{definition}

Elliptic operators are particularly convenient objects to work with as
they are invertible in a pseudo-differential sense. 
\begin{prop}
  \label{linear:prop:parametrix}
  If $a(\STPoint, \zeta)\in \Psi^m$ is elliptic, then it has a
  \emph{parametrix} $b(\STPoint, \zeta)\in \Psi^{-m}$ such that
  \begin{equation*}
    a(\STPoint, D)b(\STPoint, D) - \Identity \in \Op\Psi^{-\infty}, \qquad b(\STPoint, D)a(\STPoint, D) - \Identity \in \Op\Psi^{-\infty}.
  \end{equation*}  
\end{prop}

Finally, we can also define pseudo-differential operators on a
manifold. We begin with a key proposition regarding the behavior of
symbols under coordinate transformations (see Proposition I.7.1 of
\cite{alinhac_pseudo-differential_2007}).
\begin{prop}
  \label{linear:prop:PsiDO:manifold}
  Let $\phi:\Omega\to \Omega'$ be a smooth diffeomorphism between two
  open subsets of $\Real^d$. Moreover, let $a\in \SymClass^m$ be an
  order $m$ symbol such that the operator $a(\STPoint, D)$ has kernel
  with compact support in $\Omega\times\Omega$.

  Then the following hold. 
  \begin{enumerate}
  \item The function $a'(y, \zeta)$ defined by
    \begin{equation*}
      a'(\phi(\STPoint), \zeta) = e^{-\ImagUnit\phi(\STPoint)\cdot
        \zeta} a(\STPoint,D) e^{\ImagUnit \phi(\STPoint)\cdot \zeta},
      \qquad a' = 0 \text{ for } y\not\in \Omega',
    \end{equation*}
    is also a member of $\SymClass^m$.
  \item The kernel of $a'(\STPoint,D)$ has compact support in $\Omega'\times
    \Omega'$,
  \item For $u\in \TemperedDist(\Real^d)$,
    \begin{equation*}
      a,( D)(u\circ \phi) = (a'(\STPoint, D)u)\circ \phi.
    \end{equation*}
  \item If $a$ has the form
    \begin{equation}
      \label{linear:eq:PsiDO:manifold:a-def}
      a= a_m \mod \SymClass^{m-1}(\Real^d),
    \end{equation}
    where $a_m$ is a homogeneous symbol of order $m$, then the same is
    true for $a'$. That is, there is a homogeneous symbol $a_m'$ of
    order $m$ such that
    \begin{equation*}
      a'= a_m' \mod \SymClass^{m-1}(\Real^d),
    \end{equation*}
    and in fact
    \begin{equation}
      \label{linear:eq:PsiDO:manifold:principal-sym}
      a'_m(\phi(\STPoint), \zeta) = a_m(x, \tensor[^t]{{\chi'}}{}(\STPoint)\zeta).
    \end{equation}
  \end{enumerate} 
\end{prop}

With Proposition \ref{linear:prop:PsiDO:manifold} in hand, we can define a
pseudo-differential operator on a manifold (see Definition I.7.1 of
\cite{alinhac_pseudo-differential_2007}).
\begin{definition}  
  An operator $A:C_0^\infty(\Manifold)\to C^\infty(\Manifold)$ is
  a \emph{pseudo-differential operator of order} $m$ if for any
  coordinate system $\kappa:V\to \widetilde{V} \subset \Real^n$, the
  transported operator
  \begin{align*}
    \widetilde{A}:C_0^\infty(\widetilde{V}) &\to  C_0^\infty(\widetilde{V}),\\
    u&\mapsto A(u\circ \kappa)\circ \kappa^{-1}
  \end{align*}
  is a pseudo-differential of operator of order $m$ in
  $\widetilde{V}$. In other words, $A$ is a pseudo-differential
  operator of order $m$ if for all $\phi,\psi\in
  C^\infty_0(\widetilde{V})$, $\phi\widetilde{A}\psi\in Op\SymClass^m$.
\end{definition}

In particular, \eqref{linear:eq:PsiDO:manifold:principal-sym} shows that the
principal symbol of a pseudodifferential operator on a manifold
$\Manifold$ is a member of
$T^*\Manifold$, and is invariant under coordinate
transformations. 
These pseudo-differential operators have well-behaved mapping
properties based on their symbol.
\begin{prop}
  If $a\in \Psi^m(\Real^d)$, then the operator $a(x, D)$ is a
  well-defined mapping from $H^s(\Real^d)$ to $H^{s-m}(\Real^d)$ for
  any $s\in \Real$. 
\end{prop}

We recall below the Coifman-Meyer commutator estimate. 
\begin{prop}[See \cite{coifman_commutateurs_1978}]
  \label{linear:prop:Coifman-Meyer}
  For $f\in C^\infty$, $P\in OP\SymClass^1$,
  \begin{equation*}
    \norm*{[P,f]u}_{L^2} \le C \norm*{f}_{C^1} \norm*{u}_{L^2}.
  \end{equation*}
\end{prop}

Two standard pseudo-differential objects that will feature heavily in
what follows are the principal symbol of the Laplace-Beltrami operator
associated to a \KdS{} metric and the Hamiltonian vectorfield it
generates. 
\begin{definition}
  \label{linear:def:prin-sym-and-hamiltonian}
  For a fixed \KdS{} metric $g_b$, define $\PrinSymb_b$ to be the
  \emph{principal symbol} of the Laplace-Beltrami operator associated to
  $g_b$,
  \begin{equation*}
    \PrinSymb_b := g^{\mu\nu}\zeta_\mu\zeta_\nu. 
  \end{equation*}
  Moreover, for a fixed \KdS{} metric $g_b$, denote by
  \begin{equation*}
    H_{\PrinSymb_b}:= \sum_{\mu}\frac{d \PrinSymb_b}{d x_\mu}\frac{d }{d \zeta_\mu} - \frac{d \PrinSymb_b}{d \zeta_\mu} \frac{d}{d x_\mu},\qquad
    (x;\zeta):= (t,r,\omega; \sigma, \xi, \eta),
  \end{equation*}
  the \emph{Hamiltonian vectorfield} associated to $g_b$. 
\end{definition}

We also have the following inequality which serves as a generalization
of G\"{a}rding's inequality which we will make repeated use of in
Section \ref{linear:sec:ILED}. 
\begin{theorem}[Corollary II.8 \cite{tataru_feffermanphong_2002}]
  \label{linear:thm:sym-ineq-to-op-est}
  Let $\curlyBrace{a_j}_{j=1}^k, b\in C^{1,1}\SymClass^1$ be a finite
  set of real symbols with $|b|\le \sum |a_j|$, where
  $C^{1,1}\SymClass^1$ denotes the class of first-order symbols with
  $C^{1,1}$ coefficients. Then,
  \begin{equation*}
    \norm{B(x, D)u}_{L^2}\lesssim \sum_{j=1}^k\norm{A_j(x, D)u}_{L^2} + \norm{u}_{L^2}.
  \end{equation*}
\end{theorem}

In our application on slowly-rotating \KdS{}
backgrounds, it will be convenient to perform all the calculations
involving pseudo-differential calculus in this paper using the
Boyer-Lindquist coordinates $(t,r,\omega)$, with
$(\sigma, \xi, \FreqAngular)$
representing the respective frequency variables, with covectors
written as
\begin{equation*}
  \zeta = \sigma\,dt + \xi\,dr + \eta,\quad \eta\in T^*\Sphere^2,
\end{equation*}
where we recall that by its construction, $\tStar=t$ on a small
neighborhood of $r=3M$. 

We define two specific classes of symbols, which will come up in our
subsequent analysis.
\begin{definition}
  Let $\SymClass^n(\StaticRegionWithExtension) =
  \SymClass^n(\tStar,r,\omega;\sigma,\xi,\eta)$ denote the class of
  order-$n$ symbols on $\StaticRegionWithExtension$, and
  $S^n(\Sigma) = S^n(r,\theta;\xi,\eta)$ denote the sub-class of
  stationary, axi-symmetric order-$n$ symbols independent of $\sigma$
  and $\tStar$.  
\end{definition}
Throughout this paper, we will only use pseudo-differential operators
in $\Op S^m(\Sigma)$. That is, we will only work with
pseudo-differential operators which are pseudo-differential in the
spatial variables and strictly differential in $\tStar$. It will then
be convenient to use the following definition of negative Sobolev
spaces.
\begin{definition}
  We define the negative Sobolev spaces
  \begin{equation*}
    \norm*{h}_{H^{-k}(\DomainOfIntegration)} := \norm*{\bangle*{D_x}^{-1}h}_{L^2(\DomainOfIntegration)}. 
  \end{equation*}
\end{definition}

We remark that since we will only need to use
pseudo-differential methods in a neighborhood of the trapped set, all
of our symbols will be compactly supported in a neighborhood of
$r=3M$.

\subsection{Trapping behavior in \KdS}
\label{linear:sec:trapping-KdS}

In this section, we discuss the well-known properties of
\textit{trapped null geodesics} in \KdS{} which remain in a compact
spatial region for all time (see for example Proposition 3.1 of
\cite{dyatlov_asymptotics_2015} and Section 6.4 of
\cite{vasy_microlocal_2013}). These null geodesics represent a
fundamental high-frequency geometric obstacle to decay.  To analyze
the dynamics of the trapped set $\TrappedSet_b$ in frequency space, we
consider null-bicharacteristics rather than null-geodesics, as
null-geodesics are just the physical projection of the integral curves
of null-bicharacteristics.

It is instructive to first consider the trapped null geodesics in
\SdS, where we can write out the trapped set explicitly, and make some
fundamental observations. The trapped null geodesics in \KdS{} are in
a sense perturbations of the trapped null geodesics in \SdS. We will
make this notion more rigorous in what follows.

On \SdS, the trapped set can be located entirely physically.
\begin{lemma}
  \label{linear:lemma:trapping:SdS}
  For $g_{b_0}$ a \SdS{} background, the trapped set is given by
  \begin{equation*}
  \TrappedSet_{b_0} = \curlyBrace*{(t, r, \omega;\sigma, \xi, \eta): r=3M, \xi=0, \PrinSymb_{b_0}=0},
\end{equation*}
  where $\PrinSymb_{b_0}$ is the principal symbol of the scalar wave
  operator $\ScalarWaveOp[g_{b_0}]$. Moreover, the trapped set is
  unstable in the sense that
  \begin{equation*}
  \mu_{b_0} > 0,\quad
    \PrinSymb_{b_0} =0,\quad
    \pm (r-3M) > 0, \quad
    H_{\PrinSymb_{b_0}}r = 0 \implies \pm H_{\PrinSymb_{b_0}}^2r > 0,
\end{equation*}
  where we recall from Definition \ref{linear:def:prin-sym-and-hamiltonian}
  that we denote by $\PrinSymb_{b_0}$ the principal symbol of
  $\ScalarWaveOp[g_{b_0}]$, and by $H_{\PrinSymb_{b_0}}$ the
  Hamiltonian vectorfield of $g_{b_0}$.
\end{lemma}
\begin{proof}
  See appendix \ref{linear:appendix:lemma:trapping:SdS}.
\end{proof}

\begin{remark}
  The physical projection of $\TrappedSet_{b_0}$ is
  exactly the photon sphere, $r=3M$. 
\end{remark}

We now move onto the trapped set in the case of \KdS. In this case,
the trapped set exhibits frequency-dependent behavior.
\begin{lemma}
  \label{linear:lemma:trapping:KdS}
  For $g_{b}$ a \KdS{} background, the trapped set
  \begin{equation}
  \TrappedSet_{b} = \curlyBrace*{(t, r, \omega; \sigma, \xi, \FreqAngular): r=\rTrapping_b(\sigma, \FreqAngular), \xi=0, \PrinSymb_{b}=0},
\end{equation}
  where $\PrinSymb_{b}$ is the principal symbol of the scalar wave
  operator $\ScalarWaveOp[g_{b}]$, and $\rTrapping_b(\sigma,
  \FreqAngular)$ is a function satisfying the following properties. 
  \begin{enumerate}
  \item $\rTrapping_b(\sigma,\FreqAngular)$ lies in an $O(a)$ neighborhood
  of $r=3M$ for all $\sigma, \FreqAngular$.
\item $\rTrapping_b(\sigma, \FreqAngular)$ is smooth in $\sigma,
  \FreqAngular$, as well as the black hole parameters $b$.  
\end{enumerate}
  Moreover, the trapped set is unstable in the sense that
  \begin{equation}
  \Delta_b>0,
    \quad \PrinSymb_b=0,
    \quad \pm (r-\rTrapping_b)>0,
    \quad H_{\PrinSymb_b} r=0
    \implies \pm H^2_{\PrinSymb_b} r > 0,
\end{equation}
  where we denote by $H_{\PrinSymb_b}$ the Hamiltonian vectorfield of
  $g_b$. 
\end{lemma}
\begin{proof}
  See appendix \ref{linear:appendix:lemma:trapping:KdS}.
\end{proof}

We conclude this section by defining two important cutoff functions
that we will subsequently make use of.
\begin{lemma}
  \label{linear:lemma:freq-cutoff-construction}  
  There exist frequency cutoffs
  \begin{equation*}
  \breve{\chi}_\zeta:= \breve{\chi}_\zeta(r, \theta;\xi,\eta),\qquad
    \mathring{\chi}_\zeta:= \mathring{\chi}_\zeta(r, \theta;\xi,\eta), 
\end{equation*}
  where $(t,r,\theta, \varphi;\sigma,\xi, \FreqTheta, \FreqPhi)$ are
  the Boyer-Lindquist coordinates, 
  defined so that
  \begin{equation}
  \label{linear:eq:freq-cutoff-construction}
    \breve{\chi}_{\zeta} =
    \begin{cases}
  1 & \frac{\abs*{\xi}^2}{|\eta|^2} \ge \delta_\zeta,\\
  0 & \frac{\abs*{\xi}^2}{\abs*{\eta}^2} < \frac{1}{2} \delta_{\zeta},\\
    \end{cases}\qquad
    \mathring{\chi}_{\zeta} =
    \begin{cases}
  1 & \frac{\abs*{\xi}^2}{\abs*{\eta}^2} \le \delta_\zeta,\\
  0 & \frac{\abs*{\xi}^2}{\abs*{\eta}^2} > 2 \delta_\zeta,
    \end{cases}
\end{equation}
  and some constant $C_{\xi}>0$ sufficiently large such that for
  $r\in [\mathring{r}_-, \mathring{R}_+]$,
  \begin{equation*}
  C_\xi(H_{\PrinSymb_b}r)^2 -\PrinSymb_b
\end{equation*}
  is elliptic on the support of $\breve{\chi}_{\zeta}$
  for all $g_b$ sufficiently slowly-rotating \KdS{} black holes.   
\end{lemma}
\begin{proof}
  First, we observe that $H_{\PrinSymb_b}r = 2\rho_b^{-2}\Delta_b\xi$ in Boyer-Lindquist
  coordinates. We write
  \begin{equation*}
  \rho_b^2\PrinSymb_b^2
    = \Delta_b\xi^2 + \varkappa_b\FreqTheta^2
    - (1+\lambda_b)^2
    \left(
  \Delta_b^{-1}\left(
  (r^2+a^2)\sigma+ a\FreqPhi
\right)^2
      - \varkappa_b^{-1}\sin^{-2}\theta\left(
  a\sin^2\theta\sigma +\FreqPhi
\right)^2      
\right). 
\end{equation*}
  Observe that for $a\ll M, \Lambda$ and $r\in
  [\mathring{r}_-, \mathring{R}_+]$, we have that
  \begin{equation*}
  g_b(\KillT,\KillT) = (1+\lambda_b)^2\left(\Delta_b^{-1}(r^2+a^2)^2 - a^2 \varkappa_b^{-1}\sin^{2}\theta\right) > 0. 
\end{equation*}
  As a result, we have that
  \begin{equation*}
  \rho_b^2\PrinSymb_b + \frac{1}{2}g_b(\KillT, \KillT)\sigma^2
    \lesssim \xi^2 + |\FreqAngular|^2. 
\end{equation*}
  By the construction of $\breve{\chi}_{\zeta}$ and
  $\mathring{\chi}_\zeta$ in \eqref{linear:eq:freq-cutoff-construction}, we
  have that $(H_{\PrinSymb_b}r)^2 \gtrsim \xi^2 + |\FreqAngular|^2$
  on the support of $\breve{\chi}_{\zeta}$. As a result, we have that
  for some $C_\xi$ sufficiently large, 
  \begin{equation*}
  C_\xi(H_{\PrinSymb_b}r)^2 - \rho_b^2\PrinSymb_b \gtrsim -\frac{1}{2}g_b(\KillT, \KillT)\sigma^2 + \xi^2 + |\FreqAngular|^2.
\end{equation*}
  Since $g_b(\KillT, \KillT)<0$ and $\rho_b>0$ for $r\in
  [\mathring{r}_-, \mathring{R}_+]$ we then have that
  $C_\xi(H_{\PrinSymb_b}r)^2 - \PrinSymb_b$ is elliptic as
  desired. 
\end{proof}

\subsection{Pseudo-differential modified divergence theorem}
\label{linear:sec:int-by-parts-arg}

In this section, we introduce a pseudo-differential modification of
the main divergence property presented in equation
(\ref{linear:eq:EMTensor:divergence-property}). This modification allows us
to handle the frequency-dependent nature of trapping in the \KdS{}
family, and uses small pseudo-differential perturbations of
vectorfield multipliers and Lagrangian correctors. We emphasize that
this perturbation is only used in Section \ref{linear:sec:ILED-trapping:KdS}
to prove Theorem \ref{linear:thm:ILED-near:main}.

We first prove a convenient lemma connecting the frequency analysis
with the unperturbed divergence property in equation
\eqref{linear:eq:div-them:J-K-currents}.
\begin{lemma}
  \label{linear:lemma:divergence-prop:freq-formulation}
  Let $h$ be a complex-valued matrix function
  $h:\StaticRegionWithExtension\to \Complex^D$, $g$ be a fixed \KdS{}
  metric and $\PrinSymb$ the principal symbol of
  $\ScalarWaveOp[g]$. Then we can rewrite $\KCurrent{X,q,0}[h]$ as
  defined in \eqref{linear:eq:J-K-currents:def} as
  \begin{equation}
  \label{linear:eq:divergence-prop:freq-formulation}
  \KCurrent{X,q, 0}[{h}]
  = (\KCurrentSym{X,q}_{(2)})^{\alpha\beta}\nabla_{(\alpha}{h}\cdot \nabla_{\beta)}\overline{{h}}
  + \KCurrentSym{X,q}_{(0)}\abs*{{h}}^2,
\end{equation}
  where $(\KCurrentSym{X,q}_{(2)})^{\alpha\beta}\zeta_\alpha\zeta_\beta$
  and $\KCurrentSym{X,q}_{(0)}$ are given by
  \begin{equation*}
  (\KCurrentSym{X,q}_{(2)})^{\alpha\beta}\zeta_\alpha\zeta_\beta
    = \frac{1}{2}H_{\PrinSymb}X
    + \left(q-\frac{1}{2}\nabla_g\cdot X\right)\PrinSymb,\qquad
    \KCurrentSym{X,q}_{(0)} = -\frac{1}{2}\nabla^\alpha\p_\alpha q.
\end{equation*}
\end{lemma}
\begin{proof}
  Observe that
  \begin{equation*}
  2\DeformationTensor{X}\cdot \EMTensor[h]
    = \LieDerivative_Xg^{\alpha\beta}\nabla_{(\alpha} h\cdot \nabla_{\beta)}\overline{ h}
    - (\nabla_g\cdot X) g^{\alpha\beta}\nabla_{\alpha} h\cdot \nabla_{\beta}\overline{ h}.
\end{equation*}
  By the definition of the Hamiltonian
  vectorfield in Definition \ref{linear:def:prin-sym-and-hamiltonian}, 
  \begin{equation*}
  \LieDerivative_Xg^{\alpha\beta}\zeta_\alpha\zeta_\beta
    = H_{\PrinSymb}X.
\end{equation*}
  Equation \eqref{linear:eq:divergence-prop:freq-formulation} then follows
  from the definition of the principal symbol of
  $\ScalarWaveOp[g]$ in Definition \ref{linear:def:prin-sym-and-hamiltonian}
  and the definition of $\KCurrent{X,q,m}[h]$ in
  \eqref{linear:eq:J-K-currents:def}.   
\end{proof}

Due to the frequency-dependent nature of trapping in the \KdS{}
family, we are not able to use the divergence property in equation
\eqref{linear:eq:div-them:J-K-currents} directly to prove the desired
Morawetz estimates near the trapped set $\TrappedSet_b$. Instead, we use
an integration-by-parts variant of the divergence property that uses a
pseudo-differential perturbation of the vectorfield multipliers. 

First, we observe that  
\begin{equation}
  \label{linear:eq:LinEinsteinConj-def}
  \LinEinsteinConj_{g_b} := \PseudoSubPFixer\LinEinstein_{g_b}\PseudoSubPFixer^- = \ScalarWaveOp[g_b] + \SubPConjOp_b + \PotentialConjOp_b, \qquad
  \SubPConjOp_b := \PseudoSubPFixer \SubPOp_b\PseudoSubPFixer^{-}
  + \PseudoSubPFixer\left[\ScalarWaveOp[g_b], \PseudoSubPFixer^-\right], \qquad
  \PotentialConjOp_b := \PseudoSubPFixer\PotentialOp_b \PseudoSubPFixer^-,
\end{equation}
where $\PseudoSubPFixer$ is as constructed in \eqref{linear:eq:Q-def}, and
$\PseudoSubPFixer^-$ denotes its parametrix.
In what follows, it will also be convenient to split $\SubPConjOp_b$
into its Hermitian and skew-Hermitian components, given by
\begin{equation}
  \label{linear:eq:SubPConjOp:sym-skew-sym-split-def}
  \SubPConjOp_{b,a} := \frac{1}{2} \left(\SubPConjOp_b - \SubPConjOp_b^*\right),\qquad
  \SubPConjOp_{b,s} := \frac{1}{2}\left(\SubPConjOp_b + \SubPConjOp_b^*\right),
\end{equation}
where the adjoint is taken with respect to the
$L^2(\DomainOfIntegration)$ norm. Observe that both
$\SubPConjOp_b, \SubPConjOp_{b,a}$ and $\SubPConjOp_{b,s}$ belong to
$\Op S^{1} + \Op S^0\p_{\tStar}$, and moreover we can write
\begin{equation}
  \begin{split}
  \SubPConjOp_b = \SubPConjOp_{0}\p_{\tStar} + \SubPConjOp_{1},\qquad
  \SubPConjOp_{b,a} = \SubPConjOp_{0,a}\p_{\tStar} + \SubPConjOp_{1,a},\qquad
  \SubPConjOp_{b,s} = \SubPConjOp_{0,s}\p_{\tStar} + \SubPConjOp_{1,s}. 
  \end{split}  
\end{equation}
We will similarly define
\begin{equation*}
  \label{linear:eq:PotentialConjOp:sym-skew-sym-split-def}
  \PotentialConjOp_{b,a} := \frac{1}{2} \left(\PotentialConjOp_b - \PotentialConjOp_b^*\right),\qquad
  \PotentialConjOp_{b,s} := \frac{1}{2}\left(\PotentialConjOp_b + \PotentialConjOp_b^*\right).
\end{equation*}

\begin{prop}
  \label{linear:prop:div-thm:PDO-modification}
  Let us consider
  \begin{equation}
    \label{linear:eq:div-thm:PDO:Morawetz-Lagrangecorr-g_b-def}
    \MorawetzVF_b := \MorawetzVF_{b_0} + \widetilde{\MorawetzVF},\qquad
    \LagrangeCorr_b:= \LagrangeCorr_{b_0} + a\tilde{\LagrangeCorr},
  \end{equation}
  where $\MorawetzVF_{b_0}$ and $\LagrangeCorr_{b_0}$ are a smooth
  vectorfield and a smooth function respectively, and 
  \begin{equation}
    \label{linear:eq:div-thm:PDO:multipler-def}
    \begin{split}
      \widetilde{\MorawetzVF} &= \widetilde{\MorawetzVF}_0\p_t + \widetilde{\MorawetzVF}_1, \\
      \tilde{\LagrangeCorr} &= \tilde{\LagrangeCorr}_0 + \tilde{\LagrangeCorr}_{-1}\p_t, 
    \end{split}
  \end{equation}
  where
  $\widetilde{\MorawetzVF}_i, \tilde{\LagrangeCorr}_i \in \Op
  S^i(\Sigma)$, and $\widetilde{\MorawetzVF}_0$ and
  $\tilde{\LagrangeCorr}_0$ are self-adjoint with respect to the
  $L^2(\Sigma)$ inner product, and $\widetilde{\MorawetzVF}_1$ and
  $\tilde{\LagrangeCorr}_1$ are skew-adjoint with respect to the
  $L^2(\Sigma)$ inner product.  Moreover, let
  \begin{equation*}    
    \DomainOfIntegration:=[0,\TStar]\times \TrappingNbhd,\qquad
    \TrappingNbhd:=[\mathring{r}_-, \mathring{R}_+]\times \Sphere^2.
  \end{equation*}
  and let $h$ be a function such that for all $\tStar$,
  $h(\tStar,\cdot)$ is compactly supported on
  $\TrappingNbhd$.  Then
  \begin{align}  
    &-\Re\bangle*{\LinEinsteinConj_{g_b}{h}, (\MorawetzVF_b+\LagrangeCorr_b){h}}_{L^2(\DomainOfIntegration)}\notag \\
    ={}& \int_{\DomainOfIntegration}\KCurrent{\MorawetzVF_{b_0},\LagrangeCorr_{b_0},0}[{h}]
         - \Re \bangle*{\SubPConjOp_b[h], (\MorawetzVF_{b_0}+\LagrangeCorr_{b_0}){h}}_{L^2(\DomainOfIntegration)}
         - \Re\bangle*{\PotentialConjOp_b{h}, (\MorawetzVF_{b_0}+\LagrangeCorr_{b_0}){h}}_{L^2(\DomainOfIntegration)}
         + a\Re\KCurrentIbP{\widetilde{\MorawetzVF}, \tilde{\LagrangeCorr}}[{h}] 
         \notag\\
    & + \left.\int_{\Sigma_{\tStar}}\JCurrent{\MorawetzVF_{b_0}, \LagrangeCorr_{b_0},0}[{h}]\cdot n_{\Sigma}\right\vert_{\tStar=0}^{\tStar= \TStar}
      + \evalAt*{\Re\bangle*{g_b(\KillT, n_{\TrappingNbhd})\SubPConjOp_{0}h, \MorawetzVF_{b_0}h}_{L^2(\TrappingNbhd)}}_{\tStar=0}^{\tStar=\TStar}
      + a \Re\left.\JCurrentIbP{\widetilde{\MorawetzVF}, \tilde{\LagrangeCorr}}(\tStar)[{h}]\right\vert_{\tStar=0}^{\tStar= \TStar}
      .\label{linear:eq:ILED-near:combined-divergence-theorem}
  \end{align}
  where
  \begin{align}
    2 \KCurrentIbP{\widetilde{\MorawetzVF}, \tilde{\LagrangeCorr}}[{h}]
    ={}& \bangle*{\left[\widetilde{\MorawetzVF}, \ScalarWaveOp[g_b]\right]{h},{h}}_{L^2(\DomainOfIntegration)}
         + \bangle*{\left[\widetilde{\MorawetzVF}, \SubPConjOp_{b,s}\right]h, h}_{L^2(\DomainOfIntegration)}
         + \bangle*{\left(\widetilde{\MorawetzVF} \SubPConjOp_{b,a} + \SubPConjOp_{b,a} \widetilde{\MorawetzVF}\right){h},{h}}_{L^2(\DomainOfIntegration)}\notag\\
       & - \bangle*{\left(\tilde{\LagrangeCorr}\ScalarWaveOp[g_b]+ \ScalarWaveOp[g_b]\tilde{\LagrangeCorr}\right){h}, {h}}_{L^2(\DomainOfIntegration)}
         - 2\bangle*{\SubPConjOp_{b}{h}, \tilde{\LagrangeCorr} {h}}_{L^2(\DomainOfIntegration)} 
         -  2\bangle*{\PotentialConjOp_{b}{h},\left(\widetilde{\MorawetzVF}+\tilde{\LagrangeCorr}\right){h}}_{L^2(\DomainOfIntegration)},\notag \\                
    \JCurrentIbP{\widetilde{\MorawetzVF}, \tilde{\LagrangeCorr}}(\tStar)[{h}]
    ={}& \bangle*{n_{\TrappingNbhd}{h}, \widetilde{\MorawetzVF}{h}}_{\LTwo(\TrappingNbhd_{\tStar})}
    + \bangle*{\SubPConjOp_b{h}, \widetilde{\MorawetzVF}_0 g_b(\KillT, n_{\TrappingNbhd}){h}}_{\LTwo(\TrappingNbhd_{\tStar})}
    + \bangle*{n_{\TrappingNbhd}{h}, \tilde{\LagrangeCorr}{h}}_{\LTwo(\TrappingNbhd_{\tStar})} \notag\\
       & + \bangle*{g_{b}(\KillT, n_{\TrappingNbhd})\SubPConjOp_{0}h, \widetilde{\MorawetzVF} h }_{L^2(\TrappingNbhd_{\tStar})}
         .\label{linear:eq:ILED-near:J-K-def}
  \end{align}
\end{prop}
\begin{proof}
  See appendix \ref{linear:appendix:prop:div-thm:PDO-modification}.
\end{proof}

We can decompose $\KCurrentIbP{\widetilde{\MorawetzVF}, \tilde{\LagrangeCorr}}[{h}]$ into
its principal, subprincipal, and zeroth order components as follows
\begin{equation}
  \label{linear:eq:ILED-near:ILED-int-by-parts-decomposed}
  \KCurrentIbP{\widetilde{\MorawetzVF}, \tilde{\LagrangeCorr}}[{h}]
  = \KCurrentIbP{\widetilde{\MorawetzVF}, \tilde{\LagrangeCorr}}_{(2)}[{h}]
  + \KCurrentIbP{\widetilde{\MorawetzVF}, \tilde{\LagrangeCorr}}_{(1)}[{h}]
  + \KCurrentIbP{\widetilde{\MorawetzVF}, \tilde{\LagrangeCorr}}_{(0)}[{h}],
\end{equation}
where
\begin{equation}
  \label{linear:eq:ILED-near:ILED-int-by-parts-decomposed:decomposition-def}
  \begin{split}
  \KCurrentIbP{\widetilde{\MorawetzVF}, \tilde{\LagrangeCorr}}_{(2)}[{h}]
  ={}& \frac{1}{2}\bangle*{\left(\left[\widetilde{\MorawetzVF}, \ScalarWaveOp[g_b]\right]
   + \left(\widetilde{\MorawetzVF} \SubPConjOp_{b,a} + \SubPConjOp_{b,a}\widetilde{\MorawetzVF} \right)
   - \left(\ScalarWaveOp[g_b]\tilde{\LagrangeCorr} +
   \tilde{\LagrangeCorr}\ScalarWaveOp[g_b] \right)\right){h}, {h}}_{L^2(\DomainOfIntegration)}\\
  \KCurrentIbP{\widetilde{\MorawetzVF}, \tilde{\LagrangeCorr}}_{(1)}[{h}]
  ={}& -\frac{1}{2}\bangle*{\left[\SubPConjOp_{b,s}, \widetilde{\MorawetzVF}\right]{h},{h}}_{L^2(\DomainOfIntegration)}
   - \bangle*{\SubPConjOp_{b} {h}, \tilde{\LagrangeCorr} {h}}_{L^2(\DomainOfIntegration)}       
   - \bangle*{\PotentialConjOp_b {h}, \widetilde{\MorawetzVF} {h}}_{L^2(\DomainOfIntegration)},\\
  \KCurrentIbP{\widetilde{\MorawetzVF}, \tilde{\LagrangeCorr}}_{(0)}[{h}]
  ={}& -\bangle*{\PotentialConjOp_b {h}, \tilde{\LagrangeCorr} {h}}_{L^2(\DomainOfIntegration)}.
  \end{split}  
\end{equation}
We have a similar decomposition of the boundary terms
\begin{equation*}
  \JCurrentIbP{\widetilde{\MorawetzVF}, \tilde{\LagrangeCorr}}(\tStar)[{h}]
  = \JCurrentIbP{\widetilde{\MorawetzVF}, \tilde{\LagrangeCorr}}_{(2)}(\tStar)[{h}]
  + \JCurrentIbP{\widetilde{\MorawetzVF}, \tilde{\LagrangeCorr}}_{(1)}(\tStar)[{h}],
\end{equation*}
where
\begin{equation}
  \label{linear:eq:ILED-near:ILED-int-by-parts-decomposed:J-decomposition-def}
  \begin{split}
  \JCurrentIbP{\widetilde{\MorawetzVF}, \tilde{\LagrangeCorr}}_{(2)}(\tStar)[{h}]
  :={}& \bangle{n_{\TrappingNbhd}{h},\widetilde{\MorawetzVF}{h}}_{L^2(\TrappingNbhd_{\tStar})},\\
  \JCurrentIbP{\widetilde{\MorawetzVF}, \tilde{\LagrangeCorr}}_{(1)}(\tStar)[{h}]
  :={}& 
   \bangle*{\SubPConjOp_b{h}, \widetilde{\MorawetzVF}_{0}g_b(\KillT, n_{\TrappingNbhd}){h}}_{L^2(\TrappingNbhd_{\tStar})}
   + \bangle*{n_{\TrappingNbhd}{h}, \tilde{\LagrangeCorr}{h}}_{L^2(\TrappingNbhd_{\tStar})}
  + \bangle*{g_{b}(\KillT, n_{\TrappingNbhd})\SubPConjOp_{0,a}h, \widetilde{\MorawetzVF} h }_{L^2(\TrappingNbhd_{\tStar})}.
  \end{split}
\end{equation}

We observe that similar to Lemma
\ref{linear:lemma:divergence-prop:freq-formulation}, we have the following
symbolic representation of the principal bulk term
$\KCurrentIbP{\widetilde{\MorawetzVF}, \tilde{\LagrangeCorr}}_{(2)}$.
\begin{lemma}
  Let $h$ be a complex-valued matrix function
  $h:\StaticRegionWithExtension\to \Complex^D$, $g_b$ be a fixed
  slowly-rotating \KdS{}
  metric and $\PrinSymb_b$ be the principal symbol of
  $\ScalarWaveOp[g]$. Furthermore, let $\MorawetzVF_b$,
  $\LagrangeCorr_b$ be as defined in Proposition
  \ref{linear:prop:div-thm:PDO-modification}. Then
  \begin{equation*}
  \KCurrent{\MorawetzVF_{b_0}, \LagrangeCorr_{b_0},0}[h] + a\KCurrentIbP{\widetilde{\MorawetzVF}, \tilde{\LagrangeCorr}}[h]
\end{equation*}
  has principal symbol given by 
  \begin{equation*}
  \frac{1}{2}H_{\PrinSymb}(\MorawetzSym_{b_0} + a \tilde{\MorawetzSym}) - \SubPConjSym_b(\MorawetzSym_{b_0} + a \tilde{\MorawetzSym}) + \PrinSymb_b(\LagrangeCorrSym_{b_0}+a\tilde{\LagrangeCorrSym}),
\end{equation*}
  where
  \begin{gather*}        
  \MorawetzSym_{b_0} = \ImagUnit f_{b_0}\xi,\qquad
    \tilde{\MorawetzSym} = \tilde{\MorawetzSym}_0\sigma + \tilde{\MorawetzSym}_1,\\    
    \LagrangeCorrSym_{b_0}= \LagrangeCorr_{b_0} - \frac{1}{2}\nabla_{g_{b_0}}\cdot \MorawetzVF_{b_0}, \qquad
    \tilde{\LagrangeCorr} = \tilde{\LagrangeCorrSym}_{-1}\sigma + \tilde{\LagrangeCorrSym}_0.     
\end{gather*}
\end{lemma}
\begin{proof}
  The conclusion follows from the form of
  $\KCurrentIbP{\widetilde{\MorawetzVF},
  \tilde{\LagrangeCorr}}_{(2)}[h]$ in
\eqref{linear:eq:ILED-near:ILED-int-by-parts-decomposed:decomposition-def}
  and Lemma \ref{linear:lemma:divergence-prop:freq-formulation}.
\end{proof}

\subsection{Subprincipal symbol of $\LinEinstein_{g_b}$ at trapping}
\label{linear:sec:subprincipal-symbol}

The presence of a nontrivial non-signed subprincipal operator in
$\LinEinstein_{g_b}$ poses a considerable obstacle in proving the
desired high-frequency Morawetz estimate. Fortunately, the
subprincipal operator of $\LinEinstein_{g_b}$ possesses an appropriate
microlocal smallness at $\TrappedSet_b$ that is enough to close the
desired high-frequency Morawetz estimate in Section
\ref{linear:ILED:near}. This smallness was also critical to the proof in
\cite{hintz_global_2018} (See Theorem 4.4 in
\cite{hintz_global_2018}). We will first specify what we mean when we
refer to the subprincipal operator\footnote{\textit{A priori}, only
  the principal operator of a pseudo-differential operator is
  well-defined. }, and then uncover the desired smallness at the
trapped set.

\subsubsection{The invariant subprincipal operator}

Let $\mathcal{E}$ be a tensor bundle over a manifold $X$. The main
case of interest in this paper will be when $\mathcal{E}$ is the
cotangent bundle $T^*X$ or the bundle of symmetric two-tensors
$S^2T^*X$. The main property we are interested in is the norm of the
skew-adjoint component of the subprincipal operator. For convenience,
we remove the dependence of adjoints on a volume density by tensoring
all bundles with the half-density bundle $\Omega^{\frac{1}{2}}$ over
$X$.

It is well-known that if $P\in \SymClass^m(X,\mathcal{E}\otimes \Omega^{\frac{1}{2}})$ is a sum of homogeneous
symbols
\begin{equation*}
  p \sim \sum p_m,\qquad p_j \in S^j_{hom}(T^*X\backslash 0, \Complex^{N\times N})
\end{equation*}
with $p_j$ being a homogeneous symbol of
degree $j$ valued in complex $N\times N$ matrices, that the
subprincipal symbol 
\begin{equation}
  \label{linear:eq:subprincial-op-def:hormander-def}
  \sigma_{sub}(P) := p_{m-1} - \frac{1}{2\ImagUnit}\sum_j \p_{x_j\xi_j}p_m(x,\xi) \in S^{m-1}_{hom}(T^*X\backslash 0, \Complex^{N\times N})
\end{equation}
is well-defined under changes of coordinates (see
\cite{hormander_analysis_2007} Theorem 18.1.33). However, the
subprincipal symbol as defined above does still depend on the choice
of local trivialization of $\mathcal{E}$. We would like a
frame-independent notion of the subprincipal symbol since this would
allow us to choose convenient local frames in explicit computations.
Fortunately, as shown in \cite{hintz_resonance_2017}, there exists a
modification of \eqref{linear:eq:subprincial-op-def:hormander-def} which is
independent both of the choice of local trivialization and of local
coordinates on $\mathcal{M}$. We review the basics of the construction
here as well as the key features of the invariant subprincipal symbol
that we will use. For a more thorough discussion, we refer the reader
to Section 3.3 and Section 4 of \cite{hintz_resonance_2017}. The
results here on the invariant subprincipal symbol are specialized
cases of more general results in the literature. For the subsequent
results, we list both a reference for the general result, and provide
a proof in the appendix for the sake of completion.

\begin{definition}
  \label{linear:def:subprincipal-op-def:main-def}
  Consider $P\in \Op S^m(X, \mathcal{E}\otimes \Omega^{\frac{1}{2}})$
  with scalar principal symbol $p$. Moreover, let
  $\curlyBrace*{e_k(x)}_{k=1}^N$ be a local frame of $\mathcal{E} $,
  and define the operators
  $P_{jk}\in \Op S^m (X, \Omega^{\frac{1}{2}})$ by
  \begin{equation*}
    P_{jk}\left( \sum_k u_k(x)e_k(x)\right)  = \sum_{jk}P_{jk}(u_k)e_j(x), u_k\in C^\infty(X, \Omega^{\frac{1}{2}}).
  \end{equation*}
  Then we define the \emph{invariant subprincipal
    operator} $S_{sub}(P)\in \Diff^1(T^*X\backslash 0,
  \pi^*\mathcal{E})$  by
  \begin{equation}
    \label{linear:eq:subprincipal-op-def:main-def}
    S_{sub}(P)\left(
      \sum_k q_k(x,\zeta)e_k(x)
    \right):= \sum_{jk}(\sigma_{sub}(P_{jk})q_k)e_j - \ImagUnit\sum_{k}H_{p}q_k e_k.
  \end{equation}
  Observe that $S_{sub}(P)$ is homogeneous of degree $m-1$ with
  respect to dilations in the fibers of $T^*X\backslash 0$, and that
  in a local frame, can be understood as a matrix of first-order
  differential operators. 
\end{definition}

The main property of the invariant subprincipal symbol, and indeed its
very nomenclature comes from the fact that it is invariant under both
changes of coordinates and changes of frame.
\begin{lemma}
  Let $P\in \Op S^m(X, \mathcal{E}\otimes \Omega^{\frac{1}{2}})$ with
  scalar principal symbol $p$, and let $\{e_k(x)\}_{k=1}^N$ and
  $\{e'_k\}_{k=1}^N$ be two local frames of $\mathcal{E}$ such that
  \begin{equation*}
    e_j(x) = C(x)e_j'(x), \qquad C\in C^\infty(U,\End(\mathcal{E})).
  \end{equation*}
  Then, we have that
  \begin{equation*}
    S_{sub}^e(P) = S_{sub}^{e'}(P). 
  \end{equation*}
\end{lemma}
\begin{proof}
  We can directly compute
  \begin{equation*}
    \sigma_{sub}(C^{-1}PC)
    = (C^{e'})^{-1}\sigma_{sub}^{e'}C^{e'}
    - \ImagUnit (C^{e'})^{-1}H_{p}(C^{e'}),
  \end{equation*}
  where $C^{e'}$ is the matrix of $C$ in the frame $e'$. Then observe
  that
  \begin{equation*}
    (C^{-1}PC)^{e'} = P^e,\qquad (C^{e'})^{-1}H_p(C^{e'}) = (C^{e'})^{-1}H_p C^{e'} - H_p.
  \end{equation*}
  As a result, we have that
  \begin{equation*}
    \sigma_{sub}^e(P) - \ImagUnit H_p = (C^{e'})^{-1}(\sigma^{e'}_{sub}(P)-\ImagUnit H_p)C^{e},
  \end{equation*}
  which is exactly the desired invariance, where we remark that since
  the principal symbol $p$ of $P$ is a scalar\footnote{diagonal
    when interpreted as a $N\times N$ matrix of symbols.} and is
  well-defined independently of the choice of frame. 
\end{proof}

The main application in this paper will be to calculate
$S_{sub}(\nabla\cdot \nabla)$ the invariant subprincipal operator of
the Laplace-Beltrami operator acting on sections of the tensor bundle
at particular regions in phase space when conjugated by an appropriate
a zero-order operator. To this end, we consider the following
lemma (see Proposition 3.11 in \cite{hintz_resonance_2017} for a more
general statement).

\begin{lemma}[Proposition 3.11 in \cite{hintz_resonance_2017}]
  \label{linear:lemma:subprincipal-op-def:sub-p-basic-props}
  Let $P\in \Op S^2(X, \mathcal{E}\otimes \Omega^{\frac{1}{2}})$ be a
  pseudo-differential operator with real scalar principal symbol.
  Suppose that $Q\in \Op S^{0}(X, \mathcal{E}\otimes
  \Omega^{\frac{1}{2}})$ is an operator acting on
  $\mathcal{E}$-valued half-densities with principal symbol
  $q$. Then
  \begin{equation*}
    \sigma^{1}(\squareBrace{P,Q}) = \squareBrace*{S_{sub}(P), Q},
  \end{equation*}
  and if $Q$ is elliptic with parametric $Q^-$, then
  \begin{equation*}
    S_{sub}(QPQ^-) = q S_{sub}(P)q^{-1}. 
  \end{equation*}
  In addition,
  \begin{equation*}
    \sigma^{1}(P - P^*) = \frac{1}{2\ImagUnit}\left(S_{sub}(P) - S_{sub}(P)^*\right). 
  \end{equation*}
  In particular then,
  \begin{equation*}
    \sigma^{1}\left(QPQ^- - (QPQ^-)^*\right) = \frac{1}{2\ImagUnit}\left(q S_{sub}(P)q^{-1} - (q S_{sub}(P)q^{-1})^*\right). 
  \end{equation*}
\end{lemma}
\begin{proof}
  See appendix \ref{linear:appendix:lemma:subprincipal-op-def:sub-p-basic-props}.
\end{proof}

As the main application of interest, the invariant subprincipal symbol
of the Laplace-Beltrami operator acting on the bundle
$T_k\mathcal{M} := \otimes^k T^*\mathcal{M}$ of covariant tensors of
rank $k$ has a particularly nice form. 

\begin{lemma}[Proposition 4.1 in \cite{hintz_resonance_2017}]
  \label{linear:lemma:subprincipal-op-def:Laplacian-sub-p-gen-form}
  Let $(\mathcal{M}, g)$ be a smooth manifold equipped with a metric
  tensor\footnote{We do not restrict the signature of the metric.} $g$. 
  Let $\Laplace^{(k)} = \Trace \nabla^2 \in \Diff^2(\mathcal{M},
  T_k\mathcal{M})$ be the Laplace-Beltrami operator on $\mathcal{M}$
  acting on the bundle $T_k\mathcal{M}$.
  Then
  \begin{equation}
    \label{linear:eq:subprincipal-op-def:Laplacian-sub-p-gen-form}
    S_{sub}(\Laplace^{(k)}) = -\ImagUnit \nabla^{\pi^*T_k\mathcal{M}}_{H_G},
  \end{equation}
  where $\pi: T^*\mathcal{M}\backslash 0\to \mathcal{M}$ is the bundle
  projection, and $H_{G}$ is the Hamiltonian vectorfield of $G =
  g^{-1}$. 
\end{lemma}
\begin{proof}
  See appendix \ref{linear:appendix:lemma:subprincipal-op-def:Laplacian-sub-p-gen-form}.
\end{proof}

This in particular gives us a convenient computation regarding the
invariant subprincipal symbol of $\LaplaceAngular^{(k)}$ the
Laplace-Beltrami operator acting on $T_k\Sphere^2$.
\begin{lemma}[Proposition 9.1 in \cite{hintz_global_2018}]
  \label{linear:lemma:subprincipal-op-def:Laplace-on-sphere-computation}
  Let $\pi=\pi_{\Sphere^2}$. Away from the zero section, we split
  \begin{equation*}
    \pi^*T^*\Sphere^2 = E\oplus F, \qquad E(y,\eta) = \Span(\eta),\qquad
    F(y, \eta) = \eta^\perp. 
  \end{equation*}
  Then $\nabla^{\pi^*T_k\Sphere^2}_{H_{\abs*{\eta}^2}}$ is diagonal in
  this splitting in the sense that it preserves both the space of sections
  of $E$ and the space of sections of $F$. 
\end{lemma}
\begin{proof}
  See appendix \ref{linear:appendix:lemma:subprincipal-op-def:Laplace-on-sphere-computation}. 
\end{proof}



\subsubsection{The subprincipal operator of $\LinEinstein_{g_b}$}
\label{linear:sec:subprincipal-operator}

We are now ready to define the microlocal smallness we require at the
trapped set in proving a Morawetz estimate near trapping in Section
\ref{linear:sec:ILED}. 

\begin{lemma}
  \label{linear:lemma:subprincipal-symbol-control}
  Fix $b_0$ black hole parameters of a subextremal member of the
  \SdS{} family and some $\varepsilon_{\TrappedSet_{b_0}} > 0 $.
  Then, there exists a stationary, elliptic $\PseudoSubPFixer\in
  \SymClass^0$ defined microlocally near $\TrappedSet_{b_0}$, with
  parametrix $\PseudoSubPFixer^-$ such that
  in the $(t, r, \omega;\sigma,\xi, \FreqAngular)$
  Boyer-Lindquist coordinates for \SdS,
  \begin{equation}
    \label{linear:eq:subprincipal-symbol-control}
    \left.\frac{1}{2}|\FreqAngular|^{-1}
        \SubPConjSym_{b_0, a}\right\vert_{\TrappedSet_{b_0}}
    < \varepsilon_{\TrappedSet_{b_0}}, 
  \end{equation}
  where the operators 
  \begin{equation}
    \label{linear:eq:SubPConjOp-b0-def}
    \SubPConjOp_{b_0, a} := \frac{1}{2}\left(\SubPConjOp_{b_0} -\SubPConjOp_{b_0}^*\right),\qquad
    \SubPConjOp_{b_0}:=\PseudoSubPFixer\SubPOp_{b_0}\PseudoSubPFixer^-
    + \PseudoSubPFixer\left[\ScalarWaveOp[g_{b_0}], \PseudoSubPFixer^-\right],
  \end{equation}
  have principal symbols $\SubPConjSym_{b_0,a}$ and
  $\SubPConjSym_{b_0}$ respectively, and the adjoint in
  \eqref{linear:eq:SubPConjOp-b0-def} is taken with respect to the
  $L^2([0,\TStar]\times \Sigma)$ Hermitian inner product. 
\end{lemma}

\begin{proof}
  See appendix \ref{linear:appendix:lemma:subprincipal-symbol-control}. 
\end{proof}

\begin{remark}
  In practice we will choose
  $\varepsilon_{\TrappedSet_{b_0}} < \SpectralGap_{b_0}$, so
  that the spectral gap is strong enough to overcome any potentially
  harmful contribution arising from the subprincipal component of the
  operator.
\end{remark}

This smallness at the trapped set in particular implies the following
convenient property of the subprincipal operator in a microlocal
neighborhood of the trapped set $\TrappedSet_b$.
We start with a decomposition lemma that will prove critical in what
follows.
\begin{lemma}
  \label{linear:lemma:ILED-near:s-decomp}
  Fix some $\delta_0>0$, and let $\SubPConjSym_b$ denote the principal
  symbol of the subprincipal operator of $\LinEinstein_{g_b}$
  conjugated by $\PseudoSubPFixer$, the operator in Lemma
  \ref{linear:lemma:subprincipal-symbol-control}. Then, there exists a choice
  of $a$, $\varepsilon_{\TrappedSet_{b_0}}$, $\delta_r$ and
  $\delta_\zeta$ sufficiently small so that 
  on $\TrappingNbhd$ and the support
  of $\mathring{\chi}_{\zeta}$ as defined in \eqref{linear:eq:freq-cutoff-construction}, 
  \begin{equation*}
    \SubPConjSym_{b,a} \in \delta_0\left( S^1(\Sigma) + \sigma S^0(\Sigma)\right).
  \end{equation*}
\end{lemma}
\begin{proof}
  This follows directly by perturbing $\SubPConjSym_b$ and the smoothness
  of the trapped set on $a$.   
\end{proof}

\section{Morawetz estimates}
\label{linear:sec:ILED}

Recall from the discussion in Section \ref{linear:sec:QNM} that the quasinormal modes
(resonances) represent linear obstacles to decay. As a first step in
eliminating these linear obstacles to decay, we will prove the
existence of large, quasinormal-mode-free regions using a resolvent
estimate. 

Resolvent estimates are deeply connected to integrated local energy
decay estimates, and the primary geometric obstruction to proving
either in the current setting is the presence of \textit{trapped null
  geodesics}. Recall that we proved the energy estimates of Section
\ref{linear:sec:energy-estimates} for general strongly hyperbolic operators
without relying on any particular structure on the trapped set.
However, in this section, to deal with the trapped null geodesics, the
structure of $\LinEinstein_{g_b}$ at the trapped set will play a
critical role. It is well known that due to the presence of trapped
null geodesics, one does not expect to derive a full Morawetz
estimate, but instead one that loses derivatives at the trapped set
(see \cite{andersson_hidden_2015, dafermos_decay_2010,
  dafermos_decay_2016, klainerman_global_2020,
  hintz_non-trapping_2014, sbierski_characterisation_2015}).

To capture the loss of derivatives at the trapped set, we need to
define a new energy norm. This new energy norm will reflect that there
are three different regions of $\Sigma$: the redshift region, the
non-trapping region, and the trapping region. This division of
$\Sigma$ corresponds to the fact that the geometric difficulties of
trapping and superradiance are separated in physical space. On both
the redshift and the non-trapping region, we will prove the desired
Morawetz estimate using only physical space methods. On the trapping
region, the frequency-dependent nature of trapping will lead us to use
a pseudo-differentially modified divergence theorem in order to prove
the desired Morawetz estimate (compare with similar work in the
slowly-rotating Kerr case \cite{tataru_local_2010}).

We begin by defining some auxiliary cutoff functions.  To this end, for
fixed $b\in \BHParamNbhd$, let us define
\begin{equation*} 
  r_{\EventHorizonFuture} < \breve{r}_{-}
  < r_{\RedShift, \EventHorizonFuture}
  <r_{0} < \mathring{r}_{-} < \breve{r}_{+}
  <3M < \breve{R}_{-} < \mathring{R}_{ +}
  < R_{0} < R_{\RedShift,\EventHorizonFuture} 
  < \breve{R}_{+} < r_{\CosmologicalHorizonFuture}.
\end{equation*}
We now define the smooth physical cut-off functions
\begin{equation}
  \label{linear:eq:ILED-combine:cutoff-function-list}
  \dot{\chi}(r) := \chi_{\EventHorizonFuture}(r) + \chi_{\CosmologicalHorizonFuture}(r),\quad
  \breve{\chi}(r) := \breve{\chi}_-(r) + \breve{\chi}_+(r), \quad
  \mathring{\chi}(r),
\end{equation}
such that
\begin{equation}
  \label{linear:eq:ILED-combine:cutoff-function-definitions}
  \begin{split}
    \chi_{\EventHorizonFuture}(r)&=
    \begin{cases}
      1&r\in [r_{\EventHorizonFuture}, r_{\RedShift, \EventHorizonFuture}]\\
      0&r\in [r_{0}, r_{\CosmologicalHorizonFuture}]
    \end{cases},\\
    \breve{\chi}_{ -}(r)&=
    \begin{cases}
      1&r\in [r_{\EventHorizonFuture},\breve{r}_+]\\
      0&r\in (3M-\epsilon,r_{\CosmologicalHorizonFuture}]
    \end{cases},\\
    \mathring{\chi}(r)&=
    \begin{cases}
      1&r\in [\mathring{r}_{-}, \mathring{R}_{+}]\\
      0&r\not\in [\mathring{r}_{ -} - \epsilon, \mathring{R}_{ +}+\epsilon]
    \end{cases},\\
    \breve{\chi}_{+}(r)&=
    \begin{cases}
      1&r\in [\breve{R}_{-}, r_{\CosmologicalHorizonFuture}]\\
      0&r\in [r_{\EventHorizonFuture}, 3M+\epsilon) 
    \end{cases},\\
    \chi_{\CosmologicalHorizonFuture}(r)&=
    \begin{cases}
      1&r\in [R_{\RedShift, \CosmologicalHorizonFuture}, r_{\CosmologicalHorizonFuture}]\\ 
      0&r\in [r_{\EventHorizonFuture}, R_{0}]
    \end{cases}.
  \end{split}
\end{equation}
In what follows, we will denote $[r_{\EventHorizonFuture},
r_{\RedShift, \EventHorizonFuture})\bigcup (R_{\RedShift,
  \CosmologicalHorizonFuture}, r_{\CosmologicalHorizonFuture}]$ the
redshift region, $(\breve{r}_-, \breve{r}_+)\bigcup (\breve{R}_-,
\breve{R}_+)$ the nontrapping region, and $(\mathring{r}_{-},
\mathring{R}_{+})$ the trapping region.

\begin{definition}
  \label{linear:def:LE-norm}
  We define the \emph{Morawetz energy norms}
  $\MorawetzNorm^k(\DomainOfIntegration)$ and
  $\InducedMorawetzNorm^k(\Sigma)$ subsequently in Definition
  \ref{linear:def:Morawetz-norm} as $H^k$ norms with a degeneracy at the top
  level of derivatives.

  We define the \emph{local energy norm} by
  \begin{equation*}  
    \norm{h}_{LE^1(\DomainOfIntegration)}^2
    :={}\norm*{\dot{\chi}(r)h}_{H^1(\DomainOfIntegration)}^2
    + \norm*{\breve{\chi}(r)h}_{H^1(\DomainOfIntegration)}^2
    + \norm*{\mathring{\chi}(r)h}_{\MorawetzNorm(\DomainOfIntegration)}^2.
  \end{equation*}
  We also have its Laplace-transformed equivalent for a function $u$ on
  the spacelike slice $\Sigma$:
  \begin{equation*}
    \begin{split}
      \norm{u}_{\CombinedHk{1}(\Sigma)}^2:={}&
      \norm*{\dot{\chi}(r)u}_{\InducedHk{1}_\sigma(\Sigma)}^2
      + \norm*{\breve{\chi}(r)u}_{\InducedHk{1}_\sigma(\Sigma)}^2
      + \norm*{\mathring{\chi}(r)u}_{\InducedMorawetzNorm(\Sigma)}^2.
    \end{split}
  \end{equation*}
  We can likewise define the higher-order $LE^k$ spaces:
  \begin{align*}
    LE^k(\DomainOfIntegration)&:=\curlyBrace*{h: \RedShiftK^\alpha h\in LE^1(\DomainOfIntegration), \abs*{\alpha}\le k-1},\\
    \CombinedHk{k}(\Sigma) &:=\curlyBrace*{h: \RedShiftK^\alpha h\in \CombinedHk{1}(\Sigma), \abs*{\alpha}\le k-1}.
  \end{align*}
\end{definition}

Using this combined local energy norm, we are now ready to state the main
theorem of this section. 
\begin{theorem}
  \label{linear:thm:resolvent-estimate:main}
  Let $g_b$ be a fixed slowly-rotating \KdS{} background, and 
  $u\in H^k(\Sigma)$. Then, 
  for $k_0$ as defined in \eqref{linear:eq:threshold-reg-def}, there exists
  some $\SpectralGap>0$, $C_0>0$ such that for $k>k_0$, 
  \begin{equation}
    \label{linear:eq:resolvent-estimate:main}
    \norm{u}_{\CombinedHk{k}(\Sigma)}
    \lesssim \norm*{\widehat{\LinEinstein}_{g_b}(\sigma)u}_{\InducedHk{k-1}_\sigma(\Sigma)},
    \qquad
    \text{if }
    \Im\sigma=-\SpectralGap,\text{ or }
    \Im\sigma>-\SpectralGap, |\sigma|\ge C_0,
  \end{equation}
  for all $u$ where the norms on both sides are finite.
  
  Moreover, by adjusting the value of $\SpectralGap$ as necessary, the
  only poles of $\widehat{\LinEinstein}_{g_b}(\sigma)^{-1}$ which satisfy
  $\Im \sigma> -\SpectralGap$ in fact also satisfy $\Im \sigma \ge 0$.
\end{theorem}

\begin{remark}
  The proof of Theorem \ref{linear:thm:resolvent-estimate:main} relies
  heavily on the fact that $\LinEinstein_{g_b}$ is strongly
  hyperbolic, and the fact that we have good estimates for the scalar
  wave operator on the \SdS{} background $g_{b_0}$. Indeed, the proofs
  in the ensuing section should be thought of as perturbations of the
  arguments that can be used in the $b = b_0 = (M, 0)$ case.
\end{remark}

As alluded to by the construction of the cutoff functions at the
beginning of the section, we divide the proof of Theorem
\ref{linear:thm:resolvent-estimate:main} into three parts, corresponding to
the redshift region, the nontrapping region, and the trapping region.
We prove a resolvent estimate separately for solutions supported in
each of these regions in Sections \ref{linear:sec:ILED:redshift},
\ref{linear:sec:ILED:nontrapping}, \ref{linear:sec:ILED:nontrapping-freq}, and
\ref{linear:ILED:near} before showing that they can be appropriately combined
to yield a resolvent estimate on the whole spacetime in Section
\ref{linear:ILED:full}.

The main idea behind proving Theorem \ref{linear:thm:resolvent-estimate:main}
will be to first prove a Morawetz estimate for functions $h$ with
appropriate support, using well-chosen vectorfield multipliers
$(X,q,m)$ and the divergence theorem in
\eqref{linear:eq:div-them:J-K-currents}. We then pass from the Morawetz
estimate to a resolvent estimate using the following basic outline.
Let us assume that we have chosen $(X,q,m)$ appropriately so as to
arrive at the following inequality
\begin{equation*}
  \evalAt*{\int_{\breve{\Sigma}_{\tStar}}\JCurrent{X,q,m}[h]\cdot n_{\breve{\Sigma}_{\tStar}}}_{\tStar=0}^{\tStar=\TStar}
  + \int_{\DomainOfIntegration}\KCurrent{X,q,m}[h]
  \lesssim \norm*{\LinEinstein_{g_b}h}_{L^2(\DomainOfIntegration)}^2 + \norm*{h}_{L^2(\DomainOfIntegration)}^2,
\end{equation*}
where we will impose conditions on $\JCurrent{X,q,m}[h]$ and
$\KCurrent{X,q,m}[h]$ subsequently. 
Differentiating by $\p_{\tStar}$, we then have that 
\begin{equation}
  \label{linear:eq:ILED-outline:Morawetz-to-resolvent:differentiated-eq}
  \p_{\tStar}\int_{\breve{\Sigma}_{\tStar}}\JCurrent{X,q,m}[h]\cdot n_{\breve{\Sigma}_{\tStar}}
  + \int_{\breve{\Sigma}_{\tStar}}\KCurrent{X,q,m}[h]\,\sqrt{\GInvdtdt}
  \lesssim \norm*{\LinEinstein_{g_b}h}_{\InducedLTwo(\breve{\Sigma}_{\tStar})}^2 + \norm*{h}_{\InducedLTwo(\breve{\Sigma}_{\tStar})}^2. 
\end{equation}
Substituting in $h=e^{-\ImagUnit \sigma \tStar}u(x)$, we see that
\eqref{linear:eq:ILED-outline:Morawetz-to-resolvent:differentiated-eq}
reduces to
\begin{align*}
  &2\Im\sigma\int_{\breve{\Sigma}_{\tStar}}\JCurrent{X,q,m}\left[e^{-\ImagUnit\sigma\tStar}u\right]\cdot n_{\breve{\Sigma}_{\tStar}}
  + \int_{\breve{\Sigma}_{\tStar}}\KCurrent{X,q,m}\left[e^{-\ImagUnit\sigma\tStar}u\right]\,\sqrt{\GInvdtdt}\\
  \lesssim{}& \norm*{\LinEinstein_{g_b}e^{-\ImagUnit\sigma\tStar}u}_{\InducedLTwo(\breve{\Sigma}_{\tStar})}^2 + \norm*{e^{-\ImagUnit\sigma\tStar}u}_{\InducedLTwo(\breve{\Sigma}_{\tStar})}^2. 
\end{align*}
Multiplying both sides eliminates any $\tStar$-dependency in the equation, so
that 
\begin{equation}
  \label{linear:eq:ILED-outline:final-aux}
  2\Im\sigma\int_{\breve{\Sigma}_{\tStar}}\JLaplaceCurrent{X,q,m}\left[u\right]\cdot n_{\breve{\Sigma}_{\tStar}}
  + \int_{\breve{\Sigma}_{\tStar}}\KLaplaceCurrent{X,q,m}\left[u\right]\,\sqrt{\GInvdtdt}
  \lesssim \norm*{\widehat{\LinEinstein}_{g_b}(\sigma)u}_{\InducedLTwo(\breve{\Sigma}_{\tStar})}^2
  + \norm*{u}_{\InducedLTwo(\breve{\Sigma}_{\tStar})}^2.  
\end{equation}
To reduce to a desired resolvent estimate, we proceed in two steps. 
\begin{enumerate}
\item We will choose $(X,q,m)$ so that
  \begin{equation*}
    \KCurrent{X,q,m}[h] \ge 0,
  \end{equation*}
  which implies that 
  \begin{equation}
    \label{linear:eq:ILED-outline:K-positivity}
    \KLaplaceCurrent{X,q,m}\left[u\right] \ge 0, 
  \end{equation}
  and such that for $\GronwallExp\ge \Im\sigma
  \ge -\SpectralGap$,
  \begin{equation}
    \label{linear:eq:ILED-outline:J-K-compare}
    4\abs*{\Im\sigma\JLaplaceCurrent{X,q,m}\left[u\right]\cdot n_{\breve{\Sigma}_{\tStar}}}
    < \KLaplaceCurrent{X,q,m}\left[u\right]\,\sqrt{\GInvdtdt} + \abs*{u}^2
  \end{equation}
  Equations \eqref{linear:eq:ILED-outline:K-positivity} and
  \eqref{linear:eq:ILED-outline:J-K-compare} then imply that the left-hand
  side of \eqref{linear:eq:ILED-outline:final-aux} is positive (up to a
  lower-order term), and we have
  \begin{equation}
    \label{linear:eq:ILED-outline:after-absorbing-boundary}
    \int_{\breve{\Sigma}_{\tStar}}\KLaplaceCurrent{X,q,m}\left[u\right]\,\sqrt{\GInvdtdt}
    \lesssim \norm*{\widehat{\LinEinstein}_{g_b}(\sigma)u}_{\InducedLTwo(\breve{\Sigma}_{\tStar})}^2
    + \norm*{u}_{\InducedLTwo(\breve{\Sigma}_{\tStar})}^2.  
  \end{equation}
\item To absorb the lower-order term
  $\norm*{u}_{L^2(\breve{\Sigma}_{\tStar})}$ on the right-hand side of
  \eqref{linear:eq:ILED-outline:final-aux}, we will use the high-frequency
  condition in the resolvent estimates $\abs*{\sigma}\ge C_0$ for some
  $C_0$ sufficiently large. If $(X,q,m)$ are chosen in such a manner
  that
  \begin{equation*}
    \abs*{\sigma u}^2 \lesssim \KLaplaceCurrent{X,q,m}[u],
  \end{equation*}
  then for sufficiently large $C_0$, we have that the $L^2$ term on
  the right-hand side of
  \eqref{linear:eq:ILED-outline:after-absorbing-boundary} will be absorbed
  into the left-hand side, and we are left with exactly the desired
  resolvent estimate subject to the conditions $\GronwallExp \ge
  \Im\sigma \ge -\SpectralGap$ and $\abs*{\sigma}\ge C_0$ for some
  $C_0$ sufficiently large. 
\end{enumerate}
\begin{remark}
  In the trapping case of Section \ref{linear:sec:ILED-trapping:KdS}, it is
  not as direct to prove the desired resolvent estimate, due to the
  degeneracy at the top level of derivatives in $\KCurrent{X,q,m}[h]$,
  and we will in fact have to prove two separate resolvent estimates,
  one on $\abs*{\Im\sigma}\le \SpectralGap$, and one on $\Im\sigma >
  \frac{\SpectralGap}{2}$, but the core ideas remain the same. 
\end{remark}
\begin{remark}
  The application of Morawetz estimates to functions ${h}
  =e^{-\ImagUnit\sigma\tStar}u$ is the main tool for deriving the
  desired resolvent estimates. The choice ${h}
  =e^{-\ImagUnit\sigma\tStar}u$ allows us to handle the
  boundary terms by directly absorbing them into the bulk. This should
  be contrasted with the typical approach in the vectorfield method,
  which is to use a Killing energy estimate to control the boundary
  terms.   
\end{remark}

\subsection{Redshift region}
\label{linear:sec:ILED:redshift}

\begin{theorem}
  \label{linear:thm:ILED-redshift:main}
  Let $g$ be a slowly rotating \KdS{} metric, and define
  \begin{equation}
    \label{linear:eq:ILED-redshift:redshift-reg-def}
    \RedShiftReg :=
    \Sigma\left([r_{\EventHorizonFuture},r_{\RedShift,\EventHorizonFuture}]\right) \bigcup
    \Sigma\left([r_{\RedShift,\CosmologicalHorizonFuture},r_{\CosmologicalHorizonFuture}]\right).    
  \end{equation}  
  Then there exists a choice of
  $r_{\RedShift, \EventHorizonFuture}, r_{\RedShift,
    \CosmologicalHorizonFuture}$ 
  and some $\SpectralGap, C_0 >0$ such that for $k>k_0$, where $k_0$
  is as defined in \eqref{linear:eq:threshold-reg-def},
  \begin{equation}
    \label{linear:eq:ILED-redshift:main:resolvent-estimate}
    \norm{u}_{\InducedHk{k}_\sigma(\RedShiftReg)}
    \lesssim \norm{\widehat{\LinEinstein}_{g_b}(\sigma)u}_{\InducedHk{k-1}_\sigma(\RedShiftReg)},
    \qquad \text{if }\Im\sigma\ge -\SpectralGap,\text{ and } \abs*{\sigma}\ge C_0.
  \end{equation}
\end{theorem}

\begin{proof}
  We consider \eqref{linear:eq:redshift:morawetz-est-aux} with
  $\varepsilon_{\RedShiftN}$ and $\epsilon$ fixed such that
  \begin{equation*}
    \varepsilon_{\RedShiftN} < \frac{1}{4}\max_{\Horizon= \EventHorizonFuture, \CosmologicalHorizonFuture}\SurfaceGravity_{\Horizon},\qquad \epsilon = \frac{1}{4}\max_{\Horizon= \EventHorizonFuture, \CosmologicalHorizonFuture}\SurfaceGravity_{\Horizon},
  \end{equation*}
  we have that for ${h}$ supported on $\DomainOfIntegration$,
  \begin{equation*}
    \label{linear:eq:ILED-redshift:main:ILED}
    \p_{\tStar}\RedShiftEnergy(\tStar)[{h}] - \max_{\Horizon=\EventHorizonFuture, \CosmologicalHorizonFuture}\left(\SHorizonControl{\LinEinstein}[\Horizon]-\frac{1}{2}\SurfaceGravity_{\Horizon}\right)\RedShiftEnergy(\tStar)[{h}]
    \lesssim \norm{\LinEinstein_{g_b}{h}}_{\InducedLTwo(\Sigma_{\tStar})}^2 + \norm*{h}_{\InducedLTwo(\Sigma_{\tStar})}^2.     
  \end{equation*}

  To proceed from this version of the redshift estimate to a resolvent
  estimate, we consider   
  $h=e^{-\ImagUnit\sigma\tStar}u$, 
  \begin{equation*}
    \left(2\Im\sigma -  \max_{\Horizon=\EventHorizonFuture, \CosmologicalHorizonFuture}\left(\SHorizonControl{\LinEinstein_{g_b}}[\Horizon]-\frac{1}{2}\SurfaceGravity_{\Horizon}\right)\right)\RedShiftEnergy(\tStar)[e^{-\ImagUnit\sigma\tStar}u]
    \lesssim \norm{\LinEinstein_{g_b} e^{-\ImagUnit\sigma\tStar}u }_{\InducedLTwo(\Sigma_{\tStar})}^2 + \norm{e^{-\ImagUnit\sigma\tStar}u}_{\InducedLTwo(\Sigma_{\tStar})}^2.
  \end{equation*}
  If
  \begin{equation}
    \label{linear:eq:ILED-redshift:Im-sigma-condition}
    2\Im\sigma>\max_{\Horizon=\EventHorizonFuture,
    \CosmologicalHorizonFuture}\left(\SHorizonControl{\LinEinstein_{g_b}}[\Horizon]-\frac{1}{2}\SurfaceGravity_{\Horizon}\right),
  \end{equation}
  then the left-hand side of the above equation is
  positive. Multiplying both sides by $e^{2\Im\sigma\tStar}$ to cancel
  out any $\tStar$ dependency, we then have using
  \ref{linear:item:RedShiftEstimate:one} of Theorem
  \ref{linear:thm:redshift-energy-estimate} that
  \begin{align*}
    &\left(2\Im\sigma -  \max_{\Horizon=\EventHorizonFuture, \CosmologicalHorizonFuture}\left(\SHorizonControl{\LinEinstein_{g_b}}[\Horizon]-\frac{1}{2}\SurfaceGravity_{\Horizon}\right)\right)\left(
      \norm*{u}_{\InducedHk{1}(\Sigma)}^2 + \norm*{\sigma u}_{\InducedLTwo(\Sigma)}^2
    \right)\\
    \lesssim{}& \norm*{\widehat{\LinEinstein}_{g_b}(\sigma)u }_{\InducedLTwo(\Sigma)}^2 + \norm{u}_{\InducedLTwo(\Sigma_{\tStar})}^2.
  \end{align*}
  Then, if \eqref{linear:eq:ILED-redshift:Im-sigma-condition} is satisfied,
  and if $\abs*{\sigma} > C_0$ for some $C_0$ large enough, the
  $\norm{u}_{\InducedLTwo(\Sigma_{\tStar})}$ term on the right-hand
  side can be absorbed by the
  $\norm*{\sigma u}_{\InducedLTwo(\Sigma)}$ term on the left-hand
  side, yielding the desired resolvent estimate
  \eqref{linear:eq:ILED-redshift:main:resolvent-estimate} in the $k=1$ case.
  
  To prove the higher-order estimates, we can repeat the derivation of
  \eqref{linear:eq:ILED-redshift:main:ILED} for $\bL^{(k)}_{g_b}$ in place of
  $\LinEinstein$, where $\bL^{(k)}_{g_b}$ is the strongly hyperbolic
  operator constructed from $\LinEinstein_{g_b}$ after $k$
  commutations with the vectorfields $\{\RedShiftK_i\}_{i=1}^N$, as in
  Theorem \ref{linear:thm:redshift-commutation:main}.  Thus, to conclude, we
  need only show that for $k$ sufficiently large, in particular larger
  than $k_0$,
  $\max_{\Horizon=\EventHorizonFuture,
    \CosmologicalHorizonFuture}\left(\SHorizonControl{\bL^{(k)}_{g_b}}[\Horizon]-\frac{1}{2}\SurfaceGravity_{\Horizon}\right) < 0$.  But precisely from Theorem
  \ref{linear:thm:redshift-commutation:main}, we know that
  \begin{equation*}
    \SHorizonControl{\bL^{(k)}}[\Horizon] = \SHorizonControl{\LinEinstein_{g_b}}[\Horizon] - 2k\SurfaceGravity_{\Horizon}.
  \end{equation*}
  Thus, for $k > k_0$,
  \begin{equation*}
    \max_{\Horizon=\EventHorizonFuture,
    \CosmologicalHorizonFuture}\left(\SHorizonControl{\bL^{(k)}_{g_b}}[\Horizon]-\frac{1}{2}\SurfaceGravity_{\Horizon}\right)
  < 0,
  \end{equation*}
  as desired.
\end{proof}

\subsection{Nontrapping region}
\label{linear:sec:ILED:nontrapping}

The resolvent estimates away from trapping will be proven using the
following vectorfield 
\begin{equation}
  \label{linear:eq:ILED-nontrapping:XOuter-fOuter-def}
  \MorawetzOuterVF = \fOuter(r)\HprVF,
  \qquad \fOuter(r) = e^{\COuter(r-3M)^2}(r-3M)\Delta,
\end{equation}
where $\HprVF$ is as defined in \eqref{linear:eq:KdS:Rhat-That-def}.

\begin{theorem}
  \label{linear:thm:ILED-nontrapping:resolvent-estimate:main}
  Let $g$ be a fixed slowly-rotating \KdS{} background, and ${h}$ a
  compactly supported function on
  \begin{equation}
    \label{linear:eq:nontrapping-reg-def}
    \DomainOfIntegration:=\Real^+_{\tStar}\times\NonTrappingReg,\qquad
    \NonTrappingReg:=\Sigma\left((\breve{r}_-,\breve{r}_+)\bigcup(\breve{R}_-,
    \breve{R}_+)\right).
  \end{equation}  
  Then
  for $k_0$ the threshold regularity level defined in
  \eqref{linear:eq:threshold-reg-def}, there
  exists $\SpectralGap>0$, $C_0>0$ such that for $k>k_0$
  \begin{equation}
    \label{linear:eq:ILED-nontrapping:resolvent-estimate:main}
    \norm{u}_{\InducedHk{k}_\sigma(\NonTrappingReg)}
    \lesssim \norm{\widehat{\LinEinstein}(\sigma)u}_{\InducedHk{k-1}_\sigma(\NonTrappingReg)},\qquad
    \text{if }\GronwallExp\ge\Im\sigma \ge -\SpectralGap,\quad \abs*{\sigma} \ge C_0.
  \end{equation}
\end{theorem}

Similar to the approach taken to proving the resolvent estimates in
Theorem \ref{linear:thm:ILED-redshift:main} and Corollary
\ref{linear:corollary:naive-energy-estimate}, we will first use an energy
estimate of the form

There are two components to proving the desired Morawetz estimate. The
main difficulty will be the positive bulk term, which will come down
to fine-tuning the choice of $\COuter$ as well as the choice of
$\LagrangeCorrOuter$. After making these choices we will have to
handle the boundary terms. In the literature, the typical method for
handling boundary terms is to add a large amount of the standard
$\p_{\tStar}$-energy estimate to the Morawetz estimate to prove an
energy-Morawetz estimate.

We will handle the
boundary terms in a different way, proving
\eqref{linear:eq:ILED-nontrapping:resolvent-estimate:main} 
by absorbing the boundary terms directly
into positivity of the bulk term. We begin with the bulk term.

\begin{lemma}
  \label{linear:lemma:ILED-nontrapping:bulk-positivity}
  Fix a constant $C>0$ and a slowly-rotating \KdS{} metric
  $g$. Then there exists a choice of $\COuter$ and $\LagrangeCorrOuter$ so
  that for $\MorawetzOuterVF$ defined as in 
  (\ref{linear:eq:ILED-nontrapping:XOuter-fOuter-def}) and for ${h}$ compactly
  supported on $\NonTrappingReg$, there exists some $C_1$ such that
  \begin{equation}
    \label{linear:eq:ILED-nontrapping:bulk-positivity}
    \KCurrent{\MorawetzOuterVF, \LagrangeCorrOuter, 0}[{h}]
    - \Re\left[\MorawetzOuterVF{h} \cdot \SubPOp[\overline{{h}}]\right]
    > C\abs*{\nabla{h}}^2 - C_1\abs*{h}^2. 
  \end{equation}
\end{lemma}
\begin{remark}
  When proving a high-frequency Morawetz estimate for the scalar wave, the
  positivity requirement is merely that 
  \begin{equation*}
    \KCurrent{\MorawetzOuterVF,\LagrangeCorrOuter, 0}[{h}]>c\abs*{\nabla{h}}^2 - C_1\abs*{h}^2
  \end{equation*}
  for some $c>0$. Our requirement here that $C$ can be any positive
  real number, and that, in particular, it can be \textit{arbitrarily
    large} is due to the presence of the subprincipal operator, which
  does not carry a sign, and for which we do not have any control. In
  its absence, or in the case that it has a good sign, such strong
  positivity in the Morawetz estimate would not be necessary. To
  accommodate the subprincipal term with no sign, we need to prove
  \eqref{linear:eq:ILED-nontrapping:bulk-positivity} with $C>\sup \SubPOp$.
\end{remark}

\begin{proof}
  We first calculate the deformation tensor associated to the
  vectorfield multiplier $\fOuter(r)\HprVF$,
  \begin{equation*}
    2 \DeformationTensor{\MorawetzOuterVF}{^{\alpha\beta}}
    = \p_r\fOuter(r)g^{r(\beta}\HprVF^{\alpha)} + \fOuter(r)g^{r(\beta}\HprVF'^{\alpha)} - \fOuter(r)\HprVF g^{\alpha\beta},
  \end{equation*}
  where we recall the definition of $\HprVF$ and $\HprVF'$ in
  \eqref{linear:eq:KdS:Rhat-That-def}. 
  Then
  \begin{equation*}
    \begin{split}
      \KCurrent{\MorawetzOuterVF, 0, 0}[{h}]
      ={}&
      \frac{1}{2}\left(\Re\left(\left(\fOuter(r)\HprVF'{h} + \p_r\fOuter(r)\HprVF{h}\right) \cdot g^{r\gamma}\p_\gamma\overline{{h}}\right)
      - \fOuter(r)(\HprVF g)^{\alpha\beta}\p_{(\alpha}{h}\cdot  \p_{\beta)}\overline{{h}}\right)
       - (\nabla_g\cdot \MorawetzOuterVF )\p^\gamma{h}\cdot\p_\gamma \overline{{h}}.
    \end{split}
  \end{equation*}
  Let us rewrite this as
  \begin{equation}
    \label{linear:eq:ILED-nontrapping:K-XOuter-0:1}
    \begin{split}
      \rho^{2}\KCurrent{\MorawetzOuterVF, 0, 0}[{h}]
      ={}&   \frac{1}{2}\left(\Re\left(\fOuter(r)\HprVF'{h}\cdot  (\rho^2g)^{r\gamma}\p_\gamma\overline{{h}}
      +\p_r\fOuter(r)\HprVF{h}\cdot  (\rho^2g)^{r\gamma}\p_\gamma\overline{{h}}\right)
      - \fOuter(r)(\HprVF(\rho^2g))^{\alpha\beta}\p_{(\alpha}{h}\cdot  \p_{\beta)}\overline{{h}}\right)\\
       &- \left( \rho^2\nabla_g\cdot \MorawetzOuterVF - \rho \fOuter(r) \HprVF\rho\right) \p^\gamma{h}\cdot\p_\gamma \overline{{h}}.
    \end{split}
  \end{equation}
  We first ensure that the first line of equation
  (\ref{linear:eq:ILED-nontrapping:K-XOuter-0:1}) is positive. The
  second line will be handled by the choice of Lagrangian correction.

  Using the definition of $\fOuter(r)$ in equation
  (\ref{linear:eq:ILED-nontrapping:XOuter-fOuter-def}), we have that 
  \begin{align}
    \label{linear:eq:ILED-nontrapping:aux1}
      &\Re\left(\fOuter(r)\HprVF'{h}\cdot  (\rho^2g)^{r\gamma}\p_\gamma\overline{{h}}
      +\p_r\fOuter(r)\HprVF{h}\cdot  (\rho^2g)^{r\gamma}\p_\gamma\overline{{h}}\right)
      - \fOuter(r)(\HprVF(\rho^2g))^{\alpha\beta}\p_{(\alpha}{h}\cdot  \p_{\beta)}\overline{{h}}\notag \\
      ={}& e^{\COuter (r-3M)^2}\left(
        (2\COuter(r-3M)^2+1)\Delta^2\abs*{\HprVF h}^2 - \Re\left(\frac{r-3M}{\Delta}\HawkingVF {h}   \cdot \left(
          \p_r\Delta\HawkingVF\overline{{h}} -  2 \Delta\HawkingVF'\overline{{h}}
        \right)\right)\right) ,    
  \end{align}
  where we recall the definition of $\HawkingVF, \HawkingVF'$ from
  \eqref{linear:eq:KdS:Rhat-That-def}.
  
  Since $\abs*{r-3M}$ and $\Delta$ are bounded
  from below on $\NonTrappingReg$, there exists some choice of $\COuter$ sufficiently
  large so that on $\NonTrappingReg$, 
  \begin{equation*}
    2(\COuter(r-3M)^2 + 1)\Delta^2 > 0.
  \end{equation*}
  Thus the
  coefficient of $\abs*{\HprVF{h}}^2$ in
  \eqref{linear:eq:ILED-nontrapping:aux1} is positive. We move on to
  dealing with the last term in \eqref{linear:eq:ILED-nontrapping:aux1}. 
  Observe that
  \begin{equation*}
    r^2\p_r\Delta_{b_0} - 4r\Delta_{b_0}  = -2r^2(r-3M), \qquad
    \HawkingVF' = \frac{2}{r}\HawkingVF + O(a)(\p_{t_0}, \p_{\varphi_0}).
  \end{equation*}
  As a result, we can write
  \begin{equation*}
    \begin{split}
      -(r-3M)\Delta^{-1}\HawkingVF{h}\cdot \left(\p_r \Delta
        \HawkingVF\overline{{h}} -
        2\Delta\HawkingVF'\overline{{h}}\right)
      ={}& 2(r-3M)^2\Delta^{-1}\abs*{\HawkingVF{h}}^2
      + \Delta^{-1}O(a)\left(\abs*{\p_{t_0}{h}}^2 + \abs*{\p_{\varphi_0}{h}}^2\right).
    \end{split}
  \end{equation*}
  Now, in order to handle the terms on the second line of equation
  (\ref{linear:eq:ILED-nontrapping:K-XOuter-0:1}), we define
  \begin{equation*}
    \LagrangeCorrOuter_0
    =  \nabla_g\cdot \MorawetzOuterVF + \rho^{-1} \fOuter(r) \HprVF\rho. 
  \end{equation*}
  Recall that we have 
  \begin{equation*}
    \KCurrent{0, \LagrangeCorrOuter_0, 0}[{h}]
    = \LagrangeCorrOuter_0\p^\gamma{h}\cdot\p_\gamma\overline{{h}}
    - \frac{1}{2}\nabla^\gamma \p_\gamma \LagrangeCorrOuter_0\abs*{{h}}^2.
  \end{equation*}
  As a result, 
  \begin{equation*}
    \begin{split}
      2\KCurrent{\MorawetzOuterVF, \LagrangeCorrOuter_0, 0}[{h}]
      ={}& \rho^{-2}e^{\COuter (r-3M)^2}\left(
        \left(2\COuter(r-3M)^2+ 1\right)\abs*{\Delta\HprVF{h}}^2
        +\Delta^{-1}(r-3M)^2\abs*{\HawkingVF{h}}^2\right)\\
      &+\rho^{-2}e^{\COuter (r-3M)^2}\left(
        O(a)\left(\abs*{\p_{t_0}{h}}^2 + \abs*{\p_{\varphi_0}{h}}^2\right)
      \right)
      - \nabla^\alpha\p_\alpha \LagrangeCorrOuter_{0} \abs*{{h}}^2.
    \end{split}
  \end{equation*}
  By choosing $\COuter$ sufficiently large, we can in particular
  guarantee that there exists some $C_2$ such that 
  \begin{equation*}
    \KCurrent{\MorawetzOuterVF, \LagrangeCorrOuter_0, 0}[{h}]
    \ge C\abs*{\HprVF {h}}^2
    + C\abs*{\HawkingVF{h}}^2
    + O(a)\left(\abs*{\p_{t_0}{h}}^2 + \abs*{\p_{\varphi_0}{h}}^2\right)
    - C_2\abs*{{h}}^2.
  \end{equation*}
  There are two remaining issues. The first is the $O(a)$ errors. The
  second is that we do not have control over all the derivatives of
  ${h}$. For instance, control over $\HprVF{h}$ and $\HawkingVF {h}$
  do not yield any control over $\CartarOp{h}$. We fix these two
  problems at the same time by ``borrowing'' some positivity from
  $\HprVF h$ and $\HawkingVF h$ using the Lagrangian correction. To
  this end, consider
  \begin{equation}
    \label{linear:eq:ILED-nontrapping:q-Outer-def}
    \LagrangeCorrOuter = \LagrangeCorrOuter_0 + \LagrangeCorrOuter_1,\qquad 
    \LagrangeCorrOuter_1 =\delta_1e^{\COuter(r-3M)^2}.
  \end{equation}
  Up to lower order terms,
  \begin{align}    
    \rho^2\KCurrent{\MorawetzOuterVF, \LagrangeCorrOuter,0}[h]
    ={}&e^{\COuter(r-3M)^2}\left(\delta_1\rho^2\abs*{\CartarOp{h}}^2
         + \left(2\COuter(r-3M)^2 + 1\right)\abs*{\Delta\HprVF{h}}^2
         +\frac{(r-3M)^2-\delta_1}{\Delta}\abs*{\HawkingVF{h}}^2
         \right)\notag \\
       &+ e^{\COuter(r-3M)^2}\left(\delta_1\Delta\abs*{\HprVF {h}}^2
         + O(a)\left(\abs*{\p_{t_0}{h}}^2 + \abs*{\p_{\varphi_0}{h}}^2\right)\right), \label{linear:eq:ILED-nontrapping:K-control:aux1}
  \end{align}
  where we recall the definition of $\CartarOp h$ in \eqref{linear:eq:CartarOp:def}.
  
  For sufficiently large $\COuter$ and sufficiently small $\delta_1$,
  $a$, we have that $(r-3M)^2-\delta_1>0$ on $\breve{\Sigma}$, and the
  $O(a)$ errors are entirely controlled, so there exists some
  $\delta>0$ and $C>0$ such that 
  \begin{equation}
    \label{linear:eq:ILED-nontrapping:K-control}
    \KCurrent{\MorawetzOuterVF, \LagrangeCorrOuter,0}[h]
    > \delta e^{\COuter(r-3M)^2}\left(\COuter \abs*{\HprVF{h}}^2  + \abs*{\HawkingVF h}^2 + \abs*{\CartarOp h}^2\right)
    - C\abs*{{h}}^2.
  \end{equation}
  We next consider the contribution of the subprincipal
  operator. Using Cauchy-Schwarz, we observe that
  \begin{equation}
    \label{linear:eq:ILED-nontrapping:S-control:CS}
    \abs*{\SubPOp[h]\cdot \MorawetzOuterVF \overline{h}}
    \le e^{\COuter(r-3M)^2}\left(
      \epsilon |\nabla h|^2
      + \epsilon^{-1}\abs*{(r-3M)\HprVF h}^2\right). 
  \end{equation}
  Thus for sufficiently small $\epsilon$ and sufficiently large
  $\COuter$, we have that up to lower-order terms,
  \begin{equation*}
    \abs*{\SubPOp[h]\cdot \MorawetzOuterVF \overline{h}}
    \le \frac{1}{2}\KCurrent{\MorawetzOuterVF, \LagrangeCorrOuter,0}[h].
  \end{equation*}
  Thus, we have that for ${h}$ supported on
  $\NonTrappingReg$, up to lower-order terms,
  \begin{equation*}
    \KCurrent{\MorawetzOuterVF, \LagrangeCorrOuter,0}[h]
    - \Re\left[\SubPOp[{h}]\cdot \MorawetzOuterVF\overline{{h}}\right]
    \ge e^{\COuter(r-3M)^2}\frac{\delta}{2}\abs*{\nabla{h}}^2.
\  \end{equation*}
  We conclude the proof of Lemma
  \ref{linear:lemma:ILED-nontrapping:bulk-positivity} by further increasing
  $\COuter$ as necessary so that
  $\frac{\delta}{2}e^{\COuter\delta_r^2}>C$.
\end{proof}

We now show how to close the proof of Theorem
\ref{linear:thm:ILED-nontrapping:resolvent-estimate:main} in the $k=1$ case
given the positivity of the bulk term in Lemma
\ref{linear:lemma:ILED-nontrapping:bulk-positivity}.

\begin{proof}[Proof of Theorem
  \ref{linear:thm:ILED-nontrapping:resolvent-estimate:main} for $k=1$]
  For ${h}$ supported on
  $\DomainOfIntegration$ as defined in \eqref{linear:eq:nontrapping-reg-def},
  we define the Morawetz energy on a $\tStar$-constant spacelike slice
  by
  \begin{equation*}
    \breve{\MorawetzEnergy}(\tStar)[{h}]
    = \int_{\NonTrappingReg_{\tStar}}\JCurrent{\MorawetzOuterVF, \LagrangeCorrOuter, 0}[{h}] \cdot n_{\NonTrappingReg}.
  \end{equation*}
  It is clear upon inspection that
  \begin{equation*}
    \breve{\MorawetzEnergy}(\tStar)[{h}]
    \lesssim \norm{{h}}_{\InducedHk{1}(\NonTrappingReg)}^2
    + \norm{\KillT h}_{\InducedLTwo(\NonTrappingReg)}^2 ,
  \end{equation*}
  where we emphasize that unlike $\RedShiftEnergy(\tStar)$,
  $\MorawetzEnergy(\tStar)[h]$ does not necessarily have a sign.  We
  now apply \eqref{linear:eq:div-them:J-K-currents} to
  $\JCurrent{\MorawetzOuterVF, \LagrangeCorrOuter, 0}[{h}]$, to see
  that
  \begin{equation*}
    \nabla\cdot \JCurrent{\MorawetzOuterVF, \LagrangeCorrOuter, 0}[{h}]
    = \Re\squareBrace*{
      \MorawetzOuterVF\overline{{h}}\cdot
      \ScalarWaveOp[g]{h}
    }
    + \KCurrent{\MorawetzOuterVF, \LagrangeCorrOuter, 0}[{h}].
  \end{equation*}
  Since we have assumed ${h}$ with compact support on
  $\NonTrappingReg$,  ${h}$ vanishes on
  the horizons, and
  \begin{equation*}
    \int_{\EventHorizonFuture\bigcap \Sigma_{\tStar}}
    \JCurrent{\MorawetzOuterVF, \LagrangeCorrOuter, 0}[{h}]\cdot n_{\EventHorizonFuture}
    = \int_{\CosmologicalHorizonFuture\bigcap \Sigma_{\tStar}}
    \JCurrent{\MorawetzOuterVF, \LagrangeCorrOuter, 0}[{h}] \cdot n_{\CosmologicalHorizonFuture} = 0. 
  \end{equation*}
  Applying the divergence equation
  \eqref{linear:eq:div-them:J-K-currents}, 
  \begin{align}    
    &\left.\breve{\MorawetzEnergy}(\tStar)[{h}]\right\vert_{\tStar=0}^{\tStar= \TStar}
      + \int_{\DomainOfIntegration}\KCurrent{\MorawetzOuterVF, \LagrangeCorrOuter, 0}[{h}]
      - \Re \bangle*{\MorawetzOuterVF h , \SubPOp[{h}]}_{L^2(\DomainOfIntegration)}\notag \\
    ={}&  -\Re\bangle*{(\MorawetzOuterVF+\LagrangeCorrOuter) {h}, \LinEinstein{h}}_{L^2(\DomainOfIntegration)}
         + \Re\bangle*{(\MorawetzOuterVF+\LagrangeCorrOuter) {h}, \PotentialOp{h}}_{L^2(\DomainOfIntegration)}
         + \Re \bangle*{\LagrangeCorrOuter{h},\SubPOp[{h}]}_{L^2(\DomainOfIntegration)}.     \label{linear:eq:ILED-nontrapping:div:aux1}
  \end{align}
  Note that at this point, a Morawetz estimate would follow
  immediately
  from Lemma
  \ref{linear:lemma:ILED-nontrapping:bulk-positivity} and a Cauchy-Schwarz argument to control the
  lower-order term on the right-hand

  To prove the resolvent estimate in
  \eqref{linear:eq:ILED-nontrapping:resolvent-estimate:main}, we
  differentiate both sides of equation
  \eqref{linear:eq:ILED-nontrapping:div:aux1} by
  $\p_{\tStar}$.
  \begin{equation}
    \label{linear:eq:ILED-nontrapping:div:aux2}
    \p_{\tStar}\breve{\MorawetzEnergy}(\tStar)[h]
    + \int_{\NonTrappingReg_{\tStar}}\KCurrent{\MorawetzOuterVF, \LagrangeCorrOuter, 0}[{h}]\,\sqrt{\GInvdtdt}
    - \Re \bangle*{\MorawetzOuterVF h , \SubPOp[{h}]}_{\InducedLTwo(\NonTrappingReg_{\tStar})}
    \lesssim \norm*{\LinEinstein_{g_b}h}_{\InducedLTwo(\NonTrappingReg_{\tStar})}^2 + \norm*{h}_{\InducedLTwo(\NonTrappingReg_{\tStar})}^2. 
  \end{equation}
  We now show how to absorb the sign-less principal-level boundary
  term into the bulk term.  To this end, let
  $u\in \InducedHk{1}(\NonTrappingReg_{\tStar})$, and consider
  ${h} = e^{-\ImagUnit\sigma\tStar}u(x)$, so that
  \begin{equation*}
    \p_{\tStar}\breve{\MorawetzEnergy}(\tStar)[h] = 2\Im\sigma\breve{\MorawetzEnergy}(\tStar)[h]. 
  \end{equation*}
  Recalling that \begin{equation*}
    \JCurrent{\MorawetzOuterVF, \LagrangeCorrOuter, 0}[{h}]\cdot n_{\Sigma}
    = \MorawetzOuterVF{h} \cdot n_{\Sigma}\overline{{h}}
    + \LagrangeCorrOuter{h} \cdot n_{\Sigma}\overline{{h}},    
  \end{equation*}
  we choose $\COuter$ sufficiently large using Lemma
  \ref{linear:lemma:ILED-nontrapping:bulk-positivity} such that
  \begin{align*}
    &\int_{\NonTrappingReg_{\tStar}}\KCurrent{\MorawetzOuterVF, \LagrangeCorrOuter, 0}[{h}]\,\sqrt{\GInvdtdt}
    - \Re \bangle*{\MorawetzOuterVF h , \SubPOp[h]}_{\InducedLTwo(\Sigma_{\tStar})}
    + \Im\sigma \int_{\NonTrappingReg_{\tStar}}\JCurrent{\MorawetzOuterVF, \LagrangeCorrOuter, 0}[h]\cdot n_{\NonTrappingReg_{\tStar}}\\
    >{}& C\norm*{h}_{\HkWithT{}{1}(\NonTrappingReg_{\tStar})}^2  - C_1\norm*{h}_{\InducedLTwo(\NonTrappingReg_{\tStar})}^2. 
  \end{align*}
  Substituting this back into \eqref{linear:eq:ILED-nontrapping:div:aux2}, we
  have that for $h=e^{-\ImagUnit\sigma\tStar}u$, 
  \begin{equation*}
    \norm*{h}_{\HkWithT{1}(\NonTrappingReg_{\tStar})}^2 \lesssim \norm*{\LinEinstein_{g_b}h}_{\InducedLTwo(\NonTrappingReg_{\tStar})}^2 + \norm*{h}_{\InducedLTwo(\NonTrappingReg_{\tStar})}^2. 
  \end{equation*}
  Multiplying both sides by $e^{2\Im\sigma\tStar}$ to remove any
  $\tStar$-dependency, we then have that   
  \begin{equation*}
    \norm{u}_{\InducedHk{1}_\sigma(\NonTrappingReg_{\tStar})}
    \lesssim \norm{\widehat{\LinEinstein}(\sigma)u}_{\InducedLTwo(\NonTrappingReg_{\tStar})}
    + \norm{u}_{\InducedLTwo(\NonTrappingReg_{\tStar})}. 
  \end{equation*}
  Recalling from Definition \ref{linear:def:Hk-sigma} that
  \begin{equation*}
    \norm{u}_{\InducedHk{1}_\sigma(\NonTrappingReg)}^2 = \norm{u}_{\InducedHk{1}(\NonTrappingReg)}^2 + \norm{\sigma u}_{\InducedLTwo(\NonTrappingReg)}^2,
  \end{equation*}
  we see that for $\abs*{\sigma}$ sufficiently large, the
  $\norm{u}_{\LTwo(\NonTrappingReg_{\tStar})}$ term on the right-hand
  side can be absorbed into the
  $\norm{\sigma u}_{\LTwo(\NonTrappingReg_{\tStar})}$ left-hand side
  to conclude.
\end{proof}

Now let us prove Theorem
\ref{linear:thm:ILED-nontrapping:resolvent-estimate:main} for higher-order
$k$. To do so, we commute derivatives with the gauged linearized
Einstein operator to derive a higher order positive bulk
term. The rest of the proof is identical to the $k=0$ version. 

\begin{proof}[Proof of Theorem
  \ref{linear:thm:ILED-nontrapping:resolvent-estimate:main} for $k>1$]

  Let us define
  \begin{equation}
    \label{linear:eq:ILED-nontrapping:higher-order:aux1}
    \LinEinstein = \frac{1}{\GInvdtdt}D_{\tStar}^2 + P_1D_{\tStar} + P_2,
  \end{equation}
  where $P_i$ are order-$i$ differential operators on $\NonTrappingReg$. 

  The main idea will be to use the fact that $\KillT$ commutes with
  $\LinEinstein$, and that $P_2$ is elliptic on $\NonTrappingReg$.  We
  prove \eqref{linear:eq:ILED-nontrapping:resolvent-estimate:main} for the
  $k=2$ case. The higher-order cases follow from induction.

  First, from \eqref{linear:eq:ILED-nontrapping:resolvent-estimate:main} with
  $k=1$ have that for ${h} = e^{-\ImagUnit\sigma\tStar}u$, spatially supported on
  $\DomainOfIntegration$,
  \begin{equation}
    \label{linear:eq:ILED-nontrapping:higher-reg-bulk:aux1}
    \norm{\sigma u}_{\InducedHk{1}_\sigma(\NonTrappingReg)}
    \lesssim \norm{\sigma\widehat{\LinEinstein}(\sigma)u}_{\InducedLTwo(\NonTrappingReg)}.
  \end{equation}
  Rewriting \eqref{linear:eq:ILED-nontrapping:higher-order:aux1}, we thus have that
  \begin{equation*}
    P_2u = \widehat{\LinEinstein}(\sigma)u - \sigma P_1 u - \sigma^2 \GInvdtdt u. 
  \end{equation*}
  Since $P_2$ is elliptic on $\NonTrappingReg$, using standard
  elliptic estimates and
  \eqref{linear:eq:ILED-nontrapping:higher-reg-bulk:aux1}, we can conclude
  that
  \begin{equation*}
    \norm{u}_{\InducedHk{2}_\sigma(\NonTrappingReg)}
    \lesssim \norm{\widehat{\LinEinstein}(\sigma)u}_{\InducedHk{1}_\sigma(\NonTrappingReg)},
  \end{equation*}
  yielding the $k=2$ case. 
  Repeating the elliptic estimates and commutations with $\KillT$ as
  above yields the subsequent higher-order estimates. 
\end{proof}

\subsection{Nontrapped near $r=3M$}
\label{linear:sec:ILED:nontrapping-freq}

In the previous section, we proved the resolvent estimates away from
$r=3M$, using the vectorfield method to take advantage of the
non-trapping nature of the region. In this section, we restrict
ourselves to a physical neighborhood of $r=3M$ containing
$\TrappedSet_b$, but microlocalize away from $\TrappedSet_b$, so that
we are still able to prove a Morawetz estimate despite being in a
neighborhood of $r=3M$. As in the previous section, we fix $g=g_b$ and
drop the subscripts in what follows.

Recall from Lemma \ref{linear:lemma:trapping:KdS} that
\begin{equation*}
  \TrappedSet = \curlyBrace*{
    (t, r, \theta, \varphi;\sigma, \xi, \FreqTheta, \FreqPhi):
    \PrinSymb = 0, H_{\PrinSymb}r = 0, r = \rTrapping(\sigma,\FreqPhi)
  }.
\end{equation*}
In what follows, we work on a small neighborhood of $r=3M$ where the
$(\tStar, r, \theta, \phiStar)$ coordinates are identical to the
standard Boyer-Lindquist coordinates. As such, we use
Boyer-Lindquist coordinates in what follows, and denote
\begin{equation*}
  (\SpacetimeLoc, \zeta) = (t,r,\theta,\varphi; \sigma, \xi, \FreqTheta, \FreqPhi).
\end{equation*}

In this section then we will prove a Morawetz estimate for functions
supported physically neighborhood with Fourier transform supported in
a frequency neighborhood away from the trapped set. 

\begin{theorem}
  \label{linear:thm:ILED-nontrapping-freq:main}
  Let 
  \begin{equation*}
    \DomainOfIntegration:=\Real^+_{\tStar} \times \TrappingNbhd,\qquad
    \TrappingNbhd:= (\mathring{r}_-, \mathring{R}_+),
  \end{equation*}
  where we assume
  $3M-\mathring{r}_- \le \delta_r, \mathring{R}_+ - 3M \le \delta_r$
  for some $\delta_r$ sufficiently small. Then if
  $u\in H^1(\TrappingNbhd)$ and moreover, $\hat{u}(\xi, \eta)$ is
  supported on the region $\frac{\abs*{\xi}^2}{|\eta|^2} \ge
  \delta_\zeta$ for some $\delta_\zeta$ sufficiently small, then
  then there exist constants $\SpectralGap, C_0 > 0$ such that we have
  the following resolvent estimate
  \begin{equation}
    \label{linear:eq:ILED-nontrapping-freq:res-est}
    \norm*{u}_{\InducedHk{k}_\sigma(\TrappingNbhd)} \lesssim \norm*{\widehat{\LinEinstein}(\sigma)u}_{\InducedHk{k-1}_\sigma(\TrappingNbhd)},
    \qquad \text{if }
    -\SpectralGap\le \Im\sigma\le \GronwallExp,\quad\text{and } \abs*{\sigma}>C_0. 
  \end{equation}
\end{theorem}

Due to the introduction of the frequency cutoff away from
$\TrappedSet_b$, we will pursue the resolvent estimates here in
frequency space instead of using the physical space calculations of
Section \ref{linear:sec:ILED:nontrapping}. To this end, we will show that
$\KCurrent{X,q, 0}[h]$, which has principal symbol
$\frac{1}{2\ImagUnit}H_{\PrinSymb}X$ from Lemma
\ref{linear:lemma:divergence-prop:freq-formulation}, is elliptic.
\begin{lemma}
  \label{linear:lemma:ILED-nontrapping-freq:bulk-positivity}
  Let
  \begin{equation}
    \label{linear:eq:ILED-nontrapping-freq:op-defs}
    \MorawetzInnerNTVF = \fInnerNT(r)\p_r,\qquad \fInnerNT(r) = e^{\CInnerNT(r-3M)^2}(r-3M)\Delta,
  \end{equation}
  and define 
  \begin{equation}
    \label{linear:eq:ILED-nontrapping-freq:sym-defs}
    \begin{split}
      \MorawetzInnerNTSym &:= \fInnerNT(r)\xi = \frac{1}{2}e^{\CInnerNT(r-3M)^2}(r-3M)H_{\PrinSymb}r,\\
      \LagrangeCorrInnerNTSym &:= -e^{\CInnerNT(r-3M)^2}\frac{1}{C_\xi}\left(2\CInnerNT(r-3M)^2+1\right),
    \end{split}    
  \end{equation}
  where $C_\xi$ is the positive constant from Lemma
  \ref{linear:lemma:freq-cutoff-construction}, and $\CInnerNT>0$ is some
  positive constant to be determined.
  
  Fix some $C>0$. Then, there exists some $\fInnerNT(r)$ and
  $\CInnerNT>0$, such that for $\MorawetzInnerNTVF$ and
  $\MorawetzInnerNTSym$ as defined in
  \eqref{linear:eq:ILED-nontrapping-freq:op-defs}, for any first order symbol
  $s \in \SymClass^1(\StaticRegionWithExtension)$ satisfying
  $\abs*{s} \le C|\zeta|$, there exists some $C_{\zeta}$ sufficiently
  large such that
  \begin{equation*}
    H_{\PrinSymb}\MorawetzInnerNTSym - s\MorawetzInnerNTSym  + \LagrangeCorrInnerNTSym \PrinSymb \gtrsim |\zeta|^2,\qquad \abs*{\zeta}>C_{\zeta}.
  \end{equation*}
\end{lemma}
\begin{remark}
  Note that the inclusion of the $s \MorawetzInnerNTSym$ term with
  potentially large $s$ reflects that we will absorb both the
  contributions of the subprincipal term and the boundary terms in the
  Morawetz estimate via the positive bulk term generated by the
  principal wave component of the operator. This is entirely analogous
  to the argument taken in the previous section via the vectorfield
  argument.
\end{remark}
\begin{proof}  
  We can calculate that
  \begin{equation*}
    H_{\PrinSymb}\MorawetzInnerNTSym
    = e^{\CInnerNT(r-3M)^2}\left(
      \left(2\CInnerNT(r-3M)^2+1\right)(H_{\PrinSymb}r)^2
      + (r-3M)(H_{\PrinSymb}^2r)
    \right),
  \end{equation*}
  and hence,
  \begin{align*}
    H_{\PrinSymb}\MorawetzInnerNTSym
    - s\MorawetzInnerNTSym
    + \LagrangeCorrInnerNTSym \PrinSymb
    ={}& e^{\CInnerNT(r-3M)^2}\left(
         \left(2\CInnerNT(r-3M)^2+1\right)(H_{\PrinSymb}r)^2
         + (r-3M)(H_{\PrinSymb}^2r)\right)\\
       &+e^{\CInnerNT(r-3M)^2}\left(
         - s (r-3M)H_{\PrinSymb}
         - \frac{1}{C_\zeta}\left(2\CInnerNT(r-3M)^2 + 1\right)\PrinSymb
         \right).
  \end{align*}
  The main idea of the lemma will be to use the fact that on the
  support of $\breve{\chi}_{\zeta}$, we have already shown in Lemma
  \ref{linear:lemma:freq-cutoff-construction} that for some $C_\xi>0$
  sufficiently large, $C_\xi (H_{\PrinSymb}r)^2-\PrinSymb$ is
  elliptic. This will compensate for the fact that we are in a
  physical neighborhood of $\TrappedSet$. 
  First observe that by Cauchy-Schwarz, 
  \begin{equation*}
    \begin{split}
      \abs*{(r-3M)H_{\PrinSymb}^2r}
      &\le \epsilon (H_{\PrinSymb}^2r) + \frac{(r-3M)^2}{\epsilon}(H_{\PrinSymb}^2r),\\
      \abs*{(r-3M)s H_{\PrinSymb}r}
      &\le  \epsilon C\abs*{\zeta}^2
      + \frac{(r-3M)^2}{\epsilon}(H_{\PrinSymb}r)^2.    
    \end{split}
  \end{equation*}
  Since $(H_{\PrinSymb}r)^2 - C_\xi^{-1}\PrinSymb$ is elliptic, there exists a choice of
  $\epsilon$ sufficiently small so that 
  \begin{equation*}
    \epsilon\abs*{s}^2 +\epsilon \abs*{H_{\PrinSymb}^2r} \le (H_{\PrinSymb}r)^2- C_\xi^{-1}\PrinSymb.
  \end{equation*}
  Choosing $\epsilon$ sufficiently small as above, and
  $\CInnerNT > \frac{2}{\epsilon}$ sufficiently large,
  we have that 
  \begin{equation*}
    (2\CInnerNT(r-3M)^2+1)\left((H_{\PrinSymb}r)^2-C_\xi^{-1}\PrinSymb\right)
    + (r-3M)(H_{\PrinSymb}^2r)
    - \frac{1}{2}(r-3M)s H_{\PrinSymb}r   
  \end{equation*}
  is elliptic, as desired. 
\end{proof}

We now show how to prove the Morawetz estimate in Theorem
\ref{linear:thm:ILED-nontrapping-freq:main} for the case $k=1$, given the
ellipticity of the bulk term proven in Lemma
\ref{linear:lemma:ILED-nontrapping-freq:bulk-positivity}.
\begin{proof}[Proof of Theorem \ref{linear:thm:ILED-nontrapping-freq:main}
  for $k=1$]

  We first prove Theorem \ref{linear:thm:ILED-nontrapping-freq:main} for
  $k=1$. 
  The main difficulty, as was the case for the previous
  resolvent estimates investigated, is showing that the boundary term
  that arises after using the integration-by-parts (or divergence
  theorem) argument can be absorbed by the bulk term.
  Define, for $\LagrangeCorrInnerNTSym$ as constructed in
  \eqref{linear:eq:ILED-nontrapping-freq:sym-defs}, 
  \begin{equation*}
    \LagrangeCorrInnerNT = \frac{1}{2}\nabla_g\cdot \MorawetzInnerNTVF + \LagrangeCorrInnerNTSym,
  \end{equation*}
  so that using \eqref{linear:eq:divergence-prop:freq-formulation},
  \begin{equation*}
    \KCurrent{\MorawetzInnerNTVF, \LagrangeCorrInnerNT, 0}[h]
    =  (\KCurrentSym{\MorawetzInnerNTVF, \LagrangeCorrInnerNT}_{(2)})^{\alpha\beta }\nabla_{(\alpha}h\cdot\nabla_{\beta)}\overline{h}
    - \frac{1}{2}\nabla^\alpha\p_\alpha\LagrangeCorrInnerNT \abs*{h}^2,
  \end{equation*}
  where
  \begin{equation*}
    2(\KCurrentSym{\MorawetzInnerNTVF, \LagrangeCorrInnerNT}_{(2)})^{\alpha\beta}\zeta_\alpha\zeta_\beta
    = H_{\PrinSymb}\MorawetzInnerNTSym + \check{\LagrangeCorrSym}\PrinSymb.
  \end{equation*}
  We now also define the relevant Morawetz energy on the
  $\tStar$-constant hypersurfaces for this section
  \begin{equation*}
    \widecheck{\MorawetzEnergy}(\tStar)[{h}] = \int_{\Sigma_{\tStar}}\JCurrent{\MorawetzInnerNTVF, \LagrangeCorrInnerNT, 0}[{h}]\cdot n_{\Sigma_{\tStar}} . 
  \end{equation*}
  Applying the divergence relation in Corollary
  \ref{linear:cor:div-them:spacelike}, we have that
  \begin{align}
    -\Re\bangle*{\LinEinstein{h}, (\MorawetzInnerNTVF+\LagrangeCorrInnerNT){h}}_{L^2(\DomainOfIntegration)}
    ={}& \int_{\DomainOfIntegration}\KCurrent{\MorawetzInnerNTVF, \LagrangeCorrInnerNT, 0}[{h}]
         + \left.\widecheck{\MorawetzEnergy}(\tStar)[{h}]\right\vert_{\tStar = 0}^{\tStar = \TStar}
       - \Re \bangle*{\SubPOp[{h}], (\MorawetzInnerNTVF+\LagrangeCorrInnerNT){h}}_{L^2(\DomainOfIntegration)}
        \notag\\
    &- \Re \bangle*{\PotentialOp{h}, (\MorawetzInnerNTVF+\LagrangeCorrInnerNT){h}}_{L^2(\DomainOfIntegration)}.\label{linear:eq:ILED-nontrapping-freq:div-thm-app}
  \end{align}
  We now show the resolvent estimate
  \eqref{linear:eq:ILED-nontrapping-freq:res-est}\footnote{As was the case in
    the proof of Theorem
    \ref{linear:thm:ILED-nontrapping:resolvent-estimate:main}, a Morawetz
    estimate follows directly at this stage from the observation that
    using Lemma \ref{linear:lemma:ILED-nontrapping-freq:bulk-positivity},
    there exists a choice of $\CInnerNT$, such that for
    $\MorawetzInnerNTVF$ as defined above
    \begin{equation*}
      \norm{{h}}_{H^1(\DomainOfIntegration)}^2 \lesssim
      \int_{\DomainOfIntegration}\KCurrent{\MorawetzInnerNTVF, \LagrangeCorrInnerNT, 0}[{h}]
      - \Re \bangle*{\SubPOp[{h}], \MorawetzInnerNTVF{h}}_{L^2(\DomainOfIntegration)}
      + \norm{{h}}_{L^2(\DomainOfIntegration)}^2,
    \end{equation*}
    and an application of the Cauchy-Schwarz theorem to deduce that   
    \begin{equation*}
      \left.\widecheck{\MorawetzEnergy}(\tStar)[{h}]\right\vert_{\tStar = 0}^{\tStar = \TStar}
      +\norm{{h}}_{H^1(\DomainOfIntegration)}^2
      \lesssim
      \norm*{\LinEinstein h}_{L^2(\DomainOfIntegration)}^2
      + \norm{{h}}_{L^2(\DomainOfIntegration)}^2.    
    \end{equation*}
  }. Consider $u(x)$ such that $\hat{u}(\xi, \eta)$ is
  supported on the region $\frac{\abs*{\xi}^2}{|\eta|^2} \ge
  \delta_\zeta$ for some $\delta_\zeta$ sufficiently small, and let  
  $h = e^{-\ImagUnit\sigma\tStar}u$. Then,
  \begin{equation*}
    \p_{\tStar}\widecheck{\MorawetzEnergy}(\tStar)[{h}] = 2\Im\sigma\widecheck{\MorawetzEnergy}(\tStar)[{h}].
  \end{equation*}
  Recall that
  \begin{equation*}
    \JCurrent{\MorawetzInnerNTVF, \LagrangeCorrInnerNT, 0}[{h}]\cdot n_{\Sigma}
    = \MorawetzInnerNTVF{h} \cdot n_{\Sigma}\overline{{h}}
    + \LagrangeCorrInnerNT{h} \cdot n_{\Sigma}\overline{{h}}.
  \end{equation*}
  Then from Lemma \ref{linear:lemma:ILED-nontrapping-freq:bulk-positivity},
  we can choose $\CInnerNT$ sufficiently large so that for any
  $-\SpectralGap< \Im \sigma < \GronwallExp$,
    \begin{equation*}
    \norm{{h}}_{\HkWithT{1}(\TrappingNbhd)}^2 \lesssim \int_{\TrappingNbhd}\KCurrent{\MorawetzInnerNTVF, \LagrangeCorrInnerNT, 0}[{h}]\,\sqrt{\GInvdtdt}
    - \Re\bangle*{\SubPOp[{h}], \MorawetzInnerNTVF{h}}_{\InducedLTwo(\TrappingNbhd)}
    + \p_{\tStar} \widecheck{\MorawetzEnergy}(\tStar)[{h}] + \norm{{h}}_{\InducedLTwo(\TrappingNbhd)}^2.
  \end{equation*}
  Plugging this back into
  \eqref{linear:eq:ILED-nontrapping-freq:div-thm-app}, and applying
  Cauchy-Schwarz, we have that for $h = e^{\ImagUnit\sigma\tStar}u$,
  \begin{equation*}
    \norm*{h}_{\HkWithT{1}(\TrappingNbhd)}^2 \lesssim \norm*{\LinEinstein_{g_b}h}_{\InducedLTwo(\TrappingNbhd)}^2 + \norm*{h}_{\InducedLTwo(\TrappingNbhd)}^2. 
  \end{equation*}
  Multiplying both sides of the equation by $e^{2\Im\sigma\tStar}$ to
  get rid of any $\tStar$ dependency, we have that
  \begin{equation*}
    \norm{u}_{\InducedHk{1}_\sigma(\TrappingNbhd_{\tStar})}
    \lesssim \norm{\widehat{\LinEinstein}(\sigma)u}_{\InducedLTwo(\TrappingNbhd_{\tStar})}
    + \norm{u}_{\InducedLTwo(\TrappingNbhd_{\tStar})}. 
  \end{equation*}
  Recalling from Definition \ref{linear:def:Hk-sigma} that
  \begin{equation*}
    \norm{u}_{\InducedHk{1}_\sigma(\TrappingNbhd)}^2 = \norm{u}_{\InducedHk{1}(\TrappingNbhd)}^2 + \norm{\sigma u}_{\InducedLTwo(\TrappingNbhd)}^2,
  \end{equation*}
  we see that for $\abs*{\sigma}$ sufficiently large, the
  $\norm{u}_{\LTwo(\TrappingNbhd_{\tStar})}$ term on the right-hand
  side can be absorbed into the
  $\norm{\sigma u}_{\LTwo(\TrappingNbhd_{\tStar})}$ left-hand side
  to conclude.
\end{proof}

To prove Theorem \ref{linear:thm:ILED-nontrapping-freq:main} for higher-order
$k$, we again rely on a commutation with $\p_{\tStar}$ and an
elliptic argument, taking advantage of the fact that the
trapped set and the ergoregions in the slowly-rotating cases are
physically separated.

\begin{proof}[Proof of Theorem \ref{linear:thm:ILED-nontrapping-freq:main}
  for $k>1$.]
  We prove Theorem \ref{linear:thm:ILED-nontrapping-freq:main} for $k=2$. The
  $k>2$ case follows similarly.  Reflecting the fact that $\KillT$ is
  Killing and commutes with $\LinEinstein$, we have from
  \eqref{linear:eq:ILED-nontrapping-freq:res-est} with $k=1$ that for
  $u\in\InducedHk{k}(\TrappingNbhd)$ supported on $\TrappingNbhd$,
  \begin{equation}
    \label{linear:eq:ILED-nontrapping-freq:higher-order:aux1}
    \norm{\sigma u}_{\InducedHk{1}(\TrappingNbhd)}
    \lesssim \norm*{\sigma\widehat{\LinEinstein}(\sigma)u}_{\InducedLTwo(\TrappingNbhd)}. 
  \end{equation}
  Then, using that $\LinEinstein = P_2 + P_1D_{\tStar}
  + \frac{1}{\GInvdtdt} D^2_{\tStar}$, we have that for $u\in \InducedHk{2}(\TrappingNbhd)$, 
  \begin{equation*}
    P_2u = \widehat{\LinEinstein}(\sigma) u - \sigma P_1u -  \sigma^2 \GInvdtdt u.   
  \end{equation*}
  Recall that $P_2$ is elliptic away from the ergoregions. Since we
  are considering a region supported away from the ergoregions, we can
  apply a standard elliptic estimate to see that
  \begin{equation}
    \label{linear:eq:ILED-nontrapping-freq:higher-order:aux2}
    \norm{u}_{\InducedHk{2}(\TrappingNbhd)} \lesssim
    \norm*{\widehat{\LinEinstein}(\sigma)u}_{\InducedHk{1}_{\sigma}(\TrappingNbhd)}.
  \end{equation}
  Combining equations \eqref{linear:eq:ILED-nontrapping-freq:higher-order:aux1} and
  \eqref{linear:eq:ILED-nontrapping-freq:higher-order:aux2} allows us to
  conclude.
\end{proof}

\subsection{Trapping region}\label{linear:ILED:near}
\label{linear:sec:ILED:trapping}

We now microlocalize to the trapped set in a neighborhood of $r=3M$.
While in the previous sections we were able to prove Morawetz
estimates that controlled the full $H^k$ norm of solutions ${h}$, we
will be unable to do so in this section due to the presence of
trapping. Instead, we define new norms that account for trapping by
degenerating exactly on the trapped set. Also, we use the frequency
analysis in Section \ref{linear:sec:freq-analysis} to account for the
frequency-dependent nature of the trapped set in \KdS. The
pseudo-differential operators introduced should be compared to the
very similar pseudo-differential operators used by Tataru and
Tohaneanu in \cite{tataru_local_2010} to prove a Morawetz estimate for
solutions to the scalar wave on a Kerr
background. 

Throughout this section we will let $g_{b_0}=g(M,0)$ be a fixed \SdS{}
metric, and $g_b=g(M,a)$, be a nearby \KdS{} metric. As we discussed
in Section \ref{linear:sec:trapping-KdS}, for any $\delta_r>0$, there exists
a neighborhood of black hole parameters $\BHParamNbhd$ such that for
all $g_b$ \KdS{} backgrounds with $b\in\BHParamNbhd$, the trapped set
$\TrappedSet_b$ lies entirely within
$\curlyBrace*{|r-3M|<\delta_r}$. In addition on
$\curlyBrace*{|r-3M|<\delta_r}$, the Kerr-star coordinates
$(\tStar, r, \theta, \phiStar)$ reduce to the Boyer-Lindquist
coordinates $(t, r, \theta, \varphi)$, and we will use the
Boyer-Lindquist coordinates in the remainder of this section.

We begin with a proof of the trapped high-frequency Morawetz estimate
on \SdS{} in Section \ref{linear:sec:ILED-trapping:SdS-toy-model}, and
show how to use the basic spectral gap for the scalar wave equation
on \SdS{} to prove a spectral gap for the gauged linearized Einstein
operator linearized around a nearby \KdS{} metric.

Throughout this section, we will denote
\begin{equation}
  \label{linear:eq:ILED-trapping:trapping-reg-def}
  \DomainOfIntegration:=\Real^+_{\tStar} \times \TrappingNbhd,\qquad
  \TrappingNbhd:= \Sigma(\mathring{r}_-\le r\le \mathring{R}_+),
\end{equation}
where we assume $3M-\mathring{r}_- \le \delta_r, \mathring{R}_+ - 3M
\le  \delta_r$, so that for sufficiently slowly-rotating \KdS{}
metrics $g_b$, the trapped set $\TrappedSet_b$ lies entirely within
$\DomainOfIntegration$. 

\subsubsection{Scalar wave on \SdS}
\label{linear:sec:ILED-trapping:SdS-toy-model}

Before we prove the high-frequency Morawetz estimate for the gauged
linearized Einstein equations on \KdS, let us first review the proof
of high-frequency Morawetz  estimate in a neighborhood of the photon
sphere for the scalar wave equation on \SdS. This will serve as the
basis upon which we add the pseudo-differential modification of the
divergence theorem in Section \ref{linear:sec:int-by-parts-arg} to
prove a high-frequency Morawetz estimate for the gauged linearized Einstein
operator. 

Recall from Lemma \ref{linear:lemma:trapping:SdS} that
for $g_{b_0}$ a \SdS{} metric, the trapped set is contained exactly at
the photon sphere $r=3M$,
\begin{equation*}
  \TrappedSet_{b_0} = \curlyBrace*{(t,r,\omega;\sigma, \xi, \FreqAngular): r=3M, \xi=0, \PrinSymb_{b_0}=0}.
\end{equation*}
We then define the following norm
\begin{equation}
  \label{linear:eq:ILED-near:SdS:Morawetz-norm}
  \norm{{h}}_{\MorawetzNorm_{b_0}(\TrappingNbhd)}^2
  := \int_{\TrappingNbhd}(r-3M)^2(\abs*{\p_t{h}}^2
  + \abs*{\NablaAngular{h}}^2)
  + \abs*{\p_r{h}}^2 +\abs*{{h}}^2,
\end{equation}
where $\TrappingNbhd$ is as defined in
\eqref{linear:eq:ILED-trapping:trapping-reg-def}.  
\begin{remark}
  We remark that compared to the standard Morawetz norm defined on
  Schwarzschild (see for instance equation (1.11) in
  \cite{marzuola_strichartz_2010}), we differ by a power of $r$. This
  plays no role in our analysis as $r$ is both bounded above and below
  on the static region of \SdS.
\end{remark}
We also define the following auxiliary, non-coercive norm that will be
used in the subsequent proof of the desired Morawetz estimate,
\begin{equation*}
  \mathring{\MorawetzEnergy}_{b_0}(t)[{h}]
  = \int_{\TrappingNbhd_t}\JCurrent{\MorawetzVF_{b_0}, \LagrangeCorr_{b_0}, 0}[{h}]\cdot n_{\TrappingNbhd_t},
\end{equation*}
where $\MorawetzVF_{b_0}$ and $\LagrangeCorr_{b_0}$ will be defined in
what follows. 

With the desired norm in hand, we now state the
desired resolvent estimate for the scalar wave in \SdS:
\begin{prop}
  \label{linear:prop:ILED-near:SdS}
  Let
  \begin{equation*}
    \DomainOfIntegration:=\Real^+_t\times\TrappingNbhd,\qquad \TrappingNbhd = (3M-\delta_r, 3M+\delta_r),
  \end{equation*}
  for some sufficiently small $\delta_r$. Then if $u$ is a
  sufficiently smooth function with compact support $\TrappingNbhd$,
  then there exists some $\SpectralGap_{0}>0$ and some constant $C_0$
  such that
  \begin{equation}
    \label{linear:eq:ILED-near:SdS:resolvent-estimate}
    \norm{u}_{\MorawetzNorm_{b_0}(\TrappingNbhd)}
    \lesssim \norm*{\ScalarWaveLaplaceOp[g_{b_0}](\sigma)u}_{\InducedLTwo(\TrappingNbhd)},\qquad 
    \Im\sigma> -\SpectralGap_{0}, \text{or }
    \Im\sigma = -\SpectralGap_{0}.
  \end{equation}
\end{prop}
We will prove the result with a purely physical argument, emphasizing
that the pseudo-differential nature of the subsequent arguments in
\KdS{} reflects the frequency-dependent nature of trapping in \KdS{}
and the microlocal smallness we need at the level of the subprincipal
operator of $\LinEinstein_{g_b}$. 

There are two components to the proof of Proposition
\ref{linear:prop:ILED-near:SdS}. First, we prove the resolvent estimate
\eqref{linear:eq:ILED-near:SdS:resolvent-estimate} for $\abs*{\Im\sigma}\le
\SpectralGap$ for some $\SpectralGap>0$, and then we prove the
resolvent estimate for $\Im\sigma > \frac{\SpectralGap}{2}$. These two
steps correspond to using a Morawetz estimate and a basic energy
estimate respectively. 

The main lemma is as follows. 
\begin{lemma}
  \label{linear:lemma:ILED-near:SdS-scalar-wave:Bulk}
  For $\delta_r$ sufficiently small there exists:
  \begin{enumerate}
  \item a smooth vectorfield
    \begin{equation}
      \label{linear:eq:ILED-near:SdS:X-def}
      \MorawetzVF_{b_0} = f_{b_0}(r)\p_r, 
    \end{equation}
    where $f_{b_0}(r)$ is bounded near $r=3M$. In
    particular, we will choose
    \begin{equation}
      \label{linear:eq:ILED-near:SdS:X-def:f-def}
      f_{b_0}(r)= \frac{(r-3M)\mu_{b_0}}{r^2};
    \end{equation}
  \item a smooth function $\LagrangeCorr_{b_0}$;
    such that for ${h}$ supported in
    $\DomainOfIntegration:=\Real^+_t\times\TrappingNbhd$,
    $\TrappingNbhd = (3M-\delta_r, 3M+\delta_r)$,
    \begin{equation*}
      \KCurrent{\MorawetzVF_{b_0},\LagrangeCorr_{b_0},0}[{h}]
      \gtrsim (r-3M)^2(\abs*{\p_t {h}}^2 + \abs*{\NablaAngular {h}}^2) + \abs*{\p_r{h}}^2 + \abs*{{h}}^2.
    \end{equation*}
  \end{enumerate}  
\end{lemma}

\begin{proof}
  We can first calculate that for $\MorawetzVF_{b_0} = f_{b_0}(r)\p_r$
  as defined above in \eqref{linear:eq:ILED-near:SdS:X-def:f-def},  
  \begin{equation}
    \label{linear:eq:ILED-near:SdS:EMT-DT:1}
    \EMTensor[h] \cdot \DeformationTensor[]{\MorawetzVF_{b_0}}
    = \mu_{b_0}^2\left(\frac{1}{r^2} - \frac{2(r-3M)}{r^3}\right)\abs*{\p_r{h}}^2
    + \frac{(r-3M)^2}{r^4}\abs*{\NablaAngular{h}}^2
    - \frac{\mu_{b_0}}{2r^2} g_{b_0}^{\gamma\delta}\p_\gamma{h}\cdot \p_\delta\overline{{h}}.
  \end{equation}
  Since $\mu_{b_0} > 0$ on $\TrappingNbhd$, the first two terms on the
  right-hand side of equation \eqref{linear:eq:ILED-near:SdS:EMT-DT:1} are
  clearly positive. There are two remaining issues. The first is a
  treatment of the third term in equation
  \eqref{linear:eq:ILED-near:SdS:EMT-DT:1}, which does not have a sign. The
  second is that we wish for
  $\KCurrent{\MorawetzVF_{b_0}, 0, 0}[{h}] =
  \EMTensor[{h}]\cdot \DeformationTensor[]{\MorawetzVF_{b_0}}$
  to be coercive in all the derivatives of ${h}$, not just $\p_r$ and
  $\NablaAngular$. The key to resolving both of these issues will be
  to make an appropriate choice of $\LagrangeCorr_{b_0}$.  Consider
  \begin{equation*}
    \LagrangeCorr_0 := \frac{\mu_{b_0}}{2r^2},
  \end{equation*}
  so that
  \begin{equation*}
    \KCurrent{0, \LagrangeCorr_0, 0}[h]
    = \frac{\mu_{b_0}}{2r^2}g_{b_0}^{\gamma\delta}\p_\gamma{h}\cdot \p_\delta\overline{h}
    - \frac{1}{2}\nabla^\alpha\p_\alpha \LagrangeCorr_0 \abs*{h}^2,
  \end{equation*}
  where we can directly calculate that
  \begin{equation*}
    -\frac{1}{2}\nabla^\alpha\p_\alpha \LagrangeCorr_0(3M) =  \frac{1-9M^2\Lambda}{243M^4}>0.
  \end{equation*}
  As a result, for $\delta_r$ sufficiently small,
  \begin{equation*}
    -\frac{1}{2}\nabla^\alpha\p_\alpha \LagrangeCorr_0 > 0
  \end{equation*}
  on all of $\TrappingNbhd$. 
  
  The addition of $\KCurrent{0, \LagrangeCorr_0, 0}[{h}]$ to
  $\KCurrent{\MorawetzVF, 0, 0}[{h}]$ will eliminate the third term
  on the right-hand side of equation
  \eqref{linear:eq:ILED-near:SdS:EMT-DT:1} and introduce the coerciveness of
  the $L^2$ norm, that is,
  \begin{equation*}
    \KCurrent{\MorawetzVF_{b_0}, \LagrangeCorr_0, 0}[{h}]
    = \mu_{b_0}^2\left(\frac{1}{r^2} - \frac{2(r-3M)}{r^3}\right)\abs*{\p_r{h}}^2
    + \frac{(r-3M)^2}{r^4}\abs*{\NablaAngular{h}}^2
    - \frac{1}{2}\nabla^\alpha\p_\alpha \LagrangeCorr_0\abs*{{h}}^2.
  \end{equation*}

  To solve the issue of coerciveness of the $\p_t$ derivatives, we
  borrow positivity from the angular and radial derivatives in
  $\KCurrent{\MorawetzVF_{b_0}, \LagrangeCorr_0,0}[{h}]$ via the
  Lagrangian correction. Consider
  \begin{equation}
    \label{linear:eq:ILED-near:SdS:q1-def}
    \LagrangeCorr_1 := -\delta_1\frac{(r-3M)^2}{r^4}.
  \end{equation}
  Then
  \begin{equation*}
    \begin{split}
      \KCurrent{0,\LagrangeCorr_1,0}[{h}]
      ={}& \delta_1\left(
        \mu_{b_0}^{-1}\frac{(r-3M)^2}{r^4}\abs*{\p_t{h}}^2
        - \mu_{b_0}\frac{(r-3M)^2}{r^4}\abs*{\p_r{h}}^2
        - \frac{(r-3M)^2}{r^4}\abs*{\NablaAngular{h}}^2
      \right)\\
      &- \frac{\delta_1}{2}\nabla^\alpha\p_\alpha \frac{(r-3M)^2}{r^4}|{h}|^2.
    \end{split}      
  \end{equation*}
  Combining the above calculations, we find that defining
  \begin{equation}
    \label{linear:eq:ILED-near:SdS:q-def}
    \LagrangeCorr_{b_0}:= \LagrangeCorr_0 + \LagrangeCorr_1,
  \end{equation}
  we can calculate 
  \begin{equation}
    \label{linear:eq:ILED-near:sum-of-squares:SdS}
    \KCurrent{\MorawetzVF_{b_0}, \LagrangeCorr_{b_0},0}[{h}]
    = \delta_1\alpha_{b_0}^2\abs*{\p_t{h}}^2 + \beta_{b_0}^2 \abs*{\p_r{h}}^2
    + (1-\delta_1)\frac{(r-3M)^2}{r^2}\abs*{\NablaAngular{h}}^2 + \gamma_{b_0}^2\abs*{{h}}^2,
  \end{equation}
  where
  \begin{align*}
    \alpha_{b_0}^2 &=  \mu_{b_0}^{-1}\frac{(r-3M)^2}{r^4},\\
    \beta_{b_0}^2 &=  \mu_{b_0}^2\left(\frac{1}{r^2} - \frac{2(r-3M)}{r^3} - \delta_1\mu_{b_0}^{-1}\frac{(r-3M)^2}{r^4}\right), \\
    \gamma_{b_0}^2 &= \nabla^\alpha\p_\alpha(q_0+q_1),
  \end{align*}
  where we have used that $q_1 = O(\delta_1)$ to write
  $\gamma_{b_0}^2$ as a positive function on $\DomainOfIntegration$. 
\end{proof}
Having shown that $\KCurrent{\MorawetzVF_{b_0}, \LagrangeCorr_{b_0},0}[{h}]$
generates a non-negative bulk term, we can now move onto the proof of
Proposition \ref{linear:prop:ILED-near:SdS}.

\begin{proof}[Proof of Proposition \ref{linear:prop:ILED-near:SdS}]
  Using the divergence theorem in Corollary
  \ref{linear:cor:div-them:spacelike}, we see that for ${h}$ supported in
  a neighborhood of $r=3M$,
  \begin{equation}
    \label{linear:eq:ILED-near:SdS:div-thm}
    \p_t\mathring{\MorawetzEnergy}(t)[{h}]
    + \int_{\TrappingNbhd}\KCurrent{\MorawetzVF_{b_0}, \LagrangeCorr_{b_0},0}[{h}]\,\sqrt{\GInvdtdt}
    = -\Re\int_{\TrappingNbhd} \ScalarWaveOp[g_{b_0}]{h}\cdot \left(\MorawetzVF_{b_0}+\LagrangeCorr_{b_0}\right)\overline{{h}}\,\sqrt{\GInvdtdt}.
  \end{equation}
  Using Lemma \ref{linear:lemma:ILED-near:SdS-scalar-wave:Bulk} and the
  Cauchy-Schwarz inequality, we then have that
  \begin{equation*}
    \p_t\mathring{\MorawetzEnergy}(t)[{h}]
    + \norm{{h}}_{\MorawetzNorm(\TrappingNbhd)}^2
    \lesssim \norm*{\ScalarWaveOp[g_{b_0}]h}_{\InducedLTwo(\TrappingNbhd)}^2.
  \end{equation*}
  To prove the resolvent estimate in
  \eqref{linear:eq:ILED-near:SdS:resolvent-estimate}, we consider
  ${h} = e^{-\ImagUnit\sigma t}u(x)$, where $u$ is supported in a
  neighborhood of $r=3M$. Furthermore, recall that
  \begin{equation}
    \label{linear:eq:ILED-near:SdS:resolvent-estimate:aux1}
    \mathring{\MorawetzEnergy}(t)[{h}]
    = \int_{\TrappingNbhd}\JCurrent{\MorawetzVF_{b_0}, \LagrangeCorr_{b_0},0}[h]\cdot n_{\TrappingNbhd}
    = \Re\int_{\TrappingNbhd}f_{b_0}(r)\p_r{h}\cdot n_{\TrappingNbhd}\overline{h}
    + \Re\int_{\TrappingNbhd}\LagrangeCorr_{b_0}{h}\cdot  n_{\TrappingNbhd}\overline{h}.    
  \end{equation}
  Using the fact that $\LagrangeCorr_{b_0}$ is a smooth stationary
  function on $\TrappingNbhd$, we use integration by parts to write
  that
  \begin{equation*}
    \begin{split}
      \bangle*{\p_r(3M-r){\LagrangeCorr}_{b_0}{h},n_{\TrappingNbhd}{h}}_{\LTwo(\TrappingNbhd)}
      ={}&\bangle*{(r-3M){\LagrangeCorr}_{b_0}\p_r{h},n_{\TrappingNbhd}{h}}_{\LTwo(\TrappingNbhd)}
      + \bangle*{(r-3M){\LagrangeCorr}_{b_0}{h},n_{\TrappingNbhd}\p_r{h}}_{\LTwo(\TrappingNbhd)}\\
      &+ Err[{h}],
    \end{split}
  \end{equation*} 
  where $Err[{h}]$ consists of lower-order terms satisfying the
  estimate
  \begin{equation*}
    \abs*{Err[{h}]} \lesssim \norm{(r-3M)\nabla{h}}_{\LTwo(\TrappingNbhd)}^2 + \norm{\p_r{h}}_{\LTwo(\TrappingNbhd)}^2 + \norm{{h}}_{\LTwo(\TrappingNbhd)}^2.
  \end{equation*}
  Recall that we assumed that $h$ could be written as
  $h=e^{-\ImagUnit\sigma t}u$, it is clear that
  $\p_th = -\ImagUnit\sigma h$. Thus, we can apply Cauchy-Schwarz to control
  \begin{equation*}
    \abs*{\bangle*{(r-3M){\LagrangeCorr}_{b_0}\p_r{h},n_{\TrappingNbhd}{h}}_{\LTwo(\TrappingNbhd)}}
    + \abs*{\bangle*{(r-3M){\LagrangeCorr}_{b_0}{h},n_{\TrappingNbhd}\p_r{h}}_{\LTwo(\TrappingNbhd)}}
      \lesssim \norm{h}_{\MorawetzNorm(\TrappingNbhd)}^2.
  \end{equation*}
  We can also apply Cauchy-Schwarz to control
  \begin{equation*}
    \int_{\TrappingNbhd}\abs*{f_{b_0}(r)\p_rh\cdot n_{\TrappingNbhd}\overline{h}} \lesssim \norm{h}_{\MorawetzNorm(\TrappingNbhd)}^2.
  \end{equation*}
  Having controlled each of the terms in
  \eqref{linear:eq:ILED-near:SdS:resolvent-estimate:aux1}, we can write that 
  \begin{equation*}
    \abs*{\mathring{\MorawetzEnergy}_{b_0}(t)[{h}]} \lesssim \norm{h}_{\MorawetzNorm(\TrappingNbhd)}^2.
  \end{equation*}
  Recalling that for $h=e^{-\ImagUnit\sigma t}$,
  $\p_th = -\ImagUnit\sigma h$, we see that for any $\delta>0$, there
  exists a choice of $\SpectralGap_0>0$ such that for
  $\abs*{\Im\sigma}<\SpectralGap_0$,
  \begin{equation*}
    2\Im\sigma\mathring{\MorawetzEnergy}(t)[{h}]
    + \int_{\TrappingNbhd}\KCurrent{\MorawetzVF_{b_0}, \LagrangeCorr_{b_0},0}[{h}]\,\sqrt{\GInvdtdt}
    \gtrsim \norm*{h}_{\MorawetzEnergy(\TrappingNbhd)}.
  \end{equation*}
  Plugging this back into \eqref{linear:eq:ILED-near:SdS:div-thm}, we have
  from Cauchy-Schwarz that
  \begin{equation*}
    \norm*{h}_{\MorawetzEnergy(\TrappingNbhd)}
    \lesssim
    \norm*{\ScalarWaveOp[g_{b_0}]{h}}_{\InducedLTwo(\TrappingNbhd)}^2
    + \norm*{h}_{\InducedLTwo(\TrappingNbhd)}^2.
  \end{equation*}
  Multiplying both sides by $e^{2\Im\sigma t}$ to remove any
  $t$-dependency, and using Cauchy-Schwarz, we have that 
  \begin{equation}
    \label{linear:eq:ILED-near:SdS:Scalar-wave:ILED}
    \norm*{u}_{\MorawetzNorm_{b_0}(\TrappingNbhd)} \lesssim
    \norm*{\TransformScalarWaveOp[g_{b_0}](\sigma)u}_{\InducedLTwo(\TrappingNbhd)},
    \qquad \abs*{\Im\sigma}\le \SpectralGap_{0}.
  \end{equation}
  It then remains to prove the resolvent estimate
  \eqref{linear:eq:ILED-near:SdS:resolvent-estimate} with
  $\Im{\sigma}>\SpectralGap_{0}$. In the case of the scalar wave on
  \SdS, there is no subprincipal component to consider, and the
  Killing energy is conserved. Thus, using the equivalent of a naive
  Gronwall-type energy estimate, we have that for any $\epsilon>0$,
  there exists a constant $C(\epsilon)>0$ such that
  \begin{equation*}
    \p_{t}\EnergyKill(t)[{h}] \le 
    C(\epsilon)\norm*{\ScalarWaveOp[g_{b_0}]h}_{\InducedLTwo(\TrappingNbhd)}^2 + \epsilon \EnergyKill(t)[h].
  \end{equation*}
  Since ${h}$ is supported near $\TrappedSet_{b_0}$ (in particular, it
  is supported away from both $\EventHorizonFuture$ and
  $\CosmologicalHorizonFuture$), the Killing energy
  norm of ${h}$ controls all derivatives of ${h}$. Thus, when we plug
  in ${h} = e^{-\ImagUnit\sigma t}u(x)$ for some $u$ that
  is compactly supported on a small neighborhood of $\TrappedSet_{b_0}$
  (and in particular, compactly supported away from the event horizon
  and the cosmological horizon), we have
  \begin{equation*}
    \p_t\EnergyKill(t)[{h}] \gtrsim \Im\sigma \norm{{h}}_{\InducedHk{1}_\sigma(\TrappingNbhd)}^2. 
  \end{equation*}
  As a result, for $\Im\sigma>\frac{\SpectralGap_0}{2}$, and
  sufficiently small $\epsilon$,
  \begin{equation}
    \label{linear:eq:ILED-near:SdS:Scalar-wave:energy}
    \norm*{u}_{H^1_{\sigma}(\TrappingNbhd)} \lesssim
    \norm*{\TransformScalarWaveOp[g_{b_0}](\sigma)u}_{\InducedLTwo(\TrappingNbhd)},
    \qquad \Im\sigma>\frac{\SpectralGap_0}{2}.
  \end{equation}
  Combining the estimates in (\ref{linear:eq:ILED-near:SdS:Scalar-wave:ILED})
  and (\ref{linear:eq:ILED-near:SdS:Scalar-wave:energy}) yields a resolvent
  estimate on the entire half-plane $\Im\sigma\ge -\SpectralGap_0$
  as desired, concluding the proof of Proposition \ref{linear:prop:ILED-near:SdS}.    
\end{proof}


\subsubsection{Morawetz estimate near $r=3M$ for the gauged
  linearized Einstein operator}
\label{linear:sec:ILED-trapping:KdS}

We now turn to the problem of proving resolvent estimates for the
gauged linearized Einstein operator in \KdS. We will first need to
define the relevant norms.
To capture the idea that trapping is a
feature of the characteristic set, for  $r\in (3M - \delta_r,
3M+\delta_r)$, we factor 
\begin{equation*}
  \PrinSymb_b(r, \theta; \sigma, \xi, \FreqTheta, \FreqPhi) =
  g^{tt}(\sigma-\sigma_1(r,\theta; \xi, \FreqTheta,
  \FreqPhi))(\sigma-\sigma_2(r, \theta; \xi, \FreqTheta, \FreqPhi)),
\end{equation*}
where $\sigma_1, \sigma_2\in S_{hom}^1$ are distinct smooth
symbols. On the cones $\sigma=\sigma_i$ (i.e. on the characteristic
set), the symbol $r-\rTrapping_b$ is then equal to 
\begin{equation*}
  \ell_i = r-\rTrapping_b(\sigma_i,\FreqPhi) = r-3M - a \tilde{r}_b \left(\sigma_i, \FreqPhi\right).
\end{equation*}
To use $\ell_i$, we cut off away from the singularity at frequency $0$
and redefine:
\begin{equation*}
  \ell_i = r-\rTrapping_b(\sigma_i,\FreqPhi) = r-3M - a \chi_{\ge 1}\tilde{r}_b \left(\sigma_i, \FreqPhi\right) \in S^0(\Sigma),
\end{equation*} 
where $\chi_{\ge1}$ is a smooth symbol such that $\chi_{\ge 1}=1$ for
frequencies $\ge 2$, $\chi_{\ge1}=0$ for frequencies $\le 1$.

The symbols $\ell_i$ can then be used to define microlocally weighted
$L^2$ function spaces in a neighborhood of $r=3M$.  

\begin{definition}
  Let $h:\StaticRegionWithExtension\to \Complex^D$ such that
  $h(\tStar, \cdot)$ is compactly supported on $\TrappingNbhd$ for all
  $\tStar>0$, and let $\DomainOfIntegration$ be as defined in
  \eqref{linear:eq:ILED-trapping:trapping-reg-def}. We define
  \begin{equation*}
    \begin{split}
      \norm{h}_{L^2_{\ell_i}(\DomainOfIntegration)}^2
      &= \norm{\ell_i(D, x){h}}_{L^2(\DomainOfIntegration)}^2.
      + 
    \end{split}
  \end{equation*}
  We also define the corresponding norms over a spacelike slice.
  \begin{align*}
    \norm{{h}}_{\InducedLTwo_{\ell_i}(\TrappingNbhd)}^2
    &= \norm{\ell_i(D, x){h}}_{\InducedLTwo(\TrappingNbhd)}^2.
  \end{align*}
\end{definition}
\begin{remark}
  The symbols $\ell_i$ are nonzero outside an $O(a)$ neighborhood of
  $3M$, so $L_{\ell_i}^2$ is equivalent to the $L^2$ norm
  outside an $O(a)$ neighborhood of $r=3M$.
\end{remark}


\begin{definition}
  \label{linear:def:Morawetz-norm}
  Consider $h:\StaticRegionWithExtension\to\Complex^D$ such that
  $h(\tStar,\cdot)$ is compactly supported on $\TrappingNbhd$ for all
  $\tStar\ge 0$. Then for
  $\DomainOfIntegration$ as defined in \eqref{linear:eq:ILED-trapping:trapping-reg-def}
  we define the \emph{local Morawetz norm} by
  \begin{equation*}
    \begin{split}
      \norm{h}^2_{\MorawetzNorm_b(\DomainOfIntegration)}
      ={}& \norm{ (D_t-\sigma_2(D, x)) {h}}^2_{L^2_{\ell_1}(\DomainOfIntegration)}
      + \norm{ (D_t-\sigma_1(D, x)) {h}}^2_{L^2_{\ell_2}(\DomainOfIntegration)}
      + \norm{\p_r {h}}^2_{L^2(\DomainOfIntegration)} + \norm{{h}}^2_{L^2(\DomainOfIntegration)}.
    \end{split}
  \end{equation*}
  We also define the 
  higher-order Morawetz norms
  \begin{equation*}
    \norm{h}_{\MorawetzNorm^k_b(\DomainOfIntegration)}^2 = \sum_{|\alpha|\le {k-1}}\norm{\RedShiftK^{\alpha}h}_{\MorawetzNorm_b(\DomainOfIntegration)}^2. 
  \end{equation*}
  Observe that $\MorawetzNorm^1(\DomainOfIntegration) =
  \MorawetzNorm(\DomainOfIntegration)$. We also have the equivalent
  norms over the spacelike slice $\TrappingNbhd$,
  \begin{align*}
      \norm{h}_{\InducedMorawetzNorm_b(\TrappingNbhd)}^2
      &={} \norm{ (D_t-\sigma_2(D, x)) {h}}_{\InducedLTwo_{\ell_1}(\TrappingNbhd)}^2
      + \norm{ (D_t-\sigma_1(D, x)) {h}}_{\InducedLTwo_{\ell_2}(\TrappingNbhd)}^2
        + \norm{\p_r {h}}_{\InducedLTwo(\TrappingNbhd)}^2
        + \norm{{h}}_{\InducedLTwo(\TrappingNbhd)}^2,\\
    \norm{h}_{\InducedMorawetzNorm_b^k(\TrappingNbhd)}^2 & = \sum_{|\alpha|\le {k-1}}\norm{\RedShiftK^{\alpha}h}_{\InducedMorawetzNorm_b(\TrappingNbhd)}^2. 
  \end{align*}
\end{definition}

\begin{remark}
  Note that for the \SdS{} metric $g_{b_0}$, the Morawetz norm in
  Definition \ref{linear:def:Morawetz-norm} reduces to
  \begin{equation*}
    \begin{split}
      \norm{{h}}_{\MorawetzNorm_{b_0}(\DomainOfIntegration)}^2
      = \norm{(r-3M)\p_t{h}}_{L^2(\DomainOfIntegration)}^2
      + \norm{(r-3M)\NablaAngular {h}}_{L^2(\DomainOfIntegration)}^2
      + \norm{\p_r {h}}_{L^2(\DomainOfIntegration)}^2
      + \norm{{h}}_{L^2(\DomainOfIntegration)}^2,
    \end{split}
  \end{equation*}
  which agrees with our earlier definition in
  \eqref{linear:eq:ILED-near:SdS:Morawetz-norm} of the Morawetz norm on \SdS.
\end{remark}
\begin{remark}
  Observe that the local Morawetz norm is fine-tuned so that all
  the derivatives taken of ${h}$ have symbols which vanish exactly on
  the trapped set. This will be exploited heavily in what follows to
  prove the desired Morawetz estimate. 
\end{remark}


We state the main theorem of this section. 
\begin{theorem}
  \label{linear:thm:ILED-near:main}
  Let $\DomainOfIntegration$ and $\TrappingNbhd$ be as defined in
  \eqref{linear:eq:ILED-trapping:trapping-reg-def}, and $u$ be a sufficiently
  smooth function supported in $\TrappingNbhd$ such that
  $\hat{u}(\xi, \eta)$ is supported on the region
  $\frac{\abs*{\xi}^2}{|\eta|^2} \le \frac{1}{2}\delta_\zeta$.  Then
  for $\delta_r, \delta_\zeta$ sufficiently small, there exists
  $\SpectralGap>0$, $C_0>0$, such that for $k\ge k_0$, where $k_0$ is
  the threshold regularity level defined in
  \eqref{linear:eq:threshold-reg-def}, such that 
  \begin{equation}
    \label{linear:eq:ILED-near:main}
    \norm{u}_{\MorawetzNormk{k}_b(\TrappingNbhd)} \lesssim
    \norm*{\widehat{\LinEinstein}_{g_b}(\sigma)u}_{\InducedHk{k-1}_\sigma(\TrappingNbhd)},\qquad
    \text{if }\Im\sigma> -\SpectralGap,\text{ and } |\sigma| \ge C_0;
    \text{ or } \Im\sigma = -\SpectralGap.
  \end{equation}
\end{theorem}

Like in the case for the scalar wave on
\SdS{} in Section \ref{linear:sec:ILED-trapping:SdS-toy-model}, we will
prove the resolvent estimate in \eqref{linear:eq:ILED-near:main} in two
parts. We first prove the resolvent estimate for
$\abs*{\Im\sigma}\le \SpectralGap$ for some $\SpectralGap>0$, and
then we prove the resolvent estimate for $\Im\sigma >
\frac{\SpectralGap}{2}$. Again like in the case for the scalar wave
equation on \SdS, these two estimates correspond to a Morawetz
estimate and a $\p_{\tStar}$-energy estimate respectively.

We give a brief outline of the proof:
\begin{enumerate}
  
\item We begin with proving the resolvent estimate for
  $\abs*{\Im\sigma}\le \SpectralGap$. The bulk of the proof will be
  dedicated, as in the case in Section
  \ref{linear:sec:ILED-trapping:SdS-toy-model} for the scalar wave on \SdS{},
  to handling the bulk terms in the first line of
  \eqref{linear:eq:ILED-near:combined-divergence-theorem}. The main outline
  remains similar to the approach used in the \SdS{} case and can be
  viewed as a perturbation of the proof in Lemma
  \ref{linear:lemma:ILED-near:SdS-scalar-wave:Bulk}. The key idea is to
  extract a non-negative bulk term at the principal level that
  degenerates only at the trapped set. This is done in Lemma
  \ref{linear:lemma:ILED-near:sum-of-squares}, and relies on finding suitable
  $\MorawetzVF_b$, $\LagrangeCorr_b$ to extract a sum of squares
  expression for
  \begin{equation*}
    \frac{1}{2\ImagUnit}H_{\PrinSymb_b} \MorawetzSym_b + \PrinSymb_b \LagrangeCorrSym_b,
  \end{equation*}
  the bulk term that comes out of the principal scalar-wave component
  of $\LinEinstein_{g_b}$.  To do so, we will take pseudodifferential
  modifications of the vectorfield multiplier $\MorawetzVF_{b_0}$, and
  the Lagrangian corrector $\LagrangeCorr_{b_0}$
  constructed in Lemma \ref{linear:lemma:ILED-near:SdS-scalar-wave:Bulk} that
  take into account the more complicated (in particular
  frequency-dependent) nature of trapping in \KdS{}
  to guarantee the desired degenerate ellipticity.
  
\item We will treat the remaining terms in the integration by parts
  argument  as small perturbations of the
  positive bulk we obtained in the previous step. We first control in
  Corollary \ref{linear:coro:ILED-near:SX-control} the
  remaining terms at the principal level. These come from the
  contribution of the subprincipal symbol and the pseudo-differential
  conjugation, and have symbol $\SubPConjSym_b \MorawetzSym_b$.  It
  is critical here that for $b=b_0$, this symbol can be made
  arbitrarily small by an appropriate choice of
  $\PseudoSubPFixer$. This allows us to continue to treat it as a
  small perturbation of the scalar wave for $b$ close to $ b_0$.
  
\item We will also use the degenerate ellipticity to show that the
  lower-order bulk terms are appropriately controlled in Lemma
  \ref{linear:lemma:ILED-near:LoT-control}. Unlike in the non-trapping
  regimes, here we cannot simply control these terms via a
  high-frequency argument. Since the ellipticity that we recover at
  the principal level is degenerate at $\TrappedSet_b$, we can only
  control degenerate lower-order terms. We will get around this
  difficulty by showing that we can suitably modify the lower-order
  terms so that they respect the aforementioned degeneracy at the cost
  of a derivative.

\item The final step in proving the resolvent estimate for
  $\abs*{\Im\sigma}\le \SpectralGap$ is to show that when
  $h=e^{-\ImagUnit\sigma\tStar}u$, the boundary terms themselves can
  also be controlled by the degenerate ellipticity of the bulk term
  for some $\SpectralGap>0$ sufficiently small. We emphasize
  that in this proof, we do not appeal to energy-boundedness. In fact,
  the only time we do use a $\p_{\tStar}$-energy estimate is to show
  that the energy grows at most like $e^{\epsilon \tStar}$ for
  $\epsilon\ll 1$. This is done in Lemma
  \ref{linear:lemma:ILED-near:boundary-terms}. 

\item Finally, we proving a resolvent for
  $\abs*{\Im\sigma}\le \SpectralGap$. This step makes use of a naive
  Gronwall-based energy estimate which takes advantage of the
  smallness of the subprincipal symbol when microlocalized near
  trapping.
\end{enumerate}

We divide the proof into sections as outlined above. 
\paragraph{Principal level bulk terms}

We first handle the bulk term rising from the principal scalar wave
component of $\LinEinstein_{g_b}$. The main degenerate positivity is
the following.
\begin{lemma}
  \label{linear:lemma:ILED-near:bulk-positivity}
  Let $\MorawetzVF_{b_0}$, $\LagrangeCorr_{b_0}$ be as defined in
  Lemma \ref{linear:lemma:ILED-near:SdS-scalar-wave:Bulk} (specifically, as
  defined in \eqref{linear:eq:ILED-near:SdS:X-def} and
  \eqref{linear:eq:ILED-near:SdS:q-def} respectively).

  Then for sufficiently
  small $a$, there exists some
  \begin{equation*}
    \widetilde{\MorawetzVF}\in \Op S^1 +
      \Op S^0\p_{\tStar},\qquad
    \tilde{\LagrangeCorr} \in \Op S^0 +
      \Op S^{-1}\p_{\tStar},
  \end{equation*}
  depending smoothly on $a$ such that defining 
  \begin{equation*}
    \MorawetzVF_b := \MorawetzVF_{b_0} + a\widetilde{\MorawetzVF},\qquad
    \LagrangeCorr_b:= \LagrangeCorr_{b_0} + a\tilde{\LagrangeCorr}, 
  \end{equation*}
  there exists a Hermitian (with respect to the
  $L^2(\DomainOfIntegration)$ inner product) operator
  \begin{equation*}
    \tilde{\mathfrak{\LagrangeCorr}}_{b_0}\in (a+\delta_r)\Op S^0,
  \end{equation*}
  such that 
  \begin{align*}
    \norm*{h}_{\MorawetzNorm_b(\DomainOfIntegration)}^2
    \lesssim{}&\int_{\DomainOfIntegration}\KCurrent{\MorawetzVF_{b_0},\LagrangeCorr_{b_0},0}[h]
                + a \KCurrentIbP{\widetilde{\MorawetzVF}, \tilde{\LagrangeCorr}}[h]
                - 2\Re\bangle*{\SubPConjOp_b h, \MorawetzVF_{b} h}_{L^2(\DomainOfIntegration)}
                + 2\Re\bangle*{\SubPConjOp_bh, \tilde{\mathfrak{\LagrangeCorr}}_{b_0}h}_{L^2(\DomainOfIntegration)}\notag \\
              &+ 2\evalAt*{\Re\bangle*{\SubPConjOp_{0}h, a\widetilde{\MorawetzVF}h}_{L^2(\TrappingNbhd)}}^{\tStar=\TStar}_{\tStar=0}
                + 2\evalAt*{\Re\bangle*{\SubPConjOp_bh, a\widetilde{\MorawetzVF}_0h}_{L^2(\TrappingNbhd)}}^{\tStar=\TStar}_{\tStar=0},
  \end{align*}
\end{lemma}

The main key to proving Lemma \ref{linear:lemma:ILED-near:bulk-positivity}
will be an appropriate degenerate positivity in principal-order bulk
terms for an appropriately chosen multiplier. This stems from the fact
that the operator $\LinEinstein_{g_b}$ is strongly hyperbolic, and the
positive bulk term gained from commuting with the scalar wave
operator. The proof of the lemma below follows closely that of Tataru
and Tohaneanu in the Kerr setting in Lemma 4.3 of
\cite{tataru_local_2010}.
\begin{lemma}
  \label{linear:lemma:ILED-near:sum-of-squares}
  Let
  \begin{equation*}
    \MorawetzSym_{b_0} = \ImagUnit f_{b_0}\xi,\qquad
    \LagrangeCorrSym_{b_0} = \LagrangeCorr_{b_0} - \frac{1}{2}\nabla_{g_{b_0}}\cdot\MorawetzVF_{b_0},
  \end{equation*}
  be symbols corresponding to the choice of vectorfield multiplier and
  Lagrangian corrector in Proposition
  \ref{linear:prop:ILED-near:SdS}\footnote{We have included the divergence
    term in the Lagrange corrector for computational convenience,
    since it effectively plays a similar role.}.  Then for 
  sufficiently small $a$, there exist smooth homogeneous symbols
  $\tilde{\MorawetzSym}\in S^1 + \sigma S^0$,
  $\tilde{\LagrangeCorrSym}\in S^0+\sigma S^{-1}$ that depend smoothly
  on $a$ such that defining,
  \begin{equation*}
    \MorawetzSym_b = \MorawetzSym_{b_0} + a\tilde{\MorawetzSym},\qquad \LagrangeCorrSym_b = \LagrangeCorrSym_{b_0}+ a\tilde{\LagrangeCorrSym},
  \end{equation*}
  for $|r-3M|< \delta_r$, we have the following sum of squares
  representation
  \begin{equation}
    \label{linear:eq:ILED-near:sum-of-squares:KdS}
    \rho^2\left(
      \frac{1}{2\ImagUnit}H_{\PrinSymb_b} \MorawetzSym_{b}
      + \PrinSymb_b  
        \LagrangeCorrSym_{b}  
    \right) = \sum_{j=1}^7\SquareDecomp_j^2
  \end{equation}
  where $\SquareDecomp_j\in S^1 + \sigma S^0$. Moreover,
  \begin{enumerate}
  \item The decomposition  \eqref{linear:eq:ILED-near:sum-of-squares:KdS} 
    extends the decomposition in \eqref{linear:eq:ILED-near:sum-of-squares:SdS}
    in the sense that
    \begin{equation*}
      \begin{split}
        &(\SquareDecomp_1, \SquareDecomp_2, \SquareDecomp_3, \SquareDecomp_4, \SquareDecomp_5) \mod a(S^1_{hom} + \sigma S^0_{hom}) \\
        ={}& \left(
          \sqrt{\delta_1}\alpha_{b_0}\sigma,
          \beta_{b_0}\xi,
          \sqrt{\frac{\mu_{b_0}(1-\delta_1)}{r^2}}\alpha_{b_0}\FreqTheta,
          \sqrt{\frac{\mu_{b_0}(1-\delta_1)}{r^2}}\alpha_{b_0}\FreqPhi,
          \sqrt{\frac{\mu_{b_0}(1-\delta_1)}{r^2}}\alpha_{b_0}\xi
        \right)
        ,
      \end{split}      
    \end{equation*}
    and
    \begin{equation*}
      (\SquareDecomp_6,\SquareDecomp_7) \in \sqrt{a}(S^1_{hom}+\sigma S^0_{hom}). 
    \end{equation*}
  \item $\{\SquareDecomp_j\}_{1\le j\le 7}$ is elliptically equivalent to the family
    of symbols $(\ell_2(\sigma-\sigma_1),\ell_1(\sigma-\sigma_2),\xi)$
    in the sense that there exists a symbol valued matrix $\mathbb{M}\in
    M^{7\times 3}(S^0)$ with maximum rank $3$ everywhere such that
    \begin{equation*}
      \SquareDecomp = \mathbb{M} \mathfrak{b},\quad \mathfrak{b}=
      \begin{pmatrix}
        \ell_2(\sigma-\sigma_1)\\
        \ell_1(\sigma-\sigma_2)\\
        \xi
      \end{pmatrix}.
    \end{equation*}
  \end{enumerate}
\end{lemma}

\begin{proof}
  As discussed above, the main idea is to find appropriate
  pseudo-differential modifications of the vectorfield
  $\MorawetzVF_{b_0}$ and the Lagrangian correction
  $\LagrangeCorr_{b_0}$ that are adapted to the perturbed trapping
  dynamics of \KdS.
  
  For the sake of simplifying some of our ensuing calculations, define
  \begin{equation*}
    \LagrangeCorrSym'_{b_0} = \LagrangeCorrSym_{b_0} - 2\PoissonB{\ln
      \rho_b, \MorawetzSym_{b_0}}, \qquad \tilde{\LagrangeCorrSym}'_{b_0} =
    \tilde{\LagrangeCorrSym}_{b_0} -2\PoissonB{\ln\rho_b,
      \tilde{\MorawetzSym}}, 
  \end{equation*}
  so that
  \begin{equation*}
    \rho^2_b\left(
      \frac{1}{2\ImagUnit}H_\PrinSymb
      (\MorawetzSym_{b_0}+a\tilde{\MorawetzSym})
      +(\LagrangeCorrSym_{b_0}+a\tilde{\LagrangeCorrSym})\PrinSymb_b  
    \right)
    = \frac{1}{2\ImagUnit}H_{\RescaledPrinSymb} (\MorawetzSym_{b_0}+a\tilde{\MorawetzSym}) +
    \left(\tilde{\LagrangeCorrSym}'
      +a\tilde{\LagrangeCorrSym}'\right)(\RescaledPrinSymb_b).
  \end{equation*}
  We first choose $\tilde{\MorawetzSym}$ so that
  $H_{\RescaledPrinSymb_b} (\MorawetzSym_{b_0}+a\tilde{\MorawetzSym})$  
  vanishes at the trapped set $\TrappedSet_b$. The most immediate
  \KdS{} extension of the choice of $\MorawetzSym_{b_0}$ is the symbol
  \begin{equation*}
    \MorawetzSym_{b}':=\ImagUnit \Delta_{b}\rho_b^{-4}\left(r-\rTrapping_b(\sigma,\FreqPhi)\right)\xi
    = \ImagUnit\frac{r-\rTrapping_b(\sigma,\FreqPhi)}{2\rho_b^4}H_{\RescaledPrinSymb_{b}}r.
  \end{equation*}
  This symbol would clearly extend our choice in \SdS{} in the sense
  that
  \begin{equation*}
    \MorawetzSym_{b_0}' = \MorawetzSym_{b_0},
  \end{equation*}
  and moreover, $\MorawetzSym_{b}'$ is well defined and smooth in a
  neighborhood of the trapped set. We can calculate that on the
  characteristic set $\PrinSymb_b=0$, we have
  \begin{equation*}
    2 H_{\RescaledPrinSymb}\MorawetzSym_b' =
    \left(\frac{1}{\rho_b^4} - \frac{4(r-\rTrapping_b)\p_r\rho_b}{\rho_b^5}\right)(H_{\RescaledPrinSymb_b} r)^2
    + \frac{r-\rTrapping_b(\sigma,\FreqPhi)}{\rho_b^4}H_{\RescaledPrinSymb_b}^2r.
  \end{equation*}
  Recall from \eqref{linear:eq:Trapping-KdS:H2pr-at-Trapping} that for
  $\PrinSymb=0$, near $r=3M$, we have that
  \begin{equation*}
    H_{\RescaledPrinSymb_b}^2r
    = 2\Delta_b\p_r\left(\frac{(1+\lambda_b)^2}{\Delta_b}\left((r^2+a^2)\sigma +a\FreqPhi\right)\right).
  \end{equation*}
  Since $\rTrapping_b$ is the unique minimum of
  $\frac{(1+\lambda_b)^2}{\Delta_b}\left((r^2+a^2)\sigma
      +a\FreqPhi\right)$, and we are in a $\delta_r$
  neighborhood of $r=3M$, there exist
  positive symbols $\alpha,\beta\in \SymClass^0_{\hom}$ such that on
  $\PrinSymb_b=0$, near $r=3M$,
  \begin{equation}
    \label{linear:eq:ILED-near:sum-of-squares:KdS-ILED-sym-characteristic}
    H_{\RescaledPrinSymb}\MorawetzSym_b' =
    \alpha^2(r,\sigma, \FreqPhi) \sigma^2(r-\rTrapping_b)^2 + \beta^2(r,\sigma, \FreqPhi)\xi^2.
  \end{equation}  
  Unfortunately, the problem with $\MorawetzSym_b'$ is that it is not
  a polynomial in $\sigma$, and thus cannot be directly used in
  conjunction with our integration-by-parts or divergence theorem
  method to produce a Morawetz estimate. To overcome this difficulty,
  recall that we defined $\MorawetzSym_b$ so that it is smooth in $a$,
  and so that
  \begin{equation*}
    \MorawetzSym_b'-\MorawetzSym_{b_0}\in aS^1_{hom}.
  \end{equation*}
  Thus the Mather division theorem (Theorem
  \ref{linear:thm:Mather-division}) gives us 
  \begin{equation}
    \label{linear:eq:ILED-near:x-div-theorem}
    \frac{1}{\ImagUnit}\left(\MorawetzSym_b' - \MorawetzSym_{b_0}\right)
    = a\left(\tilde{\MorawetzSym}_1(r,\theta;\xi,\FreqTheta, \FreqPhi)
      + \tilde{\MorawetzSym}_0(r,\theta;\xi, \FreqTheta, \FreqPhi)\sigma\right)
    + a \rAux(r,\theta;\xi,\FreqTheta, \FreqPhi)\PrinSymb_b,
  \end{equation}
  where $\tilde{\MorawetzSym}_i\in S^i_{hom}$ and $\rAux \in S^{-1}_{hom}$.
  Now, we
  define
  \begin{equation*}
    \frac{1}{\ImagUnit}\tilde{\MorawetzSym} = \tilde{\MorawetzSym}_1+\tilde{\MorawetzSym}_0\sigma,
  \end{equation*}
  so that on $\PrinSymb_b=0$,
  \begin{equation*}
    \MorawetzSym_b = \MorawetzSym_{b_0} + a\tilde{\MorawetzSym} = \MorawetzSym_b'.
  \end{equation*}
  Thus $\MorawetzSym_b$ is a symbol which is a polynomial in $\sigma$
  and moreover vanishes at the trapped set $\TrappedSet_{b}$.
  
  \textit{A priori}, $H_\RescaledPrinSymb \tilde{\MorawetzSym}_b$ is a
  third degree polynomial in $\sigma$. Applying the Mather division
  theorem (Theorem \ref{linear:thm:Mather-division}) again yields that there
  exist some $\gamma_1\in S^1, \gamma_2 \in S^2$,
  $f_0\in S^0, f_{-1}\in \sigma S^{-1}$ such that
  \begin{equation*}
    \frac{1}{2\ImagUnit\rho_b^2} H_{\RescaledPrinSymb_b}(\MorawetzSym_{b_0}+a\tilde{\MorawetzSym})
    + \LagrangeCorrSym_{b_0}' (\RescaledPrinSymb_b)
    = \gamma_2+\gamma_1\sigma  + \left(
      e_{b_0} + a(f_0+f_{-1}\sigma)\right)(\sigma-\sigma_1)(\sigma-\sigma_2),
  \end{equation*}
  observing that
  \begin{equation*}
    e_{b_0} := \delta_1\alpha_{b_0}^2. 
  \end{equation*}
  is the coefficient for $\sigma^2$ in the expression for
  $\frac{1}{2\ImagUnit}H_{\PrinSymb_{b_0}}\MorawetzSym_{b_0} +
  \LagrangeCorrSym_{b_0}\PrinSymb_{b_0}$ (see \eqref{linear:eq:ILED-near:sum-of-squares:SdS}).                 
  It now remains to demonstrate that $\gamma_2+\gamma_1\sigma
  +e_{b_0}(\sigma-\sigma_1)(\sigma-\sigma_2)$ can be expressed as a
  sum of squares up to some error in
  $a(S^0+ S^{-1}\sigma)\PrinSymb_b$. If this were true, we could write
  \begin{equation}
    \label{linear:eq:ILED-near:sum-of-squares:aux-1}
    \gamma_2+\gamma_1\sigma + e_{b_0}(\sigma-\sigma_1)(\sigma-\sigma_2) = \sum \SquareDecomp_j^2 + a(g_0+g_{-1}\sigma)(\sigma-\sigma_1)(\sigma-\sigma_2). 
  \end{equation}
  We could then define $\tilde{\LagrangeCorrSym}$ such that
  \begin{equation*}
    \tilde{\LagrangeCorrSym}' = -2\left(f_0+g_0+(f_{-1}+g_{-1})\sigma\right),
  \end{equation*}
  so that the $a(S^0+\sigma S^{-1})\PrinSymb_b$ terms are all canceled.

  We now return to showing
  (\ref{linear:eq:ILED-near:sum-of-squares:aux-1}). Recall that on
  $\PrinSymb_b=0$,
  \begin{equation*}
    H_{\RescaledPrinSymb_b}(\MorawetzSym_{b_0} + a\tilde{\MorawetzSym}) = H_{\RescaledPrinSymb_b}\MorawetzSym_b'.
  \end{equation*}
  As a result of
  (\ref{linear:eq:ILED-near:sum-of-squares:KdS-ILED-sym-characteristic}), we
  now have that if $\sigma=\sigma_i$, which in particular implies that $\PrinSymb_b=0$,
  \begin{equation*}
    \gamma_2+\gamma_1\sigma =\alpha^2(r, \sigma, \FreqPhi)\sigma^2(r-\rTrapping_b)^2 + \beta^2(r,\sigma,\FreqPhi)\xi^2.
  \end{equation*}
  We can solve for $\gamma_2, \gamma_1$ explicitly now by considering
  the two-dimensional system of equations
  \begin{equation*}
    \begin{split}
      \gamma_2+\gamma_1\sigma_i &= \frac{1}{4}\alpha_i^2(\sigma_1-\sigma_2)^2 + \beta_i^2\xi^2,\\
      \alpha_i &= \frac{2|\sigma_i|}{\sigma_1-\sigma_2}\alpha(r,
      \sigma_i, \FreqPhi)(r-\rTrapping_b(\sigma_i, \FreqPhi))\in S^0,\\
      \beta_i &= \beta(r, \sigma_i, \FreqPhi) \in S^0.
    \end{split}
  \end{equation*}
  Solving the system yields 
  \begin{equation}
    \begin{split}
      \label{linear:eq:ILED-near:sum-of-squares:gamma-def}
      \gamma_2 &= \frac{1}{4}(\sigma_1-\sigma_2)(\alpha_2^2\sigma_1 - \alpha_1^2\sigma_2) + \frac{\sigma_1\beta_2^2 - \sigma_2\beta_1^2}{\sigma_1-\sigma_2}\xi^2,  \\
      \gamma_1 &= \frac{1}{4}(\sigma_1-\sigma_2)(\alpha_1^2-\alpha_2^2)+\frac{\beta_1^2 - \beta_2^2}{\sigma_1-\sigma_2}\xi^2. 
    \end{split} 
  \end{equation}
  We first add together the first two terms in $\gamma_i$ to see that
  \begin{equation}
    \label{linear:eq:ILED-near:sum-of-squares:gamma-sum-first-terms}
    \begin{split}
      (\sigma_1-\sigma_2)\left(\alpha_2^2\sigma_1 -\alpha_1^2\sigma_2 + \sigma(\alpha_1^2-\alpha_2^2)\right)
      ={}&(1-{\delta_1})(\alpha_1(\sigma-\sigma_2)-\alpha_2(\sigma-\sigma_1))^2\\
      &+ {\delta_1}\left(\alpha_1(\sigma-\sigma_2)+\alpha_2(\sigma-\sigma_1)\right)^2\\
      &- 4e_b(\sigma-\sigma_1)(\sigma-\sigma_2),
    \end{split}
  \end{equation}
  where
  \begin{equation*}
    e_b = \frac{(\alpha_1-\alpha_2)^2}{4}+ {\delta_1}\alpha_1\alpha_2. 
  \end{equation*}
  Recall that in $g_{b_0}$, $\alpha_1=\alpha_2=\alpha_{b_0}$,
  $\sigma_2=-\sigma_1$, and that $\LagrangeCorrSym_{b_0} = \delta_1\alpha_{b_0}^2$.
  This implies that
  \begin{equation}
    \label{linear:eq:ILED-near:sum-of-squares:e-diff}
    e_b-e_{b_0}\in a(S^0 + \sigma S^{-1})
  \end{equation}
  as desired.
  We now add together the second terms in the $\gamma_i$ given in \eqref{linear:eq:ILED-near:sum-of-squares:gamma-def}
  \begin{equation}
    \label{linear:eq:ILED-near:sum-of-squares:gamma-sum-second-terms}
    \begin{split}
      \frac{\sigma_1\beta_2^2-\sigma_2\beta_1^2}{\sigma_1-\sigma_2}
      + \sigma\frac{\beta_1^2-\beta_2^2}{\sigma_1-\sigma_2}
      ={}&
      \frac{1}{2}\left(\beta_1^2+\beta_2^2 -Ca\right) +
      \frac{(Ca-\beta_2^2+\beta_1^2)(\sigma-\sigma_2)^2}{2(\sigma_1-\sigma_2)^2}\\
      &+
      \frac{(Ca-\beta_1^2+\beta_2^2)(\sigma-\sigma_1)^2}{2(\sigma_1-\sigma_2)^2}
      + O(a)\PrinSymb_b.
    \end{split}
  \end{equation}
  Summing \eqref{linear:eq:ILED-near:sum-of-squares:gamma-sum-first-terms}
  and \eqref{linear:eq:ILED-near:sum-of-squares:gamma-sum-second-terms}
  together, we have that
  \begin{equation*}
    \begin{split}
      &\frac{1}{2\ImagUnit}H_{\RescaledPrinSymb_b}(
      \MorawetzSym_{b_0}+a\tilde{\MorawetzSym}) + (\RescaledPrinSymb_b)\LagrangeCorrSym'_b\\
      ={}& \frac{1-{\delta_1}}{4}\left(
        \alpha_1(\sigma-\sigma_2)-\alpha_2(\sigma-\sigma_1)
      \right)^2
      + \frac{{\delta_1}}{4}\left(\alpha_1(\sigma-\sigma_2) + \alpha_2(\sigma-\sigma_1)\right)^2\\
      &+ \frac{1}{2}\left(\beta_1^2+\beta_2^2 - Ca\right)\xi^2
      + \frac{(Ca-\beta_2^2+\beta_1^2)(\sigma-\sigma_2)^2}{2(\sigma_1-\sigma_2)^2}\xi^2\\
      &+ \frac{(Ca-\beta_1^2+\beta_2^2)(\sigma-\sigma_1)^2}{2(\sigma_1-\sigma_2)^2}\xi^2
      +a(S^0+S^{-1}\sigma)(\sigma-\sigma_1)(\sigma-\sigma_2).
    \end{split}
  \end{equation*}
  We then pick
  \begin{align*}
    \SquareDecomp_1^2 &= \frac{\delta_1}{4}\left(\alpha_1(\sigma-\sigma_2)+\alpha_2(\sigma-\sigma_1)\right)^2,\\
    \SquareDecomp_2^2 &= \frac{1}{2}\left(\beta_1^2 + \beta_2^2 - Ca\right)\xi^2,\\
    \SquareDecomp_{3}^2 &= \frac{\FreqTheta^2}{|\FreqAngular|^2 + \Delta_{b_0}\xi^2}\frac{(1-\delta_1)}{4} \left(\alpha_1(\sigma-\sigma_2)-\alpha_2(\sigma-\sigma_1)\right)^2, \\ 
    \SquareDecomp_{4}^2 &= \frac{\FreqPhi^2}{|\FreqAngular|^2 + \Delta_{b_0}\xi^2}\frac{(1-\delta_1)}{4} \left(\alpha_1(\sigma-\sigma_2)-\alpha_2(\sigma-\sigma_1)\right)^2,\\
    \SquareDecomp_{5}^2 &= \frac{\Delta_{b_0}\xi^2}{|\FreqAngular|^2 + \Delta_{b_0}\xi^2}\frac{(1-\delta_1)}{4}\left(\alpha_1(\sigma-\sigma_2)-\alpha_2(\sigma-\sigma_1)\right)^2,\\
    \SquareDecomp_6^2 &= \frac{\left(Ca-\beta_2^2 + \beta_1^2\right)\left(\sigma-\sigma_2\right)^2}{2\left(\sigma_1-\sigma_2\right)^2}\xi^2,\\
    \SquareDecomp_7^2 &= \frac{\left(Ca-\beta_1^2 + \beta_2^2\right)\left(\sigma-\sigma_1\right)^2}{2\left(\sigma_1-\sigma_2\right)^2}\xi^2,
  \end{align*}
  concluding the proof of Lemma \ref{linear:lemma:ILED-near:sum-of-squares}.
\end{proof}

Before we proceed, we note the following useful rewriting of
$\MorawetzSym_b$.
\begin{corollary}
  \label{linear:cor:ILED-near:alt-form-Xb}
  Let $\MorawetzSym_b$ be as constructed in the proof of Lemma
  \ref{linear:lemma:ILED-near:sum-of-squares}. Then we can write that
  \begin{equation}
    \label{linear:eq:ILED-near:alt-form-Xb:form}
    \MorawetzSym_b = \frac{\Delta_{b}}{\rho_b^4}(r-\rTrapping_b(\sigma_1,
    \FreqPhi))\frac{\xi(\sigma-\sigma_2)}{\sigma_1-\sigma_2} +
    \frac{\Delta_{b}}{\rho_b^4}(r-\rTrapping_b(\sigma_2, \FreqPhi))\frac{\xi(\sigma-\sigma_1)}{\sigma_2-\sigma_1}.
  \end{equation}
  As a result, if $u$ is compactly supported on $\TrappingNbhd$, then
  \begin{equation}
    \label{linear:eq:ILED-near:alt-form-Xb:bound}
    \norm{\MorawetzVF_bu}_{L^2(\TrappingNbhd)} \lesssim \norm{u}_{\MorawetzNorm_b(\TrappingNbhd)}.
  \end{equation}
\end{corollary}
\begin{proof}
  Using the observation that $\MorawetzSym_b\in \sigma S^0 + S^1$, we
  can apply the Mather division theorem (Theorem
  \ref{linear:thm:Mather-division}) to write that there exist
  $\mathfrak{e}_{(i)}$, $\mathfrak{r}_{(i)}$, for $i=1,2$ such that
  \begin{equation}
    \label{linear:cor:ILED-near:alt-form-Xb:eqn1}
    \begin{split}
      \MorawetzSym_b &= (\sigma-\sigma_1)\mathfrak{e}_{(1)} + \mathfrak{r}_{(1)},\\
      \MorawetzSym_b &= (\sigma-\sigma_2)\mathfrak{e}_{(2)} + \mathfrak{r}_{(2)}.
    \end{split}
  \end{equation}
  Moreover, using $\eqref{linear:eq:ILED-near:x-div-theorem}$, we have that
  \begin{equation*}
    \label{linear:cor:ILED-near:alt-form-Xb:eqn2}
    \begin{split}
      (r-\rTrapping_b(\sigma_1, \FreqPhi))\frac{H_{\RescaledPrinSymb_b}r}{2\rho_b^4} = \mathfrak{r}_1,\\
      (r-\rTrapping_b(\sigma_2, \FreqPhi))\frac{H_{\RescaledPrinSymb_b}r}{2\rho_b^4} = \mathfrak{r}_2.
    \end{split}
  \end{equation*}
  Solving the combined system of equations given by
  \eqref{linear:cor:ILED-near:alt-form-Xb:eqn1} and 
  \eqref{linear:cor:ILED-near:alt-form-Xb:eqn2} for $\mathfrak{e}_{(i)}$ and
  $\mathfrak{r}_{(i)}$, $i=1,2$ then yields
  \eqref{linear:eq:ILED-near:alt-form-Xb:form} using the fact that
  \begin{equation*}
    H_{\RescaledPrinSymb_b}r = 2\Delta_{b}.
  \end{equation*}
  The bound in
  \eqref{linear:eq:ILED-near:alt-form-Xb:bound} is an immediate corollary of
  Lemma \ref{linear:lemma:ILED-near:sum-of-squares}.
\end{proof}

We now illustrate how to account for the contribution of the
subprincipal operator in the principal bulk term. 
\begin{corollary}
  \label{linear:coro:ILED-near:SX-control}
  Let
  \begin{equation}
    \label{linear:eq:SubPConjSym:def}
    \SubPConjSym_{b,a} = \sigma_1(\SubPConjOp_{b,a})
  \end{equation}
  denote the principal symbol of the skew-adjoint component of the
  subprincipal operator of $\LinEinsteinConj_{g_b}$.  For $a<a_0$ of
  Lemma \ref{linear:lemma:ILED-near:sum-of-squares}, and the symbols
  $\tilde{\MorawetzSym}, \tilde{\LagrangeCorrSym}$ of Lemma
  \ref{linear:lemma:ILED-near:sum-of-squares}, there exists a choice of
  $\PseudoSubPFixer$, such that in a small neighborhood of
  $\TrappedSet_b$, (in both frequency and physical space),
  \begin{equation}
    \label{linear:eq:KCurrentIbPSym-def}
    \rho_b^2\KCurrentIbPSym{\MorawetzSym_b,\LagrangeCorrSym_b} := \rho_b^2 \left(
      \frac{1}{2\ImagUnit}H_{\PrinSymb_b}\MorawetzSym_{b} +
      \PrinSymb_b\LagrangeCorrSym_b - \SubPConjSym_{b,a}\MorawetzSym_b
    \right)
    \ge
    \frac{1}{2}\sum_{j=1}^7\SquareDecomp_{j}^2.
  \end{equation}
\end{corollary}

\begin{proof}
  We use the characterization of $\MorawetzSym_b$ in Corollary
  \ref{linear:cor:ILED-near:alt-form-Xb} to write 
  \begin{equation*}
    \SubPConjSym_b\MorawetzSym_b
    = \frac{\Delta_b}{\rho_b^4}(r-\rTrapping_b(\sigma_1, \FreqPhi))\frac{\xi(\sigma-\sigma_2) \SubPConjSym_b}{\sigma_1-\sigma_2}
    + \frac{\Delta_b}{\rho_b^4}(r-\rTrapping_b(\sigma_2, \FreqPhi))\frac{\xi(\sigma-\sigma_1) \SubPConjSym_b}{\sigma_2-\sigma_1}.
  \end{equation*}
  Recall from Lemma \ref{linear:lemma:ILED-near:s-decomp} that for any fixed
  $\delta_0>0$, there exists a choice of $a$,
  $\varepsilon_{\TrappedSet_{b_0}}$, $\delta_r$, and
  $\delta_\zeta$ such that   
  \begin{equation*}
    \abs*{\SubPConjSym_b}\lesssim \delta_0\abs*{\zeta}.
  \end{equation*}  
  
  On $\TrappingNbhd$, $|\xi|^2 + |\FreqAngular|^2
  \lesssim \sigma_1^2 + \sigma_2^2$, so it will be sufficient to bound
  \begin{equation*}
    \begin{split}
      \abs*{\frac{\Delta_b}{\rho_b^4}(r-\rTrapping_b(\sigma_1,\FreqPhi))\frac{\xi(\sigma-\sigma_2)\sigma}{\sigma_1-\sigma_2}},
      &\quad\abs*{\frac{\Delta_b}{\rho_b^4}(r-\rTrapping_b(\sigma_2,\FreqPhi))\frac{\xi(\sigma-\sigma_1)\sigma}{\sigma_2-\sigma_1}},\\
      \abs*{\frac{\Delta_b}{\rho_b^4}(r-\rTrapping_b(\sigma_1,\FreqPhi))\frac{\xi(\sigma-\sigma_2)\sigma_1}{\sigma_1-\sigma_2}},
      &\quad\abs*{\frac{\Delta_b}{\rho_b^4}(r-\rTrapping_b(\sigma_2,\FreqPhi))\frac{\xi(\sigma-\sigma_1)\sigma_1}{\sigma_2-\sigma_1}},\\
      \abs*{\frac{\Delta_b}{\rho_b^4}(r-\rTrapping_b(\sigma_1,\FreqPhi))\frac{\xi(\sigma-\sigma_2)\sigma_2}{\sigma_1-\sigma_2}},
      &\quad\abs*{\frac{\Delta_b}{\rho_b^4}(r-\rTrapping_b(\sigma_2,\FreqPhi))\frac{\xi(\sigma-\sigma_1)\sigma_2}{\sigma_2-\sigma_1}},
    \end{split}
  \end{equation*}
  by $\sum_{j=1}^7\SquareDecomp_j^2$ to conclude. 
  
  The two terms in each line are handled in an identical manner
  (simply switching $\sigma_1$, $\sigma_2$), so without loss of generality, we will only
  handle one term from each line.  Let us first consider the symbol
  given by 
  \begin{equation*}
    \frac{\Delta_b}{\rho_b^4}(r-\rTrapping_b(\sigma_1, \FreqPhi))\frac{\xi(\sigma-\sigma_2) \sigma}{\sigma_1-\sigma_2}.
  \end{equation*}
  We can immediately write 
  \begin{equation}
    \label{linear:eq:ILED-near:SubPOp:aux1}
    \begin{split}
      \frac{\Delta_b}{\rho_b^4}(r-\rTrapping_b(\sigma_1,
      \FreqPhi))\frac{\xi(\sigma-\sigma_2)
        \sigma}{\sigma_1-\sigma_2}
      ={}&
      \frac{\Delta_b}{\rho_b^4}(r-\rTrapping_b(\sigma_1, \FreqPhi))\frac{\xi (\sigma
        -\sigma_1)(\sigma-\sigma_2)}{\sigma_1-\sigma_2}\\
      & + \frac{\Delta_b}{\rho_b^4}(r-\rTrapping_b(\sigma_1, \FreqPhi))\frac{\xi(\sigma-\sigma_2)\sigma_1}{\sigma_1-\sigma_2}.
    \end{split}
  \end{equation}
  Recalling the explicit forms of
  $\curlyBrace*{\SquareDecomp_j}_{j=1}^{7}$, we see that
  applying Cauchy-Schwarz yields that
  \begin{equation*}
    \frac{\Delta_b}{\rho_b^4}(r-\rTrapping_b(\sigma_1, \FreqPhi))\frac{\xi
      (\sigma-\sigma_2)\sigma_1}{\sigma_1-\sigma_2} \lesssim \SquareDecomp_1^2 + \SquareDecomp_2^2.
  \end{equation*}
  To deal with the first term on the right-hand side of
  \eqref{linear:eq:ILED-near:SubPOp:aux1}, we again apply Cauchy Schwarz, 
  \begin{equation*}
    \frac{\Delta_b}{\rho_b^4}(r-\rTrapping_b(\sigma_1, \FreqPhi))\frac{\xi (\sigma
      -\sigma_1)(\sigma-\sigma_2)}{\sigma_1-\sigma_2}
    \le \frac{\Delta_b^2}{2\rho_b^8}\frac{\xi^2(\sigma-\sigma_1)^2}{(\sigma_1-\sigma_2)^2} + \frac{1}{2}(r-\rTrapping_b(\sigma_1,\FreqPhi))^2(\sigma-\sigma_2)^2.
  \end{equation*}
  The first term on the right-hand side is controlled by
  $\SquareDecomp_7^2$, while the second term is controlled by
  $\mathfrak{b}$.  Now let us consider the symbol
  \begin{equation*}
    \frac{\Delta_b}{\rho_b^4}(r-\rTrapping_b(\sigma_1,\FreqPhi))\frac{\xi (\sigma-\sigma_2)\sigma_1}{\sigma_1-\sigma_2}.
  \end{equation*}
  As previously mentioned, this term can be controlled directly by
  applying Cauchy-Schwarz, using the explicit forms of
  $\curlyBrace*{\SquareDecomp_j}_{j=1}^7$,
  \begin{equation*}
    \frac{\Delta_b}{\rho_b^4}(r-\rTrapping_b(\sigma_1,\FreqPhi))\frac{\xi (\sigma-\sigma_2)\sigma_1}{\sigma_1-\sigma_2} \lesssim \SquareDecomp_1^2 + \SquareDecomp_2^2.
  \end{equation*}
  We now handle the final term,
  \begin{equation*}
    \frac{\Delta_b}{\rho_b^4}(r-\rTrapping_b(\sigma_1,\FreqPhi))\frac{\xi (\sigma-\sigma_2)\sigma_2}{\sigma_1-\sigma_2}.
  \end{equation*}
  To control this term, we use that
  \begin{equation*}
    \begin{split}
      \abs*{\frac{\Delta_b}{\rho_b^4}(r-\rTrapping_b(\sigma_1,\FreqPhi))\frac{\xi (\sigma-\sigma_2)\sigma_2}{\sigma_1-\sigma_2}}
      \le{}& \abs*{\frac{\Delta_b}{\rho_b^4}(r-\rTrapping_b(\sigma_1,\FreqPhi))\frac{\xi (\sigma-\sigma_2)\sigma_1}{\sigma_1-\sigma_2}}\\
      &+ \abs*{\frac{\Delta_b}{\rho_b^4}(r-\rTrapping_b(\sigma_1,\FreqPhi))\xi (\sigma-\sigma_2)}.
    \end{split}      
  \end{equation*}
  At this point we can again use Cauchy-Schwarz to bound
  \begin{equation*}
    \abs*{\frac{\Delta_b}{\rho_b^4}(r-\rTrapping_b(\sigma_1,\FreqPhi))\frac{\xi (\sigma-\sigma_2)\sigma_1}{\sigma_1-\sigma_2}}
    + \abs*{\frac{\Delta_b}{\rho_b^4}(r-\rTrapping_b(\sigma_1,\FreqPhi))\xi (\sigma-\sigma_2)}
    \lesssim \SquareDecomp_{1}^2 + \SquareDecomp_2^2 + \mathfrak{b}^2.
  \end{equation*}
  
  Combining the above estimates, we see that if $\delta_0$ is chosen
  sufficiently small,
  \begin{equation*}
    \abs*{\SubPConjSym_b\MorawetzSym_b} \le \frac{1}{2}\sum_{j=1}^7 \SquareDecomp_j^2,
  \end{equation*}
  which concludes the proof of Corollary \ref{linear:coro:ILED-near:SX-control}. 
\end{proof}
We are now ready to prove Lemma \ref{linear:lemma:ILED-near:bulk-positivity}.
\begin{proof}[Proof of Lemma \ref{linear:lemma:ILED-near:bulk-positivity}]
  We pick
  \begin{equation*}
    \widetilde{\MorawetzVF}_i = \frac{1}{2}\left(\Op(\widetilde{\MorawetzSym}_i) - \Op(\widetilde{\MorawetzSym}_i)^*\right), \qquad
    \tilde{\LagrangeCorr}_i = \left(\Op(\tilde{\LagrangeCorrSym}_i) + \Op(\tilde{\LagrangeCorrSym}_i) \right),
  \end{equation*}
  where the adjoint is taken with respect to the $L^2(\TrappingNbhd)$
  inner product.
  
  Let us rewrite
  \begin{equation}
    \label{linear:eq:ILED-near:bulk:positivity:adjusted-X-q-def}
    \mathfrak{\MorawetzVF}_{b_0} := \MorawetzVF_{b_0} + \frac{1}{2}\nabla_{g_{b}} \cdot \MorawetzVF_{b_0},\qquad
    \tilde{\mathfrak{\LagrangeCorr}}_{b_0} := \LagrangeCorr_{b_0} - \frac{1}{2}\nabla_{g_{b}} \cdot \MorawetzVF_{b_0},
  \end{equation}
  so that $\mathfrak{\MorawetzVF}_{b_0}$ is anti-Hermitian and
  $\tilde{\mathfrak{\LagrangeCorr}}_{b_0}$ is Hermitian with
  respect to the $L^2(\DomainOfIntegration)$ inner product.
  Observe that by construction,
  \begin{equation*}    
    \tilde{\mathfrak{\LagrangeCorr}}_{b_0} = \frac{1}{2}\left(\nabla_{g_b}-\nabla_{g_{b_0}} \right)\cdot \MorawetzVF_{b_0} + \LagrangeCorr_1, 
  \end{equation*}
  where $\LagrangeCorr_1$ is as defined in
  \eqref{linear:eq:ILED-near:SdS:q1-def}.
  This directly implies that 
  \begin{equation*}
    \tilde{\mathfrak{\LagrangeCorr}}_{b_0} \lesssim a + \delta_r. 
  \end{equation*}
  Then directly by integrating by parts, we have that
  \begin{align}
    2\Re\bangle*{\SubPConjOp_b h , \left(\MorawetzVF_{b_0} + \LagrangeCorr_{b_0}\right) h}_{L^2(\DomainOfIntegration)}
    ={}& - \bangle*{\left(\mathfrak{\MorawetzVF}_{b_0}\SubPConjOp_b + \SubPConjOp_b\mathfrak{\MorawetzVF}_{b_0}\right) h ,  h}_{L^2(\DomainOfIntegration)}
         + 2\Re\bangle*{\SubPConjOp_bh, \tilde{\mathfrak{\LagrangeCorr}}_{b_0}h}_{L^2(\DomainOfIntegration)}\notag \\
    \Re\bangle*{\SubPConjOp_b h , \widetilde{\MorawetzVF}  h}_{L^2(\DomainOfIntegration)}
    ={}& - \bangle*{\left(\widetilde{\MorawetzVF}\SubPConjOp_b + \SubPConjOp_b\widetilde{\MorawetzVF}\right) h ,  h}_{L^2(\DomainOfIntegration)}         
         - \evalAt*{\Re\bangle*{\SubPConjOp_0h, \widetilde{\MorawetzVF} h}_{L^2(\TrappingNbhd)}}^{\tStar=\TStar}_{\tStar=0}\notag\\
         &- \evalAt*{\Re\bangle*{\SubPConjOp_bh, \widetilde{\MorawetzVF}_0h}_{L^2(\TrappingNbhd)}}^{\tStar=\TStar}_{\tStar=0},
         \label{linear:eq:ILED-near:bulk-positivity:IbP-S-X}
  \end{align}
  where
  \begin{equation*}
    \SubPConjOp_{b} = \SubPConjOp_0\p_{\tStar} + \SubPConjOp_1,\qquad
    \widetilde{\MorawetzVF} = \widetilde{\MorawetzVF}_0\p_{\tStar} + \widetilde{\MorawetzVF}_1, 
  \end{equation*}
  and $\widetilde{\MorawetzVF}_i, \SubPConjOp_i\in \Op S^i$.
  Then Corollary \ref{linear:coro:ILED-near:SX-control} and Theorem
  \ref{linear:thm:sym-ineq-to-op-est} show that
  \begin{equation*}
    \int_{\DomainOfIntegration}\KCurrent{\MorawetzVF_{b_0}, \LagrangeCorr_{b_0}, 0}[h]
    + a\KCurrentIbP{\widetilde{\MorawetzVF}, \tilde{\LagrangeCorr}}[h]
    + \Re\bangle*{\left(\mathfrak{\MorawetzVF}_{b_0}\SubPConjOp_b + \SubPConjOp_b\mathfrak{\MorawetzVF}_{b_0}\right) h ,  h}_{L^2(\DomainOfIntegration)}
    \gtrsim \norm*{h}_{\MorawetzNorm_b(\DomainOfIntegration)}^2. 
  \end{equation*}
  The conclusion of the lemma then follows quickly from the definition
  of $\tilde{\mathfrak{\LagrangeCorr}}_{b_0}$ and
  \eqref{linear:eq:ILED-near:bulk-positivity:IbP-S-X}. 
\end{proof}
\begin{remark}
  At first glance, the pseudo-differential operators
  $\widetilde{\MorawetzVF}$ and $\tilde{\LagrangeCorr}$ are only
  well-defined away from the singularities of the $(\theta,\phi)$
  coordinates on the sphere, namely, the poles. However, we can
  smoothly extend both $\widetilde{\MorawetzVF}$ and
  $\tilde{\LagrangeCorr}$ to the poles by repeating their
  constructions in charts covering the poles and gluing the resulting
  operators together since their construction only relies on
  $\sigma_1$, $\sigma_2$, which are smooth at the poles. 
\end{remark}

\paragraph{Lower-order bulk terms}

Next, we show that the lower-order bulk terms can be appropriately
controlled by the (degenerate) ellipticity of the principal symbol. We
start with some auxiliary lemmas that will be useful in controlling
the lower-order terms that appear in the divergence theorem argument.

The easiest terms to control will be the lowest-order bulk terms. 
\begin{lemma}
  \label{linear:lemma:ILED-near:LoT-control:lowest-order-control}
  Let $\DomainOfIntegration$ be as defined in
  \eqref{linear:eq:ILED-trapping:trapping-reg-def}. Fix $\delta>0$ and
  $\PotentialOp\in \Op S^0$. Then, for $\delta_r$ sufficiently small
  and $h$ supported on $\DomainOfIntegration$, 
  \begin{equation*}
    \abs*{\Re \bangle*{\PotentialOp h, h}_{L^2(\DomainOfIntegration)}} \lesssim \delta\left(
      \norm{h}_{L^2(\DomainOfIntegration)}^2 + \norm{\p_rh}_{L^2(\DomainOfIntegration)}^2\right),
  \end{equation*}
  and in particular,
  \begin{equation*}
    \norm*{h}_{L^2(\DomainOfIntegration)}^2
    \lesssim \delta\left(
      \norm{h}_{L^2(\DomainOfIntegration)}^2 + \norm{\p_rh}_{L^2(\DomainOfIntegration)}^2\right). 
  \end{equation*}
  Similarly, if $h(\tStar, \cdot)$ is compactly supported on
  $\TrappingNbhd$ for all $\tStar\ge 0$, then for $\delta_r$
  sufficiently small,
  \begin{equation*}
    \abs*{\Re \bangle*{\PotentialOp h, h}_{L^2(\TrappingNbhd)}}
    \lesssim \delta\left(
      \norm{h}_{L^2(\TrappingNbhd)}^2 + \norm{\p_rh}_{L^2(\TrappingNbhd)}^2\right). 
  \end{equation*}
  In particular,
  \begin{equation*}
    \norm*{h}_{L^2(\TrappingNbhd)}^2
    \lesssim \delta\left(
      \norm{h}_{L^2(\TrappingNbhd)}^2 + \norm{\p_rh}_{L^2(\TrappingNbhd)}^2\right). 
  \end{equation*}
\end{lemma}
\begin{proof}
  Using the fact that $\p_r(r-3M)=1$, we use integration by
  parts to see that
  \begin{equation*}
    -\Re \bangle*{\PotentialOp h, h}_{L^2(\TrappingNbhd)}
    = \bangle*{(r-3M)\PotentialOp\p_rh,h}_{L^2(\DomainOfIntegration)} +
    \bangle*{(r-3M)\PotentialOp h,\p_rh}_{L^2(\DomainOfIntegration)}
    Err[h],
  \end{equation*}
  where $\abs*{Err[h]} \lesssim
  \delta_r\norm{h}_{L^2(\DomainOfIntegration)}^2$. The first then
  follows by Cauchy-Schwarz and taking $\delta_r$ sufficiently
  small. Observe that since the argument only involved integration by
  parts in $\p_r$, we can repeat the argument over $\TrappingNbhd$
  instead of over $\DomainOfIntegration$ to achieve the second
  conclusion. 
\end{proof}

Throughout the proof, we will also accumulate lower-order terms of the
form $\bangle{D_x}^{-1}D_{\tStar}$ which need to be dealt with. Fortunately,
this can be done with a simple symbol decomposition. 
\begin{lemma}
  \label{linear:lemma:ILED-near:LoT-control:zero-order-mixed-term}
  Let $\DomainOfIntegration$ be as defined in
  \eqref{linear:eq:ILED-trapping:trapping-reg-def}. Fix $\delta_0>0$, and
  some $\PotentialOp \in S^0 + S^{-1}D_{\tStar}$, where $S^0$ and
  $S^{-1}$ are Hermitian with respect to the $L^2(\TrappingNbhd)$
  Hermitian inner product \footnote{The restriction to Hermitian
    $S^0$, $S^{-1}$ is unnecessary, but is sufficient here.}. Then for
  $\delta_r$ sufficiently small and $h$ such that $h(\tStar, \cdot)$
  is supported on $\TrappingNbhd$ for all $\tStar\ge 0$, we have that
  \begin{equation*}
    \abs*{\Re\bangle*{\PotentialOp h, h}_{L^2(\DomainOfIntegration)}} < \delta_0 \norm*{h}_{\MorawetzNorm_b(\DomainOfIntegration)}^2.
  \end{equation*}
\end{lemma}
\begin{proof}
  Fix some auxiliary $\delta>0$. 
  For $\PotentialOp\in S^0$ we can conclude directly using Lemma
  \ref{linear:lemma:ILED-near:LoT-control:lowest-order-control}, so it
  suffices to just consider
  $\PotentialOp = \PotentialOp_{-1}D_{\tStar}\in S^{-1}D_{\tStar}$ with
  symbol $\PotentialSym_0\sigma$. 
  The critical observation is that we can decompose
  \begin{equation*}
    \sigma = \sigma^{(1)} + \sigma^{(2)},\qquad
    \sigma^{(1)}:= \frac{\sigma-\sigma_2}{\sigma_1-\sigma_2},\qquad
    \sigma^{(2)}:= - \frac{\sigma-\sigma_1}{\sigma_1-\sigma_2}. 
  \end{equation*}
  We can now take advantage of the fact that
  \begin{equation*}
    \p_r(r-\rTrapping_b(\sigma_1, \FreqPhi)) = \p_r(r-\rTrapping_b(\sigma_2, \FreqPhi)) = 1
  \end{equation*}
  to obtain after integration by parts that
  \begin{align}
    - \bangle*{\PotentialOp h,h}_{L^2(\DomainOfIntegration)}
    ={}& \sum_{i=1,2}\bangle*{(r-\rTrapping_i)\Op\left(\sigma^{(i)}\right)\circ \p_rh,\PotentialOp_{-1} h}_{L^2(\DomainOfIntegration)}\notag\\
    &+ \sum_{i=1,2}\bangle*{(r-\rTrapping_i)\Op\left(\sigma^{(i)}\right)h,\p_r \circ \PotentialOp_{-1} h}_{L^2(\DomainOfIntegration)}\notag\\
    & + \sum_{i=1,2}\bangle*{(r-\rTrapping_i)\squareBrace*{\p_r, \Op\left(\sigma^{(i)}\right)}h,h}_{L^2(\DomainOfIntegration)}
    + Err, \label{linear:lemma:ILED-near:LoT-control:zero-order-mixed-term:IbP} 
  \end{align}
  where
  \begin{equation*}
    \rTrapping_i := \rTrapping_b(\sigma_i, \FreqPhi), 
  \end{equation*}
  and $Err$ is an error term satisfying the control
  \begin{equation*}
    \abs*{Err} < \delta\norm*{h}_{\MorawetzNorm_b(\DomainOfIntegration)}^2. 
  \end{equation*}
  The terms on the right-hand side of
  \eqref{linear:lemma:ILED-near:LoT-control:zero-order-mixed-term:IbP} can
  now be handled individually.

  To handle the first term on the right-hand side of
  \eqref{linear:lemma:ILED-near:LoT-control:zero-order-mixed-term:IbP}, we
  observe that
  \begin{equation*}
    \norm*{\Op\left(\left(\frac{(r-r_1)(\sigma-\sigma_2)}{\sigma_1-\sigma_2} - \frac{(r-r_2)(\sigma-\sigma_2)}{\sigma_1-\sigma_2}\right)\xi \right)h}_{L^2(\DomainOfIntegration)} \lesssim \norm*{h}_{\MorawetzNorm_b(\DomainOfIntegration)},
  \end{equation*}
  so that using Cauchy Schwarz and Lemma
  \ref{linear:lemma:ILED-near:LoT-control:lowest-order-control}, we have in
  fact that
  \begin{equation*}
    \abs*{
      \sum_{i=1,2}\bangle*{(r-\rTrapping_i)\Op\left(\sigma^{(i)}\right)\circ \p_rh,\PotentialOp_{-1} h}_{L^2(\DomainOfIntegration)}}
    \lesssim \delta \norm*{h}_{\MorawetzNorm_b(\DomainOfIntegration)}^2. 
  \end{equation*}  

  The second term on the right-hand side of
  \eqref{linear:lemma:ILED-near:LoT-control:zero-order-mixed-term:IbP} is
  handled directly by Cauchy-Schwarz and Lemma
  \ref{linear:lemma:ILED-near:LoT-control:lowest-order-control}, taking
  $\delta_r$ to be sufficiently small.

  Finally, to control the third term on the right-hand side of
  \eqref{linear:lemma:ILED-near:LoT-control:zero-order-mixed-term:IbP}, we
  observe that
  \begin{equation*}
    \PoissonB*{\xi, \sigma^{(1)}}
    = - \p_r(\sigma_1-\sigma_2)^{-1}(\sigma-\sigma_2)
    -\frac{\p_r\sigma_2}{\sigma_1-\sigma_2},\qquad
    \PoissonB*{\xi, \sigma^{(2)}}
    =  \p_r(\sigma_1-\sigma_2)^{-1}(\sigma-\sigma_1)
    + \frac{\p_r\sigma_1}{\sigma_1-\sigma_2}.
  \end{equation*}
  We see then that using Cauchy-Schwarz and Lemma
  \ref{linear:lemma:ILED-near:LoT-control:lowest-order-control} and taking
  $\delta_r$ to be sufficiently small, 
  \begin{align*}
    \abs*{\bangle*{(r-\rTrapping_1)\Op\left(\p_r(\sigma_1-\sigma_2)^{-1}(\sigma-\sigma_2)\right)h,\PotentialOp_{-1} h}_{L^2(\DomainOfIntegration)}}
    &< \delta \norm*{h}_{\MorawetzNorm_b(\DomainOfIntegration)}^2,\\
    \abs*{\bangle*{(r-\rTrapping_2)\Op\left(\p_r(\sigma_1-\sigma_2)^{-1}(\sigma-\sigma_1)\right)h,\PotentialOp_{-1} h}_{L^2(\DomainOfIntegration)}}
    &< \delta \norm*{h}_{\MorawetzNorm_b(\DomainOfIntegration)}^2.
  \end{align*}
  Moreover, directly by taking $\delta_r$ to be sufficiently small, we
  have that
  \begin{align*}
    &\abs*{\bangle*{(r-\rTrapping_1)\Op\left(\frac{\p_r\sigma_2}{\sigma_1-\sigma_2}\right)h,\PotentialOp_{-1} h}_{L^2(\DomainOfIntegration)}}
    +\abs*{\bangle*{(r-\rTrapping_2)\Op\left(\frac{\p_r\sigma_1}{\sigma_1-\sigma_2}\right)h,\PotentialOp_{-1} h}_{L^2(\DomainOfIntegration)}}\\
    <{}& \delta \norm*{h}_{\MorawetzNorm_b(\DomainOfIntegration)}^2. 
  \end{align*}
  This concludes the proof of Lemma
  \ref{linear:lemma:ILED-near:LoT-control:zero-order-mixed-term}.     
\end{proof}

At the level of first-order bulk terms that arise in the application
of the divergence theorem, a skew-Hermitian operator can easily be
handled via an integration by parts argument.
\begin{lemma}
  \label{linear:lemma:ILED-near:LoT-control:symmetry-control}
  Fix $\delta>0$. Let $\SubPOp = \SubPOp_0D_{\tStar} + \SubPOp_1$ be a
  first-order pseudo-differential operator such that
  $ \SubPOp_i \in \Op S^i(\Sigma)$ are skew-Hermitian with respect to
  the $\InducedLTwo(\Sigma)$ inner product respectively.
  
  Then, for $h$ such that $h(\tStar,\cdot)$ is compactly
  supported in $\TrappingNbhd$ for all $\tStar\ge 0$, with
  $\TrappingNbhd$ defined as in
  \eqref{linear:eq:ILED-trapping:trapping-reg-def},
  \begin{equation}
    \label{linear:eq:ILED-near:LoT-control:symmetry-control:bulk}
    \Re\bangle*{\SubPOp h, h}_{L^2(\DomainOfIntegration)}
    = \evalAt*{\Re\bangle*{\SubPOp_0h, h}_{\InducedLTwo(\Sigma_{\tStar})}}^{\tStar=T}_{\tStar=0}.    
  \end{equation}

  Similarly, if there exist $\sigma$ and $u$ such that $h =
  e^{-\ImagUnit\sigma\tStar}u$, and $\SubPOp$ is as specified above,
  then we also have that
  \begin{equation}
    \label{linear:eq:ILED-near:LoT-control:symmetry-control:boundary}
    \Re\bangle*{\SubPOp h, h}_{L^2(\TrappingNbhd)} = 0. 
  \end{equation}
\end{lemma}
\begin{proof}
  The proof for \eqref{linear:eq:ILED-near:LoT-control:symmetry-control:bulk}
  is a simple integration by parts exercise. The proof for
  \eqref{linear:eq:ILED-near:LoT-control:symmetry-control:boundary} is also a
  simple integration by parts exercise, where we note that
  $D_{\tStar}h = \sigma h$. 
\end{proof}

Unfortunately, not all of the first-order bulk terms that we pick up
in the application of the divergence theorem respect the symmetry
assumptions of Lemma
\ref{linear:lemma:ILED-near:LoT-control:symmetry-control}. For these
first-order terms, we will need to use a more delicate argument
relying on exchanging the lower-order nature of the bulk terms for
some degeneracy at the trapped set. Since this argument essentially
promotes these lower-order terms to become principal level errors, to
apply the control effectively will require some additional smallness
parameter. 
\begin{lemma}
  \label{linear:lemma:ILED-near:LoT-control:exchange-degeneracy-trick}
  Let $\DomainOfIntegration$ be as defined in
  \eqref{linear:eq:ILED-trapping:trapping-reg-def}, and fix some
  $\delta_0>0$. Let
  $\SubPOp\in \Op S^1(\Sigma) + \Op S^{0}(\Sigma)\p_{\tStar}$ be a first
  order pseudo-differential operator such that
  \begin{equation*}
    \SubPOp = \SubPOp_1 + \SubPOp_0D_{\tStar},
  \end{equation*}
  where $\SubPOp_1$ and $\SubPOp_0$ are Hermitian with respect to the
  $L^2(\TrappingNbhd)$ Hermitian inner product, such that the
  symbol $\SubPSym$ of $\SubPOp$ satisfies
  \begin{equation*}
    \abs*{\SubPSym_i}\le \varepsilon_{\SubPOp}|\zeta|^i. 
  \end{equation*}
  Then for $\varepsilon_{\SubPOp}$ and $\delta_r$ sufficiently small,
  \begin{equation}
    \label{linear:eq:ILED-near:LoT-control:exchange-degeneracy-trick}
    \Re\bangle*{\SubPOp h, h}_{L^2(\DomainOfIntegration)}
    \le \delta_0 \norm*{h}_{\MorawetzNorm_b(\DomainOfIntegration)}^2
    + \evalAt*{\int_{\TrappingNbhd_{\tStar}}\widetilde{J}[h]}_{\tStar=0}^{\tStar=\TStar}, 
  \end{equation}
  where $h$ is as specified in Theorem \ref{linear:thm:ILED-near:main}, and
  \begin{equation}
    \label{linear:eq:ILED-near:LoT-control:boundary-err-control}
    \abs*{\int_{\TrappingNbhd_{\tStar}}\widetilde{J}[h]}
    \lesssim \delta_0\norm*{h}_{\InducedMorawetzNorm(\Sigma_{\tStar})}^2.
  \end{equation}
\end{lemma}

\begin{proof}
  Fix some arbitrary $\delta>0$.  Observe that for $a$ sufficiently
  small with respect to $M$ and $\Lambda$, we have that
  $\sigma_1-\sigma_2$ is an elliptic operator.  We now define
  \begin{equation*}
    \rTrapping_i := \rTrapping_b(\sigma_i, \FreqPhi),
  \end{equation*}
  and  write
  \begin{equation*}
    \SubPSym = \SubPSym_0\sigma + \SubPSym_1,\qquad \SubPSym_i\in S^i. 
  \end{equation*}    
  We will handle the two terms separately. We first handle
  $\SubPSym_0\sigma$. For this term, we define
  \begin{equation*}
    \SubPSym_0^{(2)} := -\frac{\sigma_2(\sigma-\sigma_1)}{\sigma_1-\sigma_2}\SubPSym_0,\qquad
    \SubPSym_0^{(1)} := \frac{\sigma_1(\sigma-\sigma_2)}{\sigma_1-\sigma_2}\SubPSym_0,
  \end{equation*}
  so that
  \begin{equation*}
    \SubPSym_0\sigma = \SubPSym_0^{(1)} + \SubPSym_0^{(2)}.
  \end{equation*}
  Now, integrating by parts, we have that
  \begin{align}
    -\bangle*{\SubPOp_0\p_{\tStar} h, h}_{L^2(\DomainOfIntegration)}
    ={}& \sum_{i=1,2}\bangle*{(r-\rTrapping_i)\Op\left(\SubPSym_0^{(i)}\right)\circ \p_r h, h}_{L^2(\DomainOfIntegration)}
    + \sum_{i=1,2}\bangle*{(r-\rTrapping_i)\Op\left(\SubPSym_0^{(i)}\right)h, \p_r h}_{L^2(\DomainOfIntegration)}\notag \\
    &+ \sum_{i=1,2} \bangle*{(r-\rTrapping_i)\squareBrace*{\p_r, \Op\left(\SubPSym_0^{(i)}\right)}h, h}_{L^2(\DomainOfIntegration)}
    + Err[h], \label{linear:eq:ILED-near:LoT-control:exchange-degeneracy-trick:IbP:S0}
  \end{align}
  where the error terms are lower-order terms respecting the
  degeneracy in the Morawetz norm, which  for $\delta_r$ sufficiently
  small, using Lemma
  \ref{linear:lemma:ILED-near:LoT-control:lowest-order-control}, satisfies 
  \begin{equation*}
    \abs*{Err[h]}\lesssim \delta \norm*{h}_{\MorawetzNorm_b(\DomainOfIntegration)}^2.
  \end{equation*}
  We treat each of the terms on the right-hand side of
  \eqref{linear:eq:ILED-near:LoT-control:exchange-degeneracy-trick:IbP:S1}
  individually. First, we use that up to lower-order terms,
  \begin{equation*}
    \left((r-\rTrapping_i)\Op\left(\SubPSym_0^{(i)})\right)\right)^*
    =   (r-\rTrapping_i)\Op\left(\SubPSym_0^{(i)}\right)
    - \left[(r-\rTrapping_i), \Op\left(\SubPSym_0^{(i)}\right)\right], 
  \end{equation*}
  so that using integrating by parts and using Lemma
  \ref{linear:lemma:ILED-near:LoT-control:zero-order-mixed-term} and
  Cauchy-Schwarz, 
  \begin{equation*}
    \abs*{ \Re\bangle*{(r-\rTrapping_i)\Op\left(\SubPSym_0^{(i)}\right)\circ \p_r h, h}_{L^2(\DomainOfIntegration)}}
    < \abs*{ \Re\bangle*{\p_r h, (r-\rTrapping_i)\Op\left(\SubPSym_0^{(i)}\right) h}_{L^2(\DomainOfIntegration)}}
    + Err_{\DomainOfIntegration}[h] + \evalAt*{Err_{\TrappingNbhd}[h]}_{\tStar=0}^{\tStar=\TStar},
  \end{equation*}
  where using Lemma \ref{linear:lemma:ILED-near:LoT-control:lowest-order-control}, satisfies 
  \begin{equation*}
    \abs*{Err_{\DomainOfIntegration}[h]}\lesssim \delta \norm*{h}_{\MorawetzNorm_b(\DomainOfIntegration)}^2, \qquad
    \abs*{Err_{\TrappingNbhd}}\lesssim \delta\norm*{h}_{\InducedMorawetzNorm_b(\Sigma_{\TrappingNbhd})}^2.
  \end{equation*} 
  Then, using Lemma~\ref{linear:lemma:ILED-near:sum-of-squares}, we have that
  \begin{equation*}
    \abs*{ \Re\bangle*{\p_r h, (r-\rTrapping_i)\Op\left(\SubPSym_0^{(i)}\right) h}_{L^2(\DomainOfIntegration)}}
    \lesssim \varepsilon_{\SubPOp}\norm*{h}_{\MorawetzNorm_b(\DomainOfIntegration)},
  \end{equation*}
  which is controlled for $\varepsilon_{\SubPOp}$ sufficiently small. 
  Similarly, we have that
  \begin{equation*}
    \norm*{(r-\rTrapping_i)\Op\left(\SubPSym_0^{(i)}\right)h}_{L^2(\DomainOfIntegration)} \lesssim \varepsilon_{\SubPOp}\norm*{h}_{\MorawetzNorm_b(\DomainOfIntegration)},
  \end{equation*}
  which is controlled for $\varepsilon_{\SubPOp}$ sufficiently small. 
  Then using Cauchy-Schwarz and Lemma
  \ref{linear:lemma:ILED-near:LoT-control}, we have that for $\delta_r$
  sufficiently small,
  \begin{equation*}
    \sum_{i=1,2}\abs*{\bangle*{(r-\rTrapping_2)(r-\rTrapping_1)\Op\left(\SubPSym_{0}^{(i)}\right)h,h}_{L^2(\DomainOfIntegration)}}\lesssim \delta \norm*{h}_{\MorawetzNorm_b(\DomainOfIntegration)}^2. 
  \end{equation*}
  To handle the commutator term in
  \eqref{linear:eq:ILED-near:LoT-control:exchange-degeneracy-trick:IbP:S0},
  we first observe that
  \begin{align*}
    \PoissonB*{\xi, \SubPSym_0^{(2)}} &= -(\sigma-\sigma_1)\p_r\left(\frac{\sigma_2\SubPSym_0}{\sigma_1-\sigma_2}\right)
                                        +\frac{ \SubPSym_0\sigma_2\p_r\sigma_1}{\sigma_1-\sigma_2},\\
    \PoissonB*{\xi, \SubPSym_0^{(1)}} &= (\sigma-\sigma_2)\p_r\left(\frac{\sigma_1\SubPSym_0}{\sigma_1-\sigma_2}\right)
                                        -\frac{ \SubPSym_0\sigma_1\p_r\sigma_2}{\sigma_1-\sigma_2},
  \end{align*}
  Thus, we can decompose
  \begin{align}
    &\sum_{i=1,2} \bangle*{(r-\rTrapping_i)\squareBrace*{\p_r, \Op\left(\SubPSym_0^{(i)}\right)}h, h}_{L^2(\DomainOfIntegration)}\notag\\
    ={}& \bangle*{(r-\rTrapping_1)\Op\left(
      (\sigma-\sigma_2)\p_r\left(\frac{\sigma_1\SubPSym_0}{\sigma_1-\sigma_2}\right)
      \right)h,h}_{L^2(\DomainOfIntegration)}
      - \bangle*{(r-\rTrapping_1)\Op\left(
      \frac{ \SubPSym_0\sigma_1\p_r\sigma_2}{\sigma_1-\sigma_2}
      \right)h,h}_{L^2(\DomainOfIntegration)}\notag\\
    &- \bangle*{(r-\rTrapping_2)\Op\left(
      (\sigma-\sigma_1)\p_r\left(\frac{\sigma_2\SubPSym_0}{\sigma_1-\sigma_2}\right)
      \right)h,h}_{L^2(\DomainOfIntegration)}
      + \bangle*{(r-\rTrapping_2)\Op\left(
      \frac{ \SubPSym_0\sigma_2\p_r\sigma_1}{\sigma_1-\sigma_2}
      \right)h,h}_{L^2(\DomainOfIntegration)}. \label{linear:eq:ILED-near:LoT-control:exchange-degeneracy-trick:IbP:S0:commutator:decomposition}
  \end{align}
  We see that the first terms in each line of the right-hand side of
  \eqref{linear:eq:ILED-near:LoT-control:exchange-degeneracy-trick:IbP:S0:commutator:decomposition}
  are controlled by the Morawetz norm. Thus, using Cauchy-Schwarz and
  Lemma \ref{linear:lemma:ILED-near:LoT-control:lowest-order-control}, we
  have that for $\delta_r$ sufficiently small, 
  \begin{align*}
    \delta \norm*{h}_{\MorawetzNorm_b(\DomainOfIntegration)}^2\gtrsim{}
    &\abs*{\bangle*{(r-\rTrapping_1)\Op\left(
          (\sigma-\sigma_2)\p_r\left(\frac{\sigma_1\SubPSym_0}{\sigma_1-\sigma_2}\right)
        \right)h,h}_{L^2(\DomainOfIntegration)}}\\
    &+ \abs*{\bangle*{(r-\rTrapping_2)\Op\left(
          (\sigma-\sigma_1)\p_r\left(\frac{\sigma_2\SubPSym_0}{\sigma_1-\sigma_2}\right)
      \right)h,h}_{L^2(\DomainOfIntegration)}}.
  \end{align*}
  It remains to handle the second terms in each line of the right-hand
  side of
  \eqref{linear:eq:ILED-near:LoT-control:exchange-degeneracy-trick:IbP:S0:commutator:decomposition}. Without
  loss of generality, we handle
  $(\sigma_1-\sigma_2)^{-1}{\SubPSym_0\sigma_2\p_r\sigma_1}$ since
  $(\sigma_1-\sigma_2)^{-1}{\SubPSym_0\sigma_1\p_r\sigma_2}$ is
  handled identically. To this end, define
  \begin{equation*}
    \tilde{\SubPSym}_{0,2} := \frac{\SubPSym_0\sigma_2\p_r\sigma_1}{\sigma_1-\sigma_2}\in  S^1,\qquad
    \tilde{\SubPSym}_{0,2}^{(2)} := -\frac{\sigma-\sigma_1}{\sigma_1-\sigma_2} \tilde{\SubPSym}_{0,2},\qquad
    \tilde{\SubPSym}_{0,2}^{(1)} := \frac{\sigma-\sigma_2}{\sigma_1-\sigma_2} \tilde{\SubPSym}_{0,2},
  \end{equation*}
  so that
  \begin{equation*}
    \sum_{i=1,2}\tilde{\SubPSym}_{0,2}^{(i)} = \tilde{\SubPSym}_{0,2}. 
  \end{equation*}
  Then we have that
  \begin{equation*}
    -\bangle*{(r-\rTrapping_2)\Op\left(\tilde{\SubPSym}_{0,2}\right)h,h}_{L^2(\DomainOfIntegration)}
    = -\sum_{i=1,2}\bangle*{(r-\rTrapping_2) \Op\left(\tilde{\SubPSym}_{0,2}^{(i)}\right)h,h}_{L^2(\DomainOfIntegration)}.
  \end{equation*}
  Then observe that by construction,
  \begin{equation*}
    \norm*{(r-\rTrapping_2)\Op\left(\tilde{\SubPSym}_{0,2}^{(2)}\right)h}_{L^2(\DomainOfIntegration)} \lesssim \norm*{h}_{\MorawetzNorm_b(\DomainOfIntegration)}.
  \end{equation*}
  Thus, it suffices to control
  $\bangle*{(r-\rTrapping_2)\Op\left(\tilde{\SubPSym}_{0,2}^{(1)}\right)h,h}_{L^2(\DomainOfIntegration)}$.   
  To this end, we use integration by parts to write that,
  \begin{align}
    -\bangle*{(r-\rTrapping_2)\Op\left(\tilde{\SubPSym}_{0,2}^{(1)}\right)h,h}_{L^2(\DomainOfIntegration)}
    ={}& \bangle*{(r-\rTrapping_1)(r-\rTrapping_2)\Op\left(\tilde{\SubPSym}_{0,2}^{(1)}\right)\circ \p_r h,h}_{L^2(\DomainOfIntegration)} \notag\\
       &+ \bangle*{(r-\rTrapping_1)(r-\rTrapping_2)\Op\left(\tilde{\SubPSym}_{0,2}^{(1)}\right)h,\p_rh}_{L^2(\DomainOfIntegration)}\notag\\
    &+ \bangle*{(r-\rTrapping_1)\Op\left(\tilde{\SubPSym}_{0,2}^{(1)}\right)h,h}_{L^2(\DomainOfIntegration)}\notag \\
    &+ \bangle*{(r-\rTrapping_1)(r-\rTrapping_2)\squareBrace*{\p_r,\Op\left(\tilde{\SubPSym}_{0,2}^{(1)}\right)}h,h}_{L^2(\DomainOfIntegration)}. \label{linear:eq:ILED-near:LoT-control:exchange-degeneracy-trick:IbP:S0:rep1}
  \end{align}
  Observing again that
  $(r-\rTrapping_1)(r-\rTrapping_2)\Op\left(\tilde{\SubPSym}_{0,2}^{(1)}\right)$
  is Hermitian up to a $\Op S^0 + \Op^{-1}\p_{\tStar}$ term, we
  control the first three terms on the right-hand side of
  \eqref{linear:eq:ILED-near:LoT-control:exchange-degeneracy-trick:IbP:S0:rep1}
  by a combination of integration by parts, Cauchy-Schwarz, and Lemma
  \ref{linear:lemma:ILED-near:LoT-control:zero-order-mixed-term} after taking
  $\varepsilon_{\SubPOp}$ and $\delta_r$ sufficiently small. To handle
  the final commutator term, we again observe that
  \begin{equation*}
    \PoissonB*{\xi, \tilde{\SubPSym}_{0,2}^{(1)}}
    = (\sigma-\sigma_2)\p_r\left(\frac{\tilde{\SubPSym}_{0,2}}{\sigma_1-\sigma_2}\right)
    - \tilde{\SubPSym}_{0,2}\frac{\p_r\sigma_2}{\sigma_1-\sigma_2}.
  \end{equation*}
  Again we can thus decompose
  \begin{align}
    &\bangle*{(r-\rTrapping_1)(r-\rTrapping_2)\squareBrace*{\p_r,\Op\left(\tilde{\SubPSym}_{0,2}^{(1)}\right)}h,h}_{L^2(\DomainOfIntegration)}\notag\\
    ={}& \bangle*{(r-\rTrapping_1)(r-\rTrapping_2)\Op\left( (\sigma-\sigma_2)\p_r\left(\frac{\tilde{\SubPSym}_{0,2}}{\sigma_1-\sigma_2}\right)\right)h,h}_{L^2(\DomainOfIntegration)}\notag\\         
    &- \bangle*{(r-\rTrapping_1)(r-\rTrapping_2)\Op\left(\tilde{\SubPSym}_{0,2}\frac{\p_r\sigma_2}{\sigma_1-\sigma_2}\right)h,h}_{L^2(\DomainOfIntegration)}.\label{linear:eq:ILED-near:LoT-control:exchange-degeneracy-trick:IbP:S0:commutator2:decomposition}
  \end{align}
  Once again, we see that the first term on the right-hand side of
  \eqref{linear:eq:ILED-near:LoT-control:exchange-degeneracy-trick:IbP:S0:commutator2:decomposition}
  is controlled directly by the Morawetz norm, using Lemma
  \ref{linear:lemma:ILED-near:LoT-control:lowest-order-control} and taking
  $\delta_r$ sufficiently small.  To control the second term on the
  right-hand side of
  \eqref{linear:eq:ILED-near:LoT-control:exchange-degeneracy-trick:IbP:S0:commutator2:decomposition},
  we partition one final time to define
  \begin{equation*}
    \tilde{\SubPSym}_{0,2,1}:= \tilde{\SubPSym}_{0,2}\frac{\p_r\sigma_2}{\sigma_1-\sigma_2}\in S^1,\qquad
    \tilde{\SubPSym}_{0,2,1}^{(2)}:= -\frac{\sigma-\sigma_1}{\sigma_1-\sigma_2}\tilde{\SubPSym}_{0,2,1},\qquad
    \tilde{\SubPSym}_{0,2,1}^{(1)}:= \frac{\sigma-\sigma_2}{\sigma_1-\sigma_2}\tilde{\SubPSym}_{0,2,1},
  \end{equation*}
  so that
  \begin{equation*}
    \tilde{\SubPSym}_{0,2,1} = \tilde{\SubPSym}_{0,2,1}^{(2)} + \tilde{\SubPSym}_{0,2,1}^{(1)}.
  \end{equation*}
  Then we have that
  \begin{equation*}
    \bangle*{(r-\rTrapping_2)(r-\rTrapping_1)\Op\left(\tilde{\SubPSym}_{0,2,1}\right)h,h}_{L^2(\DomainOfIntegration)}
    = \sum_{i=1,2}\bangle*{(r-\rTrapping_2)(r-\rTrapping_1)\Op\left(\tilde{\SubPSym}_{0,2,1}^{(i)}\right)h,h}_{L^2(\DomainOfIntegration)}.
  \end{equation*}
  But by construction, 
  \begin{equation*}
    \norm*{(r-\rTrapping_2)(r-\rTrapping_1)\Op\left(\tilde{\SubPSym}_{0,2,1}^{(1)}\right)h}_{L^2(\DomainOfIntegration)}
    + \norm*{(r-\rTrapping_2)(r-\rTrapping_1)\Op\left(\tilde{\SubPSym}_{0,2,1}^{(2)}\right)h}_{L^2(\DomainOfIntegration)}
    \lesssim \norm*{h}_{\MorawetzNorm_b(\DomainOfIntegration)}.
  \end{equation*}
  As a result, using Cauchy-Schwarz and Lemma
  \ref{linear:lemma:ILED-near:LoT-control}, we have that
  \begin{equation*}
    \sum_{i=1,2}\abs*{\bangle*{(r-\rTrapping_2)(r-\rTrapping_1)\Op\left(\tilde{\SubPSym}_{0,2,1}^{(i)}\right)h,h}_{L^2(\DomainOfIntegration)}} < \delta \norm*{h}_{\MorawetzNorm_b(\DomainOfIntegration)}^2. 
  \end{equation*}    
  Since $\delta$ was arbitrary, we can choose it sufficiently small so
  that
  \begin{equation}
    \label{linear:eq:ILED-near:LoT-control:exchange-degeneracy-trick:S0:result}
    \abs*{\Re\bangle*{\SubPOp_0D_{\tStar} h, h}_{L^2(\DomainOfIntegration)}} < \frac{\delta_0}{2} \norm*{h}_{\MorawetzNorm_b(\DomainOfIntegration)}^2. 
  \end{equation}
 
  We now move on to handling $\Re\bangle*{\SubPOp_1 h,
    h}_{L^2(\DomainOfIntegration)}$. The main idea will be the same as
  when handling $\SubPOp_0\p_{\tStar}$. To this end, we define
  \begin{equation*}
    \SubPSym_1^{(2)}:= -\frac{\sigma-\sigma_1}{\sigma_1-\sigma_2}\SubPSym,\qquad
    \SubPSym_1^{(1)}:= \frac{\sigma-\sigma_2}{\sigma_1-\sigma_2}\SubPSym,
  \end{equation*}
  so that
  \begin{equation*}
    \SubPSym = \SubPSym_1^{(1)}+\SubPSym_1^{(2)}. 
  \end{equation*}  
  Integrating by parts, we have that
  \begin{align}
    -\bangle*{\Op(\SubPSym_1^{(i)})h, h}_{L^2(\DomainOfIntegration)}
    ={}& \bangle*{(r-\rTrapping_i)\Op\left(\SubPSym_1^{(i)}\right)\circ \p_r h, h}_{L^2(\DomainOfIntegration)}
    + \bangle*{(r-\rTrapping_i)\Op\left(\SubPSym_{1}^{(i)}\right)h, \p_r h}_{L^2(\DomainOfIntegration)}\notag \\
    &+ \bangle*{(r-\rTrapping_i)\squareBrace*{\p_r, \Op\left(\SubPSym_1^{(i)}\right)}h, h}_{L^2(\DomainOfIntegration)}
    + Err[h], \label{linear:eq:ILED-near:LoT-control:exchange-degeneracy-trick:IbP:S1}
  \end{align}
  where for $\delta_r$ sufficiently small,
  \begin{equation*}
    \abs*{Err[h]} \lesssim \delta \norm*{h}_{\MorawetzNorm_b(\DomainOfIntegration)}^2. 
  \end{equation*}
  We treat each of the terms on the right-hand side of
  \eqref{linear:eq:ILED-near:LoT-control:exchange-degeneracy-trick:IbP:S1}
  individually. Observing that
  \begin{equation*}
    \left((r-\rTrapping_i)\Op\left(\SubPSym_1^{(i)}\right)\right)^* = (r-\rTrapping_i)\Op\left(\SubPSym_1^{(i)}\right) + \Op S^0 + \Op S^{-1}\p_{\tStar},
  \end{equation*}
  we have by integration by parts that 
  \begin{equation*}
    \bangle*{(r-\rTrapping_i)\Op\left(\SubPSym_1^{(i)}\right)\circ \p_r h, h}_{L^2(\DomainOfIntegration)}
    =  \bangle*{\p_r h, (r-\rTrapping_i)\Op\left(\SubPSym_1^{(i)}\right) h}_{L^2(\DomainOfIntegration)}
    + Err_{\DomainOfIntegration}[h] + \evalAt*{Err_{\TrappingNbhd}[h]}_{\tStar=0}^{\tStar=\TStar},
  \end{equation*}
  where by Cauchy-Schwarz, and Lemma
  \ref{linear:lemma:ILED-near:LoT-control:zero-order-mixed-term} that for
  sufficiently small $\delta_r$,
  \begin{equation*}
    \abs*{Err_{\DomainOfIntegration}[h]}
    \lesssim \delta \norm*{h}_{\MorawetzNorm_b(\DomainOfIntegration)},
    \qquad \abs*{Err_{\TrappingNbhd}} \lesssim \delta \norm*{h}_{\InducedMorawetzNorm(\TrappingNbhd)}.
  \end{equation*}
  Then from Lemma \ref{linear:lemma:ILED-near:sum-of-squares}, we have that
  \begin{equation*}
    \abs*{\bangle*{\p_r h, (r-\rTrapping_i)\Op\left(\SubPSym_1^{(i)}\right) h}_{L^2(\DomainOfIntegration)}} \lesssim \varepsilon_{\SubPOp}\norm*{h}_{\MorawetzNorm(\DomainOfIntegration)}^2,
  \end{equation*}
  which is clearly controlled for $\varepsilon_{\SubPOp}$ sufficiently
  small. 

  To handle the commutator term in
  \eqref{linear:eq:ILED-near:LoT-control:exchange-degeneracy-trick:IbP:S1},
  we proceed as previously and observe that (without loss of
  generality, we will just consider the $i=1$ case)
  \begin{align*}
    \bangle*{(r-\rTrapping_1)\squareBrace*{\p_r, \Op\left(\SubPSym_1^{(1)}\right)}h, h}_{L^2(\DomainOfIntegration)}
    ={}& \bangle*{(r-\rTrapping_1)\Op\left((\sigma-\sigma_2)\p_r\left(\frac{\SubPSym_1}{\sigma_1-\sigma_2}\right)\right)h, h}_{L^2(\DomainOfIntegration)}\\
    &- \bangle*{(r-\rTrapping_1)\Op\left(\frac{\SubPSym_1\p_r\sigma_2}{\sigma_1-\sigma_2}\right)h, h}_{L^2(\DomainOfIntegration)},
  \end{align*}
  where it is clear that the first term on the right-hand side is
  well-controlled using Lemma
  \ref{linear:lemma:ILED-near:LoT-control:lowest-order-control},
  Cauchy-Schwarz, and choosing $\delta_r$ sufficiently small as
  before. To handle the second term, we can define
  \begin{equation*}
    \tilde{\SubPSym}_{1,1}:= \frac{\SubPSym_1\p_r\sigma_2}{\sigma_1-\sigma_2}\in S^1,\qquad
    \tilde{\SubPSym}_{1,1}^{(2)} := -\frac{\sigma-\sigma_1}{\sigma_1-\sigma_2}\tilde{\SubPSym}_{1,1},\qquad
    \tilde{\SubPSym}_{1,1}^{(1)} := \frac{\sigma-\sigma_2}{\sigma_1-\sigma_2}\tilde{\SubPSym}_{1,1}
  \end{equation*}
  so that
  \begin{equation*}
    \tilde{\SubPSym}_{1,1} = \sum_{i=1,2}\tilde{\SubPSym}_{1,1}^{(i)}. 
  \end{equation*}
  and
  \begin{equation*}
    \norm*{(r-\rTrapping_i)\squareBrace*{\p_r, \Op\left(\SubPSym_1^{(i)}\right)}h}_{L^2(\DomainOfIntegration)} \lesssim \norm*{h}_{\MorawetzNorm_b(\DomainOfIntegration)}.
  \end{equation*}
  Then, we see that
  \begin{equation*}
    - \bangle*{(r-\rTrapping_1)\Op\left(\tilde{\SubPSym}_{1,1}\right)h, h}_{L^2(\DomainOfIntegration)}
    = - \bangle*{(r-\rTrapping_1)\Op\left(\tilde{\SubPSym}_{1,1}^{(1)}\right)h, h}_{L^2(\DomainOfIntegration)}
    - \bangle*{(r-\rTrapping_1)\Op\left(\tilde{\SubPSym}_{1,1}^{(2)}\right)h, h}_{L^2(\DomainOfIntegration)},
  \end{equation*}
  where it is clear by construction that the first term on the
  right-hand side respects the degeneracy in the Morawetz norm and
  thus can be controlled by using Cauchy-Schwarz and Lemma
  \ref{linear:lemma:ILED-near:LoT-control}. to control the second term on the
  right-hand side, we repeat the integration by parts argument above
  to see that
  \begin{align}
    -\bangle*{(r-\rTrapping_1)\Op\left(\tilde{\SubPSym}_{1,1}^{(2)}\right)h,h}_{L^2(\DomainOfIntegration)}
    ={}& \bangle*{(r-\rTrapping_2)(r-\rTrapping_1)\Op\left(\tilde{\SubPSym}_{1,1}^{(2)}\right)\circ \p_r h,h}_{L^2(\DomainOfIntegration)} \notag\\
       &+ \bangle*{(r-\rTrapping_2)(r-\rTrapping_1)\Op\left(\tilde{\SubPSym}_{1,1}^{(2)}\right)h,\p_rh}_{L^2(\DomainOfIntegration)}\notag\\
    &+ \bangle*{(r-\rTrapping_2)\Op\left(\tilde{\SubPSym}_{1,1}^{(2)}\right)h,h}_{L^2(\DomainOfIntegration)}\notag \\
    &+ \bangle*{(r-\rTrapping_2)(r-\rTrapping_1)\squareBrace*{\p_r,\Op\left(\tilde{\SubPSym}_{1,1}^{(2)}\right)}h,h}_{L^2(\DomainOfIntegration)}. \label{linear:eq:ILED-near:LoT-control:exchange-degeneracy-trick:IbP:S1:rep1}
  \end{align}
  Again, using that
  \begin{equation*}
    \left((r-\rTrapping_2)(r-\rTrapping_1)\Op\left(\tilde{\SubPSym}_{1,1}^{(2)}\right)\right)^*
    = (r-\rTrapping_2)(r-\rTrapping_1)\Op\left(\tilde{\SubPSym}_{1,1}^{(2)}\right) + \Op S^0 + \Op S^{-1}\p_{\tStar},
  \end{equation*}
  we have from integrating by parts, Cauchy-Schwarz, Lemma
  \ref{linear:lemma:ILED-near:LoT-control:zero-order-mixed-term}, that the
  first three terms are controlled, taking $\delta_r$ to be
  sufficiently small. To handle the last term on the right-hand side
  of
  \eqref{linear:eq:ILED-near:LoT-control:exchange-degeneracy-trick:IbP:S1:rep1},
  we see that
  \begin{align*}
    &\bangle*{(r-\rTrapping_2)(r-\rTrapping_1)\squareBrace*{\p_r,\Op\left(\tilde{\SubPSym}_{1,1}^{(2)}\right)}h,h}_{L^2(\DomainOfIntegration)}\\
    = {}& -\bangle*{(r-\rTrapping_2)(r-\rTrapping_1)\Op\left(
          (\sigma-\sigma_1)\p_r\left(\frac{\tilde{\SubPSym}_{1,1}}{\sigma_1-\sigma_2}\right)
          \right)h,h}_{L^2(\DomainOfIntegration)}\\
    & + \bangle*{(r-\rTrapping_2)(r-\rTrapping_1)\Op\left(
      \frac{\tilde{\SubPSym}_{1,1}\p_r\sigma_1}{\sigma_1-\sigma_2}
      \right)h,h}_{L^2(\DomainOfIntegration)},
  \end{align*}
  where the first term on the right-hand side is handled directly as
  before using Cauchy-Schwarz, and Lemma
  \ref{linear:lemma:ILED-near:LoT-control:lowest-order-control}, taking
  $\delta_r$ to be sufficiently small. The second term on the
  right-hand side is handled by a final
  decomposition of
  \begin{equation*}    
    \tilde{\SubPSym}_{1,1,2}:= \tilde{\SubPSym}_{1,1}\frac{\p_r\sigma_1}{\sigma_1-\sigma_2}\in S^1,\qquad
    \tilde{\SubPSym}_{1,1,2}^{(2)}:= -\frac{\sigma-\sigma_1}{\sigma_1-\sigma_2}\tilde{\SubPSym}_{1,1,2},\qquad
    \tilde{\SubPSym}_{1,1,2}^{(1)}:= \frac{\sigma-\sigma_2}{\sigma_1-\sigma_2}\tilde{\SubPSym}_{1,1,2}.  
  \end{equation*}
  Then it is apparent that
  \begin{equation*}
    \sum_{i=1,2}\norm*{(r-\rTrapping_1)(r-\rTrapping_2)\Op\left(\tilde{\SubPSym}_{1,1,2}^{(i)}\right)h}_{\MorawetzNorm_b(\DomainOfIntegration)} \lesssim \norm*{h}_{\MorawetzNorm_b(\DomainOfIntegration)},
  \end{equation*}
  and we conclude the proof of Lemma
  \ref{linear:lemma:ILED-near:LoT-control:exchange-degeneracy-trick} by using
  Cauchy-Schwarz, Lemma \ref{linear:lemma:ILED-near:LoT-control}, and taking
  $\delta_r$ and $\delta$ sufficiently small.   
\end{proof}
We can achieve a similar control for the lower-order boundary terms
that come out of the integration by parts argument.
\begin{lemma}
  \label{linear:lemma:ILED-near:LoT-control:exchange-degeneracy-trick-boundary}
  Let $\TrappingNbhd$ be as defined in
  \eqref{linear:eq:ILED-trapping:trapping-reg-def}, and fix some
  $\delta_0>0$. Let
  $\SubPOp\in \Op S^1(\Sigma) + \Op S^{0}(\Sigma)D_{\tStar}$ be a first
  order pseudo-differential operator such that
  \begin{equation*}
    \SubPOp = \SubPOp_1 + \SubPOp_0D_{\tStar},
  \end{equation*}
  where $\SubPOp_1$ is Hermitian and $\SubPOp_0$ is skew-Hermitian
  with respect to the $L^2(\TrappingNbhd)$ norm respectively,
  such that the symbols $\SubPSym_i$ of $\SubPOp_i$ satisfy
  \begin{equation*}
    \abs*{\SubPSym_i}\le \varepsilon_{\SubPOp}|\zeta|^i. 
  \end{equation*}
  Then for $\varepsilon_{\SubPOp}$ and $\delta_r$ sufficiently small,
  \begin{equation}
    \label{linear:eq:ILED-near:LoT-control:exchange-degeneracy-trick-boundary}
    \abs*{\Re\bangle*{\SubPOp h, h}_{L^2(\TrappingNbhd)}}
    < \delta_0 \norm*{h}_{\InducedMorawetzNorm_b(\TrappingNbhd)}^2, 
  \end{equation}
  where $h$ is as specified in Theorem \ref{linear:thm:ILED-near:main}. 
\end{lemma}
\begin{proof}
  The proof follows exactly as the proof of
  Lemma \ref{linear:lemma:ILED-near:LoT-control:exchange-degeneracy-trick-boundary},
  observing that $D_{\tStar}e^{-\ImagUnit\sigma\tStar}u = \sigma u$. 
\end{proof}
Using Lemma
\ref{linear:lemma:ILED-near:LoT-control:exchange-degeneracy-trick-boundary},
we can actually apply the control in Lemma
\ref{linear:lemma:ILED-near:LoT-control:exchange-degeneracy-trick} to
operators in $\SubPOp^{(-1)}\p_t^2 + \SubPOp^{(0)}\p_t +
\SubPOp^{(1)}$ with the aid of some lower-order Lagrangian
correction. 
\begin{lemma}
  \label{linear:lemma:ILED-near:LoT-control:exchange-degeneracy-trick:order-minus-one}
  Fix $\delta_0>0$ and let
  \begin{equation*}
    \SubPOp = \SubPOp_{-1}D_{\tStar}^2 + \SubPOp_{0}D_{\tStar} + \SubPOp_{1},
  \end{equation*}
  where $\SubPOp_{i} \in \Op S^i$ and
  \begin{equation*}
    \abs*{\SubPSym_{i}} \le \varepsilon_{\SubPOp}\left(\abs*{\xi}^i + \abs*{\FreqAngular}^i\right),\qquad
    i=-1,0,1,
  \end{equation*}
  where $\SubPSym_i$ is the symbol of $\SubPOp_i$.  Then, there exists
  some Hermitian $\tilde{r}\in \Op S^{-1}(\Sigma)$ with principal
  symbol $\tilde{\mathfrak{r}} \in S^{-1}$ and sufficiently small $a$,
  and $\delta_r$, so that for $\varepsilon_{\SubPOp}$ sufficiently
  small,
  \begin{equation*}
    \abs*{
      \Re \bangle*{\SubPOp h, h}_{L^2(\DomainOfIntegration)}
      - 2\Re\bangle*{\ScalarWaveOp[g_b]h, \tilde{r} h}_{L^2(\DomainOfIntegration)}}
    \le \delta_0 \norm*{h}_{\MorawetzNorm(\DomainOfIntegration)}^2
    + \evalAt*{\int_{\TrappingNbhd_{\tStar}}\tilde{J}[h]}^{\tStar=\TStar}_{\tStar=0},
  \end{equation*}
  where
  \begin{equation*}
    \abs*{\int_{\TrappingNbhd_{\tStar}}\tilde{J}[h]} \lesssim \delta_0 \norm*{h}_{\InducedMorawetzNorm(\Sigma_{\tStar})}^2.
  \end{equation*}  
\end{lemma}
\begin{proof}
  In view of Lemma
  \ref{linear:lemma:ILED-near:LoT-control:exchange-degeneracy-trick}, it
  suffices to consider the case where $\SubPOp =
  \SubPOp^{-1}\p_t^2$. Denoting $\SubPSym_{-1}$ as the symbol of
  $\SubPOp_{-1}$,
  \begin{equation*}
    \SubPSym_{-1}\sigma^2
    = \SubPSym_{-1}(\sigma-\sigma_1)(\sigma-\sigma_2)
    - \SubPSym_{-1}\sigma_1\sigma_2
    + \SubPSym_{-1}\sigma(\sigma_1+\sigma_2).
  \end{equation*}
  Letting $\tilde{r} = -G_b(d\tStar,d\tStar)\SubPSym_{-1}$, we have
  that
  \begin{equation*}
    2\Re\bangle*{\ScalarWaveOp[g_b]h, \tilde{r}h}_{L^2(\DomainOfIntegration)}
    = \bangle*{\left(\ScalarWaveOp[g_b]\tilde{r} + \tilde{r}\ScalarWaveOp[g_b]\right)h, h}_{L^2(\DomainOfIntegration)}
    + \evalAt*{2\Re\bangle*{n_{\Sigma}h, \tilde{r} h}_{L^2(\TrappingNbhd_{\tStar})}}^{\tStar=\TStar}_{\tStar=0},
  \end{equation*}
  so that up to lower order terms, 
  \begin{align*}
    &\Re \bangle*{\SubPOp h, h}_{L^2(\DomainOfIntegration)}
    - 2\Re\bangle*{\ScalarWaveOp[g_b]h, \tilde{r} h}_{L^2(\DomainOfIntegration)}\\
    ={}& \bangle*{\Op\left(\SubPSym_{-1}\sigma(\sigma_1+\sigma_2) - \SubPSym_{-1}\sigma_1\sigma_2\right)h,h}_{L^2(\DomainOfIntegration)}
         + \evalAt*{2\Re\bangle*{n_{\Sigma}h, \tilde{r} h}_{L^2(\TrappingNbhd_{\tStar})}}^{\tStar=\TStar}_{\tStar=0}.
  \end{align*}
  Using Lemma
  \ref{linear:lemma:ILED-near:LoT-control:exchange-degeneracy-trick-boundary},
  we see that for $\varepsilon_{\SubPOp}$ sufficiently small and
  $\delta_r$ sufficiently small,
  \begin{equation*}
    \Re\bangle*{n_{\Sigma}h, \tilde{r} h}_{L^2(\TrappingNbhd_{\tStar})}< \frac{1}{4}\delta_0\norm*{h}_{\InducedMorawetzNorm(\Sigma_{\tStar})}^2. 
  \end{equation*}
  Moreover, directly using Lemma
  \ref{linear:lemma:ILED-near:LoT-control:exchange-degeneracy-trick}, we have
  that for $a$, $\delta_r$, and $\varepsilon_{\SubPOp}$ sufficiently
  small,
  \begin{equation*}
    \abs*{\bangle*{\Op\left(\SubPSym_{-1}\sigma(\sigma_1+\sigma_2) - \SubPSym_{-1}\sigma_1\sigma_2\right)h,h}_{L^2(\DomainOfIntegration)}}
    \lesssim \delta_0\norm*{h}_{\MorawetzNorm(\DomainOfIntegration)}^2 + \evalAt*{\int_{\TrappingNbhd_{\tStar}}\tilde{J}_1[h]}_{\tStar=0}^{\tStar=\TStar},
  \end{equation*}
  where
  \begin{equation*}
    \int_{\TrappingNbhd_{\tStar}}\tilde{J}_1[h]
    = \Re\bangle*{\SubPOp_{-1}h, n_{\Sigma} h}_{L^2(\TrappingNbhd_{\tStar})}.
  \end{equation*}
  Then, using the fact that $\SubPOp^{-1}n_{\Sigma} \in \Op
  S^{-1}\p_{\tStar}$, we can use Lemma
  \ref{linear:lemma:ILED-near:LoT-control:exchange-degeneracy-trick-boundary}
  to see that for $\varepsilon_{\SubPOp}$, and $\delta_r$ sufficiently
  small, 
  \begin{equation*}
    \abs*{\int_{\TrappingNbhd_{\tStar}}\tilde{J}_1[h] }<\frac{1}{2}\delta_0\norm*{h}_{\InducedMorawetzNorm(\TrappingNbhd_{\tStar})}^2, 
  \end{equation*}
  as desired. 
\end{proof}
The following corollary follows immediately from Lemma
\ref{linear:lemma:ILED-near:LoT-control:exchange-degeneracy-trick:order-minus-one}
and will be more useful in the context of proving a high-frequency
Morawetz estimate for $\LinEinstein_{g_b}$. 
\begin{corollary}
  \label{linear:coro:ILED-near:LoT-control:exchange-degeneracy-trick:order-minus-one:with-L}
  Fix $\delta_0>0$ and let
  \begin{equation*}
    \SubPOp = \SubPOp_{-1}D_{\tStar}^2 + \SubPOp_0D_{\tStar} + \SubPOp_1,
  \end{equation*}
  where $\SubPOp_i \in \varepsilon_{\SubPOp}\Op S^i$, where
  $\SubPSym_{i}$ is the symbol of $\SubPOp_{i}$, moreover,
  $\SubPOp_{-1}$, $\SubPOp_0$, and $\SubPOp_{1}$ are Hermitian with
  respect to the $L^2(\TrappingNbhd)$ Hermitian inner product. Then,
  there exists some Hermitian $\tilde{r} \in S^{-1}$ and sufficiently
  small $a$, and $\delta_r$, so that for $\varepsilon_{\SubPOp}$
  sufficiently small,
  \begin{equation*}
    \abs*{
      \Re \bangle*{\SubPOp h, h}_{L^2(\DomainOfIntegration)}
      - 2\Re\bangle*{\LinEinsteinConj_{g_b}h, \tilde{r} h}_{L^2(\DomainOfIntegration)}}
    < \delta_0 \norm*{h}_{\MorawetzNorm(\DomainOfIntegration)}^2
    + \evalAt*{\int_{\TrappingNbhd_{\tStar}}\tilde{J}[h]}^{\tStar=\TStar}_{\tStar=0},
  \end{equation*}
  where
  \begin{equation*}
    \abs*{\int_{\TrappingNbhd_{\tStar}}\tilde{J}[h]}
    \lesssim \norm*{h}_{H^1(\Sigma_{\tStar})}^2 + \norm*{\KillT h}_{\LTwo(\Sigma_{\tStar})}^2 ,
  \end{equation*}
  and if there exists some $\sigma\in \Complex$ such that
  $h=e^{-\ImagUnit\sigma\tStar}u$, then
  \begin{equation*}
    \abs*{\int_{\TrappingNbhd_{\tStar}}\tilde{J}[h]} \lesssim \delta_0 \norm*{h}_{\InducedMorawetzNorm(\Sigma_{\tStar})}^2.
  \end{equation*}   
\end{corollary}
\begin{proof}
  With Lemma
  \ref{linear:lemma:ILED-near:LoT-control:exchange-degeneracy-trick-boundary}
  already in hand, it suffices to show that for $\SubPOp\in
  \SubPOp^{-1}D_{\tStar}^2$ and for the same $\tilde{r}$
  constructed in Lemma
  \ref{linear:lemma:ILED-near:LoT-control:exchange-degeneracy-trick-boundary}, 
  \begin{equation*}
    2\abs*{\Re \bangle*{\SubPConjOp_{b}h,\tilde{r} h}_{L^2(\DomainOfIntegration)}}
    \lesssim \varepsilon_{\SubPOp}\norm*{h}_{\MorawetzNorm_b(\DomainOfIntegration)}^2. 
  \end{equation*}
  To this end, we write
  \begin{equation*}
    \SubPConjOp_{b} = \SubPConjOp_{1} + \SubPConjOp_{0}D_{\tStar}
  \end{equation*}
  where $\SubPConjOp_{i} \in \Op S^i(\Sigma)$. Then, observe that since
  $\tilde{r}$ is Hermitian, we have that
  \begin{equation*}
    2\Re\bangle*{\SubPConjOp_{b}h,\tilde{r} h}_{L^2(\DomainOfIntegration)}
    = \bangle*{\left( \tilde{r}\SubPConjOp_{b} + \SubPConjOp_{b}\tilde{r} \right)h, h}_{L^2(\DomainOfIntegration)}. 
  \end{equation*}
  But since
  $\tilde{\mathfrak{r}}\SubPConjSym_b \in S^0 + \sigma S^{-1}$, we can
  directly apply Lemma
  \ref{linear:lemma:ILED-near:LoT-control:exchange-degeneracy-trick} to
  conclude.
\end{proof}


We are now ready to control the lower-order terms that arise in the
divergence theorem argument. 
\begin{lemma}
  \label{linear:lemma:ILED-near:LoT-control}
  Fix $\delta_0>0$, and let $\MorawetzVF_b$, $\LagrangeCorr_b$ be
  those constructed in Lemma \ref{linear:lemma:ILED-near:sum-of-squares}.
  Then there exists a choice of $a$, $\delta_r$, $\delta_\zeta$, and
  $\varepsilon_{\TrappedSet_{b_0}}$ such that for $h$ as specified in
  Theorem \ref{linear:thm:ILED-near:main}, there exists a choice of
  $\PseudoSubPFixer$ such that for $C(\delta)\in \Real$, we have the
  following inequality
  \begin{align*}
    &\delta_0 \norm{h }_{\MorawetzNorm_b(\DomainOfIntegration)}^2
      + C(\delta_0)\norm*{\LinEinsteinConj_{g_b}h}_{L^2(\DomainOfIntegration)}^2  \\    
    >{}& \frac{1}{2}\Re\bangle*{\squareBrace*{\SubPConjOp_{b, s}, \mathfrak{\MorawetzVF}_{b_0}}h , h }_{L^2(\DomainOfIntegration)}
         + \Re\bangle*{\SubPConjOp_b[h ], \tilde{\mathfrak{\LagrangeCorr}}_{b_0} h }_{L^2(\DomainOfIntegration)}
         + \Re\bangle*{\PotentialConjOp_bh , (\MorawetzVF_{b_0}+\LagrangeCorr_{b_0})h }_{L^2(\DomainOfIntegration)}\\
    &+ a\Re\KCurrentIbP{\widetilde{\MorawetzVF}, \tilde{\LagrangeCorr}}_{(1)}[h ]
      + a\Re\KCurrentIbP{\widetilde{\MorawetzVF}, \tilde{\LagrangeCorr}}_{(0)}[h ] 
          + \left.\int_{\TrappingNbhd_{\tStar}}\tilde{J}[{h}]\right\vert_{\tStar= 0 }^{\tStar = \TStar},      
  \end{align*}
  where $\mathfrak{\MorawetzVF}_{b_0}$ and
  $\tilde{\mathfrak{\LagrangeCorr}}_{b_0}$ are as defined in
  \eqref{linear:eq:ILED-near:bulk:positivity:adjusted-X-q-def}, and
  \begin{equation*}
    \abs*{\int_{\TrappingNbhd_{\tStar}}\tilde{J}[h]}
    \lesssim \norm*{h}_{H^1(\Sigma_{\tStar})}^2 + \norm*{\KillT h}_{\LTwo(\Sigma_{\tStar})}^2 .
  \end{equation*}
  Moreover, if there exists some $\sigma\in \Complex$ such that
  $h=e^{-\ImagUnit\sigma\tStar}u$, then
  \begin{equation*}
    \abs*{\int_{\TrappingNbhd_{\tStar}}\tilde{J}[h]} \lesssim \delta_0 \norm*{h}_{\InducedMorawetzNorm(\Sigma_{\tStar})}^2.
  \end{equation*}    
\end{lemma}

\begin{remark}
  The specific choice of the smallness constants $a$, $\delta_r$,
  $\delta_\zeta$, and $\varepsilon_{\TrappedSet}$ does depend here on
  the choice of $\delta_0>0$. In particular, these smallness constants
  could degenerate as $\delta_0\to 0$. However, in application, we will
  only need to take $\delta_0$ sufficiently small so that
  \begin{equation*}
    \delta_0\norm*{h}_{\MorawetzNorm(\DomainOfIntegration)}^2
    \le \frac{1}{2} \left( \int_{\DomainOfIntegration} \KCurrent{\MorawetzVF_{b_0}, \LagrangeCorr_{b_0}, 0}[h]
      + a\KCurrentIbP{\widetilde{\MorawetzVF}, \tilde{\LagrangeCorr}}[h]\right),
  \end{equation*}
  and do not need to take $\delta_0$ arbitrarily small. 
\end{remark}

\begin{proof}
  Let us first split the lower-order terms we wish to control into the
  first-order terms
  \begin{equation}
    \label{linear:eq:ILED-near:LoT-control:first-order-terms}
    \frac{1}{2}\Re\bangle*{\squareBrace*{\SubPConjOp_{b,s}, \mathfrak{\MorawetzVF}_{b_0}}{h}, {h}}_{L^2(\DomainOfIntegration)}
    + \Re\bangle*{\SubPConjOp_b[{h}], \tilde{\mathfrak{\LagrangeCorr}}_{b_0}{h}}_{L^2(\DomainOfIntegration)}
    + \Re \bangle*{\PotentialConjOp_b{h}, \MorawetzVF_{b_0}{h}}_{L^2(\DomainOfIntegration)}
    + a\Re\KCurrentIbP{\widetilde{\MorawetzVF}, \tilde{\LagrangeCorr}}_{(1)}[{h}];
  \end{equation}
  and the zero-order terms
  \begin{equation}
    \label{linear:eq:ILED-near:LoT-control:zero-order-terms}
    \Re \bangle*{\PotentialConjOp_b{h}, \LagrangeCorr_{b_0}{h}}_{L^2(\DomainOfIntegration)}
    + a\Re\KCurrentIbP{\widetilde{\MorawetzVF}, \tilde{\LagrangeCorr}}_{(0)}[{h}].
  \end{equation}
  Let us briefly recall that in the non-trapping regimes dealt with
  previously, the lower-order terms are dealt with
  by a high-frequency argument, using the principal bulk.
  
  In the trapped regime however, this is not possible because of the
  degeneration of the principal ellipticity at the trapped set. As an
  illustration, consider the term
  $\Re\bangle*{\SubPConjOp_b{h},
    \tilde{\mathfrak{\LagrangeCorr}}_{b_0}{h}}_{L^2(\DomainOfIntegration)}$.
  If we were to naively apply Cauchy-Schwarz, we would have
  \begin{equation*}
    \abs*{\bangle*{\SubPConjOp_b{h}, \tilde{\mathfrak{\LagrangeCorr}}_{b_0}{h}}_{L^2(\DomainOfIntegration)}} \lesssim \varepsilon_{\TrappedSet_{b_0}}\left(
      \norm{{h}}_{H^1(\DomainOfIntegration)}^2 + \norm{\KillT {h}}_{L^2(\DomainOfIntegration)}^2
    \right).
  \end{equation*}
  The apparent problem is that we have no way of controlling
  $\abs*{\bangle*{\SubPConjOp_b{h},
      \tilde{\mathfrak{\LagrangeCorr}}_{b_0}{h}}_{L^2(\DomainOfIntegration)}}$ by
  $\norm{{h}}_{\MorawetzNorm_b(\DomainOfIntegration)}^2$ directly
  using Cauchy-Schwarz because neither $\SubPConjOp_b$ nor
  $\tilde{\mathfrak{\LagrangeCorr}}_{b_0}$ degenerate at the trapped set in phase space,
  while the norm $\norm{{h}}_{\MorawetzNorm_b(\DomainOfIntegration)}$
  was specifically engineered to vanish at the trapped set.

  The way around this difficulty is to invoke Lemma
  \ref{linear:lemma:ILED-near:LoT-control:exchange-degeneracy-trick},
  which trades the subprincipal nature
  of the first-order terms for a degeneracy at trapping.


  \noindent{\textbf{Step 1: The first-order terms not O(a)}}.
  We first deal with the first order terms that are not $O(a)$,
  \begin{equation*}    
    \frac{1}{2}\Re\bangle*{\squareBrace*{\SubPConjOp_{b,s}, \mathfrak{\MorawetzVF}_{b_0}}{h}, {h}}_{L^2(\DomainOfIntegration)}
      +\Re\bangle*{\SubPConjOp_b[{h}], \tilde{\mathfrak{\LagrangeCorr}}_{b_0}{h}}_{L^2(\DomainOfIntegration)}
      + \Re\bangle*{\PotentialConjOp_b{h}, \MorawetzVF_{b_0}{h}}_{L^2(\DomainOfIntegration)}.
  \end{equation*}
  We show explicitly how to use the idea of exchanging degeneracy and
  derivatives discussed above with the terms in
  \eqref{linear:eq:ILED-near:LoT-control:first-order-terms}.  Fix some
  $\delta>0$. We begin by
  controlling
  $\Re\bangle*{\SubPConjOp_{b}h,
    \tilde{\mathfrak{\LagrangeCorr}}_{b_0}h}_{L^2(\DomainOfIntegration)}$ by using the
  smallness we gain from the conjugation by $\PseudoSubPFixer$.

  Observe that we can split
  \begin{equation*}
    \Re\bangle*{\SubPConjOp_{b}h,    \tilde{\mathfrak{\LagrangeCorr}}_{b_0}h}_{L^2(\DomainOfIntegration)}
    = \Re\bangle*{\SubPConjOp_{b,s}h,    \tilde{\mathfrak{\LagrangeCorr}}_{b_0}h}_{L^2(\DomainOfIntegration)}
    + \Re\bangle*{\SubPConjOp_{b,a}h,    \tilde{\mathfrak{\LagrangeCorr}}_{b_0}h}_{L^2(\DomainOfIntegration)},
  \end{equation*}
  where $\SubPConjOp_{b,s}$ and $\SubPConjOp_{b,a}$ are as defined in
  \eqref{linear:eq:SubPConjOp:sym-skew-sym-split-def}. Using Lemma
  \ref{linear:lemma:ILED-near:LoT-control:symmetry-control}, the latter term
  reduces to
  \begin{equation*}
    \Re\bangle*{\SubPConjOp_{b,a}h,    \tilde{\mathfrak{\LagrangeCorr}}_{b_0}h}_{L^2(\DomainOfIntegration)}
    = \evalAt*{\Re\bangle*{\tilde{\mathfrak{\LagrangeCorr}}_{b_0}\SubPConjOp_{b,a}^{(0)}h, h}_{\InducedLTwo(\Sigma_{\tStar})}}_{\tStar=0}^{\tStar=\TStar},
  \end{equation*}
  where
  \begin{equation*}
    \SubPConjOp_{b,a} = \SubPConjOp_{b,a}^{(0)}D_{\tStar} + \SubPConjOp_{b,a}^{(1)},\qquad
    \SubPConjOp_{b,a}^{(i)} \in \Op S^i. 
  \end{equation*}
  Moreover, from Lemma \ref{linear:lemma:ILED-near:s-decomp}, we have that
  for $a$, $\delta_r$,
  $\delta_\zeta$, and $\varepsilon_{\TrappedSet_{b_0}}$ sufficiently
  small we have that $\norm*{\SubPConjOp_{b,a}^{(0)}}_{L^2\to
    L^2}\lesssim \delta$. 

  To control
  $\Re\bangle*{\SubPConjOp_{b,s}[h],
    \tilde{\mathfrak{\LagrangeCorr}}_{b_0}h}_{L^2(\DomainOfIntegration)}$,
  we observe from its construction in Lemma
  \ref{linear:lemma:ILED-near:bulk-positivity}  that for $a$,
  $\delta_r$ sufficiently small, 
  \begin{equation*}
    \sup_{\TrappingNbhd}\abs*{\tilde{\mathfrak{\LagrangeCorr}}_{b_0}} < \delta.
  \end{equation*}  
  Then using Lemma
  \ref{linear:lemma:ILED-near:LoT-control:exchange-degeneracy-trick}, we have
  that for $a$, $\delta_r$, $\delta_\zeta$, and
  $\varepsilon_{\TrappedSet}$ sufficiently small, we have that
  \begin{equation*}
    \abs*{\bangle*{\tilde{\mathfrak{\LagrangeCorr}}_{b_0}\SubPConjOp_{b,s}[h], h}_{L^2(\DomainOfIntegration)}}
    \le \delta \norm*{h}_{\MorawetzNorm(\DomainOfIntegration)}^2 + \evalAt*{Err_{\TrappingNbhd}[h]}_{\tStar=0}^{\tStar=\TStar}, 
  \end{equation*}
  where $\abs*{Err_{\TrappedSet}[h]}\lesssim \delta
  \norm*{h}_{\InducedMorawetzNorm(\TrappingNbhd)}^2$ for $\delta_r$,
  $\varepsilon_{\TrappedSet}$, and $a$ sufficiently small. 
  We now consider the term
  \begin{equation*}
    \Re\bangle*{\squareBrace*{\SubPConjOp_{b,s}, \mathfrak{\MorawetzVF}_{b_0}}h, h}_{L^2(\DomainOfIntegration)}.
  \end{equation*}
  Recall from its definition in
  \eqref{linear:eq:ILED-near:bulk:positivity:adjusted-X-q-def} that
  \begin{equation*}
    \mathfrak{\MorawetzVF}_{b_0} = \MorawetzVF_{b_0} + \frac{1}{2}\nabla_{g_b}\cdot \MorawetzVF_{b_0}.
  \end{equation*}
  Since $\nabla_{g_b}\cdot \MorawetzVF_{b_0}$ is just a smooth
  function in $r$ on $\TrappingNbhd$, we have that
  \begin{equation*}
    \squareBrace*{\SubPConjOp_{b,s}, \nabla_{g_b}\cdot \MorawetzVF_{b_0}} \in \Op S^0.
  \end{equation*}
  As a result, using Lemma
  \ref{linear:lemma:ILED-near:LoT-control:lowest-order-control}, we have that
  for sufficiently small $\delta_r$,
  \begin{equation*}
    \abs*{\bangle*{\squareBrace*{\SubPConjOp_{b,s}, \nabla_{g_b}\cdot \MorawetzVF_{b_0}}h, h}_{L^2(\DomainOfIntegration)}}< \delta \norm*{h}_{\MorawetzNorm_b(\DomainOfIntegration)}^2. 
  \end{equation*}
  Furthermore, using the explicit form of $\MorawetzVF_{b_0}$ in
  \eqref{linear:eq:ILED-near:SdS:X-def}, we can write that
  \begin{align*}
    \squareBrace*{\SubPConjOp_{b,s}, \MorawetzVF_{b_0}}
    &= \squareBrace*{\SubPConjOp_{b,s}, f_{b_0}(r)}\p_r
      + f_{b_0}(r)\squareBrace{\SubPConjOp_{b,s}, \p_r}.
  \end{align*}  
  Since $f_{b_0}(r)$ is a smooth function we can use Lemma
  \ref{linear:lemma:ILED-near:LoT-control:zero-order-mixed-term} to see that
  for $\delta_r$ sufficiently small,
  \begin{equation}
   \label{linear:eq:ILED-near:LoT:first-order:aux1}
    \abs*{\bangle*{\squareBrace*{\SubPConjOp_{b,s}, f_{b_0}(r)}\p_rh, h}_{L^2(\DomainOfIntegration)}}
    < \delta\norm*{h}_{\MorawetzNorm(\DomainOfIntegration)}^2.    
  \end{equation}
  The last term,  $\bangle*{\PotentialConjOp_b{h},
    \MorawetzVF_{b_0}{h}}_{L^2(\DomainOfIntegration)}$, can be
  handled directly using Cauchy-Schwarz from Lemma
  \ref{linear:lemma:ILED-near:LoT-control:lowest-order-control}.

  \noindent{\textbf{Step 2: First-order terms that are O(a)}.}
  We handle the remaining first order term,
  $a\KCurrentIbP{\widetilde{\MorawetzVF},
    \tilde{\LagrangeCorr}}_{(1)}[{h}]$ in a similar fashion. The main
  difference to the previous case is that instead of relying on the
  smallness we gain by considering the symbol decomposition of
  $\SubPConjOp_b$, we rely on the inherent smallness present in $a$.  
  
  First, we recall that
  \begin{equation*}
    \KCurrentIbP{\widetilde{\MorawetzVF}, \tilde{\LagrangeCorr}}_{(1)}[{h}]
    = \frac{1}{2}\bangle*{\left[\SubPConjOp_{b,s}, \widetilde{\MorawetzVF}\right]{h},{h}}_{L^2(\DomainOfIntegration)}
    + \bangle*{\SubPConjOp_b {h}, \tilde{\LagrangeCorr} {h}}_{L^2(\DomainOfIntegration)}
    + \bangle*{\PotentialConjOp_b {h}, \widetilde{\MorawetzVF} {h}}_{L^2(\DomainOfIntegration)}.
  \end{equation*}
  Then observe that directly using Lemma
  \ref{linear:lemma:ILED-near:LoT-control:symmetry-control} and Lemma
  \ref{linear:lemma:ILED-near:LoT-control:lowest-order-control} we have that
  for $\delta_r$ sufficiently small, 
  \begin{equation*}
    \abs*{\bangle*{\SubPConjOp_{b,a}h, \tilde{\LagrangeCorr}h }_{L^2(\DomainOfIntegration)}}
    + \abs*{\bangle*{\PotentialConjOp_{b,a}h, \tilde{\LagrangeCorr}h }_{L^2(\DomainOfIntegration)}}
    < \delta \norm*{h}_{\InducedMorawetzNorm_b(\TrappingNbhd}^2.  
  \end{equation*}
  We now move on to handling
  \begin{equation*}
    \KCurrentIbP{\widetilde{\MorawetzVF}, \tilde{\LagrangeCorr}}_{(1,s)}[{h}]
    := \frac{1}{2}\bangle*{\left[\SubPConjOp_{b,s}, \widetilde{\MorawetzVF}\right]{h},{h}}_{L^2(\DomainOfIntegration)}
    + \bangle*{\SubPConjOp_{b,s} {h}, \tilde{\LagrangeCorr} {h}}_{L^2(\DomainOfIntegration)}
    + \bangle*{\PotentialConjOp_{b,a} {h}, \widetilde{\MorawetzVF} {h}}_{L^2(\DomainOfIntegration)}.
  \end{equation*}
  Observe that
  \begin{equation*}
    \KCurrentIbP{\widetilde{\MorawetzVF}, \tilde{\LagrangeCorr}}_{(1,s)}[{h}]
    = \bangle*{\KCurrentSym{\widetilde{\MorawetzVF}, \tilde{\LagrangeCorr}}_{(1)}{h}, {h}}_{L^2(\DomainOfIntegration)}
    +\left.\bangle*{\PotentialConjOp_{b,a}{h}, \widetilde{\MorawetzVF}_0{h}}_{\LTwo(\TrappingNbhd_{\tStar})}\right\vert_{\tStar=0}^{\tStar =\TStar}
    +\left.\bangle*{\SubPConjOp_{b,s}{h}, \tilde{\LagrangeCorr}_{-1}{h}}_{\LTwo(\TrappingNbhd_{\tStar})}\right\vert_{\tStar=0}^{\tStar =\TStar},
  \end{equation*}
  where
  \begin{equation*}
    \KCurrentSym{\widetilde{\MorawetzVF}, \tilde{\LagrangeCorr}}_{(1)}
    = \left[\SubPConjOp_{b,s}, \widetilde{\MorawetzVF}\right]
    + \tilde{q}\SubPConjOp_b
    -\widetilde{\MorawetzVF}\PotentialConjOp_b
    \in \Op S^{-1}D_{\tStar}^2 + \Op S^0D_{\tStar} + \Op S^1.
  \end{equation*}
  It is clear that the accumulated boundary terms satisfy
  \begin{equation*}
    \abs*{\bangle*{\PotentialConjOp_b{h}, \widetilde{\MorawetzVF}_0{h}}_{\LTwo(\TrappingNbhd_{\tStar})}}
    +\abs*{\bangle*{\SubPConjOp_{b}{h}, \tilde{\LagrangeCorr}_{-1}{h}}_{\LTwo(\TrappingNbhd_{\tStar})}}
    \lesssim \norm{h}_{H^1(\TrappingNbhd_{\tStar})}^2 + \norm{\KillT h}_{\LTwo(\TrappingNbhd_{\tStar})}^2.
  \end{equation*}
  Moreover, using Lemma
  \ref{linear:lemma:ILED-near:LoT-control:exchange-degeneracy-trick-boundary},
  for $a$ sufficiently small, these auxiliary boundary
  terms can be controlled by
  $\delta\norm{{h}}_{\InducedMorawetzNorm(\TrappingNbhd_{\tStar})}^2$.  We can now
  use Lemma \ref{linear:lemma:ILED-near:LoT-control:lowest-order-control} to
  directly see that for $a$ and $\delta_r$ sufficiently small, we have
  that
  \begin{equation*}
    \abs*{\bangle*{a\KCurrentSym{\widetilde{\MorawetzVF}, \tilde{\LagrangeCorr}}_{(1)}{h}, {h}}_{L^2(\DomainOfIntegration)}}
    < \delta \norm*{h}_{\MorawetzNorm(\DomainOfIntegration)}^2. 
  \end{equation*}    

  \noindent{\textbf{Step 3: Zero-order terms.}}
  The zero-order terms can be directly controlled by Lemma \ref{linear:lemma:ILED-near:LoT-control:lowest-order-control}.    

  We conclude the proof of Lemma \ref{linear:lemma:ILED-near:LoT-control} by
  picking a suitable $\tilde{J}[{h}]$ to cancel out all the auxiliary
  boundary terms we picked up in the process of integrating by parts.
\end{proof}

We have now generated a principally positive bulk term up to some
error consisting entirely of boundary integrals.

\paragraph{Controlling the boundary terms}

The last step in proving a resolvent estimate for
$\abs{\Im\sigma}\le \SpectralGap$ will be to show that we can absorb
the boundary terms into the positive bulk term as in the proof
Proposition \ref{linear:prop:ILED-near:SdS}.

\begin{lemma}
  \label{linear:lemma:ILED-near:boundary-terms}
  Fix some $0<\delta_0\ll1$. Then, there exists a choice of $a$,
  $\delta_r$, $\delta_\zeta$, $\varepsilon_{\TrappedSet_{b_0}}$, and
  $\SpectralGap$ sufficiently small such that for
  ${h}=e^{-\ImagUnit\sigma t}u$ where
  $\abs*{\Im\sigma} < \SpectralGap$,
  \begin{equation*}
    e^{2\Im\sigma t}\abs*{\Im\sigma}\abs*{
      \int_{\TrappingNbhd}\JCurrent{\MorawetzVF_{b_0}, \LagrangeCorr_{b_0},0}[{h}]\cdot n_{\TrappingNbhd} 
      + a\Re\JCurrentIbP{\widetilde{\MorawetzVF}, \tilde{\LagrangeCorr}}[{h}]}
    <
    \delta_0 \norm{u}_{\InducedMorawetzNorm_b(\TrappingNbhd)}^2.
  \end{equation*}
\end{lemma}

\begin{proof}
  The main idea in this proof will be to recycle the methods we have
  already used to handle the bulk terms to analyze the boundary
  terms. Let us first fix some auxiliary $\delta>0$. 

  \noindent{\textbf{Step 1: Principal boundary term.}}
  We begin by applying the methods used to handle the principal bulk
  term to handle the principal boundary term, showing that for
  ${h} = e^{-\ImagUnit\sigma\tStar}u$,
  \begin{equation}
    \label{linear:eq:ILED-near:boundary:principal-sym-result}
    e^{2\Im\sigma\tStar}\abs*{\Re\bangle*{n_{\TrappingNbhd}{h}, \MorawetzVF_{b_0}{h}}_{\LTwo(\TrappingNbhd)}
      + a\Re\JCurrentIbP{\widetilde{\MorawetzVF},
        \tilde{\LagrangeCorr}}_{(2)}[h]}
    < \delta \norm{u}_{\InducedMorawetzNorm_b(\TrappingNbhd)}^2.
  \end{equation}
  Observe that on $\TrappingNbhd$,
  $n_{\TrappingNbhd}= \frac{1}{\sqrt{\GInvdtdt}}d\tStar$, so for ${h} = e^{-\ImagUnit\sigma\tStar}u$,
  \begin{equation*}
    e^{2\Im\sigma\tStar}\left(\Re\bangle*{n_{\TrappingNbhd}{h}, \MorawetzVF_{b_0}{h}}_{\LTwo(\TrappingNbhd)}
    + a\Re\bangle*{n_{\TrappingNbhd}{h}, \widetilde{\MorawetzVF}{h}}_{\LTwo(\TrappingNbhd)} \right)
    = -\Re\bangle*{ \sigma u, \MorawetzVF_b u}_{\LTwo(\TrappingNbhd)}.
  \end{equation*}
  As a result, choosing $\abs*{\Im\sigma}$ sufficiently small,   
  \begin{equation*}
    \abs*{\Im\sigma}\abs*{\bangle*{{h}, n_{\TrappingNbhd}\MorawetzVF_b{h}}_{\LTwo(\TrappingNbhd)}}
    < \delta\norm{{u}}_{\InducedMorawetzNorm_b(\TrappingNbhd)}^2.
  \end{equation*}
  We remark that since we use the smallness $\Im\sigma$ to conclude,
  this step necessarily restricts the size of the spectral gap. 
    
  \noindent{\textbf{Step 2: Lower-order boundary terms.}}
  We now handle the lower-order boundary terms, showing that for
  ${h}=e^{-\ImagUnit\sigma\tStar}u$,
  \begin{align}
    \delta \norm{h}_{\InducedMorawetzNorm_b(\TrappingNbhd)}^2>
    & \abs*{\Im\sigma}\left(
      \abs*{\bangle*{\LagrangeCorr_{b_0}{h}, n_{\TrappingNbhd}{h}}_{\LTwo(\TrappingNbhd)}}
      + \abs*{\bangle*{g_b(\KillT, n_{\TrappingNbhd})\SubPConjOp_{0,a}h, \MorawetzVF_{b_0}h}_{L^2(\TrappingNbhd)}}
      + a\abs*{\JCurrentIbP{\widetilde{\MorawetzVF}, \tilde{\LagrangeCorr}}_{(1)}[{h}]}      
    \right).\label{linear:eq:ILED-near:boundary:LoT-estimate}
  \end{align}
  To handle these lower-order boundary terms, we appeal to Lemma
  \ref{linear:lemma:ILED-near:LoT-control:exchange-degeneracy-trick-boundary},
  mirroring the approach taken in proving Lemma
  \ref{linear:lemma:ILED-near:LoT-control}.  Using Lemma
  \ref{linear:lemma:ILED-near:LoT-control:exchange-degeneracy-trick-boundary},
  we see that for sufficiently small\footnote{Observe that due to the
    fact that we need to use $\Im\sigma$ as a smallness parameter in
    controlling
    $\abs*{\bangle*{\LagrangeCorr_{b_0}{h},
        n_{\TrappingNbhd}{h}}_{\LTwo(\TrappingNbhd)}}$, this also
    limits the size of the spectral gap.} $\delta_r$, $\Im\sigma$,
  \begin{equation}
    \label{linear:eq:ILED-near:boundary-terms:subprinc-not-a}
    \Im\sigma \left(\abs*{\bangle*{\LagrangeCorr_{b_0}{u}, \sigma {u}}_{L^2(\TrappingNbhd)}}
      + \abs*{\bangle*{g_b(\KillT, n_{\TrappingNbhd})\SubPConjOp_{0,a}u, \MorawetzVF_{b_0}u}_{L^2(\TrappingNbhd)}}
    \right)
    < \delta  \norm{u}_{\InducedMorawetzNorm_b(\TrappingNbhd)}^2.
  \end{equation}
  To handle $a\abs*{\JCurrentIbP{\widetilde{\MorawetzVF},
      \tilde{\LagrangeCorr}}_{(1)}[{h}]}$, we can again use Lemma
  \ref{linear:lemma:ILED-near:LoT-control:exchange-degeneracy-trick-boundary}
  to see that for sufficiently small\footnote{We emphasize that here
    because of the additional smallness factor in $a$, we do not
    need to rely on the smallness in $\Im\sigma$.} $a$, $\delta_r$,
  \begin{equation}
    \label{linear:eq:ILED-near:boundary-terms:subprinc-a}
    \abs*{a e^{2\Im\sigma\tStar}\Re \JCurrentIbP{\widetilde{\MorawetzVF}, \tilde{\LagrangeCorr}}_{(1)}[h]}
    <\delta \norm{u}_{\InducedMorawetzNorm_b(\TrappingNbhd)}^2.
  \end{equation}
  Combining \eqref{linear:eq:ILED-near:boundary-terms:subprinc-not-a} and
  \eqref{linear:eq:ILED-near:boundary-terms:subprinc-a} concludes the proof
  of Lemma \ref{linear:lemma:ILED-near:boundary-terms}.
\end{proof}

\paragraph{Resolvent estimate for $\Im\sigma>\frac{\SpectralGap}{2}$}
Let us first
begin then by showing that a naive Gronwall-type energy estimate
allows us to reduce the problem to the region where
$\abs*{\Im\sigma}\le\SpectralGap$.

\begin{lemma}
  \label{linear:lemma:ILED-near:Gronwall-resolvent-estimate}
  Let 
  $u$ be as specified in the assumptions of Theorem
  \ref{linear:thm:ILED-near:main}. 
  Then, for any $\delta>0$, there exists some constant $C>0$ such that
  for $\Im\sigma > \delta$, $\abs*{\sigma}>C$,
  \begin{equation}
    \label{linear:eq:ILED-near:Gronwall-resolvent-estimate}
    \norm{u}_{\InducedLTwo(\TrappingNbhd)} \lesssim \norm*{\widehat{\LinEinsteinConj}_{g_b}(\sigma)u}_{\InducedLTwo(\TrappingNbhd)}.
  \end{equation}
\end{lemma}

\begin{proof}
  The proof of this lemma will proceed in the same way as Corollary
  \ref{linear:corollary:naive-energy-estimate}, we will take advantage of the
  smallness of $\SubPConjSym_b$ by microlocalizing to a neighborhood
  of $\TrappedSet_b$.

  Recall that $\KillT$ is uniformly timelike
  and Killing on $\TrappingNbhd$. Thus, on $\TrappingNbhd$,
  $\EnergyKill(\tStar)[{h}]$ is coercive. Moreover, using the
  divergence theorem we have that
  \begin{align*}
    \EnergyKill(\TStar)[{h}] - \EnergyKill(0)[{h}]
    ={}& - \Re\bangle{\LinEinsteinConj_{g_b} {h}, \KillT {h}}_{L^2(\DomainOfIntegration)}
         + \Re\bangle*{\SubPConjOp_b[{h}],\KillT {h}}_{L^2(\DomainOfIntegration)}
         + \Re\bangle*{\PotentialConjOp_b{h}, \KillT {h}}_{L^2(\DomainOfIntegration)}.
  \end{align*}
  Taking a $\p_{\tStar}$ derivative, we have that
  \begin{align*}
    \p_{\tStar}\EnergyKill(\tStar)[{h}]
    ={}& - \Re\bangle{\LinEinsteinConj_{g_b} {h}, \KillT {h}}_{\InducedLTwo(\TrappingNbhd)}
         + \Re\bangle*{\SubPConjOp_b[{h}],\KillT {h}}_{\InducedLTwo(\TrappingNbhd)}
         + \Re\bangle*{\PotentialConjOp_b{h}, \KillT {h}}_{\InducedLTwo(\TrappingNbhd)}.
  \end{align*}
  Using Lemma \ref{linear:lemma:ILED-near:LoT-control:symmetry-control}, we
  have that
  \begin{align*}
    \p_{\tStar}\EnergyKill(\tStar)[{h}]
    ={}& - \Re\bangle{\LinEinsteinConj_{g_b} {h}, \KillT {h}}_{\InducedLTwo(\TrappingNbhd)}
         + \Re\bangle*{\SubPConjOp_{b,a}[{h}],\KillT {h}}_{\InducedLTwo(\TrappingNbhd)}
         + \Re\bangle*{\PotentialConjOp_b{h}, \KillT {h}}_{\InducedLTwo(\TrappingNbhd)}.
  \end{align*}
  Using Lemma \ref{linear:lemma:ILED-near:s-decomp}, we have that for any
  fixed $\epsilon>0$, there exists a choice of $a$,
  $\varepsilon_{\TrappedSet}$, $\delta_r$, and $\delta_\zeta$
  sufficiently small such that
  $\abs*{\SubPConjSym_b} \le \delta_0\abs*{\zeta}$.  As a result,
  large,
  \begin{align*}
    \abs*{\bangle*{\SubPConjOp_{b,a}{h},  \KillT{h}}_{\InducedLTwo(\TrappingNbhd)}}
    &< \epsilon\norm*{h}_{\HkWithT{1}(\TrappingNbhd)}^2,\\    
    \abs*{\bangle*{\PotentialConjOp_b{h}, \KillT {h}}_{\InducedLTwo(\TrappingNbhd)}}
    &< \epsilon \norm*{h}_{\InducedLTwo(\TrappingNbhd)}^2 + C(\epsilon)\norm*{h}_{\InducedLTwo(\TrappingNbhd)}^2.       
  \end{align*}
  Recall that $\TrappingNbhd$ does not intersect the ergoregion of
  $g_b$ and that therefore $\EnergyKill(\tStar)[h]$ is strictly
  positive and coercive for $h$ such that $h(\tStar, \cdot)$ is
  supported on $\TrappingNbhd$ for all $\tStar$. As a result, for
  ${h}=e^{-\ImagUnit\sigma\tStar}u$, 
  \begin{equation*}
    \p_{\tStar}\EnergyKill(\tStar)[{h}]
    \gtrsim \Im \sigma \norm*{h}_{\HkWithT{1}(\TrappingNbhd)}^2. 
  \end{equation*}
  Since we are only considering $\Im\sigma>\delta$, the right-hand
  side is positive, and now we can now multiply through by
  $e^{2\Im\sigma\tStar}$ to remove any $\tStar$ dependency to see that
  \begin{equation*}
    \norm*{u}_{\InducedHk{1}(\TrappingNbhd)}^2
    + \norm*{\sigma u}_{\InducedLTwo(\TrappingNbhd)}^2
    \lesssim \norm*{\widehat{\LinEinsteinConj}_{g_b}(\sigma)u}_{\InducedLTwo(\TrappingNbhd)}^2
    +\norm*{u}_{\InducedLTwo(\TrappingNbhd)}^2.
  \end{equation*}
  We conclude by observing as before that the lower-order term on the
  right-hand side can be absorbed into the left-hand side for
  $\abs*{\sigma}$ sufficiently large. 
\end{proof}

\paragraph{Reduction to conjugated operator}

All of the estimates thus far in the section have been proven by
multiplying against $\LinEinsteinConj_{g_b}$ instead of
$\LinEinstein$. In this section, we show that this is in fact
sufficient to prove Theorem \ref{linear:thm:ILED-near:main}.

We first prove the following auxiliary lemma. 
\begin{lemma}
  \label{linear:lemma:ILED-near:D-t-minus-1-control}
  Let $\DomainOfIntegration$ be as defined in
  \eqref{linear:eq:ILED-trapping:trapping-reg-def}. Then for $\delta_r$
  sufficiently small and $u$ compactly supported on $\TrappingNbhd$,
  we have that for $\GronwallExp\ge \Im\sigma \ge -\SpectralGap$,
  \begin{equation*}
    \norm*{\p_{\tStar}u}_{\InducedHk{-1}(\TrappingNbhd)}^2
    \lesssim \norm*{\widehat{\LinEinstein}_{g_b}(\sigma)u}_{\InducedLTwo(\TrappingNbhd)}^2
    + \norm*{u}_{L^2(\DomainOfIntegration)}^2. 
  \end{equation*}
\end{lemma}

\begin{proof}
  Let $R_{-1}\in \Op S^{-1}(\TrappingNbhd)$ be a compactly supported
  self-adjoint operator. We use $R_{-1}$ as a Lagrangian
  multiplier. Integrating by parts, we have that 
  \begin{align*}
    2\bangle*{\GInvdtdt^{-1}\ScalarWaveOp[g_b]h, R_{-1}^2h}_{L^2(\DomainOfIntegration)}
    ={}& \norm*{R_{-1}\p_{\tStar}h}_{L^2(\DomainOfIntegration)}^2
    + O\left(\norm*{R_{-1}\p_{\tStar}h}_{L^2(\DomainOfIntegration)}\norm*{h}_{L^2(\DomainOfIntegration)} + \norm*{h}_{L^2(\DomainOfIntegration)}^2\right)\\
    & + \evalAt*{\bangle*{n_{\TrappingNbhd}h, R_{-1}h}_{\InducedLTwo(\TrappingNbhd)}}_{\tStar=0}^{\tStar=\TStar}. 
  \end{align*}
  Thus, applying Cauchy-Schwarz,
  \begin{equation}
    \label{linear:eq:ILED-near:D-t-minus-1-control:aux1}
    \norm*{R_{-1}\p_{\tStar}h}_{L^2(\DomainOfIntegration)}^2
    + \evalAt*{\bangle*{n_{\TrappingNbhd}h, R_{-1}h}_{\InducedLTwo(\TrappingNbhd)}}_{\tStar=0}^{\tStar=\TStar}
    \lesssim \norm*{\LinEinstein_{g_b}h}_{L^2(\DomainOfIntegration)}^2
    + \norm*{h}_{L^2(\DomainOfIntegration)}^2.
  \end{equation}
  Now applying \eqref{linear:eq:ILED-near:D-t-minus-1-control:aux1} to
  $h=e^{-\ImagUnit\sigma\tStar}u$, and multiplying both sides by
  $e^{2\Im\sigma\tStar}$ to remove any $\tStar$-dependency, we have
  using Cauchy-Schwarz and the boundedness of $\abs*{\Im\sigma}\le
  \GronwallExp$, that 
  \begin{equation*}
    \norm*{R_{-1}\p_{\tStar}u}_{\InducedLTwo(\TrappingNbhd)}^2
    \lesssim \norm*{\widehat{\LinEinstein}_{g_b(\sigma)}h}_{\InducedLTwo(\TrappingNbhd)}^2
    + \norm*{u}_{L^2(\DomainOfIntegration)}^2. 
  \end{equation*}
  So in fact,
  \begin{equation*}
    \norm*{\p_{\tStar}u}_{\InducedHk{-1}(\TrappingNbhd)}^2
    \lesssim \norm*{\widehat{\LinEinstein}_{g_b(\sigma)}h}_{\InducedLTwo(\TrappingNbhd)}^2
    + \norm*{u}_{L^2(\TrappingNbhd)}^2,
  \end{equation*}
  as desired.
\end{proof}

We are now ready to show that it suffices to prove Theorem
\ref{linear:thm:ILED-near:main} with $\LinEinsteinConj_{g_b}$ in place of
$\LinEinstein_{g_b}$.
\begin{lemma}
  \label{linear:lemma:ILED-near:legal-conjugation}
  Let $\DomainOfIntegration$ and $\TrappingNbhd$ be as defined in
  \eqref{linear:eq:ILED-trapping:trapping-reg-def}, and $u$ be a sufficiently
  smooth function supported in $\TrappingNbhd$ such that
  $\hat{u}(\xi, \eta)$ is supported on the region
  $\frac{\abs*{\xi}^2}{|\eta|^2} \le \frac{1}{2}\delta_\zeta$.  Then
  if there exists $\SpectralGap>0$, $C_0>0$, such that 
  \begin{equation}
    \label{linear:eq:ILED-near:main-conjugated}
    \norm{u}_{\MorawetzNorm_b(\TrappingNbhd)} \lesssim
    \norm*{\widehat{\LinEinsteinConj}_{g_b}(\sigma)u}_{\InducedLTwo_\sigma(\TrappingNbhd)},\qquad
    \text{if }\Im\sigma> -\SpectralGap,\text{ and } |\sigma| \ge C_0;
    \text{ or } \Im\sigma = -\SpectralGap,
  \end{equation}
  Then Theorem \ref{linear:thm:ILED-near:main} holds. 
\end{lemma}
\begin{proof}
  Observe that there exists some $C_{ell}>0$ such that if
  $\abs*{\sigma} > C_{ell}(\abs*{\xi} + \abs*{\FreqAngular})$, then
  $\PrinSymb_b > 0 $ is elliptic.  Then we first show that it is
  sufficient to prove Theorem \ref{linear:thm:ILED-near:main} for
  $\hat{u}(\xi, \FreqAngular)$ supported on\footnote{In fact, it is
    sufficient to prove Theorem \ref{linear:thm:ILED-near:main} for
    $\hat{u}(\xi, \FreqAngular)$ supported on
    $C_{ell}\abs*{\sigma}>\abs*{\xi, \FreqAngular} >
    \frac{1}{C_{ell}}\abs*{\sigma}$, but just proving the lower bound
  is sufficient here.}
  $\abs*{\xi, \FreqAngular} > \frac{1}{C_{ell}}\abs*{\sigma}$.
  
  To this end, let $\hat{u}(\xi, \FreqAngular)$ be supported on
  $\abs*{\sigma} > C_{ell}(\xi + \abs*{\FreqAngular})$. Then
  $\PrinSymb_b$ is elliptic, and we have directly via
  \eqref{linear:eq:div-them:J-K-currents} that
  \begin{equation*}
    -\bangle*{\LinEinstein_{g_b}h, h}_{L^2(\TrappingNbhd)}
    = \int_{\TrappingNbhd}\KCurrent{0, 1, 0}[h]\,\sqrt{\GInvdtdt}
    - \bangle*{\SubPOp_bh, h}_{L^2(\TrappingNbhd)}
    - \bangle*{\PotentialOp_bh, h}_{L^2(\TrappingNbhd)}
    + \p_{\tStar}\int_{\TrappingNbhd} \JCurrent{0, 1, 0}[h] ,
  \end{equation*}
  where 
  \begin{equation*}
    \KCurrent{0,1,0}[h] \gtrsim \abs*{\nabla h}^2 - C|h|^2, \qquad
    \abs*{\JCurrent{0,1,0}[h]} < \epsilon \abs*{\nabla h}^2 + C(\epsilon)\abs*{h}^2. 
  \end{equation*}
  Directly from Cauchy-Schwarz, we have then that
  \begin{equation*}
    \int_{\TrappingNbhd}\KCurrent{0, 1, 0}[h]\,\sqrt{\GInvdtdt}
    - \bangle*{\SubPOp_bh, h}_{L^2(\TrappingNbhd)}
    - \bangle*{\PotentialOp_bh, h}_{L^2(\TrappingNbhd)}
    \gtrsim \norm*{h}_{\HkWithT{1}(\TrappingNbhd)}^2 - \norm*{h}_{\InducedLTwo(\TrappingNbhd)}^2. 
  \end{equation*}
  Then, considering $h=e^{-\ImagUnit\sigma\tStar}u(x)$, and multiplying
  both sides by $e^{2\Im\sigma\tStar}$ and using the fact that
  $\abs*{\Im\sigma}\le \GronwallExp$ as usual, we have that
  \begin{equation*}
    \norm*{u}_{\InducedHk{1}(\TrappingNbhd)}^2 + \norm*{\sigma u}_{\InducedLTwo(\TrappingNbhd)}^2
    \lesssim  \norm*{\widehat{\LinEinstein}_{g_b}(\sigma)u}_{\InducedLTwo(\TrappingNbhd)}^2 + \norm*{u}_{\InducedLTwo(\TrappingNbhd)}^2.  
  \end{equation*}
  The desired resolvent estimate in \eqref{linear:eq:ILED-near:main} follows
  directly by taking $\abs*{\sigma}$ sufficiently large and absorbing
  the lower-order term on the right-hand side into the left-hand side.

  We now show that return to proving the lemma itself. Since it
  suffices to consider $\hat{u}(\xi, \FreqAngular)$ supported on
  $\abs*{\xi, \FreqAngular} > \frac{1}{C_{ell}}\abs*{\sigma}$, and we
  are can choose $\abs*{\sigma}$ as large as we like, we can in
  particular choose $\abs*{\sigma}$ large enough so that $\hat{u}(\xi,
  \FreqAngular)$ is supported on the region where
  \begin{equation*}
    \norm*{\PseudoSubPFixer u}_{\InducedHk{-1}(\TrappingNbhd)} \lesssim \norm*{u}_{\InducedLTwo(\TrappingNbhd)}, \qquad  \norm*{\PseudoSubPFixer^- u}_{\InducedHk{-1}(\TrappingNbhd)} \lesssim \norm*{u}_{\InducedLTwo(\TrappingNbhd)}.
  \end{equation*}
  Since $\PseudoSubPFixer$ is elliptic, we now have that 
  \begin{equation*}
    \norm*{u}_{L^2(\TrappingNbhd)}\lesssim \norm*{\PseudoSubPFixer u}_{L^2(\TrappingNbhd)} \lesssim \norm*{u}_{L^2(\TrappingNbhd)},
  \end{equation*}
  Moreover, we also have using Lemma
  \ref{linear:lemma:ILED-near:LoT-control:lowest-order-control} that for
  $\delta_r$ sufficiently small,
  \begin{equation*}
    \norm*{u}_{\InducedMorawetzNorm(\TrappingNbhd)}\lesssim \norm*{\PseudoSubPFixer u}_{\InducedMorawetzNorm(\TrappingNbhd)}.
  \end{equation*}  
  As a result, we have that 
  \begin{align*}
    \norm*{u}_{\InducedMorawetzNorm(\TrappingNbhd)}
    &\lesssim
      \norm*{\PseudoSubPFixer u}_{\InducedMorawetzNorm(\TrappingNbhd)} \notag \\
    &\lesssim \norm*{\widehat{\LinEinsteinConj}_{g_b}(\sigma)\left(\PseudoSubPFixer u\right)}_{\InducedLTwo(\TrappingNbhd)}\notag \\
    &\lesssim \norm*{\PseudoSubPFixer \widehat{\LinEinstein}(\sigma) u}
      + \norm*{\widehat{\LinEinstein}_{g_b}(\sigma) \mathfrak{r}u}_{L^2(\DomainOfIntegration)}\notag \\
    &\lesssim \norm*{\widehat{\LinEinstein}(\sigma)u}_{\InducedLTwo(\TrappingNbhd)}
      + \norm*{\left[\widehat{\LinEinstein}_{g_b}(\sigma), \mathfrak{r}\right]u}_{\InducedLTwo(\TrappingNbhd)},
  \end{align*}
  where $\mathfrak{r} = 1 - \PseudoSubPFixer\PseudoSubPFixer^- \in \Op
  S^{-\infty}$.
  Then, we observe that $\left[\widehat{\LinEinstein}_{g_b}(\sigma),
    \mathfrak{r}\right] \in \sigma S^{-1}(\Sigma)$.
  As a result, we can directly use Lemma
  \ref{linear:lemma:ILED-near:D-t-minus-1-control} to write that
  \begin{equation*}
    \norm*{u}_{\InducedMorawetzNorm(\TrappingNbhd)}
    \lesssim \norm*{\widehat{\LinEinstein}(\sigma)u}_{\InducedLTwo(\TrappingNbhd)}
    + \norm*{u}_{\InducedLTwo(\TrappingNbhd)}. 
  \end{equation*}
  But then using Lemma
  \ref{linear:lemma:ILED-near:LoT-control:lowest-order-control}, we can
  choose $\delta_r$ sufficiently small so that in fact,
  \begin{equation*}
    \norm*{u}_{\InducedMorawetzNorm(\TrappingNbhd)}
    \lesssim \norm*{\widehat{\LinEinstein}(\sigma)u}_{\InducedLTwo(\TrappingNbhd)}, 
  \end{equation*}
  as desired.
\end{proof}

\paragraph{Proving Theorem \ref{linear:thm:ILED-near:main}}

We are now ready to prove Theorem \ref{linear:thm:ILED-near:main}. We first
show how to combine the above results to conclude the $k=1$ case of
Theorem \ref{linear:thm:ILED-near:main}.

\begin{proof}[Proof of Theorem \ref{linear:thm:ILED-near:main} for $k=1$.]
  Using Lemma \ref{linear:lemma:ILED-near:sum-of-squares} and Corollary
  \ref{linear:coro:ILED-near:SX-control} we can write
  \begin{equation*}
    \begin{split}
      \int_{\DomainOfIntegration}\KCurrent{\MorawetzVF_{b_0},\LagrangeCorr_{b_0},0}[{h}]
      + a\Re\KCurrentIbP{\widetilde{\MorawetzVF}, \tilde{\LagrangeCorr}}_{(2)}[{h}]
      - \Re \bangle*{\SubPConjOp_b[{h}], \MorawetzVF_{b_0}{h}}_{L^2(\DomainOfIntegration)}
      \ge{}&\frac{1}{2}\int_{\DomainOfIntegration}\sum_{j=1}^7\abs*{\SquareDecompOp_j{h}}^2
      + \int_{\DomainOfIntegration}\nabla^\alpha\p_\alpha\LagrangeCorr_{b_0}\abs*{{h}}^2.
    \end{split}
  \end{equation*}
  Choosing $a$, and $\delta$, sufficiently small we can use
  Lemma \ref{linear:lemma:ILED-near:LoT-control} to control the lower order
  bulk terms as well and see that
  \begin{equation*}
    \begin{split}
      \int_{\DomainOfIntegration}\sum_{j=1}^7\abs*{\SquareDecompOp_j{h}}^2
      + \norm{{h}}_{L^2(\DomainOfIntegration)}^2 
      \lesssim{}&\int_{\DomainOfIntegration}\KCurrent{\MorawetzVF_{b_0},\LagrangeCorr_{b_0},0}[{h}]
      - \Re \bangle*{\SubPConjOp_b[{h}], (\MorawetzVF_{b_0}+\LagrangeCorr_{b_0}){h}}_{L^2(\DomainOfIntegration)}\\
      &- \Re\bangle*{\PotentialConjOp_b{h},(\MorawetzVF_{b_0}+\LagrangeCorr_{b_0}){h} }_{L^2(\DomainOfIntegration)}
       + a\Re\KCurrentIbP{\widetilde{\MorawetzVF}, \tilde{\LagrangeCorr}}[{h}]
      + \left.\tilde{J}(\tStar)[{h}]\right\vert_{\tStar=0}^{\tStar = \TStar}.
    \end{split}
  \end{equation*}
  Defining
  \begin{equation*}
    \mathring{\MorawetzEnergy}(\tStar)[{h}] = \int_{\Sigma_{\tStar}}\JCurrent{\MorawetzVF_{b_0}, \LagrangeCorr_{b_0},0}[{h}]\cdot n_{\TrappingNbhd}
    + a \Re \JCurrentIbP{\widetilde{\MorawetzVF}, \tilde{\LagrangeCorr}}(\tStar)[{h}]
    + \tilde{J}(\tStar)[{h}],
  \end{equation*}
  it is clear that
  \begin{equation}
    \label{linear:eq:ILED-near:physical-ineq}
    \left.\mathring{\MorawetzEnergy}(\tStar)[{h}]\right\vert_{\tStar = 0}^{\tStar = \TStar} +
    \norm{{h}}_{\MorawetzNorm_b(\DomainOfIntegration)}^2
    \lesssim \norm*{\LinEinstein_{g_b}{h}}_{L^2(\DomainOfIntegration)}^2.
  \end{equation}
  We now prove the resolvent estimate in \eqref{linear:eq:ILED-near:main} for
  $k=1$. As before, we consider $h= e^{-\ImagUnit\sigma\tStar}u(x)$,
  where $u$ satisfies the assumptions made in the statement of the
  theorem. 
  
  Differentiating both sides of \eqref{linear:eq:ILED-near:physical-ineq} by
  $\p_{\tStar}$ and multiplying by $e^{2\Im\sigma\tStar}$ to remove
  any $\tStar$-dependency, we see that
  \begin{equation*}
    \begin{split}
      \norm{u}_{\MorawetzNorm_b(\TrappingNbhd)}^2
      \lesssim{}& \norm*{\widehat{\LinEinstein}_{g_b}(\sigma)u}_{\LTwo(\TrappingNbhd)}^2
      + e^{2\Im\sigma\tStar}\abs*{\p_{\tStar}\left(\int_{\Sigma_{\tStar}}\JCurrent{\MorawetzVF_{b_0}, \LagrangeCorr_{b_0},0}[{h}]\cdot n_{\TrappingNbhd}
          + a \Re\JCurrentIbP{\widetilde{\MorawetzVF}, \tilde{\LagrangeCorr}}(\tStar)[{h}]
          + \tilde{J}(\tStar)[{h}]
        \right)}.
    \end{split}      
  \end{equation*}
  We now show that for $\SpectralGap$ sufficiently small, and $|\sigma|$
  sufficiently large, the boundary terms on the right-hand side can be
  absorbed into the bulk norm on the left-hand side.  Since
  ${h}=e^{-\ImagUnit\sigma\tStar}u$, we can now apply Lemma
  \ref{linear:lemma:ILED-near:boundary-terms}, and the bound on
  $\tilde{J}(\tStar)[{h}]$ from Lemma
  \ref{linear:lemma:ILED-near:LoT-control} to choose
  $\varepsilon_{\TrappedSet_{b_0}}$, $a$, $\SpectralGap$ sufficiently small
  so that for $\abs*{\Im\sigma} \le \SpectralGap$,
  \begin{equation*}
    \begin{split}
      e^{2\Im\sigma\tStar}\abs*{\p_{\tStar}\left(\int_{\Sigma_{\tStar}}\JCurrent{\MorawetzVF_{b_0}, \LagrangeCorr_{b_0},0}[{h}]\cdot n_{\TrappingNbhd}
          + a \Re\JCurrentIbP{\widetilde{\MorawetzVF}, \tilde{\LagrangeCorr}}(\tStar)[{h}]
          + \tilde{J}(\tStar)[{h}]
        \right)}
      \lesssim \delta \norm{u}_{\InducedMorawetzNorm_b(\TrappingNbhd)}^2.
    \end{split}
  \end{equation*}
  
  Choosing $\delta$ sufficiently small, we conclude via Lemma
  \ref{linear:lemma:ILED-near:sum-of-squares} and Proposition
  \ref{linear:prop:div-thm:PDO-modification} that for $\abs*{\sigma}$
  sufficiently large,
  \begin{equation*}
    \norm{u}_{\InducedMorawetzNorm_b(\TrappingNbhd)}
    \lesssim \norm*{\widehat{\LinEinstein}_{g_b}(\sigma)u}_{\InducedLTwo(\TrappingNbhd)},
  \end{equation*}
  as desired.
\end{proof}

We now show how to commute derivatives through the
relevant estimates to obtain higher order resolvent
estimates. The main idea here is that $\KillT$ is both Killing and
time-like on $\TrappingNbhd$. As a result, we can use a combination of
commuting with $\KillT$ and elliptic estimates to prove via induction
higher-order estimates.

\begin{proof}[Proof of Theorem \ref{linear:thm:ILED-near:main} for $k>1$.]
  We begin with the observation that since $\KillT$ is Killing, it
  commutes with $\LinEinstein_{g_b}$. As a result,
  \begin{equation*}
    \norm{\sigma u}_{\InducedMorawetzNorm_b(\TrappingNbhd)} \lesssim \norm*{\sigma\widehat{\LinEinstein}_{g_b}(\sigma)u}_{\InducedLTwo(\TrappingNbhd)}.
  \end{equation*}
  To control the rest of the derivatives, we recall that we can write 
  \begin{equation*}
    \LinEinstein_{g_b} = \GInvdtdt^{-1} D_{\tStar}^2 + P_1D_{\tStar} + P_2,  
  \end{equation*}
  where $P_i\in OPS^i(\TrappingNbhd)$. We rewrite the main equation as
  \begin{equation*}
    P_2{h} = \LinEinstein_{g_b}{h} - \GInvdtdt^{-1} D_{\tStar}^2 {h} - P_1D_{\tStar}{h}.
  \end{equation*}
  We conclude by recalling that $P_2$ is
  elliptic on $\TrappingNbhd$ and using standard elliptic estimates.  
\end{proof}
\subsection{Proof of Theorem \ref{linear:thm:resolvent-estimate:main}} \label{linear:ILED:full}

At this point, we have proven resolvent estimates for
${h}$ supported on the redshift region, the non-trapping region, and
the trapping region. In this section, the goal will be to glue these
estimates together. To this end, recall the relation between the
constants
\begin{equation*}
  r_{\EventHorizonFuture} < \breve{r}_{-}
  < r_{\RedShift,\EventHorizonFuture} <r_{0} < \mathring{r}_{-}
  < \breve{r}_{+}<3M <\breve{R}_{-}<\mathring{R}_{ +}
   < R_{0} < R_{\RedShift,\EventHorizonFuture} < \breve{R}_{+}
  < r_{\CosmologicalHorizonFuture}.
\end{equation*}
In practice, it will be useful to consider
\begin{gather*}
  \breve{r}_{-} = r_{\EventHorizonFuture}+\delta_{\Horizon}, \quad 
  r_{\bullet, \EventHorizonFuture}  = r_{\EventHorizonFuture} + 2\delta_{\Horizon} ,\quad
  \mathring{r}_{ -} = 3M- 4\delta_r,  \quad
  \breve{r}_{ +} = 3M - 2\delta_r, 
  \\
  \breve{R}_{-} = 3M + 2\delta_r,\quad
  \mathring{R}_{ +} = 3M + 4\delta_r,\quad    
  R_{\bullet, \CosmologicalHorizonFuture}  = r_{\CosmologicalHorizonFuture} - 2\delta_{\Horizon} ,\quad
  \breve{R}_+ = r_{\CosmologicalHorizonFuture} - \delta_{\Horizon}.
\end{gather*}
From Theorems \ref{linear:thm:ILED-redshift:main}, 
\ref{linear:thm:ILED-nontrapping:resolvent-estimate:main}, 
\ref{linear:thm:ILED-nontrapping-freq:main}, and 
\ref{linear:thm:ILED-near:main}, we have that for $k>\ThresholdReg$
\begin{equation*}
  \begin{split}
    \norm{u}_{\CombinedHk{k}(\Sigma)} \lesssim
    &\norm*{\widehat{\LinEinstein}_{g_b}(\sigma)(\dot{\chi}u)}_{\InducedHk{k-1}_\sigma(\Sigma)} +
    \norm*{\widehat{\LinEinstein}_{g_b}(\sigma)(\breve{\chi}u)}_{\InducedHk{k-1}_\sigma(\Sigma)} +
    \norm*{\widehat{\LinEinstein}_{g_b}(\sigma)(\mathring{\chi} u)}_{\InducedHk{k-1}_\sigma(\Sigma)}
    +\norm{u}_{\CombinedHk{k-1}(\Sigma)},
  \end{split}
\end{equation*}
where we recall the definition of the $LE^k$ norm defined in
Definition \ref{linear:def:LE-norm}.  We see here that this is not enough to
conclude, since we have to commute the cut-off functions with
$\widehat{\LinEinstein}_{g_b}(\sigma)$, producing commutation error
terms in $\InducedHk{k}_\sigma(\Sigma)$ that can not necessarily be
controlled by the left-hand side. To achieve the desired estimate on
the entire exterior region then, we will rely on a careful patching
that will rely on making a good choice of the relevant
constants. Recall that in the slowly-rotating regime, the redshift
regions and the trapping regions are physically disjoint from one
another. We then divide our analysis into two components: one
analyzing the intersection of the trapping and non-trapping regions,
and one analyzing the intersection of the non-trapping and redshift
regions.

To glue the estimates together, consider that
\begin{equation}
  \label{linear:eq:ILED-full:commutator:aux1}
  \bangle*{\chi\LinEinstein_{g_b} {h},(\MorawetzVF+\LagrangeCorr)(\chi{h})}_{L^2(\DomainOfIntegration)}
  = \bangle*{\squareBrace*{\chi, \LinEinstein_{g_b}}{h}, (\MorawetzVF+\LagrangeCorr)(\chi{h})}_{L^2(\DomainOfIntegration)}
  + \bangle*{\LinEinstein_{g_b}(\chi{h}),(\MorawetzVF+\LagrangeCorr)(\chi{h})}_{L^2(\DomainOfIntegration)}.
\end{equation}
Letting $\chi$ be one of the cutoff functions defined at the beginning
of the section, and $\MorawetzVF, \LagrangeCorr$ be a pair of the
multipliers constructed in the earlier sections, we see that the
second term on the right-hand side of equation
\eqref{linear:eq:ILED-full:commutator:aux1} can already be controlled by one
of Theorems \ref{linear:thm:ILED-redshift:main}, 
\ref{linear:thm:ILED-nontrapping:resolvent-estimate:main},
\ref{linear:thm:ILED-nontrapping-freq:main}, \ref{linear:thm:ILED-near:main}. It
remains to choose the $\chi$ in such a way such that the commutator
term is also controlled and does not disrupt the ellipticity generated
by the previously proven Morawetz estimates.

\begin{remark}
  \label{linear:rmk:ILED-combine:important-terms}
  Recalling from Lemma
  \ref{linear:lemma:EVE:GHC-quasilinear} that $\LinEinstein_{g_b}$ is
  strongly hyperbolic, we observe that
  \begin{equation*}
    [\chi, \LinEinstein_{g_b}]
    = [\chi, \ScalarWaveOp[g_b]] + [\chi, \SubPOp_b] + [\chi, \PotentialOp_b].
  \end{equation*}
  For any smooth, compactly supported $\chi=\chi(r)$, such that the
  support of $\p_r\chi$ is supported away from the trapped set,
  $[\chi, \SubPOp_b] + [\chi, \PotentialOp_b]$ is bounded as an
  operator on
  $\LTwo(\DomainOfIntegration)\to\LTwo(\DomainOfIntegration)$, they
  can be controlled by a high-frequency argument in the combined
  Morawetz estimate. Thus, it will suffice in what follows to consider
  the error term produced by $[\chi, \ScalarWaveOp[g_b]]$. The same
  result is true when considering the commutation
  $[\chi, \LinEinsteinConj_{g_b}]$.
\end{remark}

We will first show that we can combine the Morawetz estimates in the
non-trapping region and the trapping region.


\begin{prop}
  \label{linear:prop:ILED-combine:not-redshift}
  Let $g_b$ be a fixed slowly-rotating \KdS{} background, and let  
  \begin{equation*}
    \DomainOfIntegration := [0, \TStar]_{\tStar}\times
    \Sigma^{\NotRedshift},\qquad
    \Sigma^{\NotRedshift} := \Sigma\left(\breve{r}_-,
    \breve{R}_+\right).
  \end{equation*}
  Then for $u$ compactly supported on $\Sigma$ and $k\ge 1$, there
  exist constants $\SpectralGap, C_0>0$ such that we have the
  following resolvent estimate for
  $ \GronwallExp \ge \Im\sigma\ge -\SpectralGap, \abs*{\sigma} > C_0$,
  \begin{equation}
    \label{linear:eq:ILED-combine:not-redshift}
    \norm{u}_{\CombinedHk{k}(\Sigma^{\NotRedshift})}
    \lesssim \norm{\widehat{\LinEinstein}_{g_b}(\sigma)u}_{\InducedHk{k-1}_\sigma(\Sigma^{\NotRedshift})}. 
  \end{equation}
\end{prop}

As mentioned previously, we will prove this proposition by combining
the vectorfields and Lagrangian correctors used to prove Theorems
\ref{linear:thm:ILED-nontrapping:resolvent-estimate:main},
\ref{linear:thm:ILED-nontrapping-freq:main}, and \ref{linear:thm:ILED-near:main}.
Recall that in handling the Morawetz estimates in a neighborhood of
$r=3M$, we needed to localize to a neighborhood of $\TrappedSet_{b}$
in both physical space and frequency space. Now that we are
considering ${h}$ without compact support in a neighborhood of
$\TrappedSet_{b}$, we will need to add both physical cutoffs and
frequency cutoffs and show that we can combine the resolvent estimates
in the trapping and non-trapping regimes.

Recalling that on
$\Real^+\times (\mathring{r}_{ -},\mathring{R}_{ +}) \times
\Sphere^2$, the $(\tStar, r, \theta,\phiStar)$ coordinates agree with
the Boyer-Lindquist coordinates, we work in the
$(t,r,\theta,\varphi;\sigma,\xi,\FreqTheta, \FreqPhi)$ Boyer-Lindquist
coordinates in what follows.

Observe that the cutoff functions
$\chi\in \Op S^0(\StaticRegionWithExtension)$ will not effect
our integration-by-parts arguments. As such, we can repeat the proofs
of Theorems \ref{linear:thm:ILED-nontrapping:resolvent-estimate:main},
\ref{linear:thm:ILED-nontrapping-freq:main}, \ref{linear:thm:ILED-near:main} with
$\breve{\chi}{h}$, $\mathring{\chi}\breve{\chi}_{\zeta}{h}$, and
$\mathring{\chi}\mathring{\chi}_{\zeta}{h}$ in place of ${h}$
respectively.
The resulting commutation error is handled in the
following lemma.
\begin{lemma}
  \label{linear:lemma:ILED-full:Trapping-NonTrapping}
  There exists a choice of constants
  $\breve{c}, \check{c}, \mathring{c}$, such that for vectorfields
  $\MorawetzOuterVF, \MorawetzInnerNTVF, \MorawetzVF_b$, and
  Lagrangian correctors
  $\LagrangeCorrOuter, \LagrangeCorrInnerNT, \LagrangeCorr_{b}$ as
  defined above such that
  \begin{equation}
    \label{linear:eq:ILED-full:Trapping-NonTrapping:main}
    \KCurrentIbPSym{\MorawetzSym^{\NotRedshift},\LagrangeCorrSym^{\NotRedshift}} 
    \ge \frac{1}{2}\sum_{j=1}^7 \SquareDecomp_j^2
  \end{equation}
  on $r\in (\breve{r}_{-}, \breve{R}_{+})$, where
  \begin{equation*}
    \MorawetzSym^{\NotRedshift}
    = \breve{\chi}^2\MorawetzOuterSym
    + \mathring{\chi}^2\mathring{\chi}_{\zeta}^2\MorawetzInnerNTSym
    + \mathring{\chi}^2\mathring{\chi}_{\zeta}^2 \MorawetzSym_b,\qquad \MorawetzOuterSym = \frac{\fOuter}{\Delta_b}H_{\rho_b^2\PrinSymb_b}r, \qquad \MorawetzInnerNTSym = \frac{\fInnerNT}{\Delta_b}H_{\rho_b^2\PrinSymb_b}r,
  \end{equation*}
  and
  \begin{equation*}
    \LagrangeCorrSym^{\NotRedshift}
    = \breve{\chi}^2\LagrangeCorrOuter
    + \mathring{\chi}^2\mathring{\chi}_{\zeta}^2\LagrangeCorrInnerNT
    + \mathring{\chi}^2\mathring{\chi}_{\zeta}^2 \LagrangeCorrSym_b.
  \end{equation*}
\end{lemma}

\begin{proof}
  Recall that from Theorems
  \ref{linear:thm:ILED-nontrapping:resolvent-estimate:main} ,
  \ref{linear:thm:ILED-nontrapping-freq:main}, and \ref{linear:thm:ILED-near:main},
  we have already proven the lemma on the regions where the cutoff
  functions are all constant. The main difficulty is to account
  for the error terms that arise from using these cutoffs.

  These terms are of the form $H_{\PrinSymb_b}\chi$, where $\chi$ is a
  cutoff function (see Remark
  \ref{linear:rmk:ILED-combine:important-terms}). Throughout the proof, we
  will normalize the constants $\breve{c}$, $\check{c}$, and
  $\mathring{c}$ so that $\breve{c}=1$. The cutoff errors fall into
  three cases.
  \begin{enumerate}
  \item The errors arising from physically gluing together the
    resolvent estimates in a frequency neighborhood of
    $\TrappedSet_b$, i.e. on the intersection of the supports of
    $\breve{\chi}$, $\mathring{\chi}$, and $\mathring{\chi}_\zeta$. To
    handle this case, we will choose $\COuter \gg \mathring{c},
    \delta_r^{-1}$.    
  \item The errors arising from physically gluing together the
    resolvent estimates outside of a frequency neighborhood of
    $\TrappedSet_b$, i.e. on the intersection of the supports of
    $\mathring{\chi}, \breve{\chi}$, and $\breve{\chi}_\zeta$. To
    handle this case, we will choose $\COuter \gg \check{c},
    \delta_\zeta^{-1}$.  
  \item The errors arising from gluing together the resolvent
    estimates in a frequency space in a neighborhood of $r=3M$,
    i.e. on the intersection of the supports of
    $\mathring{\chi}_\zeta$, $\breve{\chi}_\zeta$, and
    $\mathring{\chi}$. To handle this case, we will choose
    $\mathring{c} \ll \frac{1}{\delta_\zeta}$ and
    $a, \delta_r \ll \frac{\mathring{c}}{\delta_\zeta}$.
  \end{enumerate}
  It suffices to address each of these cases separately. 

  \noindent{\textbf{Step 1: Physical gluing in a frequency
      neighborhood of $\TrappedSet$}.}
  Define
  \begin{equation*}
    \begin{split}
      \MorawetzSym^{\NotRedshift}_1
      &= \breve{\chi}^2\MorawetzOuterSym + \mathring{c}\mathring{\chi}^2\MorawetzSym_b,\\
      \LagrangeCorrSym^{\NotRedshift}_1
      &= \breve{\chi}^2\LagrangeCorrOuterSym + \mathring{c}\mathring{\chi}^2\LagrangeCorrSym_b.
    \end{split}
  \end{equation*}
  Then we will show that we can find $\mathring{c}$ and
  $\COuter_\Trapping$
  , such that
  \begin{equation*}
    \KCurrentIbPSym{\MorawetzSym^{\NotRedshift}_1, \LagrangeCorrSym^{\NotRedshift}_1}
    \gtrsim \sum_{j=1}^7\SquareDecomp_j^2. 
  \end{equation*}
  Using the constructions of $\breve{\chi}$,
  $\mathring{\chi}$, it suffices to prove that for any $\delta>0$,
  there exists a choice of $\mathring{c}, \COuter$ such
  that for $\frac{\xi^2}{|\FreqAngular|^2}\le 2\delta_\zeta^2$,
  \begin{align}    
    \delta \mathring{c} \KCurrentIbPSym{\MorawetzSym_b, \LagrangeCorrSym_b}
    \ge  -\MorawetzOuterSym H_{\PrinSymb_b}(\breve{\chi}^{2}),
    &\qquad r\in (\breve{r}_{ +}, 3M-\delta_r)\bigcup (3M+\delta_r,\breve{R}_{-}), \label{linear:eq:ILED-combine:Nontrapping-trapping:conditions:1}    \\
    \delta \KCurrentIbPSym{\MorawetzOuterSym, \LagrangeCorrOuterSym}
    \ge 
    \mathring{c}\abs*{\MorawetzSym_b H_{\PrinSymb_b}(\mathring{\chi}^{2})},
    &\qquad r\in (\mathring{r}_{ -}-\delta_r, \mathring{r}_{ -})\bigcup (\mathring{R}_{ +}, \mathring{R}_{ +}+\delta_r). \label{linear:eq:ILED-combine:Nontrapping-trapping:conditions:2}        
  \end{align}
  Observe that
  \begin{equation}
    \label{linear:eq:ILED-combine:nontrapping:sign-condition:aux}
    \MorawetzOuterSym H_{\PrinSymb}\breve{\chi} = e^{\COuter (r-3M)^2}(r-3M)\p_r\breve{\chi} (H_{\PrinSymb_b}r)^2,
  \end{equation}
  and that moreover, $(r-3M)\p_r\breve{\chi}\ge 0$ since
  $\p_r\breve{\chi}$ is supported away from $r=3M$. As a result,
  $\MorawetzOuterSym H_{\PrinSymb}\breve{\chi}\ge 0$, and
  \eqref{linear:eq:ILED-combine:Nontrapping-trapping:conditions:1} follows
  immediately.

  It remains to show
  \eqref{linear:eq:ILED-combine:Nontrapping-trapping:conditions:2}.  To this
  end, observe that $\breve{\chi}=1$ on $\supp\p_r\mathring{\chi}$,
  and that $\abs*{\MorawetzSym_b H_{\PrinSymb_b}(\mathring{\chi}^2)} $
  is a bounded second order symbol. Then, using Lemma
  \ref{linear:lemma:ILED-nontrapping:bulk-positivity}, we can pick $\COuter$
  sufficiently large so that
  \eqref{linear:eq:ILED-combine:Nontrapping-trapping:conditions:2} holds.

  \noindent{\textbf{Step 2: Physical gluing in a frequency
      neighborhood away from $\TrappedSet$.}}
  Next we handle the error rising from physically gluing
  together the resolvent estimates microlocalized away from
  $\TrappedSet_b$. To make this more precise, consider
  \begin{align*}
    \MorawetzSym^{\NotRedshift}_2 &= \breve{\chi}^2\MorawetzOuterSym + \check{c}\mathring{\chi}^2\MorawetzInnerNTSym,\\
    \LagrangeCorrSym^{\NotRedshift}_2 &= \breve{\chi}^2\LagrangeCorrOuterSym + \check{c}\mathring{\chi}^2 \LagrangeCorrInnerNTSym.
  \end{align*}
  Then we will show that we can find $\check{c}$, $\check{C}$ such
  that for $\frac{\xi^2}{\abs*{\FreqAngular}^2}\ge \delta_\zeta^2$. 
  \begin{align}
    \delta \check{c} \KCurrentIbPSym{\MorawetzInnerNTSym, \LagrangeCorrInnerNTSym}
    \ge  -\MorawetzOuterSym H_{\PrinSymb_b}(\breve{\chi}^{2}),
    &\qquad r\in (\breve{r}_{ +}, 3M-\delta_r)\bigcup (3M+\delta_r,\breve{R}_{-}), \label{linear:eq:ILED-combine:Nontrapping-trapping:conditions:3}    \\
    \delta \KCurrentIbPSym{\MorawetzOuterSym, \LagrangeCorrOuterSym}
    \ge 
    \check{c}\abs*{\MorawetzInnerNTSym H_{\PrinSymb_b}(\mathring{\chi}^{2})},
    &\qquad r\in (\mathring{r}_{ -}-\delta_r, \mathring{r}_{ -})\bigcup (\mathring{R}_{ +}, \mathring{R}_{ +}+\delta_r). \label{linear:eq:ILED-combine:Nontrapping-trapping:conditions:4}        
  \end{align}
  Using \eqref{linear:eq:ILED-combine:nontrapping:sign-condition:aux}, we see
  that
  \begin{equation*}
    \MorawetzOuterSym H_{\PrinSymb_b}(\breve{\chi}^2) \ge 0,
  \end{equation*}
  so on the region of interest,
  \eqref{linear:eq:ILED-combine:Nontrapping-trapping:conditions:3} is
  immediately satisfied.
  
  To prove \eqref{linear:eq:ILED-combine:Nontrapping-trapping:conditions:4},
  we observe that $H_{\PrinSymb_b}(\mathring{\chi}^2)$ is a
  bounded operator, then with the choice $\CInnerNT = \COuter$, 
  \begin{align*}
    \abs*{\MorawetzInnerNTSym H_{\PrinSymb_b}(\mathring{\chi}^2)}
    = \abs*{e^{\COuter(r-3M)^2}(r-3M)(H_{\PrinSymb_b}r)H_{\PrinSymb_b}(\mathring{\chi}^2)}. 
  \end{align*}
  Then we see that since we are only considering
  $\frac{\xi^2}{\abs*{\FreqAngular}^2}\ge
  \delta_\zeta^2$, we can pick some $\COuter$ sufficiently large
  depending on $\delta_{\zeta}$ and $\delta$ such that
  \begin{equation*}
    \abs*{(r-3M)(H_{\PrinSymb_b}r)H_{\PrinSymb_b}(\mathring{\chi}^2)}
    < \delta\left( 2\COuter(r-3M)^2+1\right)\abs*{H_{\PrinSymb_b}r}^2, 
  \end{equation*}
  As a result, we can take $\check{c}=1$ and $\CInnerNT = \COuter$,
  and see that
  \eqref{linear:eq:ILED-combine:Nontrapping-trapping:conditions:4} is
  satisfied.

  \noindent{\textbf{Step 3: Gluing together the frequency cutoffs.}}
  Having dealt with the errors arising from the spatial cutoff
  functions, it remains to deal with the errors from the frequency
  cutoffs. Because on $\breve{\Sigma}$ we do not divide the
  analysis into two frequency cases, we only have to deal with the
  error arising from using the frequency cutoff on
  $\TrappingNbhd$. That is, define
  \begin{align*}
    \MorawetzSym^{\NotRedshift}_{3}
    &= \mathring{c}\mathring{\chi}^2_{\zeta}\MorawetzSym_{b} + \check{c}\breve{\chi}^2_{\zeta}\MorawetzOuterSym,\\
    \LagrangeCorrSym^{\NotRedshift}_3 
    &= \mathring{c}\mathring{\chi}^2_{\zeta}\LagrangeCorrSym_b + \check{c}\breve{\chi}^2_{\zeta}\LagrangeCorrInnerNTSym.
  \end{align*}
  Then, we will show that there exists a choice of constants such that 
  \begin{equation*}
    \KCurrentIbPSym{\MorawetzSym^{\NotRedshift}_3,\LagrangeCorrSym^{\NotRedshift}_3}
    \gtrsim \sum_{j=1}^7\SquareDecomp_j^2.
  \end{equation*}
  Using the construction of the cutoffs, it suffices to show that
  \begin{align}
    \delta \mathring{c}\KCurrentSym{\MorawetzSym_b, \LagrangeCorrSym_b}
    > & -\check{c}\MorawetzInnerNTSym H_{\RescaledPrinSymb_b}(\breve{\chi}_{\zeta}^2), \qquad
    \supp \mathring{\chi}_{\zeta},\label{linear:eq:ILED-full:freq-trapping-commute:condition:1}\\
    \delta \check{c}\KCurrentSym{\MorawetzInnerNTSym, \LagrangeCorrInnerNTSym}
    > & \abs*{\mathring{c}\MorawetzSym_b H_{\PrinSymb_b} (\mathring{\chi}_{\zeta}^2)},
    \qquad \supp \breve{\chi}_{\zeta}. \label{linear:eq:ILED-full:freq-trapping-commute:condition:2}
  \end{align}
  We first ensure that
  \eqref{linear:eq:ILED-full:freq-trapping-commute:condition:2} holds simply
  by choosing some $\mathring{c}$ sufficiently small.

  To show that \eqref{linear:eq:ILED-full:freq-trapping-commute:condition:1}
  also holds, we will show that in fact
  \begin{equation*}
    \MorawetzInnerNTSym H_{\RescaledPrinSymb_b}(\breve{\chi}_{\zeta}^2) > - O(a)\abs*{\zeta}^2 - O(\delta_r)\xi^2,
  \end{equation*}
  so that for $a$ and $\delta_r$ sufficiently small,
  \eqref{linear:eq:ILED-full:freq-trapping-commute:condition:1} is always
  satisfied. To this end, it will be convenient to observe that
  \begin{equation}
    \label{linear:eq:ILED-full:freq-trapping:Hp-chi-freq}
    H_{\RescaledPrinSymb_b}\breve{\chi}_{\zeta}
    = \p_\xi\breve{\chi}_{\zeta}\frac{\xi H_{\RescaledPrinSymb_b}\xi}{\abs*{\eta}^2}
    + \p_{\eta}\breve{\chi}_{\zeta} \xi^2 S^{-1}(r,\theta,\varphi;\FreqTheta,\FreqPhi).
  \end{equation}  
  Now observe that
  \begin{equation*}
    \MorawetzInnerNTSym H_{\RescaledPrinSymb_b}(\breve{\chi}_\zeta^2)
    = 2\breve{\chi}_{\zeta}e^{\CInnerNT(r-3M)^2}\frac{r-3M}{\rho_b^2}\Delta_b\xi\left(
\p_\xi\breve{\chi}_{\zeta}\frac{\xi H_{\RescaledPrinSymb_b}\xi}{\abs*{\eta}^2}
    +  \p_{\eta}\xi^2 S^{-1}(r,\theta,\varphi;\FreqTheta,\FreqPhi)
    \right).
  \end{equation*}
  Observe that we can also write
  \begin{align*}
    H_{\RescaledPrinSymb_b}\xi
    &= -\p_r\Delta_{b}\xi^2 + \p_r\left(\frac{(1+\lambda_b)^2}{\Delta_b}\left((r^2+a^2)\sigma + a\FreqPhi\right)^2\right),
  \end{align*}
  where from Lemma \ref{linear:lemma:trapping:KdS} we know that 
  \begin{equation*}
    (r-\rTrapping_b)\p_r\left(\frac{(1+\lambda_b)^2}{\Delta_b}\left((r^2+a^2)\sigma + a\FreqPhi\right)^2\right)\ge 0,
  \end{equation*}
  with vanishing exactly at $\TrappedSet_b$. Observing that
  \begin{equation*}
    \abs*{\rTrapping_b - 3M}\lesssim a, 
  \end{equation*}
  we then have that
  \begin{align*}
    \MorawetzInnerNTSym H_{\RescaledPrinSymb_b}(\breve{\chi}_\zeta^2)
    &= 2\frac{\breve{\chi}_\zeta\Delta_b}{\rho_b^2}e^{\CInnerNT(r-3M)^2}\left(
      \frac{\p_\xi\breve{\chi}_\zeta\xi^2}{\abs*{\FreqAngular}^2}(r-\rTrapping_b)\p_r\left(\frac{(1+\lambda_b)^2}{\Delta_b}\left((r^2+a^2)\sigma + a\FreqPhi\right)^2\right)
      \right)\\
    & + O(\delta_r)\xi^2 + a\SymClass^2,
  \end{align*}
  where on the domain of interest,
  \begin{align*}
    \frac{\p_\xi\breve{\chi}_\zeta\xi^2}{\abs*{\FreqAngular}^2}(r-\rTrapping_b)\p_r\left(\frac{(1+\lambda_b)^2}{\Delta_b}\left((r^2+a^2)\sigma + a\FreqPhi\right)^2\right)\ge 0,\\
    \frac{\breve{\chi}_\zeta\Delta_b}{\rho_b^2}e^{\CInnerNT(r-3M)^2} > 0. 
  \end{align*}
  This concludes the proof of Lemma
  \ref{linear:lemma:ILED-full:Trapping-NonTrapping}.        
\end{proof}

Since we are proving a resolvent estimate on a region including the
trapped set, we need two integration-by-parts arguments to prove a
resolvent estimate on $\GronwallExp\ge \Im\sigma\ge
-\SpectralGap$. We need one integration-by-parts argument using
$(\MorawetzVF_b, \LagrangeCorr_b, 0)$ as multipliers in a
neighborhood of the trapped set to prove the resolvent estimate for
$\abs*{\Im\sigma}\le \SpectralGap$, and one integration-by-parts
argument using $(\KillT, 0, 0)$ as multipliers in a neighborhood of
the trapped set to prove the resolvent estimate for
$\Im\sigma > \frac{\SpectralGap}{2}$.

Thus we also need the following analogue of Lemma
\ref{linear:lemma:ILED-full:Trapping-NonTrapping}, which allows us to glue
together the nontrapping Morawetz estimate and the trapping Morawetz
estimate when $\Im\sigma \ge \frac{\SpectralGap}{2}$.
\begin{lemma}
  \label{linear:lemma:ILED-full:Trapping-NonTrapping:case-2}
  Fix $\delta_0$. Then for $\varepsilon_{\TrappedSet}$ sufficiently
  small, there exists a choice of constants
  $\breve{c}, \check{c}, \mathring{c}$, such that for the vectorfields
  $\MorawetzOuterVF$ and $\MorawetzInnerNTVF$, and
  Lagrangian correctors
  $\LagrangeCorrOuter$ and $\LagrangeCorrInnerNT$ as
  defined above,
  \begin{equation}
    \label{linear:eq:ILED-full:Trapping-NonTrapping:main:2}
    \KCurrentIbPSym{\MorawetzSym^{\NotRedshift},\LagrangeCorrSym^{\NotRedshift}} 
    \ge -\delta_0|\zeta|^2,
  \end{equation}
  on $r\in (\breve{r}_{-}, \breve{R}_{+})$, where
  \begin{equation*}
    \MorawetzSym^{\NotRedshift}
    = \breve{\chi}^2\MorawetzOuterSym
    + \mathring{\chi}^2\mathring{\chi}_{\zeta}^2\MorawetzInnerNTSym
    + \mathring{\chi}^2\mathring{\chi}_{\zeta}^2 \sigma,
    \qquad \MorawetzOuterSym = \frac{\fOuter}{\Delta_b}H_{\rho_b^2\PrinSymb_b}r,
    \qquad \MorawetzInnerNTSym = \frac{\fInnerNT}{\Delta_b}H_{\rho_b^2\PrinSymb_b}r,
  \end{equation*}
  and
  \begin{equation*}
    \LagrangeCorrSym^{\NotRedshift}
    = \breve{\chi}^2\LagrangeCorrOuter
    + \mathring{\chi}^2\mathring{\chi}_{\zeta}^2\LagrangeCorrInnerNT. 
  \end{equation*}
\end{lemma}
The proof of Lemma \ref{linear:lemma:ILED-full:Trapping-NonTrapping:case-2}
is identical to that of \ref{linear:lemma:ILED-full:Trapping-NonTrapping}.
Since the gluing procedure is also identical for the two cases, we
only explicitly work out the procedure for the case where we use
$(\MorawetzVF_b, \LagrangeCorr_b, 0)$ to prove a resolvent estimate
for $\abs*{\Im\sigma}\le \SpectralGap$.

We are now ready to handle the proof of Proposition
\ref{linear:prop:ILED-combine:not-redshift}.
\begin{proof}[Proof of Proposition
  \ref{linear:prop:ILED-combine:not-redshift}]

  Define
  \begin{equation*}
    \MorawetzEnergy^{\NotRedshift}(\tStar)[h]
    = \breve{c}\breve{\MorawetzEnergy}(\tStar)[\breve{\chi}h]
    + \check{c}\widecheck{\MorawetzEnergy}(\tStar)[\mathring{\chi}\breve{\chi}_\zeta h]
    + \mathring{c}\mathring{\MorawetzEnergy}(\tStar)[\mathring{\chi}\mathring{\chi}_\zeta h],
  \end{equation*}
  where
  \begin{align*}
    \breve{\MorawetzEnergy}(\tStar)[\breve{\chi}h]
    &= \int_{\Sigma_{\tStar}}
      \JCurrent{\MorawetzOuterVF, \LagrangeCorrOuter, 0}[\breve{\chi}h]\cdot n_{\Sigma_{\tStar}},\\
    \check{\MorawetzEnergy}(\tStar)[\mathring{\chi}\breve{\chi}_\zeta h]
    &= \int_{\Sigma_{\tStar}}
      \JCurrent{\MorawetzInnerNTVF, \LagrangeCorrInnerNT, 0}[\mathring{\chi}\breve{\chi}_\zeta h] \cdot n_{\Sigma_{\tStar}},\\
    \mathring{\MorawetzEnergy}(\tStar)[\mathring{\chi}\mathring{\chi}_\zeta h]
    &= \int_{\Sigma_{\tStar}}
      \JCurrent{\MorawetzVF_b, \LagrangeCorr_b, 0}[\mathring{\chi}\mathring{\chi}_\zeta h] \cdot n_{\Sigma_{\tStar}}.
  \end{align*}
  Then, using the pseudo-differential modification of the divergence
  theorem in \eqref{linear:eq:ILED-near:combined-divergence-theorem}, we have
  then that up to lower-order terms,
  \begin{align}    
      &\left.\MorawetzEnergy^{\NotRedshift}(\tStar)[h]\right\vert_{\tStar=0}^{\tStar = \TStar}+\int_{\DomainOfIntegration}
      \breve{c} 
        \KCurrent{\MorawetzOuterVF, \LagrangeCorrOuter, 0}[\breve{\chi}{h}]
      + \check{c}\KCurrent{\MorawetzInnerNTVF, \LagrangeCorrInnerNT, 0}[\mathring{\chi}\breve{\chi}_{\zeta}{h}]
      + \mathring{c} \KCurrent{\MorawetzVF_{b_0}, \LagrangeCorr_{b_0}, 0}[\mathring{\chi}\mathring{\chi}_{\zeta}{h}]      
      + \mathring{c}a\KCurrentIbP{\widetilde{\MorawetzVF}, \tilde{\LagrangeCorr}, 0}[\mathring{\chi}\mathring{\chi}_{\zeta}{h}]\notag \\
      & - \breve{c}\Re\bangle*{\MorawetzOuterVF (\breve{\chi}h), \SubPOp_b[\breve{\chi}h]}_{L^2(\DomainOfIntegration)}
        - \check{c}\Re\bangle*{\MorawetzInnerNTVF (\mathring{\chi}\breve{\chi}_\zeta h), \SubPOp_b[\mathring{\chi}\breve{\chi}_\zeta h]}_{L^2(\DomainOfIntegration)}
        - \mathring{c}\Re\bangle*{\MorawetzVF_{b_0}(\mathring{\chi}\mathring{\chi}_\zeta h), \SubPConjOp_b(\mathring{\chi}\mathring{\chi}_\zeta h)}_{L^2(\DomainOfIntegration)}
      \notag \\
      ={} &
      -\breve{c}\Re\bangle*{\LinEinstein_{g_b}(\breve{\chi}{h}), \MorawetzOuterVF(\breve{\chi}{h})}_{L^2(\DomainOfIntegration)}
      - \check{c}\Re\bangle*{\LinEinstein_{g_b}(\mathring{\chi}\breve{\chi}_{\zeta}{h}), \MorawetzInnerNTVF(\mathring{\chi}\breve{\chi}_{\zeta}{h})}_{L^2(\DomainOfIntegration)}
      - \mathring{c}\Re\bangle{\LinEinstein_{g_b}(\mathring{\chi}\mathring{\chi}_{\zeta}{h}), \MorawetzVF_{b}(\mathring{\chi}\mathring{\chi}_{\zeta}{h})}_{L^2(\DomainOfIntegration)}. \label{linear:eq:ILED-combined:not-redshift:div-thm}
  \end{align}
  We will discuss neither the treatment of the boundary terms nor
  of the lower-order terms in detail here, having already provided a
  treatment of them in the preceding sections\footnote{In particular,
    notice that error terms arising from commuting derivatives with
    cutoff functions for the boundary terms can be controlled by
    reducing $\SpectralGap$ if necessary, and error terms arising from
    commuting derivatives with lower-order terms can be controlled by
    increasing the high-frequency threshold $C_0$ if necessary.}. The
  terms on the left-hand side generate positive coercive (and
  degenerate at $\TrappedSet_b$) terms. It suffices then to show that
  it is possible to choose the vectorfield multipliers, the Lagrangian
  correctors, and the real weights
  $\breve{c}, \check{c}, \mathring{c}$ so that the commutation errors
  arising from commuting the cutoffs with $\LinEinstein_{g_b}$ on the
  right-hand side are controlled by the ellipticity of the left-hand
  side. Recall that since we are only interested in obtaining a
  high-frequency Morawetz estimate and $\LinEinstein_{g_b}$ is
  strongly hyperbolic, it is sufficient to consider the commutation of
  the cutoffs with $\ScalarWaveOp[g_b]$ (see Remark
  \ref{linear:rmk:ILED-combine:important-terms}). To be more precise, we only
  need to show that there exists constants
  $\breve{c}, \check{c}, \mathring{c}$ such that for ${h}$ spatially
  supported in $\Sigma^{\NotRedshift}$,
  \begin{equation*}
    \begin{split}
      &\bangle*{\squareBrace*{\ScalarWaveOp[g_b],
          \breve{c}\breve{\chi}}{h}, \MorawetzOuterVF(\breve{\chi}{h})}_{L^2(\DomainOfIntegration)}
      + \bangle*{\squareBrace*{\rho_b^2\ScalarWaveOp[g_b],
          \check{c}\mathring{\chi}\breve{\chi}_{\zeta}}{h}, \rho_b^{-2}\MorawetzInnerNTVF(\mathring{\chi}\breve{\chi}_{\zeta}{h})}_{L^2(\DomainOfIntegration)}\\
      &+\bangle*{\squareBrace*{\ScalarWaveOp[g_b],
          \mathring{c}\mathring{\chi}\mathring{\chi}_{\zeta}}{h}, \MorawetzVF_b(\mathring{c}\mathring{\chi}\mathring{\chi}_{\zeta}{h})}_{L^2(\DomainOfIntegration)}\\
      \le{}& \delta\left( \int_{\DomainOfIntegration}
      \breve{c} a\Re\KCurrent{\MorawetzOuterVF, \LagrangeCorrOuter, 0}[\breve{\chi}{h}]
      + \check{c}\KCurrent{\MorawetzInnerNTVF, \LagrangeCorrInnerNT, 0}[\mathring{\chi}\breve{\chi}_{\zeta}{h}]
      + \mathring{c} \KCurrent{\MorawetzVF_{b_0}, \LagrangeCorr_{b_0}, 0}[\mathring{\chi}\mathring{\chi}_{\zeta}{h}]      
      + \mathring{c}\KCurrentIbP{\widetilde{\MorawetzVF}, \tilde{\LagrangeCorr}}[\mathring{\chi}\mathring{\chi}_{\zeta}{h}]\right)\\
    &+\delta\left(\breve{c}\Re\bangle*{\MorawetzOuterVF (\breve{\chi}h), \SubPOp_b[\breve{\chi}h]}_{L^2(\DomainOfIntegration)}
      + \check{c}\Re\bangle*{\MorawetzInnerNTVF (\mathring{\chi}\breve{\chi}_\zeta h), \SubPOp_b[\mathring{\chi}\breve{\chi}_\zeta h]}_{L^2(\DomainOfIntegration)}\right)
    \end{split}
  \end{equation*}
  for some $0<\delta<1$, but this is exactly the content of Lemma
  \ref{linear:lemma:ILED-full:Trapping-NonTrapping}, so we conclude after a
  simple application of Cauchy-Schwarz. 
\end{proof}

Next, we move onto handling the intersection of the redshift and the
non-trapping regions.

\begin{prop}
  \label{linear:prop:ILED-combine:not-trapping}
  Let $g_b$ be a fixed slowly-rotating \KdS{} background, and let
  \begin{equation*}
    \DomainOfIntegration:= \Real^+_{\tStar}\times\Sigma^{\NotTrapping},\qquad
    \Sigma^{\NotTrapping}:= \curlyBrace*{(r,\omega)\in\Sigma: r\not\in(\breve{r}_+, \breve{R}_-)}.
  \end{equation*}
  Then for $k>k_0$, where $k_0$ is the threshold regularity level in
  \eqref{linear:eq:threshold-reg-def}, there exist constants
  $\SpectralGap, C_0>0$ such that for $u$ compactly supported in
  $\Sigma^{\NotTrapping}$,
  \begin{equation}
    \label{linear:eq:ILED-combine:not-trapping}
    \norm{u}_{\CombinedHk{k}(\Sigma^{\NotTrapping})}
    \lesssim \norm*{\widehat{\LinEinstein}_{g_b}(\sigma)u}_{\InducedHk{k-1}_\sigma(\Sigma^{\NotTrapping})}.
  \end{equation}
\end{prop}

Gluing together the non-trapping Morawetz estimate and the redshift
estimate is less nuanced than gluing together the trapping
Morawetz estimate and the non-trapping Morawetz estimate. We no longer
have frequency-dependent multipliers and thus can glue the two
estimates together using physical space methods. In addition, we only
have one cutoff function to handle instead of two. The main difficulty
with gluing the non-trapping Morawetz estimate and the redshift
estimate together then is that the Morawetz estimate degenerates in
derivatives transverse to the horizons at the horizons, but this is
exactly in the region where the redshift estimate is coercive and
positive definite. We first let
\begin{equation}
  \label{linear:eq:ILED-combine:LinEinstein-k-commute}
  \LinEinstein_{g_b}^{(k)}
  := \ScalarWaveOp[g_b]
  + \SubPOp_{b}^{(k)}
  + \PotentialOp_{b}^{(k)}
\end{equation}
denote the linear operator that
results from commuting $k$ times $\LinEinstein_{g_b}$ with
$\{\RedShiftK_i\}$ as in Theorem
\ref{linear:thm:redshift-commutation:main}. 

To proceed with the resolvent estimate, we first prove the desired
positivity of the bulk terms.

\begin{lemma}
  \label{linear:lemma:ILED-combine:nontrapping-redshift:aux-lemma}
  Let $g$ be a fixed slowly-rotating \KdS{} background, let
  $\MorawetzOuterVF, \LagrangeCorrOuter$ be as constructed in
  \eqref{linear:eq:ILED-nontrapping:XOuter-fOuter-def},
  \eqref{linear:eq:ILED-nontrapping:q-Outer-def}, and fix some
  $\dot{c}>0$. Then for ${h}$ supported outside $\TrappingNbhd$, there
  exists $\COuter_\star>0$ such that for $\COuter>\COuter_\star$, and
  $\varepsilon_{\StaticRegionWithExtension}$ sufficiently small, and
  $k>k_0$, where $k_0$ is the threshold regularity level defined in
  \eqref{linear:eq:threshold-reg-def}, such that on
  $\StaticRegionWithExtension\bigcap\{r<r_0, r>R_0\}$,
  \begin{equation}
    \KCurrent{\MorawetzOuterVF, \LagrangeCorrOuter, 0}[{h}]
    + \dot{c}\KCurrent{\RedShiftN, 0, 0}[{h}]
    - \Re \left[\SubPOp_b^{(k)}[{h}]\cdot (\MorawetzOuterVF + \dot{c}\RedShiftN)\overline{{h}}\right]
    - \Re\left[\dot{c}\left[\ScalarWaveOp[g_b], \dot{\chi}\right]h\cdot \dot{\chi}\RedShiftN \overline{h}\right]
    \gtrsim \abs*{\nabla h}^2 - C\abs*{h}^2,
  \end{equation}
  where $\SubPOp_b^{(k)}$ is the subprincipal operator of
  $\LinEinstein_{g_b}^{(k)}$, as in \eqref{linear:eq:ILED-combine:LinEinstein-k-commute}.
\end{lemma}
\begin{proof}
  We decompose
  \begin{equation*}
    \SubPOp_b^{(k)} = \SubPOp_{b, \Horizon}^{(k)}\HprVF + \widetilde{\SubPOp}_b^{(k)},
  \end{equation*}
  where $\widetilde{\SubPOp}_b^{(k)}$ is tangent to both
  $\EventHorizonFuture$ and $\CosmologicalHorizonFuture$, and we know
  that for $k>k_0$,
  $\overline{\xi}\cdot \SubPOp_{b, \Horizon}^{(k)}\xi>0$.
  
  Similarly, we decompose for $\Horizon = \EventHorizonFuture,
  \CosmologicalHorizonFuture$, 
  \begin{equation*}
    \evalAt*{\RedShiftN}_{\Horizon}
    = -\HprVF + \widetilde{\RedShiftN},
  \end{equation*}
  where $\widetilde{\RedShiftN}$ is tangent to both
  $\EventHorizonFuture$ and $\CosmologicalHorizonFuture$.

  Observe that for $k>k_0$, we have by the choice of $\RedShiftN$ in
  Proposition \ref{linear:prop:redshift:N-construction} that there exists
  some $\varepsilon_{\Horizon}>0$ such that  
  \begin{equation*}
    \KCurrent{\RedShiftN, 0, 0}[h]
    - \Re\left[\SubPOp_b^{(k)}[h]\cdot \RedShiftN \overline{h}\right]
    > \varepsilon_{\Horizon}\abs*{\nabla h}^2. 
  \end{equation*}  
  We also assume that we have already chosen $\COuter$ large enough
  that
  \begin{equation*}
    \KCurrent{\MorawetzOuterVF,\LagrangeCorrOuter, 0}[h]
    > 4\dot{c} \abs*{\p_r\dot{\chi}\Delta  \HprVF h \cdot \widetilde{\RedShiftN} \overline{h}}
    + 4 \abs*{\widetilde{\SubPOp}_{b}^{(k)}[h]\cdot \MorawetzOuterVF \overline{h}}.
  \end{equation*}
  We see that it then suffices to choose $\COuter$, $\dot{c}$, and
  $\delta_{\Horizon}$ such that
  \begin{align}       
    \KCurrent{\MorawetzOuterVF, \LagrangeCorrOuter, 0}[h]
    > - 4 \dot{c}\p_r\dot{\chi}\Delta_b\abs*{\HprVF h}^2,
    \qquad &  \supp\p_r\dot{\chi}, \label{linear:eq:ILED-full:not-trapping-gluing-conditions:1}\\
    \KCurrent{\MorawetzOuterVF, \LagrangeCorrOuter, 0}[h]
    + \dot{c}\varepsilon_{\Horizon}\abs*{\nabla h}^2 > 4 e^{\COuter(r-3M)^2}\abs*{\Delta_b (r-3M)\SubPOp_{b,\Horizon}^{(k)}}\abs*{\HprVF h}^2,
    \qquad & \{r:\dot{\chi}=1\}.   \label{linear:eq:ILED-full:not-trapping-gluing-conditions:2}  
  \end{align}
  
  We first consider
  \eqref{linear:eq:ILED-full:not-trapping-gluing-conditions:2}.  If
  $\Delta_b > \frac{4}{\COuter}$, then
  \eqref{linear:eq:ILED-full:not-trapping-gluing-conditions:2} follows
  directly from the form of
  $\KCurrent{\MorawetzOuterVF, \LagrangeCorrOuter, 0}[h]$. On the
  other hand, if $\Delta_b < \frac{4}{\COuter}$, then we choose
  $\dot{c}$ such that
  \begin{equation}
    \label{linear:eq:ILED-full:not-trapping:c-dot-choice}
    \dot{c} > e^{\COuter(r-3M)^2}\abs*{\frac{(r-3M)\SubPOp_{b,\Horizon}^{(k)}}{\varepsilon_{\Horizon} \COuter}},
  \end{equation}
  so that \eqref{linear:eq:ILED-full:not-trapping-gluing-conditions:2} is
  verified.

  We now turn our attention to
  \eqref{linear:eq:ILED-full:not-trapping-gluing-conditions:1}. Recalling the
  form of $\KCurrent{\MorawetzOuterVF, \LagrangeCorrOuter, 0}[h]$, and
  using the fact that our choice of $\dot{c}$ satisfies
  \eqref{linear:eq:ILED-full:not-trapping:c-dot-choice}, we have that
  \eqref{linear:eq:ILED-full:not-trapping-gluing-conditions:1} is satisfied if
  \begin{equation*}
    \Delta_b > 16\abs*{\frac{\SubPOp_{b,\Horizon}^{(k)} \p_r\dot{\chi} }{\varepsilon_{\Horizon}\COuter^2(r-3M)}},\qquad \supp\p_r\dot{\chi}. 
  \end{equation*}
  But from the definition of $\dot{\chi}$ in
  \eqref{linear:eq:ILED-combine:cutoff-function-list}, we see that we have
  the bound
  \begin{equation*}
    \abs*{\p_r\dot{\chi}} < \delta_{\Horizon}^{-2}. 
  \end{equation*}
  For $\delta_{\Horizon}$ sufficiently small, we have that
  \begin{equation*}
    \sup_{\supp \p_r\dot\chi}\Delta_b > \delta_{\Horizon}^2. 
  \end{equation*}
  Then we take $\COuter$ sufficiently large so that 
  \begin{equation*}
    \COuter^2 > 16\abs*{\frac{\SubPOp_{b,\Horizon}^{(k)}}{\varepsilon_{\Horizon}(r-3M)\delta_{\Horizon}^4}},  
  \end{equation*}
  so that \eqref{linear:eq:ILED-full:not-trapping-gluing-conditions:1} is
  verified. 
\end{proof}

We can now use the positivity in Lemma
\ref{linear:lemma:ILED-combine:nontrapping-redshift:aux-lemma} to glue
together the red-shift and nontrapping Morawetz estimates for the
rescaled problem. 

\begin{lemma}
  \label{linear:lemma:ILED-nontrapping:extension}
  For $\MorawetzOuterVF$ and $\LagrangeCorrOuter$ as constructed in
  Lemma \ref{linear:lemma:ILED-nontrapping:bulk-positivity}, there exists an 
  auxiliary one-form $\breve{\ZeroOCorr}$ and a
  $\dot{c}>0$ such that the following properties
  hold. 
  \begin{enumerate}
  \item On the hypersurfaces $\EventHorizonFuture_-,
    \CosmologicalHorizonFuture_+$, the  boundary flux has the
    following properties:
    \begin{equation}
      \label{linear:eq:ILED-nontrapping:extension:boundary-flux}
      \begin{split}
        \dot{c}\int_{\EventHorizonFuture_-}\JCurrent{\RedShiftN, 0, 0}[h]\cdot n_{\EventHorizonFuture_-}
      + \int_{\EventHorizonFuture_-}\JCurrent{\MorawetzOuterVF, \LagrangeCorrOuter, \breve{\ZeroOCorr}}[h]\cdot n_{\EventHorizonFuture_-}&\ge 0,\\
      \dot{c}\int_{\CosmologicalHorizonFuture_+}\JCurrent{\RedShiftN, 0, 0}[h]\cdot n_{\CosmologicalHorizonFuture_+}
      +\int_{\CosmologicalHorizonFuture_+}\JCurrent{\MorawetzOuterVF, \LagrangeCorrOuter, \breve{\ZeroOCorr}}[h]\cdot n_{\CosmologicalHorizonFuture_+}&\ge 0.
      \end{split}      
    \end{equation}
  \item For $h$ supported away from $r=3M$, the inequality in
    \eqref{linear:eq:ILED-nontrapping:bulk-positivity} continues to hold with
    $\KCurrent{\MorawetzOuterVF, \LagrangeCorrOuter,
      \breve{\ZeroOCorr}}[h]$ in place of
    $\KCurrent{\MorawetzOuterVF, \LagrangeCorrOuter, 0}[h]$. For
    $\varepsilon_{\StaticRegionWithExtension}$ sufficiently small, we
    have moreover that on all of $\Sigma$ as defined in
    \eqref{linear:eq:Sigma-def}, the following relation holds up to
    zero-order terms
    \begin{align}
      \abs*{\nabla h}^2\lesssim{}&
      \dot{c}\KCurrent{\RedShiftN, 0, 0}[\dot{\chi}h]
      + \KCurrent{\MorawetzOuterVF, \LagrangeCorrOuter,\breve{\ZeroOCorr}}[h]\notag \\
                                 &- \Re\left[
                                   \SubPOp_b^{(k)}h\cdot
                                   \left(\dot{c}\RedShiftN\circ\dot{\chi}+ \MorawetzOuterVF\right)\overline{h}\right]
      + \Re\left[\dot{c}\left[\ScalarWaveOp[g_b], \dot{\chi}\right]h\cdot\RedShiftN(\dot{\chi}\overline{h})\right], \label{linear:eq:ILED-nontrapping:extension:bulk}      
    \end{align}
    where $\SubPOp_b^{(k)}$ is the subprincipal operator of
    $\LinEinstein_{g_b}^{(k)}$ as defined in
    \eqref{linear:eq:ILED-combine:LinEinstein-k-commute}, $k>k_0$ the
    threshold regularity level as defined in
    \eqref{linear:eq:threshold-reg-def}, and $\dot{\chi}$ is the cutoff
    localizing to the redshift region as defined in
    \eqref{linear:eq:ILED-combine:cutoff-function-definitions}.
  \end{enumerate}
\end{lemma}

\begin{proof}
  We first prove \eqref{linear:eq:ILED-nontrapping:extension:bulk}. First,
  observe that it suffices to
  show \eqref{linear:eq:ILED-nontrapping:extension:bulk} with the choice
  $\breve{\ZeroOCorr}=0$. Observe that 
  \begin{equation}    
    \KCurrent{\MorawetzOuterVF, \LagrangeCorrOuter,\breve{\ZeroOCorr}}[h]      
    ={} \KCurrent{\MorawetzOuterVF, \LagrangeCorrOuter,0}[h]
    + \Re\left[\breve{\ZeroOCorr}_\alpha h\cdot\p^\alpha\overline{h}\right]
    + \frac{1}{2}\left(\nabla\cdot \breve{\ZeroOCorr}\right) \abs*{h}^2.
  \end{equation}
  As a result, we have that for any choice of one-form
  $\breve{\ZeroOCorr}$ and $\epsilon>0$, there exists some
  $C(\epsilon)>0$ such that
  \begin{equation*}
    \KCurrent{\MorawetzOuterVF, \LagrangeCorrOuter,\breve{\ZeroOCorr}}[h] \le  \KCurrent{\MorawetzOuterVF, \LagrangeCorrOuter, 0}[h] + \epsilon\abs*{\nabla h}^2 + C(\epsilon)\abs*{h}^2.
  \end{equation*}
  But then \eqref{linear:eq:ILED-nontrapping:extension:bulk} with the choice
  choice $\breve{\ZeroOCorr}=0$ is exactly the statement of Lemma
  \ref{linear:lemma:ILED-combine:nontrapping-redshift:aux-lemma}.   

  We emphasize that the
  proof of \eqref{linear:eq:ILED-nontrapping:extension:bulk} and in
  particular the choice of $\dot{c}$ is independent
  of $\breve{\ZeroOCorr}$. We now consider $\dot{c}$
  fixed and show that there exists some choice of $\breve{\ZeroOCorr}$
  such that \eqref{linear:eq:ILED-nontrapping:extension:boundary-flux}
  holds. We will just prove the statement on
  $\EventHorizonFuture_-$. A similar argument will suffice to show the
  conclusion on $\CosmologicalHorizonFuture_+$.

  Recall that $\RedShiftN$ is uniformly timelike on
  $\StaticRegionWithExtension$. Since $n_{\EventHorizonFuture_-}$ is
  also timelike, we have that on $\EventHorizonFuture_-$, 
  \begin{equation*}
    \JCurrent{\RedShiftN, 0, 0}[h]\cdot n_{\EventHorizonFuture_-} \gtrsim \abs*{\nabla h}^2.
  \end{equation*}
  Now observe that
  \begin{equation*}
    \int_{\EventHorizonFuture_-}\JCurrent{\MorawetzOuterVF, \LagrangeCorrOuter, \breve{\ZeroOCorr}}[h]\cdot n_{\EventHorizonFuture_-}
    = \int_{\EventHorizonFuture_-}\JCurrent{\MorawetzOuterVF,0, 0}[h]\cdot n_{\EventHorizonFuture_-}
    +\int_{\EventHorizonFuture_-}\JCurrent{0, \LagrangeCorrOuter, 0}[h]\cdot n_{\EventHorizonFuture_-}
    +\int_{\EventHorizonFuture_-}\JCurrent{0, 0, \breve{\ZeroOCorr}}[h]\cdot n_{\EventHorizonFuture_-}.
  \end{equation*}
  Since $\Delta_b\HprVF$ is timelike beyond the horizons (and
  vanishes at the horizons),  we have that
  \begin{equation*}
    \int_{\EventHorizonFuture_-}\JCurrent{\MorawetzOuterVF,0, 0}[h]\cdot n_{\EventHorizonFuture_-} \ge 0.
  \end{equation*}
  Recall from the definition of $\JCurrent{X,q,m}[h]$ in \eqref{linear:eq:J-K-currents:def} that
  \begin{equation}
    \label{linear:eq:ILED-combine:not-trapping:J-0-q-m}
    \JCurrent{0, \LagrangeCorrOuter, \breve{\ZeroOCorr}}[h]\cdot n_{\EventHorizonFuture_-}
    = \Re\left[\LagrangeCorrOuter h\cdot n_{\EventHorizonFuture_-}\overline{h}\right]
    - \frac{1}{2}n_{\EventHorizonFuture_-} \LagrangeCorrOuter \abs*{h}^2
    + \frac{1}{2}g(n_{\EventHorizonFuture_-}, \breve{\ZeroOCorr}) \abs*{h}^2.
  \end{equation}
  If we pick for instance $\breve{\ZeroOCorr}_\alpha =
  -\COuter'\RedShiftN_\alpha$ then
  \begin{equation*}
    \evalAt*{g(n_{\EventHorizonFuture_-}, \breve{\ZeroOCorr})}_{\EventHorizonFuture_-} > 0.
  \end{equation*}
  Picking $\COuter'$ sufficiently large, it is clear that the second
  term on the right-hand side of
  \eqref{linear:eq:ILED-combine:not-trapping:J-0-q-m} will be controlled by
  the last term on the right-hand side of
  \eqref{linear:eq:ILED-combine:not-trapping:J-0-q-m}. Moreover, for
  $\COuter'$ sufficiently large, we have, using Cauchy-Schwarz, that
  in fact the first term on the right-hand side of
  \eqref{linear:eq:ILED-combine:not-trapping:J-0-q-m} can be controlled
  \begin{equation*}
    \abs*{\LagrangeCorrOuter h \cdot n_{\EventHorizonFuture_-}\overline{h} }
    \lesssim \frac{1}{4}g(n_{\EventHorizonFuture_-}, \breve{\ZeroOCorr})\abs*{h}^2
    + \JCurrent{\MorawetzOuterVF, 0, 0}[h]\cdot n_{\EventHorizonFuture_-}
    + \dot{c}\JCurrent{\RedShiftN, 0, 0 }[h]\cdot n_{\EventHorizonFuture_-}
  \end{equation*}
  on $\EventHorizonFuture_-$.
  This concludes the proof of Lemma
  \ref{linear:lemma:ILED-nontrapping:extension}.   
\end{proof}

We are now ready to prove Proposition
\ref{linear:prop:ILED-combine:not-trapping}. 
\begin{proof}[Proof of Proposition
  \ref{linear:prop:ILED-combine:not-trapping}.]

  Define
  \begin{equation*}
    \MorawetzEnergy^{\NotTrapping}(\tStar)[h]
    = \breve{\MorawetzEnergy}(\tStar)[\breve{\chi}h]
    + \dot{c}\dot{\MorawetzEnergy}(\tStar)[\dot{\chi}h]. 
  \end{equation*}
  Then, applying the divergence theorem in
  \eqref{linear:eq:div-thm:spacetime} 
  and using
  the control of the boundary terms along $\EventHorizonFuture_-$,
  $\CosmologicalHorizonFuture_+$ present in
  \eqref{linear:eq:ILED-nontrapping:extension:boundary-flux}, we have that up
  to lower order terms
  \begin{align}
    &\left.\MorawetzEnergy^{\NotTrapping}(\tStar)[{h}]\right\vert_{\tStar=0}^{\tStar=\TStar}
      - \Re\bangle*{\MorawetzOuterVF h, \SubPOp_b[{h}]}_{L^2(\DomainOfIntegration)}
    -\dot{c}\Re\bangle*{\RedShiftN (\dot{\chi}{h}), \SubPOp_b[\dot{\chi}{h}]}_{L^2(\DomainOfIntegration)}
      + \int_{\DomainOfIntegration}\KCurrent{\MorawetzOuterVF, \LagrangeCorrOuter, \breve{\ZeroOCorr}}[{h}]
    +\dot{c}\KCurrent{\RedShiftN, 0, 0}[\dot{\chi}{h}]\notag \\    
    \le{}& \abs*{\Re\bangle*{\LinEinstein_{g_b}{h}, \MorawetzOuterVF {h}}_{\LTwo(\DomainOfIntegration)}}
    + \abs*{ \dot{c}\Re\bangle*{\LinEinstein_{g_b}(\dot{\chi}{h}), \RedShiftN(\dot{\chi}{h})}_{\LTwo(\DomainOfIntegration)}}.\label{linear:eq:ILED-nontrapping:not-trapping:div-thm}
  \end{align}
  
  To prove the resolvent estimate in
  \eqref{linear:eq:ILED-combine:not-trapping}, it suffices to differentiate
  \eqref{linear:eq:ILED-nontrapping:not-trapping:div-thm} in $\p_{\tStar}$ and
  multiply both sides by $e^{2\Im\sigma\tStar}$. Higher-order
  estimates then follow as before by commuting through with
  $\RedShiftK_i, \KillT$, and using elliptic estimates.
\end{proof}

We are now ready to prove Theorem \ref{linear:thm:ILED-near:main} for
$\abs*{\Im\sigma}\le \SpectralGap$. 
\begin{proof}[Proof of Theorem \ref{linear:thm:ILED-near:main}  for
  $\abs*{\Im\sigma}\le \SpectralGap$.]
  The theorem follows directly from
  \eqref{linear:eq:ILED-nontrapping:not-trapping:div-thm} and
  \eqref{linear:eq:ILED-combined:not-redshift:div-thm}, using the bulk
  positivity in Lemmas \ref{linear:lemma:ILED-full:Trapping-NonTrapping} and
  \ref{linear:lemma:ILED-nontrapping:extension}, and controlling boundary
  terms and lower-order terms as previously done.
\end{proof}

The proof of Theorem \ref{linear:thm:ILED-near:main} for $\abs*{\Im\sigma} >
\frac{\SpectralGap}{2}$ follows exactly like the proof of Theorem
\ref{linear:thm:ILED-near:main} for $\abs*{\Im\sigma}\le \SpectralGap$, so we
omit it here.

\section{Exponential decay up to compact perturbation}
\label{linear:sec:asymptotic-expansion}

The main goal of this section will be to prove Theorems
\ref{linear:thm:meromorphic:main-A} and
\ref{linear:thm:resolvent-estimate:inf-gen}. Together, these two theorems
imply that there are only finitely many non-decaying
$\LSolHk{k}$-quasinormal mode solutions for
$\LinEinstein=\LinEinstein_{g_b}$, giving exponential decay up to
compact perturbation, and in particular, the asymptotic expansion in
Corollary \ref{linear:coro:asymptotic-expansion}.

\subsection{Fredholm alternative for $\InfGen$ (Proof of Theorem
  \ref{linear:thm:meromorphic:main-A})}
\label{linear:sec:meromorphic-continuation}

In this section, we will prove Theorem \ref{linear:thm:meromorphic:main-A}.
We do so by analyzing the invertibility of the Laplace-transformed
operator $\widehat{\LinEinstein}(\sigma)$.  By using the Killing
energy estimate, the redshift energy estimate, and commuting with the
redshift vectorfield, we will be able to show that for sufficiently
slowly-rotating \KdS{} metrics, there exists some $\gamma$ such that
$(\widehat{\LinEinstein}(\sigma) - \gamma)^{-1}$ is a well-defined,
compact operator in the half-plane
\begin{equation}
  \label{linear:eq:meromophic:omega-range}
  \curlyBrace*{\sigma\in\Complex:\Im\sigma>
    \frac{1}{2}\max_{\Horizon=\EventHorizonFuture, \CosmologicalHorizonFuture}\left(
      \SHorizonControl{\LinEinstein}[\Horizon]
      -\left(2k + \frac{1}{2}\right)\SurfaceGravity_\Horizon \right) }.
\end{equation}
An appeal to the analytic Fredholm theorem then allows us to derive
the equivalent of Theorem \ref{linear:thm:meromorphic:main-A} for the
Laplace-transformed operator, which directly implies Theorem
\ref{linear:thm:meromorphic:main-A}, despite lacking compactness for
$(\InfGen-\sigma)^{-1}$. This follows closely the approach taken by
Warnick in proving an equivalent result on asymptotically
Schwarzschild anti-de Sitter spacetimes in
\cite{warnick_quasinormal_2015}.

In what follows, we first prove estimates for
$\widehat{\LinEinstein}(\sigma)$ with domain
$D^1(\widehat{\LinEinstein}(\sigma))$, and then for higher regularity
domains.

\subsubsection{Injectivity}

We first show that $(\widehat{\LinEinstein}(\sigma) - \gamma)^{-1}$
with domain $D^1(\widehat{\LinEinstein}(\sigma))$ is injective. 
\begin{theorem}
  \label{linear:thm:meromorphic:injective}
  For all sufficiently
  slowly-rotating \KdS{} metrics $g_b$ where $b=(M, a)$ the following
  holds. For a fixed compact domain
  \begin{equation*}
    \Omega\subset\curlyBrace*{\sigma\in \Complex: \Im\sigma >
      \frac{1}{2} \max_{\Horizon=\EventHorizonFuture, \CosmologicalHorizonFuture}\left(\SHorizonControl{\LinEinstein}[\Horizon] -\frac{1}{2}\SurfaceGravity_{\Horizon} \right)},
  \end{equation*}
  there exists some $\gamma_1$ such that
  $\widehat{\LinEinstein}(\sigma) - \gamma$ is injective for any
  $\gamma>\gamma_1$ and
  $D^1(\widehat{\LinEinstein}(\sigma))\subset \InducedHk{1}(\Sigma)$,
  for any $\sigma\in\Omega$, and that furthermore we have the estimate
  \begin{equation}
    \label{linear:eq:meromorphic:injective:main-eqn}
    \norm{u}_{\InducedHk{1}(\Sigma)} \lesssim \norm*{(\widehat{\LinEinstein}(\sigma) - \gamma )u}_{\InducedLTwo(\Sigma)}.
  \end{equation}
\end{theorem}
\begin{proof}
  It suffices to prove \eqref{linear:eq:meromorphic:injective:main-eqn} for
  smooth $u$. The rest of the conclusions then follow by a density
  argument.
  
  We first define $P_2, P_1$ as in \eqref{linear:eq:P-i-def} such that $P_i$
  is a bounded differential operators of order $i$ such that
  \begin{equation*}
    \LinEinstein u =  \frac{1}{\GInvdtdt} D_{\tStar}^2 u  +  P_1D_{\tStar} u + P_2 u.
  \end{equation*}
  We first apply the $\TFixer$-energy estimate in part
  \ref{linear:item:Killing-estimate-with-gamma:three} of Corollary
  \ref{linear:coro:Killing-estimate-with-gamma} to the function
  $e^{-\ImagUnit(\sigma + c)\tStar} u(x)$ and multiply both sides by
  $e^{2\Im c \tStar}$ to see that
  \begin{equation*}
    \begin{split}
      2\Im(c+\sigma)\EnergyKill_\gamma(\tStar)[e^{-\ImagUnit\sigma \tStar}u]
      \le{}&
      \epsilon \left(\norm*{e^{-\ImagUnit\sigma \tStar} u}^2_{\InducedHk{1}(\widetilde{\Sigma})}
        + \norm*{(\LinEinstein - \gamma) (e^{-\ImagUnit\sigma \tStar} u)}_{\InducedLTwo(\widetilde{\Sigma})}^2\right)\\
        &+ \epsilon \norm*{\left(cP_1+2c\GInvdtdt^{-1}+ c^2\GInvdtdt^{-1}\right)e^{-\ImagUnit\sigma \tStar} u}_{\InducedLTwo(\widetilde{\Sigma})}^2 
      + C(\epsilon)
      |c+\sigma|^2\norm*{e^{-\ImagUnit\sigma\tStar}u}_{\InducedLTwo(\widetilde{\Sigma})}^2\\
      &+ a C(\epsilon)\norm*{ e^{-\ImagUnit\sigma\tStar}u}_{\InducedHk{1}(\widetilde{\Sigma})}^2
      + a\gamma C\norm*{e^{-\ImagUnit\sigma\tStar}u}_{\InducedLTwo(\widetilde{\Sigma})}^2,
    \end{split}
  \end{equation*}
  where we recall $\widetilde{\Sigma}$ is as constructed in
  \eqref{linear:eq:Sigma-tilde:def}.  Now, we choose $c$ such that
  $\Im(c+\sigma)>0$ for all $\sigma\in\Omega$ 
  so that the left-hand side of the inequality is positive. We
  emphasize that this choice of $c$ very much depends on $\Omega$.

  Next, recall that
  $P_1: \InducedHk{1}(\Sigma)\to \InducedLTwo(\Sigma)$ is a bounded
  operator. Thus, given $\delta>0$, there exist constants
  $C_{\TFixer}(\delta), C_{\TFixer}$ depending on $c$, $\Omega$ and
  $P_1$, but independent of $\gamma$ such that for all $\sigma\in \Omega$,
  \begin{align}
    \EnergyKill_\gamma(\tStar)[e^{-\ImagUnit\sigma \tStar}u]
    \le{}& \delta\left(\norm*{e^{-\ImagUnit\sigma \tStar}u}^2_{\InducedHk{1}(\widetilde{\Sigma})}
           + \norm*{(\LinEinstein - \gamma)(e^{-\ImagUnit\sigma \tStar} u)}^2_{\InducedLTwo(\widetilde{\Sigma})}\right)
           + C_{\TFixer}(\delta)\norm*{e^{-\ImagUnit\sigma \tStar} u}_{\InducedLTwo(\widetilde{\Sigma})}^2\notag \\
         &+ aC_{\TFixer}(\delta)\norm*{ e^{-\ImagUnit\sigma\tStar}u}_{\InducedHk{1}(\widetilde{\Sigma})}^2
           + aC_{\TFixer}\gamma \norm*{e^{-\ImagUnit\sigma\tStar} u}^2_{\InducedLTwo(\widetilde{\Sigma})}\label{linear:eq:meromorphic:injective:Killing-first} .           
  \end{align}
  We now apply the redshift estimate in part
  \ref{linear:item:RedShiftEstimate-with-gamma:two} of Corollary
  \ref{linear:coro:redshift-energy-estimate-with-gamma} with
  \begin{equation*}
    \varepsilon_{\RedShiftN} < \frac{1}{4}\max_{\Horizon= \EventHorizonFuture, \CosmologicalHorizonFuture}\SurfaceGravity_{\Horizon}
  \end{equation*}
  to $e^{-\ImagUnit\sigma \tStar}u$ to deduce that
  \begin{align}
    &\left(2\Im\sigma - \max_{\Horizon=\EventHorizonFuture, \CosmologicalHorizonFuture}\left(\SHorizonControl{\LinEinstein}[\Horizon] - \SurfaceGravity_{\Horizon} + \varepsilon_{\RedShiftN} + \epsilon_0\right)\right)
      \RedShiftEnergy_\gamma(\tStar)[e^{-\ImagUnit\sigma\tStar}u] \notag \\ 
    \le{}& C^\RedShiftN(\epsilon_0)\left( \norm*{(\LinEinstein - \gamma)(e^{-\ImagUnit\sigma\tStar}u)}_{\InducedLTwo(\Sigma)}^2
           + \norm*{e^{-\ImagUnit\sigma\tStar}u}_{\InducedLTwo(\Sigma)}^2
           + \EnergyKill_\gamma(\tStar)[e^{-\ImagUnit\sigma\tStar} u]\right)       \label{linear:eq:meromorphic:injective:redshift-first},
  \end{align}
  where
  $\epsilon_0 = \frac{1}{4}\max_{\Horizon= \EventHorizonFuture,
    \CosmologicalHorizonFuture}\SurfaceGravity_{\Horizon}$ is chosen
  so that no matter our choice of $\Omega$,
  \begin{equation*}
    2\Im\sigma - \max_{\Horizon=\EventHorizonFuture, \CosmologicalHorizonFuture}\left(\SHorizonControl{\LinEinstein}[\Horizon] - \SurfaceGravity_{\Horizon}\right) - \varepsilon_{\RedShiftN} - \epsilon_0 > 0.
  \end{equation*}
  We will now apply
  (\ref{linear:eq:meromorphic:injective:Killing-first}) to control the
  $\TFixer$-energy norm on the right-hand side of
  (\ref{linear:eq:meromorphic:injective:redshift-first}). To do this,
  we realize that there exists $\delta_0$ sufficiently small, and
  independent of our choice of $\Omega$ such that for any
  $\delta\le\delta_0$,
  \begin{equation*}
    \delta C_{\RedShiftN}(\epsilon_0)
    \norm*{e^{-\ImagUnit\sigma\tStar}u}_{\InducedHk{1}(\Sigma)}^2
    <\epsilon_0 \RedShiftEnergy(\tStar)[e^{-\ImagUnit\sigma\tStar}u]. 
  \end{equation*}

  Now let $C_{\TFixer}(\delta_0)$ be the large constant such that
  (\ref{linear:eq:meromorphic:injective:Killing-first}) holds with the choice
  $\delta=\delta_0$. We can then find $a_0$ sufficiently small such
  that for all $a<a_0$,
  \begin{equation}
    \label{linear:eq:mermorphic:injective:a-condition}
    aC_{\RedShiftN}(\epsilon_0)C_{\TFixer}(\delta_0)\norm*{
      e^{-\ImagUnit\sigma\tStar}u}_{\InducedHk{1}(\Sigma)}^2 <
    \epsilon_0 \RedShiftEnergy(\tStar)[e^{-\ImagUnit\sigma\tStar}u],\qquad
    a C_{\RedShiftN}(\epsilon_0) C_{\TFixer} < \frac{1}{2}.
  \end{equation}
  We observe that this is possible precisely because both
  $\epsilon_0=\frac{1}{4}\max_{\Horizon= \EventHorizonFuture,
    \CosmologicalHorizonFuture}\SurfaceGravity_{\Horizon}$ and
  $\delta_0$ were chosen in a manner such that so that they are
  positive for all $a$, independent of the choice of $\Omega$. As a
  result, $C_{\RedShiftN}(\epsilon_0), C_{\TFixer}(\delta_0)$ are
  finite, and we can pick $a_0$ such that for all $a<a_0$,
  \eqref{linear:eq:mermorphic:injective:a-condition} is satisfied by
  continuity.
  
  Thus, using (\ref{linear:eq:meromorphic:injective:Killing-first})
  to control the $\TFixer$-energy in
  (\ref{linear:eq:meromorphic:injective:redshift-first}), we have that
  there exists some $C(\epsilon,\delta)$ such that
  \begin{equation*}
    \RedShiftEnergy_\gamma(\tStar)[e^{-\ImagUnit\sigma\tStar}u]
    \le C(\epsilon_0,\delta_0)
    \left(\norm*{(\LinEinstein - \gamma)\left(e^{-\ImagUnit\sigma\tStar}u\right)}_{\InducedLTwo(\Sigma)}^2
      + \norm*{e^{-\ImagUnit\sigma\tStar}u}_{\InducedLTwo(\Sigma)}^2 \right)
    + \frac{\gamma}{2}\norm*{e^{-\ImagUnit\sigma\tStar}u}_{\InducedLTwo(\Sigma)}^2. 
  \end{equation*}
  
  Picking $\gamma$ such that
  $\gamma>2 C(\epsilon_0,\delta_0)$, and
  observing that
  \begin{equation*}
    \RedShiftEnergy_\gamma(\tStar)[{h}] =
    \RedShiftEnergy_{\gamma'}(\tStar)[{h}] + 
    (\gamma-\gamma')\norm*{{h}}_{\LTwo(\Sigma)}^2,
  \end{equation*}  
  we have in fact that
  \begin{equation*}
    \RedShiftEnergy(\tStar)[e^{-\ImagUnit\sigma\tStar}u]
    \le C(\epsilon_0, \delta_0) \norm*{(\LinEinstein - \gamma)e^{-\ImagUnit\sigma\tStar}u}_{\InducedLTwo(\Sigma)}^2.
  \end{equation*}
  Multiplying by $e^{2\Im\sigma\tStar}$, both sides of the inequality
  become independent of time, and using the formula for the
  Laplace-transformed operator and \eqref{linear:eq:RedShiftEstimate:one}, we
  have that
  \begin{equation*}
    \norm*{u}_{\InducedHk{1}_\sigma(\Sigma)}\le C(\epsilon,\delta)
    \norm*{(\widehat{\LinEinstein}(\sigma) - \gamma)u}_{\InducedLTwo(\Sigma)}, 
  \end{equation*}
  as desired, concluding the proof of Theorem 
  \ref{linear:thm:meromorphic:injective}. 
\end{proof}

\subsubsection{Surjectivity}

In the previous subsection, we established injectivity of
$\widehat{\LinEinstein}(\sigma) - \gamma:
D^1(\widehat{\LinEinstein}(\sigma))\to \LTwo(\Sigma)$ for a certain
range of $\sigma$. To complete the proof of invertibility, we need to
verify surjectivity, which will follow from injectivity of the adjoint
operator.

\begin{lemma}
  Let $A$ be a closed, densely defined operator on a Hilbert space $H$
  with closed range. Then $A$ is surjective if and only if the adjoint
  $A^*$ is injective. 
\end{lemma}

We now move onto the main result of this section, the injectivity of
$\widehat{L}^\dagger(\sigma) - \gamma$. 

\begin{theorem}
  \label{linear:thm:meromorphic:surjective}
  For all sufficiently
  slowly-rotating \KdS{} metrics $g_b$ where $b=(M, a)$ the following
  holds. For a fixed compact domain
  \begin{equation*}
    \Omega\subset\curlyBrace*{\sigma\in\Complex:\Im\sigma
      > \frac{1}{2}\max_{\Horizon=\EventHorizonFuture,\CosmologicalHorizonFuture}
      \left(\SHorizonControl{\LinEinstein}[\Horizon]
        + \frac{3}{2}\SurfaceGravity_{\Horizon} \right) },
  \end{equation*}
  there exists $\gamma_1$ depending on $\Omega, g_b$, such that
  $\widehat{\LinEinstein}^\dagger(\sigma)$ is injective for
  $\gamma>\gamma_1$ and
  $D(\widehat{\LinEinstein}^\dagger(\sigma))\subset H^1(\Sigma)$ for
  all $\sigma\in\Omega$. In addition, the following estimate holds:
  \begin{equation*}
    \norm{u}_{\InducedHk{1}(\Sigma)} \lesssim \norm*{(\widehat{\LinEinstein}^\dagger(\sigma) - \gamma)u}_{\InducedLTwo(\Sigma)}.
  \end{equation*}
\end{theorem}

\begin{proof}
  The outline of the proof closely follows the proof of Theorem
  \ref{linear:thm:meromorphic:injective}. The main difference lies in that
  the domain of $\LinEinstein^\dagger$ consists of functions vanishing
  along the horizons. Thus instead of applying the estimates in part
  \ref{linear:item:RedShiftEstimate-with-gamma:two} of Corollary
  \ref{linear:coro:redshift-energy-estimate-with-gamma}, we apply the
  estimates in part \ref{linear:item:RedShiftEstimate-with-gamma:three} of
  Corollary \ref{linear:coro:redshift-energy-estimate-with-gamma}.
  
  We apply Corollary \ref{linear:coro:Killing-estimate-with-gamma}
  part \ref{linear:item:Killing-estimate-with-gamma:four} to the
  function
  $e^{-\ImagUnit(\overline{\sigma} + c)\tStar}\chi_{\bullet}u$, where
  $\chi_\bullet = \chi_\bullet(r)$ is that from Part
  \ref{linear:item:RedShiftEstimate:three} of Theorem
  \ref{linear:thm:redshift-energy-estimate} and $c$ is some constant
  which we will determine later, with $\LinEinstein^\dagger$ in place
  of $\LinEinstein$. Then,
  \begin{equation*}
    \begin{split}
      2\Im(\sigma+c)E_\gamma(\tStar)[e^{-\ImagUnit\overline{\sigma}\tStar} \chi_{\bullet}u] 
      \le {}& 
              \epsilon\left(
              \norm*{e^{-\ImagUnit\overline{\sigma}\tStar} u}_{\InducedHk{1}(\widetilde{\Sigma})}^2
              + \norm*{(\LinEinstein^\dagger - \gamma)e^{-\ImagUnit\overline{\sigma}\tStar}  u}_{\InducedLTwo(\widetilde{\Sigma})}^2
              + \norm*{\left[\LinEinstein^\dagger, \chi_{\bullet} \right]e^{-\ImagUnit\overline{\sigma}\tStar}u}_{\InducedLTwo(\widetilde{\Sigma})}^2\right)\\
            &+ \epsilon\norm*{\left(cP_1^\dagger + 2c\GInvdtdt^{-1} + c^2\GInvdtdt^{-1}\right)e^{-\ImagUnit\overline{\sigma}\tStar} u}_{\InducedLTwo(\widetilde{\Sigma})}^2 
              +C(\epsilon)|c+\overline{\sigma}|^2\norm*{e^{-\ImagUnit\overline{\sigma}\tStar} u}_{\InducedLTwo(\widetilde{\Sigma})}^2\\
            &+ aC(\epsilon)\norm*{ e^{-\ImagUnit\overline{\sigma}\tStar} u}_{\InducedHk{1}(\widetilde{\Sigma})}^2
              + a\gamma C\norm*{e^{-\ImagUnit\overline{\sigma}\tStar}u}_{\InducedLTwo(\widetilde{\Sigma})}^2.
    \end{split}
  \end{equation*}
  We choose $c$ such that $\Im(\sigma+c)>0$ so that the left-hand side
  is positive. Like in the proof of Theorem
  \ref{linear:thm:meromorphic:injective}, we use the fact that $P_1^\dagger$
  is a bounded map from $\InducedHk{1}(\Sigma)\to \LTwo(\Sigma)$, and
  that $\epsilon$ can be made arbitrarily small to conclude that given
  $\delta>0$, there exists $C(\delta)$ such that for any
  $\sigma\in\Omega$,
  \begin{align}
    \EnergyKill_\gamma(\tStar)[e^{-\ImagUnit\overline{\sigma}\tStar}\chi_{\bullet}(r) u]
    \le{} &
      \delta\left(\norm*{e^{-\ImagUnit\overline{\sigma}\tStar} u}_{\InducedHk{1}(\widetilde{\Sigma})}^2
      + \norm*{(\LinEinstein^\dagger - \gamma)e^{-\ImagUnit\overline{\sigma}\tStar} u}_{\InducedLTwo(\widetilde{\Sigma})}^2\right)\notag \\
          &+ a C_{\TFixer}(\delta)\norm*{ e^{-\ImagUnit\overline{\sigma}\tStar} u}_{\InducedHk{1}(\widetilde{\Sigma})}^2
           +  a C_{\TFixer}\gamma\norm*{e^{-\ImagUnit\overline{\sigma}\tStar} u}_{\InducedLTwo(\widetilde{\Sigma})}^2
            \label{linear:eq:meromorphic:adjoint-injective:Killing-first}.
  \end{align} 
  Next, we apply the redshift estimate in part
  \ref{linear:item:RedShiftEstimate-with-gamma:three} of Corollary
  \ref{linear:coro:redshift-energy-estimate-with-gamma}, with
  $\LinEinstein^\dagger$ in place of $\LinEinstein$, and with
  $\varepsilon_{\RedShiftN}<\frac{1}{4}\max_{\Horizon=\EventHorizonFuture,\CosmologicalHorizonFuture}\SurfaceGravity_{\Horizon}$. Recalling
  that $\SHorizonControl{\LinEinstein^\dagger}^* =
  -\SHorizonControl{\LinEinstein}$, we have that 
  \begin{align}
    &\left(
      2\Im\sigma - \max_{\Horizon=\EventHorizonFuture, \CosmologicalHorizonFuture}(\SHorizonControl{\LinEinstein}[\Horizon]
      - \SurfaceGravity_{\Horizon}) - \varepsilon_{\RedShiftN}-\epsilon_0
      \right)
      \RedShiftEnergy_\gamma(\tStar)[e^{-\ImagUnit\overline{\sigma}\tStar}u] \notag \\
    \le{}& 
           C_{\RedShiftN}(\epsilon_0)\left(
           \norm*{(\LinEinstein^\dagger - \gamma)e^{-\ImagUnit\overline{\sigma}\tStar}u}_{\InducedLTwo(\Sigma)}^2
           + \norm*{e^{-\ImagUnit\overline{\sigma}\tStar}u}_{\InducedLTwo(\Sigma)}^2
           + \EnergyKill_\gamma(\tStar)[e^{-\ImagUnit\overline{\sigma}\tStar}\chi_\bullet(r)u]
           \right),  \label{linear:eq:merormorphic:adjoint-injective:Redshift-first}
  \end{align}
  where we choose
  $\epsilon_0 =
  \frac{1}{4}\max_{\Horizon=\EventHorizonFuture,\CosmologicalHorizonFuture}\SurfaceGravity_{\Horizon}$,
  so that no matter our choice of $\Omega$,
  \begin{equation*}
    2\Im\sigma-\max_{\Horizon=\EventHorizonFuture,\CosmologicalHorizonFuture}(\SHorizonControl{\LinEinstein}[\Horizon]+\SurfaceGravity_{\Horizon})
    - \varepsilon_{\RedShiftN}-\epsilon_0 > 0.
  \end{equation*}
  We now use (\ref{linear:eq:meromorphic:adjoint-injective:Killing-first}) to
  control the
  $C(\epsilon)\EnergyKill_\gamma(\tStar)[e^{-\ImagUnit\overline{\sigma}\tStar}u]$
  term on the right-hand side of
  (\ref{linear:eq:merormorphic:adjoint-injective:Redshift-first}).
  We first recognize that there exists a $\delta_0$ such that for all
  $\delta<\delta_0$,  
  \begin{equation*}
    \delta C_{\RedShiftN}(\epsilon_0)\norm*{e^{-\ImagUnit\overline{\sigma}\tStar}u}_{\InducedHk{1}(\Sigma)}^2
    \le \epsilon_0 \RedShiftEnergy_\gamma(\tStar)[e^{-\ImagUnit\overline{\sigma}\tStar}u].
  \end{equation*}
  We then see that there exists some $a_0$ sufficiently small so that
  for all $a<a_0$,
  \begin{equation}
    \label{linear:eq:mermorphic:surjective:a-condition}
    aC_{\RedShiftN}(\epsilon_0)C_{\TFixer}(\delta_0)\norm*{
      e^{-\ImagUnit\sigma\tStar}u}_{\InducedHk{1}(\Sigma)}^2 <
    \epsilon_0 \RedShiftEnergy(\tStar)[e^{-\ImagUnit\sigma\tStar}u],\qquad
    a C_{\RedShiftN}(\epsilon_0) C_{\TFixer} < \frac{1}{2}.
  \end{equation}
  As in the proof of Theorem \ref{linear:thm:meromorphic:injective}, we
  remark that the existence of such an $a_0$ comes from the fact that
  $\epsilon_0$ is uniformly positive and bounded away from $0$ for all
  $a$ sufficiently small. Then we have that
  \begin{equation*}
    \RedShiftEnergy_\gamma(\tStar)[e^{-\ImagUnit\overline{\sigma}\tStar}u]
    \le C(\epsilon_0, \delta_0)\left(\norm*{(\LinEinstein^\dagger - \gamma)e^{-\ImagUnit\overline{\sigma}\tStar}u}_{\InducedLTwo(\Sigma)}^2
      +\norm*{e^{-\ImagUnit\overline{\sigma}\tStar}u}_{\InducedLTwo(\Sigma)}^2\right)
    + \frac{\gamma}{2}\norm*{e^{-\ImagUnit\overline{\sigma}\tStar}u}_{\InducedLTwo(\Sigma)}^2. 
  \end{equation*}
  Again, we choose $\gamma$ such that
  $\frac{\gamma}{2} >C(\epsilon_0,\delta_0)$. The $L^2$ norm on the
  right-hand side of
  (\ref{linear:eq:meromorphic:adjoint-injective:Killing-first}) can
  then be absorbed to obtain
  \begin{equation*}
    \RedShiftEnergy(\tStar)[e^{-\ImagUnit\overline{\sigma}\tStar}u]
    \le C(\epsilon,\delta)\norm*{(\LinEinstein^\dagger - \gamma)e^{-\ImagUnit\overline{\sigma}\tStar}u}_{\InducedLTwo(\Sigma)}^2. 
  \end{equation*}
  for any $\gamma>\gamma_0$.  Finally, multiplying both sides by
  $e^{2\Im\sigma\tStar}$ and using the definition of the
  Laplace-transformed operator and \eqref{linear:eq:RedShiftEstimate:one}
  concludes the proof of Theorem \ref{linear:thm:meromorphic:surjective}.
\end{proof}

Given the proofs of injectivity for $\widehat{\LinEinstein}(\sigma)$ and its
adjoint on the region
\begin{equation*}
  \Im\sigma>\max_{\Horizon=\EventHorizonFuture, \CosmologicalHorizonFuture}\frac{1}{2}\left(
    \SHorizonControl{\LinEinstein}[\Horizon]
    + \frac{3}{2}\SurfaceGravity_{\Horizon}
  \right),
\end{equation*}
we can now prove invertibility.

\begin{theorem}
  \label{linear:thm:meromorphic:invertibility:zero-order}
  Let $\LinEinstein$ be the gauged Einstein operator linearized around
  a slowly-rotating \KdS{} black hole. Fix a compact domain
  \begin{equation*}
    \Omega\subset \curlyBrace*{
      \sigma\in\Complex: 
      \Im\sigma>\frac{1}{2}\max_{\Horizon=\EventHorizonFuture, \CosmologicalHorizonFuture}\left(
        \SHorizonControl{\LinEinstein}[\Horizon] + \frac{3}{2}\SurfaceGravity_{\Horizon} \right)
    }.
  \end{equation*}
  Then
  there exists some $\gamma_1$ depending only on $\Omega$, the
  black-hole parameters $(M, a)$ such that for $\gamma>\gamma_1$,
  \begin{equation*}
    \widehat{\LinEinstein}(\sigma) - \gamma: D^1(\widehat{\LinEinstein}(\sigma))\to
    \InducedLTwo(\Sigma)
  \end{equation*}
  is invertible,  and
  \begin{equation*}
    (\widehat{\LinEinstein}(\sigma) - \gamma)^{-1}: \InducedLTwo(\Sigma)\mapsto \InducedHk{1}(\Sigma)
  \end{equation*}
  is well-defined. 
\end{theorem}

\begin{proof}
  Theorem \ref{linear:thm:meromorphic:injective} gives us that for
  sufficiently large $\gamma$,
  $(\widehat{\LinEinstein}(\sigma)-\gamma):D^1(\widehat{\LinEinstein}(\sigma))\to
  \InducedLTwo(\Sigma)$ is injective. Then by Theorem
  \ref{linear:thm:meromorphic:surjective} we also have surjectivity (after
  potentially increasing $\gamma$ if necessary). This proves the
  existence of the resolvent $(\widehat{\LinEinstein}(\sigma)
  -\gamma)^{-1}$. The desired estimate is simply a consequence of equation
  (\ref{linear:eq:meromorphic:injective:main-eqn}). 
\end{proof}

\begin{remark}
  From Theorem \ref{linear:thm:meromorphic:injective} we have that
  $(\widehat{\LinEinstein}(\sigma)-\gamma):D^1(\widehat{\LinEinstein}(\sigma))\to
  \InducedLTwo(\Sigma)$ is injective on
  $\Im\sigma>\frac{1}{2}
  \max_{\Horizon=\EventHorizonFuture,\CosmologicalHorizonFuture}
  \left(\SHorizonControl{\LinEinstein}[\Horizon] -
    \frac{1}{2}\SurfaceGravity_{\Horizon} \right)$. On the other hand,
  Theorem \ref{linear:thm:meromorphic:surjective} only gives surjectivity on
  the region
  $\Im\sigma> \frac{1}{2}
  \max_{\Horizon=\EventHorizonFuture,\CosmologicalHorizonFuture}
  (\SHorizonControl{\LinEinstein}[\Horizon] +
  \frac{3}{2}\SurfaceGravity_{\Horizon})$. While this appears to be an
  obstacle to proving invertibility on the entirety of the region
  $\Im\sigma>\frac{1}{2}\max_{\Horizon=\EventHorizonFuture,\CosmologicalHorizonFuture}\left(\SHorizonControl{\LinEinstein}[\Horizon]
    - \frac{1}{2}\SurfaceGravity_{\Horizon}\right)$, as we will see in
  the next section, it is possible to extend the range of
  invertibility to the full range on which we have shown that
  $\widehat{\LinEinstein}(\sigma)-\gamma$ is injective.
\end{remark}

\subsubsection{Extending the inverse}

Recall that in commuting the equation $\LinEinstein {h}=f$ with the
vectorfields $\RedShiftK_i$ constructed in Lemma
\ref{linear:lemma:enhanced-redshift:Ka-construction}, we were able to recover
a strongly hyperbolic operator $\bL$ satisfying $\bL {h}= \blittleF$
with $\SHorizonControl{\bL}[\Horizon]$ improved by
$2\SurfaceGravity_{\Horizon}$ with each commutation. This allowed us
to prove higher-regularity energy inequalities for
$\LinEinstein {h}=0$. We will now apply these principles to the
Laplace-transformed operator to prove that for $f$ sufficiently
regular, $\widehat{\LinEinstein}(\sigma)-\gamma$ is invertible on a
larger set.

\begin{theorem}
  \label{linear:thm:meromorphic:extension:main}
  For $g$ a sufficiently
  slowly-rotating \KdS{} metric, let $\LinEinstein$ denote the gauged
  linearized Einstein operator. Moreover, let $k$ be a positive integer and fix
  a compact domain
  \begin{equation}
    \label{linear:eq:meromorphic:extension:domain}
    \Omega\subset
    \curlyBrace*{\sigma\in\Complex:\Im\sigma>\frac{1}{2}\max_{\Horizon=\EventHorizonFuture,\CosmologicalHorizonFuture}\left(
        \SHorizonControl{\LinEinstein}[\Horizon]
        - \left(2k + \frac{1}{2} \right)\SurfaceGravity_{\Horizon}
      \right)}.
  \end{equation}
  Then there exists
  $\gamma_k$ depending on $\Omega, b, k$, such that for
  $\gamma>\gamma_k$, the equation
  \begin{equation}
    \label{linear:eq:meromorphic:extension:Laplace-transformed-eqn}
    (\widehat{\LinEinstein}(\sigma)-\gamma)u=f 
  \end{equation}
  admits a unique solution for any $f\in
  \InducedHk{k-1}(\Sigma)$. Furthermore, $u\in \InducedHk{k}(\Sigma)$ with the
  estimate
  \begin{equation}
    \label{linear:eq:meromorphic:extension:main-est}
    \norm{u}_{\InducedHk{k}(\Sigma)}
    \lesssim 
    \norm{f}_{\InducedHk{k-1}(\Sigma)}.
  \end{equation}
\end{theorem}

\begin{proof}
  Consider the Laplace-transformed commutators $\widehat{\RedShiftK}_i:
  \InducedHk{k}(\Sigma) \to \InducedHk{k-1}(\Sigma)$, defined by
  \begin{equation*}
    \widehat{\RedShiftK}_i(\sigma)u = \left.e^{\ImagUnit\sigma\tStar}\RedShiftK_ie^{-\ImagUnit\sigma\tStar}u\right\vert_{\Sigma_{\tStar}}.
  \end{equation*}
  The theorem is proven in two steps. 
  The first step to proving the theorem will then be to inductively
  commute with $\widehat{\RedShiftK}_i$ to prove that
  (\ref{linear:eq:meromorphic:extension:main-est}) holds for $\sigma$ in a
  compact domain
  \begin{equation}
    \label{linear:eq:mermomorphic:ext:first-domain}
    \Omega\subset \curlyBrace*{\sigma\in\Complex:
      \Im\sigma>\frac{1}{2}\max_{\Horizon=\EventHorizonFuture,\CosmologicalHorizonFuture}
      \left(
        \SHorizonControl{\LinearOp}[\Horizon] -
        \left(2k-\frac{3}{2}\right)\SurfaceGravity_{\Horizon}
      \right)}.
  \end{equation}
  We will then extend this to the full region in the theorem by
  taking the higher-order estimates, re-applying them to the
  lower-order estimates via an approximation argument, and then
  repeating the induction argument to achieve the desired results.

  The $k=0$ case of our induction is the content of Theorem
  \ref{linear:thm:meromorphic:invertibility:zero-order}.  Fix $k>0$, and a
  complex subset of the complex plane
  $\Omega\subset \curlyBrace*{\sigma\in\Complex:
    \Im\sigma>\max_{\Horizon=\EventHorizonFuture,
      \CosmologicalHorizonFuture}(\frac{1}{2}\SHorizonControl{\LinEinstein}[\Horizon]
    -
    \left(k-\frac{3}{4}\right)\SurfaceGravity_{\Horizon})}$. Assume
  for the sake of induction that the theorem holds for $k-1$, with
  \begin{equation*}
    \Omega\subset\curlyBrace*{\sigma\in\Complex:
      \Im\sigma>\frac{1}{2}\max_{\Horizon=\EventHorizonFuture, \CosmologicalHorizonFuture}
      \left(
        \SHorizonControl{\LinEinstein}[\Horizon]
        -\left(2k - \frac{7}{2}\right)\SurfaceGravity_{\Horizon}
      \right)}.
  \end{equation*}
  Commuting (\ref{linear:eq:meromorphic:extension:Laplace-transformed-eqn})
  with $\widehat{\RedShiftK}_i$ and applying Theorem
  \ref{linear:thm:redshift-commutation:main}, we can conclude that a solution
  $u$ to (\ref{linear:eq:meromorphic:extension:Laplace-transformed-eqn}) must
  induce a solution $\blittleU:=(u, \widehat{\RedShiftK}_i u)$ such that
  \begin{equation}
    \label{linear:eq:meromorphic:invertibility:proof-aux-one}
    (\widehat{\bL}(\sigma)-\gamma)\blittleU = \blittleF.
  \end{equation}
  Crucially, observe that 
  \begin{equation*}
    \SHorizonControl{\bL}[\Horizon] = \SHorizonControl{\LinEinstein}[\Horizon]
    -2\SurfaceGravity_{\Horizon}.
  \end{equation*}
  Since
  $\Omega\subset
  \curlyBrace{\sigma\in\Complex:\Im\sigma>
    \frac{1}{2}\max_{\Horizon=\EventHorizonFuture,
      \CosmologicalHorizonFuture}\left(
      \SHorizonControl{\bL}[\Horizon]
    - \left(2k - \frac{7}{2}
    \right)\SurfaceGravity_{\Horizon}
  \right)}$ and
  $\blittleF\in \InducedHk{k-2}(\Sigma)$, we now apply the induction
  assumption on the commuted equation in
  \eqref{linear:eq:meromorphic:invertibility:proof-aux-one}. This yields
  that for sufficiently large $\gamma$, there exists a unique
  solution $\blittleU= (u, u_i)$ to
  (\ref{linear:eq:meromorphic:invertibility:proof-aux-one}), and moreover,
  that $\blittleU\in \InducedHk{k-1}(\Sigma)$.
  The second part of Theorem \ref{linear:thm:redshift-commutation:main}
  then shows that
  \begin{equation*}
    (\widehat{\bL}'(\sigma)-\gamma)\tilde{u}=0,\qquad  \tilde{u} = u_i-\widehat{\RedShiftK}_i(\sigma)u.
  \end{equation*}
  Thus, we have that $\tilde{u}$ solves
  (\ref{linear:eq:meromorphic:extension:Laplace-transformed-eqn}).  The
  estimate from the inductive assumption,
  \begin{equation*}
    \norm{\blittleU}_{\InducedHk{k-1}(\Sigma)} \le C \norm{\blittleF}_{\InducedHk{k-2}(\Sigma)},
  \end{equation*}
  then implies
  (\ref{linear:eq:meromorphic:extension:main-est}). We can then relax the
  assumption that $f$ is smooth to an assumption that $f\in
  \InducedHk{k-1}(\Sigma)$.
  
  Now we move on to extending the domain from that in
  \eqref{linear:eq:mermomorphic:ext:first-domain} to that in
  \eqref{linear:eq:meromorphic:extension:domain}. To do so, we reconsider the
  $k=0$ case. Let us assume that $f\in \InducedHk{1}(\Sigma)$. Then
  for sufficiently large $\gamma$, we have already shown that
  (\ref{linear:eq:meromorphic:extension:Laplace-transformed-eqn}) has a
  $\InducedHk{2}(\Sigma)$ solution if
  \begin{equation*}
    \sigma\in\Omega\subset\curlyBrace*{\sigma\in\Complex:\Im\sigma >
      \frac{1}{2}\max_{\Horizon=\EventHorizonFuture,
        \CosmologicalHorizonFuture}\left(
        \SHorizonControl{\LinEinstein}[\Horizon]
        -\frac{1}{2}\SurfaceGravity_{\Horizon} \right)}.
  \end{equation*}
  Theorem
  \ref{linear:thm:meromorphic:injective} then shows that this solution is
  unique in $\InducedHk{1}(\Sigma)$, and moreover satisfies the
  estimate
  \begin{equation*}
    \norm{u}_{\InducedHk{1}(\Sigma)} \le C\norm{f}_{\InducedLTwo(\Sigma)}. 
  \end{equation*}
  Then, the fact that $\InducedHk{1}(\Sigma)$ is dense in $\InducedLTwo(\Sigma)$
  allows us to deduce that for any $f\in \InducedLTwo(\Sigma)$,
  (\ref{linear:eq:meromorphic:extension:Laplace-transformed-eqn}) has a
  unique solution for $s\in\Omega$. We have thus extended the possible
  range of $\Omega$ to the desired range for $k=1$. Doing this inductively
  yields the result for arbitrary $k$. 
\end{proof}

We can now prove a Fredholm alternative for
$\widehat{\LinEinstein}(\sigma)$.

\begin{theorem}
  \label{linear:thm:meromorphic:fredholm-alt:Laplace-transformed-op}
  For $g$ a sufficiently
  slowly-rotating \KdS{} metric, let $\LinEinstein$ be the gauged
  linearized Einstein operator on $g$. Then for any $k\in \Natural$,
  and 
  $\sigma\in \Complex$ such that
  \begin{equation*}
    \Im\sigma >\frac{1}{2}\max_{\Horizon=\EventHorizonFuture,
      \CosmologicalHorizonFuture}\left(
      \SHorizonControl{\LinEinstein}[\Horizon]
      -\left(2k+\frac{1}{2}\right)\SurfaceGravity_{\Horizon}\right) ,
  \end{equation*}
  one of the following holds: either
  \begin{enumerate}
  \item $\widehat{\LinEinstein}(\sigma)^{-1}$ exists as a bounded map from
    $\InducedHk{k-1}(\Sigma)$ to $D^{k}(\widehat{\LinEinstein}(\sigma))$, or
  \item there exists a finite-dimensional family of solutions to
    $\widehat{\LinEinstein}(\sigma)u=0$.
  \end{enumerate}
  Moreover, the latter occurs only when $\sigma\in \QNFk{k}(\LinEinstein)$,
  where $\QNFk{k}(\LinEinstein)$ is a discrete set of points with no
  accumulation point except at infinity, satisfying that
  \begin{equation*}
    \QNFk{k}(\LinEinstein)\subset \QNFk{k+1}(\LinEinstein).
  \end{equation*}
  
\end{theorem}

\begin{remark}
  The theorem above is, alternatively stated, the following: the
  function $\sigma\to \widehat{\LinEinstein}(\sigma)$ is meromorphic on
  the half-plane
  $\curlyBrace*{\sigma:\Im\sigma>\frac{1}{2}\max_{\Horizon=\EventHorizonFuture,\CosmologicalHorizonFuture}
    \left(
      \SHorizonControl{\LinEinstein}[\Horizon]
      -\left(2k + \frac{1}{2}\right)\SurfaceGravity_{\Horizon}
    \right)
  }$,
  with poles of finite order at $\QNFk{k}(\LinEinstein)$.
\end{remark}

\begin{proof}
  Let $\Omega$ be a fixed compact connected set such that
  \begin{equation*}
    \Omega\subset
    \curlyBrace*{\sigma:\Im\sigma> \frac{1}{2}\max_{\Horizon=\EventHorizonFuture,\CosmologicalHorizonFuture}
      \left(
        \SHorizonControl{\LinEinstein}[\Horizon]
      -\left(2k + \frac{1}{2}\right)\SurfaceGravity_{\Horizon}      
      \right)
    }.
  \end{equation*}
  Recall that from Theorem \ref{linear:thm:meromorphic:extension:main}, we
  have shown that there exists a $\lambda$ sufficiently large such
  that
  \begin{equation*}
    (\widehat{\LinEinstein}(\sigma)-\lambda)^{-1}:
    \InducedHk{k}(\Sigma)\to D^k(\widehat{\LinEinstein}(\sigma))\subset \InducedHk{k+1}(\Sigma)
  \end{equation*}
  is a well-defined, bounded operator for all $\sigma\in\Omega$. Then
  for the same $\lambda$, we can define the operator
  $A(\sigma):\InducedHk{k}(\Sigma)\to \InducedHk{k+1}(\Sigma)$ by
  \begin{equation*}
    A(\sigma):=-\lambda(\widehat{\LinEinstein}(\sigma)-\lambda)^{-1},
  \end{equation*}
  exists as a bounded operator on the entirety of $\Omega$. We will
  prove the main theorem by using the analytic Fredholm theorem on
  $A(\sigma)$. We first verify the conditions for the application of
  the analytic Fredholm theorem (Theorem \ref{linear:thm:AnalyticFredholm}):
  \begin{enumerate}
  \item As previously mentioned, using Theorem
    \ref{linear:thm:meromorphic:extension:main}, there exists some $\lambda$
    such that $A(\sigma)$ is a bounded operator on the entirety of
    $\Omega$.
  \item Also by Theorem \ref{linear:thm:meromorphic:extension:main}, we know
    that for $\sigma\in\Omega$, $A(\sigma)$ maps $\InducedHk{k}(\Sigma)\mapsto
    \InducedHk{k+1}(\Sigma)$. Thus, by Rellich-Kondrachov, $A(\sigma)$ is a compact operator.
  \item $A(\sigma)$ is analytic on $\Omega$, and we can calculate directly that
    \begin{equation*}
      \lim_{\sigma\to \sigma_0} \frac{A(\sigma)-A(\sigma_0)}{\sigma-\sigma_0}
      = -\lambda(\widehat{\LinEinstein}(\sigma_0)-\lambda)^{-1}
      \left(P_1 + 2\sigma_0\GInvdtdt^{-1}\right)
      (\widehat{\LinEinstein}(\sigma_0)-\lambda)^{-1},
    \end{equation*}
    which is a bounded operator on $\InducedHk{k}(\Sigma)$, where we
    define $P_1$ by writing
    \begin{equation*}
      \LinEinstein_{g_b} = \GInvdtdt^{-1}D_{\tStar}^2 + P_1D_{\tStar} + P_2.
    \end{equation*}    
  \end{enumerate}
  We have thus confirmed that $A(\sigma)$ verifies the conditions to
  apply the analytic Fredholm theorem. We now observe that
  \begin{equation*}
    \widehat{\LinEinstein}(\sigma)u=f
    \quad\iff
    \quad(1-A(\sigma))u = (\widehat{\LinEinstein}(\sigma)-\lambda)^{-1}f,
  \end{equation*}
  and that moreover, by Proposition
  \ref{linear:prop:basic-C0-semigroup-properties},
  $\widehat{\LinEinstein}(\sigma)$ always exists for some
  $\sigma\in\Omega$ after potentially extending $\Omega$. We can now
  apply the analytic Fredholm theorem to draw the conclusion that
  either
  \begin{enumerate}
  \item $\widehat{\LinEinstein}(\sigma)^{-1}$ exists as a bounded map from
    $\InducedHk{k}(\Sigma)$ to $\InducedHk{k+1}(\Sigma)$, or
  \item There exists a finite-dimensional family of solutions to
    $\widehat{\LinEinstein}(\sigma)u=0$.
  \end{enumerate}

  From Corollary \ref{linear:coro:redshift-regularity},
  we see that solutions of $\widehat{\LinEinstein}(\sigma)u=0$ are
   smooth. Thus, $\QNFk{k}(\LinEinstein)\subset \QNFk{k+1}(\LinEinstein)$.
  Finally, $1-A(\sigma)$ is a compact perturbation of the identity,
  and thus is Fredholm of index 0. Consequently, the dimension of the
  kernel and the co-kernel agree.
\end{proof}

\subsection{Identifying the spectral gap (Proof of Theorem
  \ref{linear:thm:resolvent-estimate:inf-gen})}
\label{linear:sec:Resolvent-Estimates}

As was the case for the proof of Theorem \ref{linear:thm:meromorphic:main-A},
we prove Theorem \ref{linear:thm:resolvent-estimate:inf-gen} by analyzing
$\widehat{\LinEinstein}(\sigma)$ in place of $\InfGen -
\sigma$. Recall from Lemma \ref{linear:lemma:laplace:inverse-A} that $(\InfGen
- \sigma)^{-1}(\Sigma):\LSolHk{k}(\Sigma) \to D^k(\InfGen)$ exists and is a bounded
linear transformation if and only if
$\widehat{\LinEinstein}(\sigma)^{-1}: \InducedHk{k-1}(\Sigma) \to
D^k(\widehat{\LinEinstein}(\sigma))$ exists and is a bounded linear
transformation. But from the Morawetz estimate and the corresponding
resolvent estimates in Theorem \ref{linear:thm:resolvent-estimate:main}, we
know that this is the case for $\Im\sigma\ge -\SpectralGap,
\abs*{\sigma}>C_0$, as desired.

\begin{remark}
  We briefly remark that what we have essentially shown in this
  section, is that Assumption \ref{linear:ass:QNM} could instead be reduced
  to the following two assumptions. 
  \begin{enumerate}
  \item  $\LinearOp$ is a strongly hyperbolic linear operator on a
    slowly-rotating \KdS{} background, $g_b$.
  \item For any $\varepsilon>0$, there exists an elliptic stationary
    zero-order pseudo-differential operator $\PseudoSubPFixer$ such
    that
    \begin{equation*}
      \evalAt*{\frac{1}{\abs*{\FreqAngular}}\sigma_1\left(
        \PseudoSubPFixer\LinearOp \PseudoSubPFixer^- - (\PseudoSubPFixer\LinearOp \PseudoSubPFixer^-)^*
      \right)}_{\TrappedSet_{b}} <\varepsilon \Identity,
    \end{equation*}
    where $\PseudoSubPFixer^-$ denotes the parametrix of
    $\PseudoSubPFixer$, and $\sigma_1(P)$ denotes the principal symbol
    of $P$, and $\TrappedSet_b$ denotes the trapped set of the \KdS{}
    metric $g_b$ (recall the definitions in Section \ref{linear:sec:freq-analysis}).  
  \end{enumerate}
\end{remark}

\section{Mode Stability}
\label{linear:sec:mode-stability}

Having shown that $\LinEinstein_{g_b}$ satisfies Theorem
\ref{linear:thm:meromorphic:main-A} and Theorem
\ref{linear:thm:resolvent-estimate:inf-gen}, the only remaining obstacle to
exponential decay are the finitely many residual non-decaying
resonances. To this end, we show that these resonances are
non-physical. At the level of the gauged linearized Einstein equation,
this can mean one of two things.
\begin{enumerate}
\item First, it is possible that the quasinormal mode solution does
  not itself satisfy the linearized constraint equations. As such, it
  is not a valid solution to the linearized Einstein vacuum equations.
\item Second, even if the quasinormal mode solution does satisfy the
  linearized constraint equations, if the quasinormal mode in question
  is an infinitesimal diffeomorphism of the zero
  solution, it is not truly a distinct solution.
\end{enumerate}
Our goal then will be to show that all of the non-decaying
$\LSolHk{k}$-quasinormal mode solutions fall into one of these two
categories, and are thus unphysical.

\subsection{Geometric mode stability of \SdS}

We begin by considering the case of the linearized ungauged Einstein's
equations linearized around a fixed \SdS{} background, $g_{b_0}$. On a
fixed \SdS{} background, a strong geometric mode stability
statement (GMS) exists, having first been proven by Kodama and
Ishibashi, \cite{kodama_master_2003}, but presented below in the
slightly modified form derived by Hintz and Vasy
\cite{hintz_global_2018}.

\begin{theorem} [Geometric mode stability (GMS). Referred to as
  Ungauged Einstein Mode Stability (UEMS) in \cite{hintz_global_2018}]
  \label{linear:thm:Kodama-Ishibashi}
  Let $b_0$ be the black hole parameters of a \SdS{} black
  hole. Then,
  \begin{enumerate}
  \item Let $\sigma\in \Complex, \Im\sigma\ge 0$, $\sigma\neq0$ and
    suppose that ${h}(\tStar, x)=e^{-\ImagUnit\sigma\tStar}u(x)$, 
    $u\in C^\infty(\Sigma)$ is a mode solution of the
    linearized Einstein equation
    \begin{equation*}
      D_{g_{b_0}}(\Ric - \Lambda)({h}) = 0.
    \end{equation*}
    Then there exists a 1-form $\omega(\tStar,
    x)=e^{-\ImagUnit\sigma\tStar}\omega_\sigma(x)$ with $\omega_\sigma\in
    C^\infty(\Sigma)$ such that
    \begin{equation*}
      {h} =\nabla_{g_{b_0}}\otimes \omega. 
    \end{equation*}
  \item For all $k\in\Natural$, and all generalized mode solutions
    \begin{equation*}
      {h}(\tStar, x) = \sum_{j=0}^k \tStar^j u_j(x),\quad u_j\in
      C^\infty(\Sigma),\quad 0\le j\le k,
    \end{equation*}
    of the linearized Einstein equation, there exist $b'\in T_{b_0}B$
    and $\omega\in C^\infty(M)$, such that
    \begin{equation*}
      {h} = g'_{b_0}(b') +\nabla_{g_{b_0}}\otimes \omega. 
    \end{equation*}
  \end{enumerate}
\end{theorem}

GMS states that any mode solution to the ungauged linearized vacuum
Einstein equations (linearized around a fixed \SdS{} background) that
is non-decaying is a pure gauge solution, arising either from an
infinitesimal diffeomorphism ${h} = \LieDerivative_{\omega^\sharp} g_{b_0}$, or a
linearized \KdS{} metric (or a combination of the two). In particular,
being a mode stability statement at the level of the ungauged system
of equations, GMS also implies a mode stability statement at the level
of the gauged equations. That is, applying the linearized Bianchi
equation, we immediately have the following corollary:
\begin{corollary}
  \label{linear:cor:Kodama-Ishibashi:Constraint}
  Let $b_0$ be the black hole parameters of a \SdS{} black hole.
  Let $\sigma\in \Complex, \Im\sigma\ge 0$, $\sigma\neq0$ and
  suppose that ${h}(\tStar, x)=e^{-\ImagUnit\sigma\tStar}u(x)$, 
  $u\in C^\infty(\Sigma)$ is a mode solution of the gauged
  linearized Einstein operator,
  \begin{equation}
    \label{linear:eq:mode-stab:Kodama-Ishibashi:gauged-corollary:gauged-eqn}
    \LinEinstein_{g_{b_0}}{h}=0.
  \end{equation}  
  Then one of the following must be true:
  \begin{enumerate}
  \item either ${h} = \nabla_{g_{b_0}} \otimes\omega$, or
  \item $\psi := \Constraint_{g_{b_0}}({h}) = e^{-\ImagUnit\sigma
      \tStar}v(x)$ is a non-zero mode solution to the constraint
    propagation equation
    \begin{equation*}
      \ConstraintPropagationOp_{g_{b_0}}\psi = (\VectorWaveOp[g_{b_0}]-\Lambda)\psi = 0.
    \end{equation*}
  \end{enumerate}
  A similar result holds for the quasinormal mode at 0. Indeed, let
  \begin{equation*}
    {h}(\tStar,x) = \sum_{j=0}^k \tStar^ju_{jk}(x),\quad
    u_{jk}\in C^\infty(\Sigma),\quad 0\le j\le k,
  \end{equation*}
  be a generalized mode solution of the gauged linearized Einstein
  equation
  \eqref{linear:eq:mode-stab:Kodama-Ishibashi:gauged-corollary:gauged-eqn}. Then,
  one of the following must hold true:
  \begin{enumerate}
  \item  there exist $b'\in T_{b_0}B$
    and $\omega\in C^\infty(T^*\StaticRegionWithExtension)$, such that
    \begin{equation*}
      {h} = g'_{b_0}(b') +\nabla_{g_{b_0}}\otimes \omega,
    \end{equation*}
    and moreover, ${h}$ satisfies the linearized gauge constraint
    \begin{equation*}
      \Constraint_{g_{b_0}}{h} = 0 
    \end{equation*}
    uniformly in $\StaticRegionWithExtension$;
  \item $\psi := \Constraint_{g_{b_0}}({h}) = \sum_{j=0}^k
    \tStar^jv_j(x)$ is a non-zero mode solution to the constraint 
    propagation equation
    \begin{equation*}
      \ConstraintPropagationOp_{g_{b_0}}\psi = (\VectorWaveOp[g_{b_0}]-\Lambda) \psi = 0.
    \end{equation*}
  \end{enumerate}
\end{corollary}
\begin{proof}
  Recall from Section \ref{linear:sec:linearized-EVE} that a solution
  ${h}$ to the gauged linearized Einstein equation
  $\LinEinstein_{g_{b_0}}{h} = 0$ also satisfies
  \begin{equation*}
    \ConstraintPropagationOp_{g_{b_0}}\Constraint_{g_{b_0}}({h}) = 0. 
  \end{equation*}
  Thus, either $\Constraint_{g_{b_0}}({h})$ is a non-zero mode solution to
  the constraint propagation equation
  \begin{equation*}
    \ConstraintPropagationOp_{g_{b_0}}\psi = 0,
  \end{equation*}
  or $\Constraint_{g_{b_0}}({h})=0$. But if $\Constraint_{g_{b_0}}({h})=0$,
  then ${h}$ must actually be a mode solution of the ungauged linearized
  Einstein equation. An application of Theorem
  \ref{linear:thm:Kodama-Ishibashi} allows us to conclude. 
\end{proof}

We see from Theorem \ref{linear:thm:Kodama-Ishibashi} and Corollary
\ref{linear:cor:Kodama-Ishibashi:Constraint} that there are two types of
unphysical modes: those that violate the constraint conditions, and
those that are infinitesimal diffeomorphisms of a nearby linearized
\KdS{} metric.  This leads us to define the following two categories
of unphysical $\LSolHk{k}$-quasinormal mode solutions.
\begin{definition}
  Let ${h}=e^{-\ImagUnit\sigma\tStar}u$ be an
  $\LSolHk{k}$-quasinormal mode solution
  $\LinEinstein_{g_b}{h}=0$. Then we say that ${h}$ is a
  $\LSolHk{k}$-\emph{geometric quasinormal mode solution} for
  $\LinEinstein_{g_b}$ if there exists a linearized \KdS{} metric
  $g_b'(b')$ and a one-form $\omega$ such that
  \begin{equation*}
    {h} = g_{b}'(b') + \nabla_{g_{b}}\otimes\omega,
  \end{equation*}
  and moreover, ${h}$ satisfies the linearized gauge constraint
  \begin{equation*}
    \Constraint_{g_{b}}({h}) = 0
  \end{equation*}
  uniformly on $\StaticRegionWithExtension$. 
  On the other hand, let us call any quasinormal mode solution ${h}$ such that
  \begin{equation*}
    \Constraint_{g_b}({h}) \neq 0
  \end{equation*}
  a $\LSolHk{k}$-\emph{constraint quasinormal mode solution} for
  $\LinEinstein_{g_b}$. 
\end{definition}
\begin{definition}
  For any $\sigma\in \Complex$, if there exists an
  $\LSolHk{k}$-geometric quasinormal mode solution for
  $\LinEinstein_{g_b}$ of the form $e^{-\ImagUnit\sigma\tStar}u$, then
  we call $\sigma$ a $\LSolHk{k}$-\emph{geometric quasinormal
    frequency} of $\LinEinstein_{g_b}$. Similarly, if for
  $\sigma\in \Complex$, there exists an $\LSolHk{k}$-constraint
  quasinormal mode solution $e^{-\ImagUnit\sigma\tStar}u$ of
  $\LinEinstein_{g_b}$, then we call $\sigma$ a
  $\LSolHk{k}$-\emph{constraint quasinormal frequency} of
  $\LinEinstein_{g_b}$. Given an open subset $\Xi\subset \Complex$ that contains only
  finitely many $\LSolHk{k}$-quasinormal frequencies of
  $\LinEinstein_{g_b}$, we denote by $\QNFGeok{k}(\LinEinstein_{g_b},
  \Xi)$ and $\QNFConstraintk{k}(\LinEinstein_{g_b}, \Xi)$ the set of
  $\LSolHk{k}$-geometric quasinormal frequencies and
  $\LSolHk{k}$-constraint quasinormal frequencies respectively in
  $\Xi$. 
\end{definition}
\begin{definition}
  Given an open subset $\Xi\subset \Complex$ that contains only
  finitely many $\LSolHk{k}$-quasinormal frequencies of
  $\LinEinstein_{g_b}$, we denote by $\QNMGeok{k}(\LinEinstein_{g_b},
  \Xi)$ the set of all $\LSolHk{k}$-geometric quasinormal mode solutions with
  frequency $\sigma\in \Xi$, and $\QNMConstraintk{k}(\LinEinstein_{g_b},
  \Xi)$ the set of all $\LSolHk{k}$-constraint quasinormal mode solutions with
  frequency $\sigma\in \Xi$. 
\end{definition}
\begin{remark}
  Note that $\sigma$ can be both a geometric and a constraint
  quasinormal frequency for $\LinEinstein_{g_b}$ and \textit{a priori}
  it is not clear that all $\LSolHk{k}$-quasinormal frequencies of
  $\LinEinstein_{g_b}$ are either geometric or constraint frequencies.
\end{remark}

The nomenclature reflects that $\LSolHk{k}$-geometric quasinormal mode
solutions will be handled via a gauge choice, while
$\LSolHk{k}$-constraint quasinormal mode solutions violate the
constraint conditions and thus both geometric and constraint
quasinormal modes represent unphysical quasinormal mode solutions to
the gauged linearized Einstein equations.

With this new nomenclature, we see that Corollary
\ref{linear:cor:Kodama-Ishibashi:Constraint} is a precise statement that in
the case $g_b=g_{b_0}$ is a fixed \SdS{} background, any non-decaying
$\LSolHk{k}$-quasinormal mode solution
${h} = e^{-\ImagUnit\sigma\tStar}u$ of $\LinEinstein_{g_{b_0}}$ are
either an $\LSolHk{k}$-geometric quasinormal mode or an $\LSolHk{k}$-constraint
quasinormal mode solution of $\LinEinstein_{g_{b_0}}$, and
therefore, all non-decaying $\LSolHk{k}$-quasinormal mode solutions of
$\LinEinstein_{g_{b_0}}$ are unphysical.

\subsection{Linearized stability of \SdS}
\label{linear:sec:lin-stability:SdS}

Given the GMS statement in Theorem \ref{linear:thm:Kodama-Ishibashi} and the
gauged mode stability statement in Corollary
\ref{linear:cor:Kodama-Ishibashi:Constraint} for \SdS, we can now prove the
linearized stability statement for the linearized Einstein's equations
around $g_{b_0}$.

\begin{theorem}[Stability of the linearized gauged Einstein vacuum
  equations linearized around $g_{b_0}$]
  \label{linear:thm:lin-stability:SdS}
  Fix $k>k_0$, and let $(\InducedMetric', k')\in H^{k+1}(\Sigma_0,
  S^2T^*\Sigma_0)\oplus H^{k}(\Sigma_0;S^2T^*\Sigma_0)$ be
  solutions of the linearized constraint equations, linearized around
  the initial data $(\InducedMetric_{b_0}, k_{b_0})$ of the \SdS{} space
  $(\StaticRegionWithExtension, g_{b_0})$. Let ${h}$ be a solution to
  the initial value problem
  \begin{equation*}
    \begin{cases}
      \LinEinstein_{g_{b_0}}{h} = 0 &\text{in }\StaticRegionWithExtension,\\
      \gamma_0({h}) = D_{(\InducedMetric_{b_0}, k_{b_0})}i_{b_0, \Identity}(\InducedMetric', k')&\text{on }\Sigma_0,
    \end{cases}
  \end{equation*}
  where $i_{b, \Identity}$ is defined in Proposition
  \ref{linear:prop:initial-data:ib-construction}.
  
  Then, there exist $b'\in T_{b_0}B$ and
  a 1-form $\omega\in C^\infty(\StaticRegionWithExtension, T^*\StaticRegionWithExtension)$ such that
  \begin{equation*}
    h= g_{b_0}'(b') +\nabla_{g_{b_0}}\otimes \omega +\tilde{h},
  \end{equation*}
  where $\tilde{h}$ satisfies the pointwise bounds
  \begin{equation*} 
    \begin{split}
      \norm{\tilde{h}}_{\HkWithT{k}(\Sigma_{\tStar})} &\lesssim
      e^{-\SpectralGap\tStar}\left(\norm{\InducedMetric'}_{H^{k+1}(\Sigma_0)}
        + \norm{k'}_{H^{k}(\Sigma_0)}\right),\\
      \abs{b'} + e^{-\GronwallExp\tStar}\norm{\omega}_{\HkWithT{k}(\Sigma_{\tStar})} &\lesssim \left(\norm{\InducedMetric'}_{H^{k+1}(\Sigma_0)} + \norm{k'}_{H^{k}(\Sigma_0)}\right).      
    \end{split}
  \end{equation*}
\end{theorem}

\begin{proof}
  Let $\Xi = \QNFk{k}(\LinEinstein_{g_{b_0}}, \UpperHalfSpace)$. At the
  cost of reducing $\SpectralGap$, we can ensure that in fact, $\Xi$
  consists of all $\LSolHk{k}$-quasinormal frequencies such that
  $\Im\sigma>-\SpectralGap$. We index $\Xi =
  \{\sigma_{j}\}, 1\le j \le N_{\LinEinstein_{g_{b_0}}}$. 
  Applying Corollary \ref{linear:coro:asymptotic-expansion}, we see
  that ${h}$ can be written as
  \begin{equation}
    \label{linear:eq:linear-stab:SdS:Proof1:h-expression}
    {h} = \tilde{{h}} + \sum_{1\le j \le N_{\LinEinstein_{g_{b_0}}}}\sum_{\ell = 1}^{d_j}\sum_{k=0}^{n_{j\ell}} e^{-\ImagUnit\sigma_j\tStar}\tStar^ku_{j\ell},
  \end{equation}
  where each $\sum_{\ell = 1}^{d_j}\sum_{k=0}^{n_{j\ell}}
  e^{-\ImagUnit\sigma_j\tStar}\tStar^ku_{j\ell}\in
  \QNMk{k}(\LinEinstein_{g_{b_0}}, \sigma_j)$, and moreover,
  \begin{equation*}
    \begin{split}
      \norm{\tilde{{h}}}_{\HkWithT{k}(\Sigma_{\tStar})} &\lesssim e^{-\SpectralGap\tStar}\left(\norm{\InducedMetric'}_{H^{k+1}(\Sigma_0)} + \norm{k'}_{H^{k}(\Sigma_0)}\right),\\
      \norm*{\sum_{\ell = 1}^{d_j}\sum_{k=0}^{n_{j\ell}}
        e^{-\ImagUnit\sigma_j\tStar}\tStar^ku_{j\ell}}_{\HkWithT{k}(\Sigma_{\tStar})} &\lesssim e^{\GronwallExp\tStar}\left(\norm{\InducedMetric'}_{H^{k+1}(\Sigma_0)} + \norm{k'}_{H^{k}(\Sigma_0)}\right).              
    \end{split}
  \end{equation*}
  
  Recall from the construction of the map $i_{b, \Identity}$ in
  Proposition \ref{linear:prop:initial-data:ib-construction} that for ${h}$
  inducing on $\Sigma_0$ the initial data
  $i_{b,\Identity}(\InducedMetric', k') = \gamma_0(h)$, ${h}$
  satisfies the initial gauge constraint
  \begin{equation*}
    \Constraint_{g_{b_0}}{h}\vert_{\Sigma_0}
    =\LieDerivative_{\KillT} \Constraint_{g_{b_0}}{h}\vert_{\Sigma_0}
    = 0 .
  \end{equation*}
  Then, recalling from the discussion in Section \ref{linear:sec:EVE},
  $\Constraint_{g_{b_0}}{h}$ itself satisfies a wave equation, and as
  a result, the initial data $(\InducedMetric', k')$ launches a
  solution ${h}$ to the gauged linearized EVE such that
  \begin{equation*}
    \Constraint_{g_{b_0}}({h})=0
  \end{equation*}
  for all time, and thus, ${h}$ solves not only the gauged
  linearized EVE, but also the ungauged linearized EVE,
  $D_{g_{b_0}}(\Ric-\Lambda){h} = 0$. This implies that the
  $\LSolHk{k}$-quasinormal mode solutions in the finite sum in
  equation \eqref{linear:eq:linear-stab:SdS:Proof1:h-expression} also
  satisfy the linearized constraint equations
  \begin{equation*}
    \Constraint_{g_{b_0}}\left(\sum_{1\le j \le N_{\LinEinstein_{g_{b_0}}}}\sum_{\ell = 1}^{d_j}\sum_{k=0}^{n_{j\ell}} e^{-\ImagUnit\sigma_j\tStar}\tStar^ku_{j\ell}\right) =0.
  \end{equation*}
  As a result, using
  Corollary \ref{linear:cor:Kodama-Ishibashi:Constraint}, we see that
  $\sum_{1\le j \le N_{\LinEinstein_{g_{b_0}}}}\sum_{\ell =
    1}^{d_j}\sum_{k=0}^{n_{j\ell}}
  e^{-\ImagUnit\sigma_j\tStar}\tStar^ku_{j\ell}$ must in fact be a
  linear combination of 
  $\LSolHk{k}$-geometric quasinormal mode solutions, and thus there
  exists some $b'$, $\omega$ such that
  \begin{equation*}
    \sum_{1\le j \le N_{\LinEinstein_{g_{b_0}}}}\sum_{\ell = 1}^{d_j}\sum_{k=0}^{n_{j\ell}} e^{-\ImagUnit\sigma_j\tStar}\tStar^ku_{j\ell} = g_{b_0}'(b') +\nabla_{g_{b_0}}\otimes \omega,
  \end{equation*}
  as desired. 
\end{proof}

This method of proving stability of the linearized EVE around
$g_{b_0}$ relies only on having the GMS statement in Theorem
\ref{linear:thm:Kodama-Ishibashi} and an appropriate high-frequency
resolvent estimate like in Theorem
\ref{linear:thm:resolvent-estimate:inf-gen}. When trying to generalize the
result to $g_b$, $b\in\BHParamNbhd$, we see that while we do still
have a good high-frequency resolvent estimate in Theorem
\ref{linear:thm:resolvent-estimate:inf-gen}, we do not have a version of GMS
which holds for the ungauged linearized EVE around $g_b$ for
$b=(M, \AngularMomentum)$ where $\AngularMomentum\neq 0$. Instead of
proving a version of GMS for $g_b$ directly, which is in and of itself
a difficult problem, we will take advantage of the fact that we are
only considering $b\in\BHParamNbhd$, where $\BHParamNbhd$ is a small
neighborhood of black hole parameters of $b_0$, to prove a mode
stability statement for $\LinEinstein_{g_b}$ perturbatively.

To this end, we consider an alternative proof for the linear stability
of the gauged Einstein equations linearized around $g_{b_0}$ by
showing how to use the $\LSolHk{k}$-quasinormal modes to generate an
appropriate linearized generalized harmonic gauge for the equation, in
which the contribution of the non-decaying $\LSolHk{k}$-quasinormal
modes disappears. 
\begin{prop}
  \label{linear:prop:Linear-Stability:SdS:robust}
  For $b'\in T_{b_0}B$, let $\omega_{b_0}^\Upsilon(b')$ denote the
  solution of the Cauchy problem
  \begin{equation*}
    \begin{split}
      \begin{cases}
        \Constraint_{g_{b_0}}\circ\nabla_{g_{b_0}}\otimes\omega_{b_0}^\Upsilon(b')
        = -\Constraint_{g_{b_0}}(g_{b_0}'(b')) & \text{in }\mathcal{M}\\
        \InitData(\omega_{b_0}^\Upsilon(b'))=(0,0)&\text{on }\Sigma_0.
      \end{cases} 
    \end{split}
  \end{equation*}
  Define
  \begin{equation*}
    (g_{b_0}')^\Upsilon(b') := g_{b_0}'(b') +\nabla_{g_{b_0}}\otimes \omega_{b_0}^\Upsilon(b'),
  \end{equation*}
  which solves the linearized gauged Einstein equation
  $\LinEinstein_{g_{b_0}}((g_{b_0}')^\Upsilon(b'))=0$. Fix
  $k>k_0$, and let $\SpectralGap>0$ be sufficiently small. 

  Then there exists a finite-dimensional linear subspace 
  \begin{equation*}
    \Theta\subset C^\infty_0 (\StaticRegionWithExtension, T^*\StaticRegionWithExtension),
    \qquad \Upsilon \subset \LSolHk{k}(\Sigma_0)
  \end{equation*}
  such that the following holds: for any $(f, {h}_0, {h}_1)\in
  D^{k+1,\SpectralGap}(\Omega)$ there exist unique $b'\in T_{b_0}B$,
  $\vartheta\in\Theta$, and $\upsilon\in \Upsilon$,
  such that the solution $\tilde{{h}}$ of the forward problem
  \begin{equation*} 
    \begin{cases}
    \LinEinstein_{g_{b_0}}\tilde{{h}} = f, & \text{in }\StaticRegionWithExtension,\\
    \InitData(\tilde{{h}})=({h}_0, {h}_1)  -
    \InitData((g_{b_0}')^\Upsilon(b') +\nabla_{g_{b_0}}\otimes \vartheta) - \upsilon,
    & \text{on }\Sigma_0
    \end{cases}    
  \end{equation*}
  satisfies the pointwise bound
  \begin{equation*}
    \norm{\tilde{{h}}}_{\HkWithT{k}(\Sigma_{\tStar})} \lesssim e^{-\SpectralGap\tStar}\norm{(f, {h}_0, {h}_1)}_{D^{k+1,\SpectralGap}(\StaticRegionWithExtension)},
  \end{equation*}
  and moreover, the
  map $(f, {h}_0, {h}_1)\mapsto (b', \vartheta, \upsilon)$ is linear and
  continuous.  Finally, if $({h}_0, {h}_1)$ satisfy the linearized
  constraint conditions, then in fact $\upsilon=0$. 
\end{prop}

\begin{proof}
  Consider the set of $\LSolHk{k}$-quasinormal mode solutions of
  $\LinEinstein_{g_{b_0}}$ with non-negative imaginary part
  $\QNMk{k}(\LinEinstein_{g_{b_0}}, \UpperHalfSpace)$. Recall from
  Theorem \ref{linear:thm:meromorphic:main-A} and Theorem
  \ref{linear:thm:resolvent-estimate:inf-gen} that there are only finitely
  many such $\LSolHk{k}$-quasinormal modes.
  Let us first recall that all $\LSolHk{k}$-geometric
  quasinormal mode solutions $\psi$ satisfy the
  linearized gauge constraint 
  \begin{equation*}
    \Constraint_{g_{b_0}}\psi = 0.
  \end{equation*}
  Moreover, since all all $\LSolHk{k}$- quasinormal mode solutions of
  $\LinEinstein_{g_{b_0}}$ satisfying the linearized gauge constraint are
  also mode solutions of the ungauged linearized Einstein equation, we
  know from Theorem \ref{linear:thm:Kodama-Ishibashi} that all
  $\LSolHk{k}$-quasinormal mode solutions of $\LinEinstein_{g_{b_0}}$
  satisfying the linearized gauge constraint are in fact
  $\LSolHk{k}$-geometric quasinormal modes.  On the other hand, since
  the constraint modes of $\LinEinstein_{g_{b_0}}$ are exactly those modes
  of $\LinEinstein_{g_{b_0}}$ which do not satisfy the linearized gauge
  condition, all the quasinormal modes of $\LinEinstein_{g_{b_0}}$ must
  either be an $\LSolHk{k}$-quasinormal geometric mode  solution or an
  $\LSolHk{k}$-quasinormal constraint mode solution. We can then
  characterize the $\LSolHk{k}$-geometric quasinormal mode solutions
  and the $\LSolHk{k}$-constraint quasinormal mode solutions by:
  \begin{equation*}
    \begin{split}
      \QNMGeok{k}(\LinEinstein_{g_{b_0}},\UpperHalfSpace)
      &= \curlyBrace*{{h}\in\QNMk{k}(\LinEinstein_{g_{b_0}}, \UpperHalfSpace):\Constraint_{g_{b_0}}({h}) = 0},\\
      \QNMConstraintk{k}(\LinEinstein_{g_{b_0}},\UpperHalfSpace)
      &= \curlyBrace*{{h}\in\QNMk{k}(\LinEinstein_{g_{b_0}}, \UpperHalfSpace):\Constraint_{g_{b_0}}({h}) \neq 0}.
    \end{split}
  \end{equation*}
  It is clear from the preceding discussion that
  \begin{equation*}
    \QNMk{k}(\LinEinstein_{g_{b_0}},\UpperHalfSpace)
    = \QNMGeok{k}(\LinEinstein_{g_{b_0}}, \UpperHalfSpace)+ \QNMConstraintk{k}(\LinEinstein_{g_{b_0}}, \UpperHalfSpace).
  \end{equation*}
  We index the $\LSolHk{k}$-geometric quasinormal frequencies,
  $\curlyBrace*{\sigma_{j}^\Upsilon}_{j=1}^{N_{\Upsilon}}$ such that
  $\sigma_1=0$. For each $\sigma_j^\Upsilon$, $j\ge 2$, we fix a basis
  $\curlyBrace{\nabla_{g_{b_0}}\otimes\vartheta_{j\ell}}_{\ell=1}^{N^\Upsilon_j}$
  of $\QNMGeok{k}(\LinEinstein_{g_{b_0}}, \sigma_j^\Upsilon)$,
  and define
  \begin{equation*}
    \Theta_{\neq 0}:= \Span\curlyBrace{\vartheta_{j\ell}:2\le j\le N_{b_0}^\Upsilon, 1\le \ell\le N^{\Upsilon}_j}.
  \end{equation*}
  Next, index the constraint $\LSolHk{k}$-quasinormal frequencies
  $\{\sigma_j^\Constraint\}_{i=1}^{N_{\Constraint}}$. Then for each
  $\sigma_j^{\Constraint}$, we fix a basis
  $\curlyBrace{\upsilon_{j\ell}}_{\ell=1}^{N^{\Constraint}_j}$, and
  define
  \begin{equation*}
    \Upsilon:=\Span\curlyBrace*{\upsilon_{j\ell}: 1\le j\le N^{\Constraint}_{b_0}, 1\le\ell\le N_j^{\Constraint}}.
  \end{equation*}

  
  
  It remains to deal with the geometric $\LSolHk{k}$-quasinormal mode
  solutions at $0$. At the zero quasinormal mode, we need to separate
  the gauge modes coming from a linearized \KdS{} metric and the
  remaining modes arising from an infinitesimal diffeomorphism.
  
  To do so, recall that we defined $\omega_{b_0}^\Upsilon(b')$ as the
  solution of a Cauchy problem for
  \begin{equation*}
    \Box^{\Upsilon}_{g_{b_0}} = \Constraint_{g_{b_0}}\circ\nabla_{g_{b_0}}\otimes.
  \end{equation*}
  We can calculate using the definition of
  $\Constraint,\nabla_{g_{b_0}}\otimes $, that $\Box^{\Upsilon}_{g_{b_0}}$ is
  principally $\ScalarWaveOp[g_{b_0}]$ and is equal to
  $\VectorWaveOp[g_{b_0}]$ up to an order-zero
  perturbation\footnote{We can actually calculate that
    \begin{equation*}
      \Box^{\Upsilon}_{g_{b_0}} = \VectorWaveOp[g_{b_0}] - \Lambda.
    \end{equation*}
  }.  As a result, $\Box^{\Upsilon}_{g_{b_0}}$ is a strongly
  hyperbolic operator with a well-defined $\LSolHk{k}$-quasinormal
  spectrum just like $\LinEinstein_{g_{b_0}}$. In fact, from Lemma
  \ref{linear:lemma:subprincipal-symbol-control}, we see that it also has the
  desired pseudo-differential smallness at $\TrappedSet_{b_0}$ that
  would allow us to prove the existence of a spectral gap for
  $\Box^{\Upsilon}_{g_{b_0}}$. As a result,
  \begin{equation*}
    \omega_{b_0}^\Upsilon = \omega_0 + \omega_{-} + \omega_{+},
  \end{equation*}
  where $\omega_{-}$ decays exponentially,
  $\omega_{+}\in \QNMk{k}(\Box^{\Upsilon}_{g_{b_0}},
  \UpperHalfSpaceExcZero)$, and
  $\omega_0$ is the contribution of the
  zero-mode. Then
  \begin{equation*}
    (g_{b_0}')^\Upsilon(b')^{(0)} = g_{b_0}'(b') +\nabla_{g_{b_0}}\otimes \omega_0 
  \end{equation*}
  satisfies $\LinEinstein_{g_{b_0}}((g_{b_0}')^\Upsilon(b')^{(0)}) =
  0$, and $(g_{b_0}')^\Upsilon(b')^{(0)} \in
  \QNMGeok{k}(\LinEinstein_{g_{b_0}}, 0)$. 
  
  We now define
  \begin{equation*}
    K := \curlyBrace*{(g_{b_0}')^\Upsilon(b')^{(0)}:b'\in T_{b_0}B}\subset \QNMGeok{k}(\LinEinstein_{g_{b_0}}, 0),
  \end{equation*}
  which has dimension $d_0= 4$. Then
  $\QNMGeok{k}(\LinEinstein_{g_{b_0}}, 0)\backslash K$ has a basis of
  the form
  $\curlyBrace*{\nabla_{g_{b_0}}\otimes\vartheta_\ell:1\le\ell\le d_1}$,
  where $d_1=\dim\QNMGeok{k}(\LinEinstein_{g_{b_0}}, 0)-d_0$. Define
  \begin{equation*}
    \Theta_0 := \Span\curlyBrace*{\vartheta_\ell:1\le\ell\le d_1},\quad \Theta = \Theta_0\oplus\Theta_{\neq0}.
  \end{equation*}
  We then observe that
  \begin{equation*}
    \LinEinstein_{g_{b_0}}\left((g_{b_0}')^\Upsilon(b')-(g_{b_0}')^\Upsilon(b')^{(0)}\right) = 0.
  \end{equation*}
  As a result,
  \begin{equation*}
    (g_{b_0}')^\Upsilon(b')-(g_{b_0}')^\Upsilon(b')^{(0)} \in\nabla_{g_{b_0}}\otimes \Theta + O(e^{-\SpectralGap \tStar}).
  \end{equation*}
  We can now define the space of initial value problem
  modifications,
  \begin{equation*}
    \ZCal :=
    (0, \gamma_0((g_{b_0}')^\Upsilon(T_{b_0}B)))
    +(0, \gamma_0(\nabla_{g_{b_0}}\otimes\Theta)) + (0, \gamma_0(\Upsilon)) \subset D^{\infty,\SpectralGap}(\StaticRegionWithExtension, S^2T^*\StaticRegionWithExtension).
  \end{equation*}
  The map $\lambda_{\ZCal}$ as defined in Corollary
  \ref{linear:coro:lambda-IVP} is then bijective by
  construction. $\lambda_{\ZCal}$ is surjective since any
  non-decaying asymptotic behavior in a solution of
  $\LinEinstein_{g_{b_0}}{h} = (f,{h}_0,{h}_1)$ can be removed by
  modifying $(f', {h}_0', {h}_1') = (f,{h}_0,{h}_1) + z$
  for some $z\in\ZCal$. By construction
  $\lambda_{\ZCal}$ is also injective as the dimension of $\ZCal$ is
  at most as large as the space of $\LSolHk{k}$-quasinormal mode solutions
  $\QNMk{k}(\LinEinstein_{g_{b_0}}, \UpperHalfSpace)$.
  We can also observe directly that in the case that
  $({h}_0, {h}_1)$ solve the linearized constraint equation, we in
  fact have that $\Constraint_{g_{b_0}}({h}) = 0$, and thus, there are
  no contributions from the $\LSolHk{k}$-quasinormal constraint
  modes, turning the family of initial value problem modifications
  simply into a family of gauge modifications, as desired. The
  desired estimates follow from the estimates in Corollary \ref{linear:coro:lambda-IVP}. 
\end{proof}

With Proposition \ref{linear:prop:Linear-Stability:SdS:robust}, we can now
offer an alternative proof of Theorem
\ref{linear:thm:lin-stability:SdS}.
\begin{proof}[Alternative proof of Theorem \ref{linear:thm:lin-stability:SdS}]
  Using Proposition \ref{linear:prop:Linear-Stability:SdS:robust}, we know
  that there exists some $b'\in T_{b_0}B$, $\vartheta\in\Theta$ such
  that the solution $\tilde{h}$ of the Cauchy problem
  \begin{equation*}
    \begin{split}
      \LinEinstein_{g_{b_0}}{h} &=  0\\
      \gamma_0({h}) &= i_{b_0, \Identity}(\InducedMetric', k') - \gamma_0((g_{b_0}')^{\Upsilon}(b') +\nabla_{g_{b_0}}\otimes \vartheta)
    \end{split}
  \end{equation*}
  satisfies the decay bound
  \begin{equation*}
    \norm{\tilde{h}}_{\HkWithT{k}(\Sigma_{\tStar})} \lesssim e^{-\SpectralGap\tStar}\left(
      \norm{g'}_{H^{k+1}(\Sigma_0)} + \norm{k'}_{H^{k}(\Sigma_0)}
    \right).
  \end{equation*}
  Moreover, by construction,
  \begin{equation*}
    \LinEinstein_{g_{b_0}}\left((g_{b_0}')^{\Upsilon}(b') +\nabla_{g_{b_0}}\otimes \vartheta\right) = 0.
  \end{equation*}
  As a result,
  \begin{equation*}
    h = (g_{b_0}')^{\Upsilon}(b') +\nabla_{g_{b_0}}\otimes \vartheta + \tilde{h}
  \end{equation*}
  solves 
  \begin{equation*}
    \begin{split}
      \LinEinstein_{g_{b_0}}h &=  0,\\
      \gamma_0(h) &= i_{b_0, \Identity}(\InducedMetric', k'),
    \end{split}
  \end{equation*}
  and moreover since $\Constraint_{g_{b_0}}(h) = 0$ uniformly on $\StaticRegionWithExtension$, we
  must have that actually, $h$ also solves the ungauged
  linearized Einstein equation and has
  the desired form in Theorem~\ref{linear:thm:lin-stability:SdS}.
\end{proof}

\subsection{Perturbation of the $\LSolHk{k}$-quasinormal spectrum to
  \KdS}
\label{linear:sec:KdS-QNM-perturb}

We now wish to show that Theorem \ref{linear:thm:lin-stability:SdS} holds not
only for $g_{b_0}$, but also for $g_b$ sufficiently
slowly-rotating. Unfortunately, in the \KdS{} case, we do not have a
statement of geometric mode stability like we did in the \SdS{}
case. Therefore, as previously mentioned, the direct proof of linear
stability of \SdS{} using Theorem \ref{linear:thm:Kodama-Ishibashi} and
Theorems \ref{linear:thm:meromorphic:main-A} and
\ref{linear:thm:resolvent-estimate:inf-gen} cannot be adjusted to extend to
the slowly-rotating \KdS{} case. What saves us is the alternative
method of proving the linear stability of \SdS{} in Proposition
\ref{linear:prop:Linear-Stability:SdS:robust} which uses only the fact that
all non-decaying $\LSolHk{k}$-quasinormal modes of
$\LinEinstein_{g_{b_0}}$ are either $\LSolHk{k}$-geometric or
$\LSolHk{k}$-constraint quasinormal modes. The main goal of this
section will be to utilize perturbative properties of the spectrum of
strongly hyperbolic linear operators to deduce the that all
non-decaying $\LSolHk{k}$-quasinormal modes of $\LinEinstein_{g_b}$
are also either geometric or constraint modes. From there, it will
follow that we can repeat the proof of Proposition
\ref{linear:prop:Linear-Stability:SdS:robust}, and as a result, attain a
proof for the linear stability of the gauged Einstein's equations
linearized around $g_b$.
\begin{prop}[Mode stability of $\LinEinstein_{g_b}$, version 2]
  \label{linear:prop:KdS-QNM-perturb}
  The set of $\LSolHk{k}$-geometric quasinormal frequencies of
  $\LinEinstein_{g_b}$, and $\LSolHk{k}$-constraint quasinormal frequencies of
  $\LinEinstein_{g_b}$ are continuous in $b$ in the Hausdorff distance
  sense, and in particular, for a sufficiently small neighborhood
  $\BHParamNbhd$ of $b_0$, 
  \begin{equation*}
    \begin{split}
      \dim \QNMGeok{k}(\LinEinstein_{g_{b_0}},
      \UpperHalfSpace)
      &= \dim \QNMGeok{k}(\LinEinstein_{g_b},
      \UpperHalfSpace),\\
      \dim \QNMConstraintk{k}(\LinEinstein_{g_{b_0}},
      \UpperHalfSpace)
      &= \dim \QNMConstraintk{k}(\LinEinstein_{g_b},
      \UpperHalfSpace),
    \end{split}
  \end{equation*}
  for all $b\in\BHParamNbhd$. As a result, for $b\in \BHParamNbhd$,
  all $\LSolHk{k}$-quasinormal mode solutions of $\LinEinstein_{g_b}$ are
  either $\LSolHk{k}$-geometric or $\LSolHk{k}$-constraint quasinormal
  mode solutions.
\end{prop}

\begin{proof}
  The main difficulty in applying perturbation theory to
  $\LinEinstein_{g_b}$ directly is that while $\LinEinstein_{g_b}$ is strongly
  hyperbolic, the perturbation theory established in Proposition
  \ref{linear:prop:QNM-perturb:gen} does not distinguish between
  geometric and constraint modes. Neither does it prevent the
  introduction or destruction of geometric or constraint modes.  The
  main tool that allows us to circumvent this in the proof of the
  current proposition is the observation that we can identify the
  $\LSolHk{k}$-geometric quasinormal modes of $\LinEinstein_{g_b}$ with
  the $\LSolHk{k}$-quasinormal modes of the constraint propagation
  operator $\ConstraintPropagationOp_{g_{b}}$. Since
  $\ConstraintPropagationOp_{g_b}$ is a strongly hyperbolic operator,
  we can apply perturbation theory results to
  $\ConstraintPropagationOp_{g_b}$, which gives us a perturbative way
  of treating the $\LSolHk{k}$-geometric quasinormal modes of
  $\LinEinstein_{g_b}$. 
  
  We begin by relating the non-zero $\LSolHk{k}$-quasinormal modes of
  $\LinEinstein_{g_b}$ with the non-zero $\LSolHk{k}$-quasinormal modes of
  $\ConstraintPropagationOp_{g_b}$. Define
  \begin{align*}
    \UpperHalfSpace_\epsilon := \curlyBrace*{\sigma\in\Complex:\abs*{\sigma}>\epsilon, \Im\sigma > -\epsilon},\qquad
    \mathbb{B}_\epsilon := \curlyBrace*{\sigma\in \Complex: |\sigma|<\epsilon},
  \end{align*}
  where $\epsilon$ is chosen so that
  \begin{equation*}
    \p\UpperHalfSpace_{\epsilon}\bigcap \QNMk{k}(\LinEinstein_{g_b}) =
    \emptyset,\qquad
    \mathbb{B}_{\epsilon} \bigcap \QNMk{k}(\LinEinstein_{g_b}) = \{0\}.
  \end{equation*}
  
  For any
  $\LSolHk{k}$-quasinormal mode
  $\omega = e^{-\ImagUnit\sigma\tStar}\omega_\sigma$ of
  $\ConstraintPropagationOp_{g_b}$,
  \begin{equation*}
    \ConstraintPropagationOp_{g_b}\omega = 2\nabla_{g_b}\cdot\TraceReversal_{g_b}\nabla_{g_b}\otimes  \omega = 0. 
  \end{equation*}
  Recall that $\Constraint_{g_b} =
  -\nabla_{g_b}\cdot\TraceReversal_{g_b}$. As a result, for ${h} =
  \nabla_{g_b}\otimes \omega$, we have that
  \begin{equation*}
    \LinEinstein_{g_b}{h} = 0,
  \end{equation*}
  and ${h}$ is a geometric mode of $\LinEinstein_{g_b}$. Thus, we can
  consider the mapping
  \begin{equation*}
    \begin{split}
      \QNMk{k}(\ConstraintPropagationOp_{g_b}, \UpperHalfSpace_\epsilon )&\mapsto \QNMGeok{k}(\LinEinstein_{g_b}, \UpperHalfSpace_\epsilon ),\\
      \omega &\mapsto \nabla_{g_b}\otimes \omega. 
    \end{split}
  \end{equation*}
  This map is clearly injective, as if $\omega_1, \omega_2 \in
  C^{\infty}(\StaticRegionWithExtension,
  T^*\StaticRegionWithExtension)$ satisfy
  \begin{equation*}
    \nabla_{g_b}\otimes (\omega_1-\omega_2) = 0,
  \end{equation*}
  then $\omega_1-\omega_2$ is Killing, but all the Killing vectors on
  \KdS{} are stationary, and we assumed that
  $\omega_1,\omega_2\not\in
  \QNMGeok{k}(\ConstraintPropagationOp_{g_b}, 0)$.  Now consider some
  ${h} = \nabla_{g_b}\otimes \omega\in \QNMGeok{k}(\LinEinstein_{g_b},
  \UpperHalfSpace_\epsilon )$. Then by definition, we have that
  $\nabla_{g_b}\otimes \omega$ satisfies the linearized gauge constraint,
  and in fact,
  $\omega\in \QNMk{k}(\ConstraintPropagationOp_{g_b},
  \UpperHalfSpace_\epsilon )$. As a result, we can define the
  injective map
  \begin{equation*}
    \begin{split}
      \QNMGeok{k}(\LinEinstein_{g_b}, \UpperHalfSpace_\epsilon )&\mapsto \QNMk{k}(\ConstraintPropagationOp_{g_b}, \UpperHalfSpace_\epsilon ),\\
      \nabla_{g_b}\otimes \omega &\mapsto \omega.
    \end{split}
  \end{equation*}
  As a result, we have a bijection between
  $\QNMk{k}(\ConstraintPropagationOp_{g_b}, \UpperHalfSpace_\epsilon )$
  and $\QNMGeok{k}(\LinEinstein_{g_b}, \UpperHalfSpace_\epsilon )$. This
  crucially gives a method for counting the non-zero $\LSolHk{k}$-geometric quasinormal
  modes of $\LinEinstein_{g_b}$, and we have that
  \begin{equation}
    \label{linear:eq:QNM-perturb:non-zero-geometric}
     \dim \QNMGeok{k}(\LinEinstein_{g_{b_0}}, \UpperHalfSpace_\epsilon ) = \dim \QNMGeok{k}(\LinEinstein_{g_b}, \UpperHalfSpace_\epsilon ).
   \end{equation}
  
  We now treat the zero-frequency $\LSolHk{k}$-geometric quasinormal
  mode solutions.
  Observe
  that for $b = (M, a)$ the mapping
  \begin{equation*}
    \begin{split}
      \QNMk{k}(\ConstraintPropagationOp_{g_b}, \mathbb{B}_\epsilon) &\to \QNMGeok{k}(\LinEinstein_{g_b}, \mathbb{B}_\epsilon),  \\
      \omega &\mapsto \nabla_{g_b}\otimes \omega
    \end{split}
  \end{equation*}
  is a linear map with kernel spanned by the Killing vectorfields of
  $g_b$. Recall that the Killing vectorfields of $g_b$ have basis
  $\KillT, \KillPhi$ if $b=(M, a), a\neq 0$, and basis
  $\KillT, \Omega_{ij}$ if $b=(M, 0)$, where $\Omega_{ij}$ are the
  rotation vectorfields.

  Applying the linearized second Bianchi
  identity, the map
  \begin{equation*}
     \begin{split}
       \curlyBrace*{{h} \in \QNMGeok{k}(\LinEinstein_{g_b}, \mathbb{B}_\epsilon):
         {h} = \nabla_{g_{b}}\otimes\omega, \omega\in C^\infty(\StaticRegionWithExtension, T^*\StaticRegionWithExtension)}
       &\mapsto \QNMk{k}(\ConstraintPropagationOp_{g_b}, \mathbb{B}_\epsilon ) \backslash \Ker\,\nabla_{g_b}\otimes ,\\
      \nabla_{g_b}\otimes \omega &\mapsto \omega,
    \end{split}
  \end{equation*}
  is injective. 
  Thus the family of maps
  \begin{equation*}
    \begin{split}
      \QNMk{k}(\ConstraintPropagationOp_{g_b}, \mathbb{B}_\epsilon) \backslash \Ker\,\nabla_{g_b}\otimes 
      &\to \curlyBrace*{{h} \in \QNMGeok{k}(\LinEinstein_{g_b}, \mathbb{B}_\epsilon):
        {h} = \nabla_{g_{b}}\otimes\omega, \omega\in C^\infty(\StaticRegionWithExtension, T^*\StaticRegionWithExtension)},  \\
      \omega &\mapsto \nabla_{g_b}\otimes \omega
    \end{split}
  \end{equation*}
  are bijective for all $b\in \BHParamNbhd$. 
  Furthermore, defining 
  \begin{equation*}
    d_b = \dim \frac{g_b'(T_bB)}{(\Range\, \nabla_{g_b}\otimes )\bigcap g_b'(T_bB)},
  \end{equation*}
  we have that
  \begin{equation*}
    d_b =
    \begin{cases}
      4, & a = 0,\\
      2, & a \neq 0.
    \end{cases}
  \end{equation*}
  As a result,
  \begin{equation*}
    \dim \QNMGeok{k}(\LinEinstein_{g_b}, 0) = \dim \QNMk{k}(\ConstraintPropagationOp_{g_b}, 0) - \dim \Ker\, \nabla_{g_b}\otimes  + d_b = \dim \QNMk{k}(\ConstraintPropagationOp_{g_b}, 0). 
  \end{equation*}
  We now apply Proposition \ref{linear:prop:QNM-perturb:gen} to see that for
  $b\in \BHParamNbhd$, $\BHParamNbhd$ sufficiently small, 
  \begin{equation*}
    \dim \QNMk{k}(\ConstraintPropagationOp_{g_{b_0}}, 0) = \dim \QNMk{k}(\ConstraintPropagationOp_{g_b}, 0).
  \end{equation*}
  As a result, we in fact have that
  \begin{equation*}
    \dim \QNMGeok{k}(\LinEinstein_{g_b}, 0) = \dim \QNMGeok{k}(\LinEinstein_{g_{b_0}}, 0). 
  \end{equation*}
  Combined with \eqref{linear:eq:QNM-perturb:non-zero-geometric}, we have
  that
  \begin{equation*}
    \dim \QNMGeok{k}(\LinEinstein_{g_b}, \UpperHalfSpace) = \dim \QNMGeok{k}(\LinEinstein_{g_{b_0}}, \UpperHalfSpace). 
  \end{equation*}

  We now move on to consider the constraint modes. To begin, observe
  that both $\Constraint_{g_b}$ and
  $\QNMk{k}(\LinEinstein_{g_b}, \UpperHalfSpace )$ are continuous
  in $b$. Then, we can use the lower-semicontinuity of rank to see
  that for a sufficiently small neighborhood of black hole parameters
  $\BHParamNbhd\ni b_0$,
  \begin{equation*}
    \dim \Constraint_{g_b}\left(\QNMk{k}(\LinEinstein_{g_b},\UpperHalfSpace)\right)
    \ge \dim \Constraint_{g_{b_0}}\left(\QNMk{k}(\LinEinstein_{g_{b_0}},\UpperHalfSpace)\right).
  \end{equation*}
  We can now conclude 
  that all $\LSolHk{k}$-quasinormal mode solutions of
  $\LinEinstein_{g_b}$ are either $\LSolHk{k}$-geometric or
  $\LSolHk{k}$-constraint quasinormal mode solutions.
  The continuity statement holds directly by applying Proposition
  \ref{linear:prop:QNM-perturb:gen} to $\ConstraintPropagationOp_{g_b}$ and
  $\LinEinstein_{g_b}$. 
\end{proof}

Having shown that all $\LSolHk{k}$-quasinormal modes of
$\LinEinstein_{g_b}$ are either $\LSolHk{k}$-geometric quasinormal modes
or $\LSolHk{k}$-constraint quasinormal modes, we can now directly
repeat the proof of Proposition \ref{linear:prop:Linear-Stability:SdS:robust}
with $g_b$ in place of $g_{b_0}$. 

\begin{prop}
  \label{linear:prop:linear-stab:KdS-robust}
  Fix $b\in \BHParamNbhd$. Then for $b'\in T_{b}B$, let
  $\omega_{b}^\Upsilon (b')$ denote the solution of the Cauchy problem
  \begin{equation*}
    \begin{cases}
      \left(\Constraint_{g_{b}}\circ\nabla_{g_{b}}\otimes\right)\omega_{b}^\Upsilon(b')
      = -\Constraint_{g_{b}}(g_{b}'(b')) & \text{in }\mathcal{M},\\
      \InitData(\omega_{b}^\Upsilon(b'))=(0,0)&\text{on }\Sigma_0.
    \end{cases}
  \end{equation*}
  Define
  \begin{equation*}
    (g_b')^\Upsilon (b') := g_b'(b') + \nabla_{g_b}\otimes \omega_{b}^\Upsilon(b'),
  \end{equation*}
  which solves the linearized gauged Einstein equation
  \begin{equation*}
    \LinEinstein_{g_b}((g_b')^\Upsilon (b')) = 0. 
  \end{equation*}
  Now fix some $k>k_0$. Then there exists  some small $\SpectralGap>0$
  such that there exist finite dimensional linear subspaces
  \begin{equation*}
    \Theta\subset H^{k,-\GronwallExp}(\StaticRegionWithExtension, T^*\StaticRegionWithExtension),
    \qquad \Upsilon \subset \LSolHk{k}(\Sigma_0)
  \end{equation*}
  such that the following holds: for any $(f, {h}_0, {h}_1)\in
  D^{k+1,\SpectralGap}(\StaticRegionWithExtension)$ there exist unique $b'\in T_{b}\BHParamNbhd$,
  $\vartheta\in\Theta$, and $\upsilon\in \Upsilon$,
  such that the solution of the Cauchy problem
  \begin{equation*}
    \begin{split}
      \LinEinstein_{g_b}\tilde{{h}} &= f,\\
      \InitData(\tilde{{h}})&=({h}_0, {h}_1) + \upsilon - \InitData\left((g_{b}')^\Upsilon(b') + \nabla_{g_b}\otimes \vartheta\right).  
    \end{split}
  \end{equation*}
  satisfies
  $\tilde{{h}}\in H^{k,\SpectralGap}(\StaticRegionWithExtension)$,
  the map $(f, {h}_0, {h}_1)\mapsto (b', \vartheta, \upsilon)$ is
  linear and continuous, and $\tilde{{h}}$ satisfies the pointwise
  bound
  \begin{equation*}
    \norm{\tilde{h}}_{\HkWithT{k}(\Sigma_{\tStar})} \lesssim e^{-\SpectralGap\tStar}\norm{(f, {h}_0, {h}_1)}_{D^{k+1,\SpectralGap}(\StaticRegionWithExtension)}.
  \end{equation*}
  Finally, if $f=0$, and moreover, $({h}_0, {h}_1)$ satisfy the
  linearized gauge constraint, then in fact $\upsilon=0$.
\end{prop}

\begin{proof} 
  From Proposition \ref{linear:prop:KdS-QNM-perturb}, we know that all
  $\LSolHk{k}$-quasinormal modes of $\LinEinstein_{g_b}$ are either
  geometric or constraint modes, we see that we can
  characterize the $\LSolHk{k}$-geometric quasinormal modes of
  $\LinEinstein_{g_b}$ and $\LSolHk{k}$-constraint quasinormal modes of
  $\LinEinstein_{g_b}$ by
  \begin{equation*}
    \begin{split}
      \QNMGeok{k}(\LinEinstein_{g_b}, \UpperHalfSpace) &= \curlyBrace*{h\in \QNMk{k}(\LinEinstein_{g_b}, \UpperHalfSpace): \Constraint_b({h})=0},\\
      \QNMConstraintk{k}(\LinEinstein_{g_b}, \UpperHalfSpace) &= \curlyBrace*{h \in \QNMk{k}(\LinEinstein_{g_b}, \UpperHalfSpace): \Constraint_b({h})\neq 0}. 
    \end{split}
  \end{equation*}
  This allows us to effectively repeat the proof of
  Proposition \ref{linear:prop:Linear-Stability:SdS:robust}, using the
  $\LSolHk{k}$-quasinormal modes of $\LinEinstein_{g_b}$ in place of those
  of $\LinEinstein_{g_{b_0}}$.
  We construct the span of all gauge choices that correspond to
  the non-zero $\LSolHk{k}$-quasinormal geometric mode solutions,
  \begin{equation*}
    \Theta_{\neq 0} := \Span\curlyBrace*{\vartheta_{j\ell}:
      \nabla_{g_b}\otimes \vartheta_{j\ell} \in \QNMGeok{k}(\LinEinstein_{g_b}, \sigma_j), \sigma_j\in \QNFGeok{k}(\LinEinstein_{g_b}, \UpperHalfSpace_\epsilon )};
  \end{equation*}
  the span of the zero-frequency gauge choices that correspond to
  the non-linearized \KdS{} zero-frequency $\LSolHk{k}$-geometric
  quasinormal mode solutions,
  \begin{equation*}
    \Theta_0:= \Span\curlyBrace*{ \vartheta_{\ell}:
     \nabla_{g_b}\otimes \vartheta_\ell\in \QNMGeok{k}(\LinEinstein_{g_b}, 0)\backslash \curlyBrace*{(g_b')^\Upsilon (b')^{(0)}}};
  \end{equation*}
  and the space of possible initial-data modifications corresponding
  to the $\LSolHk{k}$-constraint quasinormal mode solutions,
  \begin{equation*}
    \Upsilon_b := \Span\curlyBrace*{\gamma_0(\psi_{j\ell}): \psi_{j\ell}\in \QNMConstraintk{k}(\LinEinstein_{g_b}, \UpperHalfSpace)}.
  \end{equation*}
  Just as was the case for Proposition
  \ref{linear:prop:Linear-Stability:SdS:robust}, we now have that defining
  $\Theta = \Theta_0\oplus \Theta_{\neq0}$, we have the space
  of initial value problem modifications,
  \begin{equation*}
    \ZCal := (0, \gamma_0((g_{b}')^\Upsilon(b'))
    +(0, \gamma_0(\nabla_{g_b}\otimes \Theta)) + (0, \Upsilon_b) \subset D^{\infty,\SpectralGap}(\StaticRegionWithExtension, S^2T^*\StaticRegionWithExtension),
  \end{equation*}
  for which the map $\lambda_{\ZCal}$ as defined in Corollary
  \ref{linear:coro:lambda-IVP} is bijective.  The conclusions of the
  proposition then follow from Corollary \ref{linear:coro:lambda-IVP}. It is
  also clear from the construction that if $({h}_0, {h}_1)$ satisfy
  the linearized gauge constraint, then $\upsilon=0$. We also remark
  that applying the perturbation theory in Proposition
  \ref{linear:prop:QNM-perturb:gen} to $\ConstraintPropagationOp_{g_b}$, it
  is clear that $\Theta$ can be constructed to depend continuously
  on the black hole parameters $b$.
\end{proof}

Observe that as constructed in Proposition
\ref{linear:prop:linear-stab:KdS-robust}, the space $\Theta$ from which we
can construct $\LSolHk{k}$-geometric quasinormal mode solutions is
dependent on the \KdS{} metric $g_b$ around which we linearize. While
this is perfectly fine for the linearized stability statement, it will
be desirable in the nonlinear theory to show that in fact, we can
choose $\Theta$ independent of our choice of $g_b$.
\begin{corollary}
  \label{linear:corollary:linear-stab:Theta-b-independent}
  Let $(g'_b)^\Upsilon(b')$ be as defined in Proposition
  \ref{linear:prop:linear-stab:KdS-robust}, and fix some $k>k_0$, and
  $\SpectralGap>0$. Let $\Upsilon_b$ denote the $\Upsilon$ constructed
  for the choice of \KdS{} black hole parameters $b$, as in
  Proposition \ref{linear:prop:linear-stab:KdS-robust}.

  Then there exists a fixed finite dimensional linear subspace
  $\Theta\subset H^{k,-\GronwallExp}(\StaticRegionWithExtension,
  T^*\StaticRegionWithExtension)$ and a sufficiently small
  neighborhood $\BHParamNbhd$ of $b_0$ such that for any $b\in \BHParamNbhd$,
  for any $(f, {h}_0, {h}_1)\in
  D^{k+1,\SpectralGap}(\StaticRegionWithExtension)$ there exist unique $b'\in T_{b}B$,
  $\vartheta\in\Theta$, and $\upsilon\in \Upsilon_b$,
  such that the solution of the Cauchy problem
  \begin{equation*}
    \begin{split}
      \LinEinstein_{g_b}\tilde{{h}} &= f,\\
      \InitData(\tilde{{h}})&=({h}_0, {h}_1) + \upsilon - \InitData\left(g_{b}'(b') + \nabla_{g_b}\otimes \vartheta\right).  
    \end{split}
  \end{equation*}
  satisfies the pointwise
  bound
  \begin{equation*}
    \norm{\tilde{{h}}}_{\HkWithT{k}(\Sigma_{\tStar})} \lesssim e^{-\SpectralGap\tStar}\norm{(f, {h}_0, {h}_1)}_{D^{k+1,\SpectralGap}(\StaticRegionWithExtension)},
  \end{equation*}
  and moreover, the map $(f, {h}_0, {h}_1)\mapsto (b', \vartheta, \upsilon)$ is
  linear and continuous. In addition, if $({h}_0, {h}_1)$ satisfy the linearized gauge
  constraint, then in fact $\upsilon=0$. 
\end{corollary}
\begin{proof}
  Consider $\tilde{\ZCal}_b\subset
  D^{k,\SpectralGap}(\StaticRegionWithExtension)$ defined by
  \begin{equation*}
    \tilde{\ZCal}_b = T_bB \times \Theta_{b_0}\times \Upsilon_b.
  \end{equation*}
  We will prove the corollary by showing that 
  \begin{align*}
    \lambda_{\tilde{\ZCal}_b}:=\lambda_{IVP}\vert_{\tilde{\ZCal}_b}: \tilde{\ZCal}_b\to \QNFk{k*}(\LinEinstein_{g_{b_0}}, \UpperHalfSpace)
  \end{align*}
  is an isomorphism.
  
  We begin by parameterizing the space $\tilde{\ZCal}_b$ by picking an isomorphism
  \begin{equation*}
    \varTheta:\Real^{N_\Theta} \to \Theta_{b_0},\qquad N_{\Theta} = \dim \Theta_{b_0},
  \end{equation*}
  which induces an isomorphism 
  \begin{align*}
    \tilde{z}_b: \Real^{4 + N_{\Theta} + N_{\Upsilon}} &\to \tilde{\ZCal}_b\\
    (b', w_\Theta, \upsilon) &\mapsto \left(0, \gamma_0\left((g_b')^\Upsilon(b') + \nabla_{g_b}\otimes \varTheta(w_\Theta)\right) + \upsilon\right),
  \end{align*}
  and is clearly continuous in $b$. 
  Then, we see that the mapping
  \begin{equation*}
    \lambda_{\tilde{\ZCal}_{b}}\circ \tilde{z}_b^{-1} 
  \end{equation*}
  depends continuously on $b$. Moreover, from Proposition
  \ref{linear:prop:Linear-Stability:SdS:robust}, we know that
  $\lambda_{\tilde{\ZCal}_{b_0}}\circ \tilde{z}_{b_0}^{-1}$ is
  bijective, and as a result, for $b\in\BHParamNbhd$ sufficiently
  small, $\lambda_{\tilde{\ZCal}_{b}}\circ \tilde{z}_b^{-1}$ is
  bijective as well by semicontinuity of rank\footnote{ This can be
    clearly seen by using Proposition \ref{linear:prop:QNM-perturb:gen}, to
    parametrize the family of spaces
    $\QNMk{k*}(\LinEinstein_{g_b}, \UpperHalfSpace)$ by means of a
    continuous map
    \begin{equation*}
      \varLambda: B\times\Real^D \to \QNMk{k*}(\LinEinstein_{g_b}, \UpperHalfSpace)
    \end{equation*}
    linear in the second argument such that
    \begin{equation*}
      \varLambda(b_0, \cdot) = \lambda_{\ZCal_{b_0}}\circ \tilde{z}_{b_0}.
    \end{equation*}
    With the parametrization $\varLambda$,
    \begin{equation*}
      \varLambda(b, \cdot)^{-1}\circ\lambda_{\tilde{\ZCal}_{b}}\circ \tilde{z}_b^{-1} 
    \end{equation*}
    can be represented as a complex $4+N_{\Theta}+N_{\Upsilon}$ sized
    square matrix.
  }.
  The rest of the conclusions follow as before. 
\end{proof}

\section{Proof of the main theorem}
\label{linear:sec:proof-of-main-thm}

We can now prove Theorem \ref{linear:thm:Main}.
\begin{proof}[Proof of Theorem \ref{linear:thm:Main}]
  Let us denote
  $({h}_0, {h}_1)=D_{(\InducedMetric_b,
    k_b)}i_{b,\Identity}(\InducedMetric', k')$ the initial data for
  the gauged linearized Einstein equations in harmonic gauge. By the
  construction of $D_{(\InducedMetric_b, k_b)}i_{b,\Identity}$ in
  Corollary~\ref{linear:corollary:initial-data:ib-linearized}, we have that
  $(h_0, h_1)$ satisfies the linearized harmonic gauge constraint on
  $\Sigma_0$. Thus, using Proposition
  \ref{linear:prop:linear-stab:KdS-robust}, there exists some
  $b'\in T_b\BHParamNbhd, \vartheta \in \Theta$ satisfying the control
  in \eqref{linear:eq:Main:param-continuity} such that the solution ${h}$ of
  the Cauchy problem
  \begin{equation*}
    \begin{split}
      \LinEinstein_{g_b}\tilde{h} &= 0,\\
      \gamma_0(\tilde{h}) &= ({h}_0, {h}_1) - \gamma_0\left(g_{b}'(b') + \nabla_{g_b}\otimes \vartheta\right),
    \end{split}
  \end{equation*}
  satisfies the decay estimate
  \begin{equation*}
    \begin{split}
      \sup_{\tStar }e^{\SpectralGap \tStar}\norm{\tilde{h}}_{\HkWithT{k}(\Sigma_{\tStar})} &\lesssim \norm{({h}_0, {h}_1)}_{\LSolHk{k+1}(\Sigma_0)}.
    \end{split}
  \end{equation*}  
  Then,
  \begin{equation*}
    h = g_{b}'(b') + \nabla_{g_b}\otimes \vartheta + \tilde{h} 
  \end{equation*}
  solves
  \begin{equation*}
    \begin{split}
      \LinEinstein_{g_b}h &= 0,\\
      \gamma_0(h) &= (h_0, h_1),
    \end{split}
  \end{equation*}
  and has the desired form.   
\end{proof}

\appendix
\section{Background functional analysis}

\begin{theorem}[Analytic Fredholm Theorem] \label{linear:thm:AnalyticFredholm}
  Let $\Omega\subset \Complex$ be a domain in the complex plane. Let 
  $H$ be a Hilbert space, and $\mathcal{L}(H)$ the space of bounded
  linear operators from $H$ to itself. Also let $\iota$ denote the
  identity operator. Then, for $A:\Omega\to \mathcal{L}(H)$ a mapping such
  that
  \begin{enumerate}
  \item $A$ is analytic on $\Omega$ in the following sense:
    \begin{equation*}
      \lim_{\sigma\to\sigma_0}\frac{A(\sigma)-A(\sigma_0)}{\sigma-\sigma_0}
    \end{equation*}
    exists for all $\sigma_0\in \Omega$; and
  \item the operator $A(\sigma)$ is a compact operator for every
    $\sigma\in\Omega$. 
  \end{enumerate}
  Then we can conclude that either
  \begin{enumerate}
  \item $(\iota -A(\sigma))^{-1}$ does not exist for any
    $\sigma\in\Omega$; or
  \item There exists a discrete subset $\Sigma\subset\Omega$  (in
    particular, $\Sigma$ has no limit points in $\Omega$), such that
    $(\iota -A(\sigma))^{-1}$ exists for every $\sigma\in
    \Omega\backslash\Sigma$, and the function
    $\lambda\to(\iota-A(\sigma))^{-1}$ is analytic on
    $\Omega\backslash\Sigma$ in the sense above and the equation
    \begin{equation*}
      A(\sigma)\psi =\psi
    \end{equation*}
    has a finite-dimensional family of solutions for any $\sigma\in
    \Sigma$. 
  \end{enumerate}
\end{theorem}

\begin{lemma}[Lemma 1.3 \cite{engel_one-parameter_2000}]
  \label{linear:lemma:semigroup:resolvent-sol-op-relation}
  Let $(D^k(\InfGen), \InfGen)$ be the infinitesimal generator for the
  $C^0$-solution semigroup $\SolOp(\tStar)$. Then for 
  $\sigma\in\Complex, \tStar>0$, we have that for all $F\in
  \LSolHk{k}(\Sigma)$, 
  \begin{equation*}
    \ImagUnit\left(e^{\ImagUnit\sigma\tStar}\SolOp(\tStar)F - F\right)  = (\InfGen-\sigma)\int_0^{\tStar} e^{\ImagUnit\sigma\sStar}\SolOp(\sStar)F\,d\sStar;
  \end{equation*}
  and for all ${h}\in D^k(\InfGen)$,
  \begin{equation*}
    \ImagUnit\left(e^{\ImagUnit\sigma\tStar}\SolOp(\tStar){h} - {h}\right) = \int_0^{\tStar} e^{\ImagUnit\sigma\sStar}\SolOp(\sStar)(\InfGen - \sigma){h}\,d\sStar.
  \end{equation*}
\end{lemma}

The Mather Division Theorem introduced below is also referred to as
the Malgrange Preparation Theorem. 
\begin{theorem}[Mather Division Theorem \cite{mather_stability_1968}]
  \label{linear:thm:Mather-division}
  If $f(t,x)$ is a smooth complex function of $t\in \Real$, and $x\in
  \Real^n$ near the origin, and $k$ is the smallest integer such that
  \begin{equation*}
    \frac{d^k f}{d t^k} \neq 0, 
  \end{equation*}
  and moreover, $g$ is a smooth function near the origin, then we can
  write
  \begin{equation*}
    g = qf + r,
  \end{equation*}
  where $q$ and $r$ are smooth, and
  \begin{equation*}
    r= \sum_{0\le j<k }t^jr_j(x)
  \end{equation*}
  for some family of smooth functions $r_j(x)$. 
\end{theorem}

\section{Appendix to Section \ref{linear:sec:Setup}}

\subsection{Proof of Lemma
  \ref{linear:lemma:SdS:Kerr-star-regular-coordinates}}
\label{linear:appendix:lemma:SdS:Kerr-star-regular-coordinates}

Define the cutoff functions
  $\chi_{\Interval_{b_0},-}, \chi_{\Interval_{b_0},+}\in
  C_0^\infty(r_{\EventHorizonFuture}-\varepsilon_{\StaticRegionWithExtension},
  r_{\CosmologicalHorizonFuture}+\varepsilon_{\StaticRegionWithExtension})$
  such that 
  \begin{equation*}
    \chi_{\Interval_{b_0}, -}(r) =
    \begin{cases}
      1 & \text{near }\EventHorizonFuture\\
      0 & r>\inf\Interval_{b_0}
    \end{cases},\qquad
    \chi_{\Interval_{b_0}, +}(r) =
    \begin{cases}
      1 & \text{near }\CosmologicalHorizonFuture\\
      0 & r<\sup\Interval_{b_0}
    \end{cases}.
  \end{equation*}  
  Now we choose
  \begin{equation}
    \label{linear:eq:coordinate-transform:SdS:regular-global}
    F_{b_0}'(r) = \left(\chi_{\Interval_{b_0},+}-\chi_{\Interval_{b_0}, -}\right)\mu_{b_0}^{-1}\sqrt{1-c_{\tStar}^2\mu_{b_0}(r)},
  \end{equation}
  where
  $c_{\tStar} = \left(1 - \sqrt[3]{9\Lambda
      M^2}\right)^{-\frac{1}{2}}$. $F_{b_0}'$ clearly vanishes on
  $\Interval_{b_0}$. To check that the constant-$\tStar$ hypersurfaces
  are uniformly spacelike, we can compute that
  \begin{equation*}
    G_{b_0}(d\tStar, d\tStar) = \frac{1}{\mu_{b_0}}\left(-1 + \left(\chi_{\Interval_{b_0},+}-\chi_{\Interval_{b_0}, -}\right)^2\right)
    -\left(\chi_{\Interval_{b_0},+}-\chi_{\Interval_{b_0}, -}\right)^2c_{\tStar}^2,
  \end{equation*}
  which is uniformly negative.   
\subsection{Proof of Lemma \ref{linear:lemma:KdS:Kerr-star-regular-coordinates}}
\label{linear:appendix:lemma:KdS:Kerr-star-regular-coordinates}
We will construct an explicit coordinate system that satisfies the
conditions in the lemma.  To define $F'_b(r), \Phi_b'(r)$, fix
$\Interval_{b}=r_1, r_2$ such that
$r_{b_0,
  \EventHorizonFuture}+\epsilon_{\StaticRegionWithExtension}<r_1<r_2<r_{b_0,\CosmologicalHorizonFuture}
- \epsilon_{\StaticRegionWithExtension}$, and define smooth cutoff
functions
$\chi_{\Interval_b,\EventHorizonFuture},
\chi_{\Interval_b,\CosmologicalHorizonFuture}$ such that
\begin{equation*}
  \chi_{\Interval_b,\EventHorizonFuture}(r) :=
  \begin{cases}
    1&\text{near } \EventHorizonFuture \\
    0&r > r_1 
  \end{cases},\qquad 
  \chi_{\Interval_b,\EventHorizonFuture}(r) :=
  \begin{cases}
    1&\text{near } \CosmologicalHorizonFuture \\
    0&r < r_2 
  \end{cases}.
\end{equation*}
Then we define $F_b'(r), \Phi_b'(r)$ so that
\begin{align*}
  F_b'(r) &= F_{b_0}'(r)
  + (\chi_{\Interval_b,\EventHorizonFuture} - \chi_{\Interval_b,\CosmologicalHorizonFuture})\left(
    \mu_{b_0}^{-1} - \frac{(1+\lambda_b)(r^2+a^2)}{\Delta_b}
            \right),\\
  \Phi_b'(r) &= (\chi_{\Interval_b,\EventHorizonFuture} - \chi_{\Interval_b,\CosmologicalHorizonFuture})\left(
               \frac{(1+\lambda_b)a}{\Delta_b}
               \right).
\end{align*}
This definition ensures that $F_b$, $\Phi_b$ are well-defined up to an
additive constant. Furthermore, it is immediately clear that $F_b(r) =
\bPhi_b(r) = 0$ on $\Interval_b$, and that $F_{b}', \Phi_{b}'$ coincide
with $F_{b_0}', \Phi_{b_0}'$ when $b=b_0$. Finally, to check that the
$\tStar$-constant hypersurfaces are space-like, it suffices to observe
that 
\begin{equation*}
  G_b(d\tStar,d\tStar) = G_{b_0}(d\tStar, d\tStar) + O(a),
\end{equation*}
where the perturbation in particular follows from the fact that $r$ is
bounded. Since we have that $G_{b_0}(d\tStar, d\tStar)<0$ uniformly,
$\tStar$-constant hypersurfaces remain space-like in $g_b$ for
sufficiently small $a$.

Near the poles, we can use the smooth coordinates on $\Sphere^2$ given
by $x = \sin\theta\cos\phiStar$ and $y=\sin\theta\sin\phiStar$ to
extend the metric smoothly to the poles.

\subsection{Proof of Proposition \ref{linear:prop:redshift:N-construction}}
\label{linear:appendix:prop:redshift:N-construction}

We define the redshift vectorfield $\RedShiftY$ such that on
$\Horizon=\EventHorizonFuture, \CosmologicalHorizonFuture$,
\begin{enumerate}
\item $\RedShiftY$ is future-directed null, with $g(\RedShiftY,
  \HorizonGen_{\Horizon}) = -2$,
\item $\nabla_\RedShiftY \RedShiftY = -C_\RedShiftY(\RedShiftY+\HorizonGen_{\Horizon})$,
\item $\LieDerivative_\KillT \RedShiftY = \LieDerivative_\KillPhi
  \RedShiftY = 0$.
\end{enumerate}
Define a local frame $\omega_1$, $\omega_2$ on the spheres. Then at
$\EventHorizonFuture$ for some vectorfield $a^A$ and two-tensor
$\tensor[]{h}{_A^B}$,
\begin{equation}
  \begin{split}
    \nabla_{\HorizonGen_{\EventHorizonFuture}}\HorizonGen_{\EventHorizonFuture} &= \SurfaceGravity_{\EventHorizonFuture}\HorizonGen_{\EventHorizonFuture},\\
    \nabla_{\RedShiftY}\RedShiftY &= -C_\RedShiftY (\RedShiftY+\HorizonGen_{\EventHorizonFuture}),\\
    \nabla_{\HorizonGen_{\EventHorizonFuture}}\RedShiftY  &= -\SurfaceGravity_{\EventHorizonFuture} \RedShiftY + a^A\omega_A,\\
    \nabla_{\omega_A}\RedShiftY &= \tensor[]{h}{_A^B}\omega_B - \frac{1}{2}a^A\RedShiftY.
  \end{split}
\end{equation}
We can then calculate that at $\EventHorizonFuture$, 
\begin{gather*}
  \DeformationTensor[_{\RedShiftY\RedShiftY}]{\RedShiftY} = 2C_\RedShiftY, \quad
  \DeformationTensor[_{\HorizonGen_{\EventHorizonFuture}\HorizonGen_{\EventHorizonFuture}}]{\RedShiftY}
  = 2\SurfaceGravity_{\EventHorizonFuture}, \quad
  \DeformationTensor[_{\HorizonGen_{\EventHorizonFuture}\RedShiftY}]{\RedShiftY} = C_\RedShiftY, \\
  \DeformationTensor[_{\RedShiftY \omega_A}]{\RedShiftY} = 0, \quad
  \DeformationTensor[_{\HorizonGen_{\EventHorizonFuture} \omega_A}]{\RedShiftY} = a^A,\quad
  \DeformationTensor[_{\omega_A\omega_B}]{\RedShiftY} = \tensor[]{h}{_A^B}.
\end{gather*}
Then, at $\EventHorizonFuture$, we have that
\begin{align*}
  \KCurrent{\RedShiftY, 0, 0}[{h}]
  \ge{} &\frac{1}{2}\SurfaceGravity
          _{\EventHorizonFuture} \EMTensor(\RedShiftY,\RedShiftY)
          + \frac{1}{4}C_\RedShiftY \EMTensor(\HorizonGen_{\EventHorizonFuture}, \RedShiftY + \HorizonGen_{\EventHorizonFuture})\\
        &- c\EMTensor(\HorizonGen_{\EventHorizonFuture}, \RedShiftY + \HorizonGen_{\EventHorizonFuture})
          - c\sqrt{\EMTensor(\HorizonGen_{\EventHorizonFuture}, \RedShiftY + \HorizonGen_{\EventHorizonFuture})\EMTensor(\RedShiftY, \RedShiftY)}.
\end{align*}
Thus, for any $\epsilon>0$, for $C_\RedShiftY$ sufficiently large,  we
have that on $\EventHorizonFuture$,
\begin{equation}
  \KCurrent{\RedShiftY, 0, 0}[{h}] \ge \frac{1}{2}(\SurfaceGravity_{\EventHorizonFuture} - \epsilon)\EMTensor(\RedShiftY,\RedShiftY) + c(\epsilon)\EMTensor(\HorizonGen_{\EventHorizonFuture}, \RedShiftY + \HorizonGen_{\EventHorizonFuture})
\end{equation}
for some $c(\epsilon)>0$.  In particular, for any fixed $C>0$, we can
choose $C_\RedShiftY$ sufficiently large so that along
$\EventHorizonFuture$, 
\begin{equation}
  \label{linear:eq:redshiftY:deformTensor:Event:final-estimate}
  \KCurrent{\RedShiftY, 0, 0}[{h}]
  \ge \frac{1}{4}\SurfaceGravity_{\CosmologicalHorizonFuture}\EMTensor(\RedShiftY,\RedShiftY)
  + C\EMTensor(\HorizonGen_{\CosmologicalHorizonFuture}, \RedShiftY + \HorizonGen_{\CosmologicalHorizonFuture}). 
\end{equation}
A similar argument shows that $C_{\RedShiftY}$ can be chosen
sufficiently large so that in fact
\eqref{linear:eq:redshiftY:deformTensor:Event:final-estimate} also holds at
$\CosmologicalHorizonFuture$. 
We then define
\begin{equation*}
  \RedShiftN =
  \begin{cases}
    \HorizonGen_{\EventHorizonFuture}+\RedShiftY& r\le r_0\\
    \KillT & r\in (r_1, R_1)\\
    \HorizonGen_{\CosmologicalHorizonFuture}+\RedShiftY& r\ge R_0,
  \end{cases}
\end{equation*}
where on the regions $(r_0, r_1)$, $(R_1, R_0)$ we smoothly extend
$\RedShiftN$ as a $\KillT$, $\KillPhi$-independent timelike
vectorfield. Then we have
\eqref{linear:eq:redshift:DeformTen-redshift-weak-control},
\eqref{linear:eq:redshift:N=T}, and \eqref{linear:eq:redshift:divergence} by
construction. 

We now move onto proving
\eqref{linear:eq:redshift:DeformTen-redshift-control}.
Recalling the definition of $\TFixer$ in \eqref{linear:eq:TFixer:def}, and
\eqref{linear:eq:redshiftY:deformTensor:Event:final-estimate}, we have that on
$\EventHorizonFuture$,
\begin{equation}
  \label{linear:eq:redshiftY:deformTensor:Event:aux2}
  \KCurrent{\RedShiftY, 0, 0}[h]
  \ge \frac{1}{2}\left(\SurfaceGravity_{\EventHorizonFuture}-\varepsilon_{\RedShiftN}\right)\EMTensor(\RedShiftY,\RedShiftY)
  + C(\varepsilon_{\RedShiftN})\EMTensor(\TFixer, \RedShiftN),
\end{equation}
where we observe that $C(\varepsilon_{\RedShiftN})$ can be made
arbitrarily large by increasing $C_\RedShiftY$ as desired in the
construction of $\RedShiftY$.

Now observe that at the horizons, the unit normal $n_{\Sigma_{\tStar}}$
can be decomposed as
\begin{equation*}
  n_{\Sigma_{\tStar}} = \frac{1}{2\sqrt{\GInvdtdt}}\RedShiftY + \tilde{n},
\end{equation*}
where $\tilde{n}$ is tangent to $\EventHorizonFuture$ and
$\CosmologicalHorizonFuture$.

As a result, we have that on a neighborhood of the horizons, in
particular, for $r<r_0$, $r>R_0$,
\begin{equation*}
  \KCurrent{\RedShiftY, 0, 0}[h]
  \ge \frac{1}{\sqrt{\GInvdtdt}}\left(\SurfaceGravity_{\EventHorizonFuture} - \varepsilon_{\RedShiftN}\right)
  \JCurrent{\RedShiftN, 0, 0}[h]\cdot n_{\Sigma}
  + C(\varepsilon_{\RedShiftN})\EMTensor(\TFixer, \RedShiftN).
\end{equation*}
Now recall that $\TFixer$ is timelike on the interior of the
stationary region, only becoming null exactly at
$\EventHorizonFuture$, and that $\RedShiftN$ is globally
timelike. Then, $\EMTensor(\TFixer, \RedShiftN)$ controls all
tangential derivatives of $h$ along the horizons, and all derivatives
of $h$ away from the horizons. In particular, since $X$ is tangential
to both $\EventHorizonFuture$ and $\CosmologicalHorizonFuture$,
\begin{equation*}
  \abs*{X h }^2 \lesssim \EMTensor(\TFixer, \RedShiftN).
\end{equation*}
Choosing $C_{\RedShiftY}$ sufficiently large so that
$C(\varepsilon_{\RedShiftN})\EMTensor(\TFixer, \RedShiftN)$ is large enough to control $Xh$ then
allows us to conclude.

The fact that the divergence of $\RedShiftN$ is negative follows
directly from the computation of the components of
$\DeformationTensor{\RedShiftY}$ above.





\subsection{Proof of Lemma
  \ref{linear:lemma:enhanced-redshift:Ka-construction}}

\label{linear:appendix:lemma:enhanced-redshift:Ka-construction}

We first pick stationary, smooth $\RedShiftK_1$ such that $g(\RedShiftK_1,\HorizonGen_{\Horizon})=-1$ near $\Horizon$, for
$\Horizon=\EventHorizonFuture, \CosmologicalHorizonFuture$ and
vanishes on a neighborhood away from the horizons. 
Now consider some
$\Sigma=\Sigma_{\tStar}$. Any coordinate chart $(U, \Psi)$ on $\Sigma$
can be pushed forward to a tubular coordinate patch on
$\StaticRegionWithExtension$ by the map
\begin{align*}
  \Real_+\times U&\to \StaticRegionWithExtension\\
  (\tStar, x) &\mapsto \varphi_{\tStar}\circ \Psi(x). 
\end{align*}
In such a coordinate chart, the metric coefficients are independent
of $\tStar$. Now, pick an arbitrary point
$\STPoint=(\tStar,x)\in \Sigma$. We now split into two cases, a case
where we are not on the horizons, and a case where we are on one of
$\EventHorizonFuture$ and $\CosmologicalHorizonFuture$. 
\begin{enumerate}
\item
  $\STPoint\in \Sigma
  \backslash(\EventHorizonFuture\bigcup\CosmologicalHorizonFuture)$:
  In this case, we pick a coordinate chart $(U,\Psi)$ such that
  $\STPoint\in \Psi(U)\Subset \Sigma
  \backslash(\EventHorizonFuture\bigcup\CosmologicalHorizonFuture)$. Now
  define $\chi$ a cut-off function such that $\chi=1$ in a neighborhood of
  $p$ and $\supp \chi \subset U$. Then define the vectorfields
  \begin{equation*}
    \RedShiftK^{(\STPoint)}_\mu:= \chi \p_\mu 
  \end{equation*}
  on the tubular patch $\Real_+\times U$, where
  $\p_0=\HorizonGen_{\EventHorizonFuture}+\HorizonGen_{\CosmologicalHorizonFuture},
  \p_i=\p_{x^i}$. Each of the $\RedShiftK^{\STPoint}_\mu$ are
  stationary vectorfields since the coordinate chart is stationary,
  and they are also trivially tangent to both $\EventHorizonFuture$
  and to $\CosmologicalHorizonFuture$. Finally, since all the
  vectorfields themselves vanish near the two horizons, we also have
  that
  \begin{equation*}
    \DeformationTensor{X^{(p)}_\gamma} = 0
  \end{equation*}
  near both $\EventHorizonFuture$ and $\CosmologicalHorizonFuture$.
\item
  $\STPoint\in \EventHorizonFuture\bigcup
  \CosmologicalHorizonFuture$. The argument for both cases are
  similar, so we only give details in the case of
  $\STPoint\in \CosmologicalHorizonFuture$: We can pick a chart
  $(U, \Psi)$ on $\Sigma$ such that $\STPoint\in\Psi(U)$, and
  moreover, $\CosmologicalHorizonFuture$ is given locally by
  $x_1=0$, and $\p_{x^1} = \RedShiftY$ is normal to
  $\CosmologicalHorizonFuture$. Like in the previous case, define a
  cut-off function $\chi$ such that $\chi=1$ near $\STPoint$ and
  $\supp\chi \in U$. Then, in the tubular neighborhood
  $\Real_+\times U$, define
  \begin{equation*}
    \RedShiftK^{(\STPoint)}_0 = \chi \HorizonGen_{\CosmologicalHorizonFuture},
    \qquad \RedShiftK^{(\STPoint)}_1 = \chi x_1\p_{x^1},
    \qquad \RedShiftK^{(\STPoint)}_i = \chi \p_{x^i}\quad 2\le i 
  \end{equation*}
  By construction, all of these vectors are smooth and tangent to
  $\CosmologicalHorizonFuture$. Moreover, they vanish on and thus
  are tangent to $\EventHorizonFuture$.
\end{enumerate}

At this point, we have constructed a tubular neighborhood
$V_{(\STPoint)}$ around each point $\STPoint\in\Sigma$ and a set of
smooth vectorfields which are stationary and tangent to both
$\CosmologicalHorizonFuture$ and $\EventHorizonFuture$. Moreover,
for $\RedShiftK$ a smooth vectorfield tangent to both
$\CosmologicalHorizonFuture$ and $\EventHorizonFuture$, supported in
$V_{(\STPoint)}$, there exist stationary
$(x^{(\STPoint)})^\mu\in C^\infty(\StaticRegionWithExtension)$ such that
\begin{equation*}
  \RedShiftK= (x^{(\STPoint)})^\mu \RedShiftK^{(\STPoint)}_\mu.
\end{equation*}
Now, crucially, we realize that $\Sigma$ is compact in \KdS. This
allows us to choose a finite set of points
$\{\STPoint_i\in \Sigma\}_{i=1}^N$ such that
\begin{equation*}
  \StaticRegionWithExtension\subset \bigcup_{i=1}^N V_{(\STPoint_i)}.
\end{equation*}
Now re-index $\{\RedShiftK_i\}_{i=2}^N$ to be the union of the
$\{\RedShiftK^{(p_i)}_\mu\}_{i=1}^N$. This is now a finite set of
vectorfields such that for any
$X\in C^\infty(\StaticRegionWithExtension)$ there exist
functions $x^i\in C^\infty(\StaticRegionWithExtension)$ such that
\begin{equation*}
  X = \sum_i x^i\RedShiftK_i.
\end{equation*}
It remains to confirm the properties of the deformation
tensors. Directly from our construction of $\RedShiftY$, we have
that
\begin{equation*}
  f_{1}^{11} = \SurfaceGravity_{\mathcal{H}} \qquad \text{near }\mathcal{H},
\end{equation*}
where $\mathcal{H} =
\EventHorizonFuture,\CosmologicalHorizonFuture$.
Finally,
\begin{equation*}
  f_i^{11} = 0,\quad i\neq 1
\end{equation*}
follows directly from that $\RedShiftK_i,i\neq 1$ is tangent to
$\EventHorizonFuture$ and $\CosmologicalHorizonFuture$,  and that
\begin{equation*}
  \DeformationTensor{X} =\frac{1}{2}\LieDerivative_Xg.
\end{equation*}

  

\subsection{Proof of Lemma \ref{linear:lemma:T-Fixer:construction}}
\label{linear:appendix:lemma:T-Fixer:construction}
We can directly calculate that
\begin{equation*}
  g_b(\KillT + a\tilde{\chi}\KillPhi, \KillT + a\tilde{\chi}\KillPhi)
  = g_b(\KillT, \KillT)
  + 2a\tilde{\chi}g_b(\KillT, \KillPhi)
  + a^2\tilde{\chi}^2 g_b(\KillPhi, \KillPhi).
\end{equation*}
Using the quadratic formula, we can calculate that $ g_b(\KillT +
a\tilde{\chi}\KillPhi, \KillT + a\tilde{\chi}\KillPhi)=0$ when
\begin{equation*}
  \tilde{\chi} = \frac{-2a g_b(\KillT, \KillPhi)\pm \sqrt{4a^2\left(g(\KillT, \KillPhi) - g(\KillPhi, \KillPhi)g(\KillT, \KillT)\right)}}{2 g_b(\KillT, \KillT)}.
\end{equation*}
Observe that
\begin{equation*}
  g(\KillT, \KillPhi) - g(\KillPhi, \KillPhi)g(\KillT, \KillT)\ge 0,
\end{equation*}
with equality exactly on the horizons $\EventHorizonFuture$ and
$\CosmologicalHorizonFuture$. We define  $\tilde{\chi}$ by
\begin{equation*}
  a\tilde{\chi}(r) :=
  \begin{cases}
    \HorizonGenPhi_{\EventHorizonFuture}& r\in [r_{\EventHorizonFuture}, \sup_{\Ergoregion_{\EventHorizonFuture}}r],\\
    \HorizonGenPhi_{\CosmologicalHorizonFuture}& [\inf_{\Ergoregion_{\CosmologicalHorizonFuture}}r, r_{\CosmologicalHorizonFuture}],\\
    0& r\in [r_-,r_+],
  \end{cases}
\end{equation*}
where
\begin{equation}
  \label{linear:eq:r-minus-r-plus-def}
  \sup_{\Ergoregion_{\EventHorizonFuture}}r < r_- < r_0,\qquad
  R_0 < r_+ < \inf_{\Ergoregion_{\CosmologicalHorizonFuture}}r.
\end{equation}
It is then clear that $\TFixer$ defined with $\tilde{\chi}$ has the
desired properties.

\section{Appendix to Section \ref{linear:sec:EVE}}

\subsection{Proof of Proposition \ref{linear:prop:initial-data:ib-construction}}
\label{linear:appendix:prop:initial-data:ib-construction}

Let us define $\psi = \phi^{-1}$. Then, denote
\begin{equation*}
  \tilde{\InducedMetric}_0 = \psi^*\InducedMetric_0, \qquad \tilde{k}_0 = \psi^*k_0.
\end{equation*}

Fix $N_b$ and $X_b$ the lapse function and the shift vectorfield
respectively of $g_b$, so that
\begin{equation*}
  g_b = -N_b^2\,d\tStar^2 + (\InducedMetric_b)_{ij}(d x^i + X_b^i\,d\tStar)\otimes (dx^j + X_b^j\,d\tStar).
\end{equation*}
Then define $g $ so that
\begin{equation*}
  g = -N_b^2\,d\tStar^2 + (\tilde{\InducedMetric}_0)_{ij}(d x^i + X_b^i\,d\tStar)\otimes (dx^j + X_b^j\,d\tStar).
\end{equation*}
It remains to define $\p_{\tStar}g_{\mu\nu}$. We first define 
\begin{equation*}
  \p_{\tStar}g_{ij} = \LieDerivative_{X_b}\tilde{\InducedMetric}_0  + 2N_{b} (\tilde{k}_0)_{ij}.
\end{equation*}
It remains to determine $\p_{\tStar}g(\KillT, \cdot)$. We do
so exactly by using the gauge condition. 
Recall that in local
coordinates, the gauge condition can be written as
\begin{equation*}
  0 = g^{\alpha\beta}\p_\alpha g_{\mu\beta}
  - \frac{1}{2}g^{\alpha\beta}\p_\mu g_{\alpha\beta}
  - g_{\mu\gamma}g^{\alpha\beta}\ChristoffelTypeTwo[g_b]{\gamma}{\alpha\beta} .
\end{equation*}
Contracting this with $\KillT$, we have then that
\begin{equation}
  \frac{1}{2}g^{\tStar\tStar}\p_{\tStar}g_{\tStar\tStar}
  = - g^{i\beta}\p_i g_{\tStar\beta}
  + \frac{1}{2}g^{i j}\p_{\tStar}g_{ij}
  + g_{\tStar\gamma}g^{\alpha\beta}\ChristoffelTypeTwo[g_b]{\gamma}{\alpha\beta},
\end{equation}
which uniquely determines $\p_{\tStar}g_{\tStar\tStar}$.
On the other hand, contracting with $\p_i$, we have that
\begin{equation*}
  g^{\tStar\tStar}\p_{\tStar}g_{i\tStar} =
  \frac{1}{2}g^{\alpha\beta}\p_ig_{\alpha\beta}
  - g^{\alpha j}\p_{j}g_{i\alpha}
  +  g_{i\gamma}g^{\alpha\beta}\ChristoffelTypeTwo[g_b]{\gamma}{\alpha\beta},
\end{equation*}
which uniquely determines $\p_{\tStar}g_{\tStar i}$.
We can then define
\begin{equation*}
  i_{b,\phi}(\InducedMetric, k) = \left((g - g_b)\vert_{\Sigma_0}, \LieDerivative_{\KillT}g\vert_{\Sigma_0}\right).
\end{equation*}
Denoting $(\InducedMetric, k)$ the induced metric and second
fundamental form of $g$ on $\phi(\Sigma_0)$, it is clear that
$(\phi^*\InducedMetric, \phi^*k)=(g_0, k_0)$, and moreover, we have
constructed $(g\vert_{\Sigma_0}, \LieDerivative_\KillT
g\vert_{\Sigma_0})$ satisfying the gauge constraint. The
smoothness and mapping properties follow by construction.  

\subsection{Proof of Lemma \ref{linear:lemma:SubPOp:horizons}}
\label{linear:appendix:lemma:SubPOp:horizons}

Throughout this section, we work with the Eddington-Finkelstein
coordinates on \SdS{} defined in Section \ref{linear:sec:KdS}.  We will then
denote covectors in this section by
\begin{equation*}
  -\sigma dt_0+ \xi dr + \FreqTheta d\theta + \FreqPhi d\varphi_0.
\end{equation*}
Then we can write the metric $g_{b_0}$ and the inverse metric
$G_{b_0}$ as
\begin{equation}
  \label{linear:eq:coordinate-transform:SdS-outgoing-EF}
  \begin{split}
    g_{b_0} &= -\mu_{b_0} dt_0^2 \pm 2dt_0dr + r^2d\omega^2,\\
    G_{b_0} &= \mp 2\p_{t_0}\p_r + \mu \p_r^2 + r^{-2}\UnitSphereInvMetric
    = \pm 2\sigma\xi + \mu_{b_0} \xi^2 + r^{-2}|\eta|^2,\quad \eta\in T^*\UnitSphere^2,
  \end{split}
\end{equation}
where $\pm(r-\rCrit)>0$. Recall that in this system of coordinates, we
have that $\p_{t_0}$ is Killing, and in particular is tangential to
the two Killing horizons. Let us now define
\begin{equation*}
  e_0:=\sqrt{\mu_{b_0}}^{-1}\p_t,
  \quad e^0:= \sqrt{\mu_{b_0}}dt,
  \quad e_1:= \sqrt{\mu_{b_0}}\p_r,
  \quad e^1:= \sqrt{\mu_{b_0}}^{-1} dr. 
\end{equation*}
To achieve a smooth splitting near  $r_{b_0,\EventHorizonFuture},
r_{b_0,\CosmologicalHorizonFuture}$, we split any smooth one-form $w$
into
\begin{equation*}
  w = w^0_N \,dt_0 + w^0_{TN}\, dr + w^0_{TT},
\end{equation*}
where $w^0_N, w^0_{TN}$ are smooth, and $w^0_{TT}$ is a smooth
one-form on the 2-sphere $\Sphere^2$. This transforms between the
splitting
\begin{equation*}
  w = w_N\,e^0 + w_{TN}\,e^1+w_{TT}
\end{equation*}
via
\begin{equation*}
  \begin{pmatrix}[1.5]
    w_N\\ w_{TN}\\ w_{TT}
  \end{pmatrix}
  = \mathcal{J}_{\pm}^{(1)}
  \begin{pmatrix}[1.5]
    w^0_N\\ w^0_{TN} \\ w^0_{TT}
  \end{pmatrix},
  \quad
  \mathcal{J}_{\pm}^{(1)} =
  \begin{pmatrix}[1.5]
    \sqrt{\mu_{b_0}}^{-1} & 0 & 0\\
    \mp \sqrt{\mu_{b_0}}^{-1} & \sqrt{\mu_{b_0}} & 0 \\
    0 & 0 & 1
  \end{pmatrix}.
\end{equation*}
This induces a splitting on two tensors $u$ as
\begin{equation}
  \label{linear:eq:subprincipal-calculation:horizon-splitting}
  u = u^0_{NN}\,dt_0^2 + 2u^0_{NTN}\,dt_0dr + 2\,dt_0 u^0_{NTT} + u^0_{TNN}\,dr^2 +  2\,dr u^0_{TNT} + u^0_{TTT}, 
\end{equation}
where $u^0_{NN}, u^0_{NTN}, u^0_{TNN}$ are functions on
$T^*\StaticRegionWithExtension$, $u^0_{NTT}$ and $u^0_{TNT}$ are one
forms on the sphere, and $u^0_{TTT}$ is a two-tensor on the
sphere. This then transforms between the splitting in
\eqref{linear:eq:subprincipal-calculation:horizon-splitting} via
\begin{equation*}
  \begin{pmatrix}[1.5]
    u_{NN} \\ u_{NTN} \\ u_{NTT} \\ u_{TNN} \\ u_{TNT} \\ u_{TTT}
  \end{pmatrix}
  =
  \mathcal{J}_{\pm}^{(2)}
  \begin{pmatrix}[1.5]
    u^0_{NN} \\ u^0_{NTN} \\ u^0_{NTT} \\ u^0_{TNN} \\ u^0_{TNT} \\ u^0_{TTT}
  \end{pmatrix},
  \quad
  \mathcal{J}_{\pm}^{(2)} =
  \begin{pmatrix}[1.5]
    \mu_{b_0}^{-1} & 0 & 0 & 0 & 0 & 0\\
    \mp \mu_{b_0}^{-1} & 1 & 0 & 0 & 0 & 0\\
    0 & 0 & \sqrt{\mu_{b_0}}^{-1} & 0 & 0 & 0\\
    \mu_{b_0}^{-1} & \mp 2 & 0 & \mu_{b_0} & 0 & 0\\
    0 & 0 & \mp \sqrt{\mu_{b_0}}^{-1} & 0 & \sqrt{\mu_{b_0}} & 0 \\
    0 & 0 & 0 & 0 & 0 & 1
  \end{pmatrix}.
\end{equation*}

The threshold regularity will depend on the component of $\SubPOp_{b_0}$
transversal to the horizons. We first consider $\SubPOp_{b_0}$ at the
cosmological horizon before considering the event
horizon. In the coordinates of
(\ref{linear:eq:coordinate-transform:SdS-outgoing-EF}), we can calculate that
the non-zero Christoffel symbols are 
\begin{align*}
  \ChristoffelTypeTwo{t_0}{t_0t_0} &= -\frac{1}{2}\p_r\mu_{b_0},\\
  \ChristoffelTypeTwo{t_0}{AB}&=r\UnitSphereMetric_{AB},\\
  \ChristoffelTypeTwo{r}{t_0t_0} &= \frac{\mu_{b_0}\p_r\mu_{b_0}  }{2r},\\
  \ChristoffelTypeTwo{r}{rt_0} &= \frac{1}{2}\p_r\mu_{b_0},\\
  \ChristoffelTypeTwo{r}{AB} &= -r\mu_{b_0} \UnitSphereMetric_{AB},\\
  \ChristoffelTypeTwo{A}{rB} &= \frac{1}{r}\tensor[]{\delta}{^A_B},\\
  \ChristoffelTypeTwo{A}{BC} &= \tensor[^\UnitSphere]{\Gamma}{^A_{BC}}.
\end{align*}
Using the normalized coordinates on the $2$-sphere
\begin{equation*}
  (\thetaUnit, \phiUnit) := \left(\theta, \frac{\varphi}{\sin\theta}\right),
\end{equation*}
the subprincipal operator for the vector-wave operator is given
by: 
\begin{equation*}
  \begin{split}
    \SubPOp_{b_0}^{(1)}{h}_{t_0} &= \frac{\p_r\mu_{b_0}}{2}\p_r {h}_{t_0}
    - \frac{\p_r\mu_{b_0}}{2}\p_{t_0}{h}_{r},\\
    \SubPOp_{b_0}^{(1)}{h}_{r} &= -\frac{r\p_r\mu_{b_0}}{2}\p_r {h}_r +
    r^{-3}\UnitSphereInvMetric^{AB}(\p_B {h}_A + \p_A {h}_B),\\
    \SubPOp_{b_0}^{(1)}{h}_C &= r^{-1}\mu_{b_0}\p_C {h}_r -
    r^{-1}\mu_{b_0}\p_C {h}_r + \frac{1}{r}\p_C {h}_{t_0} -
    \frac{1}{r}\p_{t_0}{h}_C.
  \end{split}
\end{equation*}

We can similarly use the Christoffel symbols to calculate the
subprincipal operator for the linearized Einstein operator, which is
exactly the subprincipal operator for the 2-wave:
\begin{align*}
  \SubPOp_{b_0}{h}_{t_0t_0}
  ={}&
       -2\p_r\mu_{b_0}\p_r {h}_{t_0t_0}
       +2\p_r\mu_{b_0} \p_{t_0}{h}_{t_0r},\\
  \SubPOp_{b_0}{h}_{t_0r}
  ={}&
       r^{-3}\UnitSphereMetric^{AB}\p_Au_{Bt_0}
       +2\p_r\mu \p_{t_0}{h}_{rr},\\
  \SubPOp_{b_0}{h}_{t_0C}
  ={}&
       \left(-2r^{-1}\mu_{b_0}-
       \p_r\mu_{b_0}\right)\p_ru_{Ct_0}
       +2r^{-1}\p_{t_0}{h}_{t_0C}
       - 2r^{-1}\mu_{b_0}\p_C {h}_{t_0r}\\
     &- 2r^{-1}\mu_{b_0}\p_C {h}_{t_0t_0}
       -\p_r\mu_{b_0}\p_{t_0}{h}_{rC},\\
  \SubPOp_{b_0}{h}_{rr}
  ={}&
       2\p_r\mu_{b_0}\p_r{h}_{rr}
       -2r^{-3}\UnitSphereInvMetric^{AB}\p_A{h}_{rB},\\
  \SubPOp_{b_0}{h}_{rC}
  ={}&
       -2r^{-2}\mu_{b_0}\p_C{h}_{rr}
       - 2r^{-3}\UnitSphereInvMetric^{AB}\p_{B}{h}_{CA}
       - 2r^{-1}\p_C {h}_{t_0r}
       +2r^{-2}(r-3M)\p_r {h}_{t_0 r}
       + 2r^{-1}\p_{t_0}{h}_{C r},\\
  \SubPOp_{b_0}{h}_{CD}
  ={}& 2r^{-1}\mu_{b_0}\p_C {h}_{r D}
       + 2r^{-1}\p_C {h}_{t_0 D}
       + 4r^{-1}\p_r {h}_{rr}
       - 4r^{-1}\mu_{b_0}\p_r{h}_{CD}
     + 4r^{-1}\p_{t_0} {h}_{CD}.
\end{align*}
Let us now consider the case where $\Horizon=\CosmologicalHorizonFuture$.
We observe that in order to calculate
$\SHorizonControl{\LinEinstein_{g_{b_0}}}[\CosmologicalHorizonFuture]$, 
we only need to consider the
components of
$\left.\SubPOp_{b_0}\right\vert_{\CosmologicalHorizonFuture}$ in
$\p_r$ since
\begin{equation*}
  \left.g_{b_0}(\HorizonGen_{\CosmologicalHorizonFuture},
    \p_{t_0})\right\vert_{\CosmologicalHorizonFuture}=
  \left.g_{b_0}(\HorizonGen_{\CosmologicalHorizonFuture},
    \slashed{\p})\right\vert_{\CosmologicalHorizonFuture}=0.
\end{equation*}
Then recalling that
\begin{equation*}
  2\SurfaceGravity_{b_0, \CosmologicalHorizonFuture} = - \p_{r}\mu_{b_0} (r_{b_0, \CosmologicalHorizonFuture}),
\end{equation*}
we calculate that at $\CosmologicalHorizonFuture$, in the splitting of
\eqref{linear:eq:subprincipal-calculation:horizon-splitting}, we can write
\begin{equation*}
  \left.\SubPOp_{b_0}\right\vert_{\CosmologicalHorizonFuture} =
  2\SurfaceGravity_{b_0,\CosmologicalHorizonFuture}
  \begin{pmatrix}
    -2&0&0&0&0&0\\
    0&0&0&0&0&0\\
    0&0&-1&0&0&0\\
    0&0&0&2&0&0\\
    0&0&0&0&1&0\\
    0&0&0&0&0&0
  \end{pmatrix}\p_r
  + \widetilde{\SubPOp}_{b_0, \CosmologicalHorizonFuture},
\end{equation*}
where $\widetilde{\SubPOp}_{b_0, \CosmologicalHorizonFuture}$ is a
matrix-valued vectorfield tangent to $\CosmologicalHorizonFuture$.
As a result, we have that 
\begin{equation*}
  \sup_{\CosmologicalHorizonFuture}\left.g \left(\HorizonGen_{b_0, \CosmologicalHorizonFuture},
      \Re\left(\overline{\xi}\cdot \SubPOp_{b_0}\xi\right)\right)\right\vert_{\CosmologicalHorizonFuture}
  =  4\SurfaceGravity_{b_0,\CosmologicalHorizonFuture} ,
  \qquad
  \inf_{\Horizon} \left.g \left(\HorizonGen_{b_0, \CosmologicalHorizonFuture},
      \Re\left(\overline{\xi}\cdot \SubPOp_{b_0}\xi \right)\right)\right\vert_{\CosmologicalHorizonFuture}
  =  -4\SurfaceGravity_{b_0,\CosmologicalHorizonFuture}.
\end{equation*}
Similar calculations at $\EventHorizonFuture$ yield that
\begin{equation*}
  \left.\SubPOp_{b_0}\right\vert_{\EventHorizonFuture}
  =-2\SurfaceGravity_{b_0,\EventHorizonFuture}
  \begin{pmatrix}
    -2&0&0&0&0&0\\
    0&0&0&0&0&0\\
    0&0&-1&0&0&0\\
    0&0&0&2&0&0\\
    0&0&0&0&1&0\\
    0&0&0&0&0&0
  \end{pmatrix}\p_r
  + \widetilde{\SubPOp}_{b_0, \EventHorizonFuture},
\end{equation*}
where $\widetilde{\SubPOp}_{b_0, \EventHorizonFuture}$ is a
matrix-valued vectorfield tangent to $\EventHorizonFuture$,
and thus 
\begin{equation*}
  \sup_{\EventHorizonFuture}\left.g \left(\HorizonGen_{b_0, \EventHorizonFuture},
      \Re\left(\overline{\xi}\cdot \SubPOp_b\xi\right)\right)\right\vert_{\EventHorizonFuture}
  =  4\SurfaceGravity_{b_0,\EventHorizonFuture} ,
  \qquad
  \inf_{\EventHorizonFuture} \left.g \left(\HorizonGen_{b_0, \EventHorizonFuture},
      \Re\left(\overline{\xi}\cdot \SubPOp_b\xi \right)\right)\right\vert_{\EventHorizonFuture}
  =  -4\SurfaceGravity_{b_0,\EventHorizonFuture}.
\end{equation*}

\section{Appendix to Section \ref{linear:sec:QNM}}

\subsection{Proof of Proposition \ref{linear:prop:QNM:C0-semigroup}}
\label{linear:appendix:prop:QNM:C0-semigroup}

From Corollary \ref{linear:corollary:naive-energy-estimate:higher-order},
we already know that $\SolOp(\tStar)$ maps $\LSolHk{k}(\Sigma)$ to
itself. It suffices then to check that:
\begin{enumerate}
\item $\SolOp(0) = \Identity_{\LSolHk{k}(\Sigma)}$: This follows
  from the definition of $\SolOp(0)$ and the well-posedness theorem
  for the relevant Cauchy problem.
\item $\SolOp(t_0+t_1) = \SolOp(t_0)\circ\SolOp(t_1)$: This follows
  from the stationarity of the \KdS{} black hole spacetime. The
  relevant Cauchy problem is invariant under time-translations, and
  thus, solving the Cauchy problem from $[t_1, t_0+t_1]$ is
  equivalent to solving the problem on $[0, t_0]$. Thus, we can
  solve the problem from $[0, t_0+t_1]$ by first solving the problem
  from $[0, t_1]$ and then from $[t_1, t_0+t_1]$.
\item $\SolOp(\tStar)$ is continuous in the strong operator
  topology: Using Corollary
  \ref{linear:corollary:naive-energy-estimate:higher-order}, we see that
  for each $f\in \LSolHk{k}(\Sigma)$, 
  the map $\tStar\mapsto \SolOp(\tStar)f$ is a $C^0$ curve in
  $\LSolHk{k}(\Sigma)$.  
\end{enumerate}

\subsection{Proof of Proposition \ref{linear:prop:QNM:freq-characterization}}
\label{linear:appendix:prop:QNM:freq-characterization}


Let
\begin{equation}
  \label{linear:eq:QNM:freq-characterization:h-def}
  \mathbf{h}(\tStar, x) = \sum_{\ell=1}^n
  e^{-\ImagUnit\sigma\tStar}\tStar^\ell\mathbf{u}_{\ell}(x)  
\end{equation}
be a $\LSolHk{k}$-quasinormal mode solution of $\LinearOp$. Then
define
\begin{equation*}
  F(\tStar, \cdot) := \frac{1}{\ImagUnit}
  \delta_0(\tStar)\mathbf{h}(0, \cdot),
\end{equation*}
so that
\begin{equation*}
  \widehat{F}(\sigma, \cdot) = \frac{1}{\ImagUnit} \mathbf{h}(0, \cdot).
\end{equation*}
Applying Lemma \ref{linear:lemma:semigroup-resolvent-Cauchy-sol-relation},
and taking $\tStar\to+\infty$, we have that for
$\Im\sigma > \GronwallExp$,
\begin{equation*}
  (\InfGen - \sigma)^{-1}
  \widehat{F}(\sigma)
  = \int_{\Real^+}e^{\ImagUnit\sigma\sStar}\SolOp(\sStar)\mathbf{h}(0)\,d\sStar = \widehat{\mathbf{h}}(\sigma).
\end{equation*}
Taking the inverse Laplace transform on both sides and applying the
Cauchy integral formula over an appropriately deformed contour, we have that
\begin{equation*}
  \mathbf{h}(\tStar)
  = \Residue_{\sigma=\sigma_0}(e^{-\ImagUnit\sigma\tStar}(\InfGen-\sigma)^{-1}\widehat{F}(\sigma)).
\end{equation*}
Thus, to prove Proposition \ref{linear:prop:QNM:freq-characterization}, it
suffices to show that
$\widehat{F}(\sigma, \cdot)\in P(\sigma, C^\infty(\sigma))$. This is
trivially true if $n=0$. We then prove the general case by
induction. Observe that for $h$ as defined in
\eqref{linear:eq:QNM:freq-characterization:h-def}, we have that
$\mathbf{h}' = (D_{\tStar}-\InfGen)\mathbf{h}$ is a
$\LSolHk{k}$-quasinormal mode of order $n-1$. Then inductively, $\mathbf{h}'(0,
\cdot)\in P(\sigma, C^\infty(\Sigma))$. However, observe that
\begin{equation*}
  D_{\tStar}\mathbf{h} = \sum_{\ell=1}^n \sigma u_0(x) + \ImagUnit u_1(x), 
\end{equation*}
also belongs to $P(\sigma, C^\infty(\Sigma))$. Then, it is immediate
that $\mathbf{h}(0, \cdot)\in P(\sigma, C^\infty(\Sigma))$ as
desired. 

The frequency characterization of the $\LSolHk{k}$-quasinormal modes
in \eqref{linear:eq:QNM:freq-characterization}
then comes from the realization the $\LSolHk{k}$-quasinormal
spectrum is discrete and thus all poles of $(\InfGen-\sigma)^{-1}$
are of finite order.

\subsection{Proof of Proposition \ref{linear:prop:lin-theory:lambda-map-def}}
\label{linear:appendix:prop:lin-theory:lambda-map-def}

Let us first consider the case where $\Xi$ consists of only one
$\LSolHk{k}$-quasinormal frequency, $\sigma_0$, which is an
$\LSolHk{k}$-quasinormal mode of order $\ell$. Now for some fixed
$F\in H^{-1,\beta}(\Real^+; \LSolHk{k-1}(\Sigma))$, we have that
$(\InfGen - \sigma)^{-1}\widehat{F}(\sigma)$ is
holomorphic near  $\sigma_0$ if and only if
$\bangle*{\widehat{F}(\sigma),
  (\InfGen^*-\overline{\sigma})^{-1}G}_{L^2(\Sigma)}$ is holomorphic
near every $\sigma_0$ for every $G\in \LSolHk{-k}(\Sigma)$.
Using the definition of the Laplace
transform, this is equivalent to the condition that
\begin{equation*}
  \bangle*{F(\tStar,x), e^{-\ImagUnit\overline{\sigma}\tStar}(\InfGen^*-\overline{\sigma})^{-1}G(x)}_{L^2(\StaticRegionWithExtension)}
\end{equation*}
is holomorphic for every $G\in\LSolHk{-k}(\Sigma)$, near each
$\sigma_j\in \Xi$.

Since $(\InfGen^*-\overline{\sigma})^{-1}$ is meromorphic, 
for a fixed $G\in\LSolHk{-k}(\Sigma)$, we have that for
$\sigma$ in a small neighborhood of $\sigma_j$,
\begin{equation*}
  e^{-\ImagUnit\overline{\sigma}\tStar}(\InfGen^*-\overline{\sigma})^{-1}G
  = \sum_{j=1}^{\ell}(\overline{\sigma} - \overline{\sigma}_0)^{-j}{G}_{j} + \widetilde{G}(\overline{\sigma})
\end{equation*}
with ${G}_{j}\in e^{-\Im\sigma_0\tStar}\LSolHk{1-k}(\Sigma)$,
and $\widetilde{G}$ holomorphic near $\overline{\sigma}_0$ with values
in $\LSolHk{1-k}(\Sigma)$.

Thus, 
$\bangle*{\widehat{F}(\sigma),
  e^{-\ImagUnit\overline{\sigma}\tStar}(\InfGen^*-\overline{\sigma})G}_{L^2(\StaticRegionWithExtension)}$
is holomorphic at $\sigma_0$ if and only if
\begin{equation}
  \label{linear:eq:QNM:orthogonality-condition:aux1}
  \bangle*{\widehat{F}(\sigma),\residue_{\sigma=\sigma_0}\left((\overline{\sigma}-\overline{\sigma}_0)^{j-1}e^{-\ImagUnit \overline{\sigma}\tStar}(\InfGen^* - \overline{\sigma})^{-1}G\right)}_{L^2(\StaticRegionWithExtension)} = 0
\end{equation}
for each $1\le j\le \ell$. Observe that
\eqref{linear:eq:QNM:orthogonality-condition:aux1} holds for all $j>\ell$
since the poles of
$e^{-\ImagUnit\overline{\sigma}\tStar}(\InfGen^*-\overline{\sigma})^{-1}G$
are all of finite order at most $\ell$. Recalling that $G$ was
arbitrary, we see that the condition in 
\eqref{linear:eq:QNM:orthogonality-condition:aux1} is equivalent to the
condition that
\begin{equation*}
  \bangle*{\widehat{F}(\sigma), \residue_{\sigma=\overline{\sigma}_0}e^{-\ImagUnit \sigma\tStar}(\InfGen^*- \sigma)^{-1}p(\sigma)  }_{L^2(\StaticRegionWithExtension)}=0
\end{equation*}
for all polynomials $p(\sigma)$ in $\sigma$ with coefficients in
$\LSolHk{-k}(\Sigma)$. The case where $\Xi$ consists of more than one
$\LSolHk{k}$-quasinormal frequency then follows directly. 


\subsection{Proof of Proposition \ref{linear:prop:QNM-perturb:gen}}
\label{linear:appendix:prop:QNM-perturb:gen}

\begin{remark}
  Recall that the definition of the space $\LSolHk{k}(\Sigma)$, we
  relied on the construction of a family of vectorfields
  $\RedShiftK_j$ which depend on the background metric
  $g_{b(w)}$. Thus, the family of operators
  $\InfGen_w - \sigma: D^k(\InfGen_w-\sigma)\to\LSolHk{k}(\Sigma)$
  has domain and range varying in $w$. To avoid any difficulties
  that could rise in the ensuing perturbation theory, we observe
  that by construction the $\LSolHk{k}_b(\Sigma)$ norms for $b\in
  \BHParamNbhd$ are all equivalent norms. Thus, the family
  $\InfGen_w - \sigma$ is a family with differing domains, but
  identical range. Likewise, the same is true for the family
  $\widehat{\LinearOp}_{b}(\sigma):D^k(\widehat{\LinearOp}_b(\sigma))\to
  \InducedHk{k-1}(\Sigma)$. 
\end{remark}
The crucial ingredient for the first two items is that
the assumptions imply that 
\begin{equation}
  \label{linear:eq:QNM-perturb:commuted-redshift}
  \norm{\mathbf{u}}_{\InducedHk{k}(\Sigma)} \lesssim \norm*{(\InfGen_w-\sigma)\mathbf{u}}_{\InducedHk{k-1}(\Sigma)} + \norm{\mathbf{u}}_{\InducedHk{k_0}(\Sigma)}
\end{equation}
where $k>k_0$, and $k_0$ is the threshold regularity level defined in
\eqref{linear:eq:QNM-perturb:gen:threshold-reg-def}. This is shown explicitly
for the case where the family $\LinearOp_{b(w)}$ is the family of
linearized gauged Einstein operators $\LinEinstein_{g_{b(w)}}$ in
Section\footnote{We emphasize that the estimate in only assumption
  needed in Section \ref{linear:sec:energy-estimates} to prove estimates of
  the form in \eqref{linear:eq:QNM-perturb:commuted-redshift} were that the
  linearized gauged Einstein operators are strongly hyperbolic.} \ref{linear:sec:energy-estimates}.


\paragraph{Proof of part \ref{linear:prop:QNM-perturb:gen:item1}}
To prove the first statement we will show that if
$\InfGen_{w_0}-\sigma_0$ 
is invertible, then there exists a sufficiently small neighborhood
$\widetilde{W}\times \widetilde{\Omega}\ni (w_0, \sigma_0)$ such that for all
$(w, \sigma)\in \widetilde{W}\times \widetilde{\Omega}$,
$\InfGen_{w}-\sigma:
D^{k}(\InfGen)\to\LSolHk{k}(\Sigma)$
 is invertible. An
application of the analytic Fredholm theorem as in Theorem
\ref{linear:thm:meromorphic:fredholm-alt:Laplace-transformed-op} then shows
the meromorphy of $(\InfGen_w-\sigma)^{-1}$ in $\sigma$,
and we can conclude using the relation between the Laplace-transformed
operator and the infinitesimal generator in Lemma
\ref{linear:lemma:laplace:inverse-A}.

Assume for the sake of contradiction that there exists a sequence of
$w_j\to w_0$, $\sigma_j\to \sigma$ such that
$\InfGen_{w_j}-\sigma_j$ is not invertible for all
$j$. Then for each $j$, either the kernel or the cokernel of
$\InfGen_{w_j}-\sigma_j$ is nontrivial. Passing to a
subsequence, we can assume without loss of generality that
$\Ker(\InfGen_{w_j}-\sigma_j)$ is nontrivial for all $j$.
Now consider the sequence
${h}_j \in \LSolHk{k}(\Sigma)$ such that ${h}_j$
are normalized so that $\norm{{h}_j}_{\LSolHk{k}(\Sigma)}=1$, and
moreover,
\begin{equation*}
  (\InfGen_{w_j}-\sigma_j){h}_j = 0
\end{equation*}
for all $j$. Using \eqref{linear:eq:QNM-perturb:commuted-redshift}, we see
that then for all $j$,
\begin{equation}
  \label{linear:eq:QNM-perturb:gen:item1:aux1}
  1\lesssim \norm{{h}_j}_{\LSolHk{k_0}(\Sigma)}.
\end{equation}
Since we constructed ${h}_j$ as a bounded sequence in
$\LSolHk{k}(\Sigma)$, there exists a subsequence such that
${h}_j\rightharpoonup{h}_0$ weakly for some
${h}_0\in \LSolHk{k}(\Sigma)$. Now, using Rellich-Kondrachov, we
have in fact that ${h}_j\to {h}_0$ in $\LSolHk{k_0}(\Sigma)$. But
then \eqref{linear:eq:QNM-perturb:gen:item1:aux1} implies that there exists
some non-zero ${h}_0$ such that
\begin{equation*}
  (\InfGen_{w_0}-\sigma_0){h}_0 = 0,
\end{equation*}
and thus that $\Ker(\InfGen_{w_0}-\sigma_0) \neq \emptyset$, which
is a contradiction.

\paragraph{Proof of part \ref{linear:prop:QNM-perturb:gen:item2}}
We now move onto proving the second item. We first show that the map
\begin{equation*}
  I\ni(w,\sigma)\mapsto (\InfGen_w-\sigma)^{-1}\in \mathcal{L}_{weak}(\LSolHk{k}(\Sigma), \LSolHk{k}(\Sigma))
\end{equation*}
is continuous for all $k>k_0$. To do so, we see from Lemma
\ref{linear:lemma:laplace:inverse-A} that it is sufficient to show that the
map
\begin{equation*}
  I\ni(w,\sigma)\mapsto \widehat{\LinearOp}_w(\sigma)\in \mathcal{L}_{weak}(\InducedHk{k-1}(\Sigma), \InducedHk{k}(\Sigma))
\end{equation*}
is continuous for all $k>k_0$. 

Consider a sequence $f_j\in\LSolHk{k}(\Sigma)$, such that
\begin{equation*}
  \norm{f_j}_{\LSolHk{k}(\Sigma)} \le 1, 
\end{equation*}
and let $f_j\to f$ in $\LSolHk{k}$, $w_j\to w_0$,
$\sigma_j\to\sigma_0$. Assume for the sake of contradiction that
$\mathbf{h}_j$ defined by
\begin{equation*}
  \mathbf{h}_j = (\InfGen_{w_j} - \sigma_j)^{-1}f_j
\end{equation*}
is not a bounded sequence in $\InducedHk{k}(\Sigma)$. Defining
\begin{equation*}
  \psi_j = \frac{\mathbf{h}_j}{\norm{\mathbf{h}_j}_{\InducedHk{k}(\Sigma)}},
\end{equation*}
we see from equation (\ref{linear:eq:QNM-perturb:commuted-redshift}), that
\begin{equation*}
  1\lesssim \norm{\mathbf{h}_j}_{\LSolHk{k}(\Sigma)}^{-1} + \norm{\psi_j}_{\LSolHk{k_0}(\Sigma)}.
\end{equation*}
Using that $\mathbf{h}_j$ is not bounded, we find that for $j$
sufficiently large,
\begin{equation*}
  1\lesssim \norm{\psi_j}_{\LSolHk{k-1}(\Sigma)}. 
\end{equation*}
Since by construction, $\psi_j$ is a bounded sequence in
$\LSolHk{k}(\Sigma)$, there exists some
$\psi_0\neq 0\in \LSolHk{k}(\Sigma)$ and some subsequence
$\psi_{j\ell}$ of $\psi_j$ that converges weakly to $\psi_0$. We then
have that
\begin{equation*}
  \frac{f_{j\ell}}{\norm{\mathbf{h}_{j\ell}}_{\LSolHk{k}(\Sigma)}}
  = (\InfGen_{w_{j\ell}}-\sigma_{j\ell})\psi_{j\ell} \rightharpoonup (\InfGen_{w_0}-\sigma_0)\psi_0
\end{equation*}
weakly in $\LSolHk{k-1}(\Sigma)$. But then since $\mathbf{h}_{j\ell}$ is
not bounded, we have that
$\frac{f_{j\ell}}{\norm{\mathbf{h}_{j\ell}}_{\LSolHk{k}(\Sigma)}}$
converges to $0$ in $\LSolHk{k_0}(\Sigma)$, and thus
$(\InfGen_{w_0}-\sigma_0)\psi_0=0$. This is a
contradiction, and thus, $\curlyBrace{\mathbf{h}_j}$ must be bounded in
$\LSolHk{k}(\Sigma)$.
Thus, any subsequence $\mathbf{h}_{j\ell}$ converges weakly to some
$\mathbf{h}_0\in \LSolHk{k}(\Sigma)$. As a result,
\begin{equation*}
  f_{j\ell} =(\InfGen_{w_{j\ell}}-\sigma_{j\ell}) \mathbf{h}_{j\ell}
  \rightharpoonup (\InfGen_{w_0}-\sigma_{0})\mathbf{h}_0
\end{equation*}
weakly in $\LSolHk{k-1}(\Sigma)$ and
$(\InfGen_{w_{0}}-\sigma_0)\mathbf{h}_0 = f$. 
Since $\InfGen_{w_0}-\sigma_0$ is invertible, it
is in particular injective. Thus every subsequence $\mathbf{h}_{j\ell}$
must converge weakly to $\mathbf{h}_0\in \LSolHk{k}(\Sigma)$, where
$\mathbf{h}_0$ is independent of the chosen subsequence. As a result,
$\mathbf{h}_j$ itself converges weakly to $\mathbf{h}_0$ in
$\LSolHk{k}(\Sigma)$.

\paragraph{Proof of part \ref{linear:prop:QNM-perturb:gen:item3}}

For proving the remaining items, it is clear that it suffices to consider the case
where
$\Omega = \{\sigma\in\Omega: \abs*{\sigma - \sigma_0}<\delta \}$, and
the only $\LSolHk{k}$-quasinormal frequency of
$\LinearOp_{w_0}(\sigma_0)$ in $\overline{\Omega}$ is exactly
$\sigma_0$.

To prove the third item, it is sufficient to
consider the poles of $(\InfGen_w-\sigma)^{-1}$, and to
prove that the total rank, which can also be expressed as
\begin{equation*}
  d:=\sum_{\sigma\in \QNMk{k}(\LinearOp_w)\bigcap \Omega}\rank_{\zeta=\sigma}(\InfGen_w-\sigma)^{-1},
\end{equation*}
is constant for $w$ near $w_0$.

Recall from Theorem
\ref{linear:thm:meromorphic:fredholm-alt:Laplace-transformed-op} that
$(\InfGen_{w_0}-\sigma): D^k(\InfGen_{w_0})\to
\LSolHk{k}(\Sigma)$ is a zero-index operator. We can thus
construct the bases $\{u_j\}_{j=1}^{n}$, $\{f_j\}_{j=1}^{n}$ of the kernel
and the cokernel respectively, and define the operator
$R: D^k(\InfGen_w)\to \LSolHk{k-1}(\Sigma)$
\begin{equation*}
  Ru:=\sum_{j=1}^n\bangle*{u, u_j}_{\LTwo(\Sigma)}f_j.
\end{equation*}
Then, let us denote by $\mathcal{Y}_2 = \coKer
(\InfGen_{w_0}-\sigma_0)$, and let $\mathcal{Y}_1$ be
defined such that
\begin{equation*}
  \LSolHk{k}(\Sigma)= \mathcal{Y}_1\oplus \mathcal{Y}_2. 
\end{equation*}
Then,
\begin{equation*}
  P_w(\sigma) := \InfGen_w-\sigma + R: D^k(\InfGen_w)\to \LSolHk{k}(\Sigma)
\end{equation*}
is invertible for $w=w_0$, $\sigma=\sigma_0$.  Moreover,
$P_w(\sigma)$ is a zero-order perturbation of
$\InfGen_w-\sigma$, and thus equation
(\ref{linear:eq:QNM-perturb:commuted-redshift}) continues to hold with $P_w(\sigma)$ in place of
$\InfGen_w-\sigma$.  As a result, parts
\ref{linear:prop:QNM-perturb:gen:item1} and
\ref{linear:prop:QNM-perturb:gen:item2} of the current theorem hold for
$P_w(\sigma)^{-1}$ in place of $(\InfGen_w-\sigma)^{-1}$ as well.
As a result, by considering a sufficiently small parameter
space $W$, and a sufficiently small neighborhood $\Omega$ of
$\sigma_0$, we can assume that
$P_w(\sigma):D^k(\InfGen_w)\to
\LSolHk{k}(\Sigma)$ is invertible on $W\times \Omega$, with a
continuous inverse. 
Define
\begin{equation*}
  Q_w(\sigma) = 1 - RP_w(\sigma)^{-1}: \LSolHk{k}(\Sigma)\to\LSolHk{k}(\Sigma),
\end{equation*}
so that
\begin{equation*}
  \InfGen_w-\sigma = Q_w(\sigma)P_w(\sigma).
\end{equation*}
Then $\InfGen_w-\sigma$ is invertible if and only if
$Q_w(\sigma)$ is invertible. 
Using the decomposition $\LSolHk{k}(\Sigma) =
\mathcal{Y}_1 \oplus \mathcal{Y}_2$, we see that 
\begin{equation*}
  Q_w(\sigma) = 
  \begin{pmatrix}
    1 & 0\\
    Q_{w,1}(\sigma) & Q_{w,2}(\sigma)
  \end{pmatrix}
\end{equation*}
where
\begin{equation*}
  Q_{w,1}(\sigma ) = -RP_w(\sigma)^{-1}\vert_{\mathcal{Y}_1}\in \mathcal{L}(\mathcal{Y}_1, \mathcal{Y}_2)\qquad
  Q_{w,2}(\sigma) = 1 - RP_w(\sigma)^{-1}\vert_{\mathcal{Y}_2}\in \mathcal{L}(\mathcal{Y}_2, \mathcal{Y}_2).
\end{equation*}
Thus we see that the invertibility of $Q_w(\sigma)$ is equivalent
to the invertibility of $Q_{w,2}(\sigma)$, which is a family of
linear operators depending continuously on $w$ and holomorphically
on $\sigma$ acting on a fixed finite-dimensional subspace
$\mathcal{Y}_2$. 

Recall that by assumption, we have that
for any $w\in W$,
$\InfGen_w -\sigma$ is
an analytic meromorphic family, we have that
\begin{equation*}
  \text{rank}_{\zeta = \sigma}\,(\InfGen_w-\zeta)^{-1} 
  = \frac{1}{2\pi \ImagUnit} \Trace \oint_{\p\Omega} Q_{w,2}(\zeta)^{-1}\p_\zeta Q_{w,2}(\zeta)\,d\zeta,
\end{equation*}
which is integer-valued and continuous in $w$.

\paragraph{Proof of parts \ref{linear:prop:QNM-perturb:gen:item4} and \ref{linear:prop:QNM-perturb:gen:item5} }
The fourth and fifth items are proved by similar arguments, so we
only provide a detailed proof for the former.

Consider polynomials $p_1(\zeta), \cdots, p_{d}(\zeta)$ with
values in $C^\infty(\Sigma)$ such that
\begin{equation*}
  {h}_j(w_0):= \oint_{\p\Omega}e^{-\ImagUnit\zeta\tStar}(\InfGen_{w_0}- \sigma_0)^{-1}p_j(\zeta)\,d\zeta \in C^{\infty}(\Sigma)
\end{equation*}
span $\QNMk{k}(\LinearOp_{w_0}, \sigma_0)$. Using
\ref{linear:prop:QNM-perturb:gen:item3} of the proposition, we see that for
sufficiently small $w\in W$, $(\InfGen_w-\zeta)^{-1}$
exists for $\zeta\in \p\Omega$. Thus, the contour integral
\begin{equation*}
  {h}_j(w):= \oint_{\p\Omega}e^{-\ImagUnit\zeta\tStar}(\InfGen_{w}-\zeta)^{-1}p_j(\zeta)\,d\zeta \in C^{\infty}(\Sigma)
\end{equation*}
is well defined. 
Using part \ref{linear:prop:QNM-perturb:gen:item2} of the proposition,
${h}_j(w)$ depends continuously on $w$ in the topology of
$C^{\infty}(\Sigma)$. Thus, $\{{h}_j(w)\}_{j=1}^d$ is a
$d$-dimensional set of $C^\infty(\Sigma)$ functions for $w\in W$
sufficiently small.
Since $\QNMk{k}(\LinearOp_w, \Omega)$ is also $d$-dimensional
and ${h}_j(w)\in \QNMk{k}(\LinearOp_w, \Omega)$, we have that
actually
\begin{equation*}
  \QNMk{k}(\LinearOp_w, \Omega) = \Span\curlyBrace*{{h}_j(w): 1\le j\le d}.
\end{equation*}
The statement then follows from the map $W\times \Complex^d\ni (w,
(c_1,\cdots, c_d)) \to \sum c_j{h}_j(w)$. 

\section{Appendix to Section \ref{linear:sec:freq-analysis}}
\label{linear:appendix:freq-analysis}

\subsection{Proof of Lemma \ref{linear:lemma:trapping:SdS}}
\label{linear:appendix:lemma:trapping:SdS}

In what ensues, we will use regular coordinates on \SdS, where we
set $c_{b_0}=0$. This is equivalent to a coordinate transformation,
and because we work with either coordinate independent objects, like
principal symbols, or $c$-independent objects, this does not affect
our symbolic calculations.

Let $g_{b_0}$ be a \SdS{} metric. Then in the
Eddington-Finkelstein $(t_0, r, \theta, \varphi)$
coordinates,
\begin{equation*}
  \begin{split}
    g_{b_0} &= -\mu_{b_0}d t_0^2 \pm 2dt_0dr +r^2\UnitSphereMetric,\\
    G_{b_0} &= \mp 2\p_{t_0}\p_r + \mu_{b_0}\p_r^2 + r^{-2}\UnitSphereInvMetric
  \end{split}
\end{equation*}
where the $\pm$ corresponds to incoming and outgoing Eddington
Finkelstein respectively. Observe that in the incoming Eddington
Finkelstein coordinates, $g(\p_{t_0}, \p_r) = -2$ at the event
horizon, while in the outgoing Eddington-Finkelstein coordinates
$g(\p_{t_0}, \p_r)=2$ at the cosmological horizon.   

We now locate the trapped null geodesics in this coordinate system. In
anticipation of the more complicated nature of trapping in \KdS, we
analyze the trapped set of $g_{b_0}$ by considering the Hamiltonian
vectorfield $H_{\PrinSymb_{b_0}}$, where $\PrinSymb_{b} = G_{b}$ is
the principal symbol of $\ScalarWaveOp[g_{b}]$.

To begin, observe that in the
$(t_0, r, \theta, \varphi; \sigma, \xi, \FreqTheta, \FreqPhi)$
coordinates,
\begin{equation}
  \label{linear:eq:KdS-Trapping:SdS-Hp-def}
  H_{r^2\PrinSymb_{b_0}} = 2(\Delta_{b_0}\xi \pm r^2\sigma)\p_r
  - \left(\p_r\Delta_{b_0}\xi^2 \pm 4r\sigma\xi\right)\p_\xi - H_{|\FreqAngular|^2},
\end{equation}
where the sign notation is that $\pm (r-\rCrit) > 0 $. 


We then observe that $H_{r^2\PrinSymb_{b_0}}r = 0$ exactly when $\pm
\sigma = r^{-2}\Delta_{b_0}\xi$. Let us first consider the region
where $\Delta_{b_0} \le 0$. Then, for $\sigma \neq 0$, we have that if
$\Delta_{b_0}\le 0$, and in addition $H_{r^2\PrinSymb_{b_0}}r=0$, then
we necessarily have that
\begin{equation*}
  \Delta_{b_0}\xi\neq 0,
\end{equation*}
and furthermore, that
\begin{equation*}
  \PrinSymb_{b_0} < 0.
\end{equation*}
This rules out trapped null-bicharacteristics on the region where
$\Delta_{b_0}\le 0$, except exactly those null-bicharacteristics spanning
the horizons. Thus, we are left with considering only the interior of
the static region, where $\Delta_{b_0}>0$. In this region, if
$H_{r^2\PrinSymb_{b_0}}r=0$, we can calculate that
\begin{equation*}
  H^2_{r^2\PrinSymb_{b_0}}r
  = -2\Delta_{b_0}H_{r^2\PrinSymb_{b_0}}\xi
  = -2\Delta_{b_0}\xi^2\left(\p_r\Delta_{b_0} -4r^{-1}\Delta_{b_0}\right)
  = -2\Delta_{b_0}\xi^2r^4\p_r(r^{-4}\Delta_{b_0}).
\end{equation*}
Since $\sigma=0, H_{r^2\PrinSymb_{b_0}}r = 0$, and $\Delta_{b_0}=0$ would
imply $\PrinSymb_{b_0}> 0$, it suffices to consider only the case
where $\sigma\neq0$.

If $\sigma\neq 0$, and $H_{r^2\PrinSymb_{b_0}}r=0$, then $\xi\neq 0$.
We see then that if
$\sigma\neq0, \Delta_{b_0}>0$, and  $H_{r^2\PrinSymb_{b_0}}r=0$, then
$H^2_{r^2\PrinSymb_{b_0}}r = 0$ only if
\begin{equation*}
  \p_r(r^{-4}\Delta_{b_0}) = -2r^{-3}\left(1-\frac{3M}{r}\right) =0.
\end{equation*}
Since we are working under the assumption that $1-9\Lambda M^2 > 0$,
if $H_{r^2\PrinSymb_{b_0}}r = 0$, $\Delta_{b_0}>0$, and
$\PrinSymb_{b_0}=0$, then $H_{r^2\PrinSymb_{b_0}}^2r = 0$ only if
$r=3M$. Furthermore, using the explicit calculation of
$H_{r^2\PrinSymb_{b_0}}^2r$, we see that 
\begin{equation*}
  \Delta_{b_0}>0,\, \PrinSymb_{b_0}=0,\, \pm(r-3M) >0,\, H_{\PrinSymb_{b_0}}r =0\implies \pm H_{r^2\PrinSymb_{b_0}}^2r > 0.
\end{equation*}
\subsection{Proof of Lemma \ref{linear:lemma:trapping:KdS}}
\label{linear:appendix:lemma:trapping:KdS}

The approach will follow the same general outline the approach in
appendix \ref{linear:appendix:lemma:trapping:SdS} with \SdS. We again use
Eddington-Finkelstein coordinates since all our symbolic calculations
are at the principal level, and thus coordinate independent, and using
Eddington-Finkelstein significantly simplifies the calculations.  We
begin with the rescaled principal symbol:
\begin{equation*}
  \rho_b^2p_b =
  \Delta_b \xi^2
  \mp 2a(1+\lambda_b)\xi\FreqPhi
  \mp 2(1+\lambda_b)(r^2+a^2)\xi\sigma
  + \frac{(1+\lambda_b)^2}{\varkappa_b\sin^2\theta}(a(\sin^2\theta)\sigma+\FreqPhi)^2
  +\varkappa_b\FreqTheta^2. 
\end{equation*}
It is convenient to rewrite this as:
\begin{equation*}
  \rho_b^2p_b =
  \Delta_b\left( \xi \mp \frac{1+\lambda_b}{\Delta_b}\left(\left(r^2+a^2\right)\sigma
      + a\FreqPhi\right)\right)^2
  - \frac{(1+\lambda_b)^2}{\Delta_b}\left(\left(r^2+a^2\right)\sigma + a\FreqPhi\right)^2
  + \rho_b^2 \tilde{\PrinSymb}_b,
\end{equation*}
where
\begin{equation*}
  \rho_{b}^2\tilde{\PrinSymb}_b = \varkappa_b\FreqTheta^2 + \frac{(1+\lambda_b)^2}{\varkappa\sin^2\theta}(a(\sin^2\theta)\sigma + \FreqPhi^2). 
\end{equation*}
In this way, we see that
\begin{equation*}
  \Delta_b\left( \xi \mp \frac{1+\lambda_b}{\Delta_b}\left(\left(r^2+a^2\right)\sigma
      + a\FreqPhi\right)\right)^2
  - \frac{(1+\lambda_b)^2}{\Delta_b}\left(\left(r^2+a^2\right)\sigma + a\FreqPhi\right)^2
\end{equation*}
has coefficients dependent only on $r$, while the coefficients of
$\rho_b^2\tilde{\PrinSymb}_b$ depend only on $\theta$. 

An immediate observation is that $(r^2+a^2)\sigma+a\FreqPhi$ cannot
vanish along the characteristic set for $\Delta_b>0$, as if it did, then at the same
point, we must also necessarily have that
$\sigma=\xi=\FreqTheta=\FreqPhi=0$, which is not possible.

We now calculate the Hamiltonian vectorfield:
\begin{equation*}
  \begin{split}
    H_{\rho_b^2p_b}
    ={}&
    \left(\mp 2(1+\lambda_b)(r^2+a^2)\xi +
      \frac{2a^2(1+\lambda_b)^2\sin^2\theta}{\varkappa_b}\sigma -
      \frac{2a(1+\lambda_b)^2}{\varkappa_b}\FreqPhi\right)\p_{\tStar}\\
    & - (1+\lambda_b)^2\left(
      \p_\theta\varkappa_b\FreqTheta^2 +
      \p_\theta\frac{a\sin^2\theta}{\varkappa_b} \sigma^2
      +2a\p_\theta\varkappa^{-1}_b\sigma\FreqPhi +
      (1+\lambda_b)^2\p_\theta(\varkappa_b\sin^2\theta)^{-1}\FreqPhi^2 
    \right)\p_{\FreqTheta}\\
    & - 2\varkappa_b\FreqTheta\p_\theta
    -\left(\p_r\Delta_b\xi^2 \mp 4(1+\lambda_b)r\xi\sigma\right)\p_\xi
    + 2\left(\Delta_b\xi\mp a(1+\lambda_b)\FreqPhi \mp (1+\lambda_b)(r^2+a^2)\sigma\right)\p_r\\
    & + 2\left(\mp a(1+\lambda_b)\xi +
      \frac{(1+\lambda_b)^2}{\varkappa_b\sin^2\theta}\FreqPhi +
      \frac{a(1+\lambda_b)^2}{\varkappa_b}\sigma\right)\p_\phiStar .
  \end{split}
\end{equation*}
Equivalently, 
\begin{equation*}
  \begin{split}
    H_{\RescaledPrinSymb_b} ={}&
    -\left(\p_r\Delta_b\xi^2 \mp 4(1+\lambda_b)r\xi\sigma\right)\p_\xi
    \mp 2(1+\lambda_b)(r^2+a^2)\xi\p_{\tStar}\\
    & + 2 (\Delta_b\xi \mp a(1+\lambda_b)\FreqPhi \mp (1+\lambda_b)(r^2+a^2)\sigma)\p_r
    + H_{\rho_b^2\tilde{\PrinSymb}_b}.  
  \end{split}
\end{equation*}
Unlike in the \SdS{} case, we no longer have integrability of the
Hamiltonian flow directly from the conservation of $\PrinSymb$,
$\sigma$, and $\FreqAngular$. However,  we have that
\begin{equation*}
  H_{\PrinSymb_b}\PrinSymb_{b} = 0,\quad
  H_{\PrinSymb_{b}}\sigma = 0,\quad
  H_{\PrinSymb_{b}}\FreqPhi =0,\quad
  H_{\PrinSymb_b}\tilde{\PrinSymb}_b =0. 
\end{equation*}
These are now the conserved quantities of motion along the Hamiltonians
and show the integrability of the Hamiltonian flow. The fact that
$\tilde{\PrinSymb}_b=0$ comes from the so-called hidden symmetries of
the \KdS{} family, which were first observed in the Kerr family by
Carter \cite{carter_global_1968}. 

We can thus calculate
\begin{equation*}
  H_{\rho_b^2p_b}r = 2(\Delta_b\xi \mp (1+\lambda_b)((r^2+a^2)\sigma + a\FreqPhi ).
\end{equation*}
\begin{remark}
  We can show that on $\Delta_b< 0$, $H_{\rho_b^2p_b}r$ cannot
  vanish on the characteristic set. This follows from the observation
  that if both $\Delta_b< 0$, and $H_{\rho_b^2p_b}r=0$, then, 
  \begin{equation*}
    \rho_b^2p_b = -\Delta_b\xi^2 + \varkappa_b\FreqTheta^2 +
    \frac{(1+\lambda_b)^2}{\varkappa_b\sin^2\theta}(a(\sin^2\theta)\sigma
    + \FreqPhi)^2 \ge 0, 
  \end{equation*}
  with equality only if  $\sigma=\xi=\FreqTheta=\FreqPhi=0$.
  If instead, $\Delta_b=0$, then the only trapped
  null-bicharacteristic is
  \begin{equation*}
    \curlyBrace*{\Delta_b=0, \sigma=\FreqTheta= \FreqPhi=0},
  \end{equation*}
  which are exactly the null geodesics spanning the event horizon and
  the cosmological horizon. 
\end{remark}
Applying $H_{\RescaledPrinSymb_b}r$ again, we find that
\begin{equation*}
  H^2_{\RescaledPrinSymb_b}r
  = 2\Delta_bH_{\RescaledPrinSymb_b}\xi + 2 \xi \p_r\Delta_b H_{\RescaledPrinSymb_b} r
  \mp 4(1+\lambda_b)\sigma r H_{\RescaledPrinSymb_b}r. 
\end{equation*}
We can also calculate that
\begin{equation*}
  \begin{split}
    - H_{\RescaledPrinSymb_b}\xi
    ={}& \p_r\Delta_b\left( \xi \mp \frac{1+\lambda_b}{\Delta_b}\left(\left(r^2+a^2\right)\sigma
        + a\FreqPhi\right)\right)^2
    -\p_r\left(
      \frac{(1+\lambda_b)^2}{\Delta_b}\left(\left(r^2+a^2\right)\sigma
        + a\FreqPhi\right)^2
    \right)
    \\
    &\mp 2\Delta_b\left( \xi \mp \frac{1+\lambda_b}{\Delta_b}\left(\left(r^2+a^2\right)\sigma
        + a\FreqPhi\right)\right)
    \p_r\left( \frac{1+\lambda_b}{\Delta_b}\left(\left(r^2+a^2\right)\sigma
        + a\FreqPhi\right)\right)\\
    ={}& \p_r\Delta_b\left(\frac{H_{\RescaledPrinSymb_b}r}{2\Delta_b}\right)^2
    \mp \p_r\left( \frac{1+\lambda_b}{\Delta_b}\left(\left(r^2+a^2\right)\sigma
        + a\FreqPhi\right)\right) H_{\RescaledPrinSymb_b}r\\
    &-\p_r\left( \frac{(1+\lambda_b)^2}{\Delta_b}\left(\left(r^2+a^2\right)\sigma
        + a\FreqPhi\right)^2\right).
  \end{split}
\end{equation*}
We thus have that 
\begin{equation}
  \label{linear:eq:Trapping-KdS:H2pr-at-Trapping}
  H^2_{\RescaledPrinSymb_b}r
  =2\Delta_b(1+\lambda_b)^2\p_r\left(
    \Delta_b^{-1}\left(\left(r^2+a^2\right)\sigma +a\FreqPhi\right)^2
  \right).
\end{equation}
In order to locate trapped null bicharacteristics, we will need to
analyze the sign of $H_{\RescaledPrinSymb_b}^2r$. To this end, we
consider $\mathcal{F}= \Delta_b^{-1}(a\FreqPhi +
(r^2+a^2)\sigma)^2$. We will show that if $\Delta_b>0$, the critical
point $\rTrapping_b$ of $\mathcal{F}$ is unique in
$(r_{\EventHorizonFuture}, r_{\CosmologicalHorizonFuture})$ for fixed
$\FreqPhi$; and depends smoothly on $\FreqPhi$. Moreover,
$\frac{\p\mathcal{F}}{\p r}>0$ for $r>\rTrapping_b$, and
$\frac{\p\mathcal{F}}{\p r}<0$ for $r<\rTrapping_b$. This will show
then that the only trapped null bicharacteristics are
\begin{equation*}
  \TrappedSet_b = \curlyBrace*{\PrinSymb_b=0, H_{\PrinSymb_b}r = 0, r = \rTrapping_b}.
\end{equation*}
The first step will be to calculate
\begin{equation}
  \label{linear:eq:KdS-Trapping:aux-functions-def}
  \p_r\mathcal{F} = -\left( a\FreqPhi + \left(r^2+a^2\right)\sigma \right)\Delta_b^{-2}f,\quad
  f = \left( a\FreqPhi + \left(r^2+a^2\right)\sigma \right)\p_r\Delta_b - 4r\Delta_b\sigma.
\end{equation}
From this point forward, all calculations will be under the assumption
that we are working on the characteristic set on the exterior region
without boundary (so in particular, on the region where
$\PrinSymb_b=0$ and $\Delta_b>0$) since we are interested in locating
trapped null bicharacteristics and their properties, and along null
bicharacteristics, $\PrinSymb_b$ is necessarily zero. Recall that we
showed earlier that $a\FreqPhi + (r^2+a^2)\sigma$ is non-vanishing on
the characteristic set in the region where $\Delta_b>0$. Consequently,
we have that
\begin{equation*}
  \p_r\mathcal{F} = 0 \quad\text{if and only if}\quad f=0,
\end{equation*}
and 
\begin{equation*}
  f=0\implies \p_r\Delta_b\neq0.
\end{equation*}
If now in addition, $\p_r\mathcal{F}=0$, then
\begin{equation*}
  (r^2+a^2)\sigma-a\FreqPhi = \frac{4r\Delta_b\sigma}{\p_r\Delta_b}.
\end{equation*}
Thus we can calculate that
\begin{equation*}
  \p_{rr}\mathcal{F}
  = -\left(a\FreqPhi+\left(r^2+a^2\right)\sigma\right)\Delta_b^{-2}\p_r f
  = -\frac{4r\sigma}{\Delta_b(\p_r\Delta_b)^2}\p_r\Delta_b\p_rf. 
\end{equation*}
To specify the trapped set, We will show that for $g_b$ a
slowly-rotating \KdS{} metric, if $\Delta_b>0$, and
$\p_r\mathcal{F}= 0$, then $-\sigma\p_r\Delta_b\p_rf>0$.  To this end,
first calculate that
\begin{equation*}
  \p_rf = \left(
    a\FreqPhi + \left(r^2+a^2\right)\sigma
  \right)\p_{rr}\Delta_b -4 \Delta_b\sigma - 2r\sigma \p_r\Delta_b,
\end{equation*}
where
\begin{equation}
  \label{linear:eq:KdS-trapping:unstable-trapping-condition}
  \begin{split}
    \p_r\Delta_b\p_r f ={}& 4r\Delta_b\sigma\p_{rr}\Delta_b -2\sigma
    r(\p_r \Delta_b)^2 -
    4\sigma\Delta_b\p_r\Delta_b\\
    ={}&2\sigma\left(-\frac{1}{r}\left(r\p_r\Delta_b -
        4\Delta_b\right)^2-\frac{2\Delta_b}{r}\left(6Mr
        -8a^2\right)\right).
  \end{split}
\end{equation}
Then, observe that $\p_{rr}\Delta_b<0$ on the entire range where
$\Delta_b>0$. As a result from the first line of
\eqref{linear:eq:KdS-trapping:unstable-trapping-condition}
\begin{equation*}
  \p_r\Delta_b\ge 0 \implies -\sigma\p_r\Delta_b\p_rf > 0.
\end{equation*}
On the other hand, we also see from equation
\eqref{linear:eq:KdS-trapping:unstable-trapping-condition} that
\begin{equation*}
  6Mr - 8a^2 > 0 \implies -\sigma\p_r\Delta_b\p_rf > 0.
\end{equation*}
We can then evaluate that
\begin{equation*}
  \p_r\Delta_{(M,M)}\left({\frac{4}{3}M}\right) = \frac{158}{729}M + \frac{328}{729}M(1-9\Lambda M^2)>0. 
\end{equation*}
Thus, we see that for $a<M$,
$\p_r\Delta_{(M, a)}\left(\frac{4}{3}M\right) > 0$. As a result, for
$a<M$,
\begin{equation*}
  \p_r\Delta_{(M, a)} < 0 \implies 6Mr - 8a^2 > 0,
\end{equation*}
and we have that $-\sigma\p_r\Delta_b\p_rf > 0$ for
$a<M$.
Moreover, if $\p_r\mathcal{F}=0$, then
\begin{equation*}
  -\sigma \p_r\Delta_b \p_rf > 0.
\end{equation*}
Thus, if $\p_r\mathcal{F}=0$, then
\begin{equation*}
  \p_{rr}\mathcal{F} =
  -\left(a\FreqPhi+\left(r^2+a^2\right)\sigma\right)\Delta_b^{-2}\p_rf
  =
  -\frac{4r}{\Delta_b(\p_r\Delta_b)^2}\sigma\p_r\Delta_b\p_rf>0. 
\end{equation*}
This implies that the critical points of $\mathcal{F}$ are non-degenerate
minimum. In particular, as
\begin{equation*}
  \lim_{\Delta_b\to0}\mathcal{F} = \infty, 
\end{equation*}
$\rTrapping_b$, the critical point of $\mathcal{F}$,
\begin{enumerate}
\item exists and is unique in
  $(r_{\EventHorizonFuture}, r_{\CosmologicalHorizonFuture})$ for fixed
  $\FreqPhi$; 
\item depends smoothly on $\FreqPhi$;
\item lies in an $O(a)$ neighborhood of $r=3M$;
\end{enumerate}
and $\p_r\mathcal{F}$ has the same sign as $r-\rTrapping_b$, implying that
\begin{equation*}
  \Delta_b>0,
  \quad \PrinSymb_b=0
  \quad \pm (r-\rTrapping_b)>0,
  \quad H_{\PrinSymb_b} r=0,
  \implies \pm H^2_{\PrinSymb_b} r > 0.
\end{equation*}


\begin{remark}
  We remark that our above analysis does \emph{not} cover the full
  sub-extremal range of black hole parameters for \KdS{} except in the
  limiting $\Lambda=0$ case of the Kerr family.
  This is due to the fact that except in the $\Lambda=0$ case, $a=M$
  is not the extremal \KdS{} black hole, but we do not pursue this
  line of reasoning further here.
\end{remark}

\subsection{Proof of Lemma
  \ref{linear:lemma:subprincipal-op-def:sub-p-basic-props}}
\label{linear:appendix:lemma:subprincipal-op-def:sub-p-basic-props}

Let $e(x) = \curlyBrace*{e_i(x)}_{i=1}^N$ be a local frame for
$\mathcal{E}$. We can compute that
\begin{align*}
  &S_{sub}(P)\left(
  \sum_{jk}q_{jk}(x,\zeta)u_k(x,\zeta)e_j(x)
  \right)\\
  ={}& \sum_{j\ell}\left(\sum_k \sigma_{sub}(P)_{jk}q_{k\ell} - \ImagUnit H_p(q_{j\ell})\right)u_\ell e_j
     - \ImagUnit q_{j\ell}H_p(u_\ell)e_j
       - \ImagUnit q_{j\ell}u_{\ell}e_j H_p,\\
  &q S_{sub}(P)\left(
  \sum_{\ell}u_{\ell}(x,\zeta)e_\ell(x)
  \right)\\
  ={}& \sum_{j\ell}\left(
       \sum_k q_{jk}\sigma_{sub}(P)_{k\ell}
       \right) u_{\ell}e_j
       - \ImagUnit q_{j\ell}H_p(u_{\ell})e_j  -\ImagUnit q_{j\ell}u_\ell e_j H_p.
\end{align*}
As a result, we have that
\begin{align*}
  &\left(S_{sub}(P)q - q S_{sub}(P)\right)\sum_{\ell}u_{\ell}(x,\zeta)e_\ell(x)\\
  ={}& \sum_{j\ell}\left(\sum_k \left(\sigma_{sub}(P)_{jk}q_{k\ell} - q_{jk}\sigma_{sub}(P)_{k\ell} \right) - \ImagUnit H_p(q_{j\ell})\right)u_\ell e_j.
\end{align*}

\subsection{Proof of Lemma
  \ref{linear:lemma:subprincipal-op-def:Laplacian-sub-p-gen-form}}
\label{linear:appendix:lemma:subprincipal-op-def:Laplacian-sub-p-gen-form}

Both sides of \eqref{linear:eq:subprincipal-op-def:Laplacian-sub-p-gen-form}
are invariantly defined. As a result, it suffices to prove the
equality in some local coordinate system. We first consider the
left-hand side.  Fix an arbitrary point $x_0\in \mathcal{M}$ and
introduce normal coordinates $\curlyBrace*{y_i}_{i=1}^{N}$ such that the
associated Christoffel symbols vanish at $x_0$. Then
\begin{equation}
  \label{linear:eq:subprincipal-op-def:Laplacian-sub-p-gen-form:LHS}
  S_{sub}(\Laplace^{(k)})(x_0, \zeta) = -\ImagUnit H_{\abs*{\zeta}_g^2} = -2\ImagUnit g^{jk}\zeta_k\p_{y^j}. 
\end{equation}

We now evaluate the right-hand side of
\eqref{linear:eq:subprincipal-op-def:Laplacian-sub-p-gen-form}. For any
multi-index $I$, we denote
\begin{equation*}
  dx^I :=  dx^{i_1}\otimes \cdots\otimes dx^{i_k}.
\end{equation*}
Observe then that sections of $\pi^*T_k\mathcal{M}$ are of the form
$U_{I}(x, \zeta)dx^I$, and pullbacks of sections (under $\pi$) of
$T_k\mathcal{M}$ are of the form $u_I(x)dx^I$. By definition, the
pullback connection $\nabla^{\pi^*T_k\mathcal{M}}$ on pulled back
sections and extended to sections of the pullback bundle using the
Leibniz rule is given by
\begin{equation*}
  \nabla^{\pi^*T_k\mathcal{M}}_{\p_{x^j}}(u_i(x)dx^I) = \nabla^{T_k\mathcal{M}}_{\p_{x^j}}(u_i(x)dx^I),\qquad
  \nabla^{\pi^*T_k\mathcal{M}}_{\p_{\zeta^j}}(u_i(x)dx^I) = 0. 
\end{equation*}
As a result, we have that
\begin{align*}
  \nabla^{\pi^*T_k\mathcal{M}}_{\p_{x^j}}(u_I(x,\zeta)dx^I)
  &= \nabla^{T_k\mathcal{M}}_{\p_{x^j}}(u_i(\cdot, \zeta)dx^I)(x),\\
  \nabla^{\pi^*T_k\mathcal{M}}_{\p_{\zeta_k}}(u_I(x,\zeta)dx^I)
  &= \p_{\zeta_k}u_i(x, \zeta)dx^I.
\end{align*}
Thus in normal coordinates at $x_0\in \mathcal{M}$, we simply have
that
\begin{equation*}
  \nabla^{\pi^*T_k\mathcal{M}}_{\p_{x^j}} = \p_{x^j},\qquad
  \nabla^{\pi^*T_k\mathcal{M}}_{\p_{\zeta^j}} = \p_{\zeta^j}.
\end{equation*}
As a result,
\begin{equation}
  \label{linear:eq:subprincipal-op-def:Laplacian-sub-p-gen-form:RHS}
  \nabla^{\pi^*T_k\mathcal{M}}_{H_{\abs*{\zeta}_g^2}} = 2g^{jk}\zeta_k\p_{x^j}
\end{equation}
at $x_0$. Then we conclude from the equality between
\eqref{linear:eq:subprincipal-op-def:Laplacian-sub-p-gen-form:LHS} and
\eqref{linear:eq:subprincipal-op-def:Laplacian-sub-p-gen-form:RHS}.

\subsection{Proof of Lemma
  \ref{linear:lemma:subprincipal-op-def:Laplace-on-sphere-computation}}
\label{linear:appendix:lemma:subprincipal-op-def:Laplace-on-sphere-computation}

Fix a point $x\in \Sphere^2$ and consider geodesic normal coordinates
$y^1$ and $y^2$ such that the Christoffel symbols vanish at
$x$. Denote the dual variables on the fibers of $T^*\Sphere^2$ by
$\eta_1$ and $\eta_2$. Then we can calculate that in terms of the
geodesic normal coordinates and their dual variables, 
\begin{equation*}
  \evalAt*{H_{\abs*{\eta}^2}}_x = 2\evalAt*{\UnitSphereMetric^{ij}\eta_i\p_{y^j}}_x.
\end{equation*}
To see that
$\nabla^{\pi^*_{\Sphere^2}T^*\Sphere^2}_{H_{\abs*{\eta}^2}}$ preserves
sections of $E$, observe that  
\begin{equation*}
  \evalAt*{\nabla^{\pi^*_{\Sphere^2}T^*\Sphere^2}_{H_{\abs*{\eta}^2}}(\eta_k dy^k)}_x
  = 2\evalAt*{\UnitSphereMetric^{ij}\eta_i\eta_k\p_{y^j}dy^k}_x
  =0. 
\end{equation*}
On the other hand, to see that
$\nabla^{\pi^*_{\Sphere^2}T^*\Sphere^2}_{H_{\abs*{\eta}^2}}$ also
preserves sections of $F$, observe that for sections $\phi$, $\psi$ of
$\pi^*T^*\Sphere^2\to T^*\Sphere^2$, we have that
\begin{equation*}
  H_{\abs*{\eta}^2}(\UnitSphereInvMetric(\phi, \psi))
  = \UnitSphereInvMetric\left(\nabla^{\pi^*_{\Sphere^2}T^*\Sphere^2}_{H_{\abs*{\eta}^2}} \phi, \psi\right)
  + \UnitSphereInvMetric\left( \phi, \nabla^{\pi^*_{\Sphere^2}T^*\Sphere^2}_{H_{\abs*{\eta}^2}} \psi\right). 
\end{equation*}
Letting $\phi = \eta$, it is clear that if $\psi\cdot \eta = 0$, then
$\nabla^{\pi^*_{\Sphere^2}T^*\Sphere^2}_{H_{\abs*{\eta}^2}}\psi\cdot
\eta=0$, which proves that
$\nabla^{\pi^*_{\Sphere^2}T^*\Sphere^2}_{H_{\abs*{\eta}^2}}$ also
preserves the sections of $F$.

\subsection{Proof of Lemma \ref{linear:lemma:subprincipal-symbol-control}}
\label{linear:appendix:lemma:subprincipal-symbol-control}

We include the following computations for the sake of
completeness. The computations are based on the
similar proof carried out in Section 9 of \cite{hintz_global_2018}.

Let $(t,r,\omega)$ denote the Boyer-Lindquist
coordinates $(t,r)$ with $\omega$ coordinate on the sphere to be
specified.  We can then compute
that the only non-vanishing Christoffel symbols are:
\begin{align*}
  \tensor{\Gamma}{^t_{it}} &= \frac{1}{2}\mu_{b_0}^{-1}\partial_i\mu_{b_0} ,\\
  \tensor{\Gamma}{^r_{tt}} &= \frac{1}{2}\mu_{b_0}\partial_r\mu_{b_0},\\ 
  \tensor{\Gamma}{^r_{rr}} &= -\frac{1}{2}\mu_{b_0}^{-1}\partial_r\mu_{b_0},\\
  \tensor{\Gamma}{^r_{AB}} &= -r\mu_{b_0} \UnitSphereMetric_{AB},\\
  \tensor{\Gamma}{^A_{Br}} &= r^{-1}\tensor{\delta}{^A_B},\\
  \tensor{\Gamma}{^A_{BC}}&=\tensor[^{\UnitSphere^2}]{\Gamma}{^A_{BC}}.
\end{align*}
Recall that Greek letters are used to indicate spacetime indices,
lower-case Latin letters are used to indicate spatial indices, and
upper-case Latin letters are used to indicate spherical indices.  We
first calculate the subprincipal operator of the vector wave
operator. 

Let us denote by $\SubPOp_{b_0}^{(1)}$ the
subprincipal operator associated with the wave operator acting on one
forms $\VectorWaveOp[g_{b_0}]$. Then, by explicit computations in
local coordinates, we see that
\begin{equation*}
  \begin{split}
    \SubPOp_{b_0}^{(1)} {h}_t &= \mu_{b_0}^{-2}\p_r\mu_{b_0} \p_r {h}_t + \p_r\mu_{b_0} \p_t {h}_r,\\
    \SubPOp_{b_0}^{(1)} {h}_r &= -\mu_{b_0}^{-2}\p_r\mu_{b_0} \p_t {h}_t + \p_r \mu_{b_0}
    \p_r {h}_r
    - 2r^{-3}\UnitSphereMetric^{AB}\left(\p_B {h}_A + \p_A {h}_B\right),\\
    \SubPOp_{b_0}^{(1)} {h}_C &= 2r^{-1}\mu_{b_0}\tensor[]{\UnitSphereMetric}{^{B}_C}
                                \p_B {h}_r + \SubPOp[\LaplaceAngular^{(1)}]h_C,
  \end{split}
\end{equation*}
where $\LaplaceAngular^{(1)} := \slashed{\nabla}\cdot
\slashed{\nabla}$ denotes the angular Laplacian acting on one-forms,
and $\SubPOp[\LaplaceAngular^{(1)}]$ its subprincipal operator. 
To simplify our calculations at the trapped set, we define:
\begin{equation*}
  e_0 := \sqrt{\mu_{b_0}}^{-1}\p_t,\quad
  e^0:= \sqrt{\mu_{b_0}}dt,\quad
  e_1 = \sqrt{\mu_{b_0}}\p_r,\quad
  e^1:=\sqrt{\mu_{b_0}}^{-1}dr.
\end{equation*}
Then, let us decompose
\begin{equation}
  \label{linear:eq:SdS-decomp}
  \begin{split}
    T^*\mathcal{M} &= \bangle{e^0}\oplus \bangle{e^1}\oplus T^*\UnitSphere^2,\\
    S^2T^*\mathcal{M} &= \bangle{e^0e^0}
    \oplus\left(\bangle{2e^0e^1}
      \oplus \bangle{2e^0\omega} \right)
    \oplus \left(\bangle{e^1e^1}
      \oplus\bangle{2e^1\omega}\
      \oplus S^2T^*\UnitSphere^2 \right),
    \quad \omega\in T^*\UnitSphere^2.
  \end{split}
\end{equation}
Now recall that 
\begin{equation*}
  \TrappedSet_{b_0} = \curlyBrace*{r=3M, \xi =0, \sigma^2 = \mu_{b_0}r^{-2}|\FreqAngular|^2}.
\end{equation*}
We can then write the invariant subprincipal symbol of the vector-valued-wave
operator in \SdS{} using the decomposition in (\ref{linear:eq:SdS-decomp}),
\begin{align*}
  \left.S_{sub}(\VectorWaveOp[g_{b_0}])\right\vert_{\TrappedSet_{b_0}}
  ={}&
       -2\mu_{b_0}^{-1}\sigma D_{t}
       + \ImagUnit r^{-2}
       \begin{pmatrix}
         H_{\abs*{\eta}^2}& 0 & 0 \\
         0 & H_{\abs*{\eta}^2} & 0 \\
         0 & 0 & \nabla^{\pi^*_{\Sphere^2}T^*\Sphere^2}_{H_{\abs*{\eta}^2}}
       \end{pmatrix}\\
     & + \ImagUnit
       \begin{pmatrix}
         0 & -2r^{-1}\sigma & 0\\
         -2r^{-1}\sigma&0& -2\sqrt{\mu_{b_0}}r^{-3}\InteriorProd_\FreqAngular\\
         0& 2r^{-1}\sqrt{\mu_{b_0}}\FreqAngular & 0
       \end{pmatrix},
\end{align*}
where we use $\InteriorProd_\FreqAngular$ to denote the interior
product with respect to $\FreqAngular$.
In a similar vein, we can compute using local coordinates that:
\begin{equation*}
  \begin{split}
    \SubPOp_{b_0} {h}_{tt} ={}&
                                -2\p_r\mu_{b_0} \p_r {h}_{tt}
                                + 2 \p_r\mu_{b_0} \p_t {h}_{tr},
    \\
    \SubPOp_{b_0} {h}_{tr} ={}&
                                - 2 r^{-3}\UnitSphereInvMetric^{AB}\p_A {h}_{Bt} +
                                \p_r\mu_{b_0} \p_t {h}_{rr}
                                + \frac{\p_r\mu_{b_0}}{\mu_{b_0}^{2}}\p_t {h}_{tt},
    \\
    \SubPOp_{b_0} {h}_{tC} ={}&
                                \left(- 2r^{-1}\mu_{b_0} + \p_r\mu_{b_0}\right)\p_r {h}_{tC}
                                + \p_r\mu_{b_0} \p_t {h}_{Cr}
                                + 2r^{-1}\mu_{b_0} \p_C {h}_{tr}
                                + r^{-2}\SubPOp[\LaplaceAngular^{(1)}]h_{tC}
                                ,
    \\
    \SubPOp_{b_0} {h}_{rr} ={}&
                                2\p_r\mu_{b_0} \p_r {h}_{rr}
                                - 4r^{-3}\UnitSphereInvMetric^{AB}\p_A {h}_{rB}
                                + \frac{2\p_r\mu_{b_0}}{\mu_{b_0}^2}\p_t {h}_{tr},
    \\
    \SubPOp_{b_0} {h}_{rC} ={}&
                                \left(-2r^{-1}\mu_{b_0} + \p_r\mu_{b_0} \right)\p_r {h}_{rC}
                                + 2r^{-1}\mu_{b_0} \p_C {h}_{rr}
                                - 2r^{-3}\UnitSphereInvMetric^{AB}\p_A {h}_{BC}\\
                              &+ \frac{\p_r\mu_{b_0}}{\mu_{b_0}^2}\p_t {h}_{tC}
                                + r^{-2}\SubPOp[\LaplaceAngular^{(1)}]h_{rC},
    \\
    \SubPOp_{b_0} {h}_{CD} ={}&
                                -4r^{-1}\mu_{b_0} \p_r {h}_{CD}
                                + 2r^{-1}\mu_{b_0}\p_C {h}_{Dr}
                                + r^{-2}\SubPOp[\LaplaceAngular^{(2)}]h_{CD},
  \end{split}
\end{equation*}
where $\SubPOp[\LaplaceAngular^{(2)}]h_{CD}$ denotes the
subprincipal symbol of the angular Laplacian acting on two-tensors on
the sphere.

Now, using the splitting in \eqref{linear:eq:SdS-decomp} we can write that

\begin{equation*}
  \evalAt*{S_{sub}(\TensorWaveOp[g_{b_0}])}_{\TrappedSet_{b_0}}
  = -2\mu_{b_0}D_t + \ImagUnit r^{-2}
  \begin{pmatrix}
    H_{\abs*{\eta}^2} & 0 & 0 & 0 & 0 & 0 \\
    0 & H_{\abs*{\eta}^2} & 0 & 0 & 0 & 0 \\
    0 & 0 & \nabla^{\pi^*_{\Sphere^2}T^*\Sphere^2}_{H_{\abs*{\eta}^2}} & 0 & 0 & 0 \\
    0  & 0 & 0 & H_{\abs*{\eta}^2} & 0 & 0 \\
    0 & 0 & 0 & 0  & \nabla^{\pi^*_{\Sphere^2}T^*\Sphere^2}_{H_{\abs*{\eta}^2}} & 0 \\
    0 & 0 & 0 & 0 & 0   & \nabla^{\pi^*_{\Sphere^2}T_2\Sphere^2}_{H_{\abs*{\eta}^2}}
  \end{pmatrix}
  + \SubPSym_{b_0, \TrappedSet_{b_0}},
\end{equation*}
where 
\begin{equation*}
  \SubPSym_{b_0,\TrappedSet_{b_0}} =
  \begin{pmatrix}[1.5]
    0 & -4r^{-1}\sigma & 0 & 0 & 0 & 0\\
    -2r^{-1}\sigma & 0 & -2\sqrt{\mu_{b_0}}r^{-3}\InteriorProd_\FreqAngular & -2r^{-1}\sigma & 0 & 0\\
    0 & 2\sqrt{\mu_{b_0}}r^{-1} \FreqAngular & 0 & 0 & -2r^{-1}\sigma & 0 \\
    0 & -4 r^{-1}\sigma & 0 & 0 & -4\sqrt{\mu_{b_0}}r^{-3}\InteriorProd_{\FreqAngular}&0\\
    0 & 0 & -2r^{-1}\sigma & 2\sqrt{\mu_{b_0}}r^{-1}\FreqAngular & 0 & -2r^{-3}\InteriorProd_{\FreqAngular}\\
    0 & 0 & 0 & 0 & 2\sqrt{\mu_{b_0}}r^{-1}\FreqAngular & 0 
  \end{pmatrix}.
\end{equation*}

For the sake of simplifying calculations, we split
\begin{equation}
  \label{linear:eq:two-sphere-decomp}
  T^*\UnitSphere^2 = \bangle{r \abs*{\FreqAngular}^{-1}\FreqAngular}\oplus \FreqAngular^\perp,
\end{equation}
which induces the splitting
\begin{equation*}
  S^2T^*\UnitSphere^2 = \bangle{(r \abs*{\FreqAngular}^{-1}\FreqAngular)^2}\oplus 2r \abs*{\FreqAngular}^{-1}\FreqAngular\cdot \FreqAngular^\perp\oplus S^2\FreqAngular^\perp.
\end{equation*}
We thus can further split the splitting (\ref{linear:eq:SdS-decomp}) to:
\begin{equation}
  \label{linear:eq:SdS-decomp-refined}
  \begin{split}
    S^2T^*M ={}
    & \bangle{e^0e^0}\oplus\bangle{2e^0e^1}\oplus (\bangle{2e^0r\abs*{\FreqAngular}^{-1}\FreqAngular}\oplus 2e^0\cdot \FreqAngular^\perp)\\
    &\oplus \bangle{e^1e^1}\oplus (\bangle{2e^1r\abs*{\FreqAngular}^{-1}\eta}\oplus 2e^1\cdot\FreqAngular^\perp)\\
    &\oplus (\bangle{(r\abs*{\FreqAngular}^{-1}\FreqAngular)^2}\oplus 2r\abs*{\FreqAngular}^{-1}\FreqAngular\cdot \FreqAngular^{\perp}\oplus S^2\FreqAngular^\perp).
  \end{split}
\end{equation}
Using that on $\TrappedSet_{b_0}$, $r^2\mu_{b_0}^{-1}
\sigma^2 = |\FreqAngular|^2$, we can compute that in the decomposition
(\ref{linear:eq:two-sphere-decomp}), $\FreqAngular:\Real\to
T^*\UnitSphere^2$, $\InteriorProd_\FreqAngular:
T^*\UnitSphere^2\to\Real$ are given by:
\begin{equation*}
  \FreqAngular =
  \begin{pmatrix}[1.5]
    r^{-1}\abs*{\eta}\\
    0
  \end{pmatrix},\quad
  \InteriorProd_{\FreqAngular} =
  \begin{pmatrix}
    r\abs*{\eta} & 0 
  \end{pmatrix}.
\end{equation*}
Thus, $\FreqAngular: T^*\UnitSphere^2\to S^2 T^*\UnitSphere^2$,
$\InteriorProd_\FreqAngular: S^2T^*\UnitSphere^2\to T^*\UnitSphere^2$
are given by
\begin{equation*}
  \FreqAngular =
  \begin{pmatrix}[1.5]
    r^{-1}\abs*{\eta} & 0 \\
    0 & \frac{1}{2}r^{-1}\abs*{\eta}\\
    0 & 0
  \end{pmatrix}
  ,\quad \InteriorProd_\FreqAngular =
  \begin{pmatrix}[1.5]
    r\abs*{\eta} & 0 & 0\\
    0 & r\abs*{\eta} & 0
  \end{pmatrix}.
\end{equation*}
From Lemma~\ref{linear:lemma:subprincipal-op-def:Laplace-on-sphere-computation}, we
have that the only non-diagonal component of
$S_{sub}(\TensorWaveOp[g_{b_0}])$ is precisely $\SubPSym_{g_{b_0},
  \TrappedSet_{b_0}}$. Thus, it suffices to find some
$\PseudoSubPFixer$ such that
\begin{equation*}
  \PseudoSubPFixer \SubPSym_{g_{b_0},\TrappedSet_{b_0}} \PseudoSubPFixer^{-} < \varepsilon_{\TrappedSet}.
\end{equation*}

To this end, observe that in the splitting of
(\ref{linear:eq:SdS-decomp-refined}),
\begin{equation*}
  \SubPSym_{b_0, \TrappedSet_{b_0}} = \ImagUnit r^{-2}\sqrt{\mu_{b_0}}\abs*{\eta}
  \begin{pmatrix}
    0&-4&0&0&0&0&0&0&0&0\\
    -2&0&-2&0&-2&0&0&0&0&0\\
    0&2&0&0&0&-2&0&0&0&0\\
    0&0&0&0&0&0&-2&0&0&0\\
    0&-4&0&0&0&-4&0&0&0&0\\
    0&0&-2&0&2&0&0&-2&0&0\\
    0&0&0&-2&0&0&0&0&-2&0\\
    0&0&0&0&0&4&0&0&0&0\\
    0&0&0&0&0&0&2&0&0&0\\
    0&0&0&0&0&0&0&0&0&0
  \end{pmatrix}.
\end{equation*}
The crucial observation is then that $\PseudoSubPFixer_{b_0}\in
\SymClass^0$, where 
\begin{equation}
  \label{linear:eq:Q-def}
  \PseudoSubPFixer_{b_0} :=
  \begin{pmatrix}[1.5]
    0&0&0&0&0&\frac{96}{\varepsilon_{\TrappedSet_{b_0}}}&0&0&0&0\\
    0&0&0&0&0&0&-\frac{24}{\varepsilon_{\TrappedSet_{b_0}}}&0&0&0\\
    -\frac{1}{3}&0&0&0&0&-\frac{96}{\varepsilon_{\TrappedSet_{b_0}}}&0&\frac{4}{\varepsilon_{\TrappedSet_{b_0}}^2}&0&0\\
    0&0&\frac{4}{\varepsilon_{\TrappedSet_{b_0}^2}}&0&0&0&0&0&0&0\\
    \frac{1}{3}&0&0&0&0&0&0&\frac{8}{\varepsilon_{\TrappedSet_{b_0}}^2}&0&0\\
    0&0&0&0&0&0&\frac{24}{\varepsilon_{\TrappedSet_{b_0}}^3}&0&-\frac{2}{\varepsilon_{\TrappedSet_{b_0}}}&0\\
    0&0&0&-\frac{2}{\varepsilon_{\TrappedSet_{b_0}}}&0&0&0&0&0&0\\
    \frac{2}{3}&0&0&0&0&\frac{96}{\varepsilon_{\TrappedSet_{b_0}}^4}&0&-\frac{8}{\varepsilon_{\TrappedSet_{b_0}}^2}&0&1\\
    0&0&-\frac{4}{\varepsilon_{\TrappedSet_{b_0}}^2}&0&1&0&0&0&0&0\\
    0&1&0&0&0&0&0&0&0&0
  \end{pmatrix},
\end{equation}
is well defined in a neighborhood of where the splitting in
(\ref{linear:eq:SdS-decomp-refined}) is well-defined, in particular, in a
neighborhood of $\TrappedSet_{b_0}$, and furthermore, $\PseudoSubPFixer$
satisfies the property that
\begin{equation}
  \label{linear:eq:sub-op:conjugated-form}
  \left.\PseudoSubPFixer^{-1}\SubPSym_{b_0}\PseudoSubPFixer\right\vert_{\TrappedSet_{b_0}}
  = \ImagUnit  \varepsilon_{\TrappedSet_{b_0}}r^{-2}\sqrt{\mu_{b_0}}\abs*{\FreqAngular}
  \begin{pmatrix}
    0&0&0&0&0&0&0&0&0&0\\
    0&0&0&0&0&0&0&0&0&0\\
    0&0&0&1&0&0&0&0&0&0\\
    0&0&0&0&1&0&0&0&0&0\\
    0&0&0&0&0&0&0&0&0&0\\
    0&0&0&0&0&0&1&0&0&0\\
    0&0&0&0&0&0&0&1&0&0\\
    0&0&0&0&0&0&0&0&1&0\\
    0&0&0&0&0&0&0&0&0&1\\
    0&0&0&0&0&0&0&0&0&0
  \end{pmatrix}.
\end{equation}
Observe that $\PseudoSubPFixer^-$ is a constant coefficient operator on
the region where the splitting in \eqref{linear:eq:SdS-decomp-refined} is
valid (in particular, in a neighborhood of the trapped set), so
clearly commutes with $\ScalarWaveOp[g_{b_0}]$.

\subsection{Proof of Proposition \ref{linear:prop:div-thm:PDO-modification}}
\label{linear:appendix:prop:div-thm:PDO-modification}

We remark that while we will prove the proposition for the specific
operator
$\LinEinsteinConj_{g_b} = \ScalarWaveOp[g_b] + \SubPConjOp_{b} +
\PotentialConjOp_b$, the conclusions are true for the more general
strongly-hyperbolic operator $\LinearOp = \ScalarWaveOp[g_b] + S +
V$.

Integrating by parts, we have that for $h$ compactly supported on
$\DomainOfIntegration$ as defined in
\eqref{linear:eq:ILED-trapping:trapping-reg-def},
\begin{align}
  2\Re\bangle*{\LinEinsteinConj_{g_b}{h},\widetilde{\MorawetzVF}{h}}_{L^2(\DomainOfIntegration)}
  ={}& \bangle*{\ScalarWaveOp[g_b]{{h}}, \widetilde{\MorawetzVF}{h}}_{L^2(\DomainOfIntegration)}
       + \bangle*{\widetilde{\MorawetzVF}{h},\ScalarWaveOp[g_b]{{h}}}_{L^2(\DomainOfIntegration)}
       + \bangle*{\SubPConjOp_{b}[{h}], \widetilde{\MorawetzVF}{h}}_{L^2(\DomainOfIntegration)} \notag\\
     & + \bangle*{ \widetilde{\MorawetzVF}{h}, \SubPConjOp_{b}[{h}]}_{L^2(\DomainOfIntegration)} 
         + 2\Re\bangle*{ \PotentialConjOp_b{h}, \widetilde{\MorawetzVF}{h}}_{L^2(\DomainOfIntegration)} \notag\\
  ={}& \Re \bangle*{\squareBrace*{\ScalarWaveOp[g_b],\widetilde{\MorawetzVF}}{h}, {h}}_{L^2(\DomainOfIntegration)}
       - \Re\bangle*{\SubPConjOp_{b,a}[\widetilde{\MorawetzVF}{h}] + \widetilde{\MorawetzVF} \SubPConjOp_{b,a}[{h}], {h}}_{L^2(\DomainOfIntegration)}\notag \\
     & + \Re\bangle*{\squareBrace*{\SubPConjOp_{b,s},\widetilde{\MorawetzVF}}{h}, {h}}_{L^2(\DomainOfIntegration)}
       + 2\Re\bangle*{\widetilde{\MorawetzVF}{h}, \PotentialConjOp_b{h}}_{L^2(\DomainOfIntegration)}.\label{linear:eq:int-by-parts-arg:X-eqn}\\  
  2\Re\bangle*{\LinEinsteinConj_{g_b}{h}, \tilde{\LagrangeCorr}{h}}_{L^2(\DomainOfIntegration)}
  ={}& \bangle*{\left(\ScalarWaveOp[g_b]\tilde{\LagrangeCorr} + \tilde{\LagrangeCorr}\ScalarWaveOp[g_b]\right){h}, {h}}_{L^2(\DomainOfIntegration)}
       + 2\Re\bangle*{\SubPConjOp_b[{h}], \tilde{\LagrangeCorr}{h}}_{L^2(\DomainOfIntegration)}
       + 2\Re\bangle*{\PotentialConjOp_b{h}, \tilde{\LagrangeCorr}{h}}_{L^2(\DomainOfIntegration)}. \label{linear:eq:int-by-parts-arg:q-eqn}
\end{align}
We remark that \eqref{linear:eq:int-by-parts-arg:X-eqn} follows from the
observation that for $h$ compactly supported on
$\DomainOfIntegration$
\begin{align*}
  2\Re \bangle*{\SubPConjOp_b[{h}], \widetilde{\MorawetzVF}{h}}_{L^2(\DomainOfIntegration)}
  ={}& \bangle*{\SubPConjOp_b[{h}], \widetilde{\MorawetzVF}{h}}_{L^2(\DomainOfIntegration)}
  + \bangle*{\widetilde{\MorawetzVF}{h}, \SubPConjOp_b[{h}]}_{L^2(\DomainOfIntegration)}  \\
  ={}& -\bangle*{\widetilde{\MorawetzVF}\SubPConjOp_b[{h}], {h}}_{L^2(\DomainOfIntegration)}
  + \bangle*{\SubPConjOp_{b,s}\widetilde{\MorawetzVF}{h}, [{h}]}_{L^2(\DomainOfIntegration)}
  - \bangle*{\SubPConjOp_{b,a}\widetilde{\MorawetzVF}{h}, [{h}]}_{L^2(\DomainOfIntegration)}.
\end{align*}
As a result,
\begin{equation*}
  \Re\bangle*{\LinEinsteinConj_{g_b}{h}, (\widetilde{\MorawetzVF}+\tilde{\LagrangeCorr}){h}}_{L^2(\DomainOfIntegration)} = -\Re\KCurrentIbP{\widetilde{\MorawetzVF}, \tilde{\LagrangeCorr}}[{h}],
\end{equation*}
where
\begin{equation}
  \label{linear:eq:int-by=parts-arg:KCurrentIbP-def}
  \KCurrentIbP{\widetilde{\MorawetzVF}, \tilde{\LagrangeCorr}}[{h}]
  = \KCurrentIbP{\widetilde{\MorawetzVF}, \tilde{\LagrangeCorr}}_{(2)}[{h}]
  +\KCurrentIbP{\widetilde{\MorawetzVF}, \tilde{\LagrangeCorr}}_{(1)}[{h}]
  +\KCurrentIbP{\widetilde{\MorawetzVF}, \tilde{\LagrangeCorr}}_{(0)}[{h}],
\end{equation}
and
\begin{equation*}
  \begin{split}
    2\KCurrentIbP{\widetilde{\MorawetzVF}, \tilde{\LagrangeCorr}}_{(2)}[{h}]
    &= \bangle*{\squareBrace*{\widetilde{\MorawetzVF}, \ScalarWaveOp[g_b]}{h}, {h}}_{L^2(\DomainOfIntegration)}
    + \bangle*{\SubPConjOp_{b,a}[\widetilde{\MorawetzVF}{h}] + \widetilde{\MorawetzVF} \SubPConjOp_{b,a}[{h}], {h}}_{L^2(\DomainOfIntegration)}
    - \bangle*{\left(\ScalarWaveOp[g_b]\tilde{\LagrangeCorr} + \tilde{\LagrangeCorr}\ScalarWaveOp[g_b]\right){h}, {h}}_{L^2(\DomainOfIntegration)},\\
    2\KCurrentIbP{\widetilde{\MorawetzVF}, \tilde{\LagrangeCorr}}_{(1)}[{h}]
    &= \bangle*{\squareBrace*{\widetilde{\MorawetzVF},\SubPConjOp_{b,s}}{h}, {h}}_{L^2(\DomainOfIntegration)}
    - 2\bangle*{\widetilde{\MorawetzVF}{h}, \PotentialConjOp_b{h}}_{L^2(\DomainOfIntegration)}
    - 2\bangle*{\SubPConjOp_{b}[{h}], \tilde{\LagrangeCorr}{h}}_{L^2(\DomainOfIntegration)},\\
    2\KCurrentIbP{\widetilde{\MorawetzVF}, \tilde{\LagrangeCorr}}_{(0)}[{h}]
    &= - 2\bangle*{\PotentialConjOp_b{h}, \tilde{\LagrangeCorr}{h}}_{L^2(\DomainOfIntegration)}. 
  \end{split}
\end{equation*}
Unfortunately, although the geometry of the exterior domain of outer
communication of \KdS{} allows us to consider only solutions ${h}$
with compact support in $r$, we are rarely interested in ${h}$ with
compact support in $\tStar$.  Thus, we must be careful integrating by
parts in $\p_{\tStar}$, since we will necessarily pick up boundary
terms. It is because of this integration by parts in $\p_{\tStar}$ that we
require the specific form of $\widetilde{\MorawetzVF}$,
$\tilde{\LagrangeCorr}$ in
\eqref{linear:eq:div-thm:PDO:multipler-def}\footnote{Observe that an
  integration-by-parts argument integrating by parts in $\p_{\tStar}$
  is not possible for the general class of pseudo-differential
  operators.}.

Integrating by parts in $\p_{\tStar}$ we then repeat the above
calculations to see that 
\begin{equation}
  \label{linear:eq:ILED-near:Int-by-parts-1}
  - \Re\bangle*{\LinEinstein_{g_b} {h}, \left(\widetilde{\MorawetzVF} +\tilde{\LagrangeCorr} \right){h}}_{L^2(\DomainOfIntegration)}
  = \Re\KCurrentIbP{\widetilde{\MorawetzVF}, \tilde{\LagrangeCorr}}[{h}]
  + \left.\Re\JCurrentIbP{\widetilde{\MorawetzVF}, \tilde{\LagrangeCorr}}(\tStar)[{h}]\right\vert_{\tStar=0}^{\tStar = \TStar},
\end{equation}
where, 
\begin{align}
  2 \KCurrentIbP{\widetilde{\MorawetzVF}, \tilde{\LagrangeCorr}}[{h}]
  ={}& \bangle*{\left[\widetilde{\MorawetzVF},\ScalarWaveOp[g_b] \right]{h},{h}}_{L^2(\DomainOfIntegration)}
       +\bangle*{\left[\widetilde{\MorawetzVF}, \SubPConjOp_{b,s} \right]h,h}_{L^2(\DomainOfIntegration)}
       +\bangle*{\left(\widetilde{\MorawetzVF} \SubPConjOp_{b,a} + \SubPConjOp_{b,a} \widetilde{\MorawetzVF}\right){h},{h}}_{L^2(\DomainOfIntegration)}\notag\\
     &- \bangle*{\left(\tilde{\LagrangeCorr}\ScalarWaveOp[g_b]+ \ScalarWaveOp[g_b]\tilde{\LagrangeCorr}\right){h}, {h}}_{L^2(\DomainOfIntegration)}
       - 2\bangle*{\SubPConjOp_b{h}, \tilde{\LagrangeCorr} {h}}_{L^2(\DomainOfIntegration)}
       -  2\bangle*{\PotentialConjOp_b{h},\left(\widetilde{\MorawetzVF}+\tilde{\LagrangeCorr}\right){h}}_{L^2(\DomainOfIntegration)},\notag\\   
  \JCurrentIbP{\widetilde{\MorawetzVF}, \tilde{\LagrangeCorr}}(\tStar)[{h}]
  ={}& \bangle*{n_{\TrappingNbhd}{h}, \widetilde{\MorawetzVF}{h}}_{\LTwo(\TrappingNbhd_{\tStar})}
       + \bangle*{\SubPConjOp{h}, \widetilde{\MorawetzVF}_0 g_b(\KillT, n_{\TrappingNbhd}){h}}_{\LTwo(\TrappingNbhd_{\tStar})}
       + \bangle*{g_b(\KillT, n_{\TrappingNbhd})\SubPConjOp_0{h}, \widetilde{\MorawetzVF}h}_{L^2(\TrappingNbhd)}
  \notag\\
     & + \bangle*{n_{\TrappingNbhd}{h}, \tilde{\LagrangeCorr}{h}}_{\LTwo(\TrappingNbhd_{\tStar})}      . \label{linear:eq:ILED-near:pseudo-mod:J-K-def}
\end{align}
Combining this with the physical divergence theorem in Proposition
\ref{linear:prop:energy-estimates:spacetime-divergence-prop} then yields the
conclusion.

\printbibliography[heading=bibintoc]

\end{document}
